\documentclass[11pt,a4paper]{book}
\usepackage[T1]{fontenc}
\usepackage{amsmath,amsfonts,amssymb,amsthm,cite}
\usepackage{fancyhdr}
\usepackage[all]{xy}
\usepackage{feynmf}
\usepackage[frenchb]{babel} 

\addto\captionsfrancais{}
\addto\captionsfrancais{}
\addto\captionsfrancais{}
\evensidemargin=\oddsidemargin
\advance\topmargin by -\headheight
\advance\topmargin by -\headsep

\theoremstyle{plain}
\newtheorem{thm}{Theorem}[section]
\newtheorem{lem}[thm]{Lemma}
\newtheorem{assum}[thm]{Assumption}
\newtheorem{cly}[thm]{Corollary}

\newtheorem*{theofr}{Théorème}

\newtheorem{rem}[thm]{Remark}

\newtheorem{theorem}[thm]{Theorem}
\newtheorem{corollary}[thm]{Corollary}
\newtheorem{conjecture}[thm]{Conjecture}
\newtheorem{prop}[thm]{Proposition}
\newtheorem{remark}[thm]{Remark}
\newtheorem{lemma}[thm]{Lemma}

\theoremstyle{definition}
\newtheorem{defn}[thm]{Definition}
\newtheorem*{deffr}{Définition}
\newtheorem{definition}[thm]{Definition}

\DeclareMathOperator{\Op}{\mathfrak{Op}} 
\DeclareMathOperator{\cof}{cof}          
\DeclareMathOperator{\Exp}{Exp}          
\DeclareMathOperator{\Lin}{Lin}           
\DeclareMathOperator{\Tad}{Tad}                 
\DeclareMathOperator{\Dom}{Dom}          
\DeclareMathOperator{\Diff}{Diff}          
\DeclareMathOperator{\Id}{Id}              
\DeclareMathOperator{\Ima}{Im}           
\DeclareMathOperator{\Ker}{Ker}          
\DeclareMathOperator{\Res}{Res}          
\DeclareMathOperator{\Tr}{Tr}            
\DeclareMathOperator{\sgn}{sgn}       
\DeclareMathOperator{\argch}{argch}    
\DeclareMathOperator{\argsh}{argsh}    
\DeclareMathOperator{\sech}{sech}    

\newcommand{\E}{\mathcal{E}}						
\newcommand{\Th}{\Theta}
\renewcommand{\th}{\theta}
\newcommand{\A}{\mathcal{A}}             
\newcommand{\Abb}{\mathbb{A}}          
\renewcommand{\a}{\alpha}                
\renewcommand{\b}{\beta}                
\newcommand{\bfr}{\mathfrak{b}}				
\newcommand{\B}{\mathcal{B}}             
\newcommand{\C}{\mathbb{C}}              
\newcommand{\CC}{\mathcal{C}}              
\newcommand{\CR}{\mathcal{R}}           
\newcommand{\D}{\mathbb{D}}              
\newcommand{\Coo}{C^\infty}              
\newcommand{\del}{\partial}              
\newcommand{\DD}{\mathcal{D}}            
\newcommand{\eps}{\varepsilon}           
\newcommand{\F}{\mathcal{F}}						
\newcommand{\Ga}{\Gamma}                 
\newcommand{\ga}{\gamma}                 
\renewcommand{\H}{\mathcal{H}}           
\newcommand{\HH}{\mathbb{H}}          
\newcommand{\half}{{\mathchoice{\thalf}{\thalf}{\shalf}{\shalf}}}
\newcommand{\hideqed}{\renewcommand{\qed}{}} 
\newcommand{\K}{\mathcal{K}}             
\newcommand{\ka}{\kappa}								
\renewcommand{\L}{\mathcal{L}}           
\newcommand{\la}{\lambda}                
\newcommand{\N}{\mathbb{N}}              
\newcommand{\M}{\mathcal{M}}              
\newcommand{\ox}{\otimes}                
\newcommand{\om}{\omega}                
\newcommand{\pa}{\partial}
\newcommand{\Q}{\mathbb{Q}}              
\newcommand{\ol}{\overline}             
\renewcommand{\O}{\mathcal{O}}						
\newcommand{\wox}{\wh \otimes}					
\newcommand{\R}{\mathbb{R}}              
\newcommand{\RR}{\mathcal{R}}						
\newcommand{\rt}{\mathrm{t}}             
\newcommand{\rx}{\mathrm{x}}             
\newcommand{\ry}{\mathrm{y}}             
\newcommand{\rv}{\mathrm{v}}             
\newcommand{\ru}{\mathrm{u}}             

\newcommand{\set}[1]{\{\,#1\,\}}         
\newcommand{\shalf}{{\scriptstyle\frac{1}{2}}}  
\renewcommand{\S}{\mathcal{S}}          
\renewcommand{\SS}{\mathcal{S}}          
\newcommand{\T}{\mathbb{T}}              
\newcommand{\thalf}{\tfrac{1}{2}}        
\renewcommand{\Th}{\Theta}                
\renewcommand{\th}{\theta}              

\newcommand{\Ups}{\Upsilon}							

\newcommand{\vth}{\vartheta}            
\newcommand{\wh}{\widehat}              
\newcommand{\wt}{\widetilde}            
\newcommand{\Z}{\mathbb{Z}}              
\renewcommand{\.}{\cdot}                 

\def\<#1,#2>{\langle#1\,,\,#2\rangle}        
\newcommand{\norm}[1]{\left\lVert#1\right\rVert}   
\newcommand{\Uq}{{\mathcal U}_q(su(2))}  
\newcommand{\U}{\mathcal{U}}              
\newcommand{\SUq}{SU_q(2)}

\def\T{{\cal T}}

\newcommand{\be}{\begin{enumerate}}
\newcommand{\ee}{\end{enumerate}}

\newbox\ncintdbox \newbox\ncinttbox
\setbox0=\hbox{$-$} \setbox2=\hbox{$\displaystyle\int$}
\setbox\ncintdbox=\hbox{\rlap{\hbox
    to \wd2{\kern-.1em\box2\relax\hfil}}\box0\kern.1em}
\setbox0=\hbox{$\vcenter{\hrule width 4pt}$}
\setbox2=\hbox{$\textstyle\int$}
\setbox\ncinttbox=\hbox{\rlap{\hbox
    to \wd2{\kern-.14em\box2\relax\hfil}}\box0\kern.1em}
\newcommand{\ncint}{\mathop{\mathchoice{\copy\ncintdbox}
    {\copy\ncinttbox}{\copy\ncinttbox}
    {\copy\ncinttbox}}\nolimits}
\renewcommand{\up}{{\mathord{\uparrow}}}  
\newcommand{\dn}{{\mathord{\downarrow}}}
\newcommand{\ooh}{{\tfrac{3}{2}}}       
\newcommand{\oh}{{\tfrac{1}{2}}}          
\newcommand{\ssesq}{{\scriptstyle\frac{3}{2}}} 
\newcommand{\sesq}{{\mathchoice{\ooh}{\ooh}{\ssesq}{\ssesq}}} 
\newcommand{\kett}[1]{|#1\rangle\!\rangle}  
\newcommand{\piappr}{\underline{\pi}}         
\newcommand{\ul}[1]{\underline{#1}}             
\newcommand{\ket}[1]{|#1\rangle}                  
\newcommand{\sg}{\sigma}                              
\newcommand{\twobyone}[2]{\begin{pmatrix}#1\\#2\end{pmatrix}}

\begin{document}

\NoAutoSpaceBeforeFDP
\thispagestyle{empty}

\begin{center}
{\large \bf  Universit\'e  de Provence,  Aix-Marseille I}\\

\vspace{1.3cm}

{\Large \bf TH\`ESE}\\

\vspace{1.3cm}

Pr\'esent\'ee par\\

\bigskip

{\huge \bf  Cyril LEVY}\\

\vspace{1.5cm}

pour obtenir le grade de\\

\medskip

{\large \bf  Docteur de l'Universit\'e de Provence}\\

\bigskip

{\bf  Sp\'ecialit\'e}: Physique Th\'eorique, Physique Math\'ematique\\

\bigskip

{ \bf \'Ecole Doctorale}: Physique et Science de la Mati\`ere (ED 0352)\\

\vspace{1.2cm}

{\bf  \large Titre:}

\bigskip

{\huge \bf {Spectral Action }}

\bigskip

{\huge \bf {in Noncommutative Geometry }}

\bigskip

{\huge \bf { and Global Pseudodifferential Calculus}} 

\end{center}

\vspace{1cm}
\begin{center}
\bf{Soutenue publiquement le 12 Juin 2009 devant le jury composé de}
\end{center}
\medskip
\hspace{-1cm}
\begin{tabular}{ l   }
  \vspace{0.1cm}
  \large {\bf Moulay-Tahar Benameur}, Professeur à l'Université Paul Verlaine, Metz   \\
  \vspace{0.1cm}
  \large {\bf Ludwik D\c{a}browski}, Professeur à SISSA, Trieste (Rapporteur) \\
  \vspace{0.1cm}
  \large {\bf Bruno Iochum}, Professeur à l'Université de Provence, Marseille (Directeur de thèse) \\
  \vspace{0.1cm}
  \large {\bf Sylvie Paycha}, Professeur à l'Université Blaise Pascal, Clermont-Ferrand (Présidente, Rapporteur) \\
  \vspace{0.1cm}
  \large {\bf Michael Puschnigg}, Professeur à l'Université de la Méditerranée, Marseille   \\
  
  \large {\bf Thomas Sch\"ucker}, Professeur à l'Université de Provence, Marseille   
\end{tabular}

\vspace{1.2cm}

\begin{center}
{\large \bf Centre de Physique Th\'eorique\\
CNRS--UMR  6207}\\
\end{center}
\chapter*{Remerciements}

En premier lieu, je tiens à remercier Bruno Iochum, mon directeur de thèse, qui m'a conseillé, soutenu et beaucoup appris, le long de ces trois années de thèse. Ce fut un honneur et un plaisir de travailler sous sa direction. 
Je dois aussi beaucoup à Andrzej Sitarz et Driss Essouabri, avec qui j'ai eu le plaisir de collaborer au cours de cette thèse.
Tous trois ont su me faire confiance et je tiens à leur exprimer
ma plus sincère gratitude.

\medskip

Je voudrais aussi remercier les nombreuses personnes, du CPT et de l'IML, qui m'ont aidé à progresser
par des discussions stimulantes et intéressantes. Je remercie plus particulièrement Thomas Krajewski, Thomas Schücker, Michael Puschnigg, Christoph Stephan, Jan Jureit, Serge Lazzarini, Mathieu Beau, Baptiste Savoie, Jean-Philippe Michel, Pierre Duclos et Antony Wassermann.

\medskip

Un grand merci aussi aux chercheurs que j'ai pu rencontrer au cours de séminaires et conférences et qui ont été une source de connaissances, de conseils et de suggestions très utiles. Je suis notamment très reconnaissant envers Sylvie Paycha, Jean-Marie Lescure, Gerd Grubb, Ryszard Nest, Moulay-Tahar Benameur, Salah Mehdi, Giovanni Landi, Ludwik D\c{a}browski, Pierre Martinetti, Victor Gayral, Fabien Vignes-Tourneret, Fyodor Sukochev, Simon Brain, André Unterberger, et Rémi Léandre.

\medskip

J'ai pu bénéficier au CPT d'excellentes conditions de travail. En particulier, j'adresse
mes remerciements à l'ensemble du personnel administratif.

\medskip

\newpage
\tableofcontents
\newpage
{\renewcommand{\thechapter}{}\renewcommand{\chaptername}{}
\addtocounter{chapter}{-1}
\chapter{Introduction}\markboth{\sl INTRODUCTION}{\sl INTRODUCTION}}

This thesis is a presentation of the research I have done during the last three years in collaboration with 
Driss Essouabri, Bruno Iochum and Andrzej Sitarz.

\bigskip

Noncommutative geometry is a vast field of mathematics which aims at extending classical geometry in a noncommutative setting.  
Specifically, using functional analysis, operator algebras, spectral theory and
spin geometry, noncommutative geometry extends for instance
concepts of locally compact topological space, and spin Riemannian manifold. 

The interest for this mathematical field goes beyond the purely mathematical.
There are profound physical motivations behind the concepts of noncommutative geometry, which suggest that
it can be a used to describe the basic elements of physics 
such as spacetime and quantum fields. More specifically, noncommutative geometry 
appears, as a mathematical framework, well-suited to a geometric formulation of 
quantum concepts.

One may consider that noncommutative geometry was born with the following theorem of Gelfand--Naimark: any commutative $C^*$-algebra is isomorphic to the $C^*$-algebra of continuous functions on a compact space, namely the space of characters of the algebra. 
Since all the topological information of a space is contained in the set of continuous functions on the space, we can see that the notion of $C^*$-algebra is a generalization of the notion of compact topological space. 

From this fundamental result, it has been possible to go beyond purely topological concepts and construct a full noncommutative Riemannian differential geometry, with its own version of differential and integral calculus, fiber bundles, measures, spin manifolds, etc... This colossal work has been pioneered and essentially done by Alain Connes \cite{ConnesTorus,NCDG,Book,Cgeom,ConnesReality}. 

\medskip

Noncommutative geometry makes us change our point of view: it is not the set of points (the topological space) that is fundamental, but the set of functions (defined on the set of points). While equivalent in the commutative world, these two points of view are unequivalent in the noncommutative world.

Noncommutativity and this change of point of view are two crucial characteristics of the fundamental structure of quantum physics. Indeed, in quantum physics, observables do not always commute, contrarily to the case of classical physics. Moreover, $f(x)$, namely the evaluation of an observable $f$ at a point $x$, is not defined. However, the notion of observable still exists and $x(f)$ has a meaning, provided that $x$ is a state, a noncommutative equivalent of the notion of character of an algebra that can evaluate the observables $f$.

\medskip

This striking similarity between noncommutative geometry and quantum theory is very important and constitutes a source of the development of noncommutative geometry and its applications in physics. In fact, any $C^*$-algebra (commutative or not) is isomorphic to a closed subalgebra of the algebra of bounded operators on a Hilbert space. This shows that noncommutative geometry is well suited to the formalism of quantum physics. 

\medskip

Noncommutative geometry provides new mathematical notions that can be used to extend the fundamental aspect of noncommutativity of observables to spacetime itself. The two theories of fundamental interactions, namely quantum field theory (standard model) for the strong and electroweak interactions and general relativity for the gravitational interaction, are not based on the same mathematical formalism. The first one is based on a quantum point of view, while the latter one is classical. Moreover, spacetime itself is not considered the same way by these two theories. In quantum field theory, spacetime is fixed (Minkowskian flat spacetime) while in general relativity, spacetime is dynamical and gravitational effects are consequences of the curvature of the spacetime metric. These fundamental differences are not problematic if we study a physical phenomenon for which gravitational interaction can be neglected or inversely for which quantum effects can be neglected. Indeed, both theories are very powerful in their respective domain of application.   

However, for the study of phenomena that manifestly involve all known interactions (compact objects, black holes, big-bang), it is necessary to make these theories compatible, and to reunite them under a unique coherent mathematical formalism. The main idea that has been considered by theoretical physicists is to develop a generalization of quantum field theory in order to incorporate gravitation, or in other words, to quantize gravity. 

The pursuit of this goal has been realized through various approaches. One of these approaches is string theory, which is based on an increase in the number of dimensions, certain dimensions being compactified, and another approach is loop quantum gravity, which uses a spin foam structure for spacetime without using a background spacetime metric, contrarily to string theory. None of these theories have been confirmed by experiments yet, and theoretical predictions are difficult to obtain. 

\medskip 

The approach suggested by noncommutative geometry is based on a noncommutative generalization of the Lorentzian manifold which represents spacetime. By introducing noncommutativity at the spacetime level, this approach allows to apprehend the impossibility of spacetime continuity, as suggested by quantum mechanics, and the intrinsic limit associated to Planck length $l_p=\sqrt{\frac{{\cal
G}\hbar}{c^3}} \approx 10^{-33}$ cm. Actually, this approach successfully unified, at the classical level, the three interactions of the standard model with gravitation, and provided a geometric interpretation of the Higgs mechanism in particle physics. The fundamental object at the interface between noncommutative geometry and quantum field theory is the notion of spectral triple, which is a noncommutative generalization of the notion of spin Riemannian manifold, the starting point for the construction of physical theories.
By considering an ``almost commutative'' spectral triple, that is to say the product of a commutative spectral triple (a compact Riemannian spin manifold, describing the ``continuous'' spacetime), by a zero dimensional spectral triple (based on a matricial algebra), it is possible to simultaneously recover the standard model and general relativity. 
The fundamental ingredient in this unifying theory is the Chamseddine--Connes spectral action $S=\Tr \Phi(\DD/\Lambda)$, which is an action functional defined on the preceding almost commutative spectral triple and which is defined through the spectrum of the Dirac operator $\DD$ associated to this spectral triple, where $\Phi$ is a cut-off function and $\Lambda$ is a mass scale parameter. Specifically, the spectral action unifies the electroweak, strong and gravitational interactions \cite{Carminati,CC,CC1,CC2,CCM,Connesaction,CGravity,ConnesMarcolli,Schucker,Jureit1,Jureit2,Stephan}, and corresponds to the number of eigenvalues of the Dirac operator that are less or equal to a given mass scale $\Lambda$. By proceeding to a fluctuation of the metric, that is to say, a generalized gauge transformation associated to the unitary elements of the spectral triple algebra, it is possible, by expanding the spectral action in powers of $\Lambda$, to recover the standard model Lagrangian, coupled to the Einstein--Hilbert gravitational Lagrangian. 

\medskip

Another very important question associated to applications in physics of noncommutative geometry concepts concerns the renormalization in quantum field theory. Renormalization is a fundamental mathematical process that extracts physical information (finite numerical quantities) from mathematical expressions corresponding to infinite quantities. It has been shown \cite{Kreimer,ConnesKreimer} that the combinatorics of the perturbative renormalization in quantum field theory are related to the theory of Hopf algebras, which are the noncommutative Lie groups of symmetry in the language of noncommutative geometry. Moreover, Dixmier trace and Wodzicki residue, which are related to renormalization, have been extended by Connes \cite{Book,Cgeom} in the noncommutative setting of spectral triples. Thus, we expect that conceptual aspects of renormalization can be incorporated in a noncommutative formalism of quantum field theory.

\medskip

In this thesis, we discussed certain mathematical issues related to the computation of the spectral action on some fundamental noncommutative spectral triples, such as the noncommutative torus and the quantum 3-sphere $\SUq$. We also studied the question of existence of tadpoles (linear terms in the potential $A$ of the fluctuation of the metric in the spectral action) in the case of commutative Riemannian geometries, and the construction of a symbolic global pseudodifferential calculus allowing a generalization of the Weyl--Moyal product on a Schwartz space of rapidly decaying sections on a cotangent bundle of a manifold with linearization. 

\medskip

This text is divided into 5 chapters.  

The first and second chapters present the work done in collaboration with Driss Essouabri, Bruno Iochum and Andrzej Sitarz in the paper \emph{Spectral action on noncommutative torus} \cite{MCC}.

The first chapter is a review of definitions and properties about spectral triples, pseudodifferential calculus and spectral action. We establish in section 1.3 some results about residues of zeta functions $\zeta_{D_A}$ that will be used subsequently in chapter 2, 3. We also see some properties about linear terms in spectral actions (tadpoles) which will be used in chapter 4.

We compute, in the second chapter, the spectral action on noncommutative torus through heat kernel expansion and residues of Hurwitz--Epstein zeta functions.
The computation is based on expansions of terms
$\ncint \vert D_{A}\vert^{-k}$, $\zeta_{D_A}(0)$, as noncommutative integrals of certain
pseudodifferential operators. We applied these results to the $n$-noncommutative torus $\big(\Coo(\T^n_\Th),\H,\DD\big)$, which is a simple spectral triple.
The full spectral action obtained on the noncommutative torus is, in dimension 4,
$$
\SS(\DD_{A},\Phi,\Lambda)= 8\pi^2\,\Phi_{4} \, \Lambda^{4}
-\tfrac{4\pi^{2}}{3}\
\Phi(0) \,\tau(F_{\mu\nu}F^{\mu\nu})+  \mathcal{O}(\Lambda^{-2}),
$$
which shows that a new noncommutative Yang--Mills term $\tau(F_{\mu\nu}F^{\mu\nu})$, where $F_{\mu\nu}:=\delta_\mu(A_\nu)-\delta_\nu(A_\mu)-[A_\mu,A_\nu]$, replaces the standard one. 
For the computation, we had to investigate
holomorphic continuation of certain series of zeta functions. In particular, we saw that
a Diophantine constraint on the $\Th$ matrix of deformation naturally appears as crucial in the computation of
the spectral action, when the reality operator $J$ is taken into account. In other words, a Diophantine condition on $\Th$ appears when the perturbation $D\to D+A +\eps J A J^{-1}$ is considered, but is not needed for a perturbation of the type $D\to D+A$, where $A$ is a selfadjoint one form.

\medskip

The third chapter presents the work done in collaboration with Bruno Iochum and Andrzej Sitarz in the paper \emph{Spectral action on $\SUq$} \cite{MC}.
The spectral triple based on the $SU_q(2)$ quantum group can be seen as $q$-deformed noncommutative 3-sphere \cite{DLSSV}. 
We compute the full spectral action on $SU_q(2)$ by taking into account the real structure $J$.
The dimension spectrum of $SU_q(2)$ being a finite set, there is only a finite number
of terms in the spectral action expansion. Moreover, it appears that tadpoles exist on $SU_q(2)$. We show that the action depends on
$q$ and the limit $q\to 1$ does not exist automatically. When 
it exists, such limit does not lead to the associated action on the
commutative sphere $\mathbb{S}^3$. On $SU_q(2)$, we show that the sign $F$ of the Dirac operator
has very special properties: first, it commutes modulo $OP^{-\infty}$ 
with the algebra, and second, it can be seen as a selfadjoint
one-form, giving terms which are independent of $q$ in the spectral action.
In order to obtain the spectral action, we used the previous pseudodifferential calculus defined in chapter 1 and a Poincaré--Birkhoff--Witt decomposition of 1-forms.
It appears that a very special multiplicative behavior of noncommutative integrals with the real structure $J$ holds
on $SU_q(2)$: indeed, we show in section 2.4.6 that 
$$
\ncint AJBJ^{-1} |\DD|^{-3}=
\tfrac 12  \ncint  A  \vert \DD \vert ^{-3} \overline{ \ncint B
\vert \DD \vert ^{-3}}
$$
where $A$ and $B$ are $\delta$-1-forms (linear combinations of terms of the form $a[|\DD|,b]$ with $a,b\in \A$).
These results prove that the spectral action on $SU_q(2)$ is completely determined by the terms
$$
\ncint A^{q} |\DD|^{-p}, \quad 1\leq q\leq p\leq 3\,,
$$
where $A$ is a $\delta$-1-form.
We then tackle the precise computation of these noncommutative integrals. In order to realize these computations, we establish a differential calculus up to some ideal $\CR$ in pseudodifferential operators.

\medskip 
The fourth chapter presents the work done in collaboration with Bruno Iochum in the paper \emph{Tadpoles and commutative spectral triples} \cite{Tadpole}. In quantum field theory, a tadpole is a one-loop Feynman diagram with one external line, giving a contribution to the vacuum expectation value of the field. In the setting of noncommutative geometry, these diagrams are represented by linear terms in the Chamseddine--Connes spectral action.
We show in this chapter that there are no tadpoles of any order in the spectral action of compact spin manifolds without boundary, and no tadpoles of order less that 5 on manifolds with boundary with chiral boundary condition. 
Using pseudodifferential techniques and Wodzicki residue, we track zero terms in spectral actions of compact spin manifolds. We expect that most of these results can be extended to manifolds with boundary under very general boundary conditions.

\medskip
The fifth chapter presents a work \cite{PDOML} about global pseudodifferential calculus on manifold with linearization. In \cite{Gayral2}, Gayral et al. have established a remarkable link between deformation quantization and Connes' noncommutative geometry. It has been proven that Moyal planes are (noncompact) spectral triples. The Moyal product is a star-product defined on the 
Schwartz space $\SS(\R^{2n})$ and yields a Fr\'{e}chet pre-$C^*$-algebra structure on that space. The noncompact spectral triple described in \cite{Gayral2} was built on this algebra. 

We propose in this chapter the construction of a global pseudodifferential calculus that allows to extend the construction of the Moyal product to more general spaces, the main goal being the construction of new noncommutative spectral triples based on deformation quantization star-products. Specifically, the Moyal product that we obtain is defined on the Schwartz space of rapidly decaying functions on the cotangent bundle of a manifold. It corresponds to the transfer in the symbol space of the kernel convolution product through a quantization isomorphism.
We consider the case of manifolds $M$ with linearization in the sense of Bokobza-Haggiag \cite{Bokobza}, such that the associated (abstract) exponential map provides global diffeomorphisms of $M$ with $\R^n$ at any point. Cartan--Hadamard manifolds are special cases of such manifolds. The abstract exponential map encodes a notion of infinity on the manifold that allows, modulo some hypothesis of $S_\sigma$-bounded geometry, to define the Schwartz space of rapidly decaying functions, globally defined Fourier transformation and classes of symbols with uniform and decaying control over the $x$ variable. Given a linearization on the manifold with some properties of control at infinity, we construct symbol maps and $\la$-quantization, explicit Moyal star-product on the cotangent bundle, and classes of pseudodifferential operators. We show that these classes are stable under composition, and that the $\la$-quantization map gives an algebra isomorphism (which depends on the linearization) between symbols and pseudodifferential operators. We study $L^2$-continuity and give some examples. We show in particular that the hyperbolic 2-space $\HH$ has a $S_1$-bounded geometry, allowing the construction of a global symbol calculus of pseudodifferential operators, and an intrinsic Moyal product on $\S(\HH)$. 

\medskip
A summary in French of this thesis is given in appendix.

\chapter{Spectral action on spectral triples}

\section{Introduction}

The spectral action introduced by
Chamseddine--Connes \cite{CC} plays an important role in
noncommutative geometry. More
precisely, given a spectral triple $(\A,\H,\DD)$ where $\A$ is an
algebra acting on the Hilbert space $\H$ and $\DD$ is a Dirac-like
operator (see \cite{Book,Polaris}), they proposed a physical action
depending only on the spectrum of the covariant Dirac operator
\begin{equation}
\label{covDirac} \DD_{A}:=\DD + A + \epsilon \,JAJ^{-1}
\end{equation}
where $A$ is a one-form represented on $\H$,
so has the decomposition
\begin{equation}
\label{oneform}
A=\sum_{i}a_{i}[\DD,\,b_{i}],
\end{equation}
with $a_{i}$, $b_{i}\in \A$, $J$ is a real
structure on the triple corresponding to charge conjugation and
$\epsilon \in \set{1,-1}$ depending on the dimension of this triple
and comes from the commutation relation
\begin{equation}
    \label{Jcom}
J\DD=\epsilon \, \DD J.
\end{equation}

In this chapter, we revisit the notions of pseudodifferential
operators on spectral triples, zeta functions, noncommutative integral, dimension
spectrum and spectral action. The reality operator $J$ is
incorporated and we pay a particular attention to kernels of
operators which can play a role in the constant (scale invariant) term of
the spectral action.

\section{Noncommutative integration on a simple spectral triple}

\subsection{Pseudodifferential operators on spectral triples}

Noncommutative geometry, with its notion of spectral triple, provides a minimal data which allows to start doing quantum field theory. 

\begin{defn} A triplet $(A,\H,\DD)$ is called a spectral triple if $\A$ is a unital $*$-algebra faithfully represented as bounded operators on the Hilbert space $\H$, and $\DD$ is a selfadjoint operator with compact resolvent such that all commutators $[\DD,a]$ are bounded for $a\in \A$. 

A spectral triple is finitely summable, with dimension $n$, if the resolvent set of $\DD$ has characteristic values $\la_j = \O(j^{-1/n})$ when $j\to \infty$.

A spectral triple is said regular if $\A$ and $[\DD,\A]$ are in $\cap_{k\in \N}\Dom \delta^k$ where $\delta(T):=[|\DD|,T]$.
\end{defn}
If $n$ is even, we shall suppose that there is a $\Z/2$-grading $\chi$ (the chirality operator) which commutes with any element of $\A$ and anticommutes with $\DD$. 
In this formalism, $\DD$ plays the role of
the Dirac operator in classical Riemannian spin geometry. In other words, $\DD^{-1}$
will represent the (Euclidean) fermion propagator. In order to have a
charge conjugation in our theory, one adds a real structure,
which brings an antiunitary operator $J$ which commutes or anticommutes with $\DD$.
In this setting, the gauge bosons will
be seen as inner fluctuation of the Dirac operator $\DD\to \DD_A:=
\DD+A+\eps JAJ^{-1}$, where $A$ is a selfadjoint
1-form, which is an operator of the form $\sum a_i\, [\DD,b_i]$, where $a_i$ and $b_i$ are in $\A$.

The notion of real structure on a finite summable spectral triple is related to real $K$-homology, and is defined by:

\begin{defn} A real structure $J$ on a finitely summable spectral triple $(\A,\H,\DD)$ of dimension $n$ is an antilinear isometry $J$ such that $JD=\eps DJ$, $J^2=\eps'$, $J\chi =\eps'' \chi J$ (even case), where $\eps,\eps',\eps''\in \set{-1,1}$ and satisfy the following table
 \begin{center}
   \begin{tabular}{|c|cccccccc|}
   \hline
   $n$ mod 8&0&1&2&3&4&5&6&7\\
   \hline
   $J^2 = \pm 1$&+&+&\-$-$&$-$&$-$&$-$&+&+\\
   $JD= \pm DJ$&+&$-$&+&+&+&$-$&+&+\\
   $J\chi = \pm \chi J$&+& &$-$& &+& &$-$&\\
   \hline
   \end{tabular}
 \end{center}  
Moreover, any element of $J\A J^{-1}$ commutes with any element of $\A\cup [\DD,\A]$.

A spectral triple $(\A,\H,\DD)$ endowed with a real structure $J$ is called a real spectral triple.
\end{defn}

Remark that the notion of spectral triple can actually be extended to the framework of von Neumann algebras \cite{Benameur}.

For a given 1-form $A$ on a spectral triple $(\A,\H,\DD)$, we will have to compare here the kernels of $\DD$ and $\DD_{A}$. Note first that they are finite dimensional:

\begin{lemma}
\label{compres} Let $(\A,\H,\DD)$ be a spectral triple with a reality
operator $J$ and chirality $\chi$. If $A \in \Omega^1_{\DD}$ is a
one-form, the fluctuated Dirac operator $$\DD_A:= \DD+A+ \epsilon
JAJ^{-1}$$
is an operator
with compact resolvent, and in particular its kernel $\Ker \DD_A$ is
a finite dimensional space. This space is invariant by $J$ and $\chi$.
\end{lemma}

\begin{proof}
Let $T$ be a bounded operator and let $z$ be in the resolvent set of
$\DD+T$ and $z'$ be in the resolvent set of $\DD$. Then
$$
(\DD+T-z)^{-1}=(\DD-z')^{-1} \, [1-(T+z'-z)(\DD+T-z)^{-1}].
$$
Since $(\DD-z')^{-1}$ is compact by hypothesis and since the term in
bracket is
bounded, $\DD+T$ has compact resolvent. Applying this to
$T=A+\epsilon JAJ^{-1}$,
$\DD_A$ has a finite dimensional kernel (see for instance
\cite[Theorem 6.29]{Kato}).

Since according to the dimension, $J^2=\pm 1$, $J$ commutes or
anticommutes with
$\chi$,  $\chi$ commutes with the elements in the algebra $\A$ and
$\DD \chi=-\chi
\DD$ (see \cite{ConnesReality} or \cite[p. 405]{Polaris}), we get
$\DD_A \chi=-\chi
\DD_A$ and $\DD_A J=\pm J \DD_A$ which gives the result.
\end{proof}

We now set $(\A,\DD, \H,J)$ a given real regular spectral
triple of dimension $n$ and $A$ a selfadjoint one-form.
We denote
\begin{align*}
&P_0 \text{ the projection on } \Ker \DD\, ,
P_A \text{  the projection on } \Ker \DD_A\, ,\\
&D:= \DD + P_0\, , D_A:= \DD_A + P_A\, .
\end{align*}
$P_0$ and $P_A$ are thus finite-rank selfadjoint bounded operators.
Remark that
$D$ and $D_A$ are selfadjoint invertible operators with compact
inverses.

\begin{remark}
Since we only need to compute the residues and the value at 0 of the
$\zeta_{D}$,
$\zeta_{D_A}$ functions, it is not necessary to define the operators
$\DD^{-1}$ or
$\DD_A^{-1}$ and the associated zeta functions. However, we can
remark that all the
work presented here could be done using the process of Higson in
\cite{Higson}
which proves that we can add any smoothing operator to $\DD$ or
$\DD_A$ such that
the result is invertible without changing anything to the
computation of residues (see also \cite{ConnesMarcolli}, where this question is considered).
\end{remark}

Define for any $\a\in \R$
\begin{align*}
OP^0&:=\{T \, : \, t\mapsto F_t(T) \in
C^\infty\big(\R,\B(\H)\big)\},\\
OP^\a&:= \set{T \,: \, T |D|^{-\a} \in  OP^0}.
\end{align*}
where $F_t(T):=e^{it|D|}\,T\,e^{-it|D|} =
e^{it|\DD|}\,T\,e^{-it|\DD|}$ since $\vert D\vert=\vert \DD \vert
+P_0$.  Define
\begin{align*}
\delta(T)&:=[|D|,T],\\
\nabla(T)&:=[\DD^2,T],\\
\sigma_s(T)&:=|D|^{s} T |D|^{-s}, \, s\in \C.
\end{align*}
It has been shown in \cite{CM} that
$OP^0=\bigcap_{p\geq 0}\Dom(\delta^p)$.
In particular, $OP^0$ is a subalgebra of
$\B(\H)$ (while elements of $OP^{\a}$
are not necessarily bounded for $\a>0$) and
$\A\subseteq OP^0$, $J\A J^{-1}\subseteq
OP^0$, $[\DD,\A] \subseteq OP^0$. Note that $P_0 \in OP^{-\infty}$
and $\delta(OP^0)\subseteq OP^0$.

For any $t>0$, $\DD^{t}$ and $|\DD|^t$ are in $OP^t$ and
for any $\a\in \R, D^\a$ and $|D|^{\a}$ are in $OP^\a$.
By hypothesis, $|D|^{-n} \in \L^{(1,\infty)}(\H)$ so
for any $\a>n$, $OP^{-\a}\subseteq \L^1(\H)$.

\begin{lemma} \cite{CM}
    \label{propOP}

(i) For any $T\in OP^0$ and $s\in \C$, $\sigma_s(T)\in OP^0$.

(ii) For any $\a, \beta \in \R$, $OP^{\a}OP^{\beta}\subseteq
OP^{\a+\beta}$.

(iii) If $\a\leq\beta$, $OP^{\a}\subseteq OP^{\beta}$.

(iv) For any $\a$, $\delta(OP^{\a})\subseteq OP^{\a}$.

(v) For any $\a$ and $T\in OP^{\a}$, $\nabla(T) \in OP^{\a+1}$.
\end{lemma}
\begin{proof} 
\noindent $(i)$ We have $|D|T|D|^{-1} = T + \delta(T) |D|^{-1}$
and $|D|^{-1}T|D| =  T  - |D|^{-1}\delta(T)$.
A recurrence proves that for any $k\in \N$,
$|D|^{k}T|D|^{-k}= \sum_{q=0}^k \genfrac{(}{)}{0pt}{1}{q}{k} \,
\delta^q(T)
|D|^{-q}$ and we get $|D|^{-k}T|D|^{k}= \sum_{q=0}^k
(-1)^q\genfrac{(}{)}{0pt}{1}{q}{k} \, |D|^{-q} \delta^q(T)$.

As a consequence, since $T$, $|D|^{-q}$
and $\delta^q(T)$ are in $OP^0$
for any $q\in \N$, for any $k\in \Z$,
$|D|^{k}T|D|^{-k}\in OP^0$. Let us fix $p\in
\N_0$ and define $F_p(s):=\delta^p(|D|^{s}T|D|^{-s})$
for $s\in \C$. Since for $k\in \Z$, $F_p(k)$ is bounded,
a complex interpolation proves that $F_p(s)$ is bounded, which
gives $|D|^{s} T |D|^{-s} \in OP^0$.

\noindent $(ii)$ Let $T\in OP^\a$ and $T'\in OP^\beta$. Thus,
$T|D|^{-\a}$, $T'|D|^{-\beta}$ are in $OP^0$. By $(i)$ we get
$|D|^{\beta}T|D|^{-\a}|D|^{-\beta} \in OP^0$, so
$T'|D|^{-\beta}|D|^{\beta}T|D|^{-\beta-\a}\in OP^0$. Thus,
$T'T|D|^{-(\a+\beta)}\in OP^{0}$.

\noindent $(iii)$ For $T\in OP^\a$, $|D|^{\a-\beta}$ and
$T|D|^{-\a}$ are in $OP^0$, thus
$T|D|^{-\beta}=T|D|^{-\a}|D|^{\a-\beta} \in OP^0$.

\noindent $(iv)$ follows from $\delta(OP^0)\subseteq OP^0$.

\noindent $(v)$ Since $ \nabla(T)=\delta(T)|D| + |D|\delta(T)-[P_0,T]$,
the result follows from $(ii)$, $(iv)$ and the fact that $P_0$ is in
$OP^{-\infty}$.
\end{proof}

\begin{remark}
\label{oprem}
Any operator in $OP^\a$, where $\a\in \R$,
extends as a continuous linear operator from $\Dom |D|^{\a+1}$ to
$\Dom |D|$ where the $\Dom |D|^\a$ spaces have their natural norms
(see \cite{CM,Higson}).
\end{remark}

We now introduce a definition of pseudodifferential operators
in a slightly different way
than in \cite{CM,Cgeom,Higson} which in particular pays attention
to the reality operator $J$ and the kernel of $\DD$
and allows $\DD$ and $|D|^{-1}$
to be pseudodifferential operators. It is more in the spirit of
\cite{CC1}.

\medskip

\begin{definition}
\label{defpseudo} Let us define $\DD(\A)$ as the polynomial algebra
generated
by $\A$, $J\A J^{-1}$, $\DD$ and $|\DD|$.

A pseudodifferential operator is an operator $T$ such that there
exists $d\in \Z$ such that
for any $N\in \N$, there exist $p\in \N_0$, $P\in \DD(\A)$ and $R\in
OP^{-N}$ ($p$, $P$ and
$R$ may depend on $N$) such that $P\,D^{-2p}\in OP^d$ and
$$
T=P\,D^{-2p}+R\, .
$$
Define $\Psi(\A)$ as the set of pseudodifferential operators and
$\Psi(\A)^k:=\Psi(\A)\cap OP^k$.
\end{definition}
Note that if $A$ is a 1-form, $A$ and $JAJ^{-1}$ are in $\DD(\A)$ and
moreover
$\DD(\A)\subseteq \cup_{p \in \N_0} OP^p$. Since $|\DD|\in \DD(\A)$
by construction and
$P_0$ is a pseudodifferential operator, for any
$p\in \Z$, $|D|^{p}$ is a pseudodifferential operator (in $OP^{p}$).
Let us
remark also that $\DD(\A)\subseteq\Psi(\A) \subseteq \cup_{k\in \Z}
OP^{k}$.

\begin{lemma} \cite{CM,Cgeom}
\label{pdoalg}
The set of all pseudodifferential operators $\Psi(\A)$ is an algebra.
Moreover, if
$T\in \Psi(\A)^d$ and $T\in \Psi(\A)^{d'}$, then $TT'\in
\Psi(\A)^{d+d'}$.
\end{lemma}

\begin{proof} 
The non-trivial part of the proof is the stability under the product
of operators. Let $T,
T' \in \Psi(\A)$. There exist $d,d' \in \Z$ such that for any
$N\in\N$, $N> |d|+|d'|$,
there exist $P,P'$ in $\DD(\A)$, $p,p'\in \N_0$, $R\in OP^{-N-d'}$,
$R'\in OP^{-N-d}$ such
that $T=PD^{-2p}+R$, $T'=P'D^{-2p'}+R'$, $PD^{-2p}\in OP^d$ and
$P'D^{-2p'}\in
OP^{d'}$.

Thus, $TT'=PD^{-2p}P'D^{-2p'} + RP'D^{-2p'}+PD^{-2p}R' +RR'$.

We also have $RP'D^{-2p'}\in OP^{-N-d'+d'} = OP^{-N}$ and similarly,
$PD^{-2p}R'
\in OP^{-N}$. Since $RR'\in OP^{-2N}$, we get
$$
TT'\sim PD^{-2p}P'D^{-2p'} \mod OP^{-N}.
$$
If $p=0$, then $TT'\sim QD^{-2p'} \mod OP^{-N}$ where $Q=PP'\in
\DD(\A)$ and
$QD^{-2p'}\in OP^{d+d'}$. Suppose $p\neq 0$. A recurrence proves that
for any $q\in
\N_0$,
$$
D^{-2}P' \sim \sum_{k=0}^q (-1)^k \nabla^k(P') D^{-2k-2}
+(-1)^{q+1}D^{-2}
\nabla^{q+1}(P') D^{-2q-2} \mod OP^{-\infty}\, .
$$
By Lemma \ref{propOP} $(v)$, the remainder is in $OP^{d'+2p'-q-3}$,
since $P'\in
OP^{d'+2p'}$. Another recurrence gives for any $q\in \N_0$,
$$
D^{-2p}P'\sim \sum_{k_1,\cdots,k_p=0}^q (-1)^{|k|_1}
\nabla^{|k|_1}(P')
D^{-2|k|_1-2p} \, \mod  OP^{d'+2p' -q -1 - 2p}.
$$
Thus, with $q_N=N+d+d'-1$,
$$
TT'\sim \sum_{k_1,\cdots,k_p=0}^{q_N} (-1)^{|k|_1} P
\nabla^{|k|_1}(P')
D^{-2|k|_1-2(p+p')} \, \mod OP^{-N}.
$$
The last sum can be written $Q_N D^{-2r_N}$ where $r_N:=p\,
q_N+(p+p')$. Since $Q_N\in
\DD(\A)$ and $Q_N D^{-2r_N}\in OP^{d+d'}$, the result follows.
\end{proof}

It is convenient to also introduce

\begin{definition}
\label{defpseudo1} Let $\DD_{1}(\A)$ be the algebra generated by
$\A$, $J\A J^{-1}$ and $\DD$, and $\Psi_{1}(\A)$ be the set of
pseudodifferential operators
constructed as before with $\DD_{1}(\A)$ instead of $\DD(\A)$. Note
that
$\Psi_{1}(\A)$ is subalgebra of $\Psi(\A)$.
\end{definition}

Remark that $\Psi_1(\A)$ does not
necessarily contain operators such as $|D|^k$ where $k\in \Z$ is odd.
This algebra
is similar to the one defined in \cite{CC1}.

\subsection{Zeta functions, noncommutative integral and spectral action}

For any operator $B$ and if $X$ is either $D$ or
$D_A$, we define
\begin{align*}
{\zeta}_X^B(s)&:= \Tr\big(B|X|^{-s}\big ),\\
\zeta_X(s)&:= \Tr \big(|X|^{-s}\big).
\end{align*}

\medskip

{\it The dimension spectrum} $Sd(\A,\H,\DD)$ of a
spectral triple has been defined in
\cite{Cgeom, CM}. It is extended here to pay attention
to the operator $J$ and to our
definition of pseudodifferential operator.

\begin{definition}
The spectrum dimension of the spectral triple is
the subset $Sd(\A,\H,\DD)$ of all poles of the functions
$\zeta_D^P\ :\ s\mapsto \Tr \big(P |D|^{-s}\big)$
where $P$ is any pseudodifferential operator in $OP^0$.
The spectral triple $(\A,\H,\DD)$ is simple when these
poles are all simple.
\end{definition}

\begin{remark}\label{remark-spectrum} If $Sp(\A,\H,\DD)$ denotes
the set of all poles of
the functions $s\mapsto \Tr \big(P |D|^{-s}\big)$ where $P$
is any pseudodifferential operator, then,
$Sd(\A,\H,\DD) \subseteq Sp(\A,\H,\DD)$.

When $Sp(\A,\H,\DD)=\Z$, $Sd(\A,\H,\DD) = \set{n-k \ : \ k\in \N_0}$:
indeed, if  $P$ is a pseudodifferential operator
in $OP^0$, and $q\in \N$ is such that $q>n$, $P|D|^{-s}$
is in $OP^{-\Re(s)}$ so
is trace-class for $s$ in a neighborhood of $q$; as a
consequence, $q$ cannot be a pole of $s\mapsto \Tr
\big(P|D|^{-s}\big)$.
\end{remark}

\begin{remark} $Sp(\A,\H,\DD)$ is also the set of
poles of functions $s\mapsto \Tr \big(B |D|^{-s-2p}\big)$
where $p\in \N_0$ and $B\in \DD(\A)$.
\end{remark}

Introducing the notation (recall that $\nabla (T)=[\DD^2,T]$) for an
operator $T$,
$$
\eps(T):=\nabla(T)D^{-2},
$$
we get from \cite[(2.44)]{CC1} the following expansion for
$T\in OP^q$
\begin{equation}\label{one-par}
\sigma_{z}(T)\sim \sum_{r=0}^{N} g(z,r) \,\eps^r(T)  \mod OP^{-N-1+q}
\end{equation}
where $g(z,r):=
\tfrac{1}{r!}(\tfrac{z}{2})\cdots(\tfrac{z}{2}-(r-1))=\genfrac(){0pt}{1}{z/2}{r}$
with the
convention $g(z,0):=1$.

\noindent We define the noncommutative integral by
$$
\ncint T:=\underset{s=0}{\Res}\ \zeta_D^{T}(s)
=\underset{s=0}{\Res}\ \Tr\,\big(T|D|^{-s}\big).
$$
The noncommutative integral plays the role of the Wodzicki residue in a spectral triple setting.
\begin{prop}
\label{tracenc}
\cite{CM}
If the spectral triple is simple, $\ncint$ is a trace on $\Psi(\A)$.
\end{prop}
\begin{proof}
Let $P\in OP^{k_1}, \,Q\in OP^{k_2} \in \Psi(\A)$. With
$[Q,|D|^{-s}]= \big(Q-\sigma_{-s}(Q)\big) \,|D|^{-s}$ and the
equivalence
$Q-\sigma_{-s}(Q) \sim  -\sum_{r=1}^N g(-s,r) \, \eps^r(Q) \mod
OP^{-N-1+k_2}$,
we get
$$
P[Q,|D|^{-s}] \sim -\sum_{r=1}^N  g(-s,r)\,P \eps^r(Q) |D|^{-s} \mod
OP^{-N-1+k_1+k_2-\Re(s)}
$$
which gives, if we choose $N=n+k_1+k_2$,
$$
\underset{s=0}{\Res}\ \Tr\big( P
[Q,|D|^{-s}] \big)= -\sum_{r=1}^{n+k_1+k_2}\underset{s=0}{\Res}\
g(-s,r) \Tr \big( P \eps^r(Q)|D|^{-s} \big).
$$
By hypothesis $s\mapsto \Tr \big( P \eps^r(Q)|D|^{-s} \big)$ has only
simple
poles. Thus, since $s=0$ is a zero of the analytic function $s\mapsto
g(-s,r)$ for any
$r\geq 1$, we have $\underset{s=0}{\Res}\ g(-s,r) \,\Tr \big( P
 \eps^r(Q)|D|^{-s}\big) =0$, which entails that
$\underset{s=0}{\Res}\ \Tr\big(
P[Q,|D|^{-s}]\big)=0$ and thus
\begin{align*}
\ncint PQ= \underset{s=0}{\Res}\ \Tr\big( P|D|^{-s} Q\big)\, .
\end{align*}
When $s\in \C$ with $\Re(s)> 2 \max(k_1+n+1,k_2)$, the operator $P
|D|^{-s/2}$ is trace-class while $|D|^{-s/2} Q$ is bounded, so
$\Tr\big( P|D|^{-s} Q\big) = \Tr \big(|D|^{-s/2} QP|D|^{-s/2}\big) =
\Tr\big(\sigma_{-s/2}(Q P)|D|^{-s}\big)$.
Thus, using (\ref{one-par}) again,
$$
\underset{s=0}{\Res}\ \Tr\big( P |D|^{-s} Q\big)  = \ncint Q P
 +\sum_{r=1}^{n+k_1+k_2} \underset{s=0}{\Res}\ g(-s/2,r)
\Tr\big(\eps^r(Q P)
|D|^{-s}\big).
$$
As before, for any $r\geq 1$, $\underset{s=0}{\Res}\ g(-s/2,r)
\Tr\big(\eps^r(Q P) |D|^{-s}\big)=0$ since $g(0,r) =0$ and the
spectral triple
is simple. Finally,
\begin{align*}
\underset{s=0}{\Res}\ \Tr\big( P |D|^{-s} Q\big) = \ncint Q P.
\end{align*}
which yields the result.
\end{proof}

On a spectral triple $(\A,\H,\DD)$, the role of the action is played by
the ``spectral action'' as introduced by
A. Chamseddine and A. Connes:
\begin{equation}
\SS(\DD_{A},\Phi,\Lambda):=\Tr \big( \Phi( \DD_{A} /\Lambda) \big)\label{action}
\end{equation}
where $\Phi$ is any even positive cut-off function which could
be replaced by a step function up to some mathematical difficulties
investigated in \cite{Odysseus} and $\Lambda$ fixes the mass scale.
This means that $\SS$ counts the
spectral values of $\vert \DD_{A} \vert$ less than the mass scale
$\Lambda$ (note that the resolvent set of $\DD_{A}$ is compact since, by
assumption, the same is true for $\DD$, see Lemma \ref{compres}
below).
The Chamseddine--Connes spectral action principle asserts that (see \cite[p. 197]{ConnesMarcolli}) the spectral action \emph{is the fundamental action functional S that can be used both at the classical level to compare different geometric spaces and at the quantum level in the functional integral formulation, after Wick rotation to Euclidean signature}. In other words, the functional $\SS(\DD_A,\Phi,\Lambda)$ which is related to the spectrum of the Dirac operator $\DD$, contains all physical information of the (geometrized) quantum field theory associated to the triple $(\A,\H,\DD)$. It is therefore crucial to be able to compute it on some fundamental examples. This spectral action is known on few examples:
\cite{Carminati,CC1,CCM,ConnesMarcolli,MCC,GI2002,
GIV,GIVas,GW,Knecht}. We will study the case of the noncommutative torus in chapter 2, and the case of $SU_q(2)$ in chapter 3. We shall investigate in chapter 4 some questions about tadpoles, which are the linear terms in $A$ in the spectral action.
In the case of a spectral triple with simple dimension spectrum, we have (see for instance \cite[Theorem 1.145]{ConnesMarcolli})
\begin{align}
    \SS(\DD_{A},\Phi,\Lambda) \, = \,\sum_{0<k\in Sd^+} \Phi_{k}\,
    \Lambda^{k} \ncint \vert D_{A}\vert^{-k} + \Phi(0) \,
    \zeta_{D_{A}}(0) +\mathcal{O}(\Lambda^{-1}),\label{formuleaction}
\end{align}
where $\Phi_{k}= \half\int_{0}^{\infty} \Phi(t) \, t^{k/2-1} \, dt$ and
$Sd^+$ is the strictly positive part of the spectrum dimension of $(\A,\H,\DD)$. Thus,
the main problem is the computation of the $\ncint \vert D_{A}\vert^{-k}$, $\zeta_{D_{A}}(0)$ terms. We consider this question in the following section.

\section{Residues of $\zeta_{D_A}$ for a spectral triple with
simple dimension spectrum}

We fix a regular real spectral triple $(\A,\H,\DD,J)$ of
dimension $n$ and a self-adjoint 1-form $A$.

\noindent Recall that
\begin{align*}
\DD_A &:=\DD+\wt A \text{  where } \wt A:= A +\eps JAJ^{-1} ,\\
 D_A &:= \DD_A + P_A
\end{align*}
where $P_A$ is the projection on $\Ker \DD_A$. Remark that
$\wt A \in\DD(\A)\cap OP^0$ and $\DD_A\in \DD(\A)\cap OP^1$.

We denote
$$
V_A:= P_A - P_0.
$$
As the following lemma shows, $V_A$ is a smoothing operator:

\begin{lemma}
\label{finiterank}
$(i)$ $\bigcap_{k\geq 1} \Dom (\DD_A)^{k} \subseteq \bigcap_{k\geq 1}
\Dom |D|^k$.

$(ii)$ $\Ker \DD_A \subseteq \bigcap_{k\geq 1} \Dom |D|^k$.

$(iii)$ For any $\a, \beta \in \R$, $|D|^\beta P_A |D|^\a$ is bounded.

$(iv)$ $P_A \in OP^{-\infty}$.
\end{lemma}
\begin{proof}
$(i)$ Let us define for any $p\in \N$, $R_p := (\DD_A)^p -\DD^p$, so
$R_p \in OP^{p-1}$ and $R_p \big(\Dom |D|^p\big)\subseteq \Dom |D|$
(see Remark \ref{oprem}).

Let us fix $k\in \N$, $k\geq 2$. Since $\Dom \DD_A = \Dom \DD =\Dom
|D|$, we have
$$
\Dom (\DD_A)^k = \set{\phi \in \Dom |D| \ : \ (\DD^j + R_j)\,\phi
\in \Dom |D| \ , \ \forall j\ \  1\leq j\leq k-1  }.
$$
Let $\phi \in \Dom (\DD_A)^k$. We prove by recurrence that for any
$j\in\set{1,\cdots,k-1}$, $\phi \in \Dom |D|^{j+1}$:

We have $\phi\in \Dom |D|$ and $(\DD + R_1)\, \phi \in \Dom |D|$.
Thus, since $R_1\,\phi \in \Dom |D|$, $\DD \phi \in \Dom |D|$, which
proves that $\phi \in \Dom |D|^2$.
Hence, case $j=1$ is done.

Suppose now that $\phi \in \Dom |D|^{j+1}$ for a $j \in
\set{1,\cdots,k-2}$.
Since $(\DD^{j+1} + R_{j+1})\, \phi \in
\Dom |D|$, and $R_{j+1}\, \phi \in \Dom |D|$, we get $\DD^{j+1}\,
\phi \in \Dom |D|$,
which proves that $\phi\in \Dom |D|^{j+2}$.

Finally, if we set $j=k-1$, we get $\phi \in \Dom |D|^{k}$,
so $\Dom (\DD_A)^k \subseteq \Dom |D|^k$.

$(ii)$ follows from $\Ker \DD_A \subseteq \bigcap_{k\geq 1} \Dom
(\DD_A)^k$ and $(i)$.

$(iii)$ Let us first check that $|D|^\a P_A$ is bounded. We define
$D_0$
as the operator with domain $\Dom D_0 = \Ima P_A \cap \Dom |D|^\a$
and such that $D_0\, \phi = |D|^\a\, \phi.$ Since $\Dom D_0$ is
finite dimensional,
$D_0$ extends as a bounded operator on $\H$ with finite rank.
We have
$$
\sup_{\phi \in \Dom |D|^\a P_A,\ \norm{\phi}\leq 1} \norm{|D|^\a
P_A\, \phi} \leq
\sup_{\phi \in \Dom D_0,\ \norm{\phi}\leq 1}
\norm{|D|^\a\, \phi} = \norm{D_0}<\infty
$$
so $|D|^\a P_A$ is bounded. We can remark that by $(ii)$, $\Dom D_0 =
\Ima P_A$ and
$\Dom |D|^\a P_A = \H$.

Let us prove now that $P_A |D|^\a$ is bounded:
Let $\phi\in \Dom P_A |D|^\a = \Dom |D|^\a$. By $(ii)$, we have $\Ima
P_A \subseteq \Dom |D|^\a$
so we get
\begin{align*}
\norm{P_A |D|^\a\,\phi} & \leq
\sup_{\psi \in \Ima P_A,\ \norm{\psi}\leq 1} |<\psi,|D|^\a\, \phi>|
\leq \sup_{\psi \in
\Ima P_A,\ \norm{\psi}\leq 1} |<|D|^\a\psi,\phi>| \\
&\leq \sup_{\psi \in
\Ima P_A,\ \norm{\psi}\leq 1} \norm{|D|^\a\psi}\norm{\phi} =
\norm{D_0} \norm{\phi}.
\end{align*}

$(iv)$ For any $k\in \N_0$ and $t\in \R$, $\delta^k(P_A)|D|^t$ is a
linear combination of terms of
the form $|D|^\beta P_A |D|^\a$, so the result follows from $(iii)$.
\end{proof}

\begin{remark}
We will see later on the noncommutative torus example how
important is the difference
between $\DD_{A}$ and $\DD+A$. In particular, the inclusion
$\Ker \DD \subseteq \Ker \DD +
A$ is not satisfied since $A$ does not preserve $\Ker \DD$
contrarily to $\wt A$.
\end{remark}

The coefficient of the nonconstant term $\Lambda^k$ ($k>0$)
in the expansion (\ref{formuleaction}) of the spectral action
$S(\DD_A,\Phi,\Lambda)$ is equal to the
residue of $\zeta_{D_A}(s)$ at $k$. We will see in this section
how we can compute these
residues in term of noncommutative integral of certain operators.

Define for any operator $T$, $p\in \N$, $s\in \C$,
$$
K_p(T,s):=(-\tfrac{s}{2})^p\int_{0\leq t_1\leq\cdots\leq t_p\leq 1}
\sigma_{-st_1}(T)\cdots\sigma_{-st_p}(T)\, dt
$$
with $dt:=dt_1\cdots dt_p$.

Remark that if $T\in OP^\a$, then $\sigma_z(T)\in OP^\a$ for $z\in \C$
and $K_p(T,s) \in OP^{\a p}$.

Let us define
\begin{align*}
X &:= \DD_{A}^2-\DD^2 =\wt A \DD + \DD \wt A + \wt A^2 ,\\
X_V &:= X+V_A,
\end{align*}
thus $X\in \DD_1(\A)\cap OP^1$ and by Lemma \ref{finiterank},
\begin{equation}\label{xvsim}
X_V \sim X \mod OP^{-\infty}.
\end{equation}

We will use
$$
Y:=\log(D_A^2) -\log (D^2)
$$
which makes sense since $D_A^2 = \DD_A^2 + P_A$ is invertible for any
$A$.

By definition of $X_V$, we get
$$
Y= \log (D^2 + X_V) -\log (D^2).
$$
Remark that most of the results that we will present here are based on the fact that $D_A-D$ is a pseudodifferential operator of order 0.
\begin{lemma}
    \label{2dev}
\cite{CC1}

$(i)$ $Y$ is a pseudodifferential operator in $OP^{-1}$ with the
following
expansion for any $N\in\N$
$$
Y \sim \sum_{p=1}^N\sum_{k_1,\cdots,k_p
=0}^{N-p}\tfrac{(-1)^{|k|_1+p+1}}{|k|_1+p}
\nabla^{k_p}(X\nabla^{k_{p-1}}(\cdots
X\nabla^{k_1}(X)\cdots)) D^{-2(|k|_1+p)} \mod OP^{-N-1}.
$$

$(ii)$ For any $N\in\N$ and $s\in \C$,
\begin{align}
\label{expansion} |D_A|^{-s} \sim |D|^{-s} + \sum_{p=1}^N K_p(Y,s)
|D|^{-s} \mod
OP^{-N-1-\Re(s)}.
\end{align}
\end{lemma}
\begin{proof}
$(i)$ We follow \cite[Lemma 2.2]{CC1}. By functional calculus,
$Y=\int_0^\infty I(\la)\, d\la$, where
$$
I(\la)\sim\sum_{p=1}^N(-1)^{p+1}\big((D^2+\la)^{-1}X_V\big)^{p}
(D^2+\la)^{-1} \mod OP^{-N-3}.
$$
By (\ref{xvsim}), $\big((D^2+\la)^{-1}X_V\big)^{p} \sim
\big((D^2+\la)^{-1}X\big)^{p} \mod OP^{-\infty}$ and we get
$$
I(\la)\sim\sum_{p=1}^N(-1)^{p+1}\big((D^2+\la)^{-1}X\big)^{p}
(D^2+\la)^{-1} \mod OP^{-N-3}.
$$
We set $A_p(X):=\big((D^2+\la)^{-1}X\big)^{p}(D^2+\la)^{-1}$
and $L:=(D^2+\la)^{-1}\in OP^{-2}$ for a fixed $\la$.
Since $[D^2 + \la,X]\sim \nabla(X) \mod OP^{-\infty}$,
a recurrence proves that if $T$ is an operator
in $OP^{r}$, then, for $q\in \N_0$,
$$
A_1(T)=L T L \sim \sum_{k=0}^q (-1)^k\nabla^k(T) L^{k+2} \mod
OP^{r-q-5}.
$$
With $A_p(X)=LX A_{p-1}(X)$,
another recurrence gives, for any $q\in \N_0$,
$$
A_p(X)\sim \sum_{k_1,\cdots,k_p =0}^q (-1)^{|k|_1}\nabla^{k_p}
(X \nabla^{k_{p-1}}(\cdots X\nabla^{k_1}(X)\cdots)) L^{|k|_1+p+1}
\mod OP^{-q-p-3},
$$
which entails that
$$
I(\la)\sim\sum_{p=1}^N(-1)^{p+1}\sum_{k_1,\cdots,k_p
=0}^{N-p}(-1)^{|k|_1}\nabla^{k_p}(X\nabla^{k_{p-1}}
(\cdots X\nabla^{k_1}(X)\cdots))
L^{|k|_1+p+1} \mod OP^{-N-3}.
$$

With $\int_{0}^\infty (D^2+\la)^{-(|k|_1+p+1)}d\la =
\tfrac{1}{|k|_1+p}
D^{-2(|k|_1+p)}$, we get the result provided we control
the remainders. Such a control is given in \cite[(2.27)]{CC1}.

\noindent $(ii)$ We have $|D_A|^{-s}=e^{B-(s/2)Y}e^{-B}\, |D|^{-s}$
where $B:= (-s/2)\log(D^2)$. Following \cite[Theorem 2.4]{CC1},
we get
\begin{equation}
\label{egalite-DAs}
|D_A|^{-s} = |D|^{-s} + \sum_{p=1}^\infty K_p(Y,s)|D|^{-s}\, .
\end{equation}
and each $K_p(Y,s)$ is in $OP^{-p}$.
\end{proof}

\begin{corollary}
\label{eps-pdo} For any $p\in\N$ and $r_1,\cdots,r_p \in \N_0$,
$\eps^{r_1}(Y)\cdots \eps^{r_p}(Y) \in \Psi_{1}(\A)$.
\end{corollary}
\begin{proof}
If for any $q\in \N$ and $k=(k_1,\cdots,k_q)\in \N_0^q$,
$$
\Ga_q^k(X):=\tfrac{(-1)^{|k|_1+q+1}}{|k|_1+q} \nabla^{k_q}
(X\nabla^{k_{q-1}}(\cdots X\nabla^{k_1}(X)\cdots)),
$$
then, $\Ga_q^k(X) \in OP^{|k|_1+q}$. For any $N\in\N$,
\begin{equation}\label{Ydev}
Y \sim \sum_{q=1}^N\sum_{k_1,\cdots,k_q =0}^{N-q} \Ga_q^k(X)
D^{-2(|k|_1+q)} \mod OP^{-N-1}.
\end{equation}
Note that the $\Ga_q^k(X)$ are in $\DD_{1}(\A)$, which, with
(\ref{Ydev}) proves that $Y$
and thus $\eps^r(Y)= \nabla^r(Y)D^{-2r}$, are also in $\Psi_{1}(\A)$.
\end{proof}

We remark, as in \cite{Cours}, that the fluctuations leave invariant
the first term of the
spectral action (\ref{formuleaction}). This is a generalization of
the fact that in the commutative case, the noncommutative integral $\ncint |D|^{-n}$, where $n$ is the dimension of the manifold, only depends on the principal symbol of the
Dirac operator $\DD$ and this symbol is stable by adding a gauge
potential like in $\DD+A$.
Note however that the symmetrized gauge potential
$A+\epsilon JAJ^{-1}$ is always zero in
this case for any selfadjoint one-form $A$.

\begin{lemma} If the spectral triple is simple,
\begin{align}
\zeta_{D_{A}}(0)-\zeta_{D}(0)=\sum_{q=1}^{n} \tfrac{(-1)^{q}}{q}
\ncint (\wt AD^{-1})^{q}. \label{termconstanttilde}
\end{align}
\end{lemma}
\begin{proof}
Since the spectral triple is simple, equation (\ref{egalite-DAs})
entails that
$$
\zeta_{D_A}(0)-\zeta_{D}(0) = \Tr (K_1(Y,s)|D|^{-s})_{|s=0} \, .
$$
Thus, with (\ref{one-par}), we get
$\zeta_{D_A}(0)-\zeta_{D}(0) = -\half \ncint Y$. Replacing $A$ by
$\wt A$, the same proof as in \cite{CC1} gives
\begin{align*}
-\half \ncint Y = \sum_{q=1}^{n} \tfrac{(-1)^{q}}{q} \ncint (\wt
AD^{-1})^{q}.
\tag*{\qed}
\end{align*}
\hideqed
\end{proof}

\begin{lemma}
    \label{Res-zeta-n-k}
For any $k\in \N_0$,
$$
\underset{s=n-k}{\Res} \, \zeta_{D_A}(s)=
\underset{s=n-k}{\Res} \,\zeta_{D}(s) +
\sum_{p=1}^k \sum_{r_1,\cdots, r_p =0}^{k-p}
\underset{s=n-k}{\Res} \, h(s,r,p) \, \Tr\big(\eps^{r_1}(Y)
\cdots\eps^{r_p}(Y) |D|^{-s}\big),
$$
where
$$
h(s,r,p):=(-s/2)^p\int_{0\leq t_1\leq \cdots \leq t_p\leq 1}
g(-st_1,r_1)\cdots
g(-st_p,r_p) \, dt\, .
$$
\end{lemma}
\begin{proof}
By Lemma \ref{2dev} $(ii)$, $|D_A|^{-s} \sim |D|^{-s} + \sum_{p=1}^k
K_p(Y,s)
|D|^{-s} \mod OP^{-(k+1)-\Re(s)}$, where the convention
$\sum_{\emptyset}=0$ is used.
Thus, we get for $s$ in a neighborhood of $n-k$,
$$
|D_A|^{-s}-|D|^{-s} - \sum_{p=1}^k K_p(Y,s) |D|^{-s} \in
OP^{-(k+1)-\Re(s)}\subseteq
\L^1(\H)
$$
which gives
\begin{equation}
    \label{res-n-k-interm}
\underset{s=n-k}{\Res} \, \zeta_{D_A}(s)= \underset{s=n-k}{\Res}
\,\zeta_{D}(s) + \sum_{p=1}^k \underset{s=n-k}{\Res} \,\Tr
\big(K_p(Y,s) |D|^{-s}\big).
\end{equation}
Let us fix $1\leq p\leq k$ and $N\in \N$. By (\ref{one-par}) we get
\begin{align}
    \label{K_p}
K_p(Y,s)\sim (-\tfrac s2)^p \int_{0\leq t_1\leq \cdots t_p \leq 1}
\sum_{r_1,\cdots,r_p =0}^N g(-st_1,r_1)&\cdots g(-st_p,r_p)
\nonumber\\
&\eps^{r_1}(Y)\cdots \eps^{r_p}(Y)\, dt \, \mod OP^{-N-p-1}.
\end{align}

If we now take $N=k-p$, we get for $s$ in a neighborhood of $n-k$
$$
K_p(Y,s)|D|^{-s} - \sum_{r_1,\cdots,r_p
=0}^{k-p}h(s,r,p)\,\eps^{r_1}(Y)\cdots
\eps^{r_p}(Y)|D|^{-s} \in OP^{-k-1-\Re(s)} \subseteq \L^1(\H)
$$
so (\ref{res-n-k-interm}) gives the result.
\end{proof}

Our operators $|D_A|^k$ are pseudodifferential operators:
\begin{lemma}
For any $k\in \Z$, $\vert D_{A} \vert^k \in \Psi^k(\A)$.
\end{lemma}
\begin{proof}
Using (\ref{K_p}), we see that $K_p(Y,s)$ is a pseudodifferential
operator in $OP^{-p}$, so (\ref{expansion}) proves that $|D_A|^k$
is a pseudodifferential operator in $OP^k$.
\end{proof}

The following result is quite important since it shows that one
can use $\ncint$ for $D$ or $D_{A}$:
\begin{prop}
\label{ncintfluctuated}
If the spectral triple is simple,
$\underset{s=0}{\Res} \, \Tr \big(P |D_A|^{-s}\big) = \ncint P$
for any pseudodifferential operator $P$. In particular, for any $k\in
\N_0$
$$
\ncint |D_A|^{-(n-k)}=\underset{s=n-k}{\Res} \,\zeta_{D_A}(s) .
$$
\end{prop}
\begin{proof}
Suppose $P\in OP^{k}$ with $k\in \Z$ and let us fix $p\geq 1$.
With (\ref{K_p}), we see that for any $N\in \N$,
$$
PK_p(Y,s)|D|^{-s}\sim \sum_{r_1,\cdots,r_p =0}^N h(s,r,p) \,
P\eps^{r_1}(Y)\cdots \eps^{r_p}(Y)|D|^{-s} \mod OP^{-N-p-1+k-\Re(s)}.
$$
Thus if we take $N=n-p+k$, we get
$$
 \underset{s=0}{\Res} \,\Tr \big(P K_p (Y,s) |D|^{-s}\big) =
 \sum_{r_1,\cdots,r_p =0}^{n-p+k} \underset{s=0}{\Res} \,\,
h(s,r,p) \, \Tr \big(P\eps^{r_1}(Y)\cdots \eps^{r_p}(Y) |D|^{-s}\big).
$$
Since $s=0$ is a zero of the analytic function $s\mapsto h(s,r,p)$
and $s\mapsto \Tr
P\eps^{r_1}(Y)\cdots \eps^{r_p}(Y)|D|^{-s}$ has only simple poles
by hypothesis, we see
that $\underset{s=0}{\Res} \, h(s,r,p) \, \Tr
\big(P\eps^{r_1}(Y)\cdots \eps^{r_p}(Y)
|D|^{-s}\big)=0$ and
\begin{equation}
\label{res0K_p}
\underset{s=0}{\Res} \, \Tr \big(P K_p (Y,s) |D|^{-s}\big)=0.
\end{equation}
Using (\ref{expansion}), $ P|D_A|^{-s} \sim P |D|^{-s} +
\sum_{p=1}^{k+n} PK_p(Y,s) |D|^{-s} \mod OP^{-n-1-\Re(s)} $ and thus,
\begin{equation}\label{res0PD_A}
\underset{s=0}{\Res} \, \Tr (P|D_A|^{-s}) =\ncint P +
\sum_{p=1}^{k+n} \,
\underset{s=0}{\Res} \, \Tr \big(PK_p(Y,s) |D|^{-s}\big).
\end{equation}
The result now follows from (\ref{res0K_p}) and (\ref{res0PD_A}).
To get the last equality,
one uses the pseudodifferential operator $|D_A|^{-(n-k)}$.
\end{proof}

\begin{prop}
\label{invariance1} If the spectral triple is simple, then
\begin{align}
\ncint {|D_{A}|}^{-n}&=\ncint |D|^{-n}.
\end{align}
\end{prop}
\begin{proof} Lemma \ref{Res-zeta-n-k} and previous proposition for
$k=0$.
\end{proof}

\begin{lemma}
    \label{residus-particuliers}
If the spectral triple is simple,
\begin{align*}
\hspace{-2.2cm} (i) & \quad \ncint |D_A|^{-(n-1)}= \ncint |D|^{-(n-1)}
-(\tfrac{n-1}{2})\ncint X|D|^{-n-1}.\\
\hspace{-2.2cm} (ii)  & \quad \ncint |D_A|^{-(n-2)}= \ncint
|D|^{-(n-2)}+\tfrac{n-2}{2}\big(-\ncint X|D|^{-n} + \tfrac{n}{4}
\ncint X^2
|D|^{-2-n} \big).
\end{align*}
\end{lemma}

\begin{proof}
$(i)$ By (\ref{expansion}),
$$
\underset{s=n-1}{\Res} \, \zeta_{D_A}(s) -\zeta_{D}(s)=
\underset{s=n-1}{\Res} \,(-s/2)
\Tr \big(Y |D|^{-s}\big) = -\tfrac{n-1}{2} \, \underset{s=0}{\Res} \,
\Tr
\big(Y|D|^{-(n-1)}|D|^{-s} \big)
$$
where for the last equality we use the simple dimension spectrum
hypothesis. Lemma \ref{2dev} $(i)$ yields
$Y\sim XD^{-2} \mod OP^{-2}$ and $Y|D|^{-(n-1)}\sim
X|D|^{-n-1} \mod OP^{-n-1}\subseteq \L^1(\H)$. Thus,
$$
\underset{s=0}{\Res} \, \Tr \big( Y|D|^{-(n-1)}|D|^{-s}\big) =
\underset{s=0}{\Res} \,
\Tr \big(X|D|^{-n-1} |D|^{-s}\big) = \ncint  X |D|^{-n-1}.
$$
$(ii)$ Lemma \ref{Res-zeta-n-k} $(ii)$ gives
$$
\underset{s=n-2}{\Res} \, \zeta_{D_A}(s) =
\underset{s=n-2}{\Res} \, \zeta_{D}(s) +
\underset{s=n-2}{\Res} \, \sum_{r=0}^1 h(s,r,1) \,
\Tr \big(\eps^r(Y)|D|^{-s}\big) +
h(s,0,2) \,\Tr \big(Y^2 |D|^{-s}\big).
$$
We have $h(s,0,1)=-\tfrac s2$, $h(s,1,1)=
\half(\tfrac s2)^2$ and $h(s,0,2)= \half (\tfrac
s2)^2$. Using again Lemma \ref{2dev} $(i)$,
$$
Y\sim XD^{-2}-\half \nabla(X)D^{-4} -\half X^2 D^{-4} \mod OP^{-3}.
$$
Thus,
$$
\underset{s=n-2}{\Res} \, \Tr \big(Y|D|^{-s}\big) =
\ncint X|D|^{-n} -\half \ncint
(\nabla(X)+X^2)|D|^{-2-n}.
$$
Moreover, using $\ncint \nabla(X)|D|^{-k}=0$ for any $k\geq 0$
since $\ncint$ is a trace,
$$
\underset{s=n-2}{\Res} \, \Tr \big(\eps(Y)|D|^{-s}\big) =
\underset{s=n-2}{\Res} \, \Tr
\big(\nabla(X)D^{-4}|D|^{-s}\big) = \ncint \nabla(X)|D|^{-2-n}=0.
$$
Similarly, since $Y\sim XD^{-2}$ mod $OP^{-2}$ and
$Y^2\sim X^2D^{-4} \mod OP^{-3}$, we get
$$
\underset{s=n-2}{\Res} \, \Tr \big(Y^2 |D|^{-s}\big) =
\underset{s=n-2}{\Res} \, \Tr
\big(X^2D^{-4}|D|^{-s}\big) = \ncint X^2 |D|^{-2-n}.
$$
Thus,
\begin{align*}
\underset{s=n-2}{\Res} \, \zeta_{D_A}(s) =
\underset{s=n-2}{\Res} \,\zeta_{D}(s) +
&(-\tfrac {n-2}{2})(\ncint X|D|^{-n}
-\half \ncint (\nabla(X)+X^2)|D|^{-2-n})\\
& \quad +\half(\tfrac {n-2}{2})^2
\ncint \nabla(X)|D|^{-2-n}+\half(\tfrac {n-2}{2})^2
\ncint X^2 |D|^{-2-n}.
\end{align*}
Finally,
$$
\underset{s=n-2}{\Res} \, \zeta_{D_A}(s) =
\underset{s=n-2}{\Res} \,\zeta_{D}(s) +
(-\tfrac {n-2}{2}) \big(\ncint X|D|^{-n} -
\half \ncint X^2|D|^{-2-n}\big)+\half(\tfrac
{n-2}{2})^2 \ncint X^2 |D|^{-2-n}
$$
and the result follows from Proposition \ref{ncintfluctuated}.
\end{proof}

\begin{corollary}
    \label{res-n-2-A}
If the spectral triple is simple and satisfies $\ncint |D|^{-(n-2)}=
\ncint \wt A \DD |D|^{-n} = \ncint  \DD \wt A |D|^{-n}=0$, then
$$
 \ncint |D_A|^{-(n-2)} = \tfrac{n(n-2)}{4}
 \big(\ncint \wt A \DD \wt A \DD
|D|^{-n-2}+\tfrac{n-2}{n}\ncint \wt A ^2|D|^{-n}\big).
$$
\end{corollary}
\begin{proof}
By previous lemma,
$$
\underset{s=n-2}{\Res} \, \zeta_{D_A}(s) =
\tfrac{n-2}{2}\big( -\ncint \wt A^2 |D|^{-n}
+\tfrac{n}{4}\ncint (\,\wt A\DD \wt A \DD+ \DD \wt A \DD \wt A+
\wt A \DD^2 \wt A + \DD \wt A^2 \DD \,) |D|^{-n-2} \big).
$$
Since $\nabla(\wt A) \in OP^1$, the trace property of $\ncint$ yields
the result.
\end{proof}

\section{Noncommutative integrals and tadpoles} \label{tadsection}

Let $(\A,\DD,\H,J)$ be a real regular spectral triple of dimension $d$.
Recall that a one-form $A$ is a finite sum of operators like $a_1[\DD,a_2]$
where $a_i \in \A$. The set of one-forms is denoted by $\Omega_\DD^1(\A)$.

\begin{lemma}
  \label{adjoint}
Let $(\A,\DD,\H)$ be a spectral triple and $X \in \Psi(\A)$. Then 
\begin{align*}
\ncint X^*=\overline{  \ncint X}.
\end{align*}
If the spectral triple is real, then, for $X \in \Psi(\A)$, $JXJ^{-1}
\in \Psi(\A)$ and
$$
\ncint JXJ^{-1}=\ncint X^*=\overline{  \ncint X}.
$$
\end{lemma}

\begin{proof}
The first result follows from (for $s$ large enough, so the operators
are traceable)
\begin{align*}
\Tr(X^*\vert \DD\vert^{-s})=\Tr \big((\vert
\DD\vert^{-\bar{s}})X)^*\big)=\overline{ \Tr(\vert \DD \vert^{-\bar{s}}
X)}=\overline{\Tr(X\vert \DD \vert^{-\bar{s}})}.
\end{align*}

The second result is due to the anti-linearity of $J$,
$\Tr(JYJ^{-1})=\overline{\Tr(Y)}$, and $J\vert \DD \vert=\vert \DD \vert
J$, so 
\begin{align*}
\Tr(X \vert \DD \vert^{-s})=\overline{\Tr(JX \vert \DD
\vert^{-s}J^{-1})}=\overline{\Tr(JXJ^{-1}\vert \DD \vert^{-\bar{s}})}.
\tag*{\qed}
\end{align*}
\hideqed
\end{proof}

\begin{corollary}
\label{reel}
For any one-form $A=A^*$, and for $k,\,l \in \N$,
$$
\ncint A^l \,\DD^{-k} \in \R,\quad \ncint \big(A\DD^{-1}\big)^k \in \R,
\quad \ncint A^l \,\vert \DD\vert^{-k} \in \R,\quad  \ncint \chi A^l\,\vert \DD \vert ^{-k} \in \R, 
\quad \ncint A^l\,\DD \, \vert \DD \vert^{-k} \in \R.
$$
\end{corollary}

In \cite{ConnesMarcolli}, is introduced the following

\begin{definition}
\label{Deftadpole}
In $(\A,\,\H,\,\DD)$, the tadpole
$\Tad_{\DD+ A}(k)$ of order $k$, for  $k \in \set{d-l \, : \, l \in \N}$  is
the term linear in $A=A^{*}\in \Omega_\DD^1$, in the $\Lambda^k$ term
of 
\eqref{formuleaction} (considered as an infinite series) where
$\DD_{A}=\DD+ A$.

If moreover, the triple $(\A,\,\H,\,\DD,\,J)$ is real, the tadpole
$\Tad_{\DD+\tilde A}(k)$ is the term linear in $A$, in the $\Lambda^k$ term
of \eqref{formuleaction} where $\DD_{A}=\DD+\wt A$.
\end{definition}

\begin{prop}
    \label{valeurtadpole}
Let $(\A,\,\H,\,\DD)$ be a spectral triple of dimension $d$ with simple dimension spectrum. Then 
\begin{align}
    \label{tadpolen-k}
&\Tad_{\DD+A}(d-k) =-(d-k)\ncint A \DD |\DD|^{-(d-k) -2}, \quad \forall  k\neq d,\\
    \label{tadpole0}
& \Tad_{\DD+ A}(0)=-\ncint  A \DD^{-1}.
\end{align}
Moreover, if the triple is real,  $\Tad_{\DD+\wt A}= 2\Tad_{\DD+A}$.
\end{prop}

\begin{proof}
By Lemma \ref{Res-zeta-n-k} and Proposition \ref{ncintfluctuated}, we have the following
formula,
for any $k\in \N$,
$$
\ncint |\DD_A|^{-(d-k)}=
\ncint |\DD|^{-(d-k)}+
\sum_{p=1}^k \sum_{r_1,\cdots, r_p =0}^{k-p}
\underset{s=d-k}{\Res} \, h(s,r,p) \, \Tr\big(\eps^{r_1}(Y)
\cdots\eps^{r_p}(Y) |\DD|^{-s}\big),
$$
where
\begin{align*}
\hspace{1cm}&h(s,r,p):=(-s/2)^p\int_{0\leq t_1\leq \cdots \leq
t_p\leq 1}
g(-st_1,r_1)\cdots
g(-st_p,r_p) \, dt,  \\
&\eps^r (T):=\nabla(T)\DD^{-2r}, \, \nabla(T):=[\DD^{2},T],\\
&g(z,r):=\tbinom{z/2}{r} \text{ with }g(z,0):=1,\\
&Y \sim \sum_{q=1}^N\sum_{k_1,\cdots,k_q =0}^{N-q} \Ga_q^k(X)
\DD^{-2(|k|_1+q)} \mod OP^{-N-1} \text{ for any N}\in \N^*, \\
& X:=\wt A \DD + \DD \wt A+\wt A^2, \wt A :=A+\epsilon JAJ^{-1},\\
& \Ga_q^k(X):=\tfrac{(-1)^{|k|_1+q+1}}{|k|_1+q} \nabla^{k_q}
\big(X\nabla^{k_{q-1}}(\cdots X\nabla^{k_1}(X)\cdots)\big)
\,,\,\forall q\in \N^* \,,\,
k=(k_1,\cdots,k_q)\in \N^q.
\end{align*}
As a consequence, for $k\neq n$, only the terms with $p=1$ contribute
to the linear part:
$$
\Tad_{\DD + \wt A}(d-k)= \Lin_A(\ncint |\DD_A|^{-(d-k)}) =\sum_{r=0}^{k-1}
\underset{s=d-k}{\Res} \, h(s,r,1) \,
\Tr\big(\eps^{r}(\Lin_A(Y))|\DD|^{-s}\big)\, .
$$
We check that for any $N\in \N^*$, 
$$
\Lin_A(Y) \sim \sum_{l=0}^{N-1} \Ga_1^l(\wt A \DD + \DD \wt A)
\DD^{-2(l+1)} \mod OP^{-N-1}.
$$
Since $\Ga_1^l(\wt A \DD + \DD \wt A) = \tfrac{(-1)^{l}}{l+1}
\nabla^{l}
(\wt A \DD + \DD \wt A)= \tfrac{(-1)^{l}}{l+1} \{ \nabla^l(\wt A),\DD
\}$, we get, assuming the dimension spectrum to be simple
\begin{align*}
\Tad_{\DD+ \wt A}(d-k)&= \sum_{r=0}^{k-1}
\underset{s=d-k}{\Res} \, h(s,r,p) \,
\Tr\big(\eps^{r}(\Lin_A(Y))|\DD|^{-s}\big) \\
&= \sum_{r=0}^{k-1} h(d-k,r,1) \sum_{l=0}^{k-1-r}\tfrac{(-1)^{l}}{l+1}
\underset{s=d-k}{\Res} \,\Tr\big( \eps^{r}(\{ \nabla^l(\wt A),\DD\})|\DD|^{-s-2(l+1)}\big)\\
 & =  2\sum_{r=0}^{k-1} h(d-k,r,1) \sum_{l=0}^{k-1-r}\tfrac{(-1)^{l}}{l+1}
\ncint \nabla^{r+l}(\wt A) \DD |\DD|^{-(d-k + 2(r+l)) -2}  \\ 
& = -(d-k)
\ncint \wt A \DD |\DD|^{-(d-k) -2},
\end{align*}
because in the last sum, only the case $r+l=0$ remains, so $r=l=0$.

Formula \eqref{tadpole0} is a direct application of \eqref{termconstanttilde}.

The link between  $\Tad_{\DD+\wt A}$ and $\Tad_{\DD+A}$ follows from $J\DD=\epsilon 
\DD J$ and Lemma \ref{adjoint}.
\end{proof}

\begin{corollary} 
\label{Atilde=0}
In a real spectral triple $(\A,H,\DD)$, if $A=A^*\in \Omega_\DD^1(\A)$ is such that $\wt A=0$,  then $\Tad_{D+A}(k) =0$ for any $k\in \Z$, $k\leq d$. 
\end{corollary}

\begin{remark}

Note that $\tilde A=0$ for all $A=A^{*}\in \Omega_\DD^1$, 
when $\A$ is commutative and $JaJ^{-1}=a^*$, for all $a \in \A$, see \eqref{JAJ}, so one can only use $\DD_A=\DD+A$.

But we can have $\A$ commutative and $JaJ^{-1}\neq a^*$ \cite{CGravity,
Kraj}:\\
Let $\A_1=\C \oplus \C$ represented on $\H_1=\C^3$ with, for some
complex number $m\neq0$, 
\begin{align*} 
\pi_1(a)&:=  \left( \begin{array}{ccc}
b_1 & 0 & 0\\
0 & b_1 & 0\\
0 & 0 & b_2
\end{array} \right), \,\,for \,\, a=(b_1,\,b_2) \in \A \\
\DD_1&:= \left( \begin{array}{ccc}
0 & m & m\\
\bar m & 0 & 0\\
\bar m & 0 & b
\end{array} \right),  \,\,
\chi_1:=  \left( \begin{array}{ccc}
1 & 0 & 0\\
0 & -1 & 0\\
0 & 0 & -1
\end{array} \right),  
\,\,
J_1:=\left( \begin{array}{ccc}
1 & 0 & 0\\
0 & 0 & 1\\
0 & 1 & 0
\end{array} \right) \circ \, cc
\end{align*}
where $cc$ is the complex conjugation. Then ($\A_1,\,\H_1,\,\DD_1$)
is a commutative real spectral triple of dimension $d=0$ with non
zero one-forms and such that $J_1\pi_1(a)J_1^{-1}= \pi_1(a^*)$ only if $a=(b_1,b_1)$. 

Take a commutative geometry 
\big($\A_2=C^{\infty}(M), \,\H=L^2(M,S),\,\DD_2,\, \chi_2,\, J_2$\big)
defined in \ref{rieman} where $d=dim M$ is even, and then the tensor
product of the two spectral triples, namely $\A=\A_1\otimes \A_2$,
$\H=\H_1 \otimes \H_2$, $\DD=\DD_1 \otimes \chi_2 + 1\otimes \DD_2$,
$\chi=\chi_1 \otimes \chi_2$ and $J$ is either $\chi_1 J_1\otimes
J_2$ when $d\in \set{2,6}$ mod 8 or $J_1 \otimes J_2$ in other cases,
see \cite{CGravity,Vanhecke}.

Then $(\A,\, \H,\, \DD)$ is a real commutative triple of dimension $d$ 
such that $\tilde A \neq 0$ for some selfadjoint
one-forms $A$, so is not exactly like in definition \ref{rieman}.
\end{remark}

The vanishing tadpole of order 0 has the following equivalence (see \cite{CC1})
\begin{align}
  \label{equ}
\ncint A \DD^{-1}=0, \,\forall A \in \Omega_D^1(\A)
\,\Longleftrightarrow  \,\ncint a b=\ncint a \a(b), \,\forall a,b
\in \A, 
\end{align}
where $\a(b):=\DD b\DD^{-1}$, equivalence which can be generalized as

 \begin{lemma}
\label{termecumulŽ}
In a spectral triple $(\A,\,\H,\,\DD)$, for any $k\in \N$,
\begin{align*}
& \ncint (A\DD^{-1})^n=0, \,\forall A \in \Omega^{1}_{\DD}(\A)
,\,\,\forall
n\in \set{1,\cdots,k} \Longleftrightarrow  \,\ncint \prod_{j=1}^k a_j\a(b_j)
 =\ncint \prod_{j=1}^k a_jb_j, \,\, \forall
a_j,\,b_j\in \A.
\end{align*}
\end{lemma}
 
 \begin{proof}
Note that $a[\DD,b] \DD^{-1}=a\,\tilde \a (b)$ where $\tilde \a (b)
:=\a(b)-b$.

Assuming the left hand-side, we get
\begin{align*}
0&=\ncint (A\DD^{-1})^n=\ncint a_1\tilde \a(b_1)\ldots a_j\tilde \a(b_j) \ldots a_n\tilde
\a(b_n)\\
&=\ncint a_1\tilde \a(b_1)\ldots a_j \a(b_j)a_{j+1}\tilde\a(b_{j+1})
\ldots a_k\tilde \a(b_k)-\ncint a_1\tilde \a(b_1)\ldots a_jb_ja_{j+1}\tilde \a(b_j) \ldots
a_n\tilde \a(b_n)
\end{align*}
$\forall \,a_j,\,b_j \in \A$. But the last term is zero if $\ncint
{(A\DD^{-1})}^{n-1}=0$ for all $A$. By induction, we end up with 
$0=\ncint a_1\a(b_1) \cdots a_{n-1}\a(b_{n-1})\,a_n \tilde \a(b_n)$. Varying $n$ between $1$ and $k$, we get the right hand-side.
\end{proof}

\chapter{Spectral action on noncommutative torus}

\section{Introduction}

In \cite{GI2002}, the spectral action on NC-tori was only computed
only for operators of the form $\DD + A$ and computed for $\DD_{A}$
in \cite{GIVas}.  It appears that the implementation of
the real structure via $J$, does not change the spectral action, up to a
coefficient, when the torus has dimension 4.
Here we prove that this can be also directly obtained from
the Chamseddine--Connes analysis of \cite{CC1} that we follow quite
closely. Let us recall that
\begin{align*}
    \SS(\DD_{A},\Phi,\Lambda) \, = \,\sum_{0<k\in Sd^+} \Phi_{k}\,
    \Lambda^{k} \ncint \vert D_{A}\vert^{-k} + \Phi(0) \,
    \zeta_{D_{A}}(0) +\mathcal{O}(\Lambda^{-1})
\end{align*}
where $D_A = \DD_A + P_A$, $P_A$ the projection on $\Ker \DD_A$,
$\Phi_{k}= \half\int_{0}^{\infty} \Phi(t) \, t^{k/2-1} \, dt$ and
$Sd^+$ is the strictly positive part of the spectrum dimension of $(\A,\H,\DD)$.

In section 2, all previous technical points are widely used for
the computation of terms in (\ref{formuleaction}) or
(\ref{termconstanttilde}) on noncommutative torus. Most of the terms are based on residues of certain zeta functions and series of zeta functions that are studied in section 4. We show in particular that the vanishing tadpole hypothesis is satisfied on the torus.

The spectral action is obtained in section 3
and we conjecture that the noncommutative spectral action
of $\DD_{A}$ has terms proportional to the spectral action
of $\DD+A$ on the commutative torus.

Since the computation of zeta functions is crucial here, we
investigate in section 4 residues of series and integrals.
This section contains independent interesting results on the
holomorphy of series of holomorphic functions. In particular,
the necessity of a Diophantine constraint is naturally emphasized.

All these results on spectral action are quite important in physics,
especially in quantum field theory and particle physics, where one
adds to the effective action some counterterms explicitly given by
(\ref{termconstanttilde}), see for instance
\cite{Carminati,CC,CC1,CCM,Gayral,Goursac,GI2002,GIVas,
Knecht,Strelchenko,Vassilevich2002,Vassilevich2005,Vassilevich2007}.

\section{The noncommutative torus}
\subsection{Notations}

Let $\Coo(\T^n_\Th)$ be the smooth noncommutative
$n$-torus associated to a non-zero skew-symmetric deformation matrix
$\Th \in
M_n(\R)$ (see \cite{ConnesTorus}, \cite{RieffelRot}). This means that
$\Coo(\T^n_\Th)$ is the algebra generated by $n$
unitaries $u_i$, $i=1,\dots,n$ subject to the relations
\begin{equation}
\label{rel}
u_i\,u_j=e^{i\Th_{ij}}\,u_j\,u_i,
\end{equation}
and with
Schwartz coefficients: an
element $a\in\Coo(\T_\Th^n)$ can be written as
$a=\sum_{k\in\Z^n}a_k\,U_k$, where $\{a_k\}\in\SS(\Z^n)$ with
the Weyl elements defined by $U_k:=e^{-\frac i2 k.\chi
k}\,u_1^{k_1}\cdots
u_n^{k_n}$, $k\in\Z^n$, relation
\eqref{rel} reads
\begin{equation}
\label{rel1}
U_{k}U_{q}=e^{-\frac i2 k.\Theta q} \,U_{k+q}, \text{ and }
U_{k}U_{q}=e^{-i k.\Theta q} \,U_{q}U_{k}
\end{equation}
where $\chi$ is
the matrix restriction of $\Theta$ to its upper triangular part.
Thus unitary operators $U_{k}$ satisfy $U_{k}^*=U_{-k}$ and
$[U_{k},U_{l}]=-2i\sin(\frac 12 k.\Th l)\,U_{k+l}$.

Let $\tau$ be the trace on $\Coo(\T^n_\Th)$ defined by
$\tau\big( \sum_{k\in\Z^n}a_k\,U_k \big):=a_0$
and $\H_{\tau}$ be the GNS Hilbert space obtained
by completion of $ \Coo(\T_\Th^n)$
with respect of the norm induced by the scalar product
$\langle a,b\rangle:=\tau(a^*b)$.
On $\H_{\tau}=\set{\sum_{k\in\Z^n}a_k\,U_k \, : \, \{a_{k}\}_{k} \in
l^2(\Z^n) }$, we consider the left and right regular
representations of
$\Coo(\T_\Th^n)$ by bounded operators, that we denote respectively
by $L(.)$ and $R(.)$.

Let also $\delta_\mu$, $\mu\in \set{1,\dots,n}$, be the $n$ (pairwise
commuting)
canonical derivations, defined by
\begin{equation}
\delta_\mu(U_k):=ik_\mu U_k. \label{dUk}
\end{equation}

We need to fix notations: let $\A_{\Th}:=C^{\infty}(\T_{\Th}^n)$
acting on $\H:=\H_{\tau}\otimes \C^{2^m}$ with $n=2m$ or $n=2m+1$
(i.e., $m=\lfloor \tfrac n2 \rfloor$ is the integer part of $\tfrac
n2$),
the square integrable sections of the trivial spin bundle over $\T^n$.

Each element of $\A_{\Th}$ is represented on $\H$ as
$L(a)\otimes1_{2^m}$ where $L$ (resp. $R$) is the left (resp. right)
multiplication. The Tomita conjugation $J_{0}(a):=a^*$
satisfies $[J_{0},\delta_{\mu}]=0$ and we define
$J:=J_{0}\otimes C_{0}$
where $C_{0}$ is an operator on $\C^{2^m}$.
The Dirac operator is given by
\begin{align}
\label{defDirac}
\DD:=-i\,\delta_{\mu}\otimes \gamma^{\mu},
\end{align}
where we use hermitian Dirac matrices $\gamma$. It is defined and
symmetric on the
dense subset of $\H$ given by $C^{\infty}(\T_{\Th}^n) \otimes
\C^{2^{m}}$. We still note $\DD$ its selfadjoint extension. This
implies
\begin{align}
    \label{CGamma}
C_{0}\ga^{\alpha}=-\eps \ga^\alpha C_{0},
\end{align}
and
$$
\DD\ U_k \otimes e_i = k_\mu U_k \otimes \gamma^\mu e_i ,
$$
where $(e_i)$ is the canonical basis of $\C^{2^m}$. Moreover,
$C_{0}^2=\pm 1_{2^m}$
depending on the parity of $m$. Finally, one introduces the chirality
(which in the
even case is $\chi:=id \otimes (-i)^{m} \gamma^1 \cdots \gamma^{n}$)
and this yields
that $(\A_{\Th},\H,\DD,J,\chi)$ satisfies all axioms of a spectral
triple, see
\cite{Book,Polaris}.

The perturbed Dirac operator $V_{u}\DD V_{u}^*$ by
the unitary
$$
V_{u}:=\big(L(u)\otimes 1_{2^m}\big)J\big(L(u)\otimes
1_{2^m}\big)J^{-1},
$$
defined for every unitary $u \in \A$,
$uu^{*}=u^{*}u=U_{0}$,
must satisfy condition (\ref{Jcom}) (which is equivalent
to $\H$ being endowed with a structure of $\A_{\Th}$-bimodule). This
yields the necessity of a symmetrized covariant Dirac operator:
$$
\DD_{A}:=\DD + A + \epsilon J\,A\,J^{-1}$$
since
$V_{u}\DD V_{u}^{*}=\DD_{L(u)\otimes 1_{2^m}[\DD,L(u^{*})
\otimes 1_{2^m}]}$:
in fact, for $a \in \A_{\Th}$, using $J_{0}L(a)J_{0}^{-1}=R(a^*)$,
we get $$\epsilon J\big(L(a)\otimes
\gamma^{\alpha}\big)J^{-1}=-R(a^*)\otimes \gamma^{\alpha}$$
and that the representation $L$ and the
anti-representation $R$ are $\C$-linear, commute and satisfy
$$
[\delta_{\alpha},L(a)]=L(\delta_{\alpha}a),\quad
[\delta_{\alpha},R(a)]=R(\delta_{\alpha}a).
$$
This induces some covariance property for the Dirac operator:
one checks that for all $k \in \Z^{n}$,
\begin{align}
\label{puregauge1}
L(U_{k})\otimes 1_{2^m}[\DD,L(U_{k}^{*})\otimes 1_{2^m}]&=1\otimes
(-k_{\mu}\ga^{\mu}),
\end{align}
so with (\ref{CGamma}), we get $U_{k}[\DD,U_{k}^{*}]+\epsilon
JU_{k}[\DD,U_{k}^{*}]J^{-1}=0$ and
\begin{align}
\label{covariance}
V_{U_{k}} \,\DD \, V_{U_{k}}^{*}=
\DD=\DD_{L(U_{k})\otimes 1_{2^m}[\DD,L(U_{k}^{*})\otimes 1_{2^m}]}.
\end{align}
Moreover, we get the gauge transformation:
\begin{align}
\label{gaugeDirac}
V_{u} \DD_{A} V_{u}^{*}= \DD_{\gamma_{u}(A)}
\end{align}
where the gauged transform one-form of $A$ is
\begin{align}
\label{gaugetransform}
\gamma_{u}(A):=u[\DD,u^{*}]+uAu^{*},
\end{align}
with the shorthand
$L(u)\otimes 1_{2^m} \longrightarrow u$.

As a consequence, the spectral action is gauge invariant:
$$
\SS(\DD_{A},\Phi,\Lambda)=\SS(\DD_{\gamma_{u}(A)},\Phi,\Lambda).
$$

An arbitrary selfadjoint one-form $A$, can be written as
\begin{equation}
\label{connection}
A = L(-iA_{\alpha})\otimes\gamma^{\alpha},\,\, A_{\alpha}
=-A_{\alpha}^* \in
\A_{\Th},
\end{equation}
thus
\begin{equation}
\label{dirac}
\DD_{A}=-i\,\big(\delta_{\alpha}+L(A_{\alpha})-R(A_{\alpha})\big)
 \otimes \gamma^{\alpha}.
\end{equation}
Defining $$\tilde A_{\alpha}:=L(A_{\alpha})-R(A_{\alpha}),$$
we get
$\DD_{A}^2=-g^{{\alpha}_{1} {\alpha}_{2}}(\delta_{{\alpha}_{1}}+\tilde
A_{{\alpha}_{1}})(\delta_{{\alpha}_{2}}+\tilde
A_{{\alpha}_{2}})\otimes 1_{2^m} - \tfrac 12
\Omega_{{\alpha}_{1} {\alpha}_{2}}\otimes \gamma^{{\alpha}_{1}
{\alpha}_{2}}
$
where
\begin{align*}
\gamma^{{\alpha}_{1} {\alpha}_{2}}
&:=\tfrac 12(\gamma^{{\alpha}_{1}}\gamma^{{\alpha}_{2}}
-\gamma^{{\alpha}_{2}}\gamma^{{\alpha}_{1}}) ,\\
\Omega_{{\alpha}_{1} {\alpha}_{2}}
&:=[\delta_{{\alpha}_{1}}+
\tilde A_{{\alpha}_{1}},\delta_{{\alpha}_{2}}
+\tilde A_{{\alpha}_{2}}]\,
=L(F_{{\alpha}_{1} {\alpha}_{2}}) - R(F_{{\alpha}_{1} {\alpha}_{2}})
\end{align*}
with
\begin{align}
    \label{Fmunu}
F_{{\alpha}_{1} {\alpha}_{2}}:=\delta_{{\alpha}_{1}}(A_{{\alpha}_{2}})
-\delta_{{\alpha}_{2}}(A_{{\alpha}_{1}})+[A_{{\alpha}_{1}},A_{{\alpha}_{2}}].
\end{align}
In summary,
\begin{align}
\label{D2}
\DD_{A}^2=-\delta^{{\alpha}_{1} {\alpha}_{2}}
\Big(
\delta_{{\alpha}_{1}}+L(A_{{\alpha}_{1}})-R(A_{{\alpha}_{1}})\Big)
\Big(\delta_{{\alpha}_{2}}+L(A_{{\alpha}_{2}})-R(A_{{\alpha}_{2}})\Big)
\otimes 1_{2^m}
\nonumber\\
-\tfrac 12\,\big(L(F_{{\alpha}_{1} {\alpha}_{2}}) - R(F_{{\alpha}_{1}
{\alpha}_{2}})\big)
\otimes \gamma^{{\alpha}_{1} {\alpha}_{2}}.
\end{align}

\subsection{Kernels and dimension spectrum}
We now compute the kernel of the perturbed Dirac operator:
\begin{prop}
    \label{noyaux}
(i) $\Ker \DD=U_0\otimes \C^{2^m}$, so $\dim \Ker \DD =
2^m$.

(ii) For any selfadjoint one-form $A$, $\Ker \DD \subseteq \Ker
\DD_A$.

(iii) For any unitary $ u\in \A$, $\Ker \DD_{\gamma_{u}(A)}=V_{u}\,
\Ker \DD_{A}$.
\end{prop}
\begin{proof}
$(i)$ Let $\psi =\sum_{k,j} c_{k,j} \, U_k \otimes e_j \in \Ker \DD$.
Thus, $0=\DD^2
\psi =\sum_{k,i} c_{k,j} |k|^2\, U_k \otimes e_j$ which entails that
$c_{k,j}|k|^2=0$
for any $k \in \Z^n$ and $1\leq j\leq 2^m$. The result follows.

$(ii)$ Let $\psi \in \Ker \DD$. So, $\psi = U_0 \otimes v$
with $v\in \C^{2^m}$ and from (\ref{dirac}), we get
\begin{align*}
\DD_A \psi &= \DD \psi + (A+\epsilon J AJ^{-1})\psi = (A+\epsilon J
AJ^{-1})\psi=
-i[A_\a,U_0]\otimes \ga^\a v = 0
\end{align*}
since $U_{0}$ is the unit of the algebra, which proves that $\psi \in
\Ker \DD_A$.

$(iii)$ This is a direct consequence of (\ref{gaugeDirac}).
\end{proof}

\begin{corollary}
Let $A$ be a selfadjoint one-form. Then
$\Ker \DD_A=\Ker \DD$ in the following cases:

 (i) $A_{u}:=L(u)\otimes
1_{2^m}[\DD,L(u^*)\otimes 1_{2^m}]$ when $u$ is a unitary in $\A$.

 (ii) $\vert \vert A \vert \vert <\tfrac12$.

 (iii) The matrix $\tfrac{1}{2\pi}\Th$ has only integral coefficients.
 \end{corollary}

\begin{proof}
$(i)$ This follows from previous result because
$V_{u} (U_{0}\otimes v)= U_{0} \otimes v$ for any $v\in \C^{2^m}$.

$(ii)$ Let $\psi=\sum_{k,j}c_{k,j}\, U_{k}\otimes e_{j} $ be in $ \Ker
\DD_{A}$ (so $\sum_{k,j} \vert c_{k,j}\vert^2<
\infty$) and $\phi:=\sum_{j}c_{0,j}\, U_{0}\otimes e_{j}$. Thus
$\psi':=\psi-\phi \in \text{Ker }\DD_{A}$ since $\phi \in \Ker \DD
\subseteq \Ker \DD_A$ and
$$
\vert \vert \sum_{0\neq k \in \Z^n,\,j}
c_{k,j}\,k_{\alpha}\,U_{k}\otimes
\gamma^{\alpha}e_{j}\vert \vert^2=\vert \vert \DD
\psi'\vert\vert^2=\vert \vert -(A + \epsilon
JAJ^{-1})\psi'\vert \vert^2 \leq 4\vert\vert A \vert \vert^{2}\vert
\vert \psi' \vert
\vert^{2} <\vert \vert \psi' \vert \vert^{2}.
$$
Defining $X_{k}:=\sum_{\alpha}k_{\alpha}\gamma_{\alpha}$,
$X_{k}^{2}=\sum_{\alpha}\vert k_{\alpha}\vert^{2}\, 1_{2^{m}}$ is
invertible and the vectors $\set{U_{k}\otimes X_{k}e_{j}}_{0\neq k\in
\Z^{n},\,j}$ are orthogonal in $\H$, so
$$
\sum_{0\neq k\in \Z^{n},\,j}\big( \sum_{\alpha} \vert k_{\alpha}\vert
^{2} \big)\, \vert
c_{k,j}\vert^{2} < \sum_{0\neq k\in \Z^{n},\,j}\vert c_{k,j}\vert^{2}
$$
which is possible only if $c_{k,j}=0, \, \forall k,\,j$ that is
$\psi'=0$ et
$\psi=\phi \in \text{Ker }\DD$.

$(iii)$ This is a consequence of the fact that the algebra is
commutative, thus $A+\epsilon JAJ^{-1}=0$.
\end{proof}

Note that if $\wt A_{u}:=A_{u}+\epsilon JA_{u}J^{-1}$, then by
(\ref{puregauge1}), $\wt A_{U_{k}}=0$ for all
$k \in \Z^n$ and $\norm{A_{U_{k}}}=\vert k\vert$,
but for an arbitrary unitary $u\in \A$, $\wt A_{u}\ne 0$ so
$\DD_{A_{u}}\ne \DD$.

Naturally the above result is also a direct consequence of the fact
that the eigenspace of an isolated eigenvalue of an operator is not
modified by small perturbations. However, it is interesting to
compute the last result directly to emphasize the difficulty of the
general case:

Let $\psi=\sum_{l\in \Z^n, 1\leq j \leq 2^m}c_{l,j}\, U_{l}\otimes
e_{j}\in \Ker
\DD_A$, so $\sum_{l\in \Z^n, 1\leq j \leq 2^m} \vert c_{l,j}\vert^2
<Ê \infty$. We
have to show that $\psi\in$ Ker $\DD$ that is $c_{l,j}=0$ when $l\ne
0$.

Taking the scalar product of $\langle U_{k} \otimes e_{i}\vert$ with
$$
0=\DD_{A}\psi=\sum_{l,\,\a,\,j} c_{l,\,j}(l^{\a}U_{l}-i[A_{\a},U_{l}]
)\otimes
\gamma^{\a}e_{j},
$$
we obtain
$$
0=\sum_{l,\,\a,\,j} c_{l,\,j} \big(l^{\a}\delta_{k,l}-i\langle
U_{k},[A_{\a},U_{l}]\rangle \big)\langle e_{i},\gamma^{\a}e_{j}
\rangle.
$$
If $A_{\a}=\sum_{\a,l}a_{\a,l}\, U_{l} \otimes \gamma^{\a}$ with
$\set{a_{\a,l}}_{l}
\in \SS(\Z^n)$, note that $[U_{l},U_{m}]=-2i \sin(\tfrac 12 l.\Th m)
\, U_{l+m}$ and
$$
\langle U_{k},[A_{\a},U_{l}]\rangle = \sum_{l'\in
\Z^{n}}a_{\a,l'}(-2i \sin (\tfrac 12
l'.\Th l) \langle U_{k}, U_{l'+l}\rangle=-2i\, a_{\a,k-l} \,
\sin(\tfrac 12 k.\Th l).
$$
Thus
\begin{align}
    \label{contraintenoyau}
0=\sum_{l\in \Z^{n}}\sum_{\a=1}^{n}\sum_{j=1}^{2^{m}} c_{l,\,j}
\big(l^{\a}\delta_{k,l} -2a_{\a,k-l} \, \sin(\tfrac 12 k.\Th l) \big)
\, \langle
e_{i},\gamma^{\a}e_{j} \rangle, \quad \forall k\in \Z^n, \, \forall
i, 1\leq i \leq
2^{m}.
\end{align}

\medskip

{\it We conjecture that $\Ker \DD=\Ker \DD_A$ at least for generic
$\Th$'s}:

the constraints (\ref{contraintenoyau}) should imply $c_{l,j} = 0$
for all $j$ and all $l \neq 0$ meaning $\psi \in \Ker \DD$. When
$\tfrac{1}{2\pi}\Th$ has only integer coefficients, the sin part of
these constraints disappears giving the result.
\medskip
We shall use the following
\begin{definition}\label{ba}
(i) Let $\delta >0$. A vector $a \in \R^n$ is said to be
$\delta-$diophantine
if there exists $c >0$ such that $|q . a -m| \geq c \,|q|^{-\delta}$,
$\forall q \in \Z^n \setminus \set{0}$ and $\forall m \in \Z$. \\
We note ${\cal BV}(\delta )$ the set of $\delta-$diophantine
vectors and ${\cal
BV} :=\cup_{\delta >0} {\cal BV}(\delta)$ the set of diophantine vectors.\par
(ii) A matrix $\Th \in {\cal M}_{n}(\R)$ (real $n \times n$ matrices)
will be
said to be diophantine if there
exists $u \in \Z^n$ such that  ${}^t\Th (u)$ is a diophantine
vector of $\R^n$.
\end{definition}

\begin{lemma}
\label{spectrumset}
If $\tfrac{1}{2\pi} \Th$ is diophantine,
$Sp\big(\Coo(\T^n_\Th),\H,\DD\big)=\Z$ and all these poles are simple.
\end{lemma}

\begin{proof}
Let $B\in \DD(\A)$ and $p\in \N_0$. Suppose that $B$ is of the form
$$
B= a_r b_r
\DD^{q_{r-1}}|\DD|^{p_{r-1}} a_{r-1}b_{r-1}\cdots
\DD^{q_1}|\DD|^{p_1} a_1 b_1
$$
where $r\in \N$, $a_i \in \A$, $b_i\in J\A J^{-1}$, $q_i, p_i \in
\N_0$.
We note $a_i=:\sum_l a_{i,l}\,U_l$ and
$b_i=:\sum_i b_{i,l} \,U_l$. With the shorthand
$k_{\mu_1,\mu_{q_i}}:=k_{\mu_1}\cdots
k_{\mu_{q_i}}$ and $\ga^{\mu_1,\mu_{q_i}}=\ga^{\mu_1}\cdots
\ga^{\mu_{q_i}}$, we get
$$
\DD^{q_1}|\DD|^{p_1}  a_1 b_1 \, U_k \otimes e_j =  \sum_{l_1,\,l'_1}
a_{1,l_1} b_{1,l'_1}
U_{l_1}U_k U_{l'_1}
\,|k+l_1+l'_1|^{p_1}\,(k+l_1+l'_1)_{\mu_1,\mu_{q_1}} \otimes
\ga^{\mu_1,\mu_{q_1}} e_j
$$
which gives, after $r$ iterations,
$$
B U_k \otimes e_j = \sum_{l,l'} \wt a_{l} \wt b_{l} U_{l_r}\cdots
U_{l_1} U_k
U_{l'_1}\cdots U_{l'_r} \prod_{i=1}^{r-1} |k+\wh l_i+\wh
l'_i|^{p_i}(k+\wh l_{i} +\wh
l'_{i})_{\mu^{i}_1,\mu^i_{q_i}} \otimes
\ga^{\mu^{r-1}_1,\mu^{r-1}_{q_{r-1}}}\cdots
\ga^{\mu^1_1,\mu^1_{q_1}} e_j
$$
where $\wt a_l : = a_{1,l_1}\cdots a_{r,l_r}$ and $\wt b_{l'} : =
b_{1,l'_1}\cdots
b_{r,l'_r}$.

Let us note $F_\mu(k,l,l'):=\prod_{i=1}^{r-1}|k+\wh l_i+\wh
l'_i|^{p_i}
(k+\wh l_{i} +\wh l'_{i})_{\mu^{i}_1,\mu^i_{q_i}}$ and $\ga^\mu
:=\ga^{\mu^{r-1}_1,\mu^{r-1}_{q_{r-1}}}\cdots
\ga^{\mu^1_1,\mu^1_{q_1}}$. Thus, with
the shorthand $\sim_c$ meaning modulo a constant function towards the
variable $s$,
$$
\Tr \big(B|D|^{-2p-s}\big) \sim_c {\sum_k}' \, \sum_{l,l'} \wt a_l
\wt b_{l'} \,
\tau\big(U_{-k}U_{l_r}\cdots U_{l_1} U_k U_{l'_1}\cdots U_{l'_r}\big)
\tfrac{F_\mu(k,l,l')}{|k|^{s+2p}} \Tr (\ga^\mu)\, .
$$
Since $U_{l_r}\cdots U_{l_1} U_k = U_k U_{l_r}\cdots U_{l_1}
e^{-i\sum_1^r l_i .\Th
k}$ we get
$$\tau\big(U_{-k}U_{l_r}\cdots U_{l_1} U_k U_{l'_1}\cdots
U_{l'_r}\big)=
\delta_{\sum_1^r l_i+l'_i,0} \, e^{i\phi(l,l')} \, e^{-i\sum_1^r
l_i.\Th k}$$ where $\phi$
is a real valued function. Thus,
\begin{align*}
\Tr \big(B |D|^{-2p-s} \big)&\sim_c {\sum_k}' \, \sum_{l,l'}
e^{i\phi(l,l')}\,\delta_{\sum_1^r l_i+l'_i,0}\, \wt a_l \wt b_{l'}\,
\tfrac{F_\mu(k,l,l')\,e^{-i\sum_1^r l_i.\Th k}}{|k|^{s+2p}} \Tr
(\ga^\mu) \\
&\sim_c f_\mu(s)\Tr (\ga^\mu).
\end{align*}
The function $f_\mu(s)$ can be decomposed has a linear combination of
zeta function
of type described in Theorem \ref{zetageneral} (or, if $r=1$ or all
the $p_i$ are zero,
in Theorem \ref{analytic}).
Thus, $s\mapsto \Tr \big(B |D|^{-2p-s}\big)$
has only poles in $\Z$ and each pole is simple.
Finally, by linearity, we get the result.
\end{proof}
The dimension spectrum of the noncommutative torus is simple:
\begin{prop} \label{zeta(0)}

(i) If $\tfrac{1}{2\pi} \Th$ is diophantine,
the spectrum dimension of $\big(\Coo(\T^n_\Th),\H,\DD\big)$ is
equal to the set $\set{n-k \, :\,  k\in \N_0}$ and all these poles
are simple.

(ii) $\zeta_D(0)=0.$
\end{prop}
\begin{proof}
$(i)$ Lemma \ref{spectrumset} and Remark \ref{remark-spectrum}.

$(ii)$  $\zeta_D(s)={\sum}_{k\in \Z^n}
\sum_{1\leq j\leq 2^m} \<
U_k\otimes e_j, |D|^{-s}U_k\otimes e_{j}>=2^m( {\sum}'_{k\in\Z^n}
\frac{1}{|k|^{s}} + 1) =2^m(\,Z_n(s)+1).$ By (\ref{Zn0}), we get the
result.
\end{proof}

We have computed $\zeta_D(0)$ relatively easily but the
main difficulty of the present work is precisely to calculate
$\zeta_{D_A}(0)$.

\subsection{Noncommutative integral computations}

We fix a self-adjoint 1-form $A$ on the noncommutative torus of
dimension $n$.

\begin{prop}
\label{invariance} If $\tfrac{1}{2\pi}\Th$ is diophantine, then
the first elements of the expansion (\ref{formuleaction}) are
given by
\begin{align}
\ncint {|D_{A}|}^{-n}\,&=\ncint |D|^{-n}=
2^{m+1}\pi^{n/2}\,\Gamma(\tfrac{n}{2})^{-1}.\\
\ncint \vert D_{A}\vert^{n-k}&=0 \text{ for k odd}.\nonumber\\
\ncint \vert D_{A}\vert^{n-2}&=0.\nonumber
\end{align}
\end{prop}
We need few technical lemmas:
\begin{lemma}
\label{traceAD}
On the noncommutative torus, for any $t\in \R$,
$$
\ncint \wt A \DD |D|^{-t}= \ncint \DD \wt A |D|^{-t} =0.
$$
\end{lemma}

\begin{proof}
Using notations of (\ref{connection}), we have
\begin{align*}
\Tr (\wt A \DD |D|^{-s})&\sim_c {\sum}_j {\sum}'_k \langle U_k\otimes
e_j,-i k_\mu|k|^{-s}
[A_\a,U_k] \otimes \ga^\a \ga^\mu e_j \rangle \\
&\sim_c -i\Tr(\ga^\a\ga^\mu) \, {\sum}'_k k_\mu
|k|^{-s} \langle U_k,[A_\a,U_k] \rangle=0
\end{align*}
since $\langle U_k,[A_\a,U_k] \rangle = 0$. Similarly
\begin{align*}
\Tr ( \DD \wt A  |D|^{-s})&\sim_c {\sum}_j {\sum}'_k \langle
U_k\otimes
e_j,|k|^{-s}{\sum}_l a_{\a,l}\,2 \sin \tfrac{k. \Th l}{2} (l+k)_\mu
U_{l+k}
\otimes \ga^\mu \ga^\a e_j \rangle\\
&\sim_c 2\Tr(\ga^\mu \ga^\a){\sum}'_k {\sum}_l a_{\a,l}\sin \tfrac{k.
\Th
l}{2}\,(l+k)_\mu \,|k|^{-s}\langle U_k,U_{l+k} \rangle =0.
\tag*{\qed}
\end{align*}
\hideqed
\end{proof}

Any element $h$ in the algebra generated by $\A$ and $[\DD,\A]$ can
be written as a
linear combination of terms of the form ${a_1}^{p_1}\cdots
{a_n}^{p_r}$ where
$a_i$ are elements of $\A$ or $[\DD,\A]$. Such a term can be written
as a series $b:=\sum a_{1,\a_1,l_1}\cdots a_{q,\a_q,l_q} U_{l_1}\cdots
U_{l_q} \otimes \ga^{\a_1}\cdots \ga^{\a_q}$ where
$a_{i,\a_i}$ are Schwartz sequences and when $a_i=:\sum_l a_l U_l \in
\A$, we set $a_{i,\a,l}=a_{i,l}$ with $\ga^\a =1$. We define
$$
L(b):= \tau \big({\sum}_l a_{1,\a_1,l_1}\cdots a_{q,\a_q,l_q} U_{l_1}
\cdots U_{l_q}\big) \Tr (\ga^{\a_1}\cdots \ga^{\a_q}).
$$
By linearity, $L$ is defined as a linear form on the whole algebra
generated by $\A$
and $[\DD,\A]$.

\begin{lemma}
\label{tracehD}
If $h$ is an element of the algebra generated by $\A$ and $[\DD,\A]$,
$$
\Tr \big(h |D|^{-s}\big) \sim_c L(h)\,  Z_n(s).
$$
In particular, $\Tr \big(h |D|^{-s}\big)$ has at most one pole at
$s=n$.
\end{lemma}
\begin{proof} We get with $b$ of the form $\sum
a_{1,\a_1,l_1}\cdots a_{q,\a_q,l_q} U_{l_1}\cdots U_{l_q} \otimes
\ga^{\a_1}\cdots
\ga^{\a_q}$,
\begin{align*}
\Tr\big(b|D|^{-s}\big)&\sim_c {\sum_{k\in\Z^n}}' \langle  U_k, \sum_l
a_{1,\a_1,l_1}\cdots a_{q,\a_q,l_q} U_{l_1}\cdots U_{l_q}U_k \rangle
\, \Tr
(\ga^{\a_1}\cdots \ga^{\a_q})|k|^{-s} \\
&\sim_c \tau(\sum_l a_{1,\a_1,l_1}\cdots a_{q,\a_q,l_q} U_{l_1}\cdots
U_{l_q})\Tr(\ga^{\a_1}\cdots \ga^{\a_q})\, Z_n(s) =L(b) \,Z_n(s).
\end{align*}
The results follows now from linearity of the trace.
\end{proof}

\begin{lemma}
\label{traceJAJA}
If $\tfrac{1}{2\pi}\Th$ is diophantine, the function
$s\mapsto\Tr \big( \eps JAJ^{-1} A |D|^{-s} \big)$ extends
meromorphically on the whole plane with only one
possible pole at $s=n$. Moreover, this pole is simple and
$$
\underset{s=n}{\Res}\, \Tr \big(\eps JAJ^{-1} A |D|^{-s}\big) =
a_{\a,0}\,a^\a_{0}\
2^{m+1}\pi^{n/2}\,\Ga(n/2)^{-1}.
$$
\end{lemma}
\begin{proof} With $A=L(-i A_\a)\otimes \ga^\a$, we get
$\epsilon J A J^{-1}=R(i A_\a)\otimes \ga^\a$, and by multiplication
$\eps JAJ^{-1} A=R(A_\beta)
L(A_\a)\otimes \ga^{\beta}\ga^\a$. Thus,
\begin{align*}
\Tr\big(\eps JAJ^{-1} A |D|^{-s}\big)&\sim_c {\sum_{k\in\Z^n}}'
\langle U_k,A_\a
U_k A_\beta \rangle \,|k|^{-s}\Tr (\ga^{\beta}\ga^{\a}) \\
&\sim_c{\sum_{k\in\Z^n}}' \,\sum_{l}
a_{\a,l}\,a_{\beta,-l}\,e^{ik.\Th l}\,
|k|^{-s}\Tr (\ga^{\beta}\ga^{\a})\\
&\sim_c 2^m {\sum_{k\in\Z^n}}' \, \sum_{l}
a_{\a,l}\,a^\a_{-l}\,e^{ik.\Th l}\,
|k|^{-s}.
\end{align*}
Theorem \ref{analytic} $(ii)$ entails that ${\sum}'_{k\in\Z^n} \,
\sum_{l} a_{\a,l} \,a^{\a}_{-l}\,e^{ik.\Th l}\, |k|^{-s}$ extends
meromorphically
on the whole plane $\C$ with only one possible pole at $s=n$.
Moreover, this pole is
simple and we have
$$
\underset{s=n}{\Res}\, {\sum_{k\in\Z^n}}' \,\sum_{l} a_{\a,l}
\,a^\a_{-l}\,e^{ik.\Th l}\,
|k|^{-s} =  a_{\a,0}\,a^\a_{0} \, \underset{s=n}{\Res}\, Z_n(s).
$$
Equation (\ref{formule}) now gives the result.
\end{proof}

\begin{lemma}
    \label{traceXD}
If $\tfrac{1}{2\pi}\Th$ is diophantine, then for any $t\in \R$,
$$
\ncint X|D|^{-t} = \delta_{t,n}\, 2^{m+1}\big(-\sum_l
a_{\a,l}\,a^\a_{-l}+
\,a_{\a,0}\,a^\a_{0}\big)\ 2\pi^{n/2}\,\Ga(n/2)^{-1}  .
$$
where $X=\wt A\DD + \DD \wt A + {\wt A}^2$ and $A=:-i\sum_{l}
a_{\a,l}\,U_l\otimes \ga^\a$.
\end{lemma}

\begin{proof} By Lemma \ref{traceAD},
we get $\ncint X|D|^{-t}=\Res_{s=0} \Tr({\wt A}^2 |D|^{-s-t})$. Since
$A$ and
$\eps JAJ^{-1}$ commute, we have $\wt A ^2 = A^2 + JA^2J^{-1} + 2\eps
JAJ^{-1}A$.
Thus,
$$
\Tr({\wt A}^2 |D|^{-s-t})=\Tr( A^2 |D|^{-s-t})+\Tr( JA^2J^{-1}
|D|^{-s-t})+2\Tr ( \eps JAJ^{-1}A |D|^{-s-t}).
$$
Since $|D|$ and $J$
commute, we have with Lemma \ref{tracehD},
$$
 \Tr \big({\wt A}^2 |D|^{-s-t}\big)\sim_c 2L (A^2) \, Z_n(s+t) +
 2 \Tr \big(\eps JAJ^{-1}A |D|^{-s-t}\big).
$$
Thus Lemma \ref{traceJAJA} entails that $\Tr({\wt A}^2 |D|^{-s-t})$
is holomorphic at 0 if $t\neq n$. When $t=n$,
\begin{equation}
    \label{TrA^2}
\underset{s=0}{\Res}\, \Tr\big({\wt A}^2 |D|^{-s-t}\big) =
2^{m+1}\big(-\sum_l a_{\a,l} \, a^\a_{-l}+
\,a_{\a,0}\,a^\a_{0}\ \big)\, 2\pi^{n/2}\,\Ga(n/2)^{-1},
\end{equation}
which gives the result.
\end{proof}

\begin{lemma}
\label{traceAA}
If $\tfrac{1}{2\pi}\Th$ is diophantine, then
$$
\ncint \wt A \DD \wt A \DD |D|^{-2-n}=-\tfrac{n-2}{n}
\ncint \wt A^2 |D|^{-n}.
$$
\end{lemma}
\begin{proof}
With $\DD J = \eps J \DD$, we get
$$
\ncint \wt A \DD \wt A \DD |D|^{-2-n} = 2 \ncint  A \DD  A \DD
|D|^{-2-n} + 2
\ncint \eps JAJ^{-1} \DD A \DD |D|^{-2-n}.
$$
Let us first compute $\ncint  A \DD  A \DD |D|^{-2-n}$. We have, with
$A=:-i
L(A_\a)\otimes \ga^\a=: -i\sum_{l} a_{\a,l} U_l \otimes \ga^\a$,
$$
\Tr \big(A\DD A \DD |D|^{-s-2-n}\big) \sim_c -{{\sum_{k}}}'
\sum_{l_1,l_2}
a_{\a_2,l_2}\,a_{\a_1,l_1}
\,\tau(U_{-k} U_{l_2} U_{l_1} U_k) \,
\tfrac{k_{\mu_1}(k+l_1)_{\mu_2}}{|k|^{s+2+n}}
\Tr(\ga^{\a,\mu})
$$
where $\ga^{\a,\mu}:= \ga^{\a_2}\ga^{\mu_2}  \ga^{\a_1}\ga^{\mu_1}$.
Thus,
$$
\ncint  A \DD  A \DD |D|^{-2-n} = -\sum_{l} a_{\a_2,-l}\,a_{\a_1,l} \,
\underset{s=0}{\Res}\,\big({{\sum_{k}}}'
\tfrac{k_{\mu_1}k_{\mu_2}}{|k|^{s+2+n}} \big)
\Tr(\ga^{\a,\mu}).
$$
We have also, with $\eps JAJ^{-1} = iR(A_\a)\otimes \ga^{a}$,
$$
\Tr \big(\eps JAJ^{-1}\DD A \DD |D|^{-s-2-n}\big) \sim_c
{{\sum_{k}}}' \sum_{l_1,l_2}
a_{\a_2,l_2}a_{\a_1,l_1} \tau(U_{-k} U_{l_1} U_k U_{l_2})
\tfrac{k_{\mu_1}(k+l_1)_{\mu_2}}{|k|^{s+2+n}} \Tr(\ga^{\a,\mu}).
$$
which gives
$$
\ncint  \eps JAJ^{-1} \DD  A \DD |D|^{-2-n} =a_{\a_2,0}a_{\a_1,0} \,
\underset{s=0}{\Res}\,\big({{\sum_{k}}}'
\tfrac{k_{\mu_1}k_{\mu_2}}{|k|^{s+2+n}} \big)
\Tr(\ga^{\a,\mu}).
$$
Thus,
$$
\half \ncint \wt A \DD \wt A \DD |D|^{-2-n} =
\big(a_{\a_2,0}a_{\a_1,0}-\sum_{l}
a_{\a_2,-l}a_{\a_1,l}\big) \Res_{s=0}\big({{\sum_{k}}}'
\tfrac{k_{\mu_1}k_{\mu_2}}{|k|^{s+2+n}} \big) \Tr(\ga^{\a,\mu}).
$$
With ${\sum}'_k \tfrac{k_{\mu_1}k_{\mu_2}}{|k|^{s+2+n}} =
\tfrac{\delta_{\mu_1\mu_2}}{n}Z_n(s+n)$ and
$C_n:=\Res_{s=0} Z_n(s+n) = 2\pi^{n/2} \Ga(n/2)^{-1}$ we obtain
$$
\half \ncint \wt A \DD \wt A \DD |D|^{-2-n} =
\big(a_{\a_2,0}a_{\a_1,0}-\sum_{l}
a_{\a_2,-l}a_{\a_1,l}\big)\tfrac{C_n}{n}
\Tr(\ga^{\a_2}\ga^{\mu}\ga^{\a_1}\ga_\mu).
$$
Since  $\Tr(\ga^{\a_2}\ga^{\mu}\ga^{\a_1}\ga_\mu)=
2^m(2-n)\delta^{\a_2,\a_1}$,
we get
$$
\half \ncint \wt A \DD \wt A \DD |D|^{-2-n} =
2^m\big(-a_{\a,0}\,a^\a_0+\sum_{l}
a_{\a,-l}\,a^\a_l\big)\tfrac{C_n (n-2)}{n}.
$$
Lemma \ref{traceXD} now proves the result.
\end{proof}

\begin{lemma}
\label{ncint-odd-pdo}
If $\tfrac{1}{2\pi}\Th$ is diophantine, then
for any $P\in \Psi_{1}(\A)$ and $q \in \N$, $q$ odd,
$$
\ncint P |D|^{-(n-q)} = 0.
$$
\end{lemma}
\begin{proof}
There exist $B\in \DD_{1}(\A)$ and $p\in \N_0$ such that
$P= BD^{-2p}+R$ where $R$ is in $OP^{-q-1}$.
As a consequence, $\ncint P |D|^{-(n-q)} = \ncint
B|D|^{-n-2p+q}$. Assume $B= a_r b_r
\DD^{q_{r-1}}a_{r-1}b_{r-1}\cdots \DD^{q_1} a_1 b_1 $ where
$r\in \N$, $a_i \in \A$,
$b_i\in J\A J^{-1}$, $q_i\in \N$. If we prove that
$\ncint B|D|^{-n-2p+q} =0$, then
the general case will follow by linearity. We note
$a_i=:\sum_l a_{i,l}\,U_l$ and
$b_i=:\sum_l b_{i,l} \,U_l$. With the shorthand
$k_{\mu_1,\mu_{q_i}}:=k_{\mu_1}\cdots
k_{\mu_{q_i}}$ and $\ga^{\mu_1,\mu_{q_i}}=\ga^{\mu_1}\cdots
\ga^{\mu_{q_i}}$, we get
$$
\DD^{q_1}  a_1 b_1 U_k \otimes e_j =
\sum_{l_1,l'_1} \,a_{1,l_1}\, b_{1,l'_1}\, U_{l_1}U_k
U_{l'_1} \,(k+l_1+l'_1)_{\mu_1,\mu_{q_1}}
\otimes \ga^{\mu_1,\mu_{q_1}} e_j
$$
which gives, after iteration,
$$
B\, U_k \otimes e_j = \sum_{l,l'} \wt a_{l} \wt b_{l} U_{l_r}
\cdots U_{l_1} U_k
U_{l'_1}\cdots U_{l'_r} \prod_{i=1}^{r-1} (k+\wh l_{i} +\wh
l'_{i})_{\mu^{i}_1,\mu^i_{q_i}} \otimes \ga^{\mu^{r-1}_1,
\mu^{r-1}_{q_{r-1}}}\cdots \ga^{\mu^1_1,\mu^1_{q_1}} e_j
$$
where $\wt a_l : = a_{1,l_1}\cdots a_{r,l_r}$ and
$\wt b_{l'} : = b_{1,l'_1}\cdots
b_{r,l'_r}$. Let's note $Q_\mu(k,l,l'):=\prod_{i=1}^{r-1}
(k+\wh l_{i} +\wh l'_{i})_{\mu^{i}_1,\mu^i_{q_i}}$ and
$\ga^\mu :=\ga^{\mu^{r-1}_1,\mu^{r-1}_{q_{r-1}}}\cdots
\ga^{\mu^1_1,\mu^1_{q_1}}$. Thus,
$$
\ncint B\,|D|^{-n-2p+q} =\underset{s=0}{\Res}\,  {\sum_k}' \,
\sum_{l,l'} \wt a_l
\,\wt b_{l'} \, \tau\big(U_{-k}U_{l_r}\cdots U_{l_1} U_k
U_{l'_1}\cdots U_{l'_r}\big)
\,\tfrac{Q_\mu(k,l,l')}{|k|^{s+2p+n-q}} \,\Tr (\ga^\mu)\, .
$$
Since $U_{l_r}\cdots U_{l_1} U_k = U_k U_{l_r}
\cdots U_{l_1} e^{-i\sum_1^r l_i .\Th
k}$, we get
$$\tau\big(U_{-k}U_{l_r}\cdots U_{l_1} U_k U_{l'_1}\cdots
U_{l'_r}\big)=
\delta_{\sum_1^r l_i+l'_i,0} \,e^{i\phi(l,l')} \,
e^{-i\sum_1^r l_i.\Th k}$$ where $\phi$
is a real valued function. Thus,
\begin{align*}
\ncint B\,|D|^{-n-2p+q} &=\underset{s=0}{\Res}\,
{\sum_k}'\,\sum_{l,l'}
e^{i\phi(l,l')}\,\delta_{\sum_1^r l_i+l'_i,0}\, \wt a_l \,\wt b_{l'}
\,\tfrac{Q_\mu(k,l,l')e^{-i\sum_1^r l_i.\Th k}}{|k|^{s+2p+n-q}}
\Tr (\ga^\mu) \\
&=:\underset{s=0}{\Res}\, f_\mu(s)\Tr (\ga^\mu).
\end{align*}

We decompose $Q_{\mu}(k,l,l')$ as a sum
$\sum_{h=0}^r M_{h,\mu}(l,l') \, Q_{h,\mu}(k)$
where $Q_{h,\mu}$ is a homogeneous polynomial in $(k_1,\cdots,k_n)$
and $M_{h,\mu}(l,l')$ is a polynomial in
$\big((l_1)_1,\cdots,(l_{r})_n,(l'_1)_1,\cdots,(l'_{r})_n \big)$.

Similarly, we decompose $f_{\mu}(s)$ as $\sum_{h=0}^{r}
f_{h,\mu}(s)$. Theorem
\ref{analytic} $(ii)$ entails that $ f_{h,\mu}(s)$ extends
meromorphically to the whole complex plane $\C$ with only one
possible pole for $s+2p+n-q=n+d$ where
$d:=\text{deg } Q_{h,\mu}$. In other words, if $d+q-2p\neq 0$,
$f_{h,\mu}(s)$ is holomorphic at $s=0$. Suppose now $d+q-2p =0$
(note that this implies that $d$ is odd, since
$q$ is odd by hypothesis), then, by Theorem \ref{analytic} $(ii)$
$$
\underset{s=0}{\Res}\ f_{h,\mu}(s) = V \int_{u\in S^{n-1}}
Q_{h,\mu}(u)\,dS(u)
$$
where $V:=\sum_{l,l'\in Z} M_{h,\mu}(l,l')\,e^{i \phi(l,l')}\,
\delta_{\sum_1^r l_i+l'_i,0}\, \wt a_{l} \,\wt b_{l'}$ and
$Z:=\set{l,l' \, :\, \sum_{i=1}^{r} l_i=0}$.
Since $d$ is odd, $Q_{h,\mu}(-u)=-Q_{h,\mu}(u)$ and $\int_{u\in
S^{n-1}}Q_{h,\mu}(u)\,dS(u)=0$. Thus, $\underset{s=0}{\Res}
\ f_{h,\mu}(s)=0$ in any case, which gives the result.
\end{proof}

As we have seen, the crucial point of the preceding lemma is the
decomposition of the numerator of the series $f_\mu(s)$ as polynomials
in $k$. This has been possible because we restricted our
pseudodifferential
operators to $\Psi_1(\A)$.

\bigskip

{\it Proof of Proposition \ref{invariance}.}
The top element follows from Proposition \ref{invariance1} and
according to (\ref{formule}),
\begin{align*}
\ncint |D| ^{-n}= \underset{s=0}{\Res}\
\Tr\big(|D|^{-s-n}\big)=2^m\,\underset{s=0}{\Res}\,
Z_n(s+n)=\tfrac{2^{m+1}\pi^{n/2}}{\Gamma(n/2)}\, .
\end{align*}

For the second equality, we get from Lemmas \ref{tracehD}
and \ref{Res-zeta-n-k}
$$
\underset{s=n-k}{\Res}\,  \zeta_{D_A}(s)= \sum_{p=1}^k
\sum_{r_1,\cdots, r_p =0}^{k-p} h(n-k,r,p)
\ncint \eps^{r_1}(Y)\cdots\eps^{r_p}(Y) |D|^{-(n-k)}.
$$
Corollary  \ref{eps-pdo} and Lemma \ref{ncint-odd-pdo} imply that
$\ncint
\eps^{r_1}(Y)\cdots\eps^{r_p}(Y) |D|^{-(n-k)} =0$, which gives the
result.

Last equality follows from Lemma \ref{traceAA} and Corollary
\ref{res-n-2-A}.
\qed

\section{The spectral action}

Here is the main result:

\begin{theorem}
\label{main}
Consider the $n$-NC-torus $\big(\Coo(\T^n_\Th),\H,\DD\big)$ where
$n\in \N$ and
$\tfrac{1}{2\pi}\Th$ is a $n\times n$ skew-symmetric real diophantine matrix, and a selfadjoint one-form
$A=L(-iA_{\alpha})\otimes
\ga^{\alpha}$. Then, the full spectral action of
$\DD_{A}=\DD +A
+¾\epsilon JAJ^{-1}$ is

\noindent $(i)$ for $n=2$,
$$
\SS(\DD_{A},\Phi,\Lambda)=4\pi\,\Phi_{2} \, \Lambda^{2} +
\mathcal{O}(\Lambda^{-2}),
$$
\noindent $(ii)$ for $n=4$,
$$
\SS(\DD_{A},\Phi,\Lambda)= 8\pi^2\,\Phi_{4} \, \Lambda^{4}
-\tfrac{4\pi^{2}}{3}\
\Phi(0) \,\tau(F_{\mu\nu}F^{\mu\nu})+  \mathcal{O}(\Lambda^{-2}),
$$
$(iii)$ More generally, in
$$
\SS(\DD_{A},\Phi,\Lambda) \, = \,\sum_{k=0}^n \Phi_{n-k}\,
c_{n-k}(A) \,\Lambda^{n-k} +\mathcal{O}(\Lambda^{-1}),
$$
$c_{n-2}(A)=0$, $c_{n-k}(A)=0$ for $k$ odd. In particular, $c_0(A)=0$
when $n$ is odd.
\end{theorem}

\quad

This result (for $n=4$) has also been obtained in \cite{GIVas} using
the heat kernel method. It is however interesting to get the result
via direct computations of
(\ref{formuleaction}) since it shows how this formula is efficient.
As we will see, the computation of all the noncommutative integrals
require a lot of technical steps.
One of the main points, namely to isolate where the Diophantine
condition on $\Th$ is assumed, is outlined here.

\begin{remark}
Note that all terms must be gauge invariants, namely, according
to (\ref{gaugetransform}), invariant by
$A_{\alpha}\longrightarrow \gamma_{u}(A_{\alpha})=
uA_{\alpha}u^{*}+u\delta_{\alpha}(u^{*})$. A particular case is
$u=U_{k}$ where
$U_{k}\delta_{\alpha}(U_{k}^{*})=-ik_{\alpha} U_{0}$.

In the same way, note that there is no contradiction with
the commutative case where, for any selfadjoint one-form
$A$, $\DD_{A}=\DD$ (so $A$ is equivalent to 0!), since
we assume in Theorem \ref{main} that
$\Th$ is diophantine, so $\A$ cannot be commutative.
\end{remark}

\begin{conjecture}
The constant term of the
spectral action of $\DD_{A}$ on the noncommutative n-torus is
proportional to the constant term of the spectral action
of $\DD+A$ on the commutative n-torus.
\end{conjecture}

\begin{remark}
The appearance of a Diophantine condition for $\Th$ has been
characterized in dimension 2 by Connes \cite[Prop. 49]{NCDG} where in
this case, $\Th=\th\genfrac{(}{)}{0pt}{1}{\,0 \quad1}{-1\,\,\,\, 0 }$
with $\th \in \R$. In fact, the Hochschild cohomology
$H(\A_{\Th},{\A_{\Th}}^*)$ satisfies dim
$H^j(\A_{\Th},{\A_{\Th}}^*)=2$ (or $1$) for $j=1$ (or $j=2$) if and
only if the irrational number $\th$ satisfies a Diophantine condition
like $\vert 1-e^{i2\pi n \th} \vert^{-1} =\mathcal{O}(n^k)$ for some
$k$.

Recall that when the matrix $\Th$ is quite irrational (see \cite[Cor.
2.12]{Polaris}), then the C$^*$-algebra generated by $\A_{\Th}$ is
simple.
\end{remark}

\begin{remark}
It is possible to generalize above theorem to the case $\DD=-i\,{g^{\mu}}_{\nu} \, \delta_\mu \otimes \ga^\nu$ instead of \eqref{defDirac} when $g$ is a positive definite constant matrix. The formulae in Theorem \ref{main}  are still valid. 
\end{remark}

\subsection{Computations of $\ncint$}
In order to get this theorem, let us prove a few technical lemmas.

We suppose
from now on that $\Th$ is a skew-symmetric matrix in
$\mathcal{M}_n(\R)$. No other
hypothesis is assumed for $\Th$, except when it is explicitly stated.

When $A$ is a selfadjoint one-form, we define for $n\in N$, $q\in \N$,
$2\leq q \leq n$ and
$\sigma\in \{-,+\}^q$
\begin{align*}
\mathbb{A}^{+}&:=A\DD D^{-2},\\
\mathbb{A}^{-}&:= \epsilon JAJ^{-1} \DD D^{-2},\\
\mathbb{A}^{\sigma}&:=\mathbb{A}^{\sigma_q}\cdots
\mathbb{A}^{\sigma_1} .
\end{align*}

\begin{lemma}
\label{ncadmoins1}
We have for any $q\in \N$,
$$
\ncint (\wt A D^{-1})^q = \ncint (\wt A \DD D^{-2})^q =
\sum_{\sigma\in \set{+,-}^{q}}
\ncint \mathbb{A}^{\sigma}.
$$
\end{lemma}
\begin{proof} Since $P_0 \in OP^{-\infty}$, $D^{-1} = \DD D^{-2} \mod
OP^{-\infty}$
and $\ncint (\wt A D^{-1})^q = \ncint (\wt A \DD D^{-2})^q$.
\end{proof}

\begin{lemma}
\label{symetrie}
Let $A$ be a selfadjoint one-form, $n\in \N$ and $q\in \N$ with
$2\leq q \leq n$ and
$\sigma\in \{-,+\}^q$. Then
$$
\ncint \mathbb{A}^{\sigma} =\ncint \mathbb{A}^{-\sigma}.
$$
\end{lemma}
\begin{proof}Let us first check that $JP_0 = P_0 J$. Since $\DD J =
\eps J\DD$, we get $\DD J P_0=0$ so $JP_0 = P_0 J P_0$. Since $J$ is
an antiunitary operator,
we get $P_0 J = P_0 J P_0$  and finally, $P_0 J = J P_0$.
As a consequence, we get $JD^2 = D^2 J$,
$J \DD D^{-2}=\eps \DD D^{-2} J$, $J\mathbb{A}^{+}J^{-1} =
\mathbb{A}^{-}$ and
$J\mathbb{A}^{-}J^{-1} = \mathbb{A}^{+}.$

In summary, $J\mathbb{A}^{\sigma_i}J^{-1} = \mathbb{A}^{-\sigma_i}$.

The trace property of $\ncint$ now gives
\begin{align*}
\ncint \mathbb{A}^{\sigma} &= \ncint \mathbb{A}^{\sigma_q}\cdots
\mathbb{A}^{\sigma_1} = \ncint J
\mathbb{A}^{\sigma_q}J^{-1}\cdots J\mathbb{A}^{\sigma_1}J^{-1} \ncint
\mathbb{A}^{-\sigma_q} \cdots
\mathbb{A}^{-\sigma_1}=\ncint \mathbb{A}^{-\sigma}.
\tag*{\qed}
\end{align*}
\hideqed
\end{proof}
\begin{definition}

In \cite{CC1} has been introduced the vanishing tadpole hypothesis:
\begin{align}
    \label{vanishtad}
    \ncint A D^{-1}=0, \text{ for all } A\in \Omega_{\DD}^{1}(\A).
\end{align}
\end{definition}
By the following lemma, this condition is satisfied for the
noncommutative torus, a fact more or less already known within the
noncommutative community \cite{Walter}.

\begin{lemma}
\label{tadpole}
Let $n\in \N$, $A= L(-iA_{\a})\otimes \gamma^{\a}=-i\sum_{l\in
\Z^n}a_{\alpha,l} \,
U_{l}\otimes \ga^{\alpha}$, $A_{\a}\in \A_{\Th}$,
$\set{a_{\alpha,l}}_{l}\in \SS
(\Z^n)$, be a hermitian one-form.
    Then, \\
    (i) $\ncint A^p D^{-q}  =  \ncint (\epsilon JAJ^{-1})^p
D^{-q}=0$ for $p\geq0$ and $1\leq q < n$ (case $p=q=1$ is tadpole
hypothesis).\\
    (ii) If $\tfrac{1}{2\pi} \Th$ is diophantine, then $\ncint
BD^{-q}=0$ for $1\leq q < n$ and any $B$ in the algebra generated by
$\A$, $[\DD,\A]$, $J\A J^{-1}$ and $J[\DD,\A]J^{-1}$. 
\end{lemma}

\begin{proof}
$(i)$ Let us compute $$\ncint A^p(\epsilon
JAJ^{-1})^{p'}D^{-q}.$$ 
With $A=L(-i A_\a)\otimes \ga^\a$ and
$\epsilon J A
J^{-1}=R(i A_\a)\otimes \ga^\a$, we get
$$
A^p=L(-i A_{\a_1})\cdots L(-i A_{\a_p}) \otimes \ga^{\a_1}\cdots
\ga^{\a_p}
$$
and
$$
(\epsilon J A J^{-1})^{p'}=R(i A_{\a'_1})\cdots R(i A_{\a'_{p'}})
\otimes
\ga^{\a'_1}\cdots \ga^{\a'_{p'}}.
$$
We note $\wt a_{\a,l}:= a_{\a_1,l_1}\cdots a_{\a_p,l_p}$. Since
$$
L(-i A_{\a_1})\cdots L(-i A_{\a_p})R(i A_{\a'_1})\cdots R(i
A_{\a'_{p'}}) U_k= (-i)^p
\, i^{p'} \sum_{l,l'} \wt a_{\a,l} \, \wt a_{\a',l'} \, U_{l_1}\cdots
U_{l_p} U_k
U_{l'_{p'}}\cdots U_{l'_1},
$$
and
$$
U_{l_1}\cdots U_{l_p} U_k= U_k U_{l_1}\cdots U_{l_p} \, e^{-i(\sum_i
l_i).\Th k},
$$
we get, with
\begin{align*}
&U_{l,l'}:=U_{l_1}\cdots U_{l_p}U_{l'_{p'}}\cdots U_{l'_1},\\
&g_{\mu,\a,\a'}(s,k,l,l'):= e^{ik. \Th \sum_j l_j} \,
\tfrac{k_{\mu_{1}}\ldots
k_{\mu_{q}}}{\vert k \vert^{s+2q}} \, \wt a_{\a,l} \,  \wt a_{\a',l'},
\\
&\ga^{\a,\a',\mu}:=\ga^{\a_1}\cdots\ga^{\a_{p}}\ga^{\a'_1}\cdots
\ga^{\a'_{p'}} \ga^{\mu_1}\cdots \ga^{\mu_q},
\end{align*}
$$
A^p(\epsilon JAJ^{-1})^{p'}D^{-q}|D|^{-s} U_k\otimes e_i\sim_c (-i)^p
\, i^{p'}
\sum_{l,l'} g_{\mu,\a,\a'}(s,k,l,l') \, U_k U_{l,l'} \otimes
\ga^{\a,\a',\mu} e_i.
$$
Thus, $\ncint A^p(\epsilon JAJ^{-1})^{p'}D^{-q} =
\underset{s=0}{\Res}\ f(s)$
where
\begin{align*}
f(s):&=\Tr\big(A^p(\epsilon JAJ^{-1})^{p'}D^{-q}|D|^{-s}\big)\\
&\sim_c (-i)^p \, i^{p'} {\sum_{k\in\Z^{n}}}' \langle U_{k}\otimes
e_{i},\sum_{l,l'}
g_{\mu,\a,\a'}(s,k,l,l')  U_k
U_{l,l'} \otimes \ga^{\a,\a',\mu} e_i \rangle\\
&\sim_c (-i)^p \, i^{p'} {\sum_{k\in\Z^{n}}}' \, \,\tau\big(
\sum_{l,l'}
g_{\mu,\a,\a'}(s,k,l,l')
U_{l,l'} \big) \Tr(\ga^{\mu,\a,\a'})\\
&\sim_c (-i)^p \, i^{p'} {\sum_{k\in\Z^{n}}}' \sum_{l,l'}
g_{\mu,\a,\a'}(s,k,l,l') \, \tau
\big(U_{l,l'} \big) \Tr(\ga^{\mu,\a,\a'}).
\end{align*}
It is straightforward to check that the series
${\sum}'_{k,l,l'} g_{\mu,\a,\a'}(s,k,l,l') \, \tau\big( U_{l,l'}
\big)$
is absolutely summable if $\Re(s)>R$ for a $R>0$. Thus, we can
exchange the summation
on $k$ and $l,l'$, which gives
$$
f(s)\sim_c (-i)^p \, i^{p'} \sum_{l,l'}  {\sum_{k\in\Z^{n}}}'
g_{\mu,\a,\a'}(s,k,l,l') \,
\tau \big( U_{l,l'} \big) \Tr(\ga^{\mu,\a,\a'}).
$$
If we suppose now that $p'=0$, we see that,
$$
f(s)\sim_c  (-i)^p \sum_{l}  {\sum_{k\in\Z^{n}}}'   \,
\tfrac{k_{\mu_{1}}\ldots
k_{\mu_{q}}}{\vert k \vert^{s+2q}} \, \wt a_{\a,l} \,
\delta_{\sum_{i=1}^p l_i,0}
\Tr(\ga^{\mu,\a,\a'})
$$
which is, by Proposition \ref{calculres}, analytic at 0. In
particular, for
$p=q=1$, we see that $\ncint A D^{-1} =0$, i.e. the
vanishing tadpole
hypothesis is satisfied. Similarly, if we suppose $p=0$,
we get
$$
f(s)\sim_c  (-i)^{p'} \sum_{l'}  {\sum_{k\in\Z^{n}}}'   \,
\tfrac{k_{\mu_{1}}\ldots
k_{\mu_{q}}}{\vert k \vert^{s+2q}} \, \wt a_{\a,l'} \,
\delta_{\sum_{i=1}^{p'} {l'}_i,0}
\Tr(\ga^{\mu,\a,\a'})
$$
which is holomorphic at 0.

$(ii)$ Adapting the proof of Lemma \ref{ncint-odd-pdo} to our setting
(taking $q_i=0$, and 
adding gamma matrices components), we see that
$$
\ncint B\,D^{-q} =\underset{s=0}{\Res}\, {\sum_k}'\,\sum_{l,l'}
e^{i\phi(l,l')}\,\delta_{\sum_1^r l_i+l'_i,0}\, \wt a_{\a,l} \,\wt
b_{\beta,l'}
\,\tfrac{k_{\mu_1}\cdots k_{\mu_q}\,e^{-i\sum_1^r l_i.\Th
k}}{|k|^{s+2q}}
\Tr (\ga^{(\mu,\a,\beta)})
$$
where $\ga^{(\mu,\a,\beta)}$ is a complicated product of gamma matrices.
By Theorem \ref{analytic} $(ii)$, since we suppose here that
$\tfrac{1}{2\pi} \Th$ is
diophantine, this residue is 0.
\end{proof}

\subsubsection{Even dimensional case}
\begin{corollary}

Same hypothesis as in Lemma \ref{tadpole}.

(i) Case $n=2$:
    \begin{align*}
\ncint A^q D^{-q}= -\delta_{q,2}\,4\pi  \,\tau\big(A_{\a} A^{\a}
\big) \,.
\end{align*}
    (ii) Case $n=4$: with the shorthand
$\delta_{\mu_1,\ldots,\mu_4}:=
\delta_{\mu_1\mu_2}\delta_{\mu_3\mu_4}+\delta_{\mu_1\mu_3}\delta_{\mu_2\mu_4}
+\delta_{\mu_1\mu_4}\delta_{\mu_2\mu_3}$,
\begin{align*}
\ncint A^q D^{-q}= \delta_{q,4}\,\tfrac{\pi^2}{12}
\tau\big(A_{\a_1}\cdots A_{\a_4}
\big)\Tr(\ga^{\a_1}\cdots\ga^{\a_{4}}\ga^{\mu_1}\cdots\ga^{\mu_4})
\delta_{\mu_1,\ldots,\mu_4} \,.
\end{align*}
\end{corollary}

\begin{proof}
$(i,ii)$ The same computation as in Lemma \ref{tadpole} $(i)$ (with
$p'=0$, $p=q=n$) gives
\begin{align*}
\ncint A^n D^{-n}=\underset{s=0}{\Res}(-i)^{n}
\big({\sum_{k\in\Z^{n}}}'\tfrac{k_{\mu_{1}}\ldots k_{\mu_{n}}}{\vert k
    \vert^{s+2n}}\big) \, \tau\big(\sum_{l\in (\Z^n)^n}
\wt a_{\a,l} U_{l_1}\cdots U_{l_{n}}
    \big) \,
\Tr(\ga^{\a_1}\cdots\ga^{\a_{n}}\ga^{\mu_1}\cdots\ga^{\mu_n})
\end{align*}
and the result follows from Proposition \ref{calculres}.
\end{proof}

We will use few notations:

 If $n\in \N$, $q\geq 2$, $l:=(l_1,\cdots,l_{q-1})\in
(\Z^n)^{q-1}$, $\a:=(\a_1,\cdots,\a_q)\in \{1,\cdots,n\}^q$, $k\in
\Z^n \backslash
\{0\}$, $\sigma\in \{-,+\}^q$, $(a_i)_{1\leq i\leq n}\in
(\mathcal{S}(\Z^n))^n$,
\begin{align*}
&l_q:=-\sum_{1\leq j\leq q-1} l_j \, , \quad
\lambda_\sigma:=(-i)^q\prod_{j=1\ldots
q}\sigma_j \, , \quad
\wt a_{\a,l}:= a_{\a_1,l_1}\ldots a_{\a_q,l_q}\,,\\
& \phi_\sigma(k,l):=\sum_{1\leq j\leq q-1} (\sigma_j-\sigma_q)\,
k.\Th l_j +
\sum_{2\leq
j\leq q-1} \sigma_j \, (l_1+\ldots +l_{j-1}).\Th l_j \, ,\\
& g_\mu(s,k,l):=\tfrac{k_{\mu_1}(k+l_1)_{\mu_2}\ldots (k+l_1+\ldots
+l_{q-1})_{\mu_q}}{|k|^{s+2}|k+l_1|^2\ldots|k+l_1+\ldots+l_{q-1}|^2}
\, ,
\end{align*}
with the convention $\sum_{2\leq j\leq q-1} = 0$ when $q=2$, and
$g_\mu(s,k,l)=0$
whenever $\wh l_i=-k$ for a $1\leq i\leq q-1$.

\begin{lemma}\label{formegenerale}
Let $A= L(-iA_{\a})\otimes \gamma^{\a}=-i\sum_{l\in \Z^n}a_{\alpha,l}
\, U_{l}\otimes
\ga^{\alpha}$ where $A_{\a}=-A_{\a}^*\in \A_{\Th}$ and
$\set{a_{\alpha,l}}_{l}\in \SS
(\Z^n)$, with $n\in \N$, be a hermitian one-form, and let $2\leq q
\leq n$, $\sigma\in
\{-,+\}^q$.

Then, $\ncint \mathbb{A}^{\sigma}= \underset{s=0}{\Res}\ f(s)$ where
$$
f(s):=\sum_{l\in (\Z^{n})^{q-1}} {\sum_{k\in\Z^n}}' \,
\lambda_\sigma\  e^{\tfrac i2
\phi_\sigma(k,l)}\ g_\mu(s,k,l)\ \wt a_{\a,l}\
\Tr(\ga^{\a_q}\ga^{\mu_q}\cdots\ga^{\a_1}\ga^{\mu_1}).
$$
\end{lemma}
\begin{proof}
By definition, $\ncint \mathbb{A}^{\sigma}= \underset{s=0}{\Res}\
f(s)$ where
$$
\Tr(\mathbb{A}^{\sigma_q}\cdots \mathbb{A}^{\sigma_1}
|D|^{-s})\sim_c {\sum_{k\in\Z^n}}' \langle U_k\otimes
e^i ,|k|^{-s}\,\mathbb{A}^{\sigma_q}\cdots \mathbb{A}^{\sigma_1}
U_k\otimes e_i\rangle =: f(s).
$$
Let $r\in \Z^n$ and $v\in \C^{2^m}$. Since $A=L(-i A_\a)\otimes
\ga^\a$, and $\epsilon
JAJ^{-1}=R(i A_\a)\otimes \ga^{\a}$, we get
\begin{align*}
\mathbb{A}^{+}U_r\otimes v &= A\DD D^{-2} U_r\otimes
v=A\tfrac{r_\mu}{|r|^2+\delta_{r,0}} U_r \otimes
\ga^{\mu}v=-i\tfrac{r_\mu}{|r|^2+\delta_{r,0}} A_\a U_r \otimes
\ga^\a\ga^{\mu}v\,
,  \\
\mathbb{A}^{-}U_r\otimes v &=\epsilon JAJ^{-1}\DD D^{-2} U_r\otimes
v=\epsilon
JAJ^{-1}\tfrac{r_\mu}{|r|^2+\delta_{r,0}} U_r \otimes
\ga^{\mu}v=i\tfrac{r_\mu}{|r|^2+\delta_{r,0}}U_r A_\a
\otimes \ga^\a\ga^{\mu}v.
\end{align*}
With $U_l U_r=e^{\tfrac i2 r.\Th l} U_{r+l}$ and $U_r U_l=e^{-\tfrac
i2 r.\Th l}
U_{r+l}$, we obtain, for any $1\leq j\leq q$,
$$
\mathbb{A}^{\sigma_j}U_r\otimes v=\sum_{l\in \Z^n} (-\sigma_j)
\, i\, e^{\sigma_j\, \tfrac i2
r.\Th l}\, \tfrac{r_\mu}{|r|^2+\delta_{r,0}}\, a_{\a,l} \, U_{r+l}
\otimes
\ga^{\a}\ga^{\mu}v.
$$
We now apply $q$ times this formula to get
$$
|k|^{-s} \mathbb{A}^{\sigma_q}\cdots \mathbb{A}^{\sigma_1}
U_k\otimes e_i = \lambda_\sigma \sum_{l\in
(\Z^{n})^q} e^{\tfrac i2 \phi_\sigma(k,l)}\ g_\mu(s,k,l)\ \wt
a_{\a,l}\ U_{k+\sum_j
l_j} \otimes \ga^{\a_q}\ga^{\mu_q}\cdots\ga^{\a_1}\ga^{\mu_1}e_i
$$
with
\begin{align*}
\phi_\sigma(k,l)&:=\sigma_1\,  k.\Th l_1+\sigma_2 \, (k+l_1).\Th
l_2+\ldots
+\sigma_q \, (k+l_1+\ldots+l_{q-1}).\Th l_q.
\end{align*}
Thus,
\begin{align*}
f(s) &= {\sum_{k\in\Z^n}}' \, \tau \big(\lambda_\sigma
\sum_{l\in (\Z^{n})^q} e^{\tfrac i2
\phi_\sigma(k,l)}\ g_\mu(s,k,l)\ \wt a_{\a,l}\ U_{\sum_j
l_j}e^{\tfrac i2 k.\Th \sum_j
l_j}\big) \Tr(\ga^{\a_q}\ga^{\mu_q}\cdots\ga^{\a_1}\ga^{\mu_1})\\
&={\sum_{k\in\Z^n}}' \,  \lambda_\sigma \sum_{l\in (\Z^{n})^q}
e^{\tfrac i2
\phi_\sigma(k,l)}\ g_\mu(s,k,l)\ \wt a_{\a,l}\ \delta(\sum_j l_j)
\Tr(\ga^{\a_q}\ga^{\mu_q}\cdots\ga^{\a_1}\ga^{\mu_1})\\
&= {\sum_{k\in\Z^n}}' \,  \lambda_\sigma \sum_{l\in (\Z^{n})^{q-1}}
e^{\tfrac i2
\phi_\sigma(k,l)}\ g_\mu(s,k,l)\ \wt a_{\a,l}\
\Tr(\ga^{\a_q}\ga^{\mu_q}\cdots\ga^{\a_1}\ga^{\mu_1})
\end{align*}
where in the last sum $l_q$ is fixed to $-\sum_{1\leq j\leq q-1} l_j$
and thus,
$$
\phi_\sigma(k,l)=\sum_{1\leq j\leq q-1} (\sigma_j-\sigma_q) \, k.\Th
l_j + \sum_{2\leq
j\leq q-1} \sigma_j \, (l_1+\ldots +l_{j-1}).\Th l_j.
$$
By Lemma \ref{abs-som}, there exists a $R>0$ such that for any $s\in
\C$ with
$\Re(s)>R$, the family
$$\big(e^{\tfrac i2 \phi_\sigma(k,l)}\
g_\mu(s,k,l)\ \wt a_{\a,l}\big)_{(k,l)\in (\Z^n \setminus
\set{0})\times (\Z^{n})^{q-1}}$$
is absolutely summable as a linear combination of families of the
type considered in that lemma. As a consequence, we can
exchange the summations on $k$ and $l$, which gives the result.
\end{proof}

In the following, we will use the shorthand
$$
c:=\tfrac{4\pi^{2}}{3}.
$$

\begin{lemma}\label{Termaterm} Suppose $n=4$. Then, with the same
hypothesis of Lemma
\ref{formegenerale},
\begin{align*}
\hspace{-2cm}\text{(i)} \quad \quad &
\tfrac 12 \ncint (\mathbb A^+)^2= \tfrac 12 \ncint (\mathbb A^-)^2= c
\,
\sum_{l\in\Z^4} \, a_{\alpha_{1},l} \, a_{\alpha_{2},-l} \,
\big(l^{\alpha_{1}}l^{\alpha_{2}}
- \delta^{\alpha_{1}\alpha_{2}} \vert l \vert^2\big).\\
\hspace{-2cm}\text{(ii)} \quad  \quad & \hspace{-0.45cm}
-\tfrac 13\ \ncint (\mathbb A^+)^3=-\tfrac 13 \ncint (\mathbb
A^-)^3=4c
\,\sum_{l_i \in \Z^4}
a_{\a_3,-l_1-l_2}\,a^{\a_1}_{l_2}\,a_{\a_1,l_1}\ \sin \tfrac{l_1.\Th
l_2}{2}\,l_1^{\a_3}.\\
\hspace{-1cm}\text{(iii)} \quad \quad &
\tfrac 14 \ncint (\mathbb A^+)^4=\tfrac 14 \ncint (\mathbb A^-)^4=
2c\,
\sum_{l_i \in \Z^4} a_{\alpha_{1},-l_1-l_2-l_3}\,
a_{{\alpha_{2}},l_3} \, a^{\alpha_{1}}_{l_2}
\, a^{\alpha_{2}}_{l_1} \sin \tfrac{l_1 .\Th (l_2+l_3)}{2}\, \sin
\tfrac{ l_2 .\Th
l_3}{2}.
\end{align*}
\vspace{-0.3cm}

(iv) Suppose $\tfrac {1}{2\pi}\Th$ diophantine. Then the
crossed terms in
$\ncint (\mathbb A^+ + \mathbb A^-)^q$ vanish:  if $C$ is the set of
all $\sigma\in \{-,+\}^q$ with $2\leq q\leq 4$, such that there exist
$i,j$ satisfying
$\sigma_i \neq \sigma_j$, we have $ \sum_{\sigma \in C} \, \ncint
\mathbb{A}^{\sigma} =0.
$
\end{lemma}
\begin{proof}
$(i)$ Lemma \ref{formegenerale} entails that
$\ncint \mathbb{A}^{++}=
\underset{s=0}{\Res} \sum_{l\in \Z^n} - f(s,l)$ where
$$
f(s,l):=  {\sum_{k\in\Z^n}}'
\tfrac{k_{\mu_1}(k+l)_{\mu_2}}{|k|^{s+2}|k+l|^2}\ \wt
a_{\a,l}\ \Tr(\ga^{\a_2}\ga^{\mu_2}\ga^{\a_1}\ga^{\mu_1}) \, \text{
and
} \, \wt a_{\a,l}:=a_{\alpha_1,l}\, a_{\alpha_2,-l} \,.
$$
We will now reduce the computation of the residue of an expression
involving terms
like $\vert k+l\vert^{2}$ in the denominator to the computation of
residues of zeta
functions. To proceed, we use (\ref{trick-0}) into an expression like
the one
appearing in $f(s,l)$. We see that the last term on the righthandside
yields a
$Z_{n}(s)$ while the first one is less divergent by one power of $k$.
If this is not
enough, we repeat this operation for the new factor of $\vert
k+l\vert^{2}$ in the
denominator. For $f(s,l)$, which is quadratically divergent at $s=0$,
we have to repeat
this operation three times before ending with a convergent result.
All the remaining
terms are expressible in terms of $Z_{n}$ functions. We get, using
three times
(\ref{trick-0}),
\begin{equation}\label{trick-1}
\tfrac{1}{|k+l|^2}=\tfrac{1}{|k|^2} - \tfrac{2k.l+|l|^2}{|k|^4} +
\tfrac{(2k.l+|l|^2)^2}{|k|^6} - \tfrac{(2k.l+|l|^2)^3}{|k|^6|k+l|^2}
\, .
\end{equation}
Let us define
$$
f_{\a,\mu}(s,l):={\sum_{k\in\Z^n}}'
\tfrac{k_{\mu_1}(k+l)_{\mu_2}}{|k|^{s+2}|k+l|^2}\ \wt a_{\a,l}
$$
so that $f(s,l)=
f_{\a,\mu}(s,l)\Tr(\ga^{\a_2}\ga^{\mu_2}\ga^{\a_1}\ga^{\mu_1})$.
Equation (\ref{trick-1}) gives
$$
f_{\a,\mu}(s,l) = f_1(s,l) - f_2(s,l) + f_3(s,l) - r(s,l)
$$
with obvious identifications. Note that the function
$$
r(s,l) ={\sum_{k\in\Z^n}}'
\tfrac{k_{\mu_1}(k+l)_{\mu_2}(2kl+|l|^2)^3}{|k|^{s+8}|k+l|^2}\ \wt
a_{\a,l}
$$
is a linear combination of functions of the type $H(s,l)$ satisfying
the hypothesis of
Corollary \ref{res-somH}. Thus, $r(s,l)$ satisfies (H1) and with the
previously seen
equivalence relation modulo functions satisfying this hypothesis we
get
$f_{\a,\mu}(s,l) \sim f_1(s,l) - f_2(s,l) + f_3(s,l)$.

Let's now compute $f_1(s,l)$.
$$
f_1(s,l) = {\sum_{k\in\Z^n}}'
\tfrac{k_{\mu_1}(k+l)_{\mu_2}}{|k|^{s+4}}\ \wt a_{\a,l}
= \wt a_{\a,l}\, {\sum_{k\in\Z^n}}'
\tfrac{k_{\mu_1}k_{\mu_2}}{|k|^{s+4}} + 0.
$$
Proposition \ref{res-int} entails that $s\mapsto {\sum_{k\in\Z^n}}'
\tfrac{k_{\mu_1}k_{\mu_2}}{|k|^{s+4}}$ is holomorphic at 0. Thus,
$f_1(s,l)$ satisfies
(H1), and $f_{\a,\mu}(s,l) \sim - f_2(s,l) + f_3(s,l)$.

Let's now compute $f_2(s,l)$ modulo (H1). We get, using several times
Proposition
\ref{res-int},
\begin{align*}
f_2(s,l)&={\sum_{k\in\Z^n}}'
\tfrac{k_{\mu_1}(k+l)_{\mu_2}(2kl+|l|^2)}{|k|^{s+6}}\ \wt
a_{\a,l} = {\sum_{k\in\Z^n}}' \tfrac{(2kl) k_{\mu_1} k_{\mu_2}+(2kl)
k_{\mu_1}
l_{\mu_2}+|l|^2 k_{\mu_1}
k_{\mu_2}+l_{\mu_2}|l|^2k_{\mu_1}}{|k|^{s+6}}\ \wt
a_{\a,l}\\
&\sim 0 + {\sum_{k\in\Z^n}}' \tfrac{(2kl) k_{\mu_1}
l_{\mu_2}}{|k|^{s+6}}\ \wt
a_{\a,l}+{\sum_{k\in\Z^n}}' \tfrac{|l|^2 k_{\mu_1}
k_{\mu_2}}{|k|^{s+6}}\ \wt
a_{\a,l}+0\, .
\end{align*}
Recall that ${\sum}'_{k\in \Z^n} \tfrac{k_ik_j}{|k|^{s+6}} =
\tfrac{\delta_{ij}}{n}
Z_n(s+4)$. Thus,
$$
f_2(s,l) \sim 2 l^i l_{\mu_2} \wt a_{\a,l} \tfrac{\delta_{i\mu_1}}{n}
Z_n(s+4)+
|l|^2\,\wt a_{\a,l}\, \tfrac{\delta_{\mu_1\mu_2}}{n} Z_n(s+4).
$$
Finally, let us compute $f_3(s,l)$ modulo (H1) following the same
principles:
\begin{align*}
f_3(s,l)&={\sum_{k\in\Z^n}}'
\tfrac{k_{\mu_1}(k+l)_{\mu_2}(2kl+|l|^2)^2}{|k|^{s+8}}\
\wt a_{\a,l}\\&= {\sum_{k\in\Z^n}}' \tfrac{(2kl)^2 k_{\mu_1}
k_{\mu_2} + (2kl)^2
k_{\mu_1} l_{\mu_2} + |l|^4 k_{\mu_1}k_{\mu_2} + |l|^4
k_{\mu_1}l_{\mu_2} +
(4kl)|l|^2k_{\mu_1}k_{\mu_2}+(4kl)|l|^2k_{\mu_1}l_{\mu_2}}{|k|^{s+8}}\
\wt
a_{\a,l}\\&\sim 4 l^{i}l^{j} \, {\sum_{k\in\Z^n}}'\tfrac{k_i k_j
k_{\mu_1}
k_{\mu_2}}{|k|^{s+8}} \wt a_{\a,l}+ 0.
\end{align*}
In conclusion,
$$
f_{\a,\mu}(s,l) \sim - \tfrac 14 (2 l_{\mu_1} l_{\mu_2} +|l|^2\,\,
\delta_{\mu_1\mu_2}) \wt a_{\a,l} Z_n(s+4) + 4 l^{i}l^{j}\,\wt
a_{\a,l}
{\sum_{k\in\Z^n}}'\tfrac{k_i k_j k_{\mu_1} k_{\mu_2}}{|k|^{s+8}}=:
g_{\a,\mu}(s,l).
$$
Proposition (\ref{res-int}) entails that $Z_n(s+4)$ and $s\mapsto
{\sum_{k\in\Z^n}}'\tfrac{k_i k_j k_{\mu_1} k_{\mu_2}}{|k|^{s+8}}$
extend
holomorphically in a punctured open disk centered at 0. Thus,
$g_{\a,\mu}(s,l)$
satisfies (H2) and we can apply Lemma \ref{res-som} to get
$$
-\ncint (\mathbb{A}^+)^2= \underset{s=0}{\Res} \sum_{l\in \Z^n}
f(s,l)= \sum_{l\in
\Z^n}\underset{s=0}{\Res}\ g_{\a,\mu}(s,l)
\Tr(\ga^{\a_2}\ga^{\mu_2}\ga^{\a_1}\ga^{\mu_1})=:\sum_{l\in
\Z^n}\underset{s=0}{\Res}\
g(s,l).
$$
The problem is now reduced to the computation of
$\underset{s=0}{\Res}\ g(s,l)$.
Recall that Res$_{s=0} \, Z_{4}(s+4)=2\pi^2$ by (\ref{formule}) or
(\ref{formule1}),
and
$$
{\rm Res}_{s=0}\,{\sum_{k\in\Z^{n}}}'\,\tfrac{k_{i}k_{j}k_{l}k_{m}}
{\vert
k\vert^{s+8}}=(\delta_{ij}\delta_{lm}+\delta_{il}\delta_{jm}
+\delta_{im}\delta_{jl})\,\tfrac{\pi^2}{12}.
$$
Thus,
$$
\underset{s=0}{\Res}\ g_{\a,\mu}(s,l) = -\tfrac{\pi^2}{3}\,\wt
a_{\a,l}\,(l_{\mu_1}l_{\mu_2}+\tfrac 12 |l|^2\delta_{\mu_1\mu_2}).
$$
We will use
\begin{align}
    \label{Wick}
\Tr(\ga^{\mu_{1}}\cdots\ga^{\mu_{2j}})=\Tr(1)\, \sum_{\text{all
pairings of }\set{1\cdots 2j}} s(P) \, \delta_{\mu_{P_{1}}\mu_{P_{2}}}
 \delta_{\mu_{P_{3}}\mu_{P_{4}}}\cdots
 \delta_{\mu_{P_{2j-1}}\mu_{P_{2j}}}
\end{align}
where $s(P)$ is the signature of the permutation $P$ when
$P_{2m-1}<P_{2m}$ for $1 \leq m \leq n$. This gives
\begin{align}
\label{Wick1}
\Tr(\ga^{\alpha_{2}}\ga^{\mu_2}\ga^{\alpha_{1}}\ga^{\mu_1}) = 2^{m}
(\delta^{\alpha_{2}\mu_2}\delta^{\alpha_{1}
\mu_1}-\delta^{\alpha_{1}\alpha_{2}}\delta^{\mu_2\mu_1}+
\delta^{\alpha_{2}\mu_1}\delta^{\mu_2\alpha_{1}}).
\end{align}
Thus,
\begin{align*}
\underset{s=0}{\Res}\ g(s,l)&= -c\,\wt
a_{\a,l}\,(l_{\mu_1}l_{\mu_2}+\tfrac 12
|l|^2\delta_{\mu_1\mu_2}) (\delta^{\alpha_{2}\mu_2}\delta^{\alpha_{1}
\mu_1}-\delta^{\alpha_{1}\alpha_{2}}\delta^{\mu_2\mu_1}+
\delta^{\alpha_{2}\mu_1}\delta^{\mu_2\alpha_{1}})\\
&=-2c\,\wt a_{\a,l}\,(l^{\a_1}l^{\a_2}-\delta^{\a_1\a_2} |l|^2).
\end{align*}
Finally,
$$
\tfrac 12 \ncint (\mathbb{A}^+)^2= \tfrac 12 \ncint (\mathbb{A}^-)^2=
c \,
\sum_{l\in\Z^n} \, a_{\alpha_{1},l} \, a_{\alpha_{2},-l} \,
\big(l^{\alpha_{1}}l^{\alpha_{2}} - \delta^{\alpha_{1}\alpha_{2}}
\vert l \vert
^2\big).
$$

$(ii)$ Lemma \ref{formegenerale} entails that $\ncint
\mathbb{A}^{+++}=
\underset{s=0} {\Res} \sum_{(l_1,l_2)\in (\Z^n)^2} f(s,l)$ where
\begin{align*}
f(s,l)&:=  {\sum_{k\in\Z^n}}' i\,e^{\tfrac i2 l_1\Th
l_2}\,\tfrac{k_{\mu_1}(k+l_1)_{\mu_2}(k+\wh
l_2)_{\mu_3}}{|k|^{s+2}|k+l_1|^2|k+\wh
l_2|^2}\ \wt a_{\a,l}
\Tr(\ga^{\a_3}\ga^{\mu_3}\ga^{\a_2}\ga^{\mu_2}\ga^{\a_1}\ga^{\mu_1})\\
&=:f_{\a,\mu}(s,l)\Tr(\ga^{\a_3}\ga^{\mu_3}\ga^{\a_2}\ga^{\mu_2}\ga^{\a_1}\ga^{\mu_1}),
\end{align*}
and $\wt a_{\a,l}:=a_{\alpha_1,l_1}\, a_{\alpha_2,l_2}\,
a_{\alpha_3,-\wh l_2}$ with
$\wh l_2:=l_1+l_2$.

We use the same technique as in $(i)$:
\begin{align*}
\tfrac{1}{|k+l_1|^2}&=\tfrac{1}{|k|^2} -
\tfrac{2k.l_1+|l_1|^2}{|k|^4} +
\tfrac{(2k.l_1+|l_1|^2)^2}{|k|^4|k+l_1|^2}\, , \\
\tfrac{1}{|k+\wh l_2|^2}&=\frac{1}{|k|^2} - \tfrac{2k.\wh l_2+|\wh
l_2|^2}{|k|^4} +
\tfrac{(2k.\wh l_2+|\wh l_2|^2)^2}{|k|^4|k+\wh l_2|^2}
\end{align*}
and thus,
\begin{equation}\label{trick-2}
\tfrac{1}{|k+l_1|^2|k+\wh l_2|^2}=\tfrac{1}{|k|^4}
-\tfrac{2k.l_1}{|k|^6}-\tfrac{2k.\wh l_2}{|k|^6} +R(k,l)
\end{equation}
where the remain $R(k,l)$ is a term of order at most $-6$ in $k$.
Equation
(\ref{trick-2}) gives
$$
f_{\a,\mu}(s,l) = f_1(s,l) + r(s,l)
$$
where $r(s,l)$ corresponds to $R(k,l)$. Note that the function
$$
r(s,l) ={\sum_{k\in\Z^n}}'i\,e^{\tfrac i2 l_1\Th l_2}\,
\tfrac{k_{\mu_1}(k+l)_{\mu_2}(k+\wh l_2)_{\mu_3}R(k,l)}{|k|^{s+2}}\
\wt a_{\a,l}
$$
is a linear combination of functions of the type $H(s,l)$ satisfying
the hypothesis of
Corollary (\ref{res-somH}). Thus, $r(s,l)$ satisfies (H1) and
$f_{\a,\mu}(s,l) \sim
f_1(s,l)$.

Let us compute $f_1(s,l)$ modulo (H1)
\begin{align*}
f_1(s,l) &= {\sum_{k\in\Z^n}}'i\,e^{\tfrac i2 l_1\Th l_2}\,
\tfrac{k_{\mu_1}(k+l_1)_{\mu_2}(k+\wh l_2)_{\mu_3}}{|k|^{s+6}}\ \wt
a_{\a,l} -
{\sum_{k\in\Z^n}}' i\,e^{\tfrac i2 l_1\Th
l_2}\,\tfrac{k_{\mu_1}(k+l_1)_{\mu_2}(k+\wh
l_2)_{\mu_3}(2k. l_1+2k.\wh
l_2)}{|k|^{s+8}}\ \wt a_{\a,l}\\
&\sim {\sum_{k\in\Z^n}}'i\,e^{\tfrac i2 l_1\Th
l_2}\,\tfrac{k_{\mu_1}k_{\mu_2} \wh
{l_2}_{\mu_3}+k_{\mu_1}k_{\mu_3}{ l_1}_{\mu_2}}{|k|^{s+6}}\ \wt
a_{\a,l}-
{\sum_{k\in\Z^n}}' i\,e^{\tfrac i2 l_1\Th
l_2}\,\tfrac{k_{\mu_1}k_{\mu_2}k_{\mu_3}(2k.l_1+2k.\wh
l_2)}{|k|^{s+8}}\
\wt a_{\a,l}\\
&= i\,e^{\tfrac i2 l_1\Th l_2}\,\wt a_{\a,l}\big( ({l_{1}}_{\mu_2}
\delta_{\mu_1
\mu_3}+\wh {l_2}_{\mu_3} \delta_{\mu_1 \mu_2}) \, \,\tfrac 14
Z_{4}(s+4) -
2(l_{1}^{i}+\wh l_{2}^{i}){\sum_{k\in\Z^{n}}}'
\tfrac{k_{\mu_1}k_{\mu_2}
k_{\mu_3}k_{i}}{\vert k\vert^{s+8}}\big)\\&=:g_{\a,\mu}(s,l).
\end{align*}
Since $g_{\a,\mu}(s,l)$ satisfies (H2), we can apply Lemma
\ref{res-som} to
get
\begin{align*}
\ncint (\mathbb{A^+})^3&= \underset{s=0}{\Res} \sum_{(l_1,l_2)\in
(\Z^n)^2} f(s,l)\\
&= \sum_{(l_1,l_2)\in (\Z^n)^2}\underset{s=0}{\Res}\ g_{\a,\mu}(s,l)
\Tr(\ga^{\a_3}\ga^{\mu_3}\ga^{\a_2}\ga^{\mu_2}\ga^{\a_1}\ga^{\mu_1})=:\sum_{l}
X_l.
\end{align*}
Recall that $l_3:=-l_1-l_2=-\wh l_2$. By (\ref{formule1}) and
(\ref{formule2}),
\begin{align*} \underset{s=0}{\Res}\ g_{\a,\mu}(s,l) &i\,e^{\tfrac i2
l_1\Th l_2}\,\wt a_{\a,l}\big(
2(-l_{1}^{i}+l_{3}^{i})\tfrac{\pi^2}{12} (\delta_{\mu_1
\mu_2}\delta_{\mu_3
i}+\delta_{\mu_1 \mu_3}\delta_{\mu_2 i}+\delta_{\mu_1 i}\delta_{\mu_2
\mu_3})\\ &+
({l_{1}}_{\mu_2}\delta_{\mu_1 \mu_3} -{l_{3}}_{\mu_3}\delta_{\mu_1
\mu_2}) \tfrac
{\pi^2}{2}\big).
\end{align*}
We decompose $X_l$ in five terms:
$X_l= 2^m\ \tfrac{\pi^2}{2} \ i\,e^{\tfrac i2 l_{1}\Th l_{2}}\, \wt
a_{\a,l}\,(T_1+T_2+T_3+T_4+T_5)$
where
\begin{align*}
T_0 &:= \tfrac 13 (-l_1^{i}+l_3^{i})(\delta_{\mu\nu}\delta_{\rho
i}+\delta_{\mu\rho}\delta_{\nu i}+\delta_{\mu i}\delta_{\nu\rho})
+{l_1}_\nu
\delta_{\mu\rho} -{l_3}_\rho \delta_{\mu\nu},\\
T_1&:=(\delta^{\a_3\rho}\delta^{\a_2\nu}\delta^{\a_1\mu}
-\delta^{\a_3\rho}\delta^{\a_2\a_1}\delta^{\mu\nu}+
\delta^{\a_3\rho}\delta^{\a_2\mu}\delta^{\a_1\nu})T_0,\\
T_2&:=(-\delta^{\a_2\a_3}\delta^{\rho\nu}\delta^{\a_1\mu}
+\delta^{\a_2\a_3}\delta^{\a_1\rho}\delta^{\mu\nu}-
\delta^{\a_2\a_3}\delta^{\rho\mu}\delta^{\a_1\nu})T_0,\\
T_3&:=(\delta^{\a_3\nu}\delta^{\a_2\rho}\delta^{\a_1\mu}
-\delta^{\a_3\nu}\delta^{\a_1\rho}\delta^{\a_2\mu}+
\delta^{\a_3\nu}\delta^{\rho\mu}\delta^{\a_1\a_2})T_0,\\
T_4&:=(-\delta^{\a_1\a_3}\delta^{\a_2\rho}\delta^{\mu\nu}
+\delta^{\a_1\a_3}\delta^{\rho\nu}\delta^{\a_2\mu}-
\delta^{\a_1\a_3}\delta^{\rho\mu}\delta^{\a_2\nu})T_0,\\
T_5&:=(\delta^{\a_3\mu}\delta^{\a_2\rho}\delta^{\a_1\nu}
-\delta^{\a_3\mu}\delta^{\rho\nu}\delta^{\a_1\a_2}+
\delta^{\a_3\mu}\delta^{\a_1\rho}\delta^{\a_2\nu})T_0 .
\end{align*}
With the shorthand $p:=-l_1-2l_3$, $q:=2l_1+l_3$,
$r:=-p-q=-l_1+l_3$, we compute each
$T_i$, and find
\begin{align*}
3T_1&=\delta^{\a_1\a_2}(2-2^m) p^{\a_3}
+ \delta^{\a_3\a_1} q^{\a_2}-\delta^{\a_2\a_1}
q^{\a_3}+\delta^{\a_3\a_2}q^{\a_1} + \delta^{\a_3\a_2}r^{\a_1}
-\delta^{\a_2\a_1}r^{\a_3}+\delta^{\a_3\a_1}r^{\a_2},\\
3T_2&=(2^m-2)\delta^{\a_2\a_3}p^{\a_1} -2^m \delta^{\a_2\a_3}q^{\a_1}
-2^m
\delta^{\a_2\a_3}r^{\a_1},\\
3T_3&=\delta^{\a_1\a_3}p^{\a_2}-\delta^{\a_2\a_3}p^{\a_1}
+\delta^{\a_1\a_2}p^{\a_3}+2^m
\delta^{\a_2\a_1}q^{\a_3}+\delta^{\a_3\a_2}r^{\a_1}-\delta^{\a_3\a_1}r^{\a_2}
+\delta^{\a_1\a_2}r^{\a_3},\\
3T_4&=-\delta^{\a_1\a_3}2^m p^{\a_2}-\delta^{\a_1\a_3}2^m q^{\a_2}
+\delta^{\a_1\a_3}(2^m-2)r^{\a_2},\\
3T_5&=\delta^{\a_1\a_3}p^{\a_2}-\delta^{\a_1\a_2}p^{\a_3}
+\delta^{\a_3\a_2}p^{\a_1}+\delta^{\a_3\a_2}q^{\a_1}
-\delta^{\a_1\a_2}q^{\a_3}+\delta^{\a_3\a_1}q^{\a_2}
+(2-2^m)\delta^{\a_1\a_2}r^{\a_3}.
\end{align*}
Thus,
\begin{equation}
X_l = 2^m\ \tfrac{2\pi^2}{3} i\, e^{\tfrac i2 l_{1}.\Th l_{2}}\, \wt
a_{\a,l}
\,(q^{\a_3}\delta^{\a_1\a_2}
+r^{\a_2}\delta^{\a_1\a_3}+p^{\a_1}\delta^{\a_2\a_3})
\label{formuleA3-1}
\end{equation}
and
\begin{align*}
\ncint (\mathbb{A}^+)^3=i\,2 c\,  (S_1+S_2+S_3),
\end{align*}
where $S_1$, $S_2$ and $S_3$ correspond to respectively
$q^{\a_3}\delta^{\a_1\a_2}$,
$r^{\a_2}\delta^{\a_1\a_3}$ and $p^{\a_1}\delta^{\a_2\a_3}$. In
$S_1$, we permute the
$l_i$ variables the following way: $l_1\mapsto l_3$, $l_2\mapsto
l_1$, $l_3\mapsto
l_2$. Therefore, $l_3.\Th\, l_1 \mapsto l_3.\Th\, l_1$ and $q \mapsto
r$. With a
similar permutation of the $\a_i$, we see that $S_1=S_2$. We apply
the same principles
to prove that $S_1=S_3$ (using permutation $l_1\mapsto l_2$,
$l_2\mapsto l_3$,
$l_3\mapsto l_1$). Thus,
$$
\tfrac 13\ \ncint (\mathbb{A}^+)^3= i\ 2c\ \sum_{l_i} \wt a_{\a,l}\,
 e^{\tfrac i2 l_{1}.\Th l_{2}}\ (l_1-l_2)^{\a_3}
\delta^{\a_1\a_2}= S_4-S_5,
$$
where $S_4$ correspond to $l_1$ and $S_5$ to $l_2$. We permute the
$l_i$ variables in
$S_5$ the following way: $l_1\mapsto l_2$, $l_2\mapsto l_1$,
$l_3\mapsto l_3$, with a
similar permutation on the $\a_i$. Since $l_1.\Th\, l_2 \mapsto
-l_1.\Th\, l_2$, we
finally get
$$
\tfrac 13\ \ncint (\mathbb{A}^+)^3=-4c \sum_{l_i}
a_{\a_1,l_1}\,a_{\a_2,l_2}\,
a_{\a_3,-l_1-l_2}\ \sin \tfrac{l_1.\Th l_2}{2}\ l_1^{\a_3} \,
\delta^{\a_1\a_2}.
$$

$(iii)$ Lemma \ref{formegenerale} entails that
$\ncint \mathbb{A}^{++++}= \underset{s=0}{\Res} \sum_{(l_1,l_2,l_3)\in
(\Z^n)^3} f_{\mu,\a}(s,l) \Tr\ga^{\mu,\a}$
where
\begin{align*}
&\theta:=l_1.\Th l_2+ l_1.\Th l_3 + l_2. \Th l_3,\\
&\Tr\ga^{\mu,\a}:= \Tr(\ga^{\a_4}\ga^{\mu_4}\ga^{\a_3}\ga^{\mu_3}
\ga^{\a_2}\ga^{\mu2}\ga^{\a_1}\ga^{\mu_1}),\\
&f_{\mu,\a}(s,l):= {\sum_{k\in\Z^n}}'e^{\tfrac i2 \theta}\,
\tfrac{k_{\mu_1}(k+l_1)_{\mu_2}(k+\wh l_2)_{\mu_3}(k+\wh
l_3)_{\mu_4}}{|k|^{s+2}|k+l_1|^2|k+\wh l_2|^2|k+\wh l_3|^2}\ \wt
a_{\a,l},\\
&\wt a_{\a,l}:=a_{\alpha_1,l_1}\, a_{\alpha_2,l_2}\,
a_{\alpha_3,l_3}\,a_{\alpha_4,-l_1-l_2-l_3}.
\end{align*}

Using (\ref{trick-0}) and Corollary \ref{res-somH} successively, we
find
$$
f_{\mu,\a}(s,l)\sim  {\sum_{k\in\Z^{n}}}'e^{\tfrac i2 \theta}\,
\tfrac{k_{\mu_1}k_{\mu_2}k_{\mu_3}k_{\mu_4}}{\vert k\vert^{s+2}\vert
k+l_{1}\vert^{2}{\vert k+l_{1}+l_{2}\vert^2} \vert k+l_1+l_2+l_3
\vert^2} \, \wt
a_{\a,l} \sim {\sum_{k\in\Z^{n}}}' e^{\tfrac i2 \theta}\,
\tfrac{k_{\mu_1}k_{\mu_2}k_{\mu_3}k_{\mu_4}}{\vert k\vert^{s+8}} \,
\wt a_{\a,l}.
$$

Since the function ${\sum_{k\in\Z^{n}}}'e^{\tfrac i2 \theta}\,
\tfrac{k_{\mu_1}k_{\mu_2}k_{\mu_3}k_{\mu_4}}{\vert k\vert^{s+8}} \,
\wt a_{\a,l}$
satisfies (H2), Lemma \ref{res-som} entails that
$$
\ncint (\mathbb{A}^+)^4= \sum_{(l_1,l_2,l_3)\in (\Z^n)^3} e^{\tfrac
i2 \theta}\,\wt
a_{\a,l}\,\underset{s=0}{\Res}\ {\sum_{k\in\Z^{n}}}'
\tfrac{k_{\mu_1}k_{\mu_2}k_{\mu_3}k_{\mu_4}}{\vert k\vert^{s+8}}
\Tr\ga^{\mu,\a}
=:\sum_l X_l.
$$
Therefore, with (\ref{formule2}), we get
$X_l =\tfrac{\pi^2}{12}\, \wt a_{\a,l}\, e^{\tfrac i2 \th}\,
(A+B+C)$,
where
\begin{align*}
A&:=\Tr(\ga^{\a_4}\ga^{\mu_4}\ga^{\a_3}
\ga_{\mu_4}\ga^{\a_2}\ga^{\mu_2}\ga^{\a_1}\ga_{\mu_2}),\\
B&:=\Tr(\ga^{\a_4}\ga^{\mu_4}\ga^{\a_3}\ga^{\mu_2}\ga^{\a_2}\ga_{\mu_4}\ga^{\a_1}\ga_{\mu_2}),\\
C&:=\Tr(\ga^{\a_4}\ga^{\mu_4}\ga^{\a_3}
\ga_{\mu_2}\ga^{\a_2}\ga^{\mu_2}\ga^{\a_1}\ga_{\mu_4}).
\end{align*}
Using successively $\{\gamma^{\mu},\gamma^{\nu}\}=2\delta^{\mu\nu}$
and
$\gamma^\mu\gamma_\mu=2^m\ 1_{2^m}$, we see that
\begin{align*}
A&=C=4\ \Tr(\ga^{\a_4}\ga^{\a_3}\ga^{\a_2}\ga^{\a_1}),\\
B&=-4\ \big(\Tr(\ga^{\a_4}\ga^{\a_3}\ga^{\a_1}\ga^{\a_2}) +
\Tr(\ga^{\a_4}\ga^{\a_2}\ga^{\a_3}\ga^{\a_1})\big).
\end{align*}
Thus,
$
A+B+C=8\ 2^m \big( \delta^{\a_4\a_3}\delta^{\a_2\a_1}
+\delta^{\a_4\a_1}\delta^{\a_3\a_2}-2\delta^{\a_4\a_2}\delta^{\a_3\a_1}\big),
$
and
\begin{equation}
X_l =\tfrac{2\pi^2}{3}\ 2^m\ e^{\tfrac i2 \th}\, \wt a_{\a,l}\,  \big(
\delta^{\a_4\a_3}\delta^{\a_2\a_1} +\delta^{\a_4\a_1}\delta^{\a_3\a_2}
-2\delta^{\a_4\a_2}\delta^{\a_3\a_1}\big).\label{formuleA4c}
\end{equation}
By (\ref{formuleA4c}), we get
$$
\ncint (\mathbb{A}^+)^4 = 2c\ (-2T_1 +T_2+T_3),
$$
where
\begin{align*}
T_1&:=\sum_{l_1,\ldots,l_4}a_{\a_4,l_4}\,a_{\a_3,l_3}\, a_{\a_2,l_2}
\,a_{\a_1,l_1} \,
e^{\tfrac i2 \th}\
\delta_{0,\sum_i l_i}\ \delta^{\a_4\a_2}\,\delta^{\a_3\a_1},\\
T_2&:=\sum_{l_1,\ldots,l_4}a_{\a_4,l_4}\,
a_{\a_3,l_3}\,a_{\a_2,l_2}\,a_{\a_1,l_1} \,
e^{\tfrac i2 \th}\
\delta_{0,\sum_i l_i}\ \delta^{\a_4\a_3}\,\delta^{\a_2\a_1},\\
T_3&:=\sum_{l_1,\ldots,l_4}a_{\a_4,l_4}\,a_{\a_3,l_3}\,a_{\a_2,l_2}\,a_{\a_1,l_1}
\,
e^{\tfrac i2\th}\ \delta_{0,\sum_i l_i}\
\delta^{\a_4\a_1}\,\delta^{\a_3\a_2}.
\end{align*}
We now proceed to the following permutations
of the $l_i$ variables in the $T_1$ term :
$l_1\mapsto l_2$, $l_2 \mapsto l_1$,
$l_3 \mapsto l_4$, $l_4 \mapsto l_3$. While
$\sum_i l_i$ is invariant, $\th$ is modified :
$\th \mapsto l_2 .\Th l_1 + l_2 .\Th
l_4 + l_1 .\Th  l_4$. With $\delta_{0,\sum_i l_i}$
in factor, we can let $l_4$ be
$-l_1-l_2-l_3$, so that $\th \mapsto -\th$. We also
permute the $\a_i$ in the same
way. Thus,
$$
T_1=\sum_{l_1,\ldots,l_4}a_{\a_3,l_3}\,a_{\a_4,l_4} \, a_{\a_1,l_1}
\,a_{\a_2,l_2} \,
e^{-\tfrac i2 \th}\ \delta_{0,\sum_i l_i}\ \delta^{\a_3\a_1}
\,\delta^{\a_4\a_2}.
$$
Therefore,
\begin{equation}
2T_1 = 2\sum_{l_1,\ldots,l_4}a_{\a_4,l_4} \, a_{\a_3,l_3}\,
a_{\a_2,l_2}
\,a_{\a_1,l_1}\ \cos \tfrac{\th}{2}\ \delta_{0,\sum_i l_i}\
\delta^{\a_4\a_2} \,
\delta^{\a_3\a_1}\label{formuleT1}.
\end{equation}
The same principles are applied to $T_2$ and $T_3$.
Namely, the permutation
$l_1\mapsto l_1$, $l_2\mapsto l_3$, $l_3
\mapsto l_2$, $l_4 \mapsto l_4$ in $T_2$ and
the permutation $l_1\mapsto l_2$,
$l_2\mapsto l_3$, $l_3 \mapsto l_1$, $l_4 \mapsto
l_4$ in $T_3$ (the $\a_i$ variables are permuted the same way) give
\begin{align*}
T_2 &= \sum_{l_1,\ldots,l_4}a_{\a_4,l_4}a_{\a_3,l_3} a_{\a_2,l_2} \,
a_{\a_1,l_1} \,
e^{\tfrac i2 \phi}\
\delta_{0,\sum_i l_i}\ \delta^{\a_4\a_2}\,\delta^{\a_3\a_1},\\
T_3 &= \sum_{l_1,\ldots,l_4}a_{\a_4,l_4} \, a_{\a_3,l_3} a_{\a_2,l_2}
\, a_{\a_1,l_1}
\,e^{-\tfrac i2 \phi}\ \delta_{0,\sum_i l_i}\
\delta^{\a_4\a_2}\,\delta^{\a_3\a_1}
\end{align*}
where $\phi:=l_1 .\Th\ l_2 + l_1 .\Th\ l_3 - l_2 .\Th\ l_3$. Finally,
we get
\begin{align}
\ncint (\mathbb{A}^+)^4&=4c\ \sum_{l_1,\ldots,l_4}a_{{\a_{1}},l_4} \,
a_{{\a_{2}},l_3}\,
a^{\a_{1}}_{l_2} \, a^{\a_{2}}_{l_1} \, \delta_{0,\sum_i l_i}(\cos
\tfrac{\phi}{2}
-\cos \tfrac{\th}{2})\nonumber\\
&=8c\ \sum_{l_1,\ldots,l_3}a_{{\a_{1}},-l_1-l_2-l_3}\,
a_{{\a_{2}},l_3} \,
a^{\a_{1}}_{l_2} \, a^{\a_{2}}_{l_1} \sin \tfrac{l_1.\Th
(l_2+l_3)}{2}\ \sin \tfrac{
l_2 .\Th l_3}{2}.\label{formuleA4}
\end{align}

$(iv)$ Suppose $q=2$. By Lemma \ref{formegenerale}, we get
$$
\ncint \mathbb{A}^{\sigma}= \underset{s=0}{\Res} \sum_{l\in \Z^n}
\lambda_\sigma
f_{\a,\mu}(s,l) \Tr(\ga^{\a_2}\ga^{\mu_2}\ga^{\a_1}\ga^{\mu_1})
$$
where
$$
f_{\a,\mu}(s,l):=  {\sum_{k\in\Z^n}}'
\tfrac{k_{\mu_1}(k+l)_{\mu_2}}{|k|^{s+2}|k+l|^2}\,e^{i \eta \,  k.\Th
l}
\,\wt a_{\a,l}\,
$$
and $\eta:=\half (\sigma_1-\sigma_2) \in \{-1,1\}$. As in the proof of
$(i)$, since the presence of the phase does not change the
fact that $r(s,l)$ satisfies (H1), we get
$$
f_{\a,\mu}(s,l) \sim f_1(s,l) - f_2(s,l) + f_3(s,l)
$$
where
\begin{align*}
f_1(s,l) &= {\sum_{k\in\Z^n}}'
\tfrac{k_{\mu_1}(k+l)_{\mu_2}}{|k|^{s+4}}\,e^{i \eta \,
k.\Th l}\, \wt a_{\a,l},\\
f_2(s,l)&={\sum_{k\in\Z^n}}'
\tfrac{k_{\mu_1}(k+l)_{\mu_2}(2k.l+|l|^2)}{|k|^{s+6}}\,e^{i \eta \,
k.\Th l} \,  \wt
a_{\a,l} ,\\
f_3(s,l)&={\sum_{k\in\Z^n}}'
\tfrac{k_{\mu_1}(k+l)_{\mu_2}(2k.l+|l|^2)^2}{|k|^{s+8}}\,e^{i \eta \,
k.\Th l} \, \wt
a_{\a,l}.
\end{align*}
Suppose that $l=0$. Then $f_2(s,0)=f_3(s,0)=0$ and Proposition
\ref{res-int} entails
that
\begin{align*}
f_1(s,0)&= {\sum}_{k\in\Z^n}' \tfrac{k_{\mu_1}k_{\mu_2}}{|k|^{s+4}}\,
\wt a_{\a,0}
\end{align*}
is holomorphic at 0 and so is $f_{\a,\mu}(s,0)$.

Since $\tfrac {1}{2\pi}\Th$
is diophantine, Theorem \ref{analytic} $3$ gives us the result.

Suppose $q=3$. Then Lemma \ref{formegenerale} implies that
$$
\ncint \mathbb{A}^{\sigma}
= \underset{s=0}\Res\ {\sum}_{l\in (\Z^n)^{2}} \,  f_{\mu,\a}(s,l)\,
\Tr(\ga^{\mu_3}\ga^{\a_3}\cdots \ga^{\mu_1}\ga^{\a_1})
$$
where
$$
f_{\mu,\a}(s,l):= {\sum}'_{k\in \Z^n}\la_\sigma e^{ik.\Th(\eps_1
l_1+\eps_2 l_2)}
e^{\tfrac i2 \sigma_2 l_1.\Th
l_2}\tfrac{k_{\mu_1}(k+l_1)_{\mu_2}(k+l_1+
l_2)_{\mu_3}}{|k|^{s+2}|k+l_1|^2|k+l_1+l_{2}|^2}\,\wt a_{\a,l},
$$
and $\eps_i :=\half(\sigma_i - \sigma_3)\in \{-1,0,1\}$. By hypothesis
$(\eps_1,\eps_2)\neq (0,0)$. There are six possibilities for the
values of
$(\eps_1,\eps_2)$, corresponding to the six possibilities for the
values of $\sigma$:
$(-,-,+)$, $(-,+,+)$, $(+,-,+)$, $(+,+,-)$, $(-,+,-)$, and $(+,-,-)$.
As in $(ii)$, we
see that
\begin{align*}
f_{\mu,\a}(s,l)&\sim\big( {\sum_{k\in\Z^n}}' \tfrac{e^{ik.\Th(\eps_1
l_1+\eps_2
l_2)}k_{\mu_1}(k+l_1)_{\mu_2}(k+\wh l_2)_{\mu_3}}{|k|^{s+6}}\\
& \hspace{1cm}- {\sum_{k\in\Z^n}}'
\tfrac{e^{ik.\Th(\eps_1 l_1+\eps_2 l_2)}k_{\mu_1}(k+l_1)_{\mu_2}(k+\wh
l_2)_{\mu_3}(2k. l_1+2k.\wh l_2)}{|k|^{s+8}}\, \la_\sigma \,\wt
a_{\a,l}\,e^{\tfrac i2
\sigma_2 l_1.\Th l_2}.
\end{align*}
With $Z:=\{(l_1,l_2) \ : \ \eps_1 l_1 + \eps_2 l_2=0\}$, Theorem
\ref{analytic} $(iii)$
entails that $\sum_{l\in (\Z^n)^2\setminus Z}f_{\mu,\a}(s,l)$ is
holomorphic at 0.
To conclude we need to prove that
$$\sum_{\sigma} g(\sigma) := \sum_{\sigma} \sum_{l\in
Z}f_{\mu,\a}(s,l)\,
\Tr(\ga^{\mu_3}\ga^{\a_3}\cdots \ga^{\mu_1}\ga^{\a_1})$$ is
holomorphic at 0. By
definition, $\la_\sigma = i\sigma_1\sigma_2\sigma_3$
and as a consequence, we check that
\begin{align*}
g(-,-,+)=-g(+,+,-),\quad
g(+,-,+)=-g(+,-,-),\quad
g(-,+,+)=-g(-,+,-),
\end{align*}
which implies that $\sum_{\sigma} g(\sigma)=0$. The result follows.

Suppose finally that $q=4$. Again, Lemma \ref{formegenerale} implies
that
$$
\ncint \mathbb{A}^{\sigma} = \underset{s=0}\Res\ \sum_{l\in
(\Z^n)^{3}} f_{\mu,\a}(s,l)\,
\Tr(\ga^{\mu_4}\ga^{\a_4}\cdots \ga^{\mu_1}\ga^{\a_1})
$$
where
$$
f_{\mu,\a}(s,l):= {\sum_{k\in \Z^n}}'\la_\sigma\, e^{ik.\Th
\sum_{i=1}^{3}\eps_i l_i}
\, e^{\tfrac i2 (\sigma_2 l_1.\Th l_2+\sigma_3 (l_1+l_2).\Th
l_3)} \, \tfrac{k_{\mu_1}(k+l_1)_{\mu_2}(k+l_1+
l_2)_{\mu_3}(k+l_1+l_2+l_3)_{\mu_4}}
{|k|^{s+2}|k+l_1|^2|k+l_1+l_{2}|^2|k+l_1+l_2+l_3|^2}\, \wt
a_{\a,l}
$$
and $\eps_i :=\half(\sigma_i - \sigma_4)\in \{-1,0,1\}$. By hypothesis
$(\eps_1,\eps_2,\eps_3)\neq (0,0,0)$. There are fourteen
possibilities for the values
of $(\eps_1,\eps_2,\eps_3)$, corresponding to the fourteen
possibilities for the
values of $\sigma$: $(-,-,-,+)$, $(-,-,+,+)$, $(-,+,-,+)$,
$(+,-,-,+)$, $(-,+,+,+)$,
$(+,-,+,+)$, $(+,+,-,+)$, $(+,+,+,-)$, $(-,-,+,-)$, $(-,+,-,-)$,
$(+,-,-,-)$,
$(-,+,+,-)$, $(+,-,+,-)$ and $(+,+,-,-)$. As in $(ii)$, we see that,
with the
shorthand $\th_\sigma:=\sigma_2 l_1.\Th l_2+\sigma_3 (l_1+l_2).\Th
l_3$,
\begin{align*}
f_{\mu,\a}(s,l)\sim {\sum}'_{k\in \Z^n}\la_\sigma \, e^{ik.\Th
\sum_{i=1}^3\eps_i l_i}
\, e^{\tfrac i2 \th_\sigma} \,
\tfrac{k_{\mu_1}k_{\mu_2}k_{\mu_3}k_{\mu_4}}{|k|^{s+8}}\wt
a_{\a,l}=:g_{\mu,\a}(s,l)\, .
\end{align*}
With $Z_\sigma:=\{(l_1,l_2,l_3) \ : \ \sum_{i=1}^3\eps_i l_i=0\}$,
Theorem
\ref{analytic} $(iii)$, the series
$\sum_{l \in (\Z^n)^3 \setminus Z_\sigma}f_{\mu,\a}(s,l)$ is
holomorphic at 0. To conclude, we need to prove that
$$
\sum_{\sigma} g(\sigma) := \sum_{\sigma}\underset{s=0}\Res\
\sum_{l\in Z_\sigma}g_{\mu,\a}(s,l)\,
\Tr(\ga^{\mu_4}\ga^{\a_4}\cdots \ga^{\mu_1}\ga^{\a_1})=0.
$$
Let $C$ be the set of
the fourteen values of $\sigma$ and $C_7$ be the set of the seven
first values of
$\sigma$ given above. Lemma \ref{symetrie} implies
$$
\sum_{\sigma\in C} g(\sigma) = 2\sum_{\sigma\in C_7} g(\sigma).
$$
Thus, in the following, we restrict to these seven values. Let us note
$F_\mu(s):={\sum}'_{k\in
\Z^n}\tfrac{k_{\mu_1}k_{\mu_2}k_{\mu_3}k_{\mu_4}}{|k|^{s+8}}$
so that
$$
g(\sigma)=\underset{s=0} \Res \ F_\mu(s) \, \la_\sigma \, \sum_{l\in
Z_\sigma} e^{\tfrac i2
\th_\sigma} \, \wt a_{\a,l}\, \Tr(\ga^{\mu_4}\ga^{\a_4}\cdots
\ga^{\mu_1}\ga^{\a_1}).
$$
Recall from (\ref{formuleA4c}) that
$$
\underset{s=0}\Res\ F_\mu(s) \Tr(\ga^{\mu_4}\ga^{\a_4}\cdots
\ga^{\mu_1}\ga^{\a_1})
=2c \big( \delta^{\a_4\a_3}\delta^{\a_2\a_1}
+\delta^{\a_4\a_1}\delta^{\a_3\a_2}
-2\delta^{\a_4\a_2}\delta^{\a_3\a_1}\big).
$$
As a consequence, we get, with $\wt a_{\a,l}:= a_{\a_1,l_1}\cdots
a_{\a_4,l_4}$,
\begin{align*}
g(\sigma)&=2c\la_\sigma \sum_{l\in (\Z^n)^{4}} e^{\tfrac i2
\th_\sigma} \, \wt
a_{\a,l} \,\delta_{\sum_{i=1}^4 l_i,0} \,\delta_{\sum_{i=1}^3\eps_i
l_i,0}\big(
\delta^{\a_4\a_3}\delta^{\a_2\a_1} +\delta^{\a_4\a_1}\delta^{\a_3\a_2}
-2\delta^{\a_4\a_2}\delta^{\a_3\a_1}\big)\\
&=:2c\la_\sigma(T_1+T_2-2 T_3).
\end{align*}
We proceed to the following change of variable in $T_1$: $l_1\mapsto
l_1$, $l_2\mapsto
l_3$, $l_3\mapsto l_2$, $l_4\mapsto l_4$. Thus, we get
$\th_\sigma\mapsto
\psi_\sigma:=\sigma_2 l_1.\Th l_3+\sigma_3(l_1+l_3).\Th l_2$, and
$\sum_{i=1}^3\eps_i
l_i \mapsto \eps_1 l_1 + \eps_3 l_2 + \eps_2 l_3=:u_\sigma(l)$. With
a similar
permutation on the $\a_i$, we get
$$
T_1 =\sum_{l\in (\Z^n)^{4}} e^{\tfrac i2 \psi_\sigma}\,\wt
a_{\a,l}\,\delta_{\sum_{i=1}^4 l_i,0} \,\delta_{\eps_1 l_1 + \eps_3
l_2 + \eps_2
l_3,0}\, \delta^{\a_4\a_2}\delta^{\a_3\a_1}.
$$
We proceed to the following change of variable in $T_2$: $l_1\mapsto
l_2$, $l_2\mapsto
l_3$, $l_3\mapsto l_1$, $l_4\mapsto l_4$. Thus, we get
$\th_\sigma\mapsto
\phi_\sigma:=\sigma_2 l_2.\Th l_3+\sigma_3(l_2+l_3).\Th l_1$, and
$\sum_{i=1}^3\eps_i
l_i \mapsto \eps_3 l_1 + \eps_1 l_2 + \eps_2 l_3=:v_\sigma(l)$. After
a similar
permutation on the $\a_i$, we get
$$
T_2 ={\sum}_{l\in (\Z^n)^{4}} \, e^{\tfrac i2 \phi_\sigma}\,\wt
a_{\a,l}\,\delta_{\sum_{i=1}^4 l_i,0} \,\delta_{\eps_3 l_1 + \eps_1
l_2 + \eps_2
l_3,0} \, \delta^{\a_4\a_2}\delta^{\a_3\a_1}.
$$
Finally, we proceed to the following change of variable in $T_3$:
$l_1\mapsto l_2$,
$l_2\mapsto l_1$, $l_3\mapsto l_4$, $l_4\mapsto l_3$. Thus, we get
$\th_\sigma\mapsto
-\th_\sigma$, and $\sum_{i=1}^3\eps_i l_i \mapsto (\eps_2-\eps_3) l_1
+
(\eps_1-\eps_3) l_2 - \eps_3 l_3=:w_\sigma(l)$. With a similar
permutation on the
$\a_i$, we get
$$
T_3 ={\sum}_{l\in (\Z^n)^{4}}\, e^{-\tfrac i2 \th_\sigma}\,\wt
a_{\a,l}\,\delta_{\sum_{i=1}^4 l_i,0} \,\delta_{(\eps_2-\eps_3) l_1 +
(\eps_1-\eps_3)
l_2 - \eps_3 l_3,0}\delta^{\a_4\a_2}\delta^{\a_3\a_1}.
$$
As a consequence, we get
$$
g(\sigma)=2c {\sum}_{l\in (\Z^n)^{4}} K_\sigma(l_1,l_2,l_3)\,\wt
a_{\a,l}\,\delta_{\sum_{i=1}^4
l_i,0}\,\delta^{\a_4\a_2}\delta^{\a_3\a_1},
$$
where
$
K_\sigma(l_1,l_2,l_3)=\la_\sigma\big(e^{\tfrac i2
\psi_\sigma}\,\delta_{u_\sigma(l),0}+e^{\tfrac i2
\phi_\sigma}\,\delta_{v_\sigma(l),0}
-e^{\tfrac i2 \th_\sigma}\, \delta_{\sum_{i=1}^3\eps_i
l_i,0}-e^{-\tfrac i2
\th_\sigma}\, \delta_{w_\sigma(l),0}\big).
$

The computation of $K_\sigma(l_1,l_2,l_3)$ for the seven values of
$\sigma$ yields
\begin{align*}
K_{--++}(l_1,l_2,l_3)&=\delta_{l_1+l_3,0}+\delta_{l_2+l_3,0}-\delta_{l_1+l_2,0}-\delta_{l_1+l_2,0},\\
K_{-+-+}(l_1,l_2,l_3)&=\delta_{l_1+l_2,0}+\delta_{l_1+l_2,0}-\delta_{l_1+l_3,0}-\delta_{l_1+l_3,0},\\
K_{--++}(l_1,l_2,l_3)&=\delta_{l_2+l_3,0}+\delta_{l_1+l_3,0}-\delta_{l_2+l_3,0}-\delta_{l_2+l_3,0},\\
K_{---+}(l_1,l_2,l_3)&=-\big( e^{\tfrac i2 l_1.\Th l_2
}\delta_{\sum_{i=1}^3 l_i,0} +
e^{\tfrac i2 l_2.\Th l_1}\delta_{\sum_{i=1}^3 l_i,0}
- e^{\tfrac i2 l_2.\Th l_1} \delta_{\sum_{i=1}^3 l_i,0}- e^{\tfrac i2
l_1.\Th l_2} \delta_{l_3,0} \big),\\
K_{-+++}(l_1,l_2,l_3)&=-\big(e^{\tfrac i2 l_3.\Th l_2 }\delta_{l_1,0}
+ e^{\tfrac i2
l_3.\Th l_1}\delta_{l_2,0}
- e^{\tfrac i2 l_2.\Th l_3}\delta_{l_1,0} - e^{\tfrac i2 l_3.\Th
l_1}\delta_{l_2,0} \big),\\
K_{+-++}(l_1,l_2,l_3)&=-\big(e^{\tfrac i2 l_1.\Th l_2 }\delta_{l_3,0}
+ e^{\tfrac i2
l_2.\Th l_1} \delta_{l_3,0}-
e^{\tfrac i2 l_1.\Th l_3}\delta_{l_2,0} - e^{\tfrac i2 l_3.\Th
l_2}\delta_{l_1,0} \big),\\
K_{++-+}(l_1,l_2,l_3)&=-\big(e^{\tfrac i2 l_1.\Th l_3 }\delta_{l_2,0}
+ e^{\tfrac i2
l_2.\Th l_3}\delta_{l_1,0} - e^{\tfrac i2 l_1.\Th l_2}\delta_{l_3,0}
- e^{\tfrac i2
l_2.\Th l_1}\delta_{\sum_{i=1}^3 l_i,0} \big).
\end{align*}
Thus,
$$
\sum_{\sigma\in C_7} K_\sigma(l_1,l_2,l_3) =
2i(\delta_{l_3,0}-\delta_{\sum_{i=1}^3
l_i,0})\sin \tfrac {l_1.\Th l_2}{2}
$$
and
$$
\sum_{\sigma\in C_7} g(\sigma)=i 4c  \sum_{l\in (\Z^n)^{4}}
(\delta_{l_3,0}-\delta_{\sum_{i=1}^3 l_i,0})\sin \tfrac {l_1.\Th
l_2}{2}\,\wt
a_{\a,l}\,\delta_{\sum_{i=1}^4
l_i,0}\,\delta^{\a_4\a_2}\delta^{\a_3\a_1}.
$$
The following change of variables: $l_1\mapsto l_2$, $l_1\mapsto
l_2$, $l_3\mapsto
l_4$, $l_4\mapsto l_3$ gives
$$
\sum_{l\in (\Z^n)^{4}} \delta_{\sum_{1}^3 l_i,0}\sin \tfrac {l_1.\Th
l_2}{2}\,\wt
a_{\a,l}\,\delta_{\sum_{1}^4
l_i,0}\,\delta^{\a_4\a_2}\delta^{\a_3\a_1} =-\sum_{l\in
(\Z^n)^{4}} \delta_{l_3,0} \sin \tfrac {l_1.\Th l_2}{2}\,\wt
a_{\a,l}\,\delta_{\sum_{1}^4
l_i,0}\,\delta^{\a_4\a_2}\delta^{\a_3\a_1}
$$
so
$$
\sum_{\sigma\in C_7} g(\sigma)=i8c  \sum_{l\in (\Z^n)^{4}}
\delta_{l_3,0}\,\sin \tfrac
{l_1.\Th l_2}{2}\,\wt a_{\a,l}\,\delta_{\sum_{1}^4
l_i,0}\,\delta^{\a_4\a_2}\delta^{\a_3\a_1}.
$$
Finally, the change of variables: $l_2\mapsto l_4, l_4\mapsto l_2$
gives
$$
\sum_{l\in (\Z^n)^{4}} \delta_{l_3,0}\,\sin \tfrac {l_1.\Th
l_2}{2}\,\wt
a_{\a,l}\,\delta_{\sum_{1}^4
l_i,0}\,\delta^{\a_4\a_2}\delta^{\a_3\a_1} = -\sum_{l\in
(\Z^n)^{4}} \delta_{l_3,0}\,\sin \tfrac {l_1.\Th l_2}{2}\,\wt
a_{\a,l}\,\delta_{\sum_{1}^4
l_i,0}\,\delta^{\a_4\a_2}\delta^{\a_3\a_1}
$$
which entails that $\sum_{\sigma\in C_7} g(\sigma)=0$.
\end{proof}

\begin{lemma}
\label{double}
Suppose $n=4$ and $\tfrac{1}{2\pi}\Th$ diophantine. For any
self-adjoint one-form $A$,
$$
\zeta_{D_A}(0)-\zeta_{D}(0)=-c\,
\tau(F_{\a_1,\a_2}F^{\a_1\a_2}).
$$
\end{lemma}
\begin{proof}
By (\ref{termconstanttilde}) and Lemma \ref{ncadmoins1} we get
$$
\zeta_{D_A}(0)-\zeta_{D}(0) =
\sum_{q=1}^n \tfrac {(-1)^q}{q} \sum_{\sigma\in \set{+,-}^q}\ncint
\mathbb{A}^\sigma.
$$
By Lemma \ref{Termaterm} $(iv)$, we see that the crossed terms all
vanish. Thus,
with Lemma \ref{symetrie}, we get
\begin{equation}
\label{n4eq1}
\zeta_{D_A}(0)-\zeta_{D}(0) = 2 \sum_{q=1}^n
\tfrac {(-1)^q}{q} \ncint (\mathbb{A}^+)^q.
\end{equation}

By definition,
\begin{align*}
F_{\alpha_{1} \alpha_{2}}&=i\sum_{k}\big(a_{\alpha_{2},k}\,
k_{\alpha_{1}}-a_{\alpha_{1},k}\, k_{\alpha_{2}}\big)U_{k}
+\sum_{k,\,l}a_{\alpha_{1},k}\,a_{\alpha_{2},l}\,[U_{k},U_{l}]\\
&=i\sum_{k}\big[(a_{\alpha_{2},k} \, k_{\alpha_{1}} -a_{\alpha_{1},k}
\,
k_{\alpha_{2}}) -2\sum_{l}a_{\alpha_{1},k-l}\,a_{\alpha_{2},l}\,
\sin(\tfrac{k.\Th
l}{2})\big] \, U_{k}.
\end{align*}
Thus
\begin{align*}
\tau(F_{\alpha_{1}\alpha_{2}} F^{\alpha_{1}\alpha_{2}})
&=\sum_{\alpha_{1},\,\alpha_{2}=1}^{2^m}\, \sum_{k\in \Z^4}\big[
(a_{\alpha_{2},k} \, k_{\alpha_{1}}-a_{\alpha_{1},k} \,
k_{\alpha_{2}})
-2\sum_{l'\in\Z^4}a_{\alpha_{1},k-l'}\,a_{\alpha_{2},l'}\,
\sin(\tfrac{k.\Th l'}{2})\big]\\
&\hspace{2.2cm} \big[(a_{\alpha_{2},-k} \, k_{\alpha_{1}}
-a_{\alpha_{1},-k}\, k_{\alpha_{2}})
-2\sum_{l"\in
\Z^4}a_{\alpha_{1},-k-l"}\,a_{\alpha_{2},l"}\,\sin(\tfrac{k.\Th
l"}{2})\big].
\end{align*}
One checks that the term in $a^q$ of $\tau(F_{\alpha_{1}\alpha_{2}}
F^{\alpha_{1}\alpha_{2}})$ corresponds to the term $\ncint
(\mathbb{A}^+)^q$ given by
Lemma \ref{Termaterm}. For $q=2$, this is
$$
-2\sum_{l\in\Z^4,\,\alpha_{1},\,\alpha_{2}} \, a_{\alpha_{1},l} \,
a_{\alpha_{2},-l}
\, \big(l_{\alpha_{1}}l_{\alpha_{2}} - \delta_{\alpha_{1}\alpha_{2}}
\vert l \vert
^2\big).
$$
\noindent For $q=3$, we compute the crossed terms:
$$
i\sum_{k,k',l} (a_{{\a_{2}},k} \, k_{\a_{1}} - a_{{\a_{1}},k} \,
k_{\a_{2}})\
a_{k'}^{\a_1}\ a_{l}^{\a_2} \big(U_k[U_{k'},l]+[U_{k'},U_l]U_k \big),
$$
which gives the following $a^3$-term in
$\tau(F_{\alpha_{1}\alpha_{2}} F^{\alpha_{1}\alpha_{2}})$
$$
-8\sum_{l_i} a_{\a_3,-l_1-l_2}\,a^{\a_1}_{l_2}\,a_{\a_1,l_1}\ \sin
\tfrac{l_1.\Th
l_2}{2}\,l_1^{\a_3} .
$$
\noindent For $q=4$, this is
$$-4\sum_{l_i}a_{\alpha_{1},-l_1-l_2-l_3}\,
a_{{\alpha_{2}},l_3} \, a^{\alpha_{1}}_{l_2} \, a^{\alpha_{2}}_{l_1}
\sin \tfrac{l_1
.\Th (l_2+l_3)}{2}\, \sin \tfrac{ l_2 .\Th l_3}{2}
$$
which corresponds to the term $\ncint (\mathbb{A}^+)^4$.
We get finally,
\begin{equation}\label{n4eq2}
\sum_{q=1}^n
\tfrac {(-1)^q}{q} \ncint (\mathbb{A}^+)^q
=- \tfrac c2 \tau(F_{\a_1,\a_2}F^{\a_1\a_2}).
\end{equation}
Equations (\ref{n4eq1}) and (\ref{n4eq2}) yield the result.
\end{proof}

\begin{lemma}
    \label{term-n=2}
Suppose $n=2$. Then, with the same hypothesis as in Lemma
\ref{formegenerale},
\begin{align*}
\hspace{-5.9cm}\text{(i)} \quad \quad &  \ncint (\mathbb
A^+)^2= \ncint (\mathbb A^-)^2=0.
\end{align*}
(ii) Suppose $\tfrac{1}{2\pi}\Th$ diophantine. Then
$$
\ncint \mathbb A^+ \mathbb A^{-}
=  \ncint \mathbb A^- \mathbb A^+=0.
$$
\end{lemma}
\begin{proof}
$(i)$ Lemma \ref{formegenerale} entails that $\ncint \mathbb{A}^{++}=
\underset{s=0}{\Res} \sum_{l\in \Z^2} - f(s,l)$ where
$$
f(s,l):=  {{\sum}'_{k\in\Z^2}}
\tfrac{k_{\mu_1}(k+l)_{\mu_2}}{|k|^{s+2}|k+l|^2}\ \wt
a_{\a,l}\
\Tr(\ga^{\a_2}\ga^{\mu_2}\ga^{\a_1}\ga^{\mu_1})=:f_{\mu,\a}(s,l)
\Tr(\ga^{\a_2}\ga^{\mu_2}\ga^{\a_1}\ga^{\mu_1})
$$
and $\wt a_{\a,l}:=a_{\alpha_1,l}\, a_{\alpha_2,-l}$. This time,
since $n=2$, it is
enough to apply just once (\ref{trick-0}) to obtain an absolutely
convergent series.
Indeed, we get with (\ref{trick-0})
$$
f_{\mu,\a}(s,l)= {\sum_{k\in\Z^2}}'
\tfrac{k_{\mu_1}(k+l)_{\mu_2}}{|k|^{s+4}}\ \wt
a_{\a,l}- {\sum_{k\in\Z^2}}'
\tfrac{k_{\mu_1}(k+l)_{\mu_2}(2k.l+|l|^2)}{|k|^{s+4}|k+l|^2}\ \wt
a_{\a,l}.
$$
and the function $ r(s,l):= {\sum}'_{k\in\Z^2}
\tfrac{k_{\mu_1}(k+l)_{\mu_2}(2k.l+|l|^2)}{|k|^{s+4}|k+l|^2}\ \wt
a_{\a,l} $ is a
linear combination of functions of the type $H(s,l)$ satisfying the
hypothesis of
Corollary \ref{res-somH}. As a consequence, $r(s,l)$ satisfies (H1)
and
$$
f_{\mu,\a}(s,l)\sim {\sum_{k\in\Z^2}}'
\tfrac{k_{\mu_1}(k+l)_{\mu_2}}{|k|^{s+4}}\ \wt
a_{\a,l}\sim  {\sum_{k\in\Z^2}}'
\tfrac{k_{\mu_1}k_{\mu_2}}{|k|^{s+4}}\ \wt a_{\a,l}
$$ Note that the function $(s,l)\mapsto
h_{\mu,\a}(s,l):={\sum}'_{k\in\Z^2} \tfrac{k_{\mu_1}k_{\mu_2}}
{|k|^{s+4}}\ \wt a_{\a,l}$ satisfies (H2). Thus, Lemma \ref{res-som}
yields
$$
\underset{s=0}{\Res}\ f(s,l) =\sum_{l\in \Z^2} \underset{s=0}{\Res}\
h_{\mu,\a}(s,l)\Tr(\ga^{\a_2}\ga^{\mu_2}\ga^{\a_1}\ga^{\mu_1}).
$$
By Proposition \ref{calculres}, we get $\underset{s=0}{\Res}\
h_{\mu,\a}(s,l) =
\delta_{\mu_1\mu_2}\,\pi\, \wt a_{\a,l}$. Therefore,
$$
\ncint \mathbb{A}^{++}=-\pi \sum_{l\in \Z^2}\wt a_{\a,l}
\Tr(\ga^{\a_2}\ga^{\mu}\ga^{\a_1}\ga_{\mu})=0
$$
according to (\ref{Wick1}).

$(ii)$ By Lemma \ref{formegenerale}, we obtain that
$\ncint \mathbb{A}^{-+}= \underset{s=0}{\Res} \sum_{l\in \Z^2}
\lambda_\sigma
f_{\a,\mu}(s,l) \Tr(\ga^{\a_2}\ga^{\mu_2}\ga^{\a_1}\ga^{\mu_1})$
where $\la_\sigma=-(-i)^2=1$ and
$$
f_{\a,\mu}(s,l):=  {\sum_{k\in\Z^2}}'
\tfrac{k_{\mu_1}(k+l)_{\mu_2}}{|k|^{s+2}|k+l|^2}
\, e^{i \eta \,  k.\Th l} \,\wt a_{\a,l}\,
$$
and $\eta:=\half (\sigma_1-\sigma_2) =-1$. As in the proof of $(i)$,
since
the presence of the phase does not change the fact that $r(s,l)$
satisfies (H1), we
get
$$
f_{\a,\mu}(s,l) \sim {\sum_{k\in\Z^2}}'
\tfrac{k_{\mu_1}(k+l)_{\mu_2}}{|k|^{s+4}}\,e^{i \eta \, k.\Th l}\, \wt
a_{\a,l}:=g_{\a,\mu}(s,l)\, .
$$
Since $\tfrac {1}{2\pi}\Th$ is diophantine,  the functions $
s\mapsto \sum_{l\in\Z^2\backslash \{0\}} g_{\a,\mu}(s,l) $ are
holomorphic at $s=0$ by Theorem \ref{analytic} $3$. As a consequence,
$$
\ncint \mathbb{A}^{-+}=\underset{s=0}\Res\ g_{\a,\mu}(s,0)
\Tr(\ga^{\a_2}\ga^{\mu_2}\ga^{\a_1}\ga^{\mu_1})=\underset{s=0}\Res\
{\sum_{k\in\Z^2}}'
\tfrac{k_{\mu_1}k_{\mu_2}}{|k|^{s+4}}\,\wt
a_{\a,0}\,\Tr(\ga^{\a_2}\ga^{\mu_2}\ga^{\a_1}\ga^{\mu_1}).
$$
Recall from Proposition \ref{res-int} that
$\Res_{s=0}\,{\sum}'_{k\in\Z^2}\,\tfrac{k_{i}k_{j}} {\vert
k\vert^{s+4}}=\delta_{ij}\,\pi$. Thus, again with (\ref{Wick1}),
\begin{align*}
\ncint \mathbb{A}^{-+}=\wt a_{\a,0}\,
\pi\,\Tr(\ga^{\a_2}\ga^{\mu}\ga^{\a_1}\ga_{\mu})=0.
\tag*{\qed}
\end{align*}
\hideqed
\end{proof}

\begin{lemma}
    \label{termeconstantn=2}
Suppose $n=2$ and $\tfrac{1}{2\pi}\Th$ diophantine. For any self-adjoint one-form $A$,
\begin{align*}
\zeta_{D_A}(0)-\zeta_{D}(0)= 0.
\end{align*}
\end{lemma}

\begin{proof}
As in Lemma \ref{double}, we use (\ref{termconstanttilde}) and Lemma
\ref{ncadmoins1}
so the result
follows from Lemma \ref{term-n=2}.
\end{proof}

\subsubsection{Odd dimensional case}

\begin{lemma}
\label{impair}
Suppose $n$ odd and $\tfrac{1}{2\pi}\Th$
diophantine.
Then for any self-adjoint 1-form $A$ and $\sigma\in \{-,+\}^q$ with
$2\leq q\leq n$,
$$
\ncint \mathbb{A}^\sigma = 0\, .
$$
\end{lemma}
\begin{proof} Since $\mathbb{A}^\sigma \in \Psi_1(\A)$, Lemma
\ref{ncint-odd-pdo} with $k=n$ gives the result.
\end{proof}

\begin{corollary}
\label{zetaimpair}
With the same hypothesis of Lemma \ref{impair}, for any self-adjoint
one-form $A$,
$\zeta_{D_A}(0)-\zeta_{D}(0)=0.$
\end{corollary}
\begin{proof} As in Lemma \ref{double}, we use
(\ref{termconstanttilde}) and Lemma \ref{ncadmoins1}
so the result
follows from Lemma \ref{impair}.
\end{proof}

\subsection{Proof of the main result}
\begin{proof}[Proof of Theorem \ref{main}.]
$(i)$ By (\ref{formuleaction}) and Proposition \ref{invariance}, we
get
$$
\SS(\DD_{A},\Phi,\Lambda) \, = \, 4\pi \Phi_{2}\,
    \Lambda^{2}  + \Phi(0) \,
    \zeta_{D_{A}}(0) + \mathcal{O}(\Lambda^{-2}),
$$
where $\Phi_{2}= \half\int_{0}^{\infty} \Phi(t) \, dt$. By Lemma
\ref{termeconstantn=2},
$\zeta_{D_A}(0) - \zeta_{D}(0) = 0$ and from Proposition
\ref{zeta(0)}, $\zeta_{D}(0)=0$,
so we get the result.

\noindent $(ii)$ Similarly,
$\SS(\DD_{A},\Phi,\Lambda) \, =  8 \pi^2\, \Phi_{4}\,
    \Lambda^{4} + \Phi(0) \,
    \zeta_{D_{A}}(0) + \mathcal{O}(\Lambda^{-2})$
with $\Phi_{4}= \half\int_{0}^{\infty} \Phi(t) \, t \, dt$.
Lemma \ref{double} implies that
$\zeta_{D_A}(0) - \zeta_D(0)=-c\,\tau(F_{\mu\nu}F^{\mu\nu})$
and by Proposition \ref{zeta(0)},
$\zeta_{D_A}(0)=-c\,\tau(F_{\mu\nu}F^{\mu\nu})$ leading to the
result.

\noindent $(iii)$ is a direct consequence of (\ref{formuleaction}),
Propositions \ref{zeta(0)}, \ref{invariance}, and Corollary
\ref{zetaimpair}.
\end{proof}

\section{Holomorphic continuation and residues of series of zeta functions}

In the following, the prime in ${\sum}'$ means that we omit terms
with division by zero in the summand.
$B^{n}$ (resp. $S^{n-1}$) is the closed ball (resp.
the sphere) of $\R^n$ with center $0$ and radius 1 and the
Lebesgue measure on $S^{n-1}$ will be noted $dS$.

For any $x=(x_1,\dots,x_n) \in \R^n$ we denote by
$|x|=\sqrt{x_1^2+\dots+x_n^2}$ the
euclidean norm and $|x|_1 :=|x_1|+\dots +|x_n|$.

$\N =\{1, 2,\dots \}$ is the set of positive integers and $\N_0 =\N
\cup \{0\}$ the set of non negative integers.

By  $f(x,y) \ll_{y} g(x)$
uniformly in $x$, we mean that $\vert f(x,y)\vert \leq a(y)
\, \vert g(x) \vert$ for all $x$ and $y$ for some $a(y)>0$.

\subsection{Residues of series and integral}
In order to be able to compute later the residues of certain series,
we prove here the following

\begin{theorem}\label{res-int} Let $P(X)=\sum_{j=0}^{d} P_j(X)
 \in \C[X_1,\cdots,X_n]$ be
a polynomial function where $P_j$
is the homogeneous part of $P$ of degree $j$. The function
$$ \zeta^P(s):={\sum}'_{k\in\Z^n} \tfrac{P(k)}{|k|^s}, \,\,\, s \in \C
$$ has a meromorphic continuation to the whole complex plane
$\C$. \par Moreover $\zeta^P(s)$ is not entire if and only if
$\mathcal{P}_P:= \{j \  : \
\int_{u\in S^{n-1}} P_j(u)\, dS(u)\neq 0 \}\neq \varnothing$.
In that case, $\zeta^P$ has only simple poles at the points $j+n$,
$j\in \mathcal{P}_{P}$, with
$$
\underset{s=j+n}{\Res} \, \zeta^P(s) = \int_{u\in S^{n-1}} P_j(u)\,
dS(u).
$$
\end{theorem}
\medskip

The proof of this theorem is based on the following lemmas.

\begin{lemma}\label{majPQs}
For any polynomial $P\in \C[X_1,\dots,X_n]$ of total
degree $\delta (P):=\sum_{i=1}^n deg_{X_i}P$ and any $\alpha \in
\N_0^n$, we have
$$\partial^{\alpha} \left(P(x) |x|^{-s}\right)\ll_{P, \alpha, n}
  (1+|s|)^{|\alpha|_1} |x|^{-\sigma -|\alpha|_1 +\delta (P)}$$
uniformly  in  $x \in \R^n$ verifying $|x|\geq 1$, where
$\sigma=\Re(s)$.
\end{lemma}
\begin{proof}
By linearity, we may assume without loss of generality
that $P(X)=X^{\gamma}$ is a monomial. It is easy to
prove (for example by induction on $|\alpha|_1$) that for all $\alpha
\in \N_0^n$ and
$x \in \R^n \setminus \{0\}$:
$$
\partial^{\alpha} \left(|x|^{-s}\right)
=\alpha! \sum_{\genfrac{}{}{0pt}{2}{\beta, \mu \in \N_0^n}{\beta + 2
\mu =
\alpha }}
\genfrac(){0pt}{1}{-s/2}{|\beta|_1 +|\mu|_1}
\tfrac{(|\beta|_1+|\mu|_1)!}{\beta! ~ \mu!}
\tfrac{x^{\beta}}{|x|^{\sigma+2(|\beta|_1+|\mu|_1)}}.
$$
It follows that for all $\alpha \in \N_0^n$, we have uniformly in  $x
\in \R^n$
verifying $|x|\geq 1$:
\begin{equation}\label{majQs}
\partial^{\alpha} \left(|x|^{-s}\right) \ll_{\alpha,  n}
  (1+|s|)^{|\alpha|_1} |x|^{-\sigma -|\alpha|_1}\,.
\end{equation}

By Leibniz formula and (\ref{majQs}), we have uniformly in $x\in \R^n$
verifying $|x|\geq 1$:
\begin{align*}
\partial^{\alpha} \left(x^\gamma |x|^{-s}\right) & = \, \sum_{\beta
\leq \alpha}
\genfrac(){0pt}{1}{\alpha}{\beta} \,  \partial^{\beta} (x^\gamma)~
\partial^{\alpha
  -\beta} \left(|x|^{-s}\right) \\
& \ll_{\gamma, \alpha, n}  \sum_{\beta \leq \alpha; \beta \leq
\gamma} x^{\gamma
-\beta}~ (1+|s|)^{|\alpha|_1-|\beta|_1}~
|x|^{-\sigma-|\alpha|_1+|\beta|_1} \\
& \ll_{\gamma, \alpha, n}  (1+|s|)^{|\alpha|_1}~
|x|^{-\sigma-|\alpha|_1+|\gamma|_1}.
\tag*{\qed}
\end{align*}
\hideqed
\end{proof}

\begin{lemma}\label{deltaf} Let $P\in \C[X_1,\dots,X_n]$ be a
  polynomial of degree $d$.
Then, the difference
$$
\Delta_P(s):={\sum}'_{k\in\Z^n}
\tfrac{P(k)}{|k|^s}-\int_{\R^n\setminus B^{n}}
\tfrac{P(x)}{|x|^s} \, dx
$$
which is defined for $\Re(s)>d+n$, extends holomorphically on the
whole complex plane
$\C$.
\end{lemma}
\begin{proof}
We fix in the sequel a function $\psi\in C^\infty(\R^n ,\R)$
verifying for all $x\in
\R^n$
$$
0\leq \psi(x) \leq 1, \quad \psi(x)=1 \text{ if }|x|\geq 1
\quad \text{and} \quad \psi(x)=0 \text{ if } |x|\leq 1/2.
$$
The function $f(x,s) :=  \psi (x)~P(x)~ |x|^{-s}$, $x\in \R^n$ and
$s\in \C$,
is in ${\cal
  C}^\infty (\R^n \times \C)$  and depends holomorphically on $s$.

Lemma \ref{majPQs} above shows that $f$ is
a ``gauged symbol'' in the terminology of \cite[p. 4]{GSW}.
Thus \cite[Theorem 2.1]{GSW} implies that $\Delta_P(s)$ extends
holomorphically on the whole complex plane $\C$. However, to be
complete, we will give here a short proof of Lemma \ref{deltaf}:
\par It follows from the classical Euler--Maclaurin formula that for
any
function $h: \R \rightarrow \C$ of class ${\cal C}^{N+1}$
verifying $ \lim_{|t|\rightarrow +\infty} h^{(k)}(t)=0$ and $\int_{\R}
|h^{(k)} (t)| ~dt <+\infty$ for any
$k=0 \dots,N+1$, that we have
$$
\sum_{k\in \Z} h(k) = \int_{\R} h(t) + \tfrac{(-1)^N}{(N+1)!}
\int_{\R} B_{N+1}(t)~h^{(N+1)}(t) ~dt
$$
where $B_{N+1}$ is the Bernoulli function of
order $N+1$ (it is a bounded periodic function).
\par Fix $m' \in \Z^{n-1}$ and $s\in
\C$. Applying this to the function $h(t):= \psi (m',t)~P(m',t)
~|(m',t)|^{-s}$ (we use Lemma \ref{majPQs} to verify hypothesis),
 we obtain that for any $N\in \N_0$:
\begin{equation}\label{*1}
\sum_{m_n \in \Z} \psi (m',m_n) ~P(m',m_n) ~|(m',m_n)|^{-s}
= \int_{\R} \psi(m',t)
~P(m',t) ~|(m',t)|^{-s} ~dt +{\cal R}_N(m';s)
\end{equation}
where $ {\cal R}_N(m';s):=\tfrac{(-1)^N}{(N+1)!} \int_{\R} B_{N+1}(t)~
\tfrac{\partial^{N+1}}{{\partial x_n}^{N+1}} \left(\psi(m',t)~P(m',t)
~|(m',t)|^{-s}\right)~dt$.\\
By Lemma \ref{majPQs},
$$
\int_{\R} {\Big |} B_{N+1}(t)~
\tfrac{\partial^{N+1}}{{\partial x_n}^{N+1}} \left(\psi (m',t)
~P(m',t)~|(m',t)|^{-s}\right){\Big |}~dt \ll_{P,n, N}
(1+|s|)^{N+1}~ (|m'|+1)^{-\sigma
-N+ \delta(P)}.
$$
Thus $ \sum_{m' \in \Z^{n-1}} {\cal R}_N(m';s)$
converges absolutely and define a holomorphic function in
the half plane $\{\sigma
=\Re (s) > \delta(P)+n-N\}$. \par Since $N$ is an arbitrary integer,
by letting
$N\rightarrow \infty$ and using $(\ref{*1})$ above, we conclude that:
$$s\mapsto \sum_{(m',m_n) \in \Z^{n-1}\times \Z} \psi (m',m_n)
~P(m',m_n)
~|(m',m_n)|^{-s}-\sum_{m'
  \in \Z^{n-1}} \int_{\R} \psi (m',t) ~P(m',t)~|(m',t)|^{-s}~dt$$ has
a holomorphic continuation to the whole complex plane $\C$.\par
After $n$ iterations, we obtain that
$$s\mapsto {\sum}_{m\in \Z^{n}} \psi(m)~P(m)
~|m|^{-s}-\int_{\R^n} \psi(x)~P(x) ~|x|^{-s}~dx$$ has a
holomorphic continuation to the whole $\C$.\\
To finish the proof of Lemma \ref{deltaf}, it is enough to notice
that:

\hspace{1 cm} $\bullet$ $\psi(0)=0$ and  $\psi (m)=1$, $\forall m\in
\Z^n\setminus \{0\}$;

\hspace{1 cm} $\bullet$ $s\mapsto \int_{B^n} \psi(x)~P(x)~|x|^{-s}~dx
= \int_{\{x\in \R^n : 1/2\leq |x|\leq 1\}} \psi(x)~P(x)~|x|^{-s}~dx$
is a holomorphic
function on $\C$.
\end{proof}

\begin{proof}[Proof of Theorem \ref{res-int}]

Using the polar decomposition of the volume form
$dx=\rho^{n-1}\,d\rho\, dS$ in $\R^n$, we get for $\Re(s)>d+n$,
$$
\int_{\R^n \setminus B^{n}} \tfrac{P_j(x)}{|x|^s}dx =
\int_{1}^{\infty}
\tfrac{\rho^{j+n-1}}{\rho^s}\int_{S^{n-1}} P_j(u)\, dS(u) =
\tfrac{1}{j+n-s}
\int_{S^{n-1}} P_j(u)\, dS(u).
$$
Lemma \ref{deltaf} now gives the result.
\end{proof}

\subsection{Holomorphy of certain series}
Before stating the main result of this section, we give first in the
following some
preliminaries from Diophantine approximation theory:

Let us recall Definition \ref{ba}.
Let $\delta >0$. A vector $a \in \R^n$ is said to be
$\delta-$diophantine
if there exists $c >0$ such that $|q . a -m| \geq c \,|q|^{-\delta}$,
$\forall q \in \Z^n \setminus \set{0}$ and $\forall m \in \Z$. 
We note ${\cal BV}(\delta )$ the set of $\delta-$diophantine
vectors and ${\cal
BV} :=\cup_{\delta >0} {\cal BV}(\delta)$ the set of diophantine vectors.

A matrix $\Th \in {\cal M}_{n}(\R)$ (real $n \times n$ matrices)
will be
said to be diophantine if there
exists $u \in \Z^n$ such that  ${}^t\Th (u)$ is a diophantine
vector of $\R^n$.

{\bf Remark.} A classical result from Diophantine approximation
asserts
that for all $\delta >n$, the Lebesgue measure of
$\R^n \setminus {\cal BV}(\delta)$ is zero (i.e almost any element of
$\R^n$ is $\delta-$diophantine).
\par Let $\Th \in {\cal M}_n(\R)$. If its row
of index $i$ is a diophantine vector of $\R^n$ (i.e. if $L_i
\in {\cal BV}$)
then ${}^t \Th (e_i) \in {\cal BV}$ and thus $\Th$ is a diophantine matrix. It
follows that almost any matrix of ${\cal M}_n(\R)\approx \R^{n^2}$ is
diophantine.\par

\medskip

The goal of this section is to show the following
\begin{theorem}\label{analytic}
Let $P\in \C[X_1,\cdots,X_n]$ be a homogeneous polynomial of degree
$d$ and let $b$ be in
$\mathcal{S}(\Z^{n} \times \dots \times \Z^{n})$ ($q$ times,
$q\in\N$). Then,

(i)
Let $a \in \R^n$. We define $f_a(s):={\sum}'_{k\in \Z^n}
\frac{P(k)}{|k|^s}\,
e^{2\pi i k.a}$.

\quad 1.
If $a\in \Z^n$, then $f_a$ has a meromorphic
continuation to the whole complex plane
$\C$.\\ Moreover
$f_a$ is not entire if and only if
$\int_{u\in S^{n-1}} P(u)\, dS(u)\neq 0$. In that case, $f_a$
has only a simple pole at the point $d+n$, with
$\underset{s=d+n}{\Res} \, f_a(s) = \int_{u\in
  S^{n-1}} P(u)\, dS(u)$.

\quad 2.
If $a\in \R^n\setminus \Z^n$, then $f_a(s)$ extends holomorphically
to the whole
complex plane $\C$.

(ii)
 Suppose that  $\Th \in {\cal M}_{n}(\R)$ is diophantine.
 For any $(\eps_i)_i\in \{-1,0,1\}^{q}$, the function
$$
g(s):={\sum}_{l\in (\Z^n)^{q}} \, b(l) \,f_{\Th\,
\sum_i \eps_i l_i}(s)
$$
extends meromorphically to the whole complex plane  $\C$
with only one possible pole
on $s= d+n$.\par Moreover, if we set
${\cal Z}:=\{l\in(\Z^n)^{q} \ : \ \sum_{i=1}^q \eps_i
l_i= 0\}$ and $V:=\sum_{l\in {\cal Z}} \, b(l)$, then

1. If $V\int_{S^{n-1}} P(u)\, dS(u)\neq 0$, then $s=d+n$ is a
simple pole of $g(s)$ and
$$
\underset{s=d+n}{\Res} \, g(s) = V\,
\int_{u\in
  S^{n-1}} P(u)\, dS(u).
$$

2. If $V\int_{S^{n-1}} P(u)\, dS(u)=0$, then $g(s)$ extends
holomorphically to the
whole complex plane $\C$.

(iii)
Suppose that $\Th \in \mathcal{M}_n(\R)$ is diophantine. For any $(\eps_i)_i\in
\{-1,0,1\}^{q}$, the function
$$
g_0(s):={\sum}_{l\in (\Z^n)^{q}\setminus
  {\cal Z}} \,
b(l)\,f_{\Th\, \sum_{i=1}^q \eps_i l_i}(s)
$$
where ${\cal Z}:=\{l\in(\Z^n)^{q} \ : \
\sum_{i=1}^q \eps_i l_i= 0\}$ extends holomorphically to the whole
complex plane $\C$.
\end{theorem}

{\it Proof of Theorem \ref{analytic}}:
First we remark that

$\hspace{1cm}$ If $a \in \Z^n$ then
$f_a(s)={\sum}'_{k\in \Z^n} \frac{P(k)}{|k|^s}$. So,
the point $(i.1)$ follows from Theorem \ref{res-int};

$\hspace{1cm}$
$ g(s):=\sum_{l\in (\Z^n)^{q}\setminus {\cal Z}} \, b(l) \,f_{\Th\,
  \sum_i \eps_i l_i}(s)
+\left(\sum_{l\in {\cal Z}} \, b(l)\right) {\sum}'_{k\in \Z^n}
\frac{P(k)}{|k|^s}$. Thus, the point $(ii)$ rises

$\hspace{1.1cm}$easily from $(iii)$ and Theorem \ref{res-int}.

So, to complete the proof, it remains to prove the items $(i.2)$ and
$(iii)$.\par
The direct
proof of $(i.2)$ is easy but is not sufficient to deduce $(iii)$ of
which the proof is
more delicate and requires a more precise (i.e. more effective)
version of
$(i.2)$. The next lemma gives such crucial version, but before,
let us give some notations:
$$
{\cal F}:=\{\tfrac{P(X)}{(X_1^2+\dots +X_n^2+1)^{r/2}}
 \, :\,
P(X) \in \C[X_1,\dots, X_n] {\mbox { and }} r \in \N_0\}.
$$
\hspace{1cm}
  We set $g=$deg$(G) =$deg$(P) -r \in \Z$, the degree of
$G=\frac{P(X)}{(X_1^2+\dots +X_n^2+1)^{r/2}}\in
{\cal F}$.

\hspace{1cm} By convention we set deg$(0)=-\infty$.

\begin{lemma}\label{ieffective}
Let $a \in \R^n$. We assume that $d\left(a . u, \Z\right):=
\inf_{m\in \Z} |a . u  -m| >0$ for some $u \in \Z^n$.
For all $G\in {\cal F}$, we define formally,
\begin{align*}
    F_0(G;a;s):={\sum}'_{k\in\Z^n}
\tfrac{G(k)}{|k|^{s}}\, e^{2\pi i \,k . a}  \quad \text{and} \quad
F_1(G;a;s):={\sum}_{k \in \Z^n}
\tfrac{G(k)}{(|k|^2+1)^{s/2}} \,e^{2\pi i \,k . a} .
\end{align*}

Then for all $N\in \N$, all $G\in {\cal F}$ and all $i\in \{0,1\}$,
there exist
positive constants $C_i:=C_i(G,N,u)$, $B_i:=B_i(G,N,u)$ and
$A_i:=A_i(G,N,u)$ such
that $s\mapsto F_i(G;\a;s)$ extends  holomorphically to the half-plane
$\{\Re(s)>-N\}$ and verifies in it:
$$F_i(G;a;s)\leq C_i (1+|s|)^{B_i} \,
\big(d\left(a . u, \Z\right)\big)^{-A_i}.$$
\end{lemma}
\begin{remark} The important point here is that we obtain an explicit
  bound of $F_i(G;\a;s)$ in $\{\Re(s)>-N\}$ which depends on the
  vector $a$ only through $d(a.u,\Z)$, so depends on $u$ and
indirectly on $a$ (in the sequel, $a$ will vary). In particular the
constants $C_i:=C_i(G,N,u)$,
$B_i=B_i(G,N)$ and $A_i:=A_i(G,N)$ do not depend on the vector $a$
but only on
$u$. {\it This is crucial for the proof of items $(ii)$ and
$(iii)$ of Theorem \ref{analytic}!}
\end{remark}

\subsubsection{Proof of Lemma \ref{ieffective} for $i=1$:}
Let $N\in \N_0$ be a fixed integer, and set $g_0:= n+N+1$.\\
We will prove Lemma \ref{ieffective} by induction on $g=$deg$(G)\in
\Z$. More
precisely, in order to prove case $i=1$, it suffices to
prove that:

\hspace{1cm} Lemma \ref{ieffective} is true for all $G\in {\cal F}$
  verifying deg$(G)\leq -g_0$.

\hspace{1cm} Let $g\in \Z$ with $g\geq -g_0+1$.
If Lemma \ref{ieffective} is true for all $G\in {\cal F}$
  such that deg$(G)\leq g -1$,

\hspace{1cm} then it is also true for all
$G\in {\cal F}$ satisfying deg$(G)= g$.

$\bullet$ {Step 1: Checking Lemma
\ref{ieffective} for
deg$(G) \leq -g_0:= -(n+N+1)$.}\\
Let $G(X)=\frac{P(X)}{(X_1^2+\dots +X_n^2+1)^{r/2}} \in {\cal F}$
verifying deg$(G)\leq -g_0$.
It is easy to see that we have uniformly in  $s=\sigma +i\tau  \in \C$
and in $k \in
\Z^n$:
\begin{align*}
\tfrac{|G(k) \,e^{2\pi i \,k .
a}|}{(|k|^2+1)^{\sigma/2}}=&\tfrac{|P(k)|}{(|k|^2+1)^{(r+\sigma)/2}}\ll_G
\tfrac{1}{(|k|^2+1)^{(r+\sigma-deg(P))/2}}
\ll_G  \tfrac{1}{(|k|^2+1)^{(\sigma-deg(G))/2}}\ll_G
\tfrac{1}{(|k|^2+1)^{(\sigma+g_0)/2}}.
\end{align*}
It follows that $F_1(G;a;s)=\sum_{k \in \Z^n }
\frac{G(k)}{(|k|^2+1)^{s/2}} \,e^{2\pi i
\, k . a}$ converges absolutely and defines a holomorphic function in
the half plane
$\{\sigma >-N\}$. Therefore, we have for any $s\in \{\Re(s) >-N\}$:
$$|F_1(G;a;s)|\ll_G \sum_{k \in \Z^n }
\tfrac{1}{(|k|^2+1)^{(-N+g_0)/2}}\ll_G \sum_{k \in \Z^n }
\tfrac{1}{(|k|^2+1)^{(n+1)/2}}\ll_G 1.$$
Thus, Lemma \ref{ieffective} is true when deg$(G)\leq -g_0$.

$\bullet$ { Step 2: Induction. }\\
Now let $g\in \Z$ satisfying $g\geq -g_0+1$ and suppose that Lemma
\ref{ieffective} is
valid for all $G\in {\cal F}$ verifying deg$(G) \leq g-1$. Let $G\in
{\cal F}$ with deg$(G)=g$. We will prove that  $G$ also
verifies conclusions of Lemma \ref{ieffective}:\\
There exist $P\in \C[X_1,\dots,X_n]$ of degree $d\geq 0$ and
$r\in \N_0$ such that $G(X)=\frac{P(X)}{(X_1^2+\dots
+X_n^2+1)^{r/2}}$ and
$g=$deg$(G)=d-r$.\\
Since $G(k)\ll (|k|^2+1)^{g/2}$ uniformly in $k \in \Z^n$, we deduce
that
$F_1(G;a;s)$ converges absolutely in $\{\sigma=\Re(s)>n+g\}$.\\
Since $k\mapsto k+u$ is a bijection from  $\Z^n $ into $\Z^n$, it
follows
that we also have for $\Re(s)>n+g$
\begin{align*}
F_1(G;a;s)&=\sum_{k \in \Z^n }
\tfrac{P(k)}{(|k|^2+1)^{(s+r)/2}} \,e^{2\pi i \,k . a}
= \sum_{k \in \Z^n }
\tfrac{P(k+u)}{(|k+u|^2+1)^{(s+r)/2}} \,  e^{2\pi i \, (k+u) .  a}\\
&= e^{2\pi i \,u . a} \sum_{k \in \Z^n } \tfrac{P(k+u)}{(|k|^2+2 k . u
+|u|^2+1)^{(s+r)/2}} \,e^{2\pi i \,k . a}\\
&= e^{2\pi i \,u . a} \sum_{\alpha \in \N_0^n;
|\alpha|_1=\alpha_1+\dots+\alpha_n \leq
d} \tfrac{u^{\alpha}}{\alpha !} \sum_{k \in \Z^n }
\tfrac{\partial^{\alpha}
P(k)}{(|k|^2+2 k . u +|u|^2+1)^{(s+r)/2}} \, e^{2\pi i \,k . a}\\
&= e^{2\pi i \,u . a} \sum_{|\alpha|_1\leq d}
\tfrac{u^{\alpha}}{\alpha !} \sum_{k \in
\Z^n } \tfrac{\partial^{\alpha} P(k)}{(|k|^2+1)^{(s+r)/2}}
\big(1+\tfrac{2 k . u
+|u|^2}{(|k|^2+1)}\big)^{-(s+r)/2} \, e^{2\pi i \,k . a}.
\end{align*}
Let $M:= \sup(N+n+g, 0)\in \N_0$. We have uniformly in $k \in \Z^n$
$$
\big(1+\tfrac{2 k . u + |u|^2}{(|k|^2+1)}\big)^{-(s+r)/2}=
\sum_{j=0}^M \genfrac(){0pt}{1}{-(s+r)/2}{j} \tfrac{\left(2 k . u +
|u|^2\right)^j}{(|k|^2+1)^j}+
O_{M, u}\big(  \tfrac{(1+|s|)^{M+1}}{(|k|^2+1)^{(M+1)/2}}\big).
$$
Thus, for
$\sigma =\Re(s)>n+d$,
\begin{eqnarray}\label{f0dev}
F_1(G;a;s)&=& e^{2\pi i \,u . a} \sum_{|\alpha|_1 \leq d}
\tfrac{u^{\alpha}}{\alpha !}
\sum_{k \in \Z^n } \tfrac{\partial^{\alpha}
P(k)}{(|k|^2+1)^{(s+r)/2}} \big(1+\tfrac{2
k . u +|u|^2}{(|k|^2+1)}\big)^{-(s+r)/2}
e^{2\pi i \,k . a}\nonumber \\
&=& e^{2\pi i \,u . a} \sum_{|\alpha|_1 \leq d} \sum_{j=0}^M
\tfrac{u^{\alpha}}{\alpha
!}  \genfrac(){0pt}{1}{-(s+r)/2}{j} \sum_{k \in \Z^n }
\tfrac{\partial^{\alpha} P(k) \left(2 k . u
+|u|^2\right)^j } {(|k|^2+1)^{(s+r+2j)/2}} \, e^{2\pi i
  \,k . a}\nonumber \\
& & \hspace{1cm}+O_{G, M,u} \big((1+|s|)^{M+1}\sum_{k \in \Z^n }
\tfrac{1}{(|k|^2+1)^{(\sigma
+M+1-g)/2}}\big).
\end{eqnarray}
Set $I:=\left\{(\alpha ,j) \in \N_0^n \times \{0,\dots,M\} \mid
|\alpha|_1 \leq d\right\}$ and $I^*:=I\setminus \set{(0,0)}$.\\
Set also $ G_{(\alpha ,j);u}(X):= \tfrac{\partial^{\alpha} P(X)
\left(2 X . u
+|u|^2\right)^j }
{(|X|^2+1)^{(r+2j)/2}}\in {\cal F}$ for all $(\alpha ,j) \in I^*$.\\
Since $M\geq N+n+g$, it follows from (\ref{f0dev}) that
\begin{eqnarray}\label{crucial1}
(1-e^{2\pi i \,u . a})~F_1(G;a;s)= e^{2\pi i \,u . a} \sum_{(\alpha,
j)\in I^*}
\tfrac{u^{\alpha}}{\alpha !}  \genfrac(){0pt}{1}{-(s+r)/2}{j}
F_1\left(G_{(\alpha ,j);u};\a;s\right)  +R_N(G; a; u; s)
\end{eqnarray}
where $s\mapsto R_N(G; a; u; s)$ is a holomorphic function in the
half plane
$\{\sigma =\Re(s) >-N\}$, in which it satisfies the bound
$R_N(G; a; u; s)\ll_{G,N,u} 1 $.\\
Moreover it is easy to see that, for any $(\alpha, j)\in I^*$,
$$
\text{deg}\hspace{-.06cm}\left(G_{(\alpha
,j);u}\right)=\text{deg}\hspace{-.05cm}\left(\partial^{\alpha}
P \right)+j -(r+2j)\leq d-|\alpha|_1 +j -(r+2j)=g-|\alpha|_1 -j\leq
g-1.
$$
Relation (\ref{crucial1}) and the induction hypothesis imply then that
\begin{equation}\label{crucial2}
(1-e^{2\pi i \,u . a })~F_1(G;a;s) {\mbox { verifies the conclusions
of Lemma \ref{ieffective}}}.
\end{equation}
Since $ |1-e^{2\pi i \,u . a}|=2|\sin(\pi u . a)|\geq d\left(u  . a,
\Z\right)$,
then (\ref{crucial2}) implies that $F_1(G;a;s)$ satisfies
conclusions of Lemma \ref{ieffective}. This completes the induction
and the proof for $i=1$.

\subsubsection{Proof of Lemma \ref{ieffective} for $i=0$:}
Let  $N\in \N$ be a fixed integer.
Let $G(X)=\frac{P(X)}{(X_1^2+\dots +X_n^2+1)^{r/2}}
\in {\cal F}$ and $g=$ deg$(G)=d-r$ where $d\geq 0$ is the degree of
the polynomial $P$.
Set also $M:=\sup (N+g+n, 0)\in \N_0$.\par
Since $P(k)\ll |k|^d$ for $k \in \Z^n\setminus
\set{0}$, it follows that $F_0(G;a;s)$ and $F_1(G;a;s)$ converge
absolutely in the
half plane
$\{\sigma =\Re(s)>n+g\}$.\\
Moreover, we have for $s=\sigma +i\tau  \in \C$ verifying $\sigma
>n+g$:
\begin{align}\label{crucial3}
F_0(G;a;s)&= \sum_{k \in \Z^n \setminus \set{0}}
\tfrac{G(k)}{(|k|^2+1-1)^{s/2}} \, e^{2\pi i \,k . a}
={\sum_{k \in \Z^n }}' \tfrac{G(k)}{(|k|^2+1)^{s/2}}
\left(1-\tfrac{1}{|k|^2+1}\right)^{-s/2} \,
e^{2\pi i \,k . a} \nonumber \\
&= {\sum_{k \in \Z^n }}' \,  \sum_{j=0}^M
\genfrac(){0pt}{1}{-s/2}{j} (-1)^j
\tfrac{G(k)}{(|k|^2+1)^{(s+2j)/2}} \,
e^{2\pi i \,k . a} \nonumber\\
&\hspace{2cm} +O_M\big((1+|s|)^{M+1} {\sum_{k \in \Z^n }}'
\tfrac{|G(k)|}{(|k|^2+1)^{(\sigma +2M+2)/2}}\big)\nonumber \\
&=\sum_{j=0}^M  \genfrac(){0pt}{1}{-s/2}{j}
(-1)^j F_1(G;a; s+2j) \nonumber\\
&\hspace{2cm}+ O_M\big[(1+|s|)^{M+1} \big(1+{\sum_{k \in \Z^n }}'
\tfrac{|G(k)|}{(|k|^2+1)^{(\sigma +2M+2)/2}}\big)\big].
\end{align}
In addition we have  uniformly in $s=\sigma +i\tau \in \C$ verifying
$\sigma >-N$,
$$
{\sum_{k \in \Z^n }}'
\tfrac{|G(k)|}{(|k|^2+1)^{(\sigma +2M+2)/2}}\ll {\sum_{k \in \Z^n }}'
\tfrac{|k|^g}{(|k|^2+1)^{(-N +2M+2)/2}}\ll {\sum_{k \in \Z^n }}'
\tfrac{1}{|k|^{n+1}}<+\infty.
$$
So (\ref{crucial3}) and Lemma \ref{ieffective} for
$i=1$ imply that Lemma \ref{ieffective} is also true for $i=0$. This
completes the
proof of Lemma \ref{ieffective}. \qed

\subsubsection{Proof of item $(i.2)$ of Theorem \ref{analytic}:}
Since  $a\in \R^n\setminus \Z^n$, there exists $i_0\in
\{1,\dots,n\}$ such that $a_{i_0}\not \in \Z$.
In particular $d(a . e_{i_0}, \Z)=d(a_{i_0},\Z)>0$. Therefore,  $a$
satisfies the assumption of Lemma \ref{ieffective} with  $u=e_{i_0}$.
Thus, for all $N\in \N$, $s\mapsto f_{a}(s)=F_0(P;a;s)$ has a
holomorphic
continuation to the half-plane $\{\Re (s)>-N\}$. It follows, by
letting
$N\rightarrow \infty$, that $s\mapsto f_{a}(s)$ has a holomorphic
continuation
to the whole complex plane $\C$.

\subsubsection{Proof of item $(iii)$ of Theorem \ref{analytic}:}
Let $\Th\in {\cal M}_n(\R)$, $(\eps_i)_i\in \{-1,0,1\}^{q}$ and $b
\in {\cal
S}(\Z^n\times \Z^n)$. We assume that
$\Th$ is a diophantine matrix. Set ${\cal Z}:= \set{
l=(l_1,\dots,l_q)\in (\Z^n)^q \,:  \, \sum_i \eps_i l_i =0}$ and
$P\in \C[X_1,\dots,X_n]$ of degree  $d\geq 0$.\\
It is easy to see that for $\sigma >n+d$:
\begin{align*}
\sum_{l\in (\Z^n)^{q}\setminus {\cal Z}} \, |b(l)| \,  {\sum_{k \in
\Z^n }}'
\tfrac{|P(k)|}{|k|^{\sigma}} \,|e^{2\pi i \,k . \Th\, \sum_i \eps_i
l_i }| & \ll_P  \sum_{l\in (\Z^n)^q\setminus {\cal Z}} |b (l)| \,
{\sum_{k \in \Z^n }}'
\tfrac{1}{|k|^{\sigma -d}}
\ll_{P,\sigma} \sum_{l \in (\Z^n)^q \setminus {\cal Z}} |b (l)|\\
&<+\infty.
\end{align*}
So
$$
g_0(s):=\sum_{l \in (\Z^n)^q \setminus {\cal Z}}
b (l) \, f_{\Th\, \sum_i \eps_i l_i}(s)= \sum_{l \in (\Z^n)^q
\setminus
  {\cal Z}}  b(l) \, {\sum_{k \in \Z^n }}'
\tfrac{P(k)}{|k|^{s}} e^{2\pi i \,k . \Th \, \sum_i \eps_i l_i}
$$
converges absolutely in the half plane $\{\Re(s) >n+d\}$.\\
Moreover with the notations of Lemma \ref{ieffective}, we have for
all $s=\sigma
+i\tau \in \C$ verifying $\sigma >n+d$:
\begin{equation}\label{crucial4}
g_0(s)=\sum_{l \in (\Z^n)^q\setminus {\cal Z}}
b(l) f_{\Th\, \sum_i \eps_i l_i}(s)=\sum_{l
\in (\Z^n)^q  \setminus {\cal Z}} b(l) F_0(P; \Th\, {\sum}_i \eps_i
l_i; s)
\end{equation}
But $\Th$ is diophantine, so there exists $u \in \Z^n$ and
$\delta ,c>0$ such
$$
|q . \,{}^t\Th u -m| \geq c\, (1+|q|)^{-\delta},\,  \forall q \in
\Z^n \setminus \set{0},
\, \forall m\in \Z.
$$
We deduce that $\forall l \in (\Z^n)^q \setminus {\cal Z},$
$$
\quad |\big(\Th\, {\sum}_i \eps_i l_i \big) .  u -m|= |\big({\sum}_i
\eps_i l_i\big) . {}^t\Th u -m| \geq c \,
\big(1+|{\sum}_i \eps_i l_i| \big)^{-\delta} \geq c \,
(1+|l|)^{-\delta}.
$$
It follows that there exists $u \in \Z^n$, $\delta >0$ and $c>0$ such
that
\begin{equation}\label{hypothesisOK!}
\forall l \in (\Z^n)^q \setminus {\cal Z}, \quad d\big((\Th\,
{\sum}_i \eps_i l_i) .  u;
\Z\big) \geq c \, (1+|l|)^{-\delta}.
\end{equation}
{\it Therefore, for any $l \in (\Z^n)^q \setminus {\cal Z}$,
the vector $a=\Th\, \sum_i \eps_i l_i$
verifies the assumption of Lemma \ref{ieffective} with the same $u$.
Moreover
$\delta$ and $c$ in
(\ref{hypothesisOK!}) are also independent on $l$}.\\
We fix now $N\in \N$. Lemma \ref{ieffective} implies that there exist
positive
constants $C_0:=C_0(P,N,u)$, $B_0:=B_i(P,N,u)$ and $A_0:=A_0(P,N,u)$
such that for all
$l \in (\Z^n)^q \setminus {\cal Z}$, $s\mapsto F_0(P; \Th\, \sum_i
\eps_i l_i; s)$ extends
{\it holomorphically} to the half plane $\{\Re(s)>-N\}$ and verifies
in it the bound
$$
F_0(P;\Th\, \sum_i \eps_i l_i; s)\leq C_0 \left(1+|s|\right)^{B_0} \,
d\big((\Th\, {\sum}_i \eps_i l_i) .  u; \Z\big) ^{-A_0}.
$$
This and (\ref{hypothesisOK!}) imply that for any compact set $K$
included
in the half plane $\{\Re(s)>-N\}$, there exist two constants
$C:=C(P,N,c, \delta,u, K)$ and  $D
:=D(P,N,c,\delta, u)$ (independent on $l \in (\Z^n)^q \setminus {\cal
Z}$) such that
\begin{equation}\label{crucialbig}
\forall s\in K {\mbox { and }} \forall l \in (\Z^n)^q \setminus {\cal
Z}, \quad F_0(P; \Th\,
\sum_i \eps_i l_i; s)\leq C \left(1+|l|\right)^{D}.
\end{equation}
It follows that $s\mapsto \sum_{l \in (\Z^n)^q \setminus {\cal Z}}
b(l) F_0(P;
\Th\, {\sum}_i \eps_i l_i;s)$
has a holomorphic continuation to the half plane
$\{\Re(s)>-N\}$.\\
This and ( \ref{crucial4}) imply that $
s\mapsto g_0(s)=\sum_{l \in (\Z^n)^q \setminus {\cal Z}} b(l)
f_{\Th\, \sum_i \eps_i l_i}(s)$
has a holomorphic continuation to
$\{\Re(s)>-N\}$. Since $N$ is an arbitrary integer, by letting
$N\rightarrow \infty$,
it follows that $s\mapsto g_0(s)$ has a holomorphic continuation to
the whole
complex plane $\C$ which completes the proof of the theorem. \qed

\begin{remark}
\label{remdiophantine}
By equation (\ref{crucial2}), we see that a Diophantine
condition is sufficient to get Lemma \ref{ieffective}. Our Diophantine
condition appears also (in equivalent form) in Connes
\cite[Prop. 49]{NCDG} (see Remark 4.2 below). The following heuristic
argument shows that our condition seems to be necessary in order
to get the result of Theorem \ref{analytic}:\\
For simplicity we assume $n=1$ (but the argument extends easily to
any $n$).\\
Let $\theta \in \R \setminus \Q$. We know (see this reflection
formula in \cite[p. 6]{Elizaldebook})
that for any $l\in \Z\setminus\{0\}$,
$$
g_{\theta l}(s):={\sum_{k\in \Z}}' \, \tfrac{e^{2\pi i \theta l
k}}{|k|^s}=
\tfrac{\pi^{s-1/2}}{\Gamma (\tfrac{1-s}{2})}{\Gamma (\tfrac{s}{2})}
~h_{\theta l}(1-s) {\mbox { where }}
h_{\theta l}(s):={\sum_{k\in \Z}}'\,\tfrac{1}{|\theta l+ k|^s}.
$$
So, for any $(a_l) \in {\cal S}(\Z)$, the existence of meromorphic
continuation of
$g_0(s):=\sum_{l\in \Z} ' a_l \, g_{\theta l}(s)$ is equivalent to
the
existence of meromorphic continuation of
$$
h_0(s):={\sum_{l\in \Z}}' a_l \, h_{\theta l}(s)=
{\sum_{l\in \Z}} ' a_l \, {\sum_{k\in \Z}} ' \tfrac{1}{|\theta l+
k|^s}.
$$
So, for at least one $\sigma_0 \in \R$, we must have
$\tfrac{|a_l|}{|\theta l+ k|^{\sigma_0}} = O(1)
{\mbox { uniformly in }} k, l \in \Z^*.$

It follows that for any $(a_l) \in {\cal S}(\Z)$,
$|\theta l+ k| \gg |a_l|^{1/\sigma_0}$
uniformly in $k, l \in \Z^*$. Therefore, our Diophantine condition
seems
to be necessary.
\end{remark}

\subsubsection{Commutation between sum and residue}

Let $p\in \N$. Recall that $\mathcal{S}((\Z^n)^p)$ is the set of the
Schwartz
sequences on $(\Z^n)^p$. In other words, $b\in \mathcal{S}((\Z^n)^p)$
if and only if
for all $r\in \N_{0}$, $(1+|l_1|^2+\cdots
|l_p|^2)^r\,|b(l_1,\cdots,l_p)|^2$ is bounded on
$(\Z^n)^p$. We note that if $Q\in\R[X_1,\cdots,X_{np}]$ is a
polynomial, $(a_j)\in
\mathcal{S}(\Z^n)^{p}$, $b\in \mathcal{S}(\Z^n)$ and $\phi$ a
real-valued function, then
 $l:=(l_1,\cdots,l_p)\mapsto \wt a(l)\, b(-\wh l_p)\, Q(l)\,
 e^{i\phi(l)}$ is a Schwartz
sequence on $(\Z^n)^p$, where
\begin{align*}
\wt a(l) &:= a_1(l_1)\cdots a_p(l_p),\\
\wh l_i &:= l_1+\ldots +l_i.
\end{align*}

In the following, we will use several times the fact that for any
$(k,l)\in (\Z^n)^2$
such that $k\neq0$ and $k\neq -l$, we have
\begin{equation}\label{trick-0}
\frac{1}{|k+l|^2} = \frac{1}{|k|^2} -
\frac{2k.l+|l|^2}{|k|^2|k+l|^2}\,.
\end{equation}

\begin{lemma}\label{R-poly} There exists a polynomial $P\in
\R[X_1,\cdots,X_p]$ of degree $4p$ and with positive coefficients
such that for any
$k\in \Z^n$, and $l:=(l_1,\cdots,l_p)\in (\Z^n)^p$ such that $k\neq
0$ and $k\neq -\wh
l_i$ for all $1\leq i \leq p$, the following holds:
$$
\frac{1}{|k+\wh l_1|^2\ldots |k+\wh l_p|^2} \leq \frac{1}{|k|^{2p}}\
P(|l_1|,\cdots,|l_p|).
$$
\end{lemma}
\begin{proof} Let's fix $i$ such that $1\leq i\leq p$.
Using two times (\ref{trick-0}), Cauchy--Schwarz inequality and the
fact that $|k+\wh
l_i|^2\geq 1$, we get
\begin{align*}
\tfrac{1}{|k+\wh l_i|^2} &\leq \tfrac{1}{|k|^2} + \tfrac{2|k||\wh
l_i|+|\wh
l_i|^2}{|k|^4} + \tfrac{(2|k||\wh l_i|+|\wh l_i|^2)^2}{|k|^4|k+\wh
l_i|^2}\\
&\leq \tfrac{1}{|k|^2}+ \tfrac{2}{|k|^3}|\wh l_i| +
\big(\tfrac{1}{|k|^4}+\tfrac{4}{|k|^2}\big)|\wh l_i|^2 +
\tfrac{4}{|k|^3}|\wh l_i|^3 +
\tfrac{1}{|k|^4}|\wh l_i|^4.
\end{align*}
Since $|k|\geq 1$, and $|\wh l_i|^j \leq |\wh l_i|^4$ if $1\leq j\leq
4$, we find
\begin{align*}
&\tfrac{1}{|k+\wh l_i|^2} \leq \tfrac{5}{|k|^2} {\sum}_{j=0}^4 \,|\wh
l_i|^j \leq
\tfrac{5}{|k|^2} \big(1 + 4 |\wh l_i|^4\big)
\leq \tfrac{5}{|k|^2} \big(1 + 4 ({\sum}_{j=1}^p \, |l_j|)^4\big), \\
&\tfrac{1}{|k+\wh l_1|^2\ldots |k+\wh l_p|^2}  \leq
\tfrac{5^p}{|k|^{2p}} \big(1 + 4
({\sum}_{j=1}^p \,|l_j|)^{4}\big)^p.
\end{align*}
Taking $P(X_1,\cdots,X_p):= 5^p \big(1 + 4 ({\sum}_{j=1}^p
X_j)^{4}\big)^p$ now gives
the result.
\end{proof}

\begin{lemma}\label{abs-som}
Let $b\in \mathcal{S}((\Z^n)^p)$, $p\in\N$, $P_j\in
\R[X_1,\cdots,X_n]$ be a
homogeneous polynomial function of degree $j$, $k\in \Z^n$,
$l:=(l_1,\cdots,l_p)\in
(\Z^n)^p$, $r\in \N_{0}$, $\phi$ be a real-valued function on
$\Z^n\times (\Z^n)^{p}$ and
$$
h(s,k,l):=\frac{b(l)\, P_j(k)\, e^{i\phi(k,l)}}{|k|^{s+r}|k+\wh
l_1|^2\cdots |k+\wh
l_p|^2} \, ,
$$
with $h(s,k,l):=0$ if, for $k\neq 0$, one of the
denominators is zero.

For all $s\in \C$ such that $\Re(s)>n+j-r-2p$, the series
$$
H(s):={{\sum}'}_{(k,l)\in (\Z^n)^{p+1}} h(s,k,l)
$$
is absolutely summable. In particular,
$$
{\sum_{k\in\Z^n}}' \sum_{l\in (\Z^n)^p} h(s,k,l) = \sum_{l\in
  (\Z^n)^p}
{\sum_{k\in\Z^n}} ' h(s,k,l)\,.
$$
\end{lemma}

\begin{proof}
Let $s=\sigma+i\tau \in \C$ such that
$\sigma>n+j-r-2p$. By Lemma \ref{R-poly} we get, for $k\neq 0$,
$$
|h(s,k,l)|\leq |b(l)\, P_j(k)|\, |k|^{-r-\sigma-2p}\, P(l),
$$
where $P(l):=P(|l_1|,\cdots,|l_p|)$ and $P$ is a polynomial of degree
$4p$ with
positive coefficients. Thus,
$|h(s,k,l)|\leq  F(l)\, G(k)$
where $F(l):=|b(l)|\, P(l)$ and $G(k):= |P_j(k)| |k|^{-r-\sigma-2p}$.
The summability
of $\sum_{l\in (\Z^n)^p} F(l)$ is implied by the fact that $b\in
\mathcal{S}((\Z^n)^p)$. The summability of ${\sum}'_{k\in \Z^n} G(k)$
is a consequence
of the fact that $\sigma>n+j-r-2p$. Finally, as a product of two
summable series,
${\sum_{k,l}} F(l) G(k)$ is a summable series, which proves that
${\sum_{k,l}}h(s,k,l)$ is also absolutely summable.
\end{proof}

\begin{definition}
Let $f$ be a function on $D\times (\Z^n)^p$ where $D$ is an open
neighborhood of $0$ in $\C$.

We say that
$f$ satisfies (H1) if and only if there exists $\rho>0$ such that

\hspace{1cm} (i) for any $l$, $s\mapsto f(s,l)$ extends as a
holomorphic function on $U_\rho$,
where $U_\rho$ is the
open disk of center 0 and radius $\rho$,

\hspace{1cm}(ii) the series $\sum_{l\in (\Z^n)^p}
\norm{H(\cdot,l)}_{\infty,\rho}$
is a summable series, where $\norm{H(\cdot,l)}_{\infty,\rho}:=\sup_{s\in
U_{\rho}}|H(s,l)|$.\\
 We say that
$f$ satisfies (H2) if and only if there exists $\rho>0$ such that

\hspace{1cm}
(i) for any $l$, $s\mapsto f(s,l)$ extends as a holomorphic function
on
$U_\rho-\{0\}$,

\hspace{1cm}
(ii) for any $\delta$ such that $0<\delta<\rho$, the series
$\sum_{l\in (\Z^n)^p}
\norm{H(\cdot,l)}_{\infty,\delta,\rho}$ is summable, where
$\norm{H(\cdot,l)}_{\infty,\delta,\rho}:=\sup_{\delta<|s|<\rho}|H(s,l)|$.
\end{definition}

\begin{remark}
Note that (H1) implies (H2). Moreover, if $f$ satisfies (H1)
(resp. (H2) for $\rho>0$, then it is straightforward to check that
$f:s\mapsto \sum_{l\in (\Z^n)^p} f(s,l)$
extends as an holomorphic function on $U_\rho$ (resp. on $U_\rho
\setminus \set{0}$).
\end{remark}

\begin{corollary}\label{res-somH} With the same notations of Lemma
\ref{abs-som},
suppose that $r+2p-j>n$, then, the function
$H(s,l):={\sum}'_{k\in \Z^n} h(s,k,l)$ satisfies (H1).
\end{corollary}

\begin{proof} $(i)$ Let's fix $\rho>0$ such that $\rho < r+2p-j-n$.
Since $r+2p-j>n$, $U_\rho$ is inside the half-plane of absolute
convergence of
the series defined by $H(s,l)$. Thus, $s\mapsto H(s,l)$ is
holomorphic on $U_\rho$.\\
$(ii)$ Since $\big||k|^{-s}\big|\leq |k|^{\rho}$ for all $s\in
U_\rho$ and
$k\in\Z^n \setminus \set{0}$, we get as in the above proof
$$
|h(s,k,l)|\leq |b(l)\, P_j(k)| \, |k|^{-r+\rho-2p} \,
P(|l_1|,\cdots,|l_p|).
$$
Since $\rho < r+2p-j-n$, the series ${\sum}'_{k\in \Z^n} |P_j(k)|
|k|^{-r+\rho-2p}$ is
summable.

Thus, $\norm{H(\cdot,l)}_{\infty,\rho} \leq K \, F(l)$ where $K :=
{{\sum_k}}'|P_j(k)| |k|^{-r+\rho-2p}<\infty$. We have already seen
that the series
$\sum_l F(l)$ is summable, so we get the result.
\end{proof}

We note that if $f$ and $g$ both satisfy (H1) (or (H2)),
then so does $f+g$. In the
following, we will use the equivalence relation
$$
f\sim g \Longleftrightarrow f-g \text{ satisfies (H1)}.
$$

\begin{lemma}\label{res-som}
Let $f$ and $g$ be two functions on $D\times (\Z^n)^p$ where $D$ is
an open
neighborhood of $0$ in $\C$, such that $f\sim g$ and such that $g$
satisfies (H2). Then
$$
\underset{s=0}{\Res}\sum_{l \in (\Z^n)^p} f(s,l)=\sum_{l \in (\Z^n)^p}
\underset{s=0}{\Res}\ g(s,l)\, .
$$
\end{lemma}
\begin{proof} Since $f\sim g$, $f$ satisfies (H2) for a certain
$\rho>0$.
Let's fix $\eta$ such that $0<\eta<\rho$ and define $C_\eta$ as the
circle of center 0
and radius $\eta$. We have
$$
\underset{s=0}{\Res}\ g(s,l) = \underset{s=0}{\Res}\ f(s,l) = \tfrac
{1}{2\pi
i}\oint_{C_\eta} f(s,l)\, ds = \int_I u(t,l) dt\, .
$$
where $I=[0,2\pi]$ and $u(t,l):=\tfrac {1}{2\pi} \eta e^{it}
f(\eta\,e^{i t},l) $. The
fact that $f$ satisfies (H2) entails that the series $\sum_{l\in
(\Z^n)^p}
\norm{f(\cdot,l)}_{\infty,C_\eta}$ is summable. Thus, since
$\norm{u(\cdot,l)}_{\infty} = \tfrac {1}{2\pi} \eta
\norm{f(\cdot,l)}_{\infty,C_\eta}$, the series $\sum_{l\in (\Z^n)^p}
\norm{u(\cdot,l)}_{\infty}$ is summable, so, as a consequence,
$\int_I \sum_{l\in
(\Z^n)^p}   u(t,l) dt = \sum_{l\in (\Z^n)^p}  \int_I u(t,l)dt$ which
gives the result.
\end{proof}

\subsection{Computation of residues of zeta functions}

Since, we will have to compute residues of series, let us introduce
the following

\begin{definition}
\begin{align*}
\zeta(s)&:=\sum_{n=1}^{\infty} n^{-s},\\
Z_{n}(s)&:={\sum_{k\in\Z^{n}}}'\,\,\vert k\vert^{-s},\\
\zeta_{p_{1},\dots,p_{n}}(s)&:={\sum_{k\in\Z^{n}}}'\,\,
\frac{k_1^{p_{1}}\cdots
k_n^{p_{n}}}{\vert k\vert^{s}}\,\text{ , for } p_{i}\in \N,
\end{align*}
\end{definition}
\noindent where $\zeta(s)$ is the Riemann zeta function (see
\cite{HW} or
\cite{Edery}).

By the symmetry $k\rightarrow -k$, it is clear that these functions
$\zeta_{p_{1},\dots,p_{n}}$ all vanish for odd values of $p_{i}$.

Let us now compute
$\zeta_{0,\cdots,0,1_{i},0\cdots,0,1_{j},0\cdots,0}(s)$ in terms of
$Z_{n}(s)$:\\
Since $\zeta_{0,\cdots,0,1_{i},0\cdots,0,1_{j},0\cdots,0}(s)
=A_{i}(s)\,\delta_{ij}$,
exchanging the components $k_{i}$ and $k_{j}$, we get
\begin{align*}
\zeta_{0,\cdots,0,1_{i},0\cdots,0,1_{j},0\cdots,0}(s)
=\tfrac{\delta_{ij}}{n}\,Z_{n}(s-2).
\end{align*}
Similarly,
\begin{align*}
{\sum}'_{\Z^{n}}\,\tfrac{k_{1}^{2}k_{2}^{2}}{\vert k\vert^{s+8}}
=\tfrac{1}{n(n-1)}
Z_{n}(s+4)- \tfrac{1}{n-1}{\sum}'_{\Z^{n}}\, \tfrac{k_{1}^4}{\vert
k\vert^{s+8}}
\end{align*}
but it is difficult to write explicitly
$\zeta_{p_{1},\dots,p_{n}}(s)$ in terms of
$Z_{n}(s-4)$ and other $Z_{n}(s-m)$ when at least four indices
$p_{i}$ are non zero.

When all $p_{i}$ are even, $\zeta_{p_{1},\dots,p_{n}}(s)$ is a
nonzero series of
fractions $\tfrac{P(k)}{\vert k\vert ^s}$ where $P$ is a homogeneous
polynomial of
degree $p_{1}+\cdots +p_{n}$. Theorem \ref{res-int} now gives us the
following

\begin{prop}
    \label{calculres}
$\zeta_{p_{1},\dots,p_{n}}$ has a meromorphic extension to the whole
plane with a
unique pole at $n+p_{1} +\cdots +p_{n}$. This pole is simple and the
residue at this
pole is
\begin{align}
    \label{formule1}
\underset{s=n+p_{1} +\cdots
+p_{n}}{\Res} \,\zeta_{p_{1},\dots,p_{n}}(s)= 2 \,
\tfrac{\Gamma(\tfrac{p_{1}+1}{2}) \cdots
\Gamma(\tfrac{p_{n}+1}{2})}
{\Gamma
(\tfrac{n+p_{1}+ \cdots + p_{n}}{2})}
\end{align}
when all $p_{i}$ are even or this residue is zero otherwise.\\
In particular, for $n=2$,
\begin{align}
\underset{s=0}{\Res} \,\,{\sum_{k\in\Z^2}}'\,\tfrac{k_{i}k_{j}} {\vert
k\vert^{s+4}}=\delta_{ij}\,\pi\, , \label{formulen=2}
\end{align}
and for $n=4$,
\begin{align}
&\underset{s=0}{\Res} \,\,{\sum_{k\in\Z^{4}}}'\,\tfrac{k_{i}k_{j}}
{\vert k\vert^{s+6}}=\delta_{ij}\tfrac{\pi^2}{2}\, ,\nonumber\\
&\label{formule2} \underset{s=0}{\Res}
\,{\sum_{k\in\Z^{4}}}'\,\tfrac{k_{i}k_{j}k_{l}k_{m}}
{\vert k\vert^{s+8}}=(\delta_{ij}\delta_{lm}+\delta_{il}\delta_{jm}
+\delta_{im}\delta_{jl})\,\tfrac{\pi^2}{12}\, .
\end{align}
\end{prop}
\begin{proof}
Equation (\ref{formule1}) follows from Theorem (\ref{res-int})
$$
\underset{s=n+p_{1} +\cdots
+p_{n}}{\Res}\, \, \zeta_{p_{1},\dots,p_{n}}(s)=
\int_{k \in
S^{n-1}}k_{1}^{p_{1}} \cdots k_{n}^{p_{n}}\, dS(k)
$$
and standard formulae (see for instance \cite[VIII,1;22]{Schwartz}).
Equation
(\ref{formulen=2}) is a straightforward consequence of Equation
(\ref{formule1}).
Equation (\ref{formule2}) can be checked for the cases $i=j\neq l=m$
and
$i=j=l=m$.
\end{proof}
 Note that $Z_n(s)$ is an Epstein zeta function associated to the
quadratic form $q(x):=x_1^2+...+x_n^2$, so $Z_n$ satisfies the
following functional equation
$$
Z_n(s)= \pi^{s-n/2} \Gamma (n/2 -s/2)\Gamma (s/2)^{-1}\,
Z_n(n-s).
$$
Since $\pi^{s-n/2} \Gamma (n/2 -s/2) \,\Gamma (s/2)^{-1}=0$
for any negative even integer $n$ and $Z_n(s)$ is meromorphic on $\C$
with only one pole at $s=n$ with residue $2 \pi^{n/2} \Gamma
(n/2)^{-1}$ according to previous proposition, so we get $Z_n(0)=
-1$. We have proved that
\begin{align}
\label{formule}
    \underset{s=0}{\Res} \,\, Z_{n}(s+n)&=2\pi^{n/2} \, \Gamma
(n/2)^{-1},\\
Z_n(0)&= -1.
\label{Zn0}
\end{align}

\subsection{Meromorphic continuation of a class of zeta functions}

Let $n,q\in \N$, $q\geq2$, and $p=(p_1,\dots,p_{q-1}) \in
\N_0^{q-1}$.\\
Set $I:=\{ i \mid p_i\neq 0\}$ and assume that $I\neq \emptyset$ and
$${\cal I}:=\{\alpha =(\alpha_i)_{i\in I} \mid \forall i\in I ~
\alpha_i=(\alpha_{i,1},\dots, \alpha_{i,p_i})\in
\N_0^{p_i}\}=\prod_{i\in I} \N_0^{p_i}.$$

We will use in the sequel also the following notations:

\hspace{1cm} - for $x=(x_1,\dots,x_t) \in \R^t$ recall that
$|x|_1=|x_1|+\dots+|x_t|$ and $|x|=\sqrt{x_1^2+\dots+x_t^2}$;

\hspace{1cm} - for all $\alpha =(\alpha_i)_{i\in I}
  \in {\cal I} =\prod_{i\in I} \N_0^{p_i}$,
$$|\alpha|_1=\sum_{i\in I} |\alpha_i|_1 =\sum_{i\in I}
\sum_{j=1}^{p_i} |\alpha_{i,j}| {\mbox { and }}
\genfrac(){0pt}{1}{1/2}{\alpha} =\prod_{i\in I}
\genfrac(){0pt}{1}{1/2}{\alpha_{i}}=
\prod_{i\in I} \prod_{j=1}^{p_i}
\genfrac(){0pt}{1}{1/2}{\alpha_{i,j}}.$$

\subsubsection{A family of polynomials}
In this paragraph we define a family of polynomials which plays an
important role later.

Consider first the variables:

- for $X_1,\dots, X_n$ we set $X=(X_1,\dots,X_n)$;

- for any $i=1,\dots,2q$, we consider the variables $Y_{i,1},\dots,
  Y_{i,n}$ and set
$Y_i:=(Y_{i,1},\dots, Y_{i,n})$ and $Y:=(Y_1,\dots,Y_{2q})$;

- for $Y=(Y_1,\dots,Y_{2q})$, we set for any $1\leq j\leq q$,
$\wt Y_j:= Y_1+\cdots+ Y_j+  Y_{q+1}+\cdots + Y_{q+j}$.

We define for all $\alpha =(\alpha_i)_{i\in I}\in {\cal I}
=\prod_{i\in I} \N_0^{p_i}$ the polynomial
\begin{equation}
\label{palphaxy}
P_\alpha(X,Y):= \prod_{i\in I} \prod_{j=1}^{p_i}
(2\langle X, \wt Y_i\rangle + |\wt
    Y_i|^2)^{\alpha_{i,j}}.
\end{equation}

It is clear that $P_\alpha(X,Y) \in \Z[X,Y]$, deg$_X P_\alpha \leq
|\alpha|_1$ and deg$_Y P_\alpha \leq 2 |\alpha|_1$.

Let us fix a polynomial $Q\in \R[X_1,\cdots,X_n]$ and
note $d:= \deg Q$.
For $\a\in  {\cal I}$, we want to expand $P_\alpha(X,Y) \, Q(X)$ in
homogeneous polynomials in $X$ and $Y$ so defining
$$
L(\a):=\set{\beta\in\N_0^{(2q+1)n}\, :\,
|\beta|_1-d_\beta \leq 2|\a|_1 \text{ and } d_\beta\leq |\a|_1+d}
$$
where
$d_\beta := \sum_1^n \beta_i$, we set
$$
\genfrac(){0pt}{1}{1/2}{\alpha} P_\alpha(X,Y) \, Q(X) =:
\sum_{\beta\in L(\a)} c_{\a,\beta} \,  X^\beta Y^\beta
$$
where $c_{\a,\beta}\in \R$, $X^\beta:= X_1^{\beta_{1}} \cdots
X_n^{\beta_{n}}$ and $Y^{\beta}:=
Y_{1,1}^{\beta_{n+1}}\cdots Y_{2q,n}^{\beta_{(2q+1)n}}$. By
definition, $X^\beta$ is a
homogeneous polynomial of degree in $X$ equals to $d_\beta$.
We note $$M_{\a,\beta}(Y):=c_{\a,\beta} \, Y^\beta.$$

\subsubsection{Residues of a class of zeta functions}

In this section we will prove the following result, used in
Proposition \ref{zeta(0)} for
the computation of the spectrum dimension of the noncommutative
torus:

\begin{theorem}
\label{zetageneral}
(i) Let $\tfrac{1}{2\pi}\Th$ be a
diophantine matrix, and $\wt a \in \mathcal{S}
\big((\Z^{n})^{2q}\big)$. Then
$$
s\mapsto f(s):= \sum_{l\in [(\Z^n)^{q}]^2} \wt a_{l}\ {\sum_{k\in
\Z^n}}'\, \prod_{i=1}^{q-1}|k+\wt l_i|^{p_i} |k|^{-s}\, Q(k)\,
e^{ik.\Th \sum_1^{q} l_j}
$$
has a meromorphic continuation to the whole complex
 plane $\C$ with at most simple
possible poles at the points $s=n+d+|p|_1-m$ where $m\in \N_0$.

(ii) Let $m\in \N_0$ and set
$I(m):= \set{(\a,\beta)\in \mathcal{I}\times \N_0^{(2q+1)n} \, :
\, \beta\in L(\a) \text{ and } m=2|\a|_1 -d_\beta +d } $.
Then $I(m)$ is a finite set and $s=n+d+|p|_1-m$ is a pole of $f$
if and only if
$$
C(f,m):= \sum_{l\in Z} \wt a_l
\sum_{(\a,\beta)\in I(m)} M_{\a,\beta}(l)  \int_{u\in S^{n-1}}
u^\beta\, dS(u) \neq 0,
$$
with $Z:=\{l \,: \, \sum_1^{q} l_j=0 \}$ and the convention
$\sum_{\emptyset} =0$.
In that case $s=n+d+|p|_1-m$ is a simple pole of residue
$\underset{s= n+d+|p|_1 -m}{\Res} \, f(s) = C (f,m)$.
\end{theorem}

In order to prove the theorem above we need the following

\begin{lemma}
\label{zetageneral-lem}
For all $N\in \N$ we have
$$
 \prod_{i=1}^{q-1} |k+\wt l_i|^{p_i}=
\sum_{\alpha =(\alpha_i)_{i\in I}
  \in \prod_{i\in I}\{0,\dots,N\}^{p_i}}
  \genfrac(){0pt}{1}{1/2}{\alpha}\,
\tfrac{P_\alpha (k,l)}{|k|^{2|\alpha|_1-|p|_1}}
+\mathcal{O}_N(|k|^{|p|_1- (N+1)/2})
$$
uniformly in $k\in \Z^n$ and $l\in
(\Z^n)^{2q}$ verifying $|k| > U(l):=36\,
(\sum_{i=1,\, i\neq q}^{2q-1}|l_{i}|)^4$.
\end{lemma}
\begin{proof}
For $i=1,\dots,q-1$, we have uniformly in $k\in \Z^n$ and $l\in
(\Z^n)^{2q}$ verifying $|k| > U(l)$,
\begin{equation}
\label{devjustification}
\tfrac{\big|2\langle k, \wt l_i \rangle+|\wt l_i|^2\big|}{|k|^2}
\leq\tfrac{\sqrt{U(l)}}{2|k|} < \tfrac{1}{2\sqrt{|k|}}.
\end{equation}
In that case,
\begin{eqnarray*}
|k+\wt l_i|&=& \big(|k|^2+2\langle k, \wt l_i\rangle
  + |\wt l_i|^2\big)^{1/2} =
|k| \big(1+ \tfrac{2\langle k, \wt l_i\rangle
  + |\wt l_i|^2}{|k|^2}\big)^{1/2} =
 \sum_{u=0}^{\infty}  \genfrac(){0pt}{1}{1/2}{u} \,
\tfrac{1}{|k|^{2u-1}}P^i_u(k,l)
\end{eqnarray*}
where for all $i=1,\dots, q-1$ and for all $u\in \N_0$,
\begin{equation*}
P^i_u(k,l):=\big(2\langle k, \wt l_i\rangle + |\wt l_i|^2\big)^u,
\end{equation*}
with the convention $P^i_0(k,l):=1$.

In particular $P^i_u(k,l)\in \Z[k,l]$,
$\deg_{k} P^i_u\leq u$ and $\deg_{l} P^i_u\leq 2u$.
Inequality (\ref{devjustification}) implies that for all
$i=1,\dots,q-1$
and for all $u\in \N$,
$$
\tfrac{1}{|k|^{2u}}\,|P^i_u(k,l)|\leq \big(2\sqrt{|k|}\big)^{-u}
$$
uniformly in $k\in \Z^n$ and $l\in
(\Z^n)^{2q}$ verifying $|k| > U(l)$.

Let $N\in \N$. We deduce from the previous that for any
$k\in \Z^n$ and $l\in
(\Z^n)^{2q}$ verifying $|k| > U(l)$ and
for all $i=1,\dots,q-1$, we have
\begin{eqnarray*}
|k+\wt l_i|&=& \sum_{u=0}^{N}
\genfrac(){0pt}{1}{1/2}{u} \, \tfrac{1}{|k|^{2u-1}}P^i_u(k,l)+
\mathcal{O}\big(\sum_{u>N}|k|\,|\genfrac(){0pt}{1}{1/2}{u}|\,
(2\sqrt{|k|})^{-u}\big)\\
&=& \sum_{u=0}^{N} \genfrac(){0pt}{1}{1/2}{u} \,
\tfrac{1}{|k|^{2u-1}}P^i_u(k,l)+\mathcal{O}_N
\big(\tfrac{1}{|k|^{(N-1)/2}}\big).
\end{eqnarray*}
It follows that for any $N\in \N$, we have uniformly in
$k\in \Z^n$ and $l\in (\Z^n)^{2q}$ verifying $|k| > U(l)$ and
for all $i\in I$,
$$
|k+\wt l_i|^{p_i}=\sum_{\alpha_i \in
  \{0,\dots,N\}^{p_i}} \genfrac(){0pt}{1}{1/2}{\alpha_i} \,
\tfrac{1}{|k|^{2|\alpha_i|_1-p_i}}P^i_{\alpha_i} (k,l)
+\mathcal{O}_N \left(\tfrac{1}{|k|^{(N+1)/2-p_i}}\right)
$$
where $P^i_{\alpha_i} (k,l)=\prod_{j=1}^{p_i} P^i_{\alpha_{i,j}}(k,l)$
for all $\alpha_i =(\alpha_{i,1},\dots,\alpha_{i,p_i})\in
  \{0,\dots,N\}^{p_i}$ and
$$
\prod_{i\in I} |k+\wt l_i|^{p_i}=\sum_{\alpha=(\alpha_i) \in
  \prod_{i\in I} \{0,\dots,N\}^{p_i}} \genfrac(){0pt}{1}{1/2}{\alpha}
\,
\tfrac{1}{|k|^{2|\alpha|_1-|p|_1}}P_{\alpha} (k,l)+\mathcal{O}_N
\big(\tfrac{1}{|k|^{(N+1)/2 -|p|_1}}\big)
$$
where $P_{\alpha } (k,l)=\prod_{i\in I} P^i_{\alpha_{i}}(k,l)=
\prod_{i\in I} \prod_{j=1}^{p_i} P^i_{\alpha_{i,j}}(k,l)$.
\end{proof}

\medskip

{\it Proof of Theorem \ref{zetageneral}.}

$(i)$ All $n$, $q$, $p=(p_1,\dots,p_{q-1})$ and $\wt a
\in {\cal S}\left((\Z^n)^{2q}\right)$ are fixed as above and we
define formally for any $l \in (\Z^n)^{2q}$
\begin{equation}
\label{fls}
F(l,s):= {\sum_{k\in \Z^n}}' \,
\prod_{i=1}^{q-1} |k+\wt l_i|^{p_i}\, Q(k)\,
e^{ik.\Th \sum_1^{q}l_j}\,|k|^{-s}.
\end{equation}
Thus, still formally,
\begin{equation}\label{fsexpfls}
f(s):=\sum_{l\in (\Z^n)^{2q}} \wt a_l\ F(l,s).
\end{equation}
It is clear that $F(l,s)$ converges absolutely in the half
plane $\{\sigma=\Re(s) >n+d+|p|_1\}$ where $d=\deg Q$.

Let $N\in \N$. Lemma \ref{zetageneral-lem} implies
that for any $l\in (\Z^n)^{2q}$ and for $s\in \C$ such
that $\sigma >n+|p|_1+d$,
\begin{align*}
F(l,s)&= {\sum_{|k|\leq U(l)}}' \,
\prod_{i=1}^{q-1} |k+\wt l_i|^{p_i}\, Q(k)\,
e^{ik.\Th \sum_1^{q}l_j}\,|k|^{-s} \\
& \quad \quad  +\sum_{\alpha =(\alpha_i)_{i\in I}
  \in \prod_{i\in
I}\{0,\dots,N\}^{p_i}}\genfrac(){0pt}{1}{1/2}{\alpha}
\sum_{|k| > U(l)} \tfrac{1}{|k|^{s+2|\alpha|_1-|p|_1}}P_\alpha (k,l)
Q(k)\, e^{ik.\Th \sum_1^{q}l_j}
+ G_N(l,s).
\end{align*}
where $s\mapsto G_N(l,s)$ is a holomorphic function in the half-plane
$D_N:=\{\sigma > n+d+|p|_1-\tfrac{N+1}{2}\}$ and verifies in it the
bound
$G_N(l,s) \ll_{N,\sigma} 1$ uniformly in $l$.

It follows that
\begin{equation}
\label{flsexpzeta}
F(l,s)= \sum_{\alpha =(\alpha_i)_{i\in I}
  \in \prod_{i\in I}\{0,\dots,N\}^{p_i}} H_{\a}(l,s)+ R_N(l,s),
\end{equation}
where
\begin{eqnarray*}
H_{\a}(l,s)&:=&{\sum_{k\in \Z^n}}'\,
\genfrac(){0pt}{1}{1/2}{\alpha} \,
\tfrac {1}{|k|^{s+2|\alpha|_1-|p|_1}}P_\alpha (k,l)\, Q(k)\,
e^{ik.\Th \sum_1^{q}l_j},\\
R_N(l,s)&:=& {\sum_{|k|\leq U(l)}}' \,
\prod_{i=1}^{q-1} |k+\wt l_i|^{p_i}\, Q(k)\,
e^{ik.\Th \sum_1^{q}l_j}\,|k|^{-s}\\
& & \quad -{\sum_{|k|\leq U(l)}}' \quad\sum_{\alpha =
(\alpha_i)_{i\in I} \in \prod_{i\in I}\{0,\dots,N\}^{p_i}}
  \genfrac(){0pt}{1}{1/2}{\alpha} \tfrac{P_\alpha (k,l)}
{|k|^{s+2|\alpha|_1-|p|_1}}Q(k)\,
  e^{ik.\Th \sum_1^{q}l_j}
+ G_N(l,s).
\end{eqnarray*}
In particular there exists $A(N)>0$ such that
$s\mapsto R_N(l,s)$ extends holomorphically to the half-plane
$D_N$ and verifies in it the bound
$R_N(l,s) \ll_{N,\sigma} 1 +|l|^{A(N)}$ uniformly in $l$.

Let us note formally
$$
h_\a(s):= \sum_l \wt a_l\, H_\a(l,s).
$$
Equation (\ref{flsexpzeta}) and $R_N(l,s) \ll_{N,\sigma} 1
+|l|^{A(N)}$ imply that
\begin{equation}
\label{fssimN}
f(s) \sim_N \sum_{\alpha =(\alpha_i)_{i\in I}
  \in \prod_{i\in I}\{0,\dots,N\}^{p_i}} h_\a(s),
\end{equation}
where $\sim_N$ means modulo a holomorphic function in $D_N$.

Recall the decomposition
$\genfrac(){0pt}{1}{1/2}{\alpha} \,P_\a(k,l) \, Q(k)=
\sum_{\beta\in L(\a)} M_{\a,\beta}(l) \, k^\beta$ and we
decompose similarly
$h_{\a}(s) =\sum_{\beta\in L(\a)} h_{\a,\beta}(s).$
Theorem \ref{analytic} now implies that for all
$\alpha =(\alpha_i)_{i\in I}  \in \prod_{i\in I}\{0,\dots,N\}^{p_i}$
and $\beta\in L(\a)$,

\quad - the map $s\mapsto h_{\a,\beta}(s)$ has a meromorphic
continuation to the whole complex plane $\C$ with only one
simple possible pole at $s=n+ |p|_1 - 2|\a|_1 +d_\beta$,

\quad - the residue at this point is equal to
\begin{equation}
\label{res-halphaj}
\underset {s=n+ |p|_1 - 2|\a|_1 +d_\beta}{\Res}\,
h_{\a,\beta}(s) =
\sum_{l\in \mathcal{Z}} \wt a_l\, M_{\a,\beta}(l) \int_{u\in S^{n-1}}
u^\beta dS(u)
\end{equation}
where $\mathcal{Z}:=\{l\in (\Z)^{n})^{2q} \, : \, \sum_1^{q} l_j =0
\}$.
If the right hand side is zero, $h_{\a,\beta}(s)$ is holomorphic on
$\C$.

By (\ref{fssimN}), we deduce therefore that
$f(s)$ has a meromorphic continuation on the halfplane $D_N$,
with only simple possible poles in the set
$
\set{n+|p|_1 + k  \,: \,-2N|p|_1\leq k \leq d}.
$
Taking now $N\to \infty$ yields the result.

$(ii)$ Let $m\in \N_0$ and set
$I(m):= \set{(\a,\beta)\in \mathcal{I}\times \N_0^{(2q+1)n} \,:
\, \beta\in L(\a) \text{ and } m=2|\a|_1 -d_\beta +d } $.
If $(\a,\beta)\in I(m)$, then $|\a|_1 \leq m$ and
$|\beta|_1\leq 3m+d$, so $I(m)$ is finite.

With a chosen $N$ such that $2N|p|_1+d>m$, we get by (\ref{fssimN})
and (\ref{res-halphaj})
$$
\underset{s= n+d+|p|_1 -m}{\Res} \,f(s) =
\sum_{l\in \mathcal{Z}} \wt a_l
\sum_{(\a,\beta)\in I(m)} M_{\a,\beta}(l)  \int_{u\in S^{n-1}}
u^\beta\, dS(u)=C(f,m)
$$
with the convention $\sum_{\emptyset} =0$. Thus, $n+d+|p|_1 - m$ is a
pole of $f$ if and only if
$C(f,m)\neq 0$.
{\qed}

\chapter{Spectral action on $\SUq$}

\section{Introduction}

The quantum group $SU_{q}(2)$ has already a rather long history
of studies \cite{Klimyk} being one of the finest examples of quantum
deformation. This includes an approach via the noncommutative notion
of spectral triple introduced by Connes \cite{Book,ConnesMarcolli}
and various notions of Dirac operators were introduced in
\cite{Bibikov, Goswani, Pal1, Cindex, Pal3}. Finally, a real
spectral triple, which was exhibited in \cite{DLSSV}, is invariant
by left and right action of $\U_{q}(su(2))$ and satisfies almost
all postulated axioms of triples except the commutant and first-order
properties. These, however, remain valid only up to infinitesimal
of arbitrary high order. The last presentation generalizes in
a straightforward way all geometric construction details of the
spinorial spectral triple for the classical three-sphere. In
particular, both the equivariant representation and the symmetries
have a $q \to 1$ proper classical limit.

The goal of this chapter is to obtain the spectral action on $SU_q(2)$ which is a spectral triple with an invertible Dirac operator, with the control of the differential calculus generated by the Dirac
operator arising as the main difficulty. This issue of computing the spectral action was addressed in the
epilogue of \cite{Walter}.
In the case of $SU_q(2)$, we have $Sd^+=Sd=\set{1,2,3}$, so
\begin{align}
    \SS(\DD_{\Abb},\Phi,\Lambda) \, = \,\sum_{1\leq k \leq 3}
\Phi_{k}\,
    \Lambda^{k} \ncint \vert D_{\Abb}\vert^{-k} + \Phi(0) \,
    \zeta_{D_{\Abb}}(0). \label{actionsuq}
\end{align}

Note that in the case of $SU_q(2)$ there are no terms in $\Lambda^{-k}$, $k>0$ because the
dimension spectrum is bounded below by 1. 

To proceed with the computation of \eqref{actionsuq}, we introduce two presentations of one-forms.
The main ingredient is $F=$ sign ($\DD$) which appears to be a
one-form up to $OP^{-\infty}$.

In section 2, we discuss the spectral action of an arbitrary
3-dimensional spectral triple using cocycles.

In sections 3 and 4 we recall the main results on $SU_q(2)$ of
\cite{DLSSV} and show that the full spectral action with reality
operator given by \eqref{formuleaction} is completely determined by
the terms
$$
\ncint A^{q} |\DD|^{-p}, \quad 1\leq q\leq p\leq 3\,.
$$
where $A$ is a linear combination of terms of the form $a[|\DD|,b]$ with $a,b\in \A$.

In section 5, we establish a differential calculus up to some ideal
in pseudodifferential operators and apply these results to the
precise computation of previous noncommutative integrals.

Section 6 is devoted to explicit examples, while in next section are
given different comparisons with the commutative case of the 3-sphere
corresponding to $SU(2)$.

\section{Spectral action in $3$-dimension}

\subsection{Tadpole and cocycles}

Let $(\A,\H, \DD)$ be a spectral triple of dimension 3. We refer to \cite{CC1,ConnesMarcolli} for the definition of the $OP^\a$ spaces and the algebra of pseudodifferential operators $\Psi(\A)$ on a spectral triple.

For
$n\in\N^*$ and
$a_{i} \in \A$, define
\begin{align*}
\phi_n(a_0, \cdots,a_n) & :=\ncint a_0 [\DD, a_1]
D^{-1}\cdots[\DD, a_n] D^{-1}.
\end{align*}
We also use notational integrals on the universal $n$-forms
$\Omega_u^n(\A)$
defined by
$$
\int_{\phi_{n}}a_{0}da_{1}\cdots
da_{n}:=\phi_{n}(a_{0},a_{1},\cdots,a_{n}).
$$
and the reordering fact that
$(da_{0})a_{1}=d(a_{0}a_{1})-a_{0}da_{1}$.

We use the $b-B$ bicomplex defined in \cite{Book}: $b$ is the
Hochschild coboundary map (and $b'$ is truncated one) defined on
$n$-cochains $\phi$ by
\begin{align*}
b\phi(a_{0},\ldots,a_{n+1})&:=b'\phi(a_{0},\ldots,a_{n+1})+(-1)^{n+1}
\phi(a_{n+1}a_{0},a_{1},\ldots,a_{n}), \\
 b'\phi(a_{0},\ldots,a_{n+1})&:=\sum_{j=0}^n (-1)^{j}
\phi(a_{0},\ldots,a_{j}a_{j+1},\ldots,a_{n+1}).
\end{align*}
Recall that $B_{0}$ is defined on the normalized cochains $\phi_{n}$
by
$$
B_{0}
\phi_{n}(a_{0},a_{1},\ldots,a_{n-1}):=\phi_{n}(1,a_{0},
\ldots,a_{n-1}), \text{ thus }
\int_{\phi_{n}}d\omega=\int_{B_{0}\phi_{n}}\omega \text{ for }
\omega \in \Omega_u^{n-1}(\A).
$$
Then $B:=NB_{0}$, where $N:=1+\lambda+\ldots \lambda^n$ is the cyclic
skewsymmetrizer on the $n$-cochains and $\lambda$ is the cyclic
permutation $\lambda  \phi(a_{0},\ldots,a_{n}):=(-1)^n
\phi(a_{n},a_{0},\ldots,a_{n-1})$.

We will also encounter the cyclic 1-cochain $N\phi_{1}$:
\begin{align*}
N\phi_{1}(a_0,a_{1}) := \phi_1(a_0,a_{1}) - \phi_1(a_{1},a_0) \text{
and }
\int_{N\phi_{1}}a_{0}da_{1}:=N\phi_{1}(a_{0},a_{1}).
\end{align*}

\begin{remark}
    \label{commut}
Assume the integrand of $\ncint$ is in $OP^{-3}$. Since
$[D^{-1},a]=-D^{-1}[D,a]D^{-1}\in OP^{-2}$, this commutator
introduces an integrand in $OP^{-4}$ so has a vanishing integral:
under the integral, we can commute $D^{-1}$ with all $a\in \A$, but not with one-forms. Note also that since $P_0 \in OP^{-\infty}$, any integrand containing $P_0$ has a vanishing integral.
\end{remark}

\begin{lemma}
\label{cocycle}
We have
\begin{align*}
&\hspace{-5cm}(i) &&  \hspace{-5cm}b \phi_1 = - \phi_2.\\
&\hspace{-5cm}(ii) && \hspace{-5cm}b \phi_2 =0.\\
&\hspace{-5cm}(iii) && \hspace{-5cm}b \phi_3  = 0.\\
&\hspace{-5cm}(iv) && \hspace{-5cm}B \phi_1 =0.\\
&\hspace{-5cm}(v)  && \hspace{-5cm}B_0 \phi_2 =- (1-\lambda) \phi_1.\\
&\hspace{-5cm}(vi) && \hspace{-5cm}bB_{0}\phi_{2}=2
\phi_{2}+B_{0}\phi_{3}.\\
&\hspace{-5cm}(vii) && \hspace{-5cm}B \phi_{2}=0.\\
&\hspace{-5cm}(viii) && \hspace{-5cm}B_{0} \phi_{3}=Nb' \phi_{1}.\\
&\hspace{-5cm}(ix) &&  \hspace{-5cm}B \phi_{3}=3B_{0} \phi_{3}.
\end{align*}
\end{lemma}

\begin{proof}
$(i)$
\begin{align*}
b \phi_1( a_0,a_1,a_2) &= \ncint a_0 a_1 [\DD,a_2] D^{-1}
- \ncint a_0 \left(a_1 [\DD,a_2] + [\DD,a_1] a_2 \right) D^{-1}
+ \ncint a_2 a_0 [\DD,a_1] D^{-1} \\
&= \ncint a_0 [\DD,a_1] \left( D^{-1} a_2 - a_2  D^{-1} \right)
= - \ncint a_0 [\DD,a_1] D^{-1} [\DD, a_2] D^{-1} \\
&= - \phi_2(a_0, a_1, a_2)
\end{align*}
where we have used the trace property of the noncommutative integral.

$(ii)$  $b \phi_2( a_0,a_1,a_2,a_3)$
\begin{align*}
\qquad&= \ncint a_0 a_1 [\DD,a_2] D^{-1} [\DD,
a_3] D^{-1}
- \ncint a_0 \left(a_1 [\DD,a_2] + [\DD,a_1] a_2 \right) D^{-1}
[\DD, a_3] D^{-1} \\
&\quad + \ncint a_0 [\DD,a_1] D^{-1} (a_2 [\DD,a_3] + [\DD,a_2] a_3
) D^{-1}
- \ncint a_3 a_0 [\DD,a_1] D^{-1} [\DD, a_2] D^{-1} \\
&= \ncint a_0 [\DD,a_1]
\left( D^{-1} a_2 - a_2 D^{-1} \right) [\DD, a_3]  D^{-1}
+ \ncint a_0 [\DD,a_1] D^{-1} [\DD,a_2]
\left( a_3  D^{-1} - D^{-1} a_3 \right) \\
&= - \ncint a_0 [\DD,a_1] D^{-1} [\DD, a_2] D^{-1} [D, a_3]
D^{-1} + \ncint a_0 [\DD,a_1] D^{-1} [\DD, a_2] D^{-1} [D, a_3]
D^{-1} \\
&=  0.
\end{align*}

$(iii)$ Using Remark \ref{commut}, we get
$\phi_{3}(a_{0},a_{1},a_{2},a_{3})=\ncint
a_{0}[\DD,a_{1}][\DD,a_{2}][\DD,a_{3}]\vert \DD \vert^{-3}$,
so similar computations as for $\phi_{2}$ gives $b\phi_{3}=0$.

$(iv)$ \quad $B_0 \phi_1( a_0 ) = \ncint  [\DD,a_0] D^{-1} =
\ncint \left( \DD a_0 D^{-1} -a_0 \right) = 0.$
\begin{align*}
\hspace{-1.8cm} (v) \quad B_0 \phi_2( a_0,a_1) =& \ncint  [\DD,a_0]
D^{-1} [\DD,a_1] D^{-1}
= \ncint a_0 D^{-1} [\DD,a_1] - \ncint a_0 [\DD,a_1] D^{-1} \\
=& \ncint a_0 a_1 - \ncint a_0 D^{-1} a_1 \DD - \ncint a_0
[\DD,a_1] D^{-1}\\
=& - \ncint a_1 [ \DD, a_0] D^{-1} - \ncint a_0 [\DD,a_1] D^{-1}
= -  \phi_1(a_1,a_0) - \phi_1(a_0, a_1) .
\end{align*}

$(vi)$ Since $-b \lambda
\phi_{1}(a_{0},a_{1},a_{2})=\phi_{1}(a_{2},a_{0}a_{1})-
\phi_{1}(a_{1}a_{2},a_{0})+\phi_{1}(a_{1},a_{2}a_{0})$, one obtains
that
$$
-b \lambda  \phi_{1}(a_{0},a_{1},a_{2})=\ncint
a_{0}a_{1}D^{-1}a_{2}\DD+a_{0}D^{-1}a_{1}\DD a_{2}
-a_{0}D^{-1}a_{1}a_{2}\DD-a_{0}a_{1}a_{2}.
$$
So by direct expansion,
this is equal to $-\ncint a_{0}D^{-1}
[\DD,a_{1}]D^{-1}[\DD,a_{2}]$ which means that
\begin{align*}
-b\lambda \phi_{1}(a_{0},a_{1},a_{2})&=\ncint [D^{-1},a_{0}]
[\DD,a_{1}]D^{-1}[\DD,a_{2}]
-a_{0}[\DD,a_{1}]D^{-1}[\DD,a_{2}]D^{-1}\\
&=-B_{0}\phi_{3}(a_{0},a_{1},a_{2})
-\phi_{2}(a_{0},a_{1},a_{2}).
\end{align*}
Now the result follows from $(i),(v)$.

$(vii)$ $B \phi_{2}=NB_{0} \phi_{2}=-N(1-\lambda) \phi_{1}=0$ since
$N(1-\lambda)=0$.

\begin{align*}
\hspace{-0.6cm} (viii)\quad B_0 \phi_3( a_0,a_1,a_2) =
& \ncint  [\DD,a_0] D^{-1} [\DD,a_1] D^{-1} [\DD,a_2] D^{-1} \\
=& \ncint a_0 D^{-1} [\DD,a_1] D^{-1} [\DD,a_2]
-  \ncint a_0 [\DD,a_1] D^{-1} [\DD,a_2] D^{-1} \\
=& \ncint a_0 a_1 D^{-1} [\DD,a_2]
-  \ncint a_0 D^{-1} a_1 [\DD,a_2]
-  \ncint a_0 [\DD,a_1] D^{-1} [\DD,a_2] D^{-1}\\
=& \ncint  a_0 a_1 a_2 - \ncint a_0 a_1 D^{-1} a_2 \DD
-  \ncint  a_0 D^{-1} a_1 \DD a_2 + \ncint a_0 D^{-1} a_1 a_2
\DD\\
&\quad -  \ncint  a_0 [\DD,a_1] D^{-1} [\DD,a_2] D^{-1}\\
=& \ncint  a_0 a_1 a_2 - a_{2}\DD a_{1}a_{0}D^{-1}+a_{1}a_{2}\DD
a_{0}D^{-1}+a_2 \DD a_0 a_1 D^{-1}\\
&\quad - \left( a_{0}\DD a_{1}a_{2}D^{-1}-a_{0}\DD
a_{1}D^{-1}-a_{0}a_{1}\DD a_{2}D^{-1}+a_{0}a_{1}a_{2}\right).
\end{align*}
Expanding $(id + \lambda+\lambda^2)b'\phi_{1}(a_{0},a_{1},a_{2})$, we
recover previous expression.

$(ix)$ Consequence of $(viii)$.
\end{proof}

\subsection{Scale-invariant term of the spectral action}

We know from \cite{CC1} that the scale-invariant term of the action
can be written as
\begin{align}
    \label{constanttterm}
\zeta_{D_{A}}(0)-\zeta_{D}(0) = - \ncint  A D^{-1} + \tfrac{1}{2}
\ncint  A D^{-1} A
D^{-1}
- \tfrac{1}{3} \ncint  A D^{-1}  A D^{-1}  A D^{-1}.
\end{align}
In fact, this action can be expressed in dimension 3 as
contributions corresponding to tadpoles and the Yang--Mills
and Chern--Simons actions in dimension 4:
\begin{prop}
    \label{prop:action}
For any one-form $A$,
\begin{align}
    \label{constanttterm1}
\zeta_{D_{A}}(0)-\zeta_{D}(0) =  - \tfrac{1}{2} \int_{N\phi_{1}}A
+ \tfrac{1}{2} \int_{\phi_2} (dA + A^2)
- \tfrac{1}{2} \int_{\phi_3}( A dA + \tfrac{2}{3} A^3).
\end{align}
\end{prop}

To prove this, we calculate now each terms of the action.
\begin{lemma}
      \label{cocycleandoneform}
For any one-form $A$, we have

(i) $\int_{\phi_{2}}dA=\int_{B_{0}\phi_{2}} A=-\int_{\phi_{1}}A
+\int_{\lambda\phi_{1}}A$.

(ii) $\ncint AD^{-1}=\int_{\phi_{1}} A=\tfrac 12
\int_{N\phi_{1}}A-\tfrac 12\int_{\phi_{2}}dA$.

(iii) $\ncint  A D^{-1} A D^{-1} = - \int_{\phi_3} A dA +
\int_{\phi_2} A^2.$

(iv) $\ncint  A D^{-1} A D^{-1} A D^{-1} = \int_{\phi_3} A^3.$
\end{lemma}

\begin{proof}
$(i)$ and $(ii)$ follow directly from Lemma \ref{cocycle} $(v)$.

$(iii)$ With the shorthand $A=a_{i}db_{i}$ (summation on $i$)
\begin{align*}
\ncint  A D^{-1} A D^{-1} =& \ncint a_0 [\DD,b_0] D^{-1}
a_1 [\DD,b_1] D^{-1} \\
=&- \int_{\phi_3} A d A + \ncint a_0 [\DD, b_0] a_1 b_1 D^{-1}
- \ncint a_0 [\DD, b_0] a_1 D^{-1} b_1.
\end{align*}
We calculate further the remaining terms
\begin{align*}
\ncint a_0 [\DD, b_0] a_1 b_1 D^{-1} - \ncint a_0 [\DD, b_0] a_1
D^{-1} b_1
&= \ncint a_0 \DD b_0 a_1 b_1 D^{-1} - \ncint a_0 b_0 \DD a_1 b_1
D^{-1} \\
&\qquad- \ncint a_0 \DD b_0 a_1 D^{-1} b_1
+ \ncint a_0 b_0 \DD a_1 D^{-1} b_1,
\end{align*}
which are compared with
$\int_{\phi_2}A^2=\int_{\phi_{2}}a_{0}(db_{0})a_{1}db_{1}=
\int_{\phi_{2}}a_{0}d(b_{0}a_{1})db_{1}-a_{0}b_{0}da_{1}db_{1}$:
\begin{align*}
\int_{\phi_2} A^2 &= \ncint a_0 [\DD, b_0 a_1] D^{-1} [\DD, b_1]
D^{-1}
- \ncint a_0  b_0 [\DD, a_1] D^{-1} [\DD, b_1] D^{-1} \\
&= \ncint a_0 \DD b_0 a_1 b_1 D^{-1}
 - \ncint a_0 \DD b_0 a_1 D^{-1} b_1 - \ncint a_0 b_0 a_1 \DD b_1
D^{-1}
  + \ncint a_0 b_0 a_1 b_1 \\
&\qquad - \ncint a_0  b_0 \DD a_1 b_1 D^{-1}
  + \ncint a_0  b_0 \DD a_1 D^{-1} b_1 + \ncint a_0  b_0 a_1 \DD
b_1 D^{-1}
 - \ncint a_0  b_0 a_1 b_1 \\
&=  \ncint a_0 \DD b_0 a_1 b_1 D^{-1}  - \ncint  b_1 a_0 \DD b_0
a_1 D^{-1}
- \ncint a_0  b_0 \DD a_1 b_1 D^{-1} + \ncint  b_1 a_0  b_0 \DD a_1
D^{-1}.
\end{align*}

$(iv)$ Note that
\begin{align*}
\int_{\phi_{3}}A^3&=\int_{\phi_{3}}a_{0}(db_{0})a_{1}(db_{1})a_{2}db_{2}=
\int_{\phi_{3}}a_{0}d(b_{0}a_{1})d(b_{1}a_{2})db_{2}
-a_{0}b_{0}da_{1}d(b_{1}a_{2})db_{2}\\
&\qquad -a_{0}d(b_{0}a_{1}b_{1})d(a_{2}db_{2} +a_{0}b_{0}
d(a_{1}b_{1})da_{2}db_{2}\\
&=\ncint
a_{0}[\DD,b_{0}a_{1}]D^{-1}[\DD,b_{1}a_{2}]D^{-1}[\DD,b_{2}]D^{-1}
-a_{0}b_{0}[\DD,a_{1}]D^{-1}[\DD,b_{1}a_{2}]D^{-1}[\DD,b_{2}]D^{-1}\\
&\qquad-a_{0}[\DD,b_{0}a_{1}b_{1}]D^{-1}[\DD,a_{2}]D^{-1}[\DD,b_{2}]D^{-1}
+a_{0}b_{0}[\DD,a_{1}b_{1}]D^{-1}[\DD,a_{2}]D^{-1}
[\DD,b_{2}]D^{-1}.
\end{align*}
Summing up the first two terms and the last two ones gives
$$
\int_{\phi_{3}}A^3=\ncint a_{0}[\DD,b_{0}]a_{1}
D^{-1}[\DD,b_{1}a_{2}]D^{-1}[\DD,b_{2}]D^{-1}
-a_{0}[\DD,b_{0}]a_{1}b_{1}D^{-1}[\DD,a_{2}]D^{-1}[\DD,b_{2}]D^{-1}.
$$
Using Remark \ref{commut}, we can commute
under the integral $D^{-1}$ with all $a\in \A$ and similarly
$$
\ncint  A D^{-1} A D^{-1} A D^{-1} =\ncint
a_{0}[\DD,b_{0}]a_{1}D^{-1}[\DD,b_{1}]a_{2}D^{-1}[\DD,b_{2}]D^{-1}
$$
which proves $(iv)$.
\end{proof}

We deduce Proposition \ref{prop:action} from (\ref{constanttterm})
using the previous lemma.


\section{The $SU_{q}(2)$ triple}

\subsection{The spectral triple}

We briefly recall the main facts of the real spectral triple
$\big(\A(SU_q(2)),\H,\DD\big)$ introduced in \cite{DLSSV}, see also
\cite{Cindex,Pal1,Pal2}.
\vspace{0.5cm}

{\it The algebra:}

Let $\A := \A(SU_q(2))$ be
the $^*$-algebra generated polynomially by $a$ and~$b$,
subject to the following commutation rules with $0 < q < 1$:
\begin{eqnarray}
    \label{defrule}
&ba = q \,ab,  \qquad  b^*a = q\,ab^*, \qquad bb^* = b^*b,
&a^*a + q^2 \,b^*b = 1,  \qquad  aa^* + bb^* = 1\, .
\end{eqnarray}

We recall the following lemma from \cite[Lemma A2.1]{Woronowicz}:

\begin{lemma}
For any representation $\pi$ of $\A$, 
\begin{align*}
& Spect\big(\pi(bb^*)\big) = \set{0,q^{2k} \,\vert \,\vert k\in \N}
\text{ or } \pi(b)=0,\\
& Spect \big(\pi(aa^*)\big) =  \set{1,1-q^{2k} \,\vert \, k\in \N}
\text{ or $\pi(b)=0$ and $\pi(a)$ is a unitary}.
\end{align*}
\end{lemma}

This result is interesting since it shows the appearance of
discreteness for $0 \leq q<1$ while for $q=1$, $SU_q(2)=SU(2) \simeq
\mathbb{S}^3$ and the spectrum of the commuting operator $\pi(aa^*)$
and $\pi(bb^*)$ are equal to $[0,1]$. Moreover, all foregoing results
on noncommutative integrals will involve $q^2$ and not $q$. 

\vspace{0.5cm}
Any element of $\A$ can be
uniquely decomposed as a
linear combination of terms of the form
$a^\a b^\beta {b^{*}}^\gamma$ where $\a\in \Z$,
$\beta, \gamma \in \N$, with the convention
$$
a^{-\vert \a\vert} := {a^*}^{\vert \a \vert}.
$$

{\it The spinorial Hilbert space:}

$\H=\H^{\up} \oplus \H^{\dn}$ has an orthonormal
basis consisting of vectors $\ket{j\mu n\up}$
with $j = 0,\half,1,\dots$,
$\mu = -j,\dots,j$ and $n = -j^+,\dots,j^+$, together with
$\ket{j\mu n\dn}$ for $j=\half,1,\dots$, $\mu = -j,\ldots,j$ and
$n = -j^-,\dots,j^-$ (here $x^{\pm}:=x\pm \half$).

It is convenient to use a vector notation, setting:
\begin{equation}
    \label{eq:kett-defn}
\kett{j\mu n} :=\genfrac{(}{)}{0pt}{1}{\, \ket{j\mu n\up}\,}
{\,\ket{j\mu n\dn}\,}
\end{equation}
and with the convention that the lower component is zero when
$n = \pm(j + \half)$ or $j = 0$.
\vspace{0.5cm}

{\it The representation $\pi$ and its approximate $\ul \pi$:}

It is known that representation theory of $SU_q(2)$ is similar to
that of $SU(2)$ \cite{Woronowicz}. The representation $\pi$ given in
\cite{DLSSV} is:
\begin{align}
  \label{eq:reprexact}
\pi(a)\, \kett{j\mu n } &:= \alpha^+_{j\mu n}\,\kett{j^+ \mu^+ n^+}
+\alpha^-_{j\mu n}\,\kett{j^- \mu^+ n^+},\nonumber \\
\pi(b)\, \kett{j\mu n } &:= \beta^+_{j\mu n}\,\kett{j^+ \mu^+ n^-}
+\beta^-_{j\mu n}\,\kett{j^- \mu^+ n^-},\nonumber\\
\pi(a^*)\, \kett{j\mu n } &:= \tilde{\alpha}^+_{j\mu n}\,\kett{j^+
\mu^- n^-}+\tilde{\alpha}^-_{j\mu n}\,\kett{j^- \mu^- n^-},\nonumber
\\
\pi(b^*)\, \kett{j\mu n } &:= \tilde{\beta}^+_{j\mu n}\,\kett{j^+
\mu^- n^+}+\tilde{\beta}^-_{j\mu n}\,\kett{j^- \mu^- n^+}
\end{align}
where
\begin{align*}
\alpha^+_{j\mu n}&:=\sqrt{ q^{\mu +n-1/2} [j+\mu+1] }
\left( \begin{array}{cc}
q^{-j-1/2}\tfrac{ \sqrt{[j+n+3/2]}}{[2j+2]} &   0\\
q^{1/2}\tfrac{ \sqrt{[j-n+1/2]}}{[2j+1][2j+2]} &
q^{-j}\tfrac{\sqrt{[j+n+1/2]}}{[2j+1]}  
\end{array} \right),\\
\alpha^-_{j\mu n}&:=\sqrt{ q^{\mu +n+1/2} [j-\mu] }
\left( \begin{array}{cc}
q^{j+1}\tfrac{ \sqrt{[j-n+1/2]}}{[2j+1]} &   -q^{1/2}
\tfrac{\sqrt{[j+n+1/2]}}{[2j][2j+1]}\\
0 & q^{j+1/2}\tfrac{\sqrt{[j-n-1/2]}}{[2j]}  \\
\end{array} \right),\\
\beta^+_{j\mu n}&:=\sqrt{ q^{\mu +n-1/2} [j+\mu+1] }
\left( \begin{array}{cc}
\tfrac{ \sqrt{[j-n+3/2]}}{[2j+2]} &   0\\
-q^{-j-1}\tfrac{ \sqrt{[j+n+1/2]}}{[2j+1][2j+2]} &
q^{-1/2}\tfrac{\sqrt{[j-n+1/2]}}{[2j+1]}  
\end{array} \right),\\
\beta^-_{j\mu n}&:=\sqrt{ q^{\mu +n-1/2} [j-\mu] }
\left( \begin{array}{cc}
- q^{-1/2}\tfrac{\sqrt{[j+n+1/2]}}{[2j+1]} &   -q^{j}\tfrac{
\sqrt{[j-n+1/2]}}{[2j][2j+1]} \\
0 & -\tfrac{\sqrt{[j+n-1/2]}}{[2j]}  
\end{array} \right)
\end{align*}
with $\tilde{\alpha}^{\pm}_{j\mu n}:=(\alpha^{\mp}_{j^{\pm}\mu^-
n^-})^*$, $\tilde{\beta}^{\pm}_{j\mu n}:=(\beta^{\mp}_{j^{\pm}\mu^-
n^+})^*$ and with the q-number of $\alpha \in \R$ be defined as
$$
[\a]:=\tfrac{q^\a-q^{-\a}}{q-q^{-1}}.
$$
For the purpose of this chapter it is sufficient to use the approximate
spinorial $^*$-representation $\piappr$ of $SU_q(2)$ presented in
\cite{SDLSV,DLSSV} instead of the full spinorial one $\pi$. 

This approximate representation is
$$
\piappr(a) := {a}_+ + {a}_-, \quad \piappr(b) := {b}_+ + {b}_-
$$
with the following definitions, where $q_n:=\sqrt{1-q^{2n}}$:
\begin{align}
    \label{eq:rpnappr}
{a}_+ \,\kett{j\mu n}
&:= q_{j^++\mu^+} \,
\genfrac{(}{)}{0pt}{1}{q_{j^+ +n^+ +1} \qquad
0}{\quad 0 \quad \quad\quad  q_{j^+ +n}}\,\kett{j^+ \mu^+
n^+}
,\nonumber\\
{a}_- \,\kett{j\mu n}
&:= q^{2j+\mu+n+\half} \, \genfrac{(}{)}{0pt}{1}{q \quad  0}{0\quad 1
}\,\kett{j^- \mu^+ n^+},\nonumber\\
{b}_+ \,\kett{j\mu n} &:= q^{j+n-\half} q_{j^+ +\mu^+} \,
\genfrac{(}{)}{0pt}{1}{q
\quad  0}{0\quad 1 }\,\kett{j^+ \mu^+ n^-},\nonumber\\
{b}_- \,\kett{j\mu n}
&:= -q^{j+\mu} \,
\genfrac{(}{)}{0pt}{1}{q_{j^+ + n}  \quad
0\quad}{\quad0  \quad q_{j^- + n}}\,\kett{j^- \mu^+
n^-}.
\end{align}
All disregarded terms are trace-class and do not influence residue
calculations. More precisely, $\pi(x)-\piappr(x) \in \K_{q}$ where
$\K_{q}$ is the principal ideal generated by the operator
\begin{align}
\label{J_q}
J_{q}\, \kett{j\mu n}:= q^j\,\kett{j\mu n}.
\end{align}
Actually, $\K_{q}$ is independent of $q$ and is contained in all
ideals of operators such that $\mu_n=o(n^{-\a})$  (infinitesimal of
order $\a$) for any $\a>0$, and $\K_{q} \subset OP^{-\infty}$.

\vspace{0.5cm}

We define the alternative orthonormal basis $v^{j\up}_{m,l}$ and
$v^{j\dn}_{m,l}$ and the vector notation 
$$
v^j_{m,l}:=\genfrac{(}{)}{0pt}{1}{\,v^{j\up}_{m,l}\,}{\,v^{j\dn}_{m,l}\,}
\text{ where } v^{j\up}_{m,l}:=\ket{j,m-j,l-j^{+},\up},\quad
v^{j\dn}_{m,l}:=\ket{j, m-j, l-j^{-},\dn}.
$$
Here $j\in \half \N$, $0\leq m \leq 2j$, $0\leq l\leq 2j+1$
and $v^{\dn,j}_{m,l}$ is zero whenever $j=0$ or $l = 2j$ or $2j+1$.
The interest is that now, the operators $a_{\pm}$ and $b_{\pm}$
assume simpler form:
\begin{align}
a_{+}\, v^j_{m,l} &= q_{m+1} \,q_{l+1}\, v^{j^+}_{m+1,l+1}\,,\quad
a_-\,v^j_{m,l} = q^{m+l+1}\, v^{j^-}_{m,l} \,,\nonumber\\
b_+\, v^j_{m,l} &= q^{l}\,q_{m+1}\, v^{j^+}_{m+1,l} \,,\qquad \quad
b_-\, v^j_{m,l} = -q^{m}\,q_{l}\,  v^{j^-}_{m,l-1}\,.\label{defonv}
\end{align}
Thus
\begin{align}
a_{+}^*\, v^j_{m,l} &= q_{m} \,q_{l}\, v^{j^-}_{m-1,l-1}\,,\qquad
a_{-}^*\,v^j_{m,l} = q^{m+l+1}\, v^{j^+}_{m,l} \,,\nonumber\\
b_{+}^*\, v^j_{m,l} &= q^{l}\,q_{m}\, v^{j^-}_{m-1,l} \,,\qquad \quad
b_{-}^*\, v^j_{m,l} = -q^{m}\,q_{l+1}\,  v^{j^+}_{m,l+1}\,
\label{defadjointonv}.
\end{align}
Moreover, we have
\begin{align}
    \label{eq:commutrulefora,b}
&a_- a_+ = q^{2}\, a_+ \, a_-\,, \quad b_-b_{+}=q^2\,b_+b_-\,, \qquad
b_+a_+ = q\,a_+b_+\,,  \quad b_-a_{-} = q \, a_-b_-\,,\nonumber\\
&a_{-}^*a_{+}=q^2\, a_{+}a_{-}^*\,, \quad a_{-}^*a_{-}=
a_{-}a_{-}^*\,, \ \ \qquad a_{-}^*b_{+}=q\, b_{+}a_{-}^*\,,
\quad a_{-}^*b_{-}=q\,b_{-}a_{-}^*\,,\nonumber\\
&a_{+}^*a_{-}=q^2\,a_{-}a_{+}^*\,, \quad b_{-}^{*}b_{+}=
b_{+}b_{-}^*\,,\ \ \   \qquad b_{-}^*a_{+}=q\,a_{+}b_{-}^*\,,
\quad a_{-}b_{+}=q\, b_{+}a_{-}\,.
\end{align}
Note for instance that
\begin{align*}
&a_{+}a_{+}^* \, v^j_{m,l}=q_{m}^2q_{l}^2\, v^j_{m,l}\,, \quad
a_{+}^*a_{+}\, v^j_{m,l}=q_{m+1}^2q_{l+1}^2\, v^j_{m,l}\,,\\
& b_{+}b_{+}^* \, v^j_{m,l}=q^{2l}q_{m}^2\, v^j_{m,l}\,,
\quad b_{+}^*b_{+} \, v^j_{m,l}=q^{2l}q_{m+1}^2\, v^j_{m,l}\,,
\end{align*}
so applied to $v^j_{m,l}$, we get the first relation (and similarly
for the others)
\begin{align}
& a_{+}^{*}a_{+}-q^{2}\,a_{+}a_{+}^{*}+q^2\,(b_{+}^{*}b_{+}
-b_{+}b_{+}^{*})=1-q^2,\label{astuce1}\\
& a_{+}a_{+}^*+a_{-}a_{-}^*+b_{+}b_{+}^*+b_{-}b_{-}^*=1,
\label{astuce1'}\\
& a_{+}^*a_{+}+a_{-}^*a_{-}+q^2\,(b_{+}^*b_{+}+b_{-}^*b_{-})=1,
\label{astuce1''}\\
& a_{-}^{*}a_{-}-q^{2}\,a_{-}a_{-}^{*}+q^2\,b_{-}^{*}b_{-}
-q^2\, b_{-}b_{-}^{*}=0,\label{astuce2}\\
& a_{+}a_{-}^*+b_{-}^*b_{+}=0,
\hspace{3cm} a_{-}^*a_{+}+q^2\,b_{-}^*b_{+}=0, \label{astuce3}\\
& a_{-}a_{+}^*+b_{+}^*b_{-}=0,
\hspace{3cm} a_{+}^*a_{-}+q^2\,b_{+}^*b_{-}=0,\label{astuce4}\\
& b_{+}b_{+}^{*}-b_{+}^{*}b_{+}+b_{-}b_{-}^{*}
-b_{-}^{*}b_{-}=0, \label{astuce5}\\
& q\,a_{+}b_{-}-b_{-}a_{+}+q\, a_{-}b_{+}-b_{+}a_{-}=0.
\label{astuce6}
\end{align}

\vspace{0.5cm}

{\it And two others:}

Note that we also use two other infinite
dimensional
$^*$-representations $\pi_{\pm}$ of $\A$ on $\ell^2(\N)$ defined as
follows
on the orthonormal basis $\set{\eps_{n}:n\in\N}$ of $\ell^2(\N)$ by
\begin{align}
 \label{eq:pi-pm}
\pi_\pm(a) \,\eps_n &:= q_{n+1}\,\eps_{n+1},  \qquad
\pi_\pm(b) \,\eps_n := \pm q^n \,\eps_n \, .
\end{align}
These representations are irreducible but not faithful since for
instance $\pi_{\pm}(b-b^*)=0$.

\vspace{0.5cm}

{\it The Dirac operator:}

It is chosen as in the
classical case of a 3-sphere with the round metric:
\begin{align}
\DD \,\kett{j\mu n} :=
\genfrac{(}{)}{0pt}{1}{2j + \sesq \quad\,\,\,0}{\,0 \quad -2j
- \half} \, \kett{j\mu n},  \label{eq:dirac}
\end{align}
which means, with our convention, that
$\DD\,v^j_{ml}=\genfrac{(}{)}{0pt}{1}{2j + \sesq \quad\,\,\,0}{\,0
\quad -2j- \half} \, v^j_{ml}$. 
Note that this operator is invertible (and thus $D = \DD$, $P_0=0$). Moreover, it is asymptotically diagonal with linear
spectrum and

\centerline{ the eigenvalues $2j+\tfrac12$ for $j\in \tfrac12 \N$,
have
multiplicities $(2j+1)(2j+2)$,}
\centerline{the eigenvalues $-(2j+\tfrac12)$ for $j\in \tfrac12
\N^*$, have
multiplicities $2j(2j+1)$.}

So this Dirac operator coincides exactly with the classical one on the
3-sphere (see \cite{Bar, Homma}) with a gap around 0.

Let $\DD=F\vert \DD \vert$ be the polar decomposition of $\DD$, thus
\begin{align}
\vert \DD \vert \,\kett{j\mu n} &=
\genfrac{(}{)}{0pt}{1}{d_{j^+}\,\,\,0}{\,0 \quad d_j} \,
\kett{j\mu n}, \quad d_j:=2j+\half \,, \label{eq:absdirac}
\\
F \,\kett{j\mu n} &=\genfrac{(}{)}{0pt}{1}{1 \quad\,\,\,0}
{\,0 \quad -1} \, \kett{j\mu n}, \label{eq:defF}
\end{align}
and it follows from \eqref{eq:rpnappr} and \eqref{eq:defF} that
\begin{align}
    \label{Fcommutes}
F \text{ commutes with } a_{\pm},\, b_{\pm}.
\end{align}

\quad

{\it The reality operator:}

This antilinear operator $J$ is defined on the basis of $\H$ by
\begin{align}
  \label{DefJ}
J \,\vert j,\mu, n,\up \rangle:=i^{2(2j+\mu+n)}\,  \vert j,-\mu,- n,
\up\rangle, \qquad
J \,\vert j,\mu, n,\dn \rangle:=i^{2(2j-\mu-n)}\,  \vert j,-\mu,- n,
\dn\rangle
\end{align}
thus it satisfies
\begin{align*}
&J^{-1} = -J=J^* \text{ and } \DD J=J\DD, \\
&J  \, v^{j\up}_{m,l}= i^{2(m+l)-1}v^{j\up}_{2j-m,2j+1-l} \,,\qquad
J  \, v^{j\dn}_{m,l}= i^{-2(m+l)+1}v^{j\dn}_{2j-m,2j-1-l}\,.
\end{align*}
\vspace{0.5cm}

We denote

$\quad \B$ the $^*$-subalgebra of $\B(\H)$ generated
by the operators in $\delta^k(\pi(\A))$ for all $k\in \N$,

$\quad \Psi_{0}^{0}(\A)$ the algebra generated by $\delta^{k}\big(
\pi(\A)\big)$ and $\delta^{k}([\DD,\pi(\A)])$ for all $k\in \N$,

$\quad X$ the $^*$-subalgebra of $\B(\H)$ algebraically generated
by the set $\set{a_{\pm},b_{\pm}}$.

Note that $\Psi_0^0(\A)$ is a subalgebra of $\Psi^0(\A)$ (the space of pseudodifferential operators of order less or equal to zero).
\vspace{0.5cm}

{\it The Hopf map $r$}

For the explicit calculations of residues, we need
a $^*$-homomorphism $r:X \rightarrow \pi_{+}(\A) \ox \pi_{-}(\A)$
defined by the tensor product in the sense of Hopf algebras of
representations $\pi_{+}$ and $\pi_{-}$:
\begin{align}
    \label{eq:r}
r(a_+) &:= \pi_+(a) \ox\pi_-(a), &
r(a_-) &:= -q\,\pi_+(b) \ox \pi_-(b^*),\nonumber\\
r(b_+) &:= - \pi_+(a) \ox \pi_-(b), &
r(b_-) &:= - \pi_+(b) \ox \pi_-(a^*).
\end{align}
In fact, $\A$ is a Hopf $^*$-algebra under the coproduct
$\Delta(a):=a\otimes a -q\, b\otimes b^{*},\, \Delta(b):= a\otimes
b+b\otimes a^*$. These homomorphisms appeared in \cite{Woronowicz}
with the translation $\alpha \leftrightarrow a^{*},\, \gamma
\leftrightarrow -b$.
In particular, if
$U:=\genfrac{(}{)}{0pt}{1}{a \quad \,\,\,\,\,\,b}{-qb^* \quad a^*}$
is the canonical generator of the $K_{1}(\A)$-group
$(\Delta a,\, \Delta b)=
(a,\,b) \dot{\otimes}\, U$ where the last $\dot{\otimes}$
means the matrix product of tensors of components.

\vspace{0.5cm}

{\it The grading:}

According to the shift $j\rightarrow j^{\pm}$
appearing in formulae
\eqref{defonv}, \eqref{defadjointonv}, we get a $\Z$-grading on $X$
defined by the
degree $+1$ on
$a_{+},b_{+},{a_{-}}^*,{b_{-}}^*$ and $-1$ on
$a_{-},b_{-},{a_{+}}^*,{b_{+}}^*$.

Any operator $T\in X$ can be (uniquely) decomposed as
$T=\sum_{j\in J\subset \Z} T_j$ where $T_j$ is homogeneous of degree
$j$.

For $T \in X$, $T^\circ$ will denote the $0$-degree part
of $T$ for this grading and by a slight abuse of notations, we write
$r(T)^\circ$ instead of $r(T^\circ)$.
\vspace{0.5cm}

{\it The symbol map:}

We also use the $^*$-homomorphism
$\sg$: $\pi_{\pm}(\A) \to C^{\infty}(S^1)$ defined for $z \in S^1$ on
the generators
by
$$
\sg\big(\pi_{\pm}(a)\big)(z) := z, \;\;\;
\sg\big(\pi_{\pm}(a^*)\big)(z) := \bar z, \;\;\;
\sg\big(\pi_{\pm}(b)\big) (z)=
\sg\big(\pi_{\pm}(b^*)\big)(z):=0.
$$
The application $(\sg \otimes \sg) \circ r $ is defined on $X$ (and
so on $\B$) with values in
$C^{\infty}(S^1) \otimes C^{\infty}(S^1)$.
\vspace{0.5cm}

We define
$$
dT:=[\DD,T] \text{ and } \delta(T):=[\vert \DD \vert, T].
$$
\begin{lemma}
    \label{commutateur}
$a_{\pm}$, $b_{\pm}$ are bounded operators on $\H$ such that for all
$p\in \N$,

(i) $\delta(a_\pm) = \pm a_\pm \, , \quad \delta(b_\pm)=\pm b_\pm\, ,$

(ii) $\delta^p(\ul\pi(a))= a_+ + (-1)^p a_- \, , \quad
\delta^p(\ul\pi(b))= b_+ + (-1)^p b_- \, , $

(iii) $\delta(a_\pm^p) = \pm p\, a_\pm^p \, , \quad \delta(b_\pm^p) =
\pm p\, b_\pm ^p\, .$
\end{lemma}

\begin{proof}
$(i)$ By definition,
$a_\pm \, \kett{j\mu n} =\genfrac{(}{)}{0pt}{1}{\alpha_{\pm}\quad
0\,}{\,0 \quad\beta_{\pm}} \,\kett{j^\pm \mu^+ n^+}$ where the numbers
$\a_\pm$ and $\beta_\pm$ depend on $j$, $\mu$, $n$ and $q$, so we get
by \eqref{eq:absdirac}
\begin{align*}
\delta(a_\pm)  \kett{j\mu n} &=\genfrac{(}{)}{0pt}{1}{(d_{j^{+\pm}})
\a_\pm\quad 0\,}{\quad0 \qquad d_{j^\pm}\,\beta_{\pm}} \,
\kett{j^\pm \mu^+ n^+} -
\genfrac{(}{)}{0pt}{1}{(d_{j^{+}})
\a_\pm\quad 0\,}{\quad0 \qquad d_{j}\,\beta_{\pm}} \,\kett{j^\pm \mu^+
n^+} \\
&=  \genfrac{(}{)}{0pt}{1}{\pm\alpha_{\pm}\quad
0\,}{\,0 \quad \pm \beta_{\pm}} \, \kett{j^\pm \mu^+
n^+}=\pm a_\pm  \,\kett{j\mu n}
\end{align*}
and similar proofs for $b_\pm$.

$(ii)$ and $(iii)$ are straightforward consequences of $(i)$ and
definition of $\ul \pi$.
\end{proof}

\begin{remark}
\label{pseudodiff}
By Lemma \ref{commutateur}, we see that, modulo $OP^{-\infty}$, $X$
is equal to $\B$ and in particular contains $\pi(\A)$.

Using \eqref{Fcommutes}, we get that
$\B\subset \Psi_{0}^0(\A) \subset \text{algebra generated by
$\B$ and $\B F$}$.
\end{remark}
Note that, despite the last inclusion, $F$ is not a priori
in $\Psi_{0}^0(\A)$.

\quad

\subsection{The noncommutative integrals}

Recall that for any pseudodifferential operator $T$,
$\ncint T:=\underset{s=0}{\Res} \,\zeta_\DD^T(s)$ where
$\zeta_\DD^T(s):=\Tr(T \vert \DD \vert^{-s})$.

\begin{theorem}
    \label{Theo}
The dimension spectrum (without reality structure given by $J$) of
the spectral triple
$\big(\A(SU_q(2)),\H,\DD\big)$ is simple and equal to $\{1,2,3\}$. 

Moreover,  the corresponding residues for
$T \in \B$ are
\begin{align*}
&\ncint T |\DD|^{-3} = 2(\tau_1 \ox \tau_1) \bigl(r(T)^\circ \bigr),\\
&\ncint T |\DD|^{-2}= 2 \bigl(\tau_1 \ox \tau_0 +
\tau_0 \ox \tau_1 \bigr) \bigl(r(T)^\circ\bigr),\\
&\ncint T |\DD|^{-1}= \big(2 \, \tau_0 \ox \tau_0 - \tfrac{1}{2}
\, \tau_1 \ox \tau_1\big)\bigl(r(T)^\circ\bigr),\\
&\ncint F\,T |\DD|^{-3} =0,\\
&\ncint F\,T |\DD|^{-2}= 0,\\
&\ncint F\,T |\DD|^{-1}=\big(\tau_{0} \ox \tau_{1}-\tau_{1} \ox
\tau_{0}\big)\big(r(T)^\circ\bigr),
\end{align*}
where the functionals $\tau_0$, $\tau_1$ are defined for $x \in
\pi_{\pm}(\A)$ by
\begin{align*}
\tau_0(x) := \lim_{N\to\infty} \big(\Tr_N x -(N+1) \,
\tau_1(x)\big), \qquad
\tau_1(x) := \tfrac{1}{2\pi} \int_{0}^{2\pi} \sg(x)(e^{i\theta})
\,d\theta,
\end{align*}
with $\Tr_N x=\sum_{n=0}^N \langle \eps_{n},x\,\eps_{n}\rangle$.
\end{theorem}

\begin{proof}
Consequence of \cite[Theorem 4.1 and (4.3)]{SDLSV}.
\end{proof}

\begin{remark}
Since $F$ is not in $\B$, the equations of Theorem \ref{Theo} are not
valid for all $T\in\Psi_{0}^0(\A)$.\\
But when $T \in \Psi_{0}^0(\A)$, $\ncint T|\DD|^{-k}=0  \text{ for k}
\notin \set{1,2,3}$ since the dimension spectrum is $\set{1,2,3}$
\cite{SDLSV}.
\end{remark}
Compared to \cite{SDLSV} where we had
\begin{align*}
\tau_0^\up(x) := \lim_{N\to\infty} \Tr_N x - (N+\sesq) \,
\tau_1(x),\;\;\;\;\;
\tau_0^\dn(x) := \lim_{N\to\infty} \Tr_N x - (N+\half) \, \tau_1(x),
\end{align*}
we replaced them with $\tau_0$:
$$
\tau_0^\up = \tau_0 - \tfrac{1}{2} \tau_1, \;\;\;\;\;
\tau_0^\dn = \tau_0 + \tfrac{1}{2} \tau_1.
$$
Note that $\tau_{1}$ is a trace on $\pi_{\pm}(\A)$ such that
$\tau_{1}(1)=1$ and $\tau_{1}\big(\pi_{+}(aa^*)\big)=\tfrac{1}{2\pi}
\int_{0}^{2\pi}1\,d\theta=1$, while $\tau_{0}$ is not since $\tau_{0}(1)=0$ and
\begin{align}
    \label{eq:tau0}
\tau_{0}\big(\pi_{\pm}(aa^*)\big)=\lim_{N\rightarrow
\infty}\sum_{n=0}^\infty (1-q^{2n}) -(N+1)=-\tfrac{1}{1-q^2},
\end{align}
so, because of the shift, the substitution $a\leftrightarrow a^*$ gives
\begin{align}
    \label{nontrace}
\tau_{0}\big(\pi_{\pm}(a^*a)\big)=q^2 \,
\tau_{0}\big(\pi_{\pm}(aa^*)\big).
\end{align}

\subsection{The tadpole}

\begin{lemma}
For $SU_q(2)$, the condition of the vanishing tadpole (see
\cite{ConnesMarcolli}) is not satisfied.
\end{lemma}

\begin{proof}
For example, an explicit calculation gives
$ \ncint \pi(b) [\DD, \pi(b^*)] \DD^{-1} = \tfrac{2}{1 - q^2} $:

Let $x,y \in \ul \pi(\A)$. Since $[F,x]=0$, we have
$$ \ncint x [\DD,y] \DD^{-1} = \ncint x \delta(y) |\DD|^{-1}
= \tau'\big(r(x \delta(y)\big)^0) $$
where $\tau' :=2 \, \tau_0 \ox \tau_0 - \tfrac{1}{2} \, \tau_1 \ox
\tau_1$.

By Lemma \ref{commutateur}, $\ul \pi(b) \delta \big(\ul \pi(b^*)\big)
= (b_+ + b_-)\big((b_-)^* -
(b_+)^*\big)=-b_{+}{b_{+}}^*
+b_{-}{b_{-}}^*+b_{+}{b_{-}}^*-b_{-}{b_{+}}^*$.
Since only the first two terms have degree $0$, we get, using the
formulae from
Theorem \ref{Theo}
\begin{align*}
\tau'\big(r(-b_{+}{b_{+}}^*)\big)&=-\tau'\big(\pi_{+}(aa^*)\otimes
\pi_{-}(bb^*)\big)\\
&=-2\tau_{0}\big(\pi_{+}(aa^*)\big)
\tau_{0}\big(\pi_{-}(bb^*)\big)+\tfrac12\tau_{1}\big(\pi_{+}(aa^*)\big)
\tau_{1}\big(\pi_{-}(bb^*)\big)
\end{align*}
and 
$\tau_{1}\big(\pi_{-}(bb^*)\big)=0$. Similarly, using \eqref{nontrace}
$$
\tau'\big(r(b_{-}{b_{-}}^*)\big)=2\tau_{0}\big(\pi_{+}(bb^*)\big)
\tau_{0}\big(\pi_{-}(a^*a)\big)=
2q^2\tau_{0}\big(\pi_{-}(aa^*)\big)\tau_{0}\big(\pi_{+}(bb^*\big).
$$
Since $\tau_{0}\big(\pi_{\pm}(bb^*)\big)=\Tr\big(\pi_{\pm}(bb^*)\big)=
\sum_{n=0}^\infty q^{2n}=\tfrac{1}{1-q^2}$ and \eqref{eq:tau0},
\begin{align*}
\ncint \pi(b) [\DD,\pi(b^*)] \DD^{-1} =
2\tfrac{1}{1-q^2}\tfrac{1}{1-q^2}
+2q^2 \tfrac{-1}{1-q^2}  \tfrac{1}{1-q^2}
= \tfrac{2}{1-q^2}.
\tag*{\qed}
\end{align*}
\hideqed
\end{proof}
In particular the pairing of the tadpole cyclic cocycle
$\phi_{1}$ with the generator of $K_{1}$-group is nontrivial:

\begin{remark} Other examples: with the shorthand $x$ instead of
$\piappr(x)$,
\begin{align*}
&(\tau_1\otimes \tau_1)\, r\big(a\delta(a^*)^\circ \big) =-1 ,&&
(\tau_1\otimes \tau_1) r\big(a^*\delta(a)^\circ\big) =1,\\
&(\tau_0\otimes\tau_0)\, r\big(a\delta(a^*)^\circ \big) =
\tfrac{1}{q^2-1}\,, &&
(\tau_0\otimes\tau_0)\, r\big(a^*\delta(a)^\circ\big) =
\tfrac{q^2}{q^2-1}\,, \\
&\ncint a\delta(a^*)|\DD|^{-1} = \tfrac{q^2+3}{2(q^2-1)}\,,&&
\ncint a^*\delta(a)|\DD|^{-1} =  \tfrac{3q^2+1}{2(q^2-1)}\,,\\
&\ncint b\delta(b) |\DD|^{-1} =0,&&  \ncint b^*\delta(b^*)
|\DD|^{-1}=0,\\
&\ncint b\delta(b^*) |\DD|^{-1} =\tfrac{-2}{q^2-1}\,,&&  \ncint
b^*\delta(b) |\DD|^{-1}=\tfrac{-2}{q^2-1}\,.
\end{align*}

In particular, $N\phi_1$ does not vanish on 1-forms since
$\int_{N\phi_{1}}ada^*=N\phi_1(a,a^*)=-1$.
\end{remark}

Let U be the canonical generator of the $K_{1}(\A)$-group,
$U=\genfrac{(}{)}{0pt}{1}{a \quad \,\,\,\,\,\,b}{-qb^* \quad a^*}$
acting on $\H \ox \C^2$. Then for $A_{U}:=\sum_{k,l=1}^2
\pi(U_{kl})\,d\pi({U^*}_{kl})$, using above remark,
$\int_{\phi_{1}}A_{U}=-2$ as obtained in \cite[page 391]{SDLSV}: in
fact, with $P:=\tfrac 12 (1+F)$,
$$\psi_{1}(U,U^*):=2\sum_{k,l}\ncint U_{kl}\delta(U^*_{kl}) P\vert
\DD\vert^{-1}
-\ncint U_{kl}\delta^2(U^*_{kl}) P\vert \DD\vert^{-2}
+\tfrac 23 \ncint U_{kl}\delta^3(U^*_{kl}) P\vert \DD\vert^{-3}$$
satisfies $\psi_{1}(U,U^*)=2\sum_{k,l}\ncint U_{kl}\delta(U^*_{kl})
P\vert
\DD\vert^{-1}=\int_{\phi_{1}}A_{U}$.

\section{Reality operator and spectral action on $SU_q(2)$}

\subsection{Spectral action in dimension 3 with $[F,\A]\in
OP^{-\infty}$} 
Let $(\A,\H,\DD)$ a be real spectral triple of dimension 3. Assume 
 that $[F,\A]\in OP^{-\infty}$, where $F:=\DD
|\DD|^{-1}$ (we suppose $\DD$ invertible).
Let $\Abb$ be a selfadjoint one form, so $\Abb$ is  of the form
$\Abb=\sum_{i}a_i d b_i$ where
$a_{i},b_{i} \in \A$.

Thus, $\Abb \simeq AF \mod OP^{-\infty}$ where
$A:=\sum_{i}a_i\delta(b_i)$ is the $\delta$-one-form associated to
$\Abb$. Note that $A$ and
$F$ commute modulo $OP^{-\infty}$.

We define
\begin{align*}
D_\Abb &:= \DD_\Abb + P_\Abb,
\quad P_\Abb \text{ the projection on } \Ker \DD_\Abb,\\
\DD_\Abb&:= \DD+\wt \Abb, \quad \wt \Abb := \Abb + J\Abb J^{-1}.
\end{align*}

\begin{theorem}
\label{coeffsASJ}
The coefficients of the full  spectral action (with reality operator)
on any real spectral triple $(\A,\H,\DD)$ of dimension 3 such that
$[F,\A]\in OP^{-\infty}$ are
\begin{align*}
&(i) \hspace{1cm} \ncint \vert D_{\Abb}\vert^{-3}  = \ncint
|\DD|^{-3}.\\
&(ii) \hspace{.9cm} \ncint \vert D_{\Abb}\vert^{-2}  = \ncint
|\DD|^{-2} - 4\ncint A|\DD|^{-3}.\\
&(iii) \hspace{.8cm}\ncint \vert D_{\Abb}\vert^{-1}  = \ncint
|\DD|^{-1}-2\ncint A|\DD|^{-2} +2\ncint A^2 |\DD|^{-3}
+2\ncint AJAJ^{-1}|\DD|^{-3}  .\\
&(iv)  \hspace{.9cm} \zeta_{D_{\Abb}}(0)  = \zeta_{D}(0)-2\ncint
A|\DD|^{-1} +
\ncint A(A+JAJ^{-1})|\DD|^{-2} +  \ncint
\delta(A)(A+JAJ^{-1})|\DD|^{-3}\\
&\hspace{4cm}-\tfrac{2}{3}\ncint
A^3|\DD|^{-3} - 2 \ncint A^2
JAJ^{-1}|\DD|^{-3}.
\end{align*}
\end{theorem}

\begin{proof}
$(i)$ We apply Proposition \ref{invariance1}.

$(ii)$ By Lemma \ref{residus-particuliers}, we have  $\ncint \vert
D_{\Abb}\vert^{-2}=\ncint |\DD|^{-2}-\ncint (\wt \Abb \DD+\DD \wt
\Abb +\wt \Abb^2)|\DD|^{-4}$. By the trace property of the
noncommutative integral and the fact that
$\wt \Abb^2 |\DD|^{-4}$ is trace-class, we get $\ncint \vert
D_{\Abb}\vert^{-2}= \ncint |\DD|^{-2}-2\ncint \wt \Abb \DD
|\DD|^{-4}=\ncint |\DD|^{-2}-4\ncint \Abb \DD |\DD|^{-4}$. Since
$\Abb \DD \sim A|\DD| \mod OP^{-\infty}$, we get the result.

$(iii)$ By Lemma \ref{residus-particuliers} $(ii)$, we have
$$\ncint \vert
D_{\Abb}\vert^{-1}=\ncint |\DD|^{-1}-\half \ncint (\wt \Abb \DD+\DD
\wt \Abb +\wt \Abb^2)|\DD|^{-3} + \tfrac{3}{8}\ncint (\wt \Abb
\DD+\DD \wt \Abb +\wt \Abb^2)^2 |\DD|^{-5}.$$
Following arguments of $(ii)$, we get
\begin{align*}
\ncint (\wt \Abb \DD+\DD \wt \Abb +\wt \Abb^2)|\DD|^{-3}&= 4\ncint
A|\DD|^{-2} +2 \ncint A^2 |\DD|^{-3} + 2 \ncint AJAJ^{-1}
|\DD|^{-3},\\
\ncint (\wt \Abb \DD+\DD \wt \Abb +\wt \Abb^2)^2|\DD|^{-5}&=8\ncint
A^2 |\DD|^{-3} + 8\ncint AJAJ^{-1}|\DD|^{-3},
\end{align*}
and the result follows.

$(iv)$ By \eqref{termconstanttilde},
$\zeta_{D_{\Abb}}(0)  =  \sum_{j=1}^{3} \tfrac{(-1)^j}{j}\ncint (\wt
\Abb\DD^{-1})^j\, .$

Moreover, the following holds:
$\ncint \wt \Abb \DD^{-1} = 2\ncint A|\DD|^{-1}$ and
$\ncint  (\wt\Abb \DD^{-1})^2= 2\ncint (A|\DD|^{-1})^2 + 2\ncint
A|\DD|^{-1}JAJ^{-1}|\DD|^{-1}$.
Since $\delta(A)\in OP^{0}$, we can check that $\ncint
(A|\DD|^{-1})^2 = \ncint A^2 |\DD|^{-2}+\ncint \delta(A)A|\DD|^{-3}$
and, with the same argument, that $\ncint
A|\DD|^{-1}JAJ^{-1}|\DD|^{-1}=\ncint A JAJ^{-1}|\DD|^{-2}+\ncint
\delta(A)
JAJ^{-1} |\DD|^{-3}$. Thus, we get
\begin{equation}\label{ncintAD2}
\ncint  (\wt\Abb \DD^{-1})^2 = 2\ncint A(A+JAJ^{-1})|\DD|^{-2} +2
\ncint \delta(A)(A+JAJ^{-1})|\DD|^{-3}.
\end{equation}
The third term to be computed is
$$
\ncint  (\wt\Abb \DD^{-1})^3= 2\ncint (A|\DD|^{-1})^3 + 4\ncint
(A|\DD|^{-1})^2 JAJ^{-1}|\DD|^{-1}+2\ncint
A|\DD|^{-1}JAJ^{-1}|\DD|^{-1} A|\DD|^{-1}.
$$
Any operator in $OP^{-4}$ being trace-class here, we get
\begin{equation}
    \label{ncintAD3}
\ncint  (\wt\Abb \DD^{-1})^3 = 2\ncint A^3|\DD|^{-3} + 4\ncint A^2
JAJ^{-1}|\DD|^{-3}+2\ncint AJAJ^{-1} A|\DD|^{-3}.
\end{equation}
Since $\ncint AJAJ^{-1} A|\DD|^{-3}=\ncint A^2
JAJ^{-1}|\DD|^{-3}$
by trace property and the fact that $\delta(A)\in OP^{0}$,
the result follows then from (\ref{ncintAD2}) and (\ref{ncintAD3}).
\end{proof}

\begin{corollary} For the spectral action of $\Abb$ without
the reality operator (i.e. $\DD_{\Abb}=\DD+\Abb$), we get
\begin{align*}
&\ncint \vert D_{\Abb}\vert^{-2}  = \ncint |\DD|^{-2} - 2\ncint
A|\DD|^{-3},\\
&\ncint \vert D_{\Abb}\vert^{-1}  = \ncint
|\DD|^{-1}-\ncint A|\DD|^{-2} +\ncint A^2 |\DD|^{-3},\\
&\zeta_{D_{\Abb}}(0)  = \zeta_{\DD}(0) -\ncint A|\DD|^{-1} +\half
\ncint A^2|\DD|^{-2} +  \half \ncint
\delta(A)A|\DD|^{-3}-\tfrac{1}{3}\ncint
A^3|\DD|^{-3}.
\end{align*}
\end{corollary}

\subsection{Spectral action on $SU_q(2)$: main result}

On $SU_q(2)$, since $F$ commutes with $a_\pm$ and $b_\pm$, Theorem \ref{coeffsASJ} can be used for the spectral action computation.

Here is the main result of this section

\begin{theorem}
    \label{mainThmJ}
In the full spectral action \eqref{formuleaction}
(with the reality operator) of $SU_q(2)$ for a one-form $\Abb$ and
$A$ its associated $\delta$-one-form, the coefficients are:
\begin{align*}
&\ncint \vert D_{\Abb}\vert^{-3}  = 2\,, \\
& \ncint\vert D_{\Abb}\vert^{-2}  =  - 4\, \ncint A |\DD|^{-3},\\
&\ncint \vert D_{\Abb}\vert^{-1}  = -\half +2\big( \ncint
A^2|\DD|^{-3}- \ncint A|\DD|^{-2} \big)+ \big\vert\ncint
A|\DD|^{-3}\big\vert^2 ,\\
&\zeta_{D_{\Abb}}(0) =-2 \ncint A|\DD|^{-1} +\ncint A^2 |\DD|^{-2}
-\tfrac{2}{3} \ncint A^3 |\DD|^{-3}\\
&\hspace{3cm}+\overline{\ncint A |\DD|^{-3}}\big(\half\ncint
A|\DD|^{-2} -\ncint A^2 |\DD|^{-3}\big) +\half \ncint
A|\DD|^{-3}\overline{\ncint A|\DD|^{-2}}\, .
\end{align*}
\end{theorem}

In order to prove this theorem (in section 4.7), we will use a decomposition of
one-forms in the Poincar\'e-Birkhoff--Witt basis of $\A$ with an
extension of previous representations to operators like $T JT'
J^{-1}$ where $T$ and $T'$ are in $X$. 

\subsection{Balanced components and Poincar\'e--Birkhoff--Witt basis
of $\A$}
  \label{balancedcomponent}

Our objective is to compute all integrals in term of $A$ and the
computation will lead to functions of $A$ which capture certain
symmetries on $A$.

Let $\Abb=\sum_i \pi(x^i) d \pi(y^i) $ on $SU_q(2)$ be a one-form and
$A$ the associated $\delta$-one-form. The $x^i$ and $y^i$ are in $\A$
and as such
they can be uniquely written as finite sums $x^i=\sum_\a x^i_\a m^\a$
and $y^i = \sum_\beta y^i_\beta m^{\beta}$ where
$m^\a:=a^{\a_1}b^{\a_2}{b^*}^{\a_3}$ is the canonical monomial of
$\A$ with $\a,\beta \in \Z \times  \N \times  \N$ based on a fixed
Poincar\'e--Birkhoff--Witt type basis of $\A$.

\begin{remark}
Any one-form $\Abb=\sum_i \pi(x^i) d \pi(y^i) $ on $SU_q(2)$ is
characterized by a complex valued
matrix $A_{\a}^\beta=\sum_{i}x_\a^i\,y_\beta^i$ where $\a,\beta \in
\Z\times\N\times\N$. This matrix is such that
$$
A= A_{\a}^\beta\, M^{\a}_\beta
$$
where $M^{\a}_\beta:= \pi(m^\a)\delta \big(\pi(m^\beta) \big)$.

In the following, we note 
$$
\bar A:= \bar A^\beta_\a\, M^\a_\beta
$$
so for any $p\in \N$, $\ncint \bar A |\DD|^{-p} = \overline{ \ncint A
|\DD|^{-p}}$. 
\end{remark}
This presentation of one-forms is not unique modulo $OP^{-\infty}$
since, as we will see in section 5, 
$F = \sum_i x_i dy_i$ where $x_i,y_i \in\A$, thus for any generator
$z$,
$[F, z] = \sum_i x_i d( y_i z) - x_i y_i dz - z x_i dy_i = 0 \mod
OP^{-\infty}$.
We do not know however if this presentation is unique when the
$OP^{-\infty}$ part is taken into account.

\quad

The $\delta$-one-forms $M^\a_\beta$ are said to be {\itshape
canonical}. Any product of $n$ canonical $\delta$-one forms,
where $n\in \N^*$, is called a {\itshape canonical
$\delta^n$-one-form}. Thus, if $A$ is a $\delta$-one-form,
$A^n = (A^n)_{\bar \a}^{\bar \beta}\, M^{\bar\a}_{\bar \beta}$ where
$\bar \a=(\a,\a',\cdots,\a^{(n-1)})$,  $\bar
\beta=(\beta,\beta',\cdots,\beta^{(n-1)})$ are in $\Z^n\times \N^n
\times \N^n$,
$(A^n)_{\bar \a}^{\bar\beta} := A_{\a}^{\beta} \cdots
A_{\a^{(n-1)}}^{\beta^{(n-1)}}$ and $M^{\bar\a}_{\bar\beta}$ is the
canonical $\delta^n$-one form equal to $
M^{\a}_{\beta} \cdots M^{\a^{(n-1)}}_{\beta^{(n-1)}}$.

\begin{definition}
A canonical $\delta^n$-one-form is $a$-balanced if it is of the form
$$
a^{\a_1}\delta(a^{\beta_1})\cdots
a^{\a^{(n-1)}_1}\delta(a^{\beta^{(n-1)}_1})
$$
where $\sum_{i=0}^{n-1} \a^{(i)}_1+\beta^{(i)}_1 = 0$.

For any $\delta$-one-form $A$, the $a$-balanced components of $A^n$
are denoted $B_a(A^n)_{\bar \a}^{\bar\beta}$.
\end{definition}

Note that
$$
B_a(A)_{\bar \a}^{\bar\beta} = A_{-\beta_1 0 0}^{\beta_1 0 0}\,
\delta_{\a_1+\beta_1,0} \,\delta_{\a_2+\a_3+\beta_2+\beta_3,0}.
$$
\begin{definition}
A canonical $\delta^n$-one-form is balanced if it is of the form
$$
m^{\a}\delta(m^{\beta})\cdots m^{\a^{(n-1)}}\delta(m^{\beta^{(n-1)}})
$$
where $\sum_{i=0}^{n-1} \a^{(i)}_1+\beta^{(i)}_1 = 0$ and
$\sum_{i=0}^{n-1} \a^{(i)}_2+\beta^{(i)}_2  = \sum_{i=0}^{n-1}
\a^{(i)}_3+\beta^{(i)}_3$.

For any $\delta$-one-form $A$, the balanced components of $A^n$ are
denoted $B(A^n)_{\bar \a}^{\bar\beta}$.
\end{definition}
Note that
$$
B(A)_{\bar \a}^{\bar\beta} = A_{-\beta_1 \a_2 \a_3}^{\beta_1 \beta_2
\beta_3}\,\delta_{\a_1+\beta_1,0}\,
\delta_{\a_2+\beta_2,\a_3+\beta_3}.
$$
As we will show, a contribution to the $k^{th}$-coefficient in the
spectral action, is only brought by one-forms $\Abb$ such that $A^k$
is balanced (and even $a$-balanced
in the case $k=1$).

Note also that if $A$ is balanced, then $A^{k}$ for $k\geq 1$ is also
balanced, whereas the converse is false.

\subsection{The reality operator $J$ on $SU_q(2)$}
Let for any $n,p\in \N$,
\begin{align*}
 &q_n:=\sqrt{1-q^{2n}},
&q_{-n}&:= 0 \text{ if } n>0 ,\\
&q_{n,p}^\up := q_{n+1}\cdots q_{n+p}\,,
&q_{n,p}^\dn &:= q_n \cdots q_{n-(p-1)}\, ,
\end{align*}
with the convention $q_{n,0}^\up= q_{n,0}^\dn := 1$.  Thus, we have
the relations
\begin{align*}
&\pi_{\pm}(a^p)\ \eps_n = q_{n,p}^\up\  \eps_{n+p}\,, &
\pi_{\pm}({a^*}^p)&\ \eps_n  = q_{n,p}^\dn\ \eps_{n-p}\,,\\
&\pi_{\pm}(b^p)\  \eps_n  =( \pm
q^{n})^p\ \eps_n\,, &\pi_{\pm}({b^*}^p)&\  \eps_n=( \pm
q^{n})^p\ \eps_n \,,
\end{align*}
where $\eps_{k} := 0$ if $k<0$.

The sign of $x\in \R$ is denoted $\eta_x$. By convention, $a_j := a$,
$a_{\pm,j}:=a_{\pm}$ if $j\geq0$ and $a_j:=a^*$,
$a_{\pm,j}:=a_{\pm}^*$
if $j<0$.
Note that, with convention
$$
q_{n,p}^{\up_{\a_1}}:= q_{n,p}^\up
\text{ if } {\a_1}>0, \quad
q_{n,p}^{\up_{\a_1}}:= q_{n,p}^\dn \text{ if } {\a_1}<0
, \text{ and } q_{n,p}^{\up_0}:=1,
$$
we have for any ${\a_1} \in \Z$ and $p\leq
{\a_1}$,
$\pi_{\pm}(a_{\a_1}^p)\, \eps_n = q_{n,p}^{\up_{\a_1}}\
\eps_{n+\eta_{\a_1} p}$.

Recall that the reality operator $J$ is defined by
\begin{align*}
J  \, v^{j\up}_{m,l}= i^{2(m+l)-1}v^{j\up}_{2j-m,2j+1-l} \,,\qquad
J  \, v^{j\dn}_{m,l}= i^{-2(m+l)+1}v^{j\dn}_{2j-m,2j-1-l}\,,
\end{align*}
thus the real conjugate operators
$$
\wh a_\pm:=Ja_{\pm}J^{-1}, \quad \wh b_\pm:=J b_\pm J^{-1}
$$
satisfy
\begin{align*}
\wh a_+ \, v^j_{m,l} &:=
-q_{2j+1-m} \, \genfrac{(}{)}{0pt}{1}{q_{2j+2-l}\quad 0 \quad}
{\quad0 \qquad q_{2j-l}}  \, v^{j^+}_{m,l}\,, \quad
&\wh a_- \, v^j_{m,l} &:=
-q^{2j-m} \, \genfrac{(}{)}{0pt}{1}{q^{2j+2-l}\quad 0\quad}{\quad 0
\quad
q^{2j-l}} \, v^{j^-}_{m-1,l-1},\\
\wh b_+ \, v^j_{m,l} &:=
q_{2j+1-m} \, \genfrac{(}{)}{0pt}{1}{q^{2j+1-l} \quad 0 \quad}{\quad
0\quad q^{2j-1-l}} \, v^{j^+}_{m,l+1}\,, \quad
&\wh b_-  \,v^j_{m,l} &:=
-q^{2j-m} \,\genfrac{(}{)}{0pt}{1}{q_{2j+1-l}\quad 0 \quad}{\quad 0
\quad q_{2j-1-l}} \, v^{j^-}_{m-1,l}\,.
\end{align*}
So the real conjugate operator behaves differently on the
up and down part of the Hilbert space. The difference comes from the
fact that the index $l$ is not treated uniformly by $J$ on up and
down parts.

We denote $\wh X$ the algebra generated by $\set{\wh a_\pm, \wh
b_\pm}$, $\wt X$ the algebra generated by $\set{a_{\pm},b_{\pm}, \wh
a_\pm, \wh b_\pm}$ and
$\H':=\ell^2(\N) \otimes \ell^2(\Z)$ and we construct
two $*$-representations $\wh \pi_{\pm}$ of $\A$:

The representation $\wh \pi_+$ gives bounded operators
on $\H'$ while $\wh \pi_-$ represents $\A$ into $\B(\H'\otimes \C^2)$.

The representation $\wh \pi_+$ is defined on the generators by:
$$
\wh \pi_{+}(a)\, \eps_m \otimes \eps_{2j}:=q_{2j+1-m}\, \eps_m\otimes
\eps_{2j+1},
\qquad \wh \pi_{+} (b)\, \eps_m \otimes \eps_{2j}:=-q^{2j-m}\,
\eps_{m+1}\otimes \eps_{2j+1}
$$
while $\wh \pi_-$ is defined by:
\begin{align*}
 \wh \pi_{-} (a) \, \eps_l \otimes \eps_{2j}\otimes \eps_{\up\dn}&:=
-q_{2j+1\pm1-l}\, \eps_l \otimes \eps_{2j+1}\otimes \eps_{\up\dn}\,
,\\
  \wh \pi_{-} (b)\, \eps_l \otimes \eps_{2j}\otimes \eps_{\up\dn} &:=
-q^{2j\pm1-l}\, \eps_{l+1}\otimes \eps_{2j+1}\otimes \eps_{\up\dn}\, ,
 \end{align*}
where $\eps_{\up\dn}$ is the canonical basis of $\C^2$ and the $+$ in
$\pm$ corresponds to $\up$ in $\up\dn$.

The link between $\wh \pi_{\pm}$ and $\pi_{\pm}$ which explains the
notations about these intermediate objects
and the fact that $\wh \pi_{\pm}$ are representations on
different Hilbert spaces, is in the parallel between equations
\eqref{eq:r}, \eqref{eq:hat} and \eqref{eq:prime}.

Let us give immediately a few properties ($x_\beta$ equals $x$ if
the sign $\beta$ is positive and equals $x^*$ otherwise)
\begin{align*}
\wh \pi_+(a_\beta)^p \,\eps_m\otimes \eps_{2j} &=
q^{\up_\beta}_{2j-m,p} \, \eps_m\otimes\eps_{2j+\eta_\beta p}\,,\\
\wh \pi_-(a_\beta)^p\, \eps_l\otimes \eps_{2j}\otimes\eps_{\up\dn} &=
(-1)^p \, q^{\up_\beta}_{2j\pm1-l,p} \,
\eps_l\otimes\eps_{2j+\eta_\beta p}\otimes\eps_{\up\dn}\,,\\
\wh \pi_+(b_\beta)^p \,\eps_m\otimes \eps_{2j}&= (-1)^p\,
q^{(2j-m)p}\,
\eps_{m+\eta_\beta p}\otimes \eps_{2j+\eta_\beta p}\,, \\
\wh \pi_-(b_\beta)^p \,\eps_l\otimes \eps_{2j}\otimes\eps_{\up\dn}&=
(-1)^p\, q^{(2j\pm1-l)p}\, \eps_{l+\eta_\beta p}\otimes
\eps_{2j+\eta_\beta p}\otimes\eps_{\up\dn}\,.
\end{align*}

Note that the $\wh \pi_\pm$ representations still contain the shift
information, contrary to representations  $\pi_\pm$.
Moreover, $\wh \pi_\pm(b)\neq \wh \pi_\pm (b^*)$ while
$\pi_\pm(b) = \pi_\pm (b^*)$.

The operators $\wh a_\pm, \wh b_\pm$ are coded on $\H'\otimes
\H'\otimes \C^2$ as the correspondence
\begin{align}
    \label{eq:hat}
&\wh a_+ \longleftrightarrow  \wh \pi_+ (a) \otimes \wh \pi_-(a)
,\quad
&\wh a_- &\longleftrightarrow  -q \,\wh \pi_+(b^*) \otimes
\wh\pi_-(b^*) ,\nonumber\\
&\wh b_+\longleftrightarrow -\wh \pi_+(a)\otimes \wh \pi_-(b)
,\quad
&\wh b_-& \longleftrightarrow -\wh \pi_+(b^*) \otimes \wh
\pi_-(a^*).
\end{align}

We now set the following extension to $\B(\H')$ of  $\pi_+$ and to
$\B(\H'\otimes \C^2)$ of $\pi_-$ by
\begin{align*}
\pi'_{+}(a)&:=\pi_{+}(a)\otimes V, \quad \pi'_{+}(b) := \pi_+
(b)\otimes V \quad\text{ (V is the shift of }\ell^2(\Z)),\\
\pi'_{-}(a)&:=\pi_{-}(a)\otimes V \otimes 1_2, \quad \pi'_{-}(b) :=
\pi_- (b)\otimes V \otimes 1_2.
\end{align*}
So, we can define a canonical algebra morphism $\wt \rho\,$ from $\wt
X$
into the bounded operators on
$\H' \otimes\H'\otimes \C^2$.  This morphism is defined on the
generators part $\set{\wh a_{\pm},\wh b_{\pm}}$ of $\wt X$
by preceding correspondence and on the
generators part $\set{ a_{\pm}, b_{\pm}}$ by $-$see \eqref{eq:r}:
\begin{align}
    \label{eq:prime}
 & a_+ \longleftrightarrow   \pi'_+ (a) \otimes \pi'_-(a) ,\quad
 & a_- & \longleftrightarrow     -q \,\pi'_+(b^*) \otimes \pi'_-(b^*)
 ,\nonumber\\
 & b_+\longleftrightarrow       - \pi'_+(a)\otimes  \pi'_-(b) ,\quad
 & b_- & \longleftrightarrow     - \pi'_+(b^*) \otimes  \pi'_-(a^*).
\end{align}

We denote $S$ the canonical surjection from $\H'\otimes \H'\otimes\C^2$
onto $\H$. This surjection is associated to the parameters
restrictions on $m,j,l,j'$. In  particular, the index $j'$ associated
to the second $\ell^2(\N)$ in $\H'\otimes \H'\otimes \C^2$ is set to
be equal to $j$.
Any vector in $\H'\otimes \H'\otimes\C^2$ not
satisfying these restrictions is sent to 0 in $\H$.

Denote by $I$ the canonical injection
of $\H$ into $\H'\otimes \H'\otimes \C^2$ (the index $j$ is doubled).
Thus, $S\wt \rho(\cdot) I$ is the identity on
$\wt X$.

In the computation of residues of $\zeta_{\DD}^T$ functions,
we can therefore replace the operator $T$ by  $S\wt \rho(T) I$.

We now extend $\tau_0$ on $\pi'_\pm (\A) \wh \pi_\pm(\A)$: For $x,y
\in \A$, we set
\begin{align*}
\Tr_{N} \big(\pi'_+(x)\wh \pi_+(y) \big) & := \sum_{m=0}^{N}
\,\langle \eps_m\otimes \eps_{N} \,,\, \pi'_+(x)\wh
\pi_+(y)\,\eps_m\otimes \eps_{N} \rangle\, , \\
\Tr_{N}^{\up}  \big(\pi'_-(x)\wh \pi_-(y) \big) & := \sum_{l=0}^{N}
\,\langle \eps_l\otimes \eps_{N-1}\otimes \eps_{\up}\,,\,
\pi'_-(x)\wh \pi_-(y)\,\eps_l\otimes \eps_{N-1}\otimes
\eps_{\up}\rangle\, ,  \\
\Tr^\dn_{N}  \big(\pi'_-(x)\wh \pi_-(y) \big) & := \sum_{l=0}^{N}
\,\langle \eps_l\otimes \eps_{N+1}\otimes \eps_{\dn} \,,\,
\pi'_-(x)\wh \pi_-(y)\,\eps_l\otimes \eps_{N+1}\otimes\eps_{\dn}
\rangle \, .
\end{align*}
Actually, a computation on monomials of $\A$ shows that $\Tr^\dn_{N}
\big(\pi'_-(x)\wh \pi_-(y) \big)=
\Tr^\up_{N}  \big(\pi'_-(x)\wh \pi_-(y) \big)$.
For convenience, we shall note $\Tr_N\big(\pi'_-(x)\wh \pi_-(y) \big)$
this functional.

\begin{lemma}
    \label{tau0ext}
Let $x,y \in \A$. Then,

(i) $\tau_{0} \big(\pi'_\pm(x)\wh \pi_\pm(y) \big):=\lim_{N
\to \infty} U_N$ exists where $$U_{N}:=\Tr_{N}
\big( \pi'_\pm(x)\wh \pi_\pm(y)\big)
-(N+1)\tau_1\big(\pi_\pm(x)\big)\,  \tau_1\big(\pi_\pm(y)\big).$$

(ii) $U_N=\tau_{0} \big(\pi'_\pm(x)\wh \pi_\pm(y) \big) +
\mathcal{O}(N^{-k})$ for all $k>0$.
\end{lemma}

\begin{proof}
$(i)$
We can suppose that $x$ and $y$ are monomials, since the result will
follow by linearity. We will give a proof for the case of the $\pi_+$
representations, the case $\pi_-$ being similar, with minor changes.

We have $\wh \pi_+ (y) = (\wh\pi_+ a_{\beta_1})^{|\beta_1|} \,
(\wh\pi_+ b)^{\beta_2} \,(\wh\pi_+ b^*)^{\beta_3} $.
A computation gives
$$
\wh \pi_+(y) \,\eps_m\otimes \eps_{2j} = (-1)^{\beta_2+\beta_3}\,
q^{(2j-m)(\beta_2+\beta_3)}\, q^{\up_{\beta_1}}_{2j-m,|\beta_1|}\,
\eps_{m-\beta_3+\beta_2}
\otimes \eps_{2j-\beta_3+\beta_2+\beta_1}
$$
and with the notation $t_{2j,m}:=\langle \eps_m\otimes \eps_{2j}
\,,\, \pi'_\pm(x)\wh \pi_\pm(y)\,\eps_m\otimes \eps_{2j} \rangle $
and $T_{2j}:=\sum_{m=0}^{2j} t_{2j,m}$, we get
\begin{align*}
t_{2j,m}&=(-1)^{\beta_2+\beta_3}\, q^{(2j-m)(\beta_2+\beta_3)}\,
q^{\up_{\beta_1}}_{2j-m,|\beta_1|}\,
q^{\up_{\a_1}}_{m-\beta_3+\beta_2,|\a_1|}\,
q^{(m+\beta_2-\beta_3)(\a_2+\a_3)}\, \delta_{\a_1+\beta_2-\beta_3,0}\\
& \hspace{2cm} \times  \delta_{-\a_3+\a_2+\beta_1,0}\\
 &=(-1)^{\a_1}\, q^{(2j-m)(\beta_2+\beta_3)+(m-\a_1)(\a_2+\a_3)}
 \, q^{\up_{\beta_1}}_{2j-m,|\beta_1|}\,
q^{\up_{\a_1}}_{m-\a_1,|\a_1|}\, \delta_{\a_1+\beta_2-\beta_3,0} \,
\delta_{\a_2-\a_3+\beta_1,0}\\
&=: f_{\a,\beta} \ q^{2j\la }\,  t'_{2j,m} =: f_{\a,\beta}
\ q^{2j\kappa }\,  t''_{2j,2j-m}
\end{align*}
where
\begin{align}
t'_{2j,m}&:=q^{m(\kappa-\la)}\, q^{\up_{\beta_1}}_{2j-m,|\beta_1|}\,
q^{\up_{\a_1}}_{m-\a_1,|\a_1|}\label{t'2jm}\,,\\
t''_{2j,m}&:=q^{m(\la-\kappa)}\, q^{\up_{\beta_1}}_{m,|\beta_1|}\,
q^{\up_{\a_1}}_{2j-m-\a_1,|\a_1|}\label{t''2jm}\,,
\end{align}
with $\la:=\beta_2+\beta_3 \geq0$ and $\kappa:=\a_2+\a_3\geq0$. We
will now prove that if $\la\neq\kappa$, then $(T_{2j})$
is a convergent sequence.
Suppose $\kappa>\la$. Let us note $U'_{2j}:= \sum_{m=0}^{2j}
t'_{2j,m}$.
Since the $t'_{2j,m}$ are positive and $t'_{2j+1,m}\geq t'_{2j,m}$
for all $j,m$,
$U'_{2j}$ is an increasing real sequence.
The estimate
$$
U'_{2j}\leq \sum_{m=0}^{2j} q^{m(\kappa-\la)} \leq
\tfrac{1}{1-q^{\kappa-\la}}<\infty
$$
proves then that $U'_{2j}$ is a convergent sequence.
With $T_{2j}=f_{\a,\beta}\, q^{2j\la}\, U'_{2j}$, we obtain our
result.

Suppose now that $\la>\kappa$.
Let us note $U''_{2j}:= \sum_{m=0}^{2j} t''_{2j,m}$.
Since the $t''_{2j,m}$ are positive and $t''_{2j+1,m}\geq t''_{2j,m}$
for all $j,m$,
$U''_{2j}$ is an increasing real sequence.
The estimate
$$
U''_{2j}\leq \sum_{m=0}^{2j} q^{m(\la-\kappa)} \leq
\tfrac{1}{1-q^{\la-\kappa}}<\infty
$$
proves then that $U''_{2j}$ is a convergent sequence.
With $T_{2j}=f_{\a,\beta}\, q^{2j\kappa}\, U''_{2j}$, we have
again our result. Moreover, note that if $\la$
and $\kappa$ are both different from zero, the limit of $(T_{2j})$
is zero and more precisely,
\begin{align}
T_{2j}&=\mathcal{O}(q^{2j\lambda}) \text{ if } \kappa>\lambda>0
\label{ka>la},\\
T_{2j}&=\mathcal{O}(q^{2j\kappa}) \text{ if }
\lambda>\kappa>0\label{la>ka}.
\end{align}
Suppose now that $\la = \kappa \neq 0$.
In that case, $(T_{2j})$ also converges rapidly to zero.
Indeed, let us fix $q< \eps <1$. we have $\eps^{-2j\la}\,T_{2j}
= \sum_{m=0}^{2j} c_{m} d_{2j-m} =c*d(2j)$ where $c_m:=f_{\a,\beta}\,
(q/\eps)^{\la m} \, q^{\up_{\a_1}}_{m-\a_1,|\a_1|}$ and
$d_m:=(q/\eps)^{\la m}
q^{\up_{\beta_1}}_{m,|\beta_1|}$. Since both $\sum_m c_m$
and $\sum_m d_m$ are absolutely convergent series, their
Cauchy product $\sum_{2j} \eps^{-2j\la}\,T_{2j}$ is
convergent. In particular, $\lim_{j\to\infty} \eps^{-2j \la}\,
T_{2j}=0$, and
\begin{equation}
T_{2j}= \mathcal{O}(\eps^{2j\lambda})\label{la=ka}.
\end{equation}
Finally, $T_{2j}$ has a finite limit in all cases
except possibly when $\la=\kappa=0$, which is the case
when $\a_1=\a_2=\a_3=\beta_1=\beta_2=\beta_3=0$.
In that case, $t_{2j,m}=1$.

A straightforward computation gives $\tau_1 \big(\pi_\pm(x)\big)
\, \tau_1\big(\pi_{\pm}(y)\big) =
\delta_{\a_1,0}\,\delta_{\beta_1,0}\,\delta_{\a_2,0}\,
\delta_{\beta_2,0}\,\delta_{\a_3,0}\,\delta_{\beta_3,0}$.

Thus,
$$
U_{2j}= T_{2j}-(2j+1) \delta_{\a_1,0}\,\delta_{\beta_1,0}
\,\delta_{\a_2,0}\,\delta_{\beta_2,0} \, \delta_{\a_3,0}
\,\delta_{\beta_3,0}
$$
has always a finite limit when $j\to \infty$.

$(ii)$ The result is clear if $\la=\kappa=0$ (in that case
$U_N=\tau_0=0$).
Suppose $\la$ or $\kappa$ is not zero.
In that case $U_{2j} = T_{2j}$. By (\ref{la>ka}), (\ref{ka>la})
and (\ref{la=ka}), we see that if $\la>\kappa >0$ or $\kappa>\la >0$
or $\kappa=\lambda$, $(T_{2j})$ converges to 0 with a rate
in $\mathcal{O}(\eps^{2j\a})$ where $\a>0$ and $q\leq \eps<1$.
Thus, it only remains to check the cases $(\kappa>0,\lambda=0)$
and $(\kappa=0, \lambda>0)$. In the first one, we get
from (\ref{t'2jm}), $U_{2j}=f_{\a,\beta} \sum_{m=0}^{2j}
q^{m\kappa} q^{\up_{\beta_1}}_{2j-m,|\beta_1|}$. If $\beta_1 = 0$, we
are done.

Suppose $\beta_1> 0$. We have
$q^{\up_{\beta_1}}_{2j-m,|\beta_1|}=\sum_{p=0}^{\infty} l_p \,
q^{r_p}\, q^{2|p|_1(2j-m)}$ where $p=(p_1,\cdots,p_{\beta_1})$ and
$l_p=(-1)^{|p|_1}\genfrac{(}{)}{0pt}{1}{{\half}}{p}$,
$r_p:=2p_1+\cdots +2\beta_1 p_{\beta_1}$.
Thus, cutting the sum in two, we get, noting $L_{2j}:=
f_{\a,\beta}\sum_{m=0}^{2j}q^{m\kappa}$,
$$
U_{2j}-L_{2j}=f_{\a,\beta}\sum_{|p|_1>\kappa/2}l_p \,
q^{r_p}\, \tfrac{q^{4|p|_1j}-q^{(2j+1)\kappa-2|p|_1}}
{1-q^{\kappa-2|p|_1}} + f_{\a,\beta}\sum_{0\neq |p|_1
\leq \kappa/2}l_p \, q^{r_p}\,q^{4|p|_1 j}\sum_{m=0}^{2j}
q^{m(\kappa-2|p|_1)}.
$$
Since $\sum_{0\neq |p|_1\leq \kappa/2}l_p \, q^{r_p}\,q^{4|p|_1 j}
\sum_{m=0}^{2j} q^{m(\kappa-2|p|_1)}$
is in $\mathcal{O}_{j\to \infty}(jq^{4j})$,
we have, modulo a rapidly decreasing sequence,
$$
U_{2j}-L_{2j}\sim f_{\a,\beta}\sum_{|p|_1>\kappa/2}l_p \, q^{r_p}\,
\tfrac{q^{4|p|_1j}-q^{(2j+1)\kappa-2|p|_1}}
{1-q^{\kappa-2|p|_1}}=:f_{\a,\beta}q^{2\kappa j} V_{2j}
$$
with
$$
V_{2j}=\sum_{|p|_1>\kappa/2}l_p \, q^{r_p}\,
\tfrac{1-q^{(2|p|_1-\kappa)(2j+1)}}
{1-q^{2|p|_1-\kappa}} = \sum_{|p|_1>\kappa/2}\
\sum_{m=0}^{2j}\,  l_p \, q^{r_p}\,  q^{(2|p|_1-\kappa)m}.
$$
The family $v_{m,p}:= (l_p \, q^{r_p}\,
q^{(2|p|_1-\kappa)m})_{(p,m)\in I}$, where $I=\set{(p,m)
\in \N^{\beta_1}\times \N  \  :
\ |p|_1>\kappa/2 }$ is (absolutely) summable.
Indeed $|v_{m,p}|\leq |l_p| q^{r_p}\, q^{m}$
so $|v_{m,p}|$ is summable as the product of two summable families.
As a consequence, $\lim_{j\to \infty} V_{2j}$
exists and is finite, which proves that $(q^{2\kappa j} V_{2j})$,
and thus $(U_{2j}-L_{2j})$ converge rapidly to 0.

Suppose now that $\beta_1<0$. In that case,
$q^{\up_{\beta_1}}_{2j-m,|\beta_1|} =
q^{\dn}_{2j-m,|\beta_1|} = q^{\up}_{2j-(m+|\beta_1|),|\beta_1|}$
and by (\ref{t'2jm}), we get $U_{2j}=f_{\a,\beta}
\sum_{m=0}^{2j}q^{m\kappa} q^{\up}_{2j-(m+|\beta_1|),|\beta_1|} =
f_{\a,\beta}\, q^{-|\beta_1|\kappa}\,
\sum_{m=|\beta_1|}^{2j+|\beta_1|}  q^{m\kappa}
q^\up_{2j-m,|\beta_1|}$,
so the same arguments as in case $\beta_1>0$ apply here,
the summation on $m$ simply shifted
of $|\beta_1|$.

The same proof can be applied for the other case
$(\kappa=0, \lambda>0)$. This time, we only
need to use (\ref{t''2jm}) instead of (\ref{t'2jm})
and the preceding arguments follow by replacing $\kappa$
by $\la$ and $\beta_1$ by $\a_1$.
\end{proof}

\begin{remark}
Contrary to the preceding $\tau_0$, the new
functional contains the shift information.
In particular, it filters the parts of nonzero degree.
\end{remark}

If $T\in X \wh X$, $\wt \rho(T) \in \pi_+(\A)\wh
\pi_+(\A)\otimes\pi_-(\A)\wh \pi_-(\A) $.

For notational convenience, we define $\tau_1$ on $\pi'_\pm(\A)\wh
\pi_\pm(\A)$ as
$$
\tau_1\big(\pi'_\pm(x)\wh \pi_\pm(y)\big):=
\tau_1\big(\pi_\pm(x)\big)\, \tau_1\big(\pi_\pm(y)\big).
$$

In the following, the symbol $\sim_e$ means equals modulo a entire
function.

\begin{theorem}
    \label{ncinttau}
Let $T\in  X \wh X$. Then
\begin{align*}
&\hspace{-2cm}(i) \hspace{.7cm} \zeta_\DD^{T}(s)\, \sim_e
2(\tau_1\otimes \tau_1) \, \big(\wt \rho(T)\big)\,
\zeta(s-2) + 2(\tau_0\otimes \tau_1 + \tau_1\otimes  \tau_0)\,
\big(\wt\rho(T)\big)\,\zeta(s-1) \\
& \hspace{2cm}+2(\tau_0\otimes  \tau_0 -
\tfrac{1}{2}\tau_1\otimes \tau_1)\, \big(\wt \rho(T)\big)\,\zeta(s)
,\\
&\hspace{-2cm}(ii) \hspace{.5cm} \ncint T |\DD|^{-3} = 2(\tau_1 \ox
\tau_1) \bigl(\wt \rho(T)\bigl),\\
&\hspace{-2cm}(iii) \hspace{.4cm} \ncint T |\DD|^{-2}=2(\tau_0\otimes
\tau_1 + \tau_1\otimes
\tau_0)\, \big(\wt \rho(T)\big) ,\\
&\hspace{-2cm}(iv) \hspace{.5cm} \ncint T|\DD|^{-1}= 2(\tau_0\otimes
 \tau_0- \tfrac{1}{2}\tau_1\otimes \tau_1)\, \big(\wt \rho(T)\big).
\end{align*}
\end{theorem}

\begin{proof}
$(i)$ Since $T\in X \wh X$, $\wt \rho(T)$ is a linear combination of
terms
like
$\pi'_+(x)\wh \pi_+(y)\otimes \pi'_-(z)\wh\pi_-(t)$,
where $x,y,z,t \in \A$. Such a term is noted in the
following $T_+\otimes T_-$. Linear combination of
these term is implicit.
With the shorthand $T_{c_1,\cdots,c_p}:=\langle\eps_{c_1}\otimes
\cdots \otimes \eps_{c_p}, T \eps_{c_1}\otimes
\cdots \otimes \eps_{c_p}\rangle$, recalling that $v^{j,\dn}_{m,l}$
is 0
when $j=0$, or $l\geq 2j$, we get
\begin{align*}
\zeta_{\DD}^T(s)=&
\sum_{2j=0}^{\infty}\sum_{m=0}^{2j}\sum_{l=0}^{2j+1} \langle \,
\genfrac{(}{)}{0pt}{1}{{v^{j,\up}_{m,l}}}{0}\,,\, S\wt \rho(T)I\,
\genfrac{(}{)}{0pt}{1}{{v^{j,\up}_{m,l}}}{0}\rangle \,d_{j^+}^{-s} +
\langle \genfrac{(}{)}{0pt}{1}{{0}}{v^{j,\dn}_{m,l}}\,,\, S\wt
\rho(T)I\,
 \genfrac{(}{)}{0pt}{1}{{0}}{v^{j,\dn}_{m,l}}\rangle \,d_{j}^{-s} \\
 =&\sum_{2j=0}^{\infty}\sum_{m=0}^{2j}\sum_{l=0}^{2j+1}
\wt\rho(T)_{m,2j,l,2j,\up}
\,d_{j^+}^{-s}+\sum_{2j=1}^{\infty}\sum_{m=0}^{2j}\sum_{l=0}^{2j-1}
 \wt\rho(T)_{m,2j,l,2j,\dn} \, d_{j}^{-s}  \\
 =&\sum_{2j=0}^{\infty}\big( \Tr_{2j}(T_+)\, \Tr_{2j+1}^{\up}(T_-)
 + \Tr_{2j+1}(T_+)\, \Tr_{2j}^\dn(T_-)\big) \, d_{j^+}^{-s} \,.
 \end{align*}
By Lemma \ref{tau0ext} $(ii)$, for all $k>0$,
\begin{align*}
\Tr_{2j}(T_\pm) &= (2j+\tfrac{3}{2})\tau_1(T_\pm) +
\tau_{0}(T_\pm)-\half \tau_1(T_\pm) +
\mathcal{O}\big((2j)^{-k}\big),\\
\Tr_{2j+1}(T_\pm) &= (2j+\tfrac{3}{2})\tau_1(T_\pm) +
\tau_{0}(T_\pm) +\half \tau_1(T_\pm) +
\mathcal{O}\big((2j)^{-k}\big)\, .
\end{align*}
The result follows by noting that the difference of the
Hurwitz zeta function $\zeta(s,\tfrac{3}{2})$ and Riemann zeta
function $\zeta(s)$ is
an entire function.

$(ii,iii,iv)$ are direct consequences of $(i)$.
\end{proof}

\subsection{The smooth algebra $C^{\infty}(SU_q(2))$}

In \cite{Cindex,SDLSV}, the smooth algebra $C^\infty(SU_q(2)$ is
defined by pulling back the smooth structure $C^\infty(D_{q^\pm}^2)$
into the $C^*$-algebra generated by $\A$, through the morphism $\rho$
and the application $\la$ (the compression which gives an operator on
$\H$ from an operator on $l^2(\N)\ox l^{2}(\N)\ox l^2(\Z)\ox \C^2$).
The important point is that with \cite[Lemma 2, p. 69]{Cindex}, this
algebra is stable by holomorphic calculus. By defining $\rho:=\wt
\rho\circ c$ and $\la(\cdot):= S (\cdot) I$, the same lemma (with
same notation) can be applied to our setting, with $c:= \pi(x)
\mapsto \ul\pi(x)$ and
\begin{align*}
&\mathcal{C}:=C^\infty(D_{q^+}^2)\ox C^\infty(S^1) \ox
C^\infty(D_{q^+}^2)\ox C^\infty(S^1) \ox \mathcal{M}_2(\C)
\end{align*}
as algebra stable by holomorphic calculus containing the image of
$\wt\rho$. Here, we use Schwartz sequences to define the smooth
structures. We finally obtain $C^\infty(SU_q(2))$ with real structure
as a subalgebra stable by holomorphic calculus of the $C^*$-algebra
generated by $\pi(\A) \cup J\pi(\A) J^{-1}$ and containing $\pi(\A)
\cup J \pi(\A) J^{-1}$ .

\begin{corollary}
The dimension spectrum of the real spectral triple
$\big(C^{\infty}(SU_q(2)),\H,D\big)$ is simple and given by
$\{1,2,3\}$. Its KO-dimension is $3$.
\end{corollary}

\begin{proof}
Since $F$ commutes with $\ul\pi(\A)$, the pseudodifferential
operators of order 0 (without the real structure) 
are exactly (modulo $OP^{-\infty}$) the operators in $\B
+ \B F$. From Theorem \ref{Theo} we see that the dimension spectrum
of $SU_q(2)$ without taking into account the reality operator $J$ is
$\set{1,2,3}$, in other words, the possible poles of $\zeta^b_{\DD}:
s\mapsto \Tr(b F^{\eps}|\DD|^{-s})$ (with $\eps\in \set{0,1}$, $b\in
\B$) are in $\set{1,2,3}$. 
Theorem \ref{ncinttau} $(i)$ shows that the possible poles are still
$\set{1,2,3}$ when we take into account the real structure of
$SU_q(2)$, that is to say, when $\B$ is enlarged to $\B J\B J^{-1}$.
Indeed, any element of  $\B J\B J^{-1}$ is in $X\wh X$ and it is
clear from the preceding proof that adding $F$ in the previous zeta
function do not add any pole to $\set{1,2,3}$.

All arguments go true from the polynomial algebra $\A(SU_q(2))$ to
the smooth pre-C$^*$-algebra $C^{\infty}(SU_q(2))$.

KO-dimension refers just to $J^2=-1$ and $\DD\,J=J\,\DD$ since there
is no chirality because spectral dimension is 3.
\end{proof}

\subsection{Noncommutative integrals with reality operator and
one-forms on $SU_q(2)$}

The goal of this section is to obtain the following suppression of
$J$:

\begin{theorem}
\label{ncintJ}
Let $A$ and $B$ be $\delta$-one-forms. Then
\begin{align*}
&  \hspace{-3cm} (i)\hspace{0.5cm} \ncint AJBJ^{-1} |\DD|^{-3}=
\tfrac 12  \ncint  A  \vert \DD \vert ^{-3} \overline{ \ncint B
\vert \DD \vert ^{-3}},\\
&   \hspace{-3cm}  (ii)\hspace{0.4cm}\ncint AJBJ^{-1} |\DD|^{-2}=
\half  \ncint  A  \vert \DD \vert
^{-2}\,\overline{\ncint  B  \vert \DD \vert ^{-3}} + \half\ncint   A
\vert \DD \vert ^{-3}\, \overline{\ncint B \vert \DD \vert
^{-2}},\\
&   \hspace{-3cm}  (iii)\hspace{0.3cm} \ncint A^2 JBJ^{-1}
|\DD|^{-3}=\tfrac 12 \ncint  A^2  \vert \DD
\vert ^{-3} \,\overline{\ncint B  \vert \DD \vert ^{-3}},\\
&    \hspace{-3cm} (iv)\hspace{0.4cm} \ncint \delta(A)A |\DD|^{-3}=
\ncint \delta(A)JAJ^{-1}|\DD|^{-3} = 0.
\end{align*}
\end{theorem}

We gather at the beginning of this section the main notations for
technical lemmas which will follow.

For any pair $(k,p)\in \N^3\times \N^3$ such that $k_i\leq |\a_i|$,
$p_i\leq |\beta_i|$, where $\a,\beta \in \Z\times\N \times \N$, we
define
\begin{align*}
v_{k,p} &:=  g(p)\
\genfrac(){0pt}{1}{|{\a_1}|}{k_1}_{q^{2\eta_{\a_1}}}
\genfrac(){0pt}{1}{{\a_2}}{k_2}
\genfrac(){0pt}{1}{{\a_3}}{k_3}\genfrac(){0pt}{1}{|{\beta_1}|}{p_1}_{q^{2\eta_{\beta_1}}}
\genfrac(){0pt}{1}{{\beta_2}}{p_2}
\genfrac(){0pt}{1}{{\beta_3}}{p_3} \,
(-1)^{k_1+p_1+{\a_2}+{\a_3}+{\beta_2}+{\beta_3}}\, q^{\sigma_{k,p}},\\
h_{k,p} &:= {\a_1}+{\a_2}-{\a_3} -2 (\eta_{\a_1} k_1 +k_2 -k_3)
+g(p)\, , \\
g(p) &:=  {\beta_1}+{\beta_2} -{\beta_3}-2(\eta_{\beta_1} p_1 +p_2
-p_3)\, ,\\
\sigma_{k,p}&:= k_1+p_1 +\sigma^t_{k,p} +\sigma^u_{k,p}\, , \\
\sigma^t_{k,p}&:= k_1\wh k_2-\wh k_3(k_1+k_2)+\eta_{\beta_1} \wh
p_1|k|_1 +\wh p_2(|k|_1+p_1) -\wh p_3(|k|_1+p_1+p_2)\, , \\
\sigma^u_{k,p}&:=(k_3+\eta_{\beta_1} \wh p_1 -p_2 +p_3)(k_1+\wh k_2
+\wh k_3) -k_2(k_1+\wh k_2) + (p_1+\wh p_2)(-p_2+p_3) +\wh p_3p_3 \,
,\\
t_{k,p}&= a_{\a_1}^{\wh k_1} \,a^{\wh k_2}\, {a^*}^{\wh k_3}\,
a_{\beta_1}^{\wh p_1}\, a^{\wh p_2} \,{a^*}^{\wh p_3}
\,b^{|k|_1+|p|_1}    \, ,\\
u_{k,p}&= a_{\a_1}^{\wh k_1} \,{a^*}^{k_2} \,a^{k_3}
\,a_{\beta_1}^{\wh p_1} \,{a^*}^{p_2} \,a^{p_3}\,b^{|\wt k|_1+|\wt
p|_1}   \, .
\end{align*}
where we used the notation
$$
\wh k_i:= |{\a_i}|-k_i, \,  \wh
p_i:=|{\beta_i}|-p_i,
$$
so $0\leq \wh k_i \leq |\a_i|,\, 0\leq  \wh
p_i \leq |\beta_i|$. We will also use the shorthand $\wt k:=(k_1,\wh
k_2,\wh k_3)$.

For ${\beta_1}\in \Z$ and $j\in \N$, we define
\begin{align*}
w_1({\beta_1},j) &:= \sum_{n=0}^\infty \big( q^{2jn}
(q_{n,|{\beta_1}|}^{\up_{\beta_1}})^{2}-\delta_{j,0}\big),\\
w_{\beta}^\a&:= 2{\beta_1}\,
q^{{\beta_1}(2{\a_3}+{\beta_3}-{\beta_2})}
\,w_1({\beta_1},{\a_3}+{\beta_3}).
\end{align*}

We introduce the following notations:
\begin{align*}
q^+_{k,p,n}&:=q^{n(|k|_1+|p|_1)} q_{n+r^+_{k,p} -\eta_{\a_1} \wh
k_1,\wh k_1}^{\up_{\a_1}} q_{n-\wh k_3+
\eta_{\beta_1} \wh p_1+\wh p_2-\wh p_3,\wh k_2}^\up
q_{n+\eta_{\beta_1} \wh p_1+\wh p_2-\wh p_3
,\wh k_3}^\dn q_{n+\wh p_2-\wh p_3,\wh p_1}^{\up_{\beta_1}} q_{n-\wh
p_3,\wh p_2}^{\up} q_{n,\wh p_3}^\dn\, ,\\
q^-_{k,p,n}&:=q^{n(|\wt k|_1+|\wt p|_1)} q_{n+r^-_{k,p}-\eta_{\a_1}
\wh k_1,\wh k_1}^{\up_{\a_1}}
q_{n+k_3+\eta_{\beta_1} \wh p_1-p_2 +p_3 ,k_2}^\dn
q_{n+\eta_{\beta_1} \wh p_1-p_2+p_3,k_3}^\up q_{n-p_2+p_3,\wh
p_1}^{\up_{\beta_1}} q_{n+p_3,p_2}^\dn q_{n,p_3}^\up\\
&\hspace{1cm}\times (-1)^{|\wt k|_1+|\wt p|_1},\\
r^+_{k,p}&:=\eta_{\a_1} \wh k_1+\wh k_2-\wh k_3+\eta_{\beta_1} \wh
p_1  +\wh p_2  -\wh p_3 \, , \\
r^-_{k,p}&:=\eta_{\a_1} \wh k_1 -k_2+k_3+\eta_{\beta_1} \wh
p_1-p_2+p_3\, .
\end{align*}
Thus, $\pi_+(t_{k,p})\eps_n = q^+_{k,p,n} \eps_{n+ r^+_{k,p}}$ and
$\pi_-(u_{k,p})\eps_n = q^-_{k,p,n} \eps_{n+ r^-_{k,p}}$.

\begin{lemma}
    \label{technicalxdy}
We have
$$
r\left( (M_{\beta}^\a)^\circ \right) =
\sum_{k,p}\delta_{h_{k,p},0}\, v_{k,p}\, \pi_+(t_{k,p}) \otimes
\pi_-(u_{k,p})
$$
where the summation is done on $k_i, p_i$ in $\N$ such that $k_i\leq
{|\a_i|}, p_i \leq {|\beta_i|}$ for $i\in \set{1,2,3}$.
\end{lemma}

\begin{proof}
Since $\ul\pi(m^\a)=
(a_++a_-)^{\a_1}(b_++b_-)^{\a_2}(b_+^*+b_-^*)^{\a_3}$, with
$v_k:=\genfrac(){0pt}{1}{|{\a_1}|}{k_1}_{q^{2\eta_{\a_1}}}
\genfrac(){0pt}{1}{{\a_2}}{k_2}
\genfrac(){0pt}{1}{{\a_3}}{k_3}$,
$$
\ul \pi(m^\a) = {\sum}_k v_k\, c_k
\text{ where }
c_k:=a_{+,{\a_1}}^{|{\a_1}|-k_1}\, a_{-,{\a_1}}^{k_1}
\,b_+^{{\a_2}-k_2} \,b_-^{k_2}\, {b_+^*}^{{\a_3}-k_3}
\,{b_-^*}^{k_3}.
$$
By Lemma \ref{commutateur} $(iii)$ we see that
$\delta(\ul \pi(m^\beta)) = \sum_p w_p\, d_p$ where we introduce

$w_p:=\genfrac(){0pt}{1}{|{\beta_1}|}{p_1}_{q^{2\eta_{\beta_1}}}\genfrac(){0pt}{1}{{\beta_2}}{p_2}
\genfrac(){0pt}{1}{{\beta_3}}{p_3}$ and
$d_p:=g(p)\,a_{+,{\beta_1}}^{|{\beta_1}|-p_1}\, a_{-,{\beta_1}}^{p_1}
\,b_+^{{\beta_2}-p_2} \,b_-^{p_2}\, {b_+^*}^{{\beta_3}-p_3}
\,{b_-^*}^{p_3} $.

As a consequence, $(M^{\a}_\beta)^\circ = \sum_{k,p}
\delta_{h(k,p),0}\  g(p)\  v_k \,w_p \,c_{k,p}$ where
\begin{equation}
    \label{ckp}
c_{k,p}=a_{+,{\a_1}}^{\wh k_1} \,a_{-,{\a_1}}^{k_1} \,b_+^{\wh k_2}\,
b_-^{k_2} \,{b_+^*}^{\wh k_3}\, {b_-^*}^{k_3}\,
a_{+,{\beta_1}}^{\wh p_1} \,a_{-,{\beta_1}}^{p_1}\,
b_+^{\wh p_2} \,b_-^{p_2}\, {b_+^*}^{\wh p_3} \,{b_-^*}^{p_3}
\end{equation}
With (\ref{ckp}), we get
$r(c_{k,p})= (-1)^{k_1+p_1+{\a_2}+{\a_3}+{\beta_2}+{\beta_3}}\,
q^{k_1+p_1}\, \pi_+(t'_{k,p}) \otimes \pi_-(u'_{k,p})$ where
\begin{align*}
t'_{k,p}&=  a_{\a_1}^{\wh k_1}\,b^{k_1}\,a^{\wh k_2}\, b^{k_2}\,
{a^*}^{\wh k_3}\,b^{k_3}
a_{\beta_1}^{\wh p_1}\,b^{p_1}\,a^{\wh p_2} \,b^{p_2} \,{a^*}^{\wh
p_3}\,b^{p_3}\, ,\\
u'_{k,p}&=  a_{\a_1}^{\wh k_1} \,b^{k_1}\,b^{\wh k_2}\,{a^*}^{k_2}
\,{b}^{\wh k_3}a^{k_3} \,a_{\beta_1}^{\wh p_1}\,b^{p_1}\,b^{\wh
p_2}\, {a^*}^{p_2} \,{b}^{\wh p_3}\,a^{p_3}\,  .
\end{align*}
A recursive use of relation $b a_j =q^{\eta_j} a_j b$ yields the
result.
\end{proof}

\begin{lemma}
    \label{ncintLin}
We have

(i) $(\tau_1 \otimes \tau_1) \big(r(M^{\a}_\beta)^\circ\big) =
{\beta_1}\
\delta_{{\a_1},-{\beta_1}}\,\delta_{{\a_2},0}\,\delta_{{\a_3},0}\,
\delta_{{\beta_2},0}\, \delta_{{\beta_3},0}\, .$

(ii) $(\tau_1 \otimes \tau_0 + \tau_0 \otimes \tau_1)
\big(r(M^{\a}_\beta)^\circ\big)
=\delta_{{\a_1},-{\beta_1}} \,
\delta_{{\a_2}+{\beta_2},{\a_3}+{\beta_3}}
\, w_\beta^\a.$

In particular, if $A$ is a $\delta$-one-form, we have
\begin{align*}
\ncint A |\DD|^{-3} &= 2 {\beta_1 } \, A_{-\beta_1 0 0 }^{\beta_1 0 0
} \,, \\ 
\ncint A |\DD|^{-2} &= 2 w^{\a}_\beta\, B(A)^{\beta}_\a  \,.
\end{align*}
where we implicitly summed on all $\a,\beta$ indices.
\end{lemma}

\begin{proof} See Appendix A.
\end{proof}

With notations of Lemma \ref{technicalxdy}, it is direct to check
that for given $\bar\a=(\a,\a',\cdots ,\a^{(n-1)})$
and $\bar\beta=(\beta,\beta',\cdots ,\beta^{(n-1)})$,
\begin{equation}
    \label{rproduct}
r\big((M^{\bar\a}_{\bar\beta} )^\circ \big)= \sum_{K,P}
\delta_{h_{K,P},0}\,v_{K,P}\, \pi_+(t_{K,P}) \otimes
\pi_-(u_{K,P})
\end{equation}
where $K=(k,k',\cdots k^{(n-1)})$, $P=(p,p',\cdots,p^{(n-1)})$ with
$0\leq
k^{(j)}_i \leq |\a^{(j)}_i|$, $0\leq p^{(j)}_i\leq |\beta^{(j)}_i|$,
\begin{align*}
t_{K,P}&:=t_{k,p} \,t_{k',p'}\cdots t_{k^{(n-1)},p^{(n-1)}} \,\,,&&
u_{K,P}:=u_{k,p}u_{k',p'}\cdots u_{k^{(n-1)},p^{(n-1)}} \,\,,\\
v_{K,P}&:=v_{k,p} \,v_{k',p'}\cdots v_{k^{(n-1)},p^{(n-1)}} \,\,,&&
h_{K,P}:= h_{k,p} + h_{k',p'}+\cdots h_{k^{(n-1)},p^{(n-1)}}\,\,.
\end{align*}

In the following, we will use the shorthands
$A_i:=\a_i+\a'_i+\cdots +\a^{(n-1)}_i$,
$B_i:=\beta_i+\beta'_i+\cdots+\beta^{(n-1)}_i$.

In the case $n=2$, we also note $r^{\pm}_{K,P}:=  r^{\pm}_{k,p}+
r^{\pm}_{k',p'}$ and
$q^{\pm}_{K,P,n}:=q^{\pm}_{k',p',n}
\,q^{\pm}_{k,p,n+r^{\pm}_{k',p'}}$.

Thus,  we have $\pi_+(t_{K,P})\, \eps_m = q^+_{K,P,m} \, \eps_{m+
r^+_{K,P}}$ and $\pi_-(u_{K,P}) \, \eps_m = q^-_{K,P,n} \, \eps_{m+
r^-_{K,P}}$.

We also introduce, still for $n=2$,
\begin{align*}
v_{\beta_1,\a'_1,\beta'_1}(l,j)&:=\sum_{n=0}^\infty\big(q^{l+2nj}\,
q^{\up_{-\beta'_1-\a'_1-\beta_1}}_{n+\beta'_1+\a'_1+\beta_1,|\beta'_1+\a'_1+\beta_1|}
q^{\up_{\beta_1}}_{n+\beta'_1+\a'_1,|\beta_1|}
q^{\up_{\a'_1}}_{n+\beta'_1,|\a'_1|}
q^{\up_{\beta'_1}}_{n,|\beta'_1|}-\delta_{j,0}\big),
\\
V^{\bar \a}_{\bar \beta}
&:=2[\beta_1\beta'_1+(\beta_2-\beta_3)(\beta'_2-\beta'_3)]\,
q^{2{\beta_1}({\a_2}+{\a_3})+2{\beta'_1}({\a'_2}+{\a'_3})}\,\\
&\hspace{2cm}\times
v_{\beta_1,\a'_1,\beta'_1}((\a_2+\beta_2+\a_3+\beta_3)(\a'_1+\beta'_1),A_3+B_3).
\end{align*}

\begin{lemma} 
\label{ncintA2} 
We have

$(i) \quad (\tau_1 \otimes \tau_1)\, \big(r(M^{\a}_\beta
M^{\a'}_{\beta'})^\circ\big) =
{\beta_1}\beta'_1\
\delta_{A_1,-B_1}\,\delta_{A_2,0}\,\delta_{A_3,0}\,
\delta_{B_2,0} \delta_{B_3,0\, .}$

$(ii) \quad (\tau_1 \otimes \tau_0+\tau_0 \otimes \tau_1)\,  \big(
r(M^{\a}_\beta M^{\a'}_{\beta'})^\circ\big)
=\delta_{A_2+B_2,A_3+B_3} \delta_{A_1,-B_1} V^{\bar \a}_{\bar
\beta}$\,.

$(iii) \quad (\tau_1 \otimes \tau_1)\, \big(r(M^{\a}_\beta
M^{\a'}_{\beta'}M^{\a''}_{\beta''})^0 \big) =
{\beta_1}\beta'_1\beta''_1\
\delta_{A_1,-B_1}\,\delta_{A_2,0}\,\delta_{A_3,0}\,
\delta_{B_2,0} \delta_{B_3,0}.$

$(iv) \quad (\tau_{1}\otimes \tau_1)\  \big(r(\delta(M^{\a}_\beta)
M^{\a'}_{\beta'})^0 \big)
=-(\a'_1+\beta'_1){\beta_1}\beta'_1\
\delta_{A_1,-B_1}\,\delta_{A_2,0}\,\delta_{A_3,0}\,
\delta_{B_2,0} \delta_{B_3,0} .$

$(v)$ In particular, if $A$ is a $\delta$-one-form,
\begin{align*}
&\ncint A^2|\DD|^{-3}= 2\beta_1\beta'_1\ B_a(A^2)^{\bar \beta}_{\bar
\a} \,,\\
&\ncint A^2|\DD|^{-2} = 2 V_{\bar \beta}^{\bar\a}\  B(A^2)^{\bar
\beta}_{\bar \a}\, ,\\
&\ncint A^3|\DD|^{-3}= 2\beta_1\beta'_1\beta''_1\ B_a(A^3)^{\bar
\beta}_{\bar \a}\,,\\
&\ncint \delta(A) A |\DD|^{-3} =\ncint A\delta(A) |\DD|^{-3}=0.
\end{align*}
\end{lemma}

\begin{proof} See Appendix B.
\end{proof}

For a given $\delta$-1-form $A$, we say that $A$ is homogeneous of
degree in $a$ equal to $n\in \Z$ 
if it is a linear combination of $M_\beta^\alpha$ such that
$\a_1+\beta_1=n$. From Lemma \ref{ncintA2} $(iv)$ we get,

\begin{corollary}
    \label{commutation1}
Let $A$, $A'$ be two $\delta$-1-forms, then
\begin{align*}
 &\ncint (A |\DD|^{-1})^2 = \ncint A^2 |\DD|^{-2}\, ,\\
& \ncint  A |\DD|^{-1} A' |\DD|^{-1} = \ncint A A' |\DD|^{-2}
- n \ncint A A' |\DD|^{-3} \, , \text{ when $A'$
homogenous of degree $n$.}\\
\end{align*}
\end{corollary}

\begin{lemma}
\label{ncintJlem}
We have

(i)  $(\tau_1\otimes \tau_1)\,  \wt \rho \big(  M^{\a}_\beta J
M^{\a'}_{\beta'}
J^{-1}\big)={\beta_1}\beta'_1\
\delta_{\a_1,-\beta_1}\,\delta_{\a'_1,-\beta'_1}\,\delta_{A_2,0}\,\delta_{A_3,0}\,
\delta_{B_2,0} \delta_{B_3,0}\,.$

(ii)  $(\tau_0 \otimes \tau_1+\tau_1 \otimes \tau_0)\,  \wt \rho
\big(  M^{\a}_{\beta} J  M^{\a'}_{\beta'}  J^{-1}\big)
=\delta_{\a_1,-\beta_1}\,\delta_{\a'_1,-\beta'_1}
\big(\beta'_1 w^{\a}_{\beta}
\delta_{\a'_2+\beta'_2+\a'_3+\beta'_3,0}\,\delta_{\a_2+\beta_2,\a_3+\beta_3}$

\hspace{8cm} $+\beta_1
w^{\a'}_{\beta'}\delta_{\a_2+\beta_2+\a_3+\beta_3,0}
\,\delta_{\a'_2+\beta'_2,\a'_3+\beta'_3}\big)$.

(iii) $(\tau_1\otimes \tau_1)\,  \wt \rho \big(  M^{\a}_{\beta}
M^{\a'}_{\beta'} J M^{\a''}_{\beta''}
J^{-1}\big)={\beta_1}\beta'_1\beta''_1\
\delta_{\a_1+\a'_1,-\beta_1-\beta'_1}\,\delta_{\a''_1,-\beta''_1}\,
\delta_{A_2,0}\,\delta_{A_3,0}\,
\delta_{B_2,0} \delta_{B_3,0}\,.$

(iv) $(\tau_{1}\otimes \tau_1)\   \wt \rho \big(
\delta(M^{\a}_{\beta})
 J M^{\a'}_{\beta'} J^{-1}\big)
=-(\a'_1+\beta'_1){\beta_1}\beta'_1\
\delta_{\a_1,-\beta_1}\,\delta_{\a'_1,-\beta'_1}\,\delta_{A_2,0}\,
\delta_{A_3,0}\,
\delta_{B_2,0} \delta_{B_3,0}$ .

$(v)$ In particular, if $A$ and $A'$ are $\delta$-one forms, 
\begin{align*}
&\ncint AJA'J^{-1} |\DD|^{-3}= 2  (\beta_1
A_{-\beta_100}^{\beta_100})(\beta_1 \bar
{A'}_{-\beta_100}^{\beta_100}) ,\\
&\ncint AJA'J^{-1} |\DD|^{-2}= 2(\beta_1
\bar {A'}^{\beta_100}_{-\beta_100})(w^\a_\beta B(A)^{\beta}_\a)
+2(\beta_1 A^{\beta_100}_{-\beta_100})(w^\a_\beta B(\bar
{A'})^{\beta}_\a),\\
&\ncint A^2 JA'J^{-1} |\DD|^{-3}=2(\beta_1
\bar {A'}_{-\beta_100}^{\beta_100}) (\beta_1\beta'_1
B_a(A^2)^{\bar\beta}_{\bar\a}), \\
&\ncint \delta(A)JAJ^{-1} = 0.
\end{align*}

\end{lemma}

\begin{proof}
See Appendix C.
\end{proof}

\begin{lemma}
\label{tech-separation}
Let $\beta,\beta' \in \Z$. Then,
$$\lim_{2j\to\infty} \sum_{m=0}^{2j} \big((q^{\up_{\beta}}_{m,|\beta|}
q^{\up_{\beta'}}_{2j-m,|\beta'|})^2 -1\big) =
\sum_{m=0}^{\infty}\big((q^{\up_{\beta}}_{m,|\beta|})^2 -1\big) +
\sum_{m=0}^{\infty}\big((q^{\up_{\beta'}}_{m,|\beta'|})^2 -1\big).$$
\end{lemma}

\begin{proof} See Appendix D.
\end{proof}

\begin{proof}[Proof of Theorem \ref{ncintJ}]
The result follows from Lemmas \ref{ncintLin}, \ref{ncintA2} $(v)$
and \ref{ncintJlem} $(v)$.
\end{proof}

\subsection{Proof of Theorem \ref{mainThmJ} and corollaries}

\begin{lemma}
\label{ncintD}
We have on $SU_q(2)$,

(i) \hspace{0.3cm} $\ncint  |\DD|^{-3} = 2 $.

(ii) \hspace{0.2cm} $\ncint  |\DD|^{-2} = 0 $.

(iii) \hspace{0.1cm}  $\ncint  |\DD|^{-1} = -\half $.

(iv) \hspace{0.1cm} $\zeta_\DD(0)=0$.
\end{lemma}

\begin{proof} 
$(iv)$ We have by definition
$$
\zeta_\DD(s) := \Tr(\vert \DD
\vert^{-s})=\sum_{2j=0}^{\infty}\sum_{m=0}^{2j}\sum_{l=0}^{2j+1}\,
\langle v^{j}_{m,l}\,,\, |\DD|^{-s}\, v^{j}_{m,l}\rangle.
$$
Since $|\DD|^{-s} v^j_{m,l} =\genfrac{(}{)}{0pt}{1}{d_{j^+}^{-s}
\quad
0\,\,}{\, 0  \quad d_j^{-s}} \, v^j_{m,l}$ where
$d_j:=2j+\half$, we get
$$
\zeta_\DD(s) = \sum_{2j=0}^{\infty} (2j+1)(2j+2)\,d_{j^+}^{-s}+
\sum_{2j=1}^{\infty} (2j+1)(2j) \,d_j^{-s}=2\sum_{2j=0}^{\infty}
(2j+1)(2j) \,d_j^{-s}.
$$
With the equalities $(2j+1)(2j) = d_{j}^2 - \tfrac{1}{4}$ and
$\zeta(s,\half)=(2^s-1)\zeta(s)$
(here $\zeta(s,x):=\sum_{n\in\N}\tfrac{1}{(n+x)^s}$ is the Hurwitz
zeta function and $\zeta(s):=\zeta(s,1)$ is the Riemann zeta
function) we get
\begin{equation}\label{zetaDs}
\zeta_\DD(s) = 2(2^{s-2}-1)\zeta(s-2)-\half(2^s-1)\zeta(s)
\end{equation}
which entails that $\zeta_\DD(0)=0$.

$(i,ii,iii)$ are direct consequences of equation (\ref{zetaDs}). 
\end{proof}
\begin{proof}[Proof of Theorem \ref{mainThmJ}] It is a consequence of
Lemma \ref{ncintD} and Theorems \ref{coeffsASJ}, \ref{ncintJ}.
\end{proof}
As we have seen, the computation of the noncommutative integral on
$SU_q(2)$ leads to certain functions of $A$ which filter some symmetry
on the degree in $a$, $a^*$, $b$, $b^*$ of the canonical
decomposition.
Precisely, it is the balanced features that appear and the following
functions of $A^n$, $n\in \set{1,2,3}$:
\begin{align}
\label{I_pA^n}
\ncint A^n |\DD|^{-p}
\end{align}
where $1\leq n\leq p \leq 3$.
In the next section, we describe a method to compute these integrals.

\begin{corollary}
Let $u$ be a unitary in $C^\infty\big(SU_q(2)\big)$ and
$\ga_{u}(\Abb):=\pi(u)\Abb \pi(u^*)+\pi(u)d\pi(u^{*})$ be a
gauge-variant of $\Abb$. Then the
following terms of Theorem \ref{mainThmJ} are gauge invariant
\begin{align*}
&\ncint A |\DD|^{-3}\, ,\qquad
\ncint A^2 |\DD|^{-3} -\ncint A |\DD|^{-2}\,,\qquad 
&-2\ncint A |\DD|^{-1}  + \ncint A^2 |\DD|^{-2} -\tfrac{2}{3} \ncint
A^3 |\DD|^{-3}.
\end{align*}
\end{corollary}

\begin{proof}
It is sufficient to remark that all terms
$\ncint \vert D_{\Abb}\vert ^{-k}$ and $\zeta_{\DD_{\Abb}}(0)$
in the spectral action
\eqref{formuleaction} are gauge invariant. This can also be seen
via the computation $D_{\ga_{u}(\Abb)}=V_{u}\DD V_{u}^*
+V_{u}P_{0}V_{u}^*$ where $P_{0}$ is the projection on Ker$\,\DD$ and
$V_{u}=\pi(u)J\pi(u)J^{-1}$ and $\ncint \vert D_{\Abb}\vert ^{n-k}
=\Res_{s=n-k} \, \Tr \, \big(\vert D_{\Abb}\vert ^{n-k}\big)$
(see Proposition \ref{noyaux} $(iii)$ and Proposition \ref{ncintfluctuated}).
\end{proof}

\begin{corollary}
    \label{sansJ}
In the case of the spectral action without the reality operator (i.e.
$D_{\Abb}=\DD+\Abb$), we get
\begin{align*}
&  \ncint \vert D_{\Abb}\vert^{-3} =2,\qquad \ncint
\vert D_{\Abb}\vert^{-2}  =  - 2\, \ncint A |\DD|^{-3}\, ,\qquad
\ncint \vert D_{\Abb}\vert^{-1}  = -\half - \ncint A |\DD|^{-2} +
\ncint A^2 |\DD|^{-3},\\
&\zeta_{D_{\Abb}}(0) = - \ncint A |\DD|^{-1} + \half \ncint A^2
|\DD|^{-2}   -\tfrac{1}{3} \ncint A^3 |\DD|^{-3} \, .
\end{align*}
As a consequence, if $\Abb$ is a one-form such that $\ncint A
|\DD|^{-3} =0$,
then the scale invariant term of the
spectral action with or without $J$ is exactly the same
modulo a global factor of 2.
\end{corollary}

\section{Differential calculus on $SU_q(2)$ and applications}

\subsection{The sign of $\DD$ }

There are multiple differential calculi on $SU_{q}(2)$, see
\cite{Woronowicz, Klimyk}. Due to \cite[Theorem 3]{Schmuedgen},
the $3D$ and $4D_{\pm}$ differential calculi do not coincide with the
one considered here: the right multiplication of one-forms by an
element in
the algebra $A$ is a consequence of the chosen Dirac operator which
was introduced according to some equivariance properties with respect
to the duality between the two Hopf algebras $SU_{q}(2)$ and
$\U_{q}(su(2))$.

It is known that the Fredholm module associated to $(\A,\, \H,\,
\DD)$ is one-summable since $[F,\pi(x)]$ is trace-class for all $x \in
\A$. In fact, more can be said about $F$
\footnote{Note that a similar result for a different spectral triple
over $SU_q(2)$ when $q=0$ was obtained in \cite[eq. (48)]{Cindex}}:

\begin{prop}
Since
\begin{align}
\label{eq:F2}
\tfrac{1}{1-q^{2}}\, \left(\ul\pi(a^*) \,d\ul\pi(a)+ q^2
\,\ul\pi(b) \,d\ul\pi(b^*)+q^2 \, \ul\pi(a) \,d\ul\pi(a^*)+ q^2
\,\ul\pi(b^*) \,d\ul\pi(b) \right) = F,
\end{align}
$F$ is a central one-form modulo $OP^{-\infty}$.
\end{prop}

\begin{proof}
Forgetting $\piappr$, this follows from
\begin{align}
a^* \,\delta a+ q^2 \,b \,\delta b^* +q^2 \, &a \,\delta a^* + q^2
\,b^* \,\delta b \nonumber\\
&=(a_{+}^*+a_{-}^*)(a_{+}-a_{-})+q^2\,(b_{+}+b_{-})(b_{-}^*-b_{+}^*)+q^2\,
(a_{+}+a_{-})(a_{-}^*-a_{+}^*)\nonumber\\
&\hspace{3.9cm} +q^2\,(b_{+}^*+b_{-}^*)(b_{+}-b_{-})\nonumber\\
&=[a_{+}^{*}a_{+}-q^{2}\,a_{+}a_{+}^{*}+q^2\,b_{+}^{*}b_{+}
-q^2\,b_{+}b_{+}^{*}]+R=(1-q^2) +R
\label{CalculF}
\end{align}
by \eqref{astuce1} where we check that the remainder $R$ is zero:
\begin{align*}
R=&-[a_{+}^*a_{-}+q^2\,b_{+}^*b_{-}]+[a_{-}^*a_{+}+q^2\,b_{-}^*b_{+}]
-[a_{-}^{*}a_{-}-q^{2}\,a_{-}a_{-}^{*}+q^2\,b_{-}^{*}b_{-}
-q^2\, b_{-}b_{-}^{*}] \\
&\hspace{1.5cm}+(q^2\,a_{+}a_{-}^*+q^2\,q_{-}^*b_{+})
-(a_{+}^*a_{-}+q^2\, b_{+}^*b_{-}),
\end{align*}
thus, applying \eqref{astuce2}, \eqref{astuce3},
\eqref{astuce4}, $R=+(q^2\,a_{+}a_{-}^*+q^2\,q_{-}^*b_{+})
-(a_{+}^*a_{-}+q^2\, b_{+}^*b_{-})=0$ using commutation
relations \eqref{eq:commutrulefora,b}.

Now, replacing $\delta$ by $d$ in \eqref{CalculF} gives \eqref{eq:F2}
since $F$ commute with $a_{\pm},\, b_{\pm}$ and $F$ is central
by \eqref{Fcommutes}.
\end{proof}

\begin{prop}
The one-form in \eqref{eq:F2} is in fact exactly a function of the
Dirac operator $\DD$:
\begin{align}
    \label{eq:F3}
\pi(a^*) \,d\pi(a)+ q^2 \,\pi(b) \,d\pi(b^*)+q^2
\, \pi(a) \,d\pi(a^*) + q^2 \,\pi(b^*) \,d\pi(b)
= \xi_q(\DD)=F\, \xi_q (\vert \DD \vert),
\end{align}
where $ \xi_q(s) := q \tfrac{[2s] - 2s}{[s+1/2] [s-1/2]}$.

Moreover,  $F=\lim_{q \to 0} \xi_q(\DD)$.
\end{prop}

\begin{proof}
First, let us observe that the one-form $\omega$ in (\ref{eq:F3}) is
invariant under the action of the $\Uq \times \Uq$: $h \triangleright
\omega=\epsilon(h)\, \omega$ for any $h\in \Uq \times \Uq$. For
instance, using notations of \cite{DLSSV}
\begin{align*}
e \triangleright \omega= q^{\oh} a^* db + q^2 \left( - q^{\oh-1} b
da^*
+ q^{-\oh} b da^* - q^{-1-\oh} a^* db \right) = 0=\epsilon(e) \,
\omega.
\end{align*}
Therefore, since both the representation $\pi$ as well as the
operator $\DD$ are equivariant, the image of $\omega$ must
be diagonal in the spinorial base.  
A tedious computation with the full spinorial representation $\pi$
given in \eqref{eq:reprexact} yields
\begin{align*}
\langle v^{j \up}_{ml} ,\,\omega\, v^{j \up}_{ml} \rangle &=
\tfrac{q^{8j+8}-q^{8j+6}-(4j+3)\,q^{4j+6}+(8j+6)\,q^{4j+4}-(4j+3)\,q^{4j+2}-q^2+1}{(q^{4j+4}-1)(q^{4j+2}-1)}=\xi_q(2j+\tfrac
32),\\ 
\,\langle v^{j \dn}_{ml} ,\,\omega\, v^{j \dn}_{ml} \rangle
&=\tfrac{-q^{8j+4}+q^{8j+2}+(4j+1)\,q^{4j+4}-(8j+2)\,q^{4j+2}+(4j+1)\,q^{4j}+q^2-1}{(q^{4j+2}-1)(q^{4j}-1)}=-\xi_q(2j+\tfrac
12).
\end{align*}
These expressions have a clear $q=0$ limit equal respectively to 1
and -1, so $\omega \to F$ as $q \to 0$.
\end{proof}

In the $q=1$ limit, these expressions yields identically
$0$, which is confirmed by the fact that all one-forms are central,
it could be expressed as $d( a a^* + b b^*) = d\,1$.

Note that since the invariant one-form we constructed differs
by $OP^{-\infty}$ from $F$, hence any commutator with it will
be itself in $OP^{-\infty}$. 

We do not know if a central form $\omega$ is automatically invariant
by the action of both $U_q(su(2))$, that is: 
$h \triangleright \omega = \epsilon(h) \omega$.

\begin{prop}
The order one calculus up to $OP^{-\infty}$ is not universal.
\end{prop}

\begin{proof}
Let us take the one-form $\omega_F$ from (\ref{eq:F2}),
which gives $F$. Then, for any $x \in \A(SU_q(2))$ we
have
$ \ul\pi(x \omega_F - \omega_F x) = 0$.
\end{proof}

Note that since $\omega_F \sim  (1-q^2)^{-1} \omega \mod OP^{-\infty}$, we get $1\sim (1-q^2)^{-1}\xi_q(|\DD|) \mod OP^{-\infty}$.

\begin{corollary}
Still modulo $OP^{-\infty}$, $1 \in \pi\big(\Omega^2_{u}(\A)\big)$.
\end{corollary}

\begin{proof}
$1=F^2$ is by definition in $\pi\big(\Omega^2_{u}(\A)\big)$.
\end{proof}
In fact, one checks, using \eqref{astuce1}, \eqref{astuce2},
\eqref{astuce5} that
\begin{align}
    \label{eq:1is2form}
q^2\, da\, da^*-da^*\, da=1-q^2
\end{align}
showing again that $1 \in \pi\big(\Omega^2_{u}(\A)\big)$.

Similarly, using
\eqref{eq:commutrulefora,b} and \eqref{astuce1'},
\eqref{astuce5}, \eqref{astuce6}, we get still up to $OP^{-\infty}$
\begin{align}
    \label{2forms}
&q \,da \, db = db \, da, && q \,da \, db^* = db^* \, da, \nonumber\\
&da^* \, db = q \, db \,da^*, && da^* \, db^* = q \, db^* \, da^*
\nonumber\\
&db \, db^* = db^* \,db, &&  da \, da^* +db \, db^* = -1.
\end{align}
The use of the last equality of \eqref{2forms} and
\eqref{eq:1is2form} gives

\begin{prop}
Up to $OP^{-\infty}$, $F$ is not a (universal) closed one-form, as
\begin{align}
\label{dF}
da^* \,da+q^2 \, da \,da^*+ q^2 \,db^* \,db+ q^2 \,db \,db^*  =
-1-q^2.
\end{align}
\end{prop}

\subsection{The ideal $\CR$}

In order to perform explicit calculations of all terms of
the spectral action, we observe that each $\delta$-one-form
could be expressed in terms of $x \delta(z) y$, where $z$ is one
of the generators $a,a^*,b,b^*$ and $x,y$ are some elements of
the algebra $\A(SU_q(2))$.

Then, for the computation of $\ncint x dz y |\DD|^{-1}$
we can use the trace property of the noncommutative integral to get:
$$\ncint x \delta(z) y \,|\DD|^{-1} = \ncint yx \delta(z) \,|\DD|^{-1} +
\ncint x \delta(z) \,|\DD|^{-1} \delta(y) \, |\DD|^{-1}.$$

Therefore, the problem of calculating the tadpole-like integral could
be in effect reduced to the calculation of much simpler integrals:
$ \ncint x \delta(z) |\DD|^{-1}$  for all generators $z$ and the integrals
of higher order in $|\DD|^{-1}$.

However, it appears that the calculations of higher-order terms
simplify
a lot, when we further restrict the algebra by introducing an ideal,
which is {\em invisible} to the parts of integral at dimension $2$ and
$3$. For instance, consider the space of
pseudodifferential operators
$T \in \Psi^{0}(\A)$ of order less or equal to zero (see \cite{CM}),
which satisfy
\begin{align}
    \label{ideal}
\ncint T\, t \,|\DD|^{-2} =\ncint t\,T \,|\DD|^{-2}
= \ncint T\, t \,|\DD|^{-3} = \ncint t\, T \,|\DD|^{-3} = 0,
 \;\forall t \in \Psi^0_{0}(\A).
\end{align}
The elements $a_{-}, \,\, b_- b_+,\,
\, b_{-}b_{+}^*$ and their adjoints
are in this space up to $OP^{-\infty}$: this is due to the fact that
in Theorem \ref{Theo},
$\tau_{1}\otimes \tau_1\big(r(x)\big)=0$ when $r(x)\in
\pi_{\pm}(\A)\otimes \pi_{\pm}(\A)$ mod $OP^{-\infty}$ 
contains tensor products of $\pi_\pm (b)$ or $\pi_\pm(b^*)$ since
these elements are in the kernel of $\sigma$.
\begin{definition}
    \label{Rem:onR}
Let $R$ be the kernel in $X$ of $(\sigma \otimes \sigma)\circ r$
where $r$ is the Hopf-map defined in \eqref{eq:r} and $\sigma$ is the
symbol map and let $\CR$  be the vector space generated by $R$ and
$R\,F$.
\end{definition}

Note that $R$ is a $^*$-ideal in $X$ and $$a_{-}, \,\,b_-
b_+(=q^2\,b_+ b_-),\,
\,b_{-}b_{+}^* \text{ are in }\CR.$$

By construction and Theorem \ref{Theo}, any $T \in \CR$ satisfies
\eqref{ideal} and $\CR$ is invariant by $F$.

Moreover, by \eqref{astuce2}, $[b_{-},b_{-}^*] \in R$, so by
\eqref{astuce1} and \eqref{astuce5},
$a_{+}^*a_{+}-q^2\,a_{+}a_{+}^*-(1-q^2) \in R$ and by
\eqref{astuce6}, $q\,a_{+}b_{-}-b_{-}a_{+} \in R$.

It is interesting to quote, thanks to Theorem \ref{Theo}
that if $x\in R$, then $\ncint F\,x\,\vert\DD\vert^{-1}=0$ while a
priori, $\ncint x\, \vert \DD \vert^{-1}\neq 0$.

Note that $F \in \Psi^0(\A)$ also satisfies \eqref{ideal} by
Theorem \ref{Theo} while $F \notin \CR$ since $F^2=1$.

Moreover other elements are in $\CR$ like for instance
$d(b^*b)= d(bb^*)$:
\begin{align*}
\delta(bb^*)&=-\delta (aa^*)=-\delta a\,a^*-a\,\delta
a^*=-(a_{+}-a_{-})(a_{+}^*+a_{-}^*)
-(a_{+}+a_{-})(a_{-}^*-a_{+}^*)\\
&=2(a_{+}a_{-}^*-a_{-}a_{+}^*)
\end{align*}
is in $R$ since $a_{-}\in R$  yielding $d(bb^*)\in R\,F$.

We do not know if $\CR$ is equal to the subset of the algebra
generated by $\B$ and $\B\, F$ satisfying \eqref{ideal}.

\begin{lemma}
$\CR$ is a $*$-ideal in $\Psi^0(\A)$ which is invariant by
$F$, $d$, $\delta$.
\end{lemma}

\begin{proof}
Since $R$ is an ideal in $X=\B$ mod $OP^{-\infty}$ (see Remark
\ref{pseudodiff}), $\CR$ appears to be an ideal in $\Psi^0(\A)
\subset \text{algebra generated by } \B \text{ and } \B\,F$.
Since $\CR$ is invariant by $F$, its invariance by $d$ follows from
its invariance by $\delta$ which is true on the generators of $R$.
\end{proof}

Note that, according to Theorem \ref{ncintLin},
$\ncint da \,|\DD|^{-2}= \ncint da \,|\DD|^{-3}=0$ while
$\ncint a^*da \,|\DD|^{-3}=2$ which emphasizes the importance the quantifiers "for all" in \eqref{ideal}.

\begin{lemma}
For any $t \in \Psi_{0}^0(\A)$ and $T \in \CR$, we have 
$\ncint t \,T \,|\DD|^{-1} = \ncint T \,t \,|\DD|^{-1}$.
\end{lemma}

\begin{proof}
For any $t\in \B$, we have
 $\ncint T \,t \,|\DD|^{-1} = \ncint t \,T \,|\DD|^{-1} +
\ncint T \,|\DD|^{-1} \,\delta(t) \,|\DD|^{-1}$ and moreover
$\ncint T \,|\DD|^{-1}\, \delta(t) \, |\DD|^{-1} = \ncint T \, \delta(t)
\,|\DD|^{-2}
- \ncint T \, \delta^2(t) \,|\DD|^{-3}$ which comes from
\begin{align*}
|\DD|^{-1}\delta(t)|\DD|^{-1}&=\delta(t)|\DD|^{-2}
+[|\DD|^{-1},\delta(t)]|\DD|^{-1}
=\delta(t)|\DD|^{-2}-|\DD|^{-1}\delta^2(t)|\DD|^{-2}\\
&=\delta(t)|\DD|^{-2}-\delta^2(t)|\DD|^{-3}+|\DD|^{-1}\delta^3(t)|\DD|^{-3}.
\end{align*}
So we get the result because $T$ satisfies \eqref{ideal}.
\end{proof}

\begin{lemma}
    \label{table:commrule}
If $\simeq$ means equality up to the ideal $\CR$, the following rules
with $d(.)=[\DD,.]$ of the first-order differential
calculus hold (suppressing $\piappr$)
\begin{center}
\begin{tabular}{l l l l}
 $a \,da \simeq da \,a$,
& $a^* \,da \simeq - da^* \,a$,
& $b \,da   \simeq q \, da \,b$,
& $b^* \,da \simeq q \,da \,b^*$, \\
 $a\, da^*\simeq - da\,a^*$,
& $a^*\,da^* \simeq da^*\,a^*$,
& $b\,da^* \simeq q^{-1}\,da^*\,b$,
& $b^*\, da^* \simeq q^{-1}\,da^*\,b^*$,\\
 $a \,db \simeq q^{-1} \,db \,a$,
& $a^* \,db \simeq q \,db \,a^*$,
& $b \,db \simeq db \,b$,
& $b^* \,db \simeq db \, b^*\simeq -b\,db^*$, \\
 $a\,db^* \simeq q^{-1}\,db^* \, a$,
& $a^*\, db^* \simeq q\,db^*\, a^*$,
& $b\,db^* \simeq db^*\,b \simeq -b^*\,db$,
& $b^* \, db^* \simeq db^* \, b^*$.
\end{tabular}
\end{center}
Moreover
\begin{align}
    \label{eq:F}
& a^* \,da - q^2 da \, a^* \simeq (1 - q^2) \, F  ,
& q^2 \,a \,da^* - da^* \,a \simeq (1 - q^2) \, F.
\end{align}
\end{lemma}

\begin{proof}
The table follows from relations
\eqref{defrule} and Lemma \ref{commutateur} with \eqref{Fcommutes}
(one can also use \eqref{eq:commutrulefora,b}).
For instance, since $a_{-}\in \CR$, using the fact that $\CR$ is
invariant by $F$,
\begin{align*}
b\, da &=(b_{+}+b_{-})(a_{+}-a_{-})\,F\simeq
(b_{+}+b_{-})(a_{+}+a_{-})\,F=ba\,F=q\,ab \,F
\simeq q  \,(a_{+}-a_{-})\,F\,b\\
&=q \,da \,b
\end{align*}
or similarly, $a^*\,da=(a_{+}^*+a_{-}^*)(a_{+}-a_{-})F\simeq
(a_{+}^*-a_{-}^*)(a_{+}+a_{-})F=-da^*\,a$.

The second equivalence of \eqref{eq:F}
is just the adjoint of the first one that we prove now:
\begin{align*}
a^* \,da - q^2 da \, a^* &=(a_{+}^*+a_{-}^*)(a_{+}-a_{-})F-q^{2}
\,(a_{+}-a_{-})F(a_{+}^*+a_{-}^*) \\
&\simeq
(a_{+}^*+a_{-}^*)(a_{+}+a_{-})F-q^{2}\,(a_{+}+a_{-})
(a_{+}^*+a_{-}^*)F= (a^{*}a-q^2\,aa^{*})\,F\\
&=(1-q^{2})\,F.
\tag*{\qed}
\end{align*}
\hideqed
\end{proof}

\begin{remark}
The rule written above remains if $dx$ is replaced
by $\delta(x)$ and $F$ by $1$.
\end{remark}

Working modulo $\CR$ simplifies the writing of a one-form:
\begin{lemma}
    \label{1formmodR}
(i) Every one-form $A$ can be, up to elements from $\CR$, presented as
$$
A \simeq x_a \, da + da^* \, x_{a^*} + x_b\, db + db^* \,
x_{b^*},
$$
where all $x_*$ are the elements of $\A$.

(ii) When $A$ is selfadjoint, $A$ can be
written up to $\CR$ (not in a unique way, though) as
$$
A \simeq x_a \,da - da^* \,(x_a)^* + x_b \,db - db^* \,(x_b)^*\, ,
$$
where $x_a,x_b$ are arbitrary elements of $\A$.
\end{lemma}

\begin{proof}
$(i)$ A basis for one-forms consists of the following
forms:
$a^{\alpha} b^{\beta} (b^*)^{\gamma} \,
d \big( a^{\alpha'} b^{\beta'} (b^*)^{\gamma'} \big)$,
where $\alpha, \alpha' \in \Z$ and
$\beta,\gamma,\beta',\gamma' \in \N$.

Using the Leibniz rule and the commutation rules within
the algebra (up to the $\CR$ according to Lemma
\ref{table:commrule}),
we reduce the problem to the case of the forms:
$ \big( a^{\alpha} b^{\beta} (b^*)^{\gamma} \big)\,
dx \,\big( a^{\alpha'} b^{\beta'} (b^*)^{\gamma'} \big)$,
where $x$ can be either of the generators $a,a^*,b,b^*$. If
$x=b$ or $x=b^*$, the straightforward application of the
rules of the differential calculus leads to the answer that
the one-form could be expressed as:
$ a^{\alpha} b^{\beta} (b^*)^{\gamma} \,db$ and
$db^* \,a^{\alpha} b^{\beta} (b^*)^{\gamma}.$

Similar considerations for the case $x=a,a^*$ lead to the
remaining terms.

Note that the presentation is not unique,
since there still might remain terms, which are in $\CR$,
for example: $b^* db + db^* b=d(bb^*)\in \CR$.

$(ii)$ is direct.
\end{proof}

Next we can start explicit calculation of the integrals, beginning
with the tadpole terms. 

Application of the Leibniz rule yields a presentation of one-forms
which is different from the one of the previous lemma.
Each $\delta$-one-form could be expressed, as a finite sum of the
terms
$x \delta(z) y$, where $z$ is one of the generators $a,a^*,b,b^*$
and $x,y$ are some elements of the algebra $\A\big(SU_q(2)\big)$.

\begin{prop}
For all $x,y \in \A\big(SU_q(2)\big)$ and $z\in \{a,a^*,b,b^*\}$ we
have
$$\ncint x \delta(z) y \,|\DD|^{-1} = \ncint yx \delta(z) \,|\DD|^{-1} +
\ncint x \delta(z)  \delta(y) \, |\DD|^{-2}- \ncint
x\delta(z)\delta^2(y) \,|\DD|^{-3}.$$
\end{prop}
\begin{proof}
This is just the application of the trace property of the
noncommutative integral, together with the identity: 
$|\DD|^{-1} \delta(z) |\DD|^{-1} = - \left[ |\DD|^{-1}, z \right]$.
\end{proof}
\begin{remark}
The computation of tadpole-like integrals is reduced to the following integrals:
$ \ncint x \delta(z)|\DD|^{-1}$ for all generators $z$ and the
integrals of higher order in $|\DD|^{-2}$. However, the calculations of
higher-order terms simplify a lot when we use the relations which
hold up to the ideal $\CR$: this erases parts of the integral depending
on $\vert \DD \vert ^{-2}$ and 
$\vert\DD\vert^{-3}$. Thus, beside $ \ncint x \delta(z)\,|\DD|^{-1}$,
we only need to compute $\ncint x\delta(z)\delta(z') \,|\DD|^{-2}$
where $z$ and $z'$ are generators, since all the $|\DD|^{-3}$
integrals have already been explicitly computed in section 4.6 (these
integrals do not depend on $q$).

Besides the tadpole, the only integrals that need to be computed are
$\ncint A\,|\DD|^{-2}$ and $\ncint A^2 \,|\DD|^{-2}$ where $A$ is a
$\delta$-1-form. Working modulo $\CR$ and using again Leibniz rule,
we only need to compute $\ncint x\delta(z) \, |\DD|^{-2}$ and  the
previous integrals $\ncint x\delta(z)\delta(z') \, |\DD|^{-2}$.
\end{remark}
\subsubsection{Operators $L_q$ and $M_q$}

In the notation $v^{j}_{l,m}$ of $\H$, we have already used the $j$
dependence in \eqref{J_q} with $J_q \,v^{j}_{m,l} := q^{j}
\,v^{j}_{m,l}$.

Let $L_q$ and $M_q$ be the similar diagonal operators
\begin{align*}
&L_q \,v^{j}_{m,l} := q^{2l} \,v^{j}_{m,l}\,, \\
&M_q \,v^{j}_{m,l}: = q^{2m} \,v^{j}_{m,l}\,.
\end{align*}
We immediately get
\begin{lemma}
\label{lemmaLqMq}
For $n\in \N^*$, $\ncint (L_q)^n \,|\DD^{-2}| = \ncint (M_q)^n
\,|\DD^{-2}| = \frac{2}{1-q^{2n}}\,.$
\end{lemma}

\begin{proof} We have
\begin{align*}
\Tr \big(L_q^n |\DD|^{-2-s} \big)&= \sum_{2j=0}^{\infty}
\sum_{m=0}^{2j} \sum_{l=0}^{2j+1} \langle v^j_{m,l},L_q^n
|\DD|^{-2-s}v^j_{m,l}\rangle \\
&= \sum_{2j=0}^\infty (2j+1)\tfrac{1-q^{2n(2j+2)}}{1-q^{2n}}
d_{j^+}^{-2-s}+ \sum_{2j=0}^\infty
(2j+1)\tfrac{1-q^{2n(2j+2)}}{1-q^{2n}} d_{j}^{-2-s} \\
&\sim_0 \tfrac{1}{1-q^{2n}}\big( \zeta(s+1,\tfrac{3}{2}) +
\zeta(s+1,\tfrac{1}{2} \big) \sim_e \tfrac{2}{1-q^{2n}} \zeta(s+1) \,
.
\end{align*}
where $\sim_0$ means modulo a function holomorphic at 0. This gives
the result for $L_q^n$ and a  similar computation can be done for
$M_q^n$.
\end{proof}

The interest of these operators stems from 
\begin{lemma}
   \label{uptoR}
We have $L_qM_q \in \CR$. Moreover,
\begin{align*}
& b \, \delta b^* \simeq  M_q - L_q, \quad   b^* \delta b \simeq L_q
- M_q,
\quad b b^* \simeq L_q + M_q,\\
& a \,\delta(a^*) \simeq -aa^*\simeq L_q + M_q -1,\quad a^*\, \delta
a \simeq a^*a \simeq 1 - q^2 (L_q + M_q) ,\\
& da\, d a^* \simeq L_q +M_q-1,\quad da^*\,da \simeq q^2(L_q
+M_q)-1,\\ 
& b^{n-2} (b^*)^n \, db \, db \simeq  (L_q)^n + (M_q)^n, \\
& b^{n-1} \,(b^*)^{n-1} \, db \, db^* \simeq - (L_q)^n - (M_q)^n, \\
& b^{n} (b^*)^{n-2} \, db^* \, db^* \simeq (L_q)^n + (M_q)^n.
\end{align*}
\end{lemma}

\begin{proof}
Since $L_qM_q=q^2\,a_-a_-^* \in \CR$, we compute up to the ideal $\CR$
\begin{align*}
b \, \delta b^* & = (b_+ +b_-) (b_-^*-b_+^*) \simeq -b_+b_+^*
+b_-b_-^*=M_q -L_q + L_q M_q(1-q^2)  \simeq  M_q - L_q
\end{align*}
and similarly for the other relations.
\end{proof}

\subsubsection{Automorphisms of the algebra and symmetries of
integrals}
\begin{prop}
\label{calculavecLM}
For any $n \in \N^*$,  
\begin{align*}
& \ncint (bb^*)^n \,|\DD|^{-1} =
\tfrac{-2(1+q^{2n})}{(1-q^{2n})^2}\,, \\
& \ncint (bb^*)^{n} b^* \,\delta b\, |\DD|^{-1} = 
\ncint (bb^*)^{n} b \,\delta b^*\, |\DD|^{-1} =
\tfrac{2}{1-q^{2n+2}}\,, \\
& \ncint (bb^*)^n a \, \delta a^* \,|\DD|^{-1}  =
\tfrac{-2q^{4n+2}-2q^{4n}-2q^{2n+2}+6q^{2n}}{(1-q^{2n})^2(1-q^{2n+2})}\,,\\
& \ncint (bb^*)^n a^* \, \delta a \,|\DD|^{-1}  =
\tfrac{6q^{2n+2}-2q^{2n}-2q^2-2}{(1-q^{2n})^2(1-q^{2n+2})}\,.
\end{align*}
\end{prop}
Note that the knowledge of these integrals is enough for the
computation of any term of the form $\ncint x\delta(z) |\DD|^{-1}$,
where $z$ is a generator, since any other $\delta$-one-form will be
unbalanced.

To show this proposition, we will use few symmetries, properties of
the ideal $\CR$ and replacement of $\delta$-one-forms in terms of
$L_q,\, M_q$ as above.

Let $U$ be the following unitary operator on the Hilbert space:

\begin{equation*}
U \,v^{j\up}_{m,l} = (-1)^{m+l} \,v^{j^+\dn}_{l,m}, \;\;\;
U \,v^{j\dn}_{m,l} = (-1)^{m+l} \,v^{j^-\up}_{l,m}.
\end{equation*}
Then, by explicit computations we have
\begin{align*}
U^* a U =  a, \qquad U^* a^* U = a^*,  &\qquad U^* b U =  b^*, \qquad
U^* b^* U = b,\qquad \text{and }\quad U^*  \DD U = - \DD. 
\end{align*}

\begin{lemma}
Each noncommutative integral \eqref{I_pA^n} of an element of the
algebra or differential
forms is (up to sign) invariant under the algebra automorphism $\rho$
defined by
\begin{align}
\rho(a):=a, \;\; \rho(a^*):=a^*,\;\; \rho(b):=b^*,\;\; \rho(b^*):=b.
\label{autominv}
\end{align}
\end{lemma}

\begin{proof}
For any homogeneous polynomial $p$ and any $k\in \N$,
\begin{align*}
 \ncint p(a,a^*,b,b^*,\DD) \,\DD^{-k} &= \ncint U^* p(a,a^*,b,b^*,
\DD) \,\DD^{-k} U \\
& = (-1)^k \ncint p(U^* a U, U^* a^* U, U^* b U, U^* b^* U,U^* \DD U)
\,\DD^{-k} \\
& = (-1)^{k+d} \ncint p(\rho(a),\rho(a^*),\rho(b),\rho(b^*),\DD)
\,\DD^{-k},
\end{align*}
where $d$ is the degree of $p$ with respect to $\DD$.
\end{proof}

\begin{corollary}
\label{cor-01}
For any $n\in \N$,
$ \ncint(bb^*)^{n}\, b^* \,d b\, \DD^{-1}
=  \ncint(bb^*)^{n} \,b \, d b^* \,\DD^{-1}.$
\end{corollary}

\begin{lemma} 
\label{use-l1}
For any $x,y \in \Psi^0(\A)$,
\begin{align*}
& \hspace{-4cm} (i) \quad \, \ncint x y |\DD|^{-1} = \ncint yx
|\DD|^{-1} + \ncint x \delta(y) |\DD|^{-2} - \ncint x \delta^2(y)\,
|\DD|^{-2}.\\
&  \hspace{-4cm}(ii) \quad \ncint z \,x \,\DD^{-1} y\,\DD^{-1}
=\ncint z \,xy \,\DD^{-2}, \text{ if $z \in \A$ contains $b$ or
$b^*$.}
\end{align*}
\end{lemma}

\begin{proof}
$(i)$ is direct consequence of the trace property of $\ncint$ and the
fact that $OP^{-4}$ operators are trace-class.

$(ii)$ We calculate:
\begin{align*}
 \ncint z \,x \,\DD^{-1} y \,\DD^{-1} &=
 \ncint z \,x\left( y \, \DD^{-1}
- \DD^{-1} [ \DD, y ] \,\DD^{-1} \right) \DD^{-1}=  \ncint z \,x y
\,\DD^{-2}
  -  \ncint z \,x \DD^{-1} [\DD, y] \,\DD^{-2} \\
 &= \ncint z \,x y \,\DD^{-2}.
\end{align*}
The last step is based on the observation that any integral with
$\DD^{-3}$ vanishes if the expression integrated contains $b$
or $b^*$.
\end{proof}

\begin{lemma}
For any $n \in \N$,
\begin{align*}
&\hspace{-7.5cm} (i)  \quad \, \ncint(bb^*)^{n} \, b^* \,db
\,\DD^{-1}= \tfrac{2}{1-q^{2n+2}}\,.\\
&\hspace{-7.5cm} (ii)  \quad \ncint (bb^*)^{n} d(bb^*) \,\DD^{-1}=0.\\
& \hspace{-7.5cm} (iii) \! \quad  \ncint (bb^*)^n |\DD|^{-1} =
\tfrac{-2(1+ q^{2n})}{(1-q^{2n})^2}\,. 
\end{align*}
\end{lemma}

\begin{proof}
$(i)$ With $n>1$, we begin with $\ncint d \left( (b\,b^*)^n \right)
\DD^{-1} = 0, $
which follows directly from the trace property of the noncommutative
integral. Expanding the expression using the Leibniz rule and the
commutation
\begin{equation}
x \DD^{-1} = \DD^{-1} x + \DD^{-1} [\DD,x] \DD^{-1}, \label{bcom}
\end{equation}
we obtain
\begin{align*}
0  =& \sum_{k=0}^{n-1} \ncint  b^k \, db \, b^{n-k-1} (b^*)^n
\,\DD^{-1}
+ \sum_{k=0}^{n-1} \ncint b^n (b^*)^k \, db^* \, (b^*)^{n-k-1}
\,\DD^{-1} \\
=& n \left( \ncint b^{n-1} (b^*)^n \, db \,\DD^{-1} +
           \ncint b^{n} (b^*)^{n-1} \, db^* \,\DD^{-1}  \right) \\
&\qquad + \sum_{k=0}^{n-1} \ncint
\left( b^k \, db \, \DD^{-1} d( b^{n-k-1} (b^*)^n) \DD^{-1}
+ b^n (b^*)^k \, db^* \, \DD^{-1} d ((b^*)^{n-k-1}) \DD^{-1} \right) .
\end{align*}
Using Lemma \ref{use-l1}, 
\begin{align*}
0=& n \ncint(bb^*)^{n-1} ( b^* \,db + b \, db^*) \,\DD^{-1}  \\
& \quad+ \ncint \left( \oh n (n-1) b^{n-2} (b^*)^n \, db \, db
+ n^2 b^{n-1} \,(b^*)^{n-1} \, db \, db^*
+ \oh n (n-1) b^{n} (b^*)^{n-2} \, db^* \, db^* \right) \DD^{-2}.
\end{align*}
The integrals with $\DD^{-2}$ could be easily calculated when
we restrict ourselves to calculations modulo the ideal $\CR$:
\begin{align*}
 n \ncint(bb^*)^{n-1} & ( b^* \,db + b \, db^*) \,\DD^{-1} = 
 - 2 \left( n(n-1) - 2 n^2 + n(n-1) \right) \tfrac{1}{1-q^{2n}}
= 4 n \tfrac{1}{1-q^{2n}}\, .
\end{align*}
Hence $ \ncint(bb^*)^{n-1} ( b^* \,db + b \, db^*) \, \DD^{-1} =
\tfrac{4}{1-q^{2n}}$, 
which together with Corollary \ref{cor-01} proves $(i)$.

$(ii)$ In a similar way, 
$\ncint (bb^*)^{n-1} d( bb^*) \, \DD^{-1}= 0 = \ncint(bb^*)^{n-1} d(
aa^*) \, \DD^{-1}$ implies:
\begin{align*}
0 = & \sum_{k=0}^{n-1} (b b^*)^{n-k-1} d(bb^*) (b b^*)^{k}\, \DD^{-1}
\\ 
  =&  n \ncint (bb^*)^{n-1} d(bb^*) \,\DD^{-1}
+ \oh n(n-1) \ncint (bb^*)^{n-2} d(bb^*) \, d(bb^*) \,\DD^{-2}\\
  = & n \ncint (bb^*)^{n-1} d(bb^*) \,\DD^{-1}, 
\end{align*}
where in the last step we used that $d(bb^*) \in \CR$. The identity
$(ii)$ now follows from the equality $aa^* = 1 - bb^*$.

$(iii)$ Using Lemma \ref{use-l1}, we get
\begin{align*}
A_n &:= \ncint (bb^*)^n |\DD|^{-1} = \ncint (bb^*)^n (aa^* + bb^*)
|\DD|^{-1} \\
& \hspace{-2.3cm} \text{and we push now $a^*$ through $|\DD|^{-1}$ and
from cyclicity
of the trace through $(bb^*)^n$,} \\
& \,= A_{n+1} + \ncint (bb^*)^n q^{2n} a^* a \,|\DD|^{-1} +
\ncint (bb^*)^n q^{2n} a \delta(a^*) \,|\DD|^{-2}  \\
& \hspace{-2.3cm} \text{the last term being calculated explicitly,
since up to ideal
$\CR$, $a \delta(a^*) \simeq L_q + M_q -1$,} \\
& \, = A_{n+1} (1-q^{2n+2}) + q^{2n} A_n +
4 \left( \tfrac{1}{1-q^{2n+2}} - \tfrac{1}{1-q^{2n}} \right),
\end{align*}
which leads to
\begin{equation*}
A_n (1-q^{2n}) + \tfrac{4}{1-q^{2n}} = A_{n+1} (1-q^{2n+2}) +
\tfrac{4}{1-q^{2n+2}}\,.
\end{equation*}
Assuming $A_n = \tfrac{f_n}{(1-q^{2n})^2}$ we have 
$\frac{f_n + 4}{1-q^{2n}} = \tfrac{f_{n+1}+4}{1-q^{2n+2}}\,,$ 
and taking into account that $A_0 = - 2 \tfrac{1+q^2}{(1-q^2)^2}$, 
we obtain
$A_n = - 2 \tfrac{1+ q^{2n}}{(1-q^{2n})^2}\,.$
\end{proof}

Finally, to get Proposition \ref{calculavecLM}, it remains to prove

\begin{lemma}
For $n\geq 1$,
\begin{align*}
\ncint (bb^*)^n a \, da^* \,\DD^{-1} & =
\tfrac{-2q^{4n+2}-2q^{4n}-2q^{2n+2}+6q^{2n}}{(1-q^{2n})^2(1-q^{2n+2})}\,,\\
\ncint (bb^*)^n a^* \, da \,\DD^{-1} & =
\tfrac{6q^{2n+2}-2q^{2n}-2q^2-2}{(1-q^{2n})^2(1-q^{2n+2})}\,.
\end{align*}
\end{lemma}

\begin{proof}
First, using the Leibniz rule, (\ref{bcom}) and Lemma \ref{use-l1} we
have (for $n \geq 1$)
$$ \ncint (bb^*)^n a \, da^*\, \DD^{-1} =
- q^{2n} \ncint (bb^*)^n a^* \, da - \ncint (bb^*)^n da \, da^*
\,\DD^{-2}. $$

Further, we use the identity (\ref{eq:F2}):
$$ \ncint (bb^*)^n \left( a^* \, da + q^2 a \, da^*
   + q^2 b\, db^* + q^2 b^* \, db \right)\, \DD^{-1}
   = (1-q^2) \ncint (bb^*)^n \,|\DD|^{-1}.$$
taking into account that $F\, \DD = |\DD|$.

These equations give together a system of linear equations
\begin{align*}
\ncint (bb^*)^n a \, da^* \,\DD^{-1} + q^{2n} \ncint (bb^*)^n a^* \,
da \,\DD^{-1}
&= -4 \left( \tfrac{1}{1-q^{2n+2}} - \tfrac{1}{1-q^{2n}} \right), \\
q^2 \ncint (bb^*)^n a \, da^* \, \DD^{-1} + \ncint (bb^*)^n a^* \, da
\,\DD^{-1}
&= -2 (1-q^2) \tfrac{1+q^{2n}}{(1-q^{2n})^2} - \tfrac{4
q^2}{1-q^{2n+2}}
\end{align*}
which is solved by the expressions stated in the lemma.
\end{proof}

\subsubsection{The noncommutative integrals at $|\DD|^{-2}$}

We need to separate this task into two problems. First, we shall calculate all integrals of the type $\ncint x\, \delta(z)|\DD|^{-2}$,
with $x \in \A(SU_q(2))$ and $z$ being one of the generators. The
second
problem is to calculate  $\ncint x \,\delta(y) \, \delta(z)\,
|\DD|^{-2}$, with
both $y$ and $z$ being the generators $\{a,a^*,b,b^*\}$.

\begin{lemma} The only a priori non-vanishing integrals of the
type $\ncint x \, \delta(z)  \,|\DD|^{-2}$ are for $n\in \N$:
\begin{align*}
&\ncint (bb^*)^n b^*\delta(b) \,|\DD|^{-2} =\ncint (bb^*)^n
b\delta(b^*) \,|\DD|^{-2}=0 ,\\
& \ncint (bb^*)^n a\delta(a^*) \,|\DD|^{-2} =
\tfrac{4q^{2n}(1-q^2)}{(q^{2n+2}-1)(1-q^{2n})}\, , \quad n>0\\
&\ncint (bb^*)^n a^*\delta(a) \,|\DD|^{-2} =
\tfrac{4(1-q^2)}{(1-q^{2n+2})(1-q^{2n})} \,.
\end{align*}
\end{lemma}
\begin{proof}
Since $a\delta(a^*)\simeq L_q+M_q-1$ and $(bb^*)^n\simeq
L_q^n+M_q^n$, we get 
$$
(bb^*)^n a\delta(a^*) \simeq L_q^{n+1}+M_q^{n+1}-L_q^n-M_q^n
$$
and the second result is obtained from Lemma \ref{lemmaLqMq}. The
other integrals are computed in a similar way.
\end{proof}

\begin{lemma}
The only a priori non-vanishing integrals of the
type $\ncint x \, dy \, dz \,|\DD|^{-2}$ are for $n\in \N$:
\begin{align*}
&\ncint (bb^*)^{n} \, (b^*)^2 db\, db \,|\DD|^{-2} =
\tfrac{4}{1-q^{2n+4}}, \\
&\ncint (bb^*)^{n} \, db\, db^*  \, |\DD|^{-2} =
\tfrac{4}{1-q^{2n+2}}, \\
&\ncint (bb^*)^{n} \, (a^* b^*) (da \,db)\,|\DD|^{-2} = 0, \\
&\ncint (bb^*)^{n} \, (a b^*) (da^* \,db)\,|\DD|^{-2} = 0, \\
&\ncint (bb^*)^{n} \, (a^* b) (da \,db^*)\,|\DD|^{-2} = 0, \\
&\ncint (bb^*)^{n} \, (a b) (da^* \,db^*)\,|\DD|^{-2} = 0,\\
&\ncint (bb^*)^{n} \, (da \,da^*)\,|\DD|^{-2}
=\tfrac{4(q^{2n+2}-q^{2n})}{(1-q^{2n+2})(1-q^{2n})}\,, \quad n>0\\
&\ncint (bb^*)^{n} \, (da^* \,da) \,|\DD|^{-2} 
=\tfrac{4(q^2-1)}{(1-q^{2n+2})(1-q^{2n})}\,.
\end{align*}
\end{lemma}

\begin{proof}
This follows from Lemma \ref{lemmaLqMq} with the equivalences up to
$\CR$ gathered in Lemma \ref{uptoR}.
\end{proof}

\section{Examples of spectral action}

It is clear from Theorem \ref{mainThmJ} that any one-form of
the form $ada$, $bdb$, $adb$, $a^*db$, etc... do not contribute to
the spectral action. Indeed, only
the balanced parts of one-forms give a possibly nonzero term
in the coefficients.
We give in the following table the values of the terms $\ncint A^n\,\vert \DD
\vert^{-p}$ and the full $\zeta_{\DD_{\Abb}}(0)$ for a few examples.

\begin{table}[ht]
\caption{Values of noncommutative integrals}
\centering 
\begin{tabular}{c c c c c c c c} 
\hline\hline 
\\
$\Abb$ & $\ncint A\,\vert \DD \vert^{-3}$  & $\ncint A^2\,\vert \DD
\vert^{-3}$ & $\ncint A^3\,\vert \DD \vert^{-3}$ & $\ncint A\,\vert
\DD \vert^{-2}$ & $\ncint A^2\,\vert \DD \vert^{-2}$
& $\ncint A\,\vert \DD \vert^{-1}$ & $\zeta_{\DD_{\Abb}}(0)$ \\
[0.5ex] 
\hline 
\\
$a^*da$ & 2 & 2 & 2 & $\tfrac{4q^2}{q^2-1}$ &
$\tfrac{4q^2(q^2+2)}{q^4-1}$ & $\tfrac{3q^2+1}{2(q^2-1)}$ &
$\tfrac{11q^4+36q^2+13}{3(q^4-1)}$ \\
$b^*db$ & 0 & 0 & 0 & 0 & $\tfrac{-4}{q^4-1}$ & $\tfrac{-2}{q^2-1}$ &
$\tfrac{4q^2}{q^4-1}$\\
$ada^*$ & $-2$ & $2$ & $-2$ & $\tfrac{-4}{q^2-1}$ &
$\tfrac{4(2q^2+1)}{q^4-1}$ & $\tfrac{q^2+3}{2(q^2-1)}$ &
$\tfrac{13q^4+36q^2+11}{3(q^4-1)}$  \\
$bdb^*$ & 0 & 0 & 0 & 0 & $\tfrac{-4}{q^4-1}$ & $\tfrac{-2}{q^2-1}$ &
$\tfrac{4q^2}{q^4-1}$\\
\\
\hline 
\end{tabular}
\label{table:nonlin}
\end{table}

\noindent 1) Clearly the spectral action depends on $q$: for instance,
\begin{align*}
&\SS(\DD_{a^*da},\Phi,\Lambda) \, = \,2\,\Phi_{3}\,\Lambda^{3}
-8\,\Phi_{2}\,\Lambda^{2}+\,\tfrac{q^2+15}{2(1-q^2)}\,\Phi_{1}\,\Lambda^{1}+\,
  \tfrac{11q^4+36q^2 +13}{3(q^4-1)}\,\Phi(0).
\end{align*}
2) Moreover, for $B:=a\,\delta a^*$ and $A:=B+B^*$, we get since $B
\simeq B^*\mod \CR$,
\begin{align}
     \label{symmetrization}
\ncint A^p |\DD|^{-k} = 2^{p} \ncint B^p |\DD|^{-k}, \quad 1\leq
p\leq k\leq 3\, .
\end{align}
Thus the spectral action of the selfadjoint one-form
$\mathbb{A}:=ada^*+(ada^*)^*$ is
\begin{align*}
\SS(\DD_{\mathbb{A}},\Phi,\Lambda) \, = \,2\,\Phi_{3}\,\Lambda^{3}+
16\,\Phi_2\,\Lambda^2+ \tfrac{q^2-33}{2(1-q^2)}\,\Phi_1\,\Lambda^1+ 
  \tfrac{122q^4+168 q^2-2}{3(q^4-1)}\,\Phi(0) .
\end{align*}
3) When $B_n:=(bb^*)^n \,b\, \delta b^*$, then by Lemma
\eqref{uptoR}, $B_n \simeq B_n^*$, so for $A_n:=B_n+B_n^*$, the
equation \eqref{symmetrization} is still valid and $\ncint B_n^p
\,|\DD|^{-k}$ are all zero but $\ncint B_n
\,|\DD|^{-1}=\tfrac{2}{1-q^{2n+2}}$ and $\ncint B_n^2
\,|\DD|^{-2}=\tfrac{4}{1-q^{4n+4}}$, so
\begin{align}
   \label{strange}
\SS(\DD_{\mathbb{A}_n},\Phi,\Lambda) \, =
\,2\,\Phi_{3}\,\Lambda^{3}-\half\, \Phi_1 \Lambda^1 +
\tfrac{8}{1+q^{2n+2}} \,\Phi(0).
\end{align}
Remark that this spectral action still exists as $q \to 1$!

Note however that the symmetrization process \eqref{symmetrization}
does not work in general, for instance if $B:=a \, \delta b$ and
$A:=B+B^*$, then $\ncint A^2
|\DD|^{-1}=\tfrac{8(q^4-q^2-1)}{(1-q^4)^2}$ while $\ncint B^2
|\DD|^{-1}=0$ or $\ncint [B,\,B^*] |\DD|^{-1}= \tfrac{4}{1-q^4}\,$.

\noindent 4) The spectral action can be also independent of $q$: for instance,
if $\mathbb{A} = \tfrac{1}{1-q^2}\,\xi(\DD)$ is the $q$-dependent
selfadjoint one-form given in \eqref{eq:F3}, then,
$$
\SS(\DD_{\mathbb{A}},\Phi,\Lambda) \, = \,2\,\Phi_{3}\,\Lambda^{3}
-8\,\Phi_{2}\,\Lambda^{2}+\tfrac{15}{2}\, \Phi_1\,
\Lambda^1-\tfrac{13}{3}.
$$

\section{The commutative sphere $\mathbb{S}^3$}

Since $ SU(2) \simeq \mathbb{S}^3$, we get a concrete spinorial
representation of the algebra $\A:=C^{\infty}(\mathbb{S}^3)$ on the
same Hilbert space $\H$ and same Dirac operator $\DD$ with
\eqref{eq:reprexact} where $q=1$ which means that $q$-numbers are
trivial: $[\alpha]=\alpha$. So
\begin{align}
  \label{eq:reprexacq=1t}
\pi(a)\, \kett{j\mu n } &:= \alpha^+_{j\mu n}\,\kett{j^+ \mu^+ n^+}
+\alpha^-_{j\mu n}\,\kett{j^- \mu^+ n^+},\nonumber \\
\pi(b)\, \kett{j\mu n } &:= \beta^+_{j\mu n}\,\kett{j^+ \mu^+ n^-}
+\beta^-_{j\mu n}\,\kett{j^- \mu^+ n^-},\nonumber\\
\pi(a^*)\, \kett{j\mu n } &:= \tilde{\alpha}^+_{j\mu n}\,\kett{j^+
\mu^- n^-}+\tilde{\alpha}^-_{j\mu n}\,\kett{j^- \mu^- n^-},\nonumber
\\
\pi(b^*)\, \kett{j\mu n } &:= \tilde{\beta}^+_{j\mu n}\,\kett{j^+
\mu^- n^+}+\tilde{\beta}^-_{j\mu n}\,\kett{j^- \mu^- n^+}
\end{align}
where
\begin{align*}
\alpha^+_{j\mu n}&:=\sqrt{ j+\mu+1 }
\left( \begin{array}{cc}
\tfrac{ \sqrt{j+n+3/2}}{2j+2} &   0\\
\tfrac{ \sqrt{j-n+1/2}}{(2j+1)(2j+2} & \tfrac{\sqrt{j+n+1/2]}}{2j+1}  
\end{array} \right),\\
\alpha^-_{j\mu n}&:=\sqrt{ j-\mu }
\left( \begin{array}{cc}
\tfrac{ \sqrt{j-n+1/2}}{2j+1} &  -\tfrac{\sqrt{j+n+1/2}}{2j(2j+1)}\\
0 & \tfrac{\sqrt{j-n-1/2}}{2j}  \\
\end{array} \right),\\
\beta^+_{j\mu n}&:=\sqrt{j+\mu+1 }
\left( \begin{array}{cc}
\tfrac{ \sqrt{j-n+3/2}}{2j+2} &   0\\
-\tfrac{ \sqrt{j+n+1/2}}{(2j+1)(2j+2)} & \tfrac{\sqrt{j-n+1/2}}{2j+1}
\end{array} \right),\\
\beta^-_{j\mu n}&:=\sqrt{ j-\mu }
\left( \begin{array}{cc}
-\tfrac{\sqrt{j+n+1/2}}{2j+1} &  - \tfrac{ \sqrt{j-n+1/2}}{2j(2j+1} \\
0 & -\tfrac{\sqrt{j+n-1/2}}{2j}  
\end{array} \right)
\end{align*}
with $\tilde{\alpha}^{\pm}_{j\mu n}:=(\alpha^{\mp}_{j^{\pm}\mu^-
n^-})^*$, 
$\tilde{\beta}^{\pm}_{j\mu n}:=(\beta^{\mp}_{j^{\pm}\mu^- n^+})^*$.

Note that the representation on the vectors $v^j_{m,l}$ is not as
convenient as in \eqref{eq:rpnappr}.

One can check that the generators $\pi(a)$, $\pi(b)$ and their
adjoint commute and that $[x,[\DD,y]]=0$ for any $x,\,y \in \A$.

\subsection{Translation of Dirac operator}

In general the Dirac operator is defined in a more symmetric way than
what we did. So, although not absolutely necessary here,  we define
for the interested reader the unbounded self-adjoint translated
operator $\DD'$ on $\H$ by the constant $\lambda$ as 
$$
\DD':= \DD+\lambda.
$$
For instance, this gives for $\lambda=-\half$ in the case of
$\mathbb{S}^3$, see \cite{Homma}, 
$\DD'\, v^j_{m,l} = (2j+1) \genfrac{(}{)}{0pt}{1}{1 \quad
0}{ \,0 \,\,\, -1}  v^j_{m,l}$ so $v^j_{m,l}$ is an eigenvector of
$\vert \DD' \vert$. We define $D:=\DD + P_0$ and $D':=\DD' + P'_0$ where $P_0$ is the projection on $\Ker \DD$ and $P_0'$ is the projection on $\Ker \DD'$.

As the following lemma shows, the computation of the noncommutative
integrals involving $\DD$ can be reduced to the computation of
certain noncommutative integrals involving $\DD'$:

\begin{lemma}
   \label{translation}
If $\ncint' T := \Res_{s=0} \Tr \big(T |D'|^{-s}\big)$, then for
any 1-form $\Abb$ on a spectral triple of dimension $n$, 
\begin{align*}
&\ncint A \,|D|^{-(n-2)} =   \ncint' A\,|D'|^{-(n-2)} +\lambda\,
(n-2) \ncint' A \,\DD'|D'|^{-n} + \lambda^2\, \tfrac{(n-1)(n-2)}{2}
\ncint'A\, |D'|^{-n},\\
&\ncint A \,D^{-(n-2)} =   \ncint' A\,D'^{-(n-2)} +\lambda\,
(n-2) \ncint' A \,D'^{-(n-1)} + \lambda^2\, \tfrac{(n-1)(n-2)}{2}
\ncint'A\, D'^{-n}.
\end{align*}
\end{lemma}

\begin{proof}
Recall from Proposition \ref{ncintfluctuated} that for any
pseudodifferential operator $P$,
$$
\ncint P |D|^{-r} = \Res_{s=0} \Tr \big(P|D|^{-r}|D'|^{-s}
\big).
$$
Moreover by Lemma \ref{2dev}, for any $s\in \C$ and $N\in \N^*$
\begin{equation}
      \label{D-s}
|D|^{-s}= |D'|^{-s} + \sum_{p=1}^N K_{p,s} \,Y^p\, |D'|^{-s}
\mod OP^{-N-1-\Re(s)}
\end{equation}
where 
$Y= \sum_{k=1}^{N}\tfrac{(-1)^{k+1}}{k} (-2\la D'+\lambda^2)^k
D'^{-2k} \mod OP^{-N-1}$ 
and $K_{p,s}$ are complex numbers that can be explicitly computed.
Precisely, we find
$K_{p,s}= (-\tfrac{s}{2})^p \,V(p)$
where $V(p)$ is the volume of the $p$-simplex. Since the spectral
dimension is $n$, we work modulo $OP^{-(n+1)}$, and for $s=n-2$, we get
from (\ref{D-s}): 
$|D|^{-(n-2)}= |D'|^{-(n-2)} +\lambda (n-2)\DD'|D'|^{-n} +
\lambda^2 \tfrac{(n-1)(n-2)}{2}|D'|^{-n}  \mod OP^{-(n+1)}$.

As a consequence, we have for $P\in OP^0$ (the $OP^0$ spaces are the
same for  $\DD$ or $\DD'$),
\begin{align*}
\ncint P |D|^{-(n-2)} = \ncint' P|D'|^{-(n-2)} +\lambda (n-2)
\ncint' P \DD'|D'|^{-n} + \lambda^2 \tfrac{(n-1)(n-2)}{2} \ncint'P
|D'|^{-n}.
\end{align*}
Since $A$ and $AF$ are in $OP^{0}$, we get both formulae.
\end{proof}

\subsection{Tadpole and spectral action on $\mathbb{S}^3$} 
We consider now the commutative spectral triple
$(C^{\infty}(\mathbb{S}^3),\,\H,\, \DD)$. It is 1-summable since 
$\langle \! \langle j\mu n\,s\,|\, [F,\pi(x)] \,\kett{j\mu n\,s}=0$
when $x=a,a^*,b,b^*$ for any $j,\,\mu,\,n,\,s=\up,\dn$.

All integrals of above lemma are zero for $\mathbb{S}^3$:
\begin{prop}
  \label{notadpole}
There is no tadpole of any order on the commutative real spectral
triple $(C^{\infty}(\mathbb{S}^3),\,\H,\, \DD)$. More generally, for any one-form $\Abb$, $\ncint \Abb F |\DD|^{-p}=0$ for
$p\in \N$.
\end{prop}

\begin{proof}
Since the representation is real, that is any
matrix elements of the generators are real, so must be the trace of
$\Abb F |\DD|^{-p}$. Hence $\ncint \Abb F |\DD|^{-p}= \ncint \Abb^* F |\DD|^{-p}$.

The reality operator $J$ introduced in \eqref{DefJ} satisfies, when
$q=1$, the commutative relation $JxJ^{-1}=x^*$ for $x \in \A$.
Thus $ J \Abb J^{-1} = - \Abb^*$ and 
$ \ncint \Abb F |\DD|^{-p} = \ncint J \left( \Abb^* F |\DD|^{-p} \right) J^{-1} 
= - \ncint \Abb F |\DD|^{-p}$ and $\ncint \Abb F |\DD|^{-p}=0$.
\end{proof}

For any selfadjoint one-form $\Abb$, $\DD_{\Abb}:=\DD+\wt \Abb=\DD$. Thus, the
spectral action for the real spectral triple
$\big(C^{\infty}(\mathbb{S}^3),\,\H,\, \DD \big)$ for $\DD_\Abb$  is
trivialized by
\begin{align}
   \label{spectralactioncomm} 
\SS(\DD_{\Abb},\Phi,\Lambda) \, = \,2\,\Phi_{3}\,\Lambda^{3}
-\tfrac{1}{2}\, \Phi_1\, \Lambda^1+\mathcal{O}(\Lambda^{-1}).
\end{align}
But it is more natural to compare with the spectral action of
$\DD+\Abb$.  This is  obtained respectively from Lemma \ref{ncintD} and
general heat kernel approach \cite{Gilkey2}:
\begin{align*}
\SS(\DD + \Abb,\Phi,\Lambda) \, = \,2\,\Phi_{3}\,\Lambda^{3}   + \ncint
\vert \DD+\Abb \vert^{-1} \,  \Phi_1\,
\Lambda^1+\mathcal{O}(\Lambda^{-1}) 
\end{align*}
since all terms of \eqref{formuleaction} in $\Lambda^{n-k}$ are zero
for $k$ odd and $\zeta_{\DD+\Abb}(0)=0$ when $n$ is odd: as a
verification, $\ncint \vert \DD+\Abb \vert^{-2}$ is zero according to
Lemma \ref{residus-particuliers}, Lemmas \ref{ncintD} and  Proposition
\ref{notadpole}. Similarly, $\zeta_{\DD+\Abb}(0)=0$ because in
\eqref{termconstanttilde}, all terms with $k$ odd are zero (same proof as in
Proposition \ref{notadpole}) but for $k$ even, it is not that easy to
show that $\ncint \Abb \DD^{-1} \Abb \DD^{-1}= 0$.

Moreover, the curvature term does not depend on $\Abb$:
\begin{lemma}
\label{commAF}
For any one-form $\Abb$ on a commutative spectral triple of dimension
$n$ based on a compact Riemannian spin$^c$ manifold without boundary,
we have 
\begin{align}
  \label{curvature}
 \ncint |\DD+\Abb|^{-(n-2)}= \ncint \vert \DD\vert^{-(n-2)} .
\end{align}
\end{lemma}

\begin{proof}
Follows from \cite[first formula page 511]{Polaris}  with
$\rho:=\Abb=\Abb^*$, $N(\rho)=\rho$ (the constraint $J\rho J^{-1}=\pm \rho$
is not used).

One can also use \cite[Proposition 1.149]{ConnesMarcolli}.
\end{proof}
From Lemma \ref{residus-particuliers},
 $\ncint |\DD+\Abb|^{-(n-2)}= \ncint \vert \DD\vert^{-(n-2)} +
\tfrac{n(n-2)}{4}\, \ncint (\Abb F)^2|\DD|^{-3}+
\tfrac{(n-2)^2}{4}\,\ncint \Abb^2 |\DD|^{-3}$    using $X:=\Abb\DD+\DD
\Abb+\Abb^2$ and $[\vert \DD \vert, \Abb]\in OP^0$, but  again, it is not that
easy to show that the last two terms cancel:  for instance here,
for $\mathbb{B}=b[\DD,b^*]$, we obtain by direct computation (using the
easiest translated Dirac operator $\DD'$)
\begin{align*}
\Tr \big(\mathbb{B}^2 \vert \DD'\vert^{-3-s}\big)&=\Tr \big((\mathbb{B}^*)^2 \vert
\DD'\vert^{-3-s}\big)=\tfrac12 \Tr \big(\mathbb{B}\mathbb{B}^* \vert
\DD'\vert^{-3-s}\big)= \tfrac12 \Tr \big(\mathbb{B}^*\mathbb{B} \vert
\DD'\vert^{-3-s}\big) \\
&=\tfrac 43\sum_{2j \in \N} \tfrac{j+1}{(2j+1)^{2+s}}\,,
\end{align*}
so $\ncint \mathbb{B}^2 \vert \DD' \vert^{-3}=\tfrac23$. Similarly, one checks
that $\ncint  (\mathbb{B}F)^2 \vert \DD \vert^{-3}= \tfrac 12 \ncint \mathbb{B}F\mathbb{B}^*F
\vert \DD \vert ^{-3}=-\tfrac 29$. Thus if $\Abb:=\mathbb{B}+\mathbb{B}^*$, $\ncint \Abb^2
\vert \DD\vert^{-3}=\ncint \Abb^2 \vert \DD'\vert^{-3}=4$ and $\ncint
(\Abb F)^2 \vert \DD \vert^{-3}=-\tfrac 43$ which yields
\eqref{curvature}.

Thus for any one-form $\Abb$ on the 3-sphere,
\begin{align*}
 \quad\SS (\DD +\Abb, \Phi, \Lambda) \, = \,2\,\Phi_{3}\,\Lambda^{3}   -
\tfrac 12 \,\Phi_1\, \Lambda^1+\mathcal{O}(\Lambda^{-1},\Abb)
\end{align*}
which as \eqref{spectralactioncomm} is not identical to
\eqref{strange}  which contains a nonzero constant term $\Lambda^0$
for $q=1$.

\section{Appendix}

\subsection*{A. Proof of Lemma \ref{ncintLin}}

$(i)$ Using same notations of Lemma \ref{technicalxdy}, we obtain by
definition of $\tau_1$,
\begin{align}
\tau_1\big(\pi_+(t_{k,p})\big) &= \delta_{k,0} \,\delta_{p,0}\
\delta_{{\a_1}+{\a_2}-{\a_3} +{\beta_1} +{\beta_2}-{\beta_3},0} \, ,
\label{tau1t} \\
\tau_1\big(\pi_-(u_{k,p})\big) &= \delta_{\wt k,0} \,\delta_{\wt
p,0}\  \delta_{{\a_1}-{\a_2}+{\a_3} +{\beta_1}
-{\beta_2}+{\beta_3},0}\, .\label{tau1u}
\end{align}
We get  $\tau_1\big(\pi_+(t_{k,p})\big)
\,\tau_1\big(\pi_-(u_{k,p})\big) = \delta_{k,0}\,\delta_{p,0}\,
\delta_{{\a_2},0} \,\delta_{{\a_3},0} \,
\delta_{{\beta_2},0} \, \delta_{{\beta_3},0} \,
\delta_{{\a_1},-{\beta_1}}$,
so Lemma  \ref{technicalxdy} gives the result.

$(ii)$ Since $\pi_+(t_{k,p})\eps_n = q^+_{k,p,n} \eps_{n+ r^+_{k,p}}$
and $\pi_-(u_{k,p})\eps_n = q^-_{k,p,n} \eps_{n+ r^-_{k,p}}$, we get,
\begin{align}
\tau_0\big(\pi_+(t_{k,p})\big)&=\delta_{r^+_{k,p},0}\,
\sum_{n=0}^\infty \big(q^+_{k,p,n}-\, \delta_{k,0} \, \delta_{p,0}\,
\delta_{{\a_1}+{\a_2}-{\a_3} +{\beta_1} +{\beta_2}-{\beta_3},0}\big)
\, ,\label{tau0tkp}\\
\tau_0\big(\pi_-(u_{k,p})\big)&=\delta_{r^-_{k,p},0}\,
\sum_{n=0}^\infty \big(q^-_{k,p,n}-\, \delta_{\wt k,0}\, \delta_{\wt
p,0} \,
\delta_{{\a_1}-{\a_2}+{\a_3} +{\beta_1} -{\beta_2}+{\beta_3},0}\big)
\, .\label{tau0ukp}
\end{align}
With (\ref{tau1t}) and (\ref{tau0ukp}) we get
\begin{align*}
\tau_1\big(\pi_+(t_{k,p})\big) \, \tau_0\big(\pi_-(u_{k,p})\big) &=
\delta_{k,0}\,\delta_{p,0}\,
 \delta_{{\a_2}+{\beta_2},{\a_3}+{\beta_3}}
\,\delta_{{\a_1},-{\beta_1}} \sum_{n=0}^\infty \big( \delta_{k,0}
\,\delta_{p,0}\, q^-_{k,p,n}-\delta_{\a_3+\beta_3,0}\big)\\ &=
\delta_{k,0}\,\delta_{p,0}\,
 \delta_{{\a_2}+{\beta_2},{\a_3}+{\beta_3}}
\,\delta_{{\a_1},-{\beta_1}} w_1(\beta_1,\a_3+\beta_3).
 \end{align*}

Using (\ref{tau1u}) and (\ref{tau0tkp}),
\begin{align*}
\tau_0\big(\pi_+(t_{k,p})\big) \, \tau_1\big(\pi_-(u_{k,p})\big) &=
\delta_{\wt k,0}\,\delta_{\wt p,0}\,
 \delta_{{\a_2}+{\beta_2},{\a_3}+{\beta_3}}
\,\delta_{{\a_1},-{\beta_1}} \sum_{n=0}^\infty \big( \delta_{\wt k,0}
\,\delta_{\wt p,0}\, q^+_{k,p,n}-\delta_{\a_3+\beta_3,0}\big)\\ &=
\delta_{\wt k,0}\,\delta_{\wt p,0}\,
 \delta_{{\a_2}+{\beta_2},{\a_3}+{\beta_3}}
\,\delta_{{\a_1},-{\beta_1}} w_1(\beta_1,\a_3+\beta_3).
 \end{align*}
Lemma \ref{technicalxdy} yields the result.

\subsection*{B. Proof of Lemma \ref{ncintA2}}
We have
\begin{align}
\tau_1(\pi_+(t_{K,P})) &= \delta_{K,0}\,\delta_{P,0}\,
\delta_{A_1+A_2-A_3+B_1+B_2-B_3,0}\label{tau1tKP}\,,\\
\tau_1(\pi_-(u_{ K, P})) &= \delta_{\wt K,0}\,\delta_{\wt P,0}\,
 \delta_{A_1-A_2+A_3+B_1-B_2+B_3,0}\label{tau1uKP} \,.
\end{align}
and
 \begin{align}
\tau_0(\pi_+(t_{K,P})) &= \delta_{r^+_{K,P},0} \sum_{n=0}^\infty
\big(q^+_{K,P,n}  -\delta_{K,0}\,\delta_{P,0} \,
\delta_{A_1+A_2-A_3+B_1+B_2-B_3,0}\big)\label{tau0tKP}\,,\\
\tau_0(\pi_-(u_{ K, P})) &=  \delta_{r^-_{K,P},0} \sum_{n=0}^\infty
\big(q^-_{K,P,n}  -\delta_{\wt K,0} \, \delta_{\wt P,0} \,
 \delta_{A_1-A_2+A_3+B_1-B_2+B_3,0}\big)\label{tau0uKP}\,.
\end{align}
$(i)$ Equations (\ref{tau1tKP}) and (\ref{tau1uKP}) give $(\tau_1
\otimes \tau_1)\, r(\ul A\, \ul A')^0
=\delta_{A_1,-B_1}\delta_{A_2,0}\delta_{A_3,0}\delta_{B_2,0}\delta_{B_3,0}\,
\la_{0,0}$. A computation of $v_{0,0}$ with
$\delta_{A_1,-B_1}\,\delta_{A_2,0}\,\delta_{A_3,0}\,\delta_{B_2,0}\,\delta_{B_3,0}=1$
gives the result.

$(ii)$ Equations (\ref{tau1tKP}) and (\ref{tau0uKP}) yield
\begin{align*}
\tau_1(\pi_+(t_{K,P}))\,\tau_0(\pi_-(u_{K,P}))&=\,
\delta_{K,0}\,\delta_{P,0} \,\delta_{A_2+B_2,A_3+B_3}
\,\delta_{A_1,-B_1} \\ &\hspace{2cm}\times
v_{\beta_1,\a'_1,\beta'_1}((\a_2+\beta_2+\a_3+\beta_3)(\a'_1+\beta'_1),A_3+B_3).
\end{align*}

Equations (\ref{tau0tKP}) and (\ref{tau1uKP}) yield
\begin{align*}
\tau_0(\pi_+(t_{K,P}))\,\tau_1(\pi_-(u_{K,P}))&=\, \delta_{\wt
K,0}\,\delta_{\wt P,0} \,\delta_{A_2+B_2,A_3+B_3} \,\delta_{A_1,-B_1}
\\ &\hspace{2cm}\times
v_{\beta_1,\a'_1,\beta'_1}((\a_2+\beta_2+\a_3+\beta_3)(\a'_1+\beta'_1),A_3+B_3)
\end{align*}
and the result follows.

$(iii)$ With (\ref{rproduct}) a direct computation gives
\begin{align}
\tau_1(\pi_+(t_{K,P})) &= \delta_{K,0}\,\delta_{P,0}
\,\delta_{A_1+A_2-A_3+B_1+B_2-B_3,0}\label{tau1tKP3} \,,\\
\tau_1(\pi_-(u_{ K, P})) &= \delta_{\wt K,0}\,\delta_{\wt P,0}\,
 \delta_{A_1-A_2+A_3+B_1-B_2+B_3,0}\label{tau1uKP3}\,.
\end{align}
Using (\ref{tau1tKP3}) and (\ref{tau1uKP3}), $(\tau_1 \otimes
\tau_1)\, \big(r(\ul A\, \ul A'\,  \ul A'')^\circ \big)
=\delta_{A_1,-B_1}\,\delta_{A_2,0}\,\delta_{A_3,0}\,\delta_{B_2,0}\,\delta_{B_3,0}\,
v_{0,0}$. A computation of $v_{0,0}$ with
$\delta_{A_1,-B_1}\,\delta_{A_2,0}\,\delta_{A_3,0}\,\delta_{B_2,0}\,\delta_{B_3,0}=1$
gives the result.

$(iv)$ We have $\delta(M^{\a}_\beta) M^{\a'}_{\beta'} =
\delta(x)\delta(y)
x'\delta(y') + x\delta^2(y) x'\delta(y')$ where $x,x',y,y'$ are
monomials ($\ul \pi$ omitted). Since
\begin{align*}
\ul \pi(x) = \sum_{k}
\genfrac(){0pt}{1}{\a}{k} a_{+,\a_1}^{\wh k_1} a_{-,\a_1}^{k_1}
b_+^{\wh k_2} b_-^{k_2} {b_+^*}^{\wh k_3} {b_-^*}^{k_3}=:\sum_k
\genfrac(){0pt}{1}{\a}{k} c_k,
\end{align*}
we get $\delta(\ul \pi(x)) =\sum_{k}g(k) \genfrac(){0pt}{1}{\a}{k}
\,c_k$.

Similarly, $\delta(\ul \pi(y)) =\sum_{p}g(p)
\genfrac(){0pt}{1}{\beta}{p}\,c_p$ and $\delta^2(\ul \pi(y))
=\sum_{p}g(p)^2 \genfrac(){0pt}{1}{\beta}{p}c_p$.

Thus, with $c_{K,P} :=c_k \,c_p \,c_{k'} \,c_{p'} $,
\begin{align*}
&\delta(x)\delta(y) x'\delta(y') = \sum_{K,P}
g(k)\,g(p)\,g(p')\genfrac(){0pt}{1}{\a}{K}\genfrac(){0pt}{1}{\beta}{P}\,
c_{K,P}\,,\\
&x\delta^2(y) x'\delta(y') = \sum_{K,P}
g(p)^2g(p')\genfrac(){0pt}{1}{\a}{K}\genfrac(){0pt}{1}{\beta}{P}\,
c_{K,P}\,,\\
& r(\delta(M^{\a}_\beta) M^{\a'}_{\beta'})^0 =
\sum_{K,P}\delta_{h_{K,P},0}
\big(g(k)+g(p)\big)\, g(p)\,
g(p')\,\genfrac(){0pt}{1}{\a}{K}\genfrac(){0pt}{1}{\beta}{P}
\,r(c_{K,P})=: \sum_{K,P} \la_{K,P}\  r(c_{K,P})\, .
\end{align*}
Since
$r(c_k) = (-q)^{k_1} (-1)^{{\a_2}+{\a_3}}\pi_+(t_k) \otimes
\pi_-(u_k)$ with $t_k,u_k$ defined by
$$
t_k:=a_{\a_1}^{\wh k_1}\, b^{k_1}\,a^{\wh k_2}\,b^{k_2}\,{a^*}^{\wh
k_3}\,b^{k_3} \text{  and } u_k:=a_{\a_1}^{\wh k_1}\,b^{k_1}\,b^{\wh
k_2}\,{a^*}^{k_2}\,{b}^{\wh k_3}\,a^{k_3},
$$
we get
$$
r(\delta(M^{\a}_\beta) M^{\a'}_{\beta'})^0 = \sum_{K,P}\la_{K,P}\,
(-q)^{k_1+k'_1+p_1+p'_1} (-1)^{A_2+A_3+B_2+B_3}\pi_+(t_{K,P}) \otimes
\pi_-(u_{K,P})
$$
where $t_{K,P}=t_k t_p t_{k'} t_{p'}$ and  $u_{K,P}=u_k u_p u_{k'}
u_{p'}$. Direct computations yield
\begin{align*}
\tau_1 \big(\pi_+(t_{K,P})\big) &= \delta_{K,0}\,\delta_{P,0}\,
\delta_{A_1+A_2-A_3+B_1+B_2-B_3,0}\,,\\
\tau_1 \big(\pi_-(u_{ K, P})\big) &= \delta_{\wt K,0}\,\delta_{\wt
P,0}\,
 \delta_{A_1-A_2+A_3+B_1-B_2+B_3,0}\,.
\end{align*}
The result follows.

$(v)$ For the last equality, note that by $(iv)$
$$
\ncint \delta(A) A |\DD|^{-3} = -2 \sum_{\a_1,\a'_1,\beta_1,\beta'_1}
(\a'_1+\beta'_1)
{\beta_1}\beta'_1\, A^{\beta_1 0 0}_{\a_1 0 0 0}\,A^{\beta'_1 0
0}_{\a'_1 0 0}\, \delta_{\a_1+\a'_1+\beta_1+\beta'_1,0}.
$$
The following change of variables $\a_1 \leftrightarrow \a'_1$,
$\beta_1\leftrightarrow \beta'_1$, implies by symmetry that this is
equal to zero.

\subsection*{C. Proof of Lemma \ref{ncintJlem}}

$(i)$ Following notations of Lemma \ref{technicalxdy}, we have
$$
 M^{\a}_\beta J M^{\a'}_{\beta'} J^{-1}= \sum_{K,P} v_{K,P}
\,c_{k,p}Jc_{k',p'}J^{-1}
$$
where $K=(k,k')$, $P=(p,p')$, $\la_{K,P}=
g(p)g(p')v_{k}v_{k'}w_pw_{p'}$. Thus,
$$
 \wt\rho(M^{\a}_\beta J M^{\a'}_{\beta'} J^{-1})=
(-1)^{A_2+A_3+B_2+B_3}\sum_{K,P}(-q)^{k_1+k'_1+p_1+p'_1} \la_{K,P} \,
T^+_{K,P}\otimes T^-_{K,P}
$$
where $T^+_{K,P}:=\pi'_+(t_k t_p)\wh \pi_{+}(t_{k'} t_{p'})$
 and $T^-_{K,P}:=\pi'_-(u_k u_p)\wh \pi_{-}(u_{k'} u_{p'})$ with
\begin{align*}
t_k&:=a_{\a_1}^{\wh k_1}\,{b^{*}}_{\a_1}^{k_1}\,
 a^{\wh k_2}\, {b^*}^{k_2}\, {a^*}^{\wh k_3} b^{k_3}, \\
u_k&:=a_{\a_1}^{\wh k_1}\,{b^{*}}_{\a_1}^{k_1}\,
 b^{\wh k_2}\, {a^*}^{k_2}\, {b^*}^{\wh k_3} a^{k_3} .
\end{align*}
A direct computation leads to
\begin{align*}
\tau_1(T^+_{K,P})&=\delta_{K,0} \,\delta_{P,0}\
\delta_{{\a_1}+{\a_2}-{\a_3} +{\beta_1} +{\beta_2}-{\beta_3},0} \,
\delta_{{\a'_1}+{\a'_2}-{\a'_3} +{\beta'_1}
+{\beta'_2}-{\beta'_3},0}\, ,\\
\tau_1(T^-_{K,P})&=\delta_{\wt K,0} \,\delta_{\wt
P,0}\  \delta_{{\a_1}-{\a_2}+{\a_3} +{\beta_1}
-{\beta_2}+{\beta_3},0} \, \delta_{{\a'_1}-{\a'_2}+{\a'_3} +{\beta'_1}
-{\beta'_2}+{\beta'_3},0}
\end{align*}
which gives the result.

$(ii)$
Using the commutation relations on $\A$, we see that there are real
functions of $(K,P)$, denoted $\sigma^t_{K,P}$ and $\sigma^u_{K,P}$
such that
\begin{align*}
T^+_{K,P} &= q^{\sigma_{K,P}^t}\, \pi'_+(t_{k,p})\wh \pi_+(t_{k',p'}),
\\
T^-_{K,P} &= q^{\sigma_{K,P}^u}\, \pi'_-(u_{k,p}) \wh
\pi_-(u_{k',p'}),\\
t_{k,p}&:= a_{\a_1}^{\wh k_1}\,
 a^{\wh k_2}\,  {a^*}^{\wh k_3}\,a_{\beta_1}^{\wh p_1}\,
 a^{\wh p_2}\,  {a^*}^{\wh p_3}\,
{b^{*}}_{\a_1}^{k_1}\,{b^{*}}_{\beta_1}^{p_1}\,
{b^*}^{k_2+p_2}b^{k_3+p_3},\\
u_{k,p}&:=a_{\a_1}^{\wh k_1}\,
 {a^*}^{ k_2}\,  {a}^{k_3}\,a_{\beta_1}^{\wh p_1}\,
 {a^*}^{ p_2}\,  {a}^{p_3}\,
{b^{*}}_{\a_1}^{k_1}\,{b^{*}}_{\beta_1}^{p_1}\, {b}^{\wh k_2+\wh
p_2}{b^*}^{\wh k_3+\wh p_3}.
\end{align*}

We have, under the hypothesis $\tau_1(T^-_{K,P})=1$,
\begin{align*}
&\wh \pi_+(t_{k',p'})\eps_{m,2j} =
(-1)^{\la'}q^{(2j-m)\la'}\,q^{\up_{\a'_1}}_{2j-m-s+\beta'_1,|\a'_1|}\,
q^{\up_{\beta'_1}}_{2j-m-s,|\beta'_1|}\,\eps_{m+s,2j}\, ,\\
&s:=- \a'_2 + \a'_3 -\beta'_2+\beta'_3=\a'_1+\beta'_1 \,,\\
&\la':=\a'_2+\a'_3+\beta'_2+\beta'_3\,,\\
&\la:=\a_2+\a_3+\beta_2+\beta_3 \,\\
&\tau_1(T^+_{K,P})=  \delta_{\la,0}\, \delta_{\la',0}\, .
\end{align*}
and then,
\begin{align*}
&(T^+_{K,P})_{m,2j}=q^{\sigma^t_{K,P}+s\la}(-1)^{\la'}
q^{(2j-m)\la'+m\la}\,F_m \,F'_{2j-m}\,
\delta_{A_1+B_1,0}\,,\\
&F'_{2j-m}:=q^{\up_{\a'_1}}_{2j-m-\a'_1,|\a'_1|}\,
q^{\up_{\beta'_1}}_{2j-m-\a'_1-\beta'_1,|\beta'_1|}\,,\\
&F_{m}:=
q^{\up_{\a_1}}_{m-\a_1,|\a_1|}q^{\up_{\beta_1}}_{m-\beta_1-\a_1,|\beta_1|}\,.
\end{align*}
Following the proof of Lemma \ref{tau0ext}, we see that
$\tau_0(T^+_{K,P})$ is possibly nonzero only in the two cases
$\la'=0$ or $\la=0$.

Suppose first $\la=\la'=0$. In that case, we have
$$
\tau_0(T^+_{K,P})= \lim_{2j\to\infty} \sum_{m=0}^{2j}
\big((q^{\up_{\beta_1}}_{m,|\beta_1|}
q^{\up_{\beta'_1}}_{2j-m,|\beta'_1|})^2 -1\big) =
\sum_{m=0}^{\infty}\big((q^{\up_{\beta_1}}_{m,|\beta_1|})^2 -1\big) +
\sum_{m=0}^{\infty}\big((q^{\up_{\beta'_1}}_{m,|\beta'_1|})^2 -1\big)
$$
where the second equality comes from Lemma \ref{tech-separation}.

In the case $(\la=0,\la'>0)$, we get $\a'_1=-\beta'_1$ and thus,
$$
(T^+_{K,P})_{m,2j}=q^{\sigma^t_{K,P}} q^{m\la}\,
(q^{\up_{\beta_1}}_{m,|\beta_1|}
q^{\up_{\beta'_1}}_{2j-m,|\beta'_1|})^2
\delta_{\a_1+\beta_1,0}.
$$
Let us note $U_{2j} = \sum_{m=0}^{2j} q^{m\la}\,
(q^{\up_{\beta_1}}_{m,|\beta_1|}
q^{\up_{\beta'_1}}_{2j-m,|\beta'_1|})^2 $
and $L_{2j} = \sum_{m=0}^{2j} q^{m\la}\,
(q^{\up_{\beta_1}}_{m,|\beta_1|})^2$.

Suppose $\beta'_1>0$.
Since $(q^{\up_{\beta'_1}}_{2j-m,|\beta'_1|})^2-1 = \sum_{|p|_1 \neq
0,
p_i\in \set{0,1}} (-1)^{|p|_1}q^{r_p}\, q^{2(2j-m)|p|_1}$ where we
have
$r_p=2+\cdots +2\beta'_1$.
As in the proof of Lemma \ref{tau0ext} $(ii)$, we can conclude that
$U_{2j}-L_{2j}$ converges to 0.
The case $\beta'_1\leq 0$ is similar.

In the other case $(\la>0,\la'=0)$, the
arguments are the same, replacing $\la$ by $\la'$ and $\a_1$,
$\beta_1$ by $\a'_1$,
$\beta'_1$. Finally,
\begin{align*}
&\tau_0(T^+_{K,P})\tau_1(T^-_{K,P}) = \delta_{\wt K,0}\,\delta_{\wt
P,0}\,\delta_{\a_1,-\beta_1}\,\delta_{\a'_1,-\beta'_1}
(\delta_{\la',0}\,\delta_{\a_2+\beta_2,\a_3+\beta_3}\,
s_{\a,\beta}+\delta_{\la,0}\,\delta_{\a'_2+\beta'_2,\a'_3+\beta'_3}
\,s_{\a',\beta'}),\\
&s_{\a\beta}:= q^{\beta_1(\a_3-\a_2)}\,\sum_{m=0}^{\infty}\big(
q^{m\la}\, (q^{\up_{\beta_1}}_{m,|\beta_1|})^2
-\delta_{\la,0}\big)\, .
\end{align*}
A similar computation of $\tau_0(T^-_{K,P})$ can be done following
the same arguments. We find
eventually
$$
\tau_1(T^+_{K,P})\tau_0(T^-_{K,P}) = \delta_{ K,0}\,\delta_{
P,0}\,\delta_{\a_1,-\beta_1}\,\delta_{\a'_1,-\beta'_1}
(\delta_{\la',0}\,\delta_{\a_2+\beta_2,\a_3+\beta_3}\,
s_{\a,\beta}+\delta_{\la,0}\,\delta_{\a'_2+\beta'_2,\a'_3+\beta'_3}
\,s_{\a',\beta'})
$$
and the result follows.

$(iii)$ The same arguments of $(i)$ apply here with minor changes.

$(iv)$ follows from a slight modification of the proof of Lemma
\ref{ncintA2} $(iv)$.

$(v)$ is a straightforward consequence of $(i,ii,iii,iv)$. 

\subsection*{D. Proof of Lemma \ref{tech-separation}}

We give a proof for $\beta$ and $\beta'>0$, the other
cases being similar. 

Since
$(q^{\up_{\beta}}_{m,|\beta|})^2= \sum_{p_i\in \set{0,1}}
(-1)^{|p|_1}q^{r_p} q^{2|p|_1 m}$ where
$p=(p_1,\cdots,p_{\beta})$ and $r_p:=2(p_1+\cdots+\beta p_{\beta})$,
we get, with the notations
$\la_{p,p'}:=(-1)^{|p+p'|_1}q^{r_p+r_{p'}}$ and $U_{2j}:=
\sum_{m=0}^{2j}\,(q^{\up_{\beta}}_{m,|\beta|}q^{\up_{\beta'}}_{2j-m,|\beta'|})^2
-1$,
\begin{align*}
 U_{2j}&=
\sum_{m=0}^{2j}\sum_{|p+p'|_1>0} \la_{p,p'}
q^{2|p|_1 m +2|p'|_1(2j-m)}\\
&= \sum_{|p|_1\geq |p'|_1, |p|_1>0} \la_{p,p'} V_{2j,p,p'} +
\sum_{|p|_1< |p'|_1, |p'|_1>0}\la_{p,p'}
V'_{2j,p,p'}
\end{align*}
where
$$
V_{2j,p,p'}=q^{4j|p'|_1}\sum_{m=0}^{2j}
q^{2(|p|_1-|p'|_1) m}, \qquad V'_{2j,p,p'}=q^{4j|p|_1}\sum_{m=0}^{2j}
q^{2(|p'|_1-|p|_1) m}.
$$
It is clear that $V_{2j,p,p'}$ has $0$ for limit when $j\to \infty$
when $|p'|_1>0$, and
$V'_{2j,p,p'}$ has $0$ for limit when $j\to \infty$ when $|p|_1>0$.
As a consequence,
\begin{align*}
 U_{2j}&= \sum_{|p|_1>0} \la_{p,0} V_{2j,p,0} + \sum_{
|p'|_1>0}\la_{0,p'}
V'_{2j,0,p'} + o(1).
\end{align*}
The result follows as

$\sum_{m=0}^{2j}\big((q^{\up_{\beta}}_{m,|\beta|})^2 -1\big) =
\sum_{|p|_1>0} \la_{p,0} V_{2j,p,0}$ and
$\sum_{m=0}^{2j}\big((q^{\up_{\beta'}}_{m,|\beta'|})^2 -1\big)
=\sum_{ |p'|_1>0}\la_{0,p'}
V'_{2j,0,p'}$.

\chapter{Tadpoles and commutative spectral triples}

\section{Introduction}

The history of the noncommutative residue is now rather long
\cite{Kassel}, so we sketch it only briefly: after some approaches by
Adler \cite{Adler} and  Manin \cite{Manin} on the Korteweg-de Vries
equation using a trace on the algebra of formal pseudodifferential
operators in one dimension, and of Guillemin with his "soft" proof of
Weyl's law on the eigenvalues of an elliptic operator
\cite{Guillemin}, the noncommutative residue  in any dimension was essentially initiated par Wodzicki in his thesis \cite{Wodzicki2}. This residue gives, up to a multiplicative factor, the unique non-trivial trace on the algebra of pseudodifferential operators. Then, a link between this residue and the Dixmier trace was given by Connes in \cite{Connesaction}. Thanks to Connes again \cite{Book,Cgeom}, the setting of classical pseudodifferential operators on Riemannian
manifolds without boundary was extended to a noncommutative geometry
where the manifold is replaced by a not necessarily commutative
algebra $\A$ plus a Dirac-like operator $\DD$ via the notion of
spectral triple $(\A,\,\H,\,\DD)$ where $\H$ is the Hilbert space
acted upon by $\A$ and $\DD$. The previous Dixmier trace is extended to the algebra of pseudodifferential operators naturally associated to the triple $(\A,\,\H,\,\DD)$. This spectral point of view appears quite natural in the general framework of noncommutative geometry which goes beyond Riemannian geometry. From a physicist's point of view, this framework has many advantages: the spectral approach is motivated by quantum physics but not only since classical observables and infinitesimals are now on the same footing and even Dixmier's trace is related to renormalization.  It is amazing to observe that most of classical geometrical notions like those defined in relativity or particle physics can be extended in
this really noncommutative setting. Among others, some physical
actions still make sense as in \cite{Connesaction}  where Dixmier's trace is used to compute the Yang--Mills action in the context of noncommutative differential geometry. Another example is the Einstein--Hilbert action: on a compact
spin Riemannian 4-manifold, $\ncint \DD^{-2}$ coincides (up to a
universal scalar) with the Einstein--Hilbert action, where $\ncint$ is
precisely the noncommutative residue, a point first noticed by Connes; then, there was a brute force proof \cite{Kastler} and generalization \cite{KW} (see also \cite{Ack}) of this fact which is particularly relevant here.

Since then, the case of compact manifolds with
boundary has been studied, making clearer the links between
noncommutative residues, Dixmier's trace and heat kernel expansion. This was achieved 
using Boutet de Monvel's algebra \cite{FGLS, Schrohe, GSc}, in the
case of conical singularities \cite{Schrohe1, Lescure} or when the
symbols are log-polyhomogeneous \cite{Lesch}. Besides, there are some applications of noncommutative residues for such manifolds to classical gravity \cite{Wang} and to the unification of gravity with fundamental interactions \cite{CC2}.
Needless to say that in field theory, the
one-loop divergencies, anomalies and different asymptotics
of the effective action are directly obtained from the heat kernel
method \cite{V}, so all of the above quoted mathematical results have
profound applications to physics. 

\medskip

We are interested in possible cancellation of terms in the Chamseddine--Connes spectral action formula 
\eqref{formuleaction}. We focus
essentially on commutative spectral triples, for which we show that there
are no tadpoles (see Definition \ref{Deftadpole}). In particular, terms like $\ncint A\DD^{-1}$ are zero: in field
theory, $\DD^{-1}$ is the Feynman propagator and $A\DD^{-1}$ is a
one-loop graph with fermionic internal line and only one external
bosonic line $A$ looking like a tadpole. More generally, the tadpoles are the $A$-linear terms in \eqref{formuleaction}.
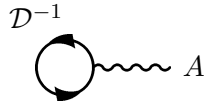
\begin{fmffile}{tadpolegraph}
\[
{\begin{picture}(100,60)
\put(0,0){\begin{fmfgraph}(50,60)
\fmfleft{l}
\fmfright{r}
\fmffreeze
\fmf{photon,tension=3}{r,w}
\fmf{photon,tension=3}{w,v}
\fmf{fermion,right}{v,l}
\fmf{fermion,right}{l,v}
\end{fmfgraph}}
\put(-10,45){\mbox{${\cal D}^{-1}$}}
\put(55,27){\mbox{$A$}}
\end{picture}} 
\]
\end{fmffile}
\vspace{-1cm}

In \cite{Ponge}, computations of $\ncint \vert \DD \vert^{-k}$ for some values of $k$ are presented and formula like \eqref{termconstanttilde} also appear in \cite{LP} in the context of pseudodifferential elliptic operators.
As a starting point, we investigate in section \ref{tadboundsec} the existence of tadpoles for manifolds with boundaries, considering following Chamseddine and Connes \cite{CC2} the case of a chiral boundary condition on the Dirac operator. One of their original motivations was to show that the first two terms in the spectral action come with the right ratio and sign for their coefficients as in the modified Euclidean action used in gravitation. We generalize this approach to the perturbed Dirac operator by an internal fluctuation, ending up with no tadpoles up to order 5.

However, this approach stems from explicit computations of first heat kernel coefficients, so we cannot conclude that other integrals of the same type as tadpoles are zero. It is then natural to restrict to manifolds without  boundary via a different method. 

After some useful facts using the link between, we conclude in section \ref{commtadsection}, and using the results of section \ref{tadsection} and pseudodifferential techniques, that a lot of terms in \eqref{formuleaction} are zero.

\section{Tadpoles and compact spin manifolds with boundary}\label{tadboundsec}

Let $M$ be a smooth compact Riemannian $d$-dimensional manifold with
smooth boundary $\del M$ and let $V$ be a given smooth vector bundle on
$M$. We denote $dx$ (resp. $dy$) the Riemannian volume form on $M$ (resp. on $\del M$).

Recall that a differential operator $P$ is of Laplace type if it has locally the form
\begin{equation}
P = - (g^{\mu\nu} \del_\mu \del_\nu + \mathbb{A}^\mu\del_\mu +\mathbb{B})  \label{Lapl}
\end{equation}
where $(g^{\mu\nu})_{1\leq \mu,\nu\leq d}$ is the inverse matrix associated
to the metric $g$ on $M$, and $\mathbb{A}^\mu$ and $\mathbb{B}$ are smooth
$L(V)$-sections on $M$ (endomorphisms). A differential operator $D$ is of Dirac type if $D^2$ is of Laplace type, or equivalently if it has locally the following form
$$
D = -i \ga^\mu \del_\mu + \phi
$$
where $(\ga^\mu)_{1\leq \mu\leq d}$ gives $V$ a Clifford module structure: $\set{\ga^\mu,\ga^\nu}=2g^{\mu\nu}\Id_V$, ${(\ga^\mu)}^*=\ga^\mu$.

A particular case of Dirac operator is given by the following formula 
\begin{equation}
D = -i\ga^\mu (\del_\mu + \om_\mu) \label{phiDirac}
\end{equation}
where the $\om_\mu$ are in $C^\infty\big(L(V)\big)$.

If $P$ is a Laplace type operator of the form (\ref{Lapl}), then (see \cite[Lemma 1.2.1]{Gilkey2})
there is an unique connection $\nabla$ on $V$ and an unique
endomorphism $E$ such that $P = L(\nabla,E)$ where by definition
\begin{align*}
&L(\nabla,E) :=  -(\Tr_g \nabla^2  + E), \quad \nabla^2(X,Y):= [\nabla_X,\nabla_Y] -\nabla_{\nabla^{g}_X Y} \, ,
\end{align*}
$X,Y$ are vector fields on $M$ and $\nabla^g$ is the Levi-Civita connection on $M$. Locally 
$$
\Tr_g \nabla^2 := g^{\mu\nu}(\nabla_\mu \nabla_\nu -\Ga^{\rho}_{\mu\nu} \nabla_\rho)
$$
where $\Ga^{\rho}_{\mu \nu}$ are the Christoffel coefficients of $\nabla^g$.
Moreover (with local frames of $T^*M$ and $V$), $\nabla =
dx^\mu\ox (\del_\mu +\om_\mu)$ and $E$ are related to $g^{\mu\nu}$,
$\mathbb{A}^\mu$ and $\mathbb{B}$ through
\begin{align}
\om_\nu&=  \half g_{\nu\mu}(\mathbb{A}^\mu +g^{\sg\eps} \Ga_{\sg \eps}^{\mu}\Id )
\label{omeganu}\, ,\\
E&=  \mathbb{B}-g^{\nu\mu}(\del_{\nu} \om_\mu +\om_\nu\om_\mu -\om_\sg
\Ga_{\nu\mu}^\sg ) \label{EEquation}  \, .
\end{align}

Suppose that $P=L(\nabla,E)$ is a Laplace type operator on $M$, and assume that $\chi$ is an endomorphism of $V_{\del M}$ so that
$\chi^2=\Id_V$. We extend $\chi$ on a collar neighborhood $\CC$ of $\del M$
in $M$ with the condition $\nabla_d \,(\chi) = 0$ where the
$d^{th}$-coordinate here is the radial coordinate (the geodesic
distance of a point in $M$ to the boundary $\del M$).

Let $V_\pm:=\Pi_{\pm} V$ be the sub-bundles of $V$ on $\CC$ where $\Pi_\pm:=\half(\Id_V\pm \chi)$ are the projections on the $\pm1$ eigenvalues of $\chi$.
We also fix an auxiliary endomorphism $S$ on ${V_+}_{\del M}$ extended to $\CC$.

This allows to define the mixed boundary operator $\B=\B(\chi,S)$ as
\begin{equation}
\B s := \Pi_+ (\nabla_{d} +S)\Pi_+ s_{ |\del M} \oplus
\Pi_{-}s_{|\del M}\,, \quad s\in C^\infty(V) .\label{Mixed}
\end{equation}
These boundary conditions generalize Dirichlet ($\Pi_-=\Id_V$) and Neumann--Robin ($\Pi_+=\Id_V$) conditions.

We define $P_\B$ as the realization of $P$ on $\B$, that is to say the closure of $P$ defined
on the space of smooth sections of $V$ satisfying the boundary condition $\B s=0$.
\medskip

We are interested in the behavior of the heat kernel coefficients $a_{d-n}$ defined through its expansion as $\Lambda\to \infty$ (see \cite[Theorem 1.4.5]{Gilkey2})
$$
\Tr(e^{-\Lambda^{-2} D^2_\B}) \sim \sum_{n\geq 0} \Lambda^{d-n}\,
a_{d-n}(D,\B)
$$ 
where $D$ is a self-adjoint Dirac type operator. Moreover, we will use a perturbation $D\to D+A$, where $A$ is a 1-form (a linear combination of terms of the type $f[D,g]$, where $f$ and $g$ are smooth functions on $M$). More precisely, we investigate the linear
dependence of these coefficients with respect to $A$. It is clear
that, since $A$ is a differential operator of order 0, a perturbation
$D\mapsto D+A$ transforms a Dirac type operator into another Dirac
type operator. 

This perturbation has consequences on the $E$ and $\nabla$ terms:

\begin{lemma}
\label{Perturbation}
Let $D$ be a Dirac type operator locally of the form
(\ref{phiDirac}) such that $\nabla_\mu:=\partial_\mu + \om_\mu$ is
connection compatible with the Clifford action $\ga$. Let $A$ be a 1-form associated to $D$, so that $A$ is locally of the form $-i\ga^\mu a_\mu$ with $a_\mu \in C^\infty(U)$, $(U,x_\mu)$ being a local coordinate frame on $M$.  

Then $(D+A)^2=L(\nabla^A,E^A)$ and $D^2=L(\nabla,E)$ where, 
\begin{align*}
&\om^A_{\mu}= \om_\mu + a_\mu\, ,\,\,  \text{thus }
\nabla^A_\mu=\nabla_\mu+a_\mu \,\text{Id}_V,\\
&E^A = E +\tfrac{1}{4}[\ga^\mu, \ga^\nu]F_{\mu \nu} , \quad E=\tfrac{1}{2}\ga^\mu \ga^\nu [\nabla_{\mu}, \nabla_\nu] , \quad F_{\mu\nu}:=\partial_\mu(a_\nu)-\partial_\nu(a_\mu)
\end{align*}
Moreover, the curvature of the connection $\nabla^A\,$is
$\Omega_{\mu \nu}^{A}=\Omega_{\mu \nu}+F_{\mu \nu}$, where $\Omega_{\mu\nu}=[\nabla_\mu,\nabla_\nu]$.

In particular $\Tr E^A=\Tr E$.
\end{lemma}

\begin{proof}
This is quoted in \cite[equation (3.27)]{V}.

$(D+A)^2=L(\nabla^A,E^A):=-g^{\mu\nu}(\nabla^A_\mu \nabla^A_\nu-\Gamma_{\mu\nu}^\rho \nabla^A_\rho)-E^A$ and we get with $\nabla^A_\mu:=\nabla_\mu+a_\mu \text{Id}_V$:
\begin{align}
-(D+A)^2&=\ga^\mu\nabla^A_\mu\ga^\nu\nabla^A_\nu=\ga^\mu[\nabla^A_\mu,\ga^\nu]\nabla^A_\nu+\ga^\mu\ga^\nu \nabla^A_\mu \nabla^A_\nu \nonumber \\
&=\ga^\mu[\nabla_\mu,\ga^\nu]\nabla^A_\nu+\tfrac{1}{2}(\ga^\mu\ga^\nu+\ga^\nu\ga^\nu)\nabla^A_\mu\nabla^A_\nu + \tfrac{1}{2}\ga^\mu \ga^\nu[\nabla^A_\mu,\nabla^A_\nu] \nonumber \\
&=-\ga^\mu\ga^\rho{\Gamma_{\mu \rho}}^\nu \nabla^A_\nu+g^{\mu
\nu}\nabla^A_\mu\nabla^A_\nu
+\tfrac{1}{2}\ga^\mu \ga^\nu[\nabla_\mu +a_\mu Id_V,\nabla_\nu +a_\nu \text{Id}_V]. \label{nabla}
\end{align}
Since $\Gamma_{\mu\nu}^\rho=\Gamma_{\nu\mu}^\rho$, we get by comparison, 
\begin{align*}
E^A&=\tfrac{1}{2}\ga^\mu \ga^\nu[\nabla_\mu +a_\mu \,\text{Id}_V,\nabla_\nu
+a_\nu \,\text{Id}_V]=\tfrac{1}{2}\ga^\mu \ga^\nu \big([\nabla_\mu ,\nabla_\nu
]+\partial_\mu(a_\nu)-\partial_\nu(a_\mu)\big) \\
&=\tfrac{1}{2}\ga^\mu \ga^\nu [\nabla_\mu ,\nabla_\nu
] +\tfrac{1}{4}[\ga^\mu, \ga^\nu] \big(\partial_\mu(a_\nu)-\partial_\nu(a_\mu)\big).
\tag*{\qed}
\end{align*}
\hideqed
\end{proof}
Remark that even if quadratic terms in $A^2$ appear in the local presentation of the perturbation $D^2\to (D+A)^2$ (in the $b$ term), these terms do not appear in the invariant formulation $(\nabla,E)$
since they are hidden in $\nabla^A_\mu\nabla^A_\nu$ of \eqref{nabla}. 

\medskip

In the following, $D$ and $A$ are fixed and satisfy the hypothesis of Lemma \ref{Perturbation}. Indices $i$, $j$, $k$, and $l$ range from 1 through
the dimension $d$ of the manifold and index a local orthonormal frame
$\{ e_1,...,e_d\}$ for the tangent bundle. Roman indices
$a$, $b$, $c$, range from 1 through $d-1$ and index a local orthonormal
frame for the tangent bundle of the boundary
$\partial M$. The vector field $e_d$ is chosen to be the inward-pointing unit normal
vector field. Greek indices are associated to coordinate frames. 

Let $R_{ijkl}$, $\rho_{ij}:=R_{ikkj}$ and $\tau:=\rho_{ii}$ be respectively the components of the Riemann tensor, Ricci tensor and scalar curvature of the Levi-Civita connection. 
Let $L_{ab}:=(\nabla_{e_{a}}e_{b},e_{d})$ be the second fundamental form of the hypersurface $\partial M$ in $M$. 
Let
``;''  denote multiple covariant differentiations with respect to $\nabla^{A}$ and 
``:'' denote multiple covariant differentiations with respect to $\nabla$ and the Levi-Civita
connection of $M$.

We will look at a chiral boundary condition. This is a mixed boundary condition natural to consider in order to preserve the existence of chirality on $M$ and its boundary $\del M$ which are compatible with the (selfadjoint) Clifford action: we assume that the operator $\chi$ is selfadjoint and satisfies the following relations:
\begin{align}
\{\chi,\ga^d\} = 0 \,, \qquad  [\chi, \ga^a] = 0 \label{chigamma},\, \forall a \in \set{1,\cdots,d-1}\,.
\end{align}

This condition was shown in \cite{CC2} a natural assumption to enforce the hermiticity of the realization of the Dirac operator. It is known \cite[Lemma 1.5.3]{Gilkey2} that ellipticity is preserved.

Since $\ga^d$ is invertible, $\dim V_+ = \dim V_-$ and $\Tr\chi =0$. 

For an even-dimensional oriented manifold, there is a natural candidate $\chi$ satisfying \eqref{chigamma}, namely 
$$
\chi:={\chi_{}}_{\pa M}=(-i)^{d/2-1} \ga(e_1)\cdots \ga(e_{d-1}) \, .
$$
This notation is compatible with \eqref{chi}.
Recall that 
\begin{align}
\Tr ( \ga^{i_1} \cdots \ga^{i_{2k+1}} ) = 0 \, ,\,\, \forall k \in \N, \quad\Tr (\ga^{i} \ga^{j} ) = \dim V \, \delta ^{ij}\, . \label{tracegammaimpair}
\end{align}

The natural realization of this boundary condition for the Dirac type operator $D+A$ is the operator $(D+A)_\chi$ which acts as $D+A$ on the domain $\set{s\in C^\infty(V) \, : \, \Pi_- s_{|\del M} =0 }$. It turns out (see \cite[Lemma 7]{BG2}) that the natural boundary operator $B_\chi^A$ defined by 
$$
\B_\chi^A s:= \Pi_- (D+A) s_{|\del M} \oplus \Pi_{-}s_{|\del M}\, 
$$  
is a boundary operator of the form (\ref{Mixed}) provided that 
$S = \half \Pi_+ (-i[\ga^d,A]-L_{aa}\chi)\Pi_+$. 

\begin{lemma}
Actually, $S$ and $\chi_{;a}$ are independent of the perturbation $A$:
\label{Schi}

(i) $S = -\half L_{aa}\, \Pi_+$\,.

(ii) $\chi_{;a} = \chi_{:a}$.
\end{lemma}

\begin{proof} 
$(i)$ Since $A$ is locally of the form $-i \ga^j a_j$ with $a_j\in C^\infty(U)$,  we obtain from (\ref{chigamma}), 
\begin{align*}
\chi[\ga^d,A] = -i a_j \, \chi[\ga^d,\ga^j] = -i \sum_{j<d} a_j \,\chi[\ga^d,\ga^j] = i\sum_{j<d} a_j \,[\ga^d,\ga^j]\chi = -[\ga^d,A] \chi
\end{align*}
and the result as a consequence of $\Pi_+\,[\ga^d,A]=[\ga^d,A]\,\Pi_-$ and $\Pi_+\Pi_-=0$.

$(ii)$ We have $\nabla^A_{i} = \nabla_i+a_i \Id_V$ where $A=:-i\ga^j a_j$, and since  
$(\nabla^A_i \chi) s = \nabla^A_i (\chi s) - \chi( \nabla^A_i s)$ 
for any $s\in C^\infty(V)$, using Lemma \ref{Perturbation}, $\nabla^A_i (\chi) =[\nabla_i + a_i \Id_V, \chi]=[\nabla_i,\chi] = \nabla_i (\chi)$.
\end{proof}

While $S$ is not sensitive to the perturbation $A$, the boundary operator $\B_\chi^A$ depends a priori on $A$. We shall denote $\B_\chi$ the boundary operator $\B_\chi^A$ when $A=0$.

The coefficients $a_{d-k}$ for $0\leq k\leq 4$ have been computed in \cite{BG1} for general mixed boundary conditions in the case of Laplace type operators and in \cite[Lemma 8]{BG2} for Dirac type operators with chiral boundary conditions. We recall here these coefficients in our setting:

\begin{prop}
\label{ThmGilkey}\begin{align*}
& a_{d}(D+A,\B_\chi^A)=(4\pi)^{-d/2}\,\int_M \Tr_V 1 \, dx \, , \\
& a_{d-1}(D+A,\B_\chi^A)=0 \, ,  \\
& a_{d-2}(D+A,\B_\chi^A)=\tfrac{(4\pi)^{-d/2}}{6}\big\{
     \int_M \Tr_V(6E^A+\tau) \,dx + \int_{\del M}\Tr_V(2L_{ aa}+12S)\, dy\,\big\}, \\
& a_{d-3}(D+A,\B_\chi^A)=\tfrac{(4 \pi )^{
       -(d-1)/2}}{384}  \int_{\del M}\Tr_V\big\{96 \chi E^A + 3 L_{aa}^2 + 6 L_{ab}^2 
+ 96 S L_{aa} + 192 S^2 -12 \chi_{;a}^2 \} \, dy ,\\
& a_{d-4}(D+A,\B_\chi^A)=\tfrac{(4 \pi )^{-d/2}}{360} \big\{\int_M \Tr_V 
      \big \{ 60\tau E^A + 180 (E^A)^2 +30 (\Omega_{ij}^{A})^2 + 5\tau^2 -2\rho^2+2R^2 \big\} \, dx \\
      &\hspace{4.3cm} + \int_{\del M} \Tr_V \big\{ 180 \chi E^A_{;d}+120E^A L_{aa} +720 S E^A+60 \chi \chi_{;a} \Omega^A_{ad} +T\big \} \, dy\, \big\}.
\end{align*}
where 
\begin{align*}
&T:=20 \tau L_{aa}+  4 R_{adad} L_{bb} -12 R_{adbd}L_{ab}+4R_{abcb}L_{ac} + \tfrac{1}{21}\big(160 L_{aa}^3-48L_{ab}^2 L_{cc}+272 L_{ab}L_{bc}L_{ac}\\
&\hspace{1cm} +120 \tau S+144 S L_{aa}^2+48 S L_{ab}^2+480 (S^2 L_{aa}+S^3)-42\chi_{;a}^2 L_{bb}+6\chi_{;a}\chi_{;b}L_{ab}-120 \chi_{;a}^2 S\big)
\end{align*}
is independent of $A$.
\end{prop}

The following proposition shows that there are no tadpoles up to order 5 in manifolds endowed with a chiral boundary condition.
\begin{theorem}
Let $M$ be an even $d$-dimensional compact oriented spin Riemannian manifold with smooth boundary $\del M$ and spin bundle $V$. Let $D:=-i\ga^j\nabla_j $ be the classical Dirac operator, and $\chi={\chi_{}}_{\pa M}=(-i)^{d/2-1} \ga(e_1)\cdots \ga(e_{d-1})$ where $(e_i)_{1\leq i\leq d}$ is a local orthonormal frame of $TM$.
 
The perturbation $D\to D+A$ where $A=-i\ga^j a_j$ is a 1-form for $D$, induces, under the chiral boundary condition, the following perturbations on the heat kernel coefficients where we set 
$c_{d-k}(A):=a_{d-k}(D+A,\B_\chi^A)-a_{d-k}(D,\B_{\chi})$:

(i) $c_{d}(A)=c_{d-1}(A)=c_{d-2}(A)=c_{d-3}(A) = 0$.

(ii) $c_{d-4}(A) =- \tfrac{1}{6(2\pi)^{d/2}}\int_M F_{\mu\nu}F^{\mu\nu} \,dx.$

In other words, the coefficients $a_{d-k}$ for $0\leq k\leq 3$ are unperturbed, $a_{d-4}$ is only perturbed by quadratic terms in $A$ and there are no linear terms in $A$ in $a_{d-k}(D+A,\B_\chi^A)$ for $k\leq 5$.
\end{theorem}

\begin{remark} When $A$ is selfadjoint, all coefficients $a_{d-k}(D+A,\B_\chi^A)$ and $a_{d-k}(D,\B_{\chi})$ are real while linear contributions in $A$ are purely imaginary, modulo traces of $\ga$ and $\chi$ matrices and their covariant derivatives. Since the invariant terms appearing as integrands of $\int_{M}$ and $\int_{\del M}$ in the coefficients at higher order are polynomial in $S$, $\chi$, $R$, $E^A$ and $\Omega^A$, and their covariant derivatives, one expects no linear terms in $A$ at any order.
\end{remark}

\begin{proof}
$(i)$ The fact that $c_{d}(A)=c_{d-1}(A)=0$ follows from Proposition \ref{ThmGilkey}. 

Since by Lemma \ref{Schi}, $c_{d-2}(A) = (4\pi)^{-d/2}\int_M \Tr_V(E^A-E_A)\,  dx$, we get $c_{d-2}(A)=0$ because $\Tr_V E^A =\Tr_V E$ by Lemma \ref{Perturbation}. 
     
From Proposition \ref{ThmGilkey} and Lemma \ref{Schi}, we get
$c_{d-3}(A)= \tfrac{1}{4}(4 \pi )^{-(d-1)/2} \int_{\del M}\Tr_V\big\{ \chi (E^A-E)\big\}$.

Since $\chi(E^A-E) = (-i)^{d/2} \ga^1\cdots\ga^{d-1} [\ga^j,\ga^k]F_{jk}$, (\ref{tracegammaimpair}) yields  $\Tr_V \chi(E^A-E)=0$ because $d$ is even.

$(ii)$ Since $\Tr_V(E^A-E)=0$ and $\Tr_V \chi(E^A-E) =0$, we obtain $\Tr_V S (E^A-E)=0$ from Lemma \ref{Schi}. Thus, using Proposition \ref{ThmGilkey} and Lemma \ref{Schi},
\begin{align*}
&c_{d-4}(A)=\tfrac{(4 \pi )^{-d/2} }{360}\big\{\int_M \Tr_V
      \big \{ 180( (E^A)^2 -E^2) +30 \big((\Omega_{ij}^{A})^2-(\Omega_{ij})^2\big) \big\} \, dx \\
      &\hspace{4cm} + \int_{\del M} \Tr_V \big\{ 180 \chi( E^A_{;d}-E_{:d})+60 \chi \chi_{;a} (\Omega^A_{ad}-\Omega_{ad}) \} \, dy\, \big\}.
\end{align*}
We obtain locally $\Tr_V \big((E^A)^2-E^2\big) = \tfrac{1}{16}\Tr([\ga^\mu,\ga^\nu][\ga^\rho,\ga^\sg])F_{\mu\nu}F_{\rho\sg}$ using Lichn\'erowicz' formula $E=-\tfrac{1}{4}\tau$. 
Since $\Tr_V([\ga^\mu,\ga^\nu][\ga^\rho,\ga^\sg])=4. 2^{d/2}(g^{\mu\sg}g^{\nu\rho}-g^{\mu\rho}g^{\nu\sg})$, 
\begin{align*}
\Tr_V \big((E^A)^2-E^2\big) = -\,2^{d/2-1} F_{\mu\nu}F^{\mu\nu}.
\end{align*}
$\nabla$ being the spin connection associated to the spin structure of $M$, we have $\Omega_{ij}=\tfrac{1}{4}\ga^{k}\ga^l R_{ijkl}$. So $R_{ijkl}=-R_{ijlk}$ implies $\Tr_V \Omega_{ij}=0$. Hence, with Lemma \ref{Perturbation}, 
\begin{align*}
\Tr_V \big((\Omega_{ij}^A)^2- \Omega_{ij}^2\big) =2^{d/2}  F_{ij}^2 =2^{d/2} F_{\mu\nu}F^{\mu\nu}\, .
\end{align*}
Moreover, $E^A_{;d} = [\nabla_d+a_d,E^A] = [\nabla_d,E+\tfrac{1}{4}[\ga^i,\ga^j]F_{ij}] = E_{:d} + \tfrac{1}{4}[\nabla_d,[\ga^i,\ga^j]] F_{ij}$. 

Using $[\nabla_i,\ga^i]=\ga(\nabla_i e_j)$ and (\ref{tracegammaimpair}),
$$
\Tr_V \big( \chi( E^A_{;d}-E_{:d}) \big)= (-i)^{d/2} \tfrac{1}{2}\, F_{ij} \Tr_V \big\{ \ga^1\cdots \ga^{d-1} \big(\ga(\nabla_d e_i) \ga^j+\ga^i \ga(\nabla_d e_j)\big) \big\} =0  \, .
$$ 
It remains to check that $\Tr_V \big(\chi \chi_{:a} (\Omega^A_{ad}-\Omega_{ad})\big)=0$. 
Let $\chi_M=-i\chi \ga^d$ be the grading operator (see \eqref{chi}). Since $\chi_{M}$ commutes with the spin connection operator $\nabla$ (see \cite[p. 396]{Polaris}),
$$
0=[\nabla_a,\chi_M]=[\nabla_a,\chi\ga^d] = \chi_{:a}\ga^d +\chi[\nabla_a,\ga^d]=\chi_{:a}\ga^d +\chi\ga(\nabla_a e_d)
$$
and thus $\chi\chi_{:a}= -\ga(\nabla_a e_d)\ga^d  = -\Ga_{ad}^j \ga^j \ga^d$, where 
$\Ga_{ad}^j=-\Ga_{aj}^{d}$ since $(e_j)$ is an orthonormal frame. 
So $\Tr_V(\chi \chi_{:a}) = -\Ga_{ad}^{j} \delta^{jd}=-\Ga_{ad}^{d}=0$. 
Finally, the result on $c_{d-4}$ follows from  Lemma \ref{Perturbation} as $\Tr_V\big(\chi \chi_{:a} (\Omega^A_{ad}-\Omega_{ad})\big)=\Tr_V(\chi \chi_{:a})F_{ad}$.

The coefficient $a_{d-5}(D+A,\B_\chi^A)$ is computed in \cite{BGKV}. One can check directly as above that linear terms in $A$ are not present. The computation uses the fact that the trace of the following terms $\chi E^A_{;dd}$, $E^A_{;d}S$, $\chi (E^A)^2$, $E^A S^2$, $\chi_{;a}\chi_{;b} \Omega^A_{ab}$, $\chi_{;a}^2 E^A$, do not have linear terms in $A$.
\end{proof}

\medskip

In the following, we investigate the above conjecture with Chamseddine--Connes pseudodifferential calculus applied to compact spin manifolds without boundary and Riemannian spectral triples. We also see, using Wodzicki residue, how to compute some noncommutative integrals in this setting.

\section{Commutative spectral triples} \label{commtadsection}

\subsection {Commutative geometry}

\begin{definition}
\label{rieman}
Consider a commutative spectral triple given by a compact
Riemannian spin manifold  $M$ of dimension $d$ without boundary and
its Dirac operator $\DD$ associated to the Levi-Civita connection.
This means $\big( \A:=C^{\infty}(M),\, \H:=L^2(M,S),\,\DD \big)$
where $S$
is the spinor bundle over $M$. This triple is real since, due to the
existence of a spin structure, the charge conjugation operator
generates an anti-linear isometry $J$ on $\H$ such that
$$
JaJ^{-1}=a^*,\quad \forall a \in \A,
$$
and when $d$ is even, the grading is given by the chirality matrix
\begin{align}
\label{chi}
\chi_M:=(-i)^{d/2}\,\ga^1\ga^2\cdots\ga^d.
\end{align}

Such a triple is said to be a commutative geometry (see \cite{CGravity} and \cite{CReconstruction} for the role of $J$ in the nuance between spin and spin$^c$ manifold).
\end{definition}
Since $JaJ^{-1}=a^*$ for $a\in\A$, we get that in a commutative geometry, 
\begin{align}
\label{JAJ}
JAJ^{-1}=-\epsilon\,A^* , \quad \forall A\in \Omega^1_\DD(\A).
\end{align}

\subsection{No tadpoles}

The appearance of tadpoles never occurs in commutative geometries, as
quoted in \cite[Lemma 1.145]{ConnesMarcolli} for the dimension $d=4$. 
This fact means that a given geometry $(\A,\, \H,\,\DD)$ is a
critical point for the spectral action \eqref{formuleaction}.

\begin{theorem} 
\label{proptadpoles} 
There are no tadpoles on a commutative geometry, namely, for any one-form $A=A^*\in \Omega_{\DD}^1(\A)$, $Tad_{\DD+A}(k)=0$, for any $k \in \Z$,
$k\leq d$.
\end{theorem}

\begin{proof}
Since $\tilde A=0$ when $A=A^*$ by \eqref{JAJ}, the result follows from Corollary \ref{Atilde=0}.
\end{proof}

There are similar results in the following 

\begin{lemma}
\label{similar}
Under same hypothesis, for any $k,\,l\in \N$

(i) $\ncint A\,\DD^{-k}=-\epsilon^{k+1}\,\ncint A\,\DD^{-k}$,

(ii) $\ncint \chi A\,\DD^{-k}=-\epsilon^{k+1}\ncint\chi A \,\DD^{-k}$,

(iii) $\ncint A^l \vert \DD \vert^{-k}=(-\epsilon)^l\,\ncint {A^l}\vert \DD \vert^{-k}$,

(iv) $\ncint \chi A^l \vert \DD \vert^{-k}=(-\epsilon)^l\,\ncint {\chi A^l}\vert \DD \vert^{-k}$.
\end{lemma}

\begin{proof}
\begin{align*}
\ncint A\,\DD^{-k}&=\overline{\ncint JA\,\DD^{-k}J^{-1}}=\overline{\ncint
JAJ^{-1}( \epsilon^k \DD^{-k})}=-\epsilon^{k+1}\,\overline{\ncint
A^*\,\DD^{-k}}=-\epsilon^{k+1}\,\ncint \,\DD^{-k}A \\
&=-\epsilon^{k+1}\,\ncint A\,\DD^{-k}. 
\end{align*}
The same argument gives the other equalities using $\chi A=-A \chi$ and $\chi \vert \DD \vert=\vert \DD \vert \chi$.
\end{proof}

\begin{lemma}
   \label{componentofzeta}
For any one-form $A$,
$\ncint \big( A \,\DD^{-1}\big)^k=0$ when $k\in \N$ is odd.
\end{lemma}

\begin{proof}
We have
\begin{align}
\label{trick}
\ncint \big( A \,\DD^{-1}\big)^k&=\overline{ \ncint J\big(
A\DD^{-1}\big)^kJ^{-1}}=\overline{ \ncint \big( JAJ^{-1}\,
J\DD^{-1}J^{-1}\big)^k}=(-1)^k\epsilon^{2k}\overline{ \ncint
\big(A^*\DD^{-1}\big)^k} \nonumber  \\
&=(-1)^k\ncint (A\DD^{-1})^k
\end{align}
(which shows again that $\ncint A \DD^{-1}=0$).
\end{proof}

\subsection{Miscellaneous for commutative geometries}

To show that more noncommutative integrals, where the use of the operator $J$ in the trick \eqref{trick} is not sufficient, are nevertheless zero, we need to use the Wodzicki residue (see \cite{Wodzicki1,Wodzicki}):
in a chosen coordinate system and local trivialization $(x,\xi)$ of
$T^*M$, this residue is
\begin{align}
\label{wres}
wres_{x}(X) & :=\int_{S_x^*M} \Tr\big(\sigma_{-d}^X\,(x,\xi)\big)\,
\vert d\xi \vert\,\vert dx^1\wedge\cdots \wedge dx^d\vert,
\end{align}
where $\sigma_{-d}^X\,(x,\xi)$ is the symbol of the classical
pseudodifferential operator $X$ in the chosen coordinate frame $(x_1,\cdots,x_d)$, which is homogeneous of degree
$-d:=-\text{dim}(M)$ and taken at point $(x,\,
\xi)\in T^*(M)$, $d\xi$ is the normalized restriction of the volume
form to the unit sphere $S_x^*M \simeq \mathbb{S}^{d-1}$, so we
assume $d\geq 2$ to get $S_x^*M$ connected.

This $wres_x(X)$ appears to be a one-density not depending on the
local
representation of the symbol (see \cite{Wodzicki,Polaris}), so 
\begin{align}
Wres(X):=\int_M wres_x(X)
\end{align}
is well defined. 

The noncommutative integral $\ncint$ coincides with the Wodzicki
residue, up to a scalar: since both $\ncint$ and $Wres$ are
traces on the set of pseudodifferential operators, the uniqueness of
the trace \cite{Wodzicki} gives the proportionality  
\begin{align}
\label{Wres}
\ncint X=c_d\,Wres(X)
\end{align}
where $c_d$ is a constant depending only on $d$. Computing separately $\ncint \vert \DD \vert^{-d}$ and $Wres( \vert \DD \vert^{-d})$, we get $c_d>0$. (Note that $\ncint$ is not a positive functional, see Lemma \ref{scalarcurvature}.)

Lemma \ref{adjoint} follows for instance from the fact that $\int_M
wres_x(X^*)=\overline{\int_M wres_x(X)}$.

Note that $Wres$ is independent of the metric.

As noticed by Wodzicki, $\ncint X$ is equal to $-2$ times the
coefficient in log $t$ of the asymptotics of $\Tr(X\, e^{-t\,\DD^2}$)
as $t \rightarrow 0$. It is remarkable that this coefficient is
independent of $\DD$ and this gives a close relation between the
$\zeta$ function and the heat-kernel expansion with $Wres$. Actually, by
\cite[Theorem 2.7]{GS}
\begin{align}
\label{heat}
\Tr(X\, e^{-t\,\DD^2}) \sim_{t\rightarrow 0^+} \sum_{k=0}^{\infty}
a_k\,t^{(j-ord(X)-d)/2} + \sum_{k=0}^{\infty} (-a'_k\,\log
t+b_k)\,t^k,
\end{align}
so $\ncint X=2 a'_0$. Since, via Mellin transform, $\Tr(X\,\DD^{-2s})
=\tfrac{1}{\Gamma (s)} \int_0^\infty t^{s-1}\, \Tr(X\, e^{-t\,\DD^2})
\,dt$, the non-zero coefficient $a'_k$, $k\neq0$ create a pole of
$\Tr(X\,\DD^{-2s}) $ of order $k+2$ since $\int_0^1 t^{s-1}
\log(t)^k=\tfrac{(-1)^k k!}{s^{k+1}}$ and 
\begin{align}
\label{Gamma}
\Gamma(s)=\frac{1}{s} +\gamma+s\,g(s)
\end{align} 
where $\gamma$ is the Euler constant and the function $g$ is also holomorphic around zero.

We have $\ncint 1=0$ and more generally, $Wres(P)=0$ for all zero-order pseudodifferential projections \cite{Wodzicki1}.

For extension to log-polyhomogeneous pseudodifferential operators,
see \cite{Lesch}.

When $M$ has a boundary, some $a'_k$ are non zero, the dimension
spectrum can be non simple (even if it is simple for the Dirac operator, see for instance \cite{Lescure}).

On a spectral triple $(\A,\, \H,\, \DD)$, changing the
product on $\A$ may or not affect the dimension spectrum: for
instance, there is no change  when one goes from the commutative
torus to the noncommutative one, while the dimension
spectrum of $SU_q(2$) which is bounded from below, does not coincide
with the dimension spectrum of the sphere $\mathbb{S}^3$
corresponding to $q=1$ .

\vspace{0.3cm}
We first introduce few necessary notations. In the following we fix a local coordinate frame $(U,(x_i)_{1\leq i \leq n})$ which is normal at $x_0\in M$, and denote $\sigma_{k}^X$ the $k$-homogeneous symbol of any classical pseudodifferential operator $X$ on $M$, in this local coordinate frame. The Dirac operator is locally of the form---compatible with \eqref{phiDirac}
\begin{align}
\label{Dirac}
\DD=-i\ga (dx^j)\,\big(\partial_{x^j}+\omega_j(x)\big)
\end{align}
where $\omega_j$ is the spin connection, $\ga$ is the Clifford
multiplication of one-forms \cite[page 392]{Polaris}. Here we make the choice of gauge given
by $h:=\sqrt{g}$ which gives \cite[Exercise 9.6]{Polaris}
$$
\omega_i=-\tfrac{1}{4}\, \big( \Gamma_{ij}^k
\,g_{kl}-\partial_{x^j}(h_j^\alpha) \delta_{\a \beta}\,h_l^\beta
\big)\, \ga(dx^j)\,\ga(dx^l),\quad
\ga(dx^j)=\sqrt{g^{-1}}^{\,jk}\,\gamma_k
$$
where $\gamma^j=\gamma_j$ are the selfadjoint constant $\gamma$
matrices satisfying $\set{\ga^i ,\ga^j}=\delta^{ij}$. Thus
$$
\sigma^{\DD}(x,\xi)=\sqrt{g^{-1}}^{\,jk}\,\gamma_k \big(\xi_j-i\,\omega_j(x)\big).
$$

We have chosen normal (or geodesic) coordinates around the base point $x_0$. Since 
\begin{align*}
& g_{ij}(x)=g_{ij}(x_0)+\tfrac{1}{3}R_{ijkl}\,x^kx^l +o(\vert \vert x
\vert\vert^3),\\
& g^{ij}(x)=g^{ij}(x_0)-\tfrac{1}{3}{{{{R}^i}_k}{}^j}_l\,x^kx^l +o(\vert \vert x \vert\vert^3),\\
& g_{ij}(x_0)=\delta_{ij}, \quad\Gamma_{ij}^k(x_0)=0,
\end{align*}
the matrices $h(x)$ and $h^{-1}(x)$ have no linear terms in $x$. Thus
$$
\om_i(x_0)=0.
$$
We could also have said that parallel translation of a basis of the
cotangent bundle along the radial geodesics emanating from $x_0$
yields a trivialization (this is the radial gauge) such that
$\om_i(x_0)=0$.
In particular, using product formulae for symbols and the fact that in
the decomposition $D=\DD+P$, $P\in OP^{-\infty}$, we get for $k \in
\N$
\begin{align}
&\sigma_1^{\DD}(x,\xi)=\sqrt{g^{-1}}^{\,jk}(x)\,
\gamma_k\xi_j=\ga(\xi),
&&\sigma_1^{\DD}(x_0,\xi)=\gamma^j\xi_j,\label{sigma1}\\
&\sigma_0^{\DD}(x,\xi)=-i\sqrt{g^{-1}}^{\,jk}(x)\,\gamma_k\om_j(x),
&&\sigma_0^{\DD}(x_0,\xi)=0,\label{sigma0}\\
&\partial_{x^k}\sigma_1^{\DD}(x_0,\xi)=0, \label{deriveedesigma1}\\
&\sigma_{-1}^{\DD^{-1}}(x,\xi)=\sqrt{g^{-1}}^{\,jk}(x)\,\gamma_{j}\xi_k \, \vert
\vert \xi \vert \vert_x^{-2}, && \vert \vert \xi \vert \vert_x^2:=g^{jk}(x)\,\xi_j \xi_k \label{sigma-1}\\
&\partial_{x^k}\sigma_{-1}^{\DD^{-1}}(x_0,\xi)=0. &&
\label{deriveesigma-1}
\end{align}
We will use freely the fact that the symbol of a one-form $A$ can be written as 
\begin{align}
\label{sigmaA}
\sigma^A(x,\xi)=\sigma_0^A(x)=-i \, a_k(x) \, \gamma^k
\end{align}
with $a_k(x) \in i\R$ when $A=A^*$.

When $d$ is even (so $\epsilon=1$), remark that for $k=l$ and
$A_i=a_i[\DD,b_i]$ and $a=\prod_{i=1}^k a_i$, then by \cite[page 231
(actually, $\chi$ is missing)]{CM}, \cite{Ponge2} or \cite[p.
479]{Polaris} when $k=d$, ($M$ is supposed to be oriented)
$$
\ncint \chi A_1\cdots A_k \vert \DD\vert^{-k}=c'_k \int_M
\hat{A}(R)^{(d-k)}\wedge a db_1 \wedge \cdots \wedge db_k
$$
where $\hat{A}(R)$ is the $\hat{A}$-genus associated to the Riemannian
curvature $R$. Since  we have $\hat{A}(R)\in \oplus_{j\in \N}\Omega^{4j}(M,\R)$,
$\ncint \chi A^k \vert D\vert^{-k}$ can be non zero only when
$k=d-4j$. For instance in dimension $d=$2, for $j=0$,
$$
\sigma_{-2}^{\chi A_1A_2 \DD^{-2}}(x,\xi)=\sigma_0^{\chi
A_1A_2}(x)\,\sigma_{-2}^{\DD^{-2}}(x,\xi)=-a_1(x)\,a_2(x) \,\chi
g^{jk}(x)\gamma_{j}\gamma_{k}\,\tfrac{1}{g^{lm}(x)\xi_{l}\xi_{m}}.
$$
Thus $wres_{x}(\chi A_1A_2 \DD^{-2})=-2\,a_1(x)\,a_2(x)\, \sqrt{\text{det}\, g_x} \, \Tr(\chi \gamma^{j}\gamma^{k})$,  so if $ \nu_g$ is the Riemannian density, 
\begin{align}
\label{d=2}
\ncint \chi A_1A_2\,\DD^{-2}=-2 c_d\,\, \Tr(\chi \gamma^{j}\gamma^{k})\int_M\,a_1a_2\, \nu_g .
\end{align}
Actually, this last equality is nothing else than Wodzicki--Connes' trace theorem, see \cite[section
7.6]{Polaris}, and this is equal to $c'_d\int_M   a_1a_2db_1\wedge db_2$ as claimed above.

\vspace{0.3cm}
We introduce a few subspaces of the pseudodifferential operators space
$\Psi (M)$. Let
\begin{align*}
\B_e&:=\set{P \in \Psi(M) \,:\, \sigma_j^P \in E_j, \, \forall j\in
\Z}\quad \text{e for even}, \\
\B_o & :=\set{P \in \Psi(M) \,:\, \sigma_j^P \in O_j, \, \forall j\in
\Z}\quad \text{o for odd},
\end{align*}
such that, for $m=2^{[d/2]}$,

\begin{align*}
&E_j := \set{ f\in C^{\infty} \big(U\times
\R^d\backslash\set{0},\M_m(\C) \big) \ : \   f(x,\xi) = \sum_{i\in I}
\tfrac{\xi^{\b^i}}{\norm{\xi}_x^{2k_i}} h_i(x) \ , \ I\neq \emptyset,
\\
&\hskip3cm k_i\in \N, \ \b^i \in \N^d  \ , |\b^i|-2k_i = j \ , \
h_i\in C^{\infty}(U,\M_m(\C))}\, , \\
& O_j := \set{ f\in C^{\infty} \big(U\times
\R^d\backslash\set{0},\M_m(\C) \big) \ : \   f(x,\xi) = \sum_{i\in I}
\tfrac{\xi^{\b^i}}{\norm{\xi}_x^{2k_i+1}} h_i(x) \ , \ I\neq
\emptyset, \\
&\hskip3cm k_i\in \N, \ \b^i \in \N^d,  \  |\b^i|-(2k_i+1) = j \ , \
h_i\in C^{\infty}(U,\M_m(\C))}\, .
\end{align*}

\begin{lemma} 
\label{Ejlem}
For any $j,\,j' \in \Z$ and  $\a\in \N^d$, 

$(i)$ $E_j E_{j'}\subseteq E_{j+j'} \text{ and } \del_\xi^{\a} E_j
\subseteq E_{j-|\a|},\,\,\del_x^\a E_j\subseteq E_j$.

$(ii)$ $O_j O_{j'}\subseteq E_{j+j'}$ and $\del_\xi^{\a} O_j
\subseteq O_{j-|\a|}$, $\del_x^\a O_j\subseteq O_j$.

$(iii)$ $O_j E_{j'}$ and $E_{j'} O_j$ are included in $O_{j+j'}$. 

$(iv)$ $\B_e$ is a sub-algebra of $\Psi(M)$.

$(v)$  $\B_e\B_e$, $\B_o \B_o$ are included in $\B_e$, and $\B_e
\B_o$, $\B_o \B_e$ are included in $\B_o$.
\end{lemma}
\begin{proof}
$(i)$ Let $f \in E_j$ and $\a \in \N^d$. We have, if $f(x,\xi) =
\sum_{i\in I} \tfrac{\xi^{\b^i}}{\norm{\xi}_x^{2k_i}} \,h_i(x)$, 
$$
\del^\a_\xi f = \sum_{i\in I}
\del^\a_{\xi}(\tfrac{\xi^{\b^i}}{\norm{\xi}_x^{2k_i}}) \,h_i(x)=
\sum_{i\in I} \sum_{\ga\leq \a} \tbinom{\a}{\ga}
\del^{\a-\ga}_{\xi}(\xi^{\b^i})
\del^\ga_\xi(\tfrac{1}{\norm{\xi}_x^{2k_i}}) \,h_i(x).
$$
We check by induction that we can write
$$
\del^\ga_\xi(\tfrac{1}{\norm{\xi}_x^{2k_i}}) =
\tfrac{1}{\norm{\xi}_x^{2k_i(|\ga|+1)}} \sum_{p} \la_{p}
\prod_{j=1}^{|\ga|} \del^{\b^{j,p}}_\xi \norm{\xi}_x^{2k_i}
$$
where $\la_{p}$ are real numbers, the sum on indices $p$ is
finite, and $\sum_{j=1}^{|\ga|} \b^{j,p} = \ga$. As a consequence,
since $\norm{\xi}_x^{2k_i}=(g^{kl}(x)\xi_k\xi_l)^{k_i}$ is a
homogeneous polynomial in $\xi$ of degree $2k_i$, we get 
$\del^\a_\xi f \in E_{j-|\a|}$. The inclusions $E_j E_{j'}\subseteq
E_{j+j'}$, $\del_x^\a E_j\subseteq E_j$ are straightforward.

$(ii)$ The proof is similar to $(i)$ since by induction
$$
\del^\ga_\xi(\tfrac{1}{\norm{\xi}_x}) =
\tfrac{1}{\norm{\xi}_x^{2|\ga|+1}} \sum_{p} \la_{p}
\prod_{j=1}^{|\ga|} \del^{\b^{j,p}}_\xi \norm{\xi}_x^{2}
$$
where $\la_{p}$ are real numbers, the sum on the indices $p$ is
finite and $\sum_{j=1}^{|\ga|} \b^{j,p} = \ga$.

$(iii)$ Straightforward.

$(iv)$ The product symbol formula for two classical
pseudodifferential operators $P\in \Psi^p(M)$, $Q\in \Psi^q(M)$ gives
\begin{equation}
\sg^{PQ}_{p+q-j} = \sum_{\a\in \N^d}\, \sum_{k\geq 0, \ |\a|+k\leq j}
i^{|\a|}\tfrac{(-1)^{|\a|}}{\a!} \, \del_{\xi}^\a
\sg^P_{p-j+|\a|+k}\, \del^\a_x \sg^Q_{q-k} \, .
\label{prodsymbol}
\end{equation}
The presence of the factor $i^{|\a|}$ that will be crucial in later
arguments like Lemma \ref{Calg}.

If $P,Q \in \B_e$, we see that by $(i)$, $\del_{\xi}^\a
\sg^{P}_{p-j+|\a|+k}\in E_{p-j+k}$ and $\del^\a_x \sg^Q_{q-k} \in
E_{q-k}$. 
Again by $(i)$, we obtain $\del_{\xi}^\a \sg^P_{p-j+|\a|+k}\,
\del^\a_x \sg^Q_{q-k} \in E_{p+q-j}$, so the result follows from
(\ref{prodsymbol}).

$(v)$ A similar argument as $(iv)$ can be applied, using $(ii)$ to
obtain $\B_o\B_o\subseteq \B_e$ and $(iii)$ to get $\B_o\B_e
\subseteq \B_o$, $\B_e\B_o\subseteq \B_o$.
\end{proof}

$\B_e$ and $\B_o$ are stable by inverse:

\begin{lemma}
\label{Binversion}
Let $P\in \B_e$ (resp. $\B_o$) be an elliptic classical
pseudodifferential operator in $\Psi^p(M)$ with
$\sg^P_p(x,\xi)=\norm{\xi}_x^p$, $p\in \N$. Then any parametrix
$P^{-1}$ of $P$ is in $\B_e$ (resp. $\B_o$).
\end{lemma}
\begin{proof} 
Assume $P \in \B_e$ so $p$ is even. From the parametrix equation $P
P^{-1} = 1$, we obtain  $\sg_{-p}^{P^{-1}} = (\sg_p^{P})^{-1}=
\norm{\xi}_x^{-p} \in E_{-p}$. Moreover, using (\ref{prodsymbol}), we
see that for any $j\in \N^*$,
\begin{align}
\sg_{-p-j}^{P^{-1}} = -(\sg_p^{P})^{-1}\big( \sum_{0\leq k<j}
\sg_{p-j+k}^P \, \sg_{-p-k}^{P^{-1}} 
+ \sum_{0<|\a|\leq j} \sum_{k=0}^{j-|\a|}
i^{|\a|}\tfrac{(-1)^{|\a|}}{\a!}  \,\del_{\xi}^\a \sg_{p-j+|\a|+k}^P
\, \del^\a_x \sg_{-p-k}^{P^{-1}} \,\big)
\label{parametrix}
\end{align}
We prove by induction that for any $j\in \N$, $\sg_{-p-j}^{P^{-1}}
\in E_{-p-j}$: suppose that for a $j\in \N^*$, we have for any
$j'<j$, $\sg_{-p-j'}^{P^{-1}} \in E_{-p-j'}$. We then directly check
with Lemma \ref{Ejlem} and (\ref{parametrix}) that
$\sg_{-p-j}^{P^{-1}} \in E_{-p-j}$.

The case $P\in \B_o$ is similar.
\end{proof}

\begin{lemma}
\label{DiracB}
For any $k\in \Z$, $\DD^{k} \in \B_e$ and when $k$ is odd, $|\DD|^k\in
\B_o$.
\end{lemma}

\begin{proof}
Since $\DD\in \B_e$, $\DD^{-2}$ is in $\B_e$ by Lemma
\ref{Binversion} and \ref{Ejlem} and so is $\DD^k$.

Using \eqref{prodsymbol} for the equation $|\DD||\DD|=\DD^2$, we
check that $\sg_1^{|\DD|}(x,\xi)=\norm{\xi}_x$ and for any $j\in
\N^*$,
\begin{align}
\sg_{1-j}^{|\DD|}=&\tfrac{1}{2\norm{\xi}_x}
\big(\,\sg_{2-j}^{\DD^2}-\sum_{0<k<j}\sg_{1-j+k}^{|\DD|} \,
\sg_{1-k}^{|\DD|} 
+ \sum_{0<|\a|\leq
j}\sum_{k=0}^{j-|\a|}i^{|\a|}\tfrac{(-1)^{|\a|}}{\a!} \, \del^\a_\xi
\sg_{1-j+|\a|+k}^{|\DD|} \del_{x}^\a \sg_{1-k}^{|\DD|} \, \big).
\label{sqrt}
\end{align}
Again, a straightforward induction argument shows that for any $j\in
\N$, $\sg_{1-j}^{|\DD|} \in O_{1-j}$, and thus $|\DD|\in \B_o$. The
result follows as above.
\end{proof}

In the next four lemmas, we emphasize the fact that only some of the results could be obtained using the trick \eqref{trick} with the operator $J$. 

\begin{lemma}
\label{lemadxi}

(i) If $d$ is odd, then for any $P\in \B_e$, $\ncint P = 0$.

(ii) If $d$ is even, then for any $P\in \B_o$, $\ncint P = 0$.

(iii) For any pseudodifferential operator $P \in \Psi_1(\A)$,

- when $d$ is odd, then $\ncint P =0$,

- when $d$ is even, then $\ncint P |\DD|^{-1} = 0$.
\end{lemma}

\begin{proof}
$(i)$ Since $\sg^P_{-d} \in E_{-d}$, $\sg^P_{-d}(x,\xi)=\sum_{i\in I}
\tfrac{\xi^{\b^i}}{\norm{\xi}_x^{2k_i}} \,h_i(x)$ where 
$|\b^i|$ are odd. The integration on the cosphere in \eqref{wres}
therefore vanishes.

$(ii)$ The same argument can be applied.

$(iii)$ Direct consequence of $(i)$ and $(ii)$.
\end{proof}

\begin{remark} Lemma \ref{lemadxi} (iii) entails for instance that $\ncint B|\DD|^{-(2k+1)}$ where $B$ is a polynomial in $\A$ and $\DD$ and $k\in \N$, always vanish in even dimension, while $\ncint B \DD^{-2k}$ always vanish in odd dimension. In other words, $\ncint B |D|^{-(d-q)} =0$ for any odd integer $q$.
\end{remark}

We shall now pay attention to the real or purely imaginary nature
(independently of the appearance of gamma matrices) of homogeneous
symbols of a given pseudodifferential operator. Let 
$$
\CC:=\set{P\in \Psi^p(M) \, : \,   \sg_{p-j}^{P} \in I_{j},\, \forall
j\in \N}
$$
where $I_k=I_e$ if $k$ is even and $I_{k}=I_o$ if $k$ is odd, with 
\begin{align*}
& I_e:=\set{f\in C^{\infty} \big(U\times \R^n,\M_m(\C) \big)
\,: \, f = \ga_{k_1}\cdots \ga_{k_q} \, h(x,\xi) \ , \ h \text{ real valued}}, \\
& I_o:=\set{f\in C^{\infty} \big(U\times \R^n,\M_m(\C) \big) \, : \, f = i \,\,\ga_{k_1}\cdots \ga_{k_q} \,h(x,\xi) \ , \ h \text{ real valued}}.
\end{align*}

\begin{lemma}
\label{Calg}
(i) $\CC$ is a sub-algebra of $\Psi(M)$.

\noindent (ii) If $P\in \CC$ is  hypo-elliptic then $P^{-1}\in \CC$.

\noindent (iii) $\DD^{k}\in \CC$ and $|\DD|^k\in \CC$ for any $k\in
\Z$.
\end{lemma}

\begin{proof}
$(i)$ Consequence of (\ref{prodsymbol}).

\noindent $(ii)$ Consequence of (\ref{parametrix}).

\noindent $(iii)$ It is clear that $\DD \in \CC$ and the fact that
$|\DD|\in \CC$ is a consequence of (\ref{sqrt}).
\end{proof}

\begin{lemma}
\label{lemadk}
Let $k\in \N$ odd. Then any element $B$ of the polynomial algebra generated by $\A$ and $[D,\A]$ satisfies $\ncint B |\DD|^{-(d-k)} =\ncint B F|\DD|^{-(d-k)}=0$.
\end{lemma}

\begin{proof}
We may assume that $B$ is selfadjoint so $\ncint B \DD^{-(d-k)} \in
\R$.

By Lemma \ref{Calg}, $\sg_{-d}^{B \DD^{-(d-k)}} = \sg_0^{B} \sg_{-d}^{\DD^{-(d-k)}} 
\hspace{-0.1cm}
\in I_{k}$. Thus $\ncint A\DD^{-k}\in i \R$ and the result
follows. The case $\ncint B F|\DD|^{-(d-k)}$ is similar.
\end{proof}

We now look at the information given by the gamma matrices.

\begin{lemma}
For any one-form $A$, $\ncint A |\DD|^{-q}=0$, $q \in \N$ in either of the following cases:

- $d \neq 1 \mod 8$ and $d\neq 5 \mod 8$,

- ($d = 1 \mod 8$ or $d = 5 \mod 8$) and ($q$ is even or $q\geq \tfrac{d+3}{2}$).
\end{lemma}

\begin{proof} In the case $d \neq 1 \mod 8$ and $d\neq 5 \mod 8$, the result follows from the fact that $\eps=1$.

The case $d$ even and $q$ odd or $d$ odd and $q$ even is done by Lemma \ref{lemadxi} $(iii)$.

Suppose that $d$ is even and $q$ is even.
If $q=2k$, with a recurrence and the symbol product formula,
we see that $\sigma_{2k-j}^{D^{2k}}$ and all its derivatives are linear
combinations of terms of the form $f(x,\xi)\ox
\ga^{j_1}\cdots\ga^{j_i}$ where $i$ is even and less than $2j$ (with
the convention $\ga^{j_1}\cdots\ga^{j_i}=1$ if $i=0$). We call
$(P_{j})$ this property. The parametrix equation $ \DD^{2k}\DD^{-2k}=1$
entails that $\sigma_{-2k}^{\DD^{-2k}}=(\sigma_{2k}^{\DD^{2k}})^{-1}$ and
for any $j\geq 1$,
\begin{align*}
&\sigma_{-2k-j}^{\DD^{-2k}} = -
\sigma_{-2k}^{\DD^{-2k}}\big( \,\sum_{r=\max\{j-2k,0\}}^{j-1}
\sigma_{2k-(j-r)}^{\DD^{2k}}\,\sigma_{-2k-r}^{\DD^{-2k}}\nonumber\\
& \hspace{4cm}+ \sum_{1\leq |\a|\leq 2k}
\,\sum_{r=\max\{j-2k,0\}}^{j-|\a|}
\tfrac{(-i)^{|\a|}}{\a!} \,
\del^\a_{\xi}\sigma_{2k-(j-|\a|-r)}^{\DD^{2k}}\
\del^\a_{x}\sigma_{-2k-r}^{\DD^{-2k}}\,\big). 
\end{align*}
Note that $\sigma_{-2k}^{\DD^{-2k}}$ satisfies $(P_0)$. By recurrence, 
this formula shows  that
$\sigma_{-2k-j}^{\DD^{-2k}}$ satisfies $(P_{j})$ for any $j\in \N$. In
particular, $\sigma_{-d}^{\DD^{-2k}}$ satisfies $(P_{-2k+d})$ and the
result follows then from \eqref{sigmaA} and the product of an odd
number (different from the dimension) of gamma matrices is traceless.

Suppose now that $d$ is odd, $q$ is odd and $d\geq q$. In that situation, any odd number of gamma matrices $\ga^{i_1}\cdots \ga^{i_r}$ is traceless when $r<d$.

Using \eqref{prodsymbol} for the equation $|\DD|^{-q}|\DD|^{-q}=\DD^{-2q}$, we
check that $\sg_{-q}^{|\DD|^{-q}}(x,\xi)=\norm{\xi}_x^{-q}$ and for any $j\in
\N^*$,
\begin{align*}
\sg_{-q-j}^{|\DD|^{-q}}=&\tfrac{1}{2\norm{\xi}_x^{-q}}
\big(\,\sg_{-2q-j}^{\DD^{-2q}}-\sum_{0<k<j}\sg_{-q-j+k}^{|\DD|^{-q}} \,
\sg_{-q-k}^{|\DD|^{-q}} 
+ \sum_{0<|\a|\leq
j}\sum_{k=0}^{j-|\a|}i^{|\a|}\tfrac{(-1)^{|\a|}}{\a!} \, \del^\a_\xi
\sg_{-q-j+|\a|+k}^{|\DD|^{-q}} \del_{x}^\a \sg_{-q-k}^{|\DD|^{-q}} \, \big) .
\end{align*}
We saw that each $\sigma_{-2q-j}^{|\DD|^{-2q}}$ satisfies ($P_j$), that is to say, is a linear combination of terms of the form $f(x,\xi)\ox \ga^{j_1}\cdots \ga^{j_i}$ where $i$ is even and less than $2j$.
Again, a straightforward induction argument shows that for any $j\in
\N$, $\sg_{-q-j}^{|\DD|^{-q}}$ satisfies ($P_j$). In particular $\sg_{-d}(A|\DD|^{-q})$ is a linear combination of terms of the form $f(x,\xi)\ox \ga^{j_1}\cdots \ga^{j_r}$ where $r\leq 2(d-q)+1$ is odd. This yields the result.
\end{proof}

The fact that $\ncint A \,\DD^{-d+1}=0$, consequence of Lemmas
\ref{lemadxi} and \ref{lemadk} is also a consequence of the fact that
$\sigma_{-d}^{\DD^{-d+1}}(x_0,\xi)=0$:

\begin{lemma}
  \label{symbol-k-1}
For all $k\in \N^*$, we have
$\sigma_{k-1}^{\DD^k}(x_0,\xi)=\sigma_{-k-1}^{\DD^{-k}}(x_0,\xi)=0$.
\end{lemma}

\begin{proof}
We already know that $ \sigma_0^\DD(x_0,\xi)=0$, see \eqref{sigma0}. We
proceed by
recurrence, assuming $ \sigma_{k-1}^{\DD^k}(x_0,\xi)=0$ for
$k=1,\cdots,n$. Then $\sigma_n^{\DD^{n+1}}=\sigma_n^{\DD^n}
\sigma_0^\DD+\sigma_{n-1}^{\DD^n} \sigma_1^\DD -i\,
\partial_{\xi_k}\sigma_n^{\DD^n} \,\partial_{x^k} \sigma_1^\DD$, thus by
\eqref{sigma0} and \eqref{deriveedesigma1}, $\sigma_n^{\DD^{n+1}}
(x_0,\xi)=0$.

Since $\DD\DD^{-1}=1$ yields $\sigma_{-2}^{\DD^{-1}}
(x_0,\xi)=-\big(\sigma_{-1}^{\DD^{-1}}  \,\sigma_0^{\DD} \big)
(x_0,\xi)=0$, we  assume $\sigma_{-k-1}^{\DD^{-k}} (x_0,\xi)=0$ for
$k=1,\cdots n$. Then 
$\sigma_{-n-2}^{\DD^{-n-1}}=\sigma_{-n}^{\DD^{-n}}
\sigma_{-2}^{\DD^{-1}}+\sigma_{-n-1}^{\DD^{-n}} \sigma_{-1}^{\DD^{-1}} -i\,
\partial_{\xi_k}\sigma_{-n}^{\DD^{-n}} \,\partial_{x^k}
\sigma_{-1}^{\DD^{-1}}$. Using \eqref{deriveesigma-1} and recurrence
hypothesis, $\sigma_{-n-2}^{\DD^{-n-1}}(x_0,\xi)=0$.
\end{proof}

\begin{remark}

Regularity of $\zeta_X(s):=\Tr(\vert X \vert^{-s})$ at point 0 when
$X$ is an elliptic selfadjoint differential operator of order one
(see \cite{Gilkey}){\rm:} 
\end{remark}
One checks that   
$\zeta_X(s)=\tfrac{1}{\Gamma(s)}\int_0^{\infty} t^{s-1}
\Tr(e^{-t\vert X \vert})\, dt$ for $\Re (s)>d$. Because of the
asymptotic expansion 
\begin{align}
\label{exp}
\Tr(e^{-t \vert X \vert})=t^{-d}\,\sum_{n=0}^N t^n \, a_n[X]
+\mathcal{O}(t^{N+1-d})
\end{align}
and meromorphic extension to the whole complex plane,
$\underset{s=d-n}{\Res} \,\zeta_X(s)=\tfrac{a_n[X]}{\Gamma(d-n)}$. In
particular, $\zeta_X(s) =\Gamma(s)^{-1}\big( \tfrac{a_{d}[X]}{s}
+f(s)\big)$, where $f$ is holomorphic around $s=0$. By \eqref{Gamma}
we get
that $\zeta_X(s)$ is regular around zero and $\zeta_X(0)=a_{d}[X]$ if
$d$ is even and  $\zeta_X(0)=0$ if $d$ is odd.

\begin{corollary}
  \label{difference}
$\zeta_{\DD+A}(0)=\zeta_{\DD}(0)=0$ when $d=dim(M)$ is odd.

When $d$ is even, $\zeta_{\DD+A}(0)-\zeta_{\DD}(0)=\sum_{k=1}^{d/2}
\tfrac{1}{2k} \ncint (A\,\DD^{-1})^{2k}$.
\end{corollary}

\begin{proof}
The result follows from \eqref{termconstanttilde} and Lemma
\ref{componentofzeta}.
\end{proof}

A proof of \eqref{termconstanttilde} also follows from
$\sigma^{\log(1+A\DD^{-1})} \sim \sum_{k=1}^{\infty} \tfrac{(-1)^k}{k}
\, \sigma^{{(A\DD^{-1})}^k}$ with $\log (X):=\tfrac{\pa\,}{\pa
z}_{\vert_{z=0}} X^z$, so $Wres \big(\log(1+A\DD^{-1})
\big)=\sum_{k=1}^d \tfrac{(-1)^k}{k} \, Wres \big({{A\DD^{-1})}}^k
\big)$ since ${(A\DD^{-1})}^k$ has zero Wodzicki residue if $k>d$ and
moreover $\zeta_{\DD+A}(0)=-Wres \big( \log(\DD+A) \big)$. Actually, the
important point is that $\det(X):=e^{Wres \big(\log(X) \big)}$ is multiplicative (see
\cite{LP}). Moreover, such determinant is different from the $\zeta$-determinant $e^{-\zeta_X'(0)}$ used for instance by Hawking \cite{Hawking} in his regularization via the partition function which suffers from conformal anomalies.

\vspace{0.3cm}
The fact that in the asymptotic expansion of the heat kernel
\eqref{exp}, the term $a_2[\DD+A]$ depends only on the scalar
curvature, so independent of $A$ is reflected in 

\begin{lemma}
   \label{dim2}
In any spectral triple of dimension 2 (commutative or not) with
vanishing tadpoles of order zero (i.e. \eqref{equ} is satisfied),
$\zeta_{\DD+A}(0)=\zeta_{\DD}(0)$ for any one-form $A$.
\end{lemma}

\begin{proof}
Let $a_1,\,a_2,\,b_1,\,b_2\in \A$. Then, with $A_1=a_1[\DD,b_1]$,
 
$$
\ncint A_1 \, \DD^{-1}\, a_2[\DD,b_2]\,\DD^{-1}=\ncint A_1
[\DD^{-1},a_2] [\DD,b_2]D^{-1}+ \ncint A_1 a_2 \DD^{-1}[\DD,b_2] \DD^{-1}\, .
$$

The first term is zero since the integrand is in $OP^{-3}$, while the second term is equal to $\ncint
\big(a_1\a(b_1a_2)-a_1b_1\a(a_2)\big)\big(\a(b_2)-b_2\big)$,  so is
zero using $\a(x)\a(y)=\a(xy)$, $\ncint xy=\ncint x \a(y)$ by \eqref{equ} and the fact that $\ncint$ is a trace. 
Thus $\ncint
\big(A \DD^{-1}\big)^2=0$ and Corollary \ref{difference} yields the result.
\end{proof}

Note that $\zeta_{\DD+A}(0)-\zeta_{\DD}(0)$ is usually non zero:
consider for instance the flat 4-torus and as a generic selfadjoint
one-form $A$, take 
$$
A:= \phi \in[0,2\pi[^4 \,\mapsto -i\gamma^{\a}\,{\sum}_{l\in \Z^4}
\,a_{\a, l}\,e^{\,i\, l^k \phi_k},
$$
where $a_{\a, l}$ is in the Schwartz space $\SS(Z^4)$ and
$a_{\a,l}=-\overline{a_{\a,-l}}$. We have by Lemma \ref{Termaterm}, 
(with $c=\tfrac{8 \pi^2}{3}$,  $\vert l\vert^2={\sum}_k {l^k}^2$ and
$\Th=0$)
\begin{align*}
\zeta_{\DD+A}(0)-\zeta_{\DD}(0)=\ncint (A\DD^{-1})^2=c\,\sum_{l\in
\Z^4} a_{\a_1,l}\, a_{\a_2,-l}\, (l^{\a_1}l^{\a_2} -\delta^{\a_1\a_2} \vert l \vert^2)
\end{align*}
since $\ncint (A \DD^{-1})^4=0$.

This last  equality suggests that Lemma \ref{dim2} can be extended:

\begin{prop}
  \label{AD-1max} 
For any one-form $A$,  $\ncint \,(A\DD^{-1})^{d}=0$ if $d=dim(M)$.
\end{prop}

\begin{proof}
As in the proof of Lemma \ref{dim2}, $\DD^{-1}$ commutes with the
element in the algebra as the integrand is in $OP^{-d}$. So for a
family of $a_i,b_i \in \A$ and using $a:=\prod_{i=1}^d a_i$,
$$
\ncint \prod_{i=1}^d \big(a_i[\DD,b_i]\,D^{-1}\big)=\ncint
\big(\prod_{i=1}^d a_i \big) \prod_{i=1}^d \big(
[\DD,b_i]\,\DD^{-1}\big)=\ncint a \prod_{i=1}^d \big(\a(b_i)-b_i \big).
$$
We obtain, since $\a(b_i)-b_i \in OP^{-1}$, 
$$
\sigma_{-d}^{a\prod_{i=1}^d \a(b_i)-b_i} =  a\,\prod_{i=1}^d
\sigma_{-1}^{\a(b_i)-b_i} =  a\,\prod_{i=1}^d \sigma_{-1}^{\a(b_i)}
.$$
Moreover,  $\sigma_{-1}^{\DD b_i\DD^{-1}}(x_0,\xi)=0$: we already know by
Lemma \ref{symbol-k-1} that $\sigma_{-2}^{\DD^{-1}}(x_0,\xi)=0$, by
\eqref{sigma-1} that $\del_{x^k} \sigma_{-1}^{\DD^{-1}}(x_0,\xi)=0$ for
all $k$,  and $\sigma_0^{\DD b_i}(x_0,\xi) =b_i(x_0) \,
\sigma_0^\DD(x_0,\xi)=0$ giving the claim and the result. 
\end{proof}   

\vspace{0.2cm}
This proposition does not survive in noncommutative spectral triples,  see for instance \cite[Table 1]{MC}.

\medskip

Note that for a one-form $A$, $\ncint A^d\,\DD^{-d} \neq \ncint
(A\,D^{-1})^{-d}=0$: in dimension $d=2$, as in \eqref{d=2},
$$
\ncint A^2\,\DD^{-2}=-2c_d\, \Tr(\ga^k \ga^l)\int_M\,a_{k} a_{l}\, \nu_g.
$$

It is known (see \cite[Proposition 1.153]{ConnesMarcolli}) that the
$d-2$ term (for $d=4$) in the spectral action expansion $\ncint|\DD +
A|^{-2}$ is independent of the perturbation $A$. This is why the
Einstein--Hilbert action $S(\DD)=\ncint |\DD|^{-d+2}=-c \int_M \tau \sqrt{g} \, dx$ 
(see \cite[Theorem 11.2]{Polaris})  is so fundamental. Here $\tau $ is the scalar curvature 
(positive on the sphere) and $c$ is a positive constant. 

We give here another proof of this result. 

\begin{lemma}
\label{scalarcurvature}
We have $\ncint |\DD + A|^{-d+2} = \ncint |\DD|^{-d+2} =-c\int_M \tau \sqrt{g} \, dx$ with $c=\tfrac{d-2}{24} \, \ncint \vert \DD \vert^{-d}$.
\end{lemma}

\begin{proof}
We get from Lemma \ref{residus-particuliers} $(ii)$, the following
equality, where $X:=A\DD+\DD A+A^2$:
$$
\ncint |\DD + A|^{-d+2}-\ncint |\DD|^{-d+2} =
\tfrac{(d-2)}{2}\,\big(\tfrac{d}{4}\ncint X^2 |\DD|^{-d-2} - \ncint X
|\DD|^{-d}\big).
$$
Since the tadpole terms vanish, we have $\ncint X |\DD|^{-d} = \ncint
A^2 |\DD|^{-d}$. Moreover, since mod $OP^1$, $X^2 = (A\DD)^2 + (\DD A)^2 + A\DD^2
A + \DD A^2\DD$, we get with $[\DD^2,A]\in OP^1$,
$$
\ncint X^2 |\DD|^{-d-2} = 2\ncint (A\DD)^2|\DD|^{-d-2} + 2 \ncint A^2 |\DD|^{-d} 
$$ 
which yields
$$
\ncint |\DD + A|^{-d+2}-\ncint |\DD|^{-d+2} =
\tfrac{d(d-2)}{4} \big(\ncint (A\DD)^2 |\DD|^{-d-2} -\tfrac{2-d}{d} \ncint A^2 |\DD|^{-d}\big).
$$
Thus, it is sufficient to check that
$$ 
\int_{S_{x_0}^*M} \Tr\big(\sigma_{-d}((A\DD)^2 |\DD|^{-d-2})(x_0,\xi)\big)\, d\xi= \tfrac{2-d}{d}\int_{S_{x_0}^*M} \Tr\big(\sigma_{-d}(A^2 |\DD|^{-d})\,(x_0,\xi)\big) d\xi.
$$

A straightforward computation yields, with $A=:-i a_\mu \ga^\mu$, and $\sg_1^\DD(x_0,\xi) = \ga^\mu \xi_\mu$,
\begin{align*}
&\int_{S^*_{x_0} M}\sg_{-d}((A\DD)^2 |\DD|^{-d-2})(x_0,\xi)\, d\xi =-
\tfrac{1}{d}\, a_\mu a_\tau \Tr(\ga^{\mu} \ga^\nu \ga^\tau \ga_{\nu})
\,\text{Vol}(S^{d-1})\, ,\\
&\int_{S^*_{x_0} M} \sg_{-d}(A^2 |\DD|^{-d})(x_0,\xi)\, d\xi = -a_\mu a_\tau
\Tr(\ga^\mu \ga^\tau) \,\text{Vol}(S^{d-1})\, .
\end{align*}
Now, $\ncint |\DD + A|^{-d+2} = \ncint |\DD|^{-d+2}$ follows from the equality $\Tr(\ga^{\mu} \ga^\nu
\ga^\tau \ga_{\nu})= (2-d) \Tr(\ga^\mu \ga^\tau)$.
The constant $c$ is given in \cite[Theorem 11.2 and normalization (11.2)]{Polaris}.
\end{proof}
\begin{remark}
In \cite[Definition 1.143]{ConnesMarcolli}, the above result justifies the definition of a scalar curvature for $(\A,\H,\DD)$ as $\mathcal{R}(a):=\ncint a \vert \DD \vert^{-d+2}$ for $a \in \A$. This map is of course a trace on $\A$ for a commutative geometry. But for the triple associated to $SU_q(2)$, this not a trace since
\begin{align*}
\mathcal{R}(aa^*)=\ncint aa^* \, \vert \DD\vert^{-1}= \tfrac{-q^4+6q^2+3}{2{(1-q^2)}^2}  \quad \text{while } \quad
\mathcal{R}(a^*a)=\ncint a^*a \,\vert \DD\vert^{-1}= \tfrac{3q^4+6q^2-1}{2{(1-q^2)}^2} \,.
\end{align*}
\end{remark}

\chapter{Global pseudodifferential calculus on manifolds with linearization}

\section{Introduction}

It has been proven by Gayral et al. in \cite{Gayral2} that Moyal planes are (noncompact) spectral triples. In other words, one can consider Moyal planes as noncommutative Riemannian spin noncompact manifolds. A bridge between the world of deformation quantization and the world of noncommutative geometry has thus been constructed. This suggests that the paradigm of spectral triples is a good environment to deal with quantization problems.
The Moyal product is defined on the 
Schwartz space $\SS(\R^{2n})$ of rapidly decaying functions by
\begin{equation}
f \star g(x) := (\pi\th)^{-2n} \int_{\R^{4n}} \,dy \,dz\,f(y)\, g(z)\,
e^{\frac{2i}{\th}(x-y)\,\.\,S(x-z)}
\end{equation}
where $\th\in \R^*$ and $S =\genfrac{(}{)}{0pt}{1}{0
\  -1_n}{1_n\  0}$, and gives to $\S(\R^{2n})$ a Fr\'{e}chet pre-$C^*$-algebra structure. The noncompact spectral triple described in \cite{Gayral2} is based on this algebra, and extensions to isospectral deformations have been established \cite{GIV,Yang}.

The extension of this remarkable construction to more general symplectic manifolds, for instance when $\R^{2n}$ is replaced by the cotangent bundle $T^*M$ of a general manifold $M$ is an open problem. We propose here the study of a pseudodifferential calculus that allows to extend the construction of the Moyal product to more general spaces. The main idea is to use a global pseudodifferential calculus on a manifold $M$ that gives us a full algebra isomorphism between symbols and operators.

Classically, a pseudodifferential operator on a (smooth, finite dimensional) 
manifold is defined through local charts and the notion of pseudodifferential operator
on open subsets of $\R^n$ \cite{Shubin,Treves}. In this setting, the full symbol of a
pseudodifferential operator is a coordinate dependent notion. However, 
the \emph{principal} symbol can be globally defined as a function on the cotangent bundle.
Naturally, the question of a full coordinate free definition of the symbol calculus of pseudodifferential 
operators on a manifold has been considered. One approach, based on the ideas of Bokobza-Haggiag \cite{Bokobza}, Widom \cite{Widom1,Widom2} 
and Drager \cite{Drager} 
allows such a calculus if one provides the manifold with a linear connection. Parallel transport along geodesics and the exponential map to connect any two points sufficiently close on the manifold are then used for the definitions and properties of local phase functions and oscillatory integrals.
Safarov \cite{Safarov} has formulated a version 
of a full coordinate free symbol calculus and $\la$-quantization ($0\leq \la \leq 1$) using invariant oscillatory integral over the cotangent bundle and 
determined by the linear connection. Pflaum \cite{PflaumQ,Pflaum} developed a complete symbol calculus on any Riemannian manifold using normal coordinates and microlocal lift on the test functions on manifolds with arbitrary Hermitian bundles. Sharafutdinov \cite{Shara1,Shara2} constructed a similar global pseudodifferential calculus, based on coordinate invariant geometric symbols. Further results in the same direction, connection to Weyl quantization and application to physics has been considered in Fulling and Kennedy \cite{Fulling}, Fulling \cite{Fulling0} and G\"{u}nt\"{u}rk \cite{Gunturk}. Connection between complete symbol calculus, deformation quantization and star-products on the cotangent bundle has also been made (see for instance Gutt \cite{Gutt}, Bordemann, Neumaier and Waldmann \cite{Bordeman} and Voronov \cite{Voronov1,Voronov2}). Getzler \cite{Getzler} used a global pseudodifferential calculus in the context of the Atiyah-Singer index theorem on supermanifolds.

All these pseudodifferential calculi are based on symbol (functions of $(x,\th)\in T^*M$) estimates over the covariable $\th$ while the dependence on the variable $x$ is only controlled locally uniformly on compact sets. This is well suited for the case of a compact manifold. For non-compact manifolds, we have to impose a uniform control over $x$ in order to obtain $L^2(M)$ continuity of operators of order 0 and compactness of the remainder operators if the control over $x$ is decaying. In other words, any global pseudodifferential calculus adapted to non-compact manifolds and sensitive to non-local effects needs to encode the behaviour ``at infinity'' of symbols. On the Euclidean space $\R^n$, several types of pseudodifferential calculi have been defined: standard pseudodifferential calculus with uniform control over $x$ (see for instance H\"{o}rmander \cite{Hormander}, Beals \cite{Beals}, Shubin \cite{Shubin3}), isotropic calculus with simultaneous decay of the $x$ and $\th$ variables (Shubin \cite{Shubin,Shubin2}, Melrose \cite{Melrose}), and $SG$-pseudodifferential calculus with separated decay of the $x$ and $\th$ variables (Shubin \cite{Shubin2}, Parenti \cite{Parenti}, Cordes \cite{Cordes1,Cordes2}, Schrohe \cite{Schrohe2}), which is invariant under a special class of diffeomorphisms and can be extended to an adapted class of manifolds, namely the $SG$-manifolds (Schrohe \cite{Schrohe2}). This class of manifolds contains the non-compact manifolds ``with exits'' and adapted pseudodifferential calculus has been developed (see for instance Cordes \cite{Cordes1}, Schulze \cite{Schulze}, Maniccia and Panarese \cite{Maniccia2}). Another approach, based on Lie structures at infinity, has been investigated to study the geometry of pseudodifferential operators on non-compact manifolds. Describing the geometry at infinity of the basis manifold by a Lie algebra of vector fields, an adapted pseudodifferential calculus has been constructed (see for instance Melrose \cite{Melrose2}, Mazzeo and Melrose \cite{Mazzeo}, Ammann, Lauter and Nistor \cite{Ammann}). Let us also mention the groupoid approach: by associating to any manifold with corners a smooth Lie groupoid and by building a pseudodifferential calculus on Lie groupoids, the $b$-calculus of Melrose on manifolds with corners can be generalized (see Monthubert \cite{Monthubert}).

Our purpose is to construct a global pseudodifferential calculus that generalizes the standard and $SG$ calculi on $\R^n$, on manifolds with linearization. These manifolds provide a natural geometric setting to deal simultaneously with the questions of a global isomorphism between symbols and pseudodifferential operators, and the non-local effects associated to non-compact manifolds.

We define in section 2 a manifold with linearization (or exponential manifold) as a pair $(M,\exp)$ where $M$ is a smooth real finite-dimensional manifold and $\exp$ is an abstract exponential map, a smooth map from the tangent bundle onto $M$ that satisfies, besides the usual properties of an exponential map associated to a connection $\nabla$ on $TM$, the property that at each point $x\in M$, $\exp_x$ is a diffeomorphism. Any Cartan--Hadamard manifold with its canonical exponential map is an exponential manifold. These diffeomorphisms are used to define topological vector spaces of functions on the manifold (or on $TM$, $T^*M$, $M\times M$) that generalize, for instance, the notions of rapidly decaying function on $\R^n$ or of tempered distribution, provided that we add a hypothesis of ``$\O_M$-bounded geometry'' on the exponential map.  
In section 3, we use linearizations in the spirit of Bokobza-Haggiag \cite{Bokobza}, to define symbol and quantization maps. This leads to topological isomorphisms between tempered distributional sections on $T^* M$ and $M\times M$, if we consider polynomially controlled (at infinity) linearizations ($\O_M$-linearizations). In particular, we extend the usual (explicit) Moyal product (or $\la$-product, for the $\la$-quantization) on any exponential manifold with $\O_M$-bounded geometry on which we set a $\O_M$-linearization. We get the following $\la$-product formula, giving a Fr\'{e}chet algebra structure to $\S(T^*M)$, 
$$ 
a\circ_\la b\,(x,\eta) = \int_{T_x(M) \times M} d\mu_{x}(\xi)d\mu(y)\int_{V^\la_{x,\xi,y}} d\mu_{x,\xi,y}^*(\th,\th')\, g^\la_{x,\xi,y}\,e^{2\pi i \om^\la_{x,\xi,y}(\eta,\th,\th')} a(y^\la_{x,\xi},\th)\,b(y^{1-\la}_{x,-\xi},\th')
$$
where $a,b\in \S(T^*M)$ and the other notations are detailed in Proposition \ref{la-product}.

In section 4, we define the symbol and amplitudes spaces for our pseudodifferential calculus. Symbol spaces can be defined in an intrinsic way on the exponential manifold with the help of "symbol-like" control ($S_\sg$-bounded geometry, see Definition \ref{ssigmadef}) of the coordinate change diffeomorphisms $\psi_{z,z'}^{\bfr,\bfr'}$ associated to the exponential map $\exp$ on $M$.
For practical reasons the definition of amplitudes here is slightly different from the usual functions of the parameters $x,y$ and $\th$. Instead, our amplitudes generalize functions of the form $(\rx,\zeta,\vth)\mapsto a(\rx,\rx+\zeta,\vth)$, where $a$ is a standard amplitude of the Euclidian pseudodifferential calculus. We establish continuity and regularity results for operators of the following form (which can be seen, for some forms of $\Ga$, as special Fourier integral operators on $\R^n$):  
$$
\langle \Op_{\Ga}(a) , u\rangle :=  \int_{\R^{3n}} e^{2\pi i \langle \vth,\zeta \rangle } \Tr\big(a(\rx,\zeta,\vth)\, \Ga(u)^*(\rx,\zeta)) \, d\zeta\,d\vth\,  d\rx
$$
where $\Ga$ is a topological isomorphism on $\S(\R^{2n},L(E_z))$ (here $E_z$ is a fixed fiber of the Hermitian bundle $E\to M$, so $L(E_z)$ can be identified with $\M_{\dim E_z}(\C)$), $a$ is in a $\O_{f,z}$ space (see Definition \ref{OFZ}) and $u\in \S(\R^n,E_z)$. In particular, results of Proposition \ref{ampliOP} and \ref{amplContinu} and Lemma \ref{noyauReste} are believed to be new.

With the help of a hypothesis of a control of symbol type over the derivative of the linearization ($S_\sg$-linearizations), we obtain in section \ref{pdosection} an intrinsic definition (Theorem \ref{lambdainv}) of pseudodifferential operators $\Psi_{\sg}^{l,m}$ on $M$.
We see in section \ref{linkstd} a condition $(H_V)$ on the linearization that entails that any pseudodifferential operator on $M$, when transferred in a frame $(z,\bfr)$, is a standard pseudodifferential operator on $\R^n$. This condition yields a $L^2$-continuity result in Proposition \ref{L2cont}. 
The last part of section 4 is devoted to the derivation of a symbol product asymptotic formula for the composition of two pseudodifferential operators. The main result is Theorem \ref{compo}: under a special hypothesis $(C_\sg)$ on the linearization (see Definition \ref{Csigma}), we have the following asymptotic formula for the normal symbol (transferred in a frame $(z,\bfr)$) of the product of two pseudodifferential operators
$$
\sigma_{0}(AB)_{z,\bfr} \sim \sum_{\b,\ga \in \N^n} c_\b c_\ga \del_{\zeta,\vth}^{\ga,\ga}\big( a(\rx,\vth)\del^{\b}_{\zeta'}\big(e^{2\pi i \langle \vth,\varphi_{\rx,\zeta}(\zeta')\rangle} (\del^{\b}_{\vth'} f_b)(\rx,\zeta,\zeta',L_{\rx,\zeta}(\vth))\big)_{\zeta'=0} \tau^{-1}_{\rx,\zeta} \big)_{\zeta=0}
$$
where  $a:=\sigma_0(A)_{z,\bfr}$, $b:=\sigma_0(B)_{z,\bfr}$, and other notations are defined in section \ref{composec}.

Finally, we give in Section \ref{exsec} two possible settings (besides the usual standard calculus on the Euclidian $\R^n$) in which the previous calculus applies. The first is based on the Euclidian space $\R^n$, with a ``deformed'' (non-bilinear, non-flat) $S_\sg$-linearization. The second example is the hyperbolic plane (or Poincar\'{e} half-plane) $\HH$. We prove in particular that $\HH$ has a $S_1$-bounded geometry. This allows to define a global Fourier transform, Schwartz spaces $\S(\HH)$, $\S(T^*\HH)$, $\S(T\HH)$, $\B(\HH)$ and the space of symbols $S_1^{l,m}(T^*\HH)$. Moreover we can then define in an intrinsic way a global complete pseudodifferential calculus on $\HH$, and Moyal product, for any specified $S_\sg$-linearization on $\HH$.

\section{Manifolds with linearization and basic function spaces}
\subsection{Abstract exponential maps, definitions and notations}

The notion of linearization on a manifold was first introduced by Bokobza-Haggiag in \cite{Bokobza} and corresponds to a smooth map $\nu$ from $M\times M$ into $TM$ such that $\pi \circ \nu = \pi_1$, $\nu(x,x)=0$ for any $x\in M$ and $(d_y\nu)_{y=x}=\Id_{T_x M}$. In all the following, we shall work with ``global'' linearizations, in the following sense:

\begin{defn} A manifold with linearization (or exponential manifold) is a pair $(M,\exp)$ where $M$ is a smooth manifold and $\exp$ a smooth map from $TM$ into $M$ such that:

\noindent $(i)$ for any $x\in M$, $\exp_x:T_x M \to M$ defined as $\exp_x(\xi):=\exp(x,\xi)$, is a global diffeomorphism between $T_x M$ and $M$,

\noindent $(ii)$ for any $x\in M$, $\exp_x(0)=x$ and $(d\exp_x)_0=\Id_{T_xM}$.

\noindent The map $\exp$ will be called the exponential map, and $(x,y)\mapsto \exp_x^{-1}(y)$ the linearization, of the exponential manifold $(M,\exp)$. We shall sometimes use the shorthand $e_x^\xi:=\exp_x (\xi)$.
\end{defn}

Note that the term ``exponential manifold'' used here is not to be confused with the notion of ``exponential statistical manifold'' used in stochastic analysis. 
Remark that if $\exp\in C^{\infty}(TM,M)$ satisfies $(i)$, then defining $\Exp:=\exp\circ\ T$ where $T(x,\xi):=\exp_{x}^{-1}(x)+(d\exp^{-1}_x)_x\xi$, we see that $(M,\Exp)$ is an exponential manifold.

We will say that $(M,\nabla)$ (resp. $(M,g)$) is exponential, where $M$ is a smooth manifold with connection $\nabla$ on $TM$ (resp. with pseudo-Riemannian metric $g$), if $(M,\exp)$ where $\exp$ is the canonical exponential map associated to $\nabla$ (resp. to $g$) is an exponential manifold, or in other words, if for any $x\in M$, $\exp_x$ is a diffeomorphism from $T_x M$ onto $M$. Note that $(M,\nabla)$ (resp. $(M,g)$) is exponential if and only if
\begin{itemize}
\item $M$ is geodesically complete
\item For any $x,y\in M$, there exists one and only one maximal geodesic $\ga$ such that $\ga(0)=x$ and $\ga(1)=y$.
\item For any $x\in M$, $\exp_x$ is a local diffeomorphism.
\end{itemize}

\begin{rem} $\R^n$ (with its standard metric of signature $(p,n-p)$) is an exponential manifold and any $n$-dimensional real exponential manifold is diffeomorphic to $\R^n$. In particular, an exponential manifold cannot be compact. 
A Cartan--Hadamard manifold is a Riemannian, complete, simply connected manifold with nonpositive sectional curvature. It is a consequence of the Cartan--Hadamard theorem (see for instance \cite[Theorem 3.8]{Lang}) that any Cartan--Hadamard manifold is exponential.
\end{rem}

\begin{rem} The exponential structure can be transported by diffeomorphism: if $(M,\exp_M)$ is an exponential manifold, $N$ a smooth manifold and $\varphi: M\to N$  is a diffeomorphism, then $(N,\exp_N:=\varphi\circ \exp_M \circ \  T\varphi^{-1})$ is an exponential manifold.
\end{rem}
\begin{assum} We suppose from now on that $(M,\exp)$ is an exponential $n$-dimensional real manifold.
\end{assum}

For any $x,y\in M$, we define $\ga_{xy}$ as the curve $\R\to M$, $t\mapsto \exp_{x}(t\exp_x^{-1}y)$, and $\wt\ga_{xy}(t):=\ga_{yx}(1-t)$. Note that $\ga_{xy}(0)=x$ and $\ga_{xy}(1)=y$. If the exponential map is derived from a linear connection, we have for any $t\in \R$, $\ga_{xy}(t)=\wt \ga_{xy}(t)$. In the general case, this is only true for $t=0$ and $t=1$. 

The abstract exponential map $\exp$ provides the manifold $M$ with a notion of ``points at infinity'' and ``straight lines'' ($\ga_{xy}$). It can be seen as a generalization to manifolds of the useful properties of $\R^n$ for the study of the behaviour of functions at infinity. The abstract exponential map $\exp$ formalizes the fact that our straight lines never stop and connect any two different points.

The diffeomorphism $\exp_z^{-1}$, for a given $z\in M$, is not stricto sensu a chart, since it maps $M$ onto $T_z M$, which is diffeormorphic but not equal to $\R^n$. In order to obtain a chart, one needs to choose a linear basis of $T_z M$. If $z\in M$ and $\bfr$ is a basis of $T_z M$ we will call the pair $(z,\bfr)$ a (normal) frame. For any frame $(z,\bfr)$, we define $n_z^{\mathfrak{b}}:=L_\bfr\circ \exp_z^{-1}$ with $L_\bfr$ the linear isomorphism from $T_z M$ onto $\R^n$ associated to $\mathfrak{b}$. As a consequence, the pair $(M,n_z^\bfr)$ is a chart which is a global diffeomorphism from $M$ onto $\R^n$.

We denote $\psi_{z,z'}^{\bfr,\bfr'}:= n_z^\bfr \circ (n_{z'}^{\bfr'})^{-1}$ the normal coordinate change diffeomorphism from $\R^n$ onto $\R^n$ and
$(\del_{i,z,\bfr})_{i\in \N_n}$ and $(dx^{i,z,\bfr})_{i\in \N_n}$ (whith $\N_n:=\set{1,\cdots,n}$) the global frame vector fields and 1-forms associated to the chart $n_z^\bfr$. We also note $n^\bfr_{z,*}$ the diffeomorphism from $T^*M$ onto $\R^{2n}$ defined by $n^\bfr_{z,*}(x,\th)=(n^\bfr_{z}(x),\wt M_{z,x}^\bfr(\th))$ where $(\wt M_{z,x}^\bfr(\th)_i)_{i\in \N_n}$ are the components of $\th$ in $(dx^{i,z,\bfr}_{x})_{i\in \N_n}$ and $n_{z,T}^\bfr : (x,\xi)\to (n_z^\bfr(x), M^\bfr_{z,x}(\xi))$ the diffeomorphism from $TM$ onto $\R^{2n}$, where $( M ^\bfr_{z,x}(\xi)_i)_{i\in \N_n}$ are the coordinates of $\xi$ in the basis $({\del_{i,z,\bfr}}_x)_{i\in \N_n}$. We have $M^{\bfr}_{z,x}  = (dn_{z}^\bfr)_x$ and $\wt M^{\bfr}_{z,x}  =\, ^t(dn_{z}^\bfr)_x^{-1}$.
The diffeomorphism from $M\times M$ onto $\R^{2n}$ defined by $(x,y)\mapsto (n_z^\bfr(x),n_z^\bfr(y))$ will be noted $n_{z,M^2}^\bfr$.

We denote $(\del_{i,z,\bfr})_{i\in \N_{2n}}$ the family of vector fields on $T^*M$ (resp. $TM$, $M\times M$) associated to the chart $n^\bfr_{z,*}$ (resp. $n^\bfr_{z,T}$, $n^\bfr_{z,M^2}$) onto $\R^{2n}$. 
We suppose in all the following that $\mathfrak{E}$ is an arbitrary normed finite dimensional complex vector space. If $\nu$ is a ($2n$)-multi-index, we define the following operator on $C^{\infty}(T^*M,\mathfrak{E})$ (resp. $C^{\infty}(TM,\mathfrak{E})$, $C^{\infty}(M\times M,\mathfrak{E})$): 
$$
\del_{z,\bfr}^\nu:= \prod_{k=1}^{2n}\,\del_{k,z,\bfr}^{\nu_k}.
$$
If $\a$ and $\beta$ are $n$-multi-indices, we denote $(\a,\beta)$ the $2n$-multi-index obtained by concatenation. If $\a$ is a $n$-multi-index, $\del_{z,\bfr}^\a$ is a linear operator on $C^\infty(M,\mathfrak{E})$. 
We fix the shorcut $\langle \rx \rangle:=(1+\norm{\rx}^2)^{1/2}$ for any $\rx\in \R^p$, $p\in \N$. 
We will use the convention $\rx^\a := \rx_1^{\a_1} \cdots \rx_p^{\a_p}$ for $\rx\in \R^p$ and $\a$ $p$-multi-index, with $0^0:=1$. 
If $f$ is continuous function from $\R^p$ to a normed vector space and $g$ is a continuous function from $\R^p$ to $\R$, we denote $f=\O(g)$ if and only if there exist $r>0$, $C>0$ such that for any $\rx\in \R^{p}\backslash B(0,r)$, $\norm{f(\rx)}\leq C |g(\rx)|$. In the case where $g$ is strictly positive on $\R^p$, this is equivalent to: there exists $C>0$ such that for any  $\rx\in \R^{p}$, $\norm{f(\rx)}\leq C g(\rx)$.  
We also introduce the following shorthands, for given $(z,\bfr)$, $x,y\in M$, $\th\in T^*_x (M)$, $\xi\in T_x(M)$:
\begin{align*}
&\langle x\rangle_{z,\bfr} := \langle{n_{z}^\bfr(x)}\rangle,\qquad
\langle \th \rangle_{z,\bfr,x} := \langle{\wt M_{z,x}^\bfr(\th)}\rangle,\qquad
\langle \xi \rangle_{z,\bfr,x} := \langle{ M_{z,x}^\bfr(\xi)}\rangle, \\
&\langle x,y\rangle_{z,\bfr}:=\langle (n_z^\bfr(x),n_z^\bfr(y))\rangle,\quad  \langle x,\th\rangle_{z,\bfr}:=\langle (n_z^\bfr(x),\wt M_{z,x}^{\bfr}(\th))\rangle, \quad \langle x,\xi\rangle_{z,\bfr}:=\langle (n_z^\bfr(x), M_{z,x}^{\bfr}(\xi))\rangle\, .
\end{align*}
If $f$ and $g$ are in $C^0(\R^p, \R^{p'})$ we denote $f\asymp g$ the equivalence relation defined by: $\langle f\rangle =\O(\langle g\rangle)$ and $\langle g \rangle =\O(\langle f \rangle )$.
 
\subsection{Parallel transport on an Hermitian bundle}

Let $E$ be an hermitian vector bundle (with typical fiber $\mathbb{E}$ as a finite dimensional complex vector space) on the exponential manifold $(M,\exp)$. $E$ admits a (non-unique) connection $\nabla^E$ compatible with the hermitian metric \cite{Berline}. It is a differential operator from $C^\infty(M,E)$ (the space of smooth sections of $E\to M$) to $C^\infty(M,T^*M\otimes E)$
such that for any smooth function $f$ on $M$ and smooth $E$-sections $\psi$, $\psi'$, 
\begin{align*}
&\nabla^E(f\psi)= df\ox\psi+ f\nabla^E \psi \, ,\\
&d(\psi|\psi')= (\nabla^E \psi|\psi')+(\psi|\nabla^E\psi')\, ,
\end{align*}
where $(\psi|\psi')$ is the hermitian pairing of $\psi$ and $\psi'$. We will note $|\psi|^2:=(\psi|\psi)$. The sesquilinear form $(\cdot|\cdot)_x$ of $E_x$ is antilinear in the second variable by convention. The operator $\nabla^E$ can be (uniquely) extended as an operator acting on $E$-valued differential forms on $M$. 
If $\gamma$ is a curve on $M$ defined on an interval $J$ and $\ga^*E$ the associated pullback bundle on $J$, there exists a natural connection (the pullback of $\nabla^E$) on $\ga^*E$, noted $\nabla^{\ga^*E}$ compatible with $\nabla^E$.

Let us fix $x,y\in M$ and $\ga :\, J \to M$ a curve such that $\ga(0)=x$ and $\ga(1)=y$. For any $v\in E_x$, there exists an unique smooth section $\beta$ of $\ga^*E\to J$ such that $\beta(0)=v$ and $\nabla^{\ga^*E} \beta=0$. Clearly, $\beta(1) \in E_{y}$ and we can define a linear isomorphism $\tau_{\ga}$ from $E_x$ to $E_y$ as $\tau_{\ga}(v)=\beta(1)$.
The map $\tau_{\ga}$ is the parallel transport map associated to $\ga$ from $E_x$ to $E_y$. 
The compatibility of $\nabla^E$ with the hermitian metric entails that the maps $\tau_{\ga}$ are in fact isometries for the hermitian structures on $E_x$ and $E_y$. 

The vector bundle $L(E)\to M$, defined by $L(E)_x:=L(E_x)$ (the space of endomorphisms on $E_x$), is lifted to $T^*M$, $TM$ and $M\times M$ by setting the fiber at $(x,\th)$ to $L(E_x)$ for $T^*M$ or $TM$, and the fiber at $(x,y)$ to $L(E_y,E_x)$ for $M\times M$. The canonical projection from $T^*M$ or $TM$ to $M$ is noted $\pi$.

We denote $\tau_{xy}:=\tau_{\ga_{xy}}$. Remark that $\tau_{xy}^{-1}=\tau_{\wt\ga_{yx}}$.  
We define $\tau_z : x \mapsto \tau_{zx}$ and $\tau_z^{-1} : x\mapsto \tau_{zx}^{-1}=\tau_{zx}^*$. 

If $u\in C^\infty(M,E)$ and $z\in M$, we denote $u^z (x):= (\tau_{z}^{-1}\, u)(x)$ for any $x\in M$.
If $a$ is a section of $L(E)\to T^*M$ or $L(E)\to TM$, we denote $a^z := (\tau_z^{-1} \circ \pi) \, a\,(\tau_{z}\circ \pi)$. If $a$ is a section of $L(E)\to M\times M$, we denote $a^z(x,y):= \tau_z^{-1}(x)\,a(x,y)\, \tau_{z}(y)$.
We also define $\tau^z:= (x,y)\mapsto \tau_z^{-1}(y)\tau(x,y) \tau_z(x) \in L(E_z)$. Noting $\pi_1(x,y):=x$, $\pi_2(x,y):=y$, we get $a^z= (\tau_z^{-1}\circ \pi_1)\,a\, (\tau_z\circ \pi_2) $ and $\tau^z=(\tau_z^{-1}\circ \pi_2)\, a\, (\tau_z\circ \pi_1)$.

Parallel transport on $E$ has the following smoothness property: 

\begin{lem} 
(i) 
The map $\tau : (x,y) \mapsto \tau_{xy}$ (resp. $\tau^{-1} : (x,y) \mapsto \tau_{xy}^{-1}$) is a smooth section of the vector bundle $L(E)^\vee\to M\times M$ where the fiber at $(x,y)$ is $L(E_x,E_y)$ (resp. of the vector bundle $L(E)\to M\times M$).

\noindent (ii) $\tau_z \in C^\infty(M,L(E_z,E))$ and $\tau_z^{-1} \in C^{\infty}(M,L(E,E_z))$.

\noindent (iii) $\tau^z \in C^\infty(M\times M,L(E_z))$.
\end{lem}

\begin{proof} $(i)$ The map $G: TM \to M\times M$ defined by $G(v):= (\pi(v),\exp(v))$ is a local diffeomorphism since the Jacobian of $G$ at $v_0=(x_0,\xi_0)\in TM$ is equal to the Jacobian of $\exp_{x_0}$ at $\xi_0$. Since it is also bijective (with inverse $G^{-1}(x,y):= (x,\exp_x^{-1}(y))$), it is a (global) diffeomorphism $TM\to M\times M$.
The map $b(x,y,t):= (x,t\exp_x^{-1}(y))$ is thus a smooth map from $M\times M\times \R$ to $TM$, and we get a smooth parametrization by $M\times M$ of the following family of curves: $c (x,y)\mapsto \big(\ga_{xy}: t\mapsto \exp b(x,y,t))$. This parametrization leads (see \cite[p. 17]{Dumitrescu}) to a smooth bundle homomorphism between $c^*(\cdot)(0)E\to M\times M$ and $c^*(\cdot)(1) E\to M\times M$, so a smooth section $\tau:(x,y)\mapsto \tau_{xy}$ of $L(E_x,E_y)\to M\times M$. The case of $\tau^{-1}$ is similar, by taking $b^{-1}(x,y,t):=b(x,y,1-t)$.

\noindent $(ii,iii)$ are straightforward consequences of $(i)$.
\end{proof}

\begin{cly} If $u$ is in the space $C^\infty(M,E)$, then $u^z \in C^\infty(M,E_z)$. Similarly, if $a\in C^\infty({T^*M,L(E)})$ (resp. $C^\infty(TM,L(E))$, $C^\infty(M\times M, L(E))$), then $a^z \in C^\infty({T^*M, L(E_z)})$ (resp. $C^\infty({TM, L(E_z)})$, $C^\infty(M\times M ,L(E_z))$). 
\end{cly}

\begin{rem} The vector bundle $E$ on $M$ is trivializable and the parallel transport provides a $M$-indexed family of trivializations, since for any $z\in M$, the pair $f_z :E\mapsto M \times \mathbb{E} , (x,v)\mapsto (x,\tau_{xz}(v))$, $\Id : M\mapsto M, x\mapsto x$, is a vector bundle isomorphism 
from $E\to M$ onto $M \times \mathbb{E}\to M$. Note that if $\exp$ is derived from a connection, $\tau_{xy}^{-1}=\tau_{yx}$ for any $x,y\in M$.
\end{rem}

\subsection{$\O_M$ and $S_\sigma$-bounded geometry}

Classically, in Riemannian geometry, bounded geometry hypothesis gives boundedness on the covariant derivative of the Riemann curvature of the basis manifold.
For the following pseudodifferential calculus, we shall need some hypothesis of that kind, formulated not with the curvature but with the exponential diffeomorphisms (``normal'' coordinate transition maps). 
The hypothesis that we will need for pseudodifferential symbol calculus is actually not simply the boundedness condition on the derivatives of the transition maps, which is a classical consequence of bounded geometry. For symbol calculus, we will require that the $n^{\text{th}}$-derivatives are not only bounded, but decrease to zero at infinity as $\norm{x}^{-\sigma (n-1)}$ where $\sigma$ is a parameter in $[0,1]$. Or, in other words, the normal coordinate change maps behave as ``symbols'' or order 1. Thus, we introduce the following

\begin{defn}
\label{ssigmadef}
Let $\sigma \in [0,1]$. The exponential manifold $(M,\exp)$ is said to have a $S_\sigma$-bounded geometry if for any $(z,\bfr)$, $(z',\bfr')$, and any $n$-multi-index $\a\neq 0$, 
\begin{align*}
&(S_\sigma 1)\qquad \del^{\a} \psi_{z,z'}^{\bfr,\bfr'}(\rx) = \mathcal{O}(\langle \rx \rangle^{-\sigma(|\a|-1)})\, ,
\end{align*}
and a $\O_M$-bounded geometry if for any $(z,\bfr)$, $(z',\bfr')$, and any $n$-multi-index $\a$, there exist $p_\a\geq 1$ such that
\begin{align*}
&(\O_M 1)\qquad \del^{\a} \psi_{z,z'}^{\bfr,\bfr'}(\rx) = \mathcal{O}(\langle \rx \rangle^{p_\a})  \, .
\end{align*}
\end{defn}

We shall be working with $\O_M$-bounded geometry for the definition of function spaces and Fourier transform and with $S_{\sigma}$-bounded geometry (for a $\sigma \in [0,1]$) for pseudodifferential symbol calculus. 

\begin{defn} The triple $(M,\exp,E)$ where $(M,\exp)$ is exponential and $E$ is a hermitian vector bundle on $M$ has a $S_\sigma$-bounded geometry if $(M,\exp)$ has a $S_\sigma$-bounded geometry and for any $(z,\bfr)$, $z',z''$, and any $n$-multi-index $\a$,  
$$
(S_\sigma 2)\qquad \del_{z,\bfr}^\a \tau_{z'}^{-1}\tau_{z''}(x) = \mathcal{O}( \langle x \rangle_{z,\bfr}^{-\sigma |\a|}) \, ,
$$
and a $\O_M$-bounded geometry if $(M,\exp)$ has a $\O_M$-bounded geometry and for any $(z,\bfr)$, $(z',\bfr')$, and any $n$-multi-index $\a$, there exist $p_\a\geq 1$ such that
\begin{align*}
&(\O_M 2)\qquad \del_{z,\bfr}^\a \tau_{z'}^{-1}\tau_{z''}(x) = \mathcal{O}(\langle x \rangle_{z,\bfr}^{p_\a})  \, .
\end{align*}
\end{defn}

Clearly, if $\sigma\leq \sigma'$, since $(S_{\sigma'} i)\Rightarrow (S_\sigma i)$, we have $S_{\sigma'} $-bounded $\Rightarrow$ $S_\sigma$-bounded $\Rightarrow$ $\O_M$-bounded. 
Note that $S_\sigma$-bounded geometry on the vector bundle entails that the derivatives of the transport transition maps $\tau_{z}^{-1}\tau_{z'}$ (smooth from $M$ to $L(E_{z'},E_z$)) are bounded (for $S_0$-bounded geometry) or decrease to zero with an order equal to the order of the derivative (for $S_1$-bounded geometry). Remark also that if $E$ is a trivial bundle and $\nabla^E=d$, then $(S_1 2)$ is automatically satisfied since the maps $\tau_{z}$ are all equal to the constant $x\mapsto Id_{\mathbb{E}}$.

\begin{lem}
\label{B1-lem}
Let $\sigma\in [0,1]$ and $(z,\bfr)$, $(z',\bfr')$ be given frames.

\noindent (i) If $(M,\exp)$ has a $S_\sigma$-bounded geometry, there exist $K,C,C'>0$ such that for any $\rx \in \R^n$, $x\in M$, $\th\in T^*_x(M)$, $\xi \in T_x(M)$,
\begin{align}
 \psi_{z,z'}^{\bfr,\bfr'} \asymp \Id_{\R^n}  \qquad \text{ and } \qquad \langle x \rangle_{z,\bfr} \leq K \langle x \rangle_{z',\bfr'}\, ,\label{i-first} \\
 \langle \th \rangle_{z,\bfr,x} \leq C \langle \th \rangle_{z',\bfr',x} \qquad \text{ and } \qquad \langle \xi \rangle_{z,\bfr,x} \leq C' \langle \xi \rangle_{z',\bfr',x}\, ,\label{i-bis}
\end{align}
and if $(M,\exp)$ has a $\O_M$-bounded geometry, there exist $K,K',K'',C,C'>0$ and $q\geq 1$ such that for any $\rx \in \R^n$, $x\in M$, $\th\in T^*_x(M)$, $\xi \in T_x(M)$,
\begin{align}
K' \langle \rx \rangle^{1/q} \leq \langle \psi_{z,z'}^{\bfr,\bfr'}(\rx)\rangle \leq K''\langle \rx \rangle^q  \qquad \text{ and } \qquad \langle x \rangle_{z,\bfr} \leq K \langle x \rangle^q_{z',\bfr'}\, ,\label{i-first-om} \\
 \langle \th \rangle_{z,\bfr,x} \leq C \langle x \rangle^q_{z',\bfr'}\langle \th \rangle_{z',\bfr',x} \qquad \text{ and } \qquad \langle \xi \rangle_{z,\bfr,x} \leq C' \langle x \rangle^q_{z',\bfr'}\langle \xi \rangle_{z',\bfr',x}\, ,\label{i-bis-om}
\end{align}

\noindent (ii) For any given $n$-multi-indice $\a$, we can write 
$$
\del_{z,\bfr}^{\a} = \sum_{0\leq |\a'|\leq |\a|} f_{\a,\a'}\, \del_{z',\bfr'}^{\a'}
$$
where the $(f_{\a,\a'})$ are smooth real functions on $M$ such that for each $n$-multi-indices $\a,\a'$,   

(a) if $(M,\exp)$ has a $S_\sigma$-bounded geometry, there exists $C_{\a}>0$ such that for any $x\in M$,
$| f_{\a,\a'} (x)| \leq C_{\a} \langle x\rangle_{z,\bfr}^{-\sigma (|\a|-|\a'|)}$,

(b) if $(M,\exp)$ has a $\O_M$-bounded geometry, there exist $C_{\a}>0$ and $q_{\a}\geq 1$ such that for any $x\in M$, $| f_{\a,\a'} (x)| \leq C_{\a} \langle x\rangle_{z,\bfr}^{q_{\a}}$.

\end{lem}
\begin{proof} $(i)$   Suppose that $(M,\exp)$ has a $S_\sigma$-bounded geometry. Taylor formula implies that 
$\norm{\psi_{z,z'}^{\bfr,\bfr'}(\rx)} \leq \norm{\psi_{z,z'}^{\bfr,\bfr'}(0)}  + C_0 \norm{\rx}$ for any $\rx\in \R^n$, where $C_0:= \sup_{\rx\in \R^n} \norm{(d\psi_{z,z'}^{\bfr,\bfr'})_\rx}$. As a consequence $\psi_{z,z'}^{\bfr,\bfr'}(\rx)=\O(\norm{\rx})$ and thus, there is $K''>0$ such that $\langle \psi_{z,z'}^{\bfr,\bfr'}(\rx) \rangle \leq K'' \langle \rx \rangle$. The same argument for $\psi_{z',z}^{\bfr',\bfr}=(\psi_{z,z'}^{\bfr,\bfr'})^{-1}$ gives $\psi_{z,z'}^{\bfr,\bfr'} \asymp \Id_{\R^n}$ and $\langle x \rangle_{z,\bfr} \leq K \langle x \rangle_{z',\bfr'}$ follows immediately. 
Since $x\mapsto \norm{\wt M_{z,x}^\bfr ( \wt M_{z',x}^{\bfr'})^{-1}}=\norm{(d \psi_{z',z}^{\bfr',\bfr})_{n_z^\bfr(x)}}$ and $x\mapsto \norm{ M_{z,x}^\bfr (  M_{z',x}^{\bfr'})^{-1}}=\norm{(d \psi_{z,z'}^{\bfr,\bfr'})_{n_{z'}^{\bfr'}(x)}}$ are bounded functions, (\ref{i-bis}) follows. The case where $(M,\exp)$ has a $\O_M$-bounded geometry is similar.

\noindent $(ii)$ We have for any $f\in C^\infty(M,\mathfrak{E})$, 
$$
\del_{z,\bfr}^\a (f) = \del^\a(f\circ (n_{z}^\bfr)^{-1}) \circ n_z^\bfr = \del^\a(f\circ (n_{z'}^{\bfr'})^{-1} \circ \psi_{z',z}^{\bfr',\bfr}) \circ n_z^\bfr\, .
$$
We now apply the multivariate Faa di Bruno formula obtained by G.M. Constantine and T.H. Savits in \cite{Constantine}, that we reformulated for convenience in Theorem \ref{FaaCS}. This formula entails that for any $n$-multi-index $\a\neq 0$,
$$
\del^\a(f\circ (n_{z'}^{\bfr'})^{-1} \circ \psi_{z',z}^{\bfr',\bfr}) =  \sum_{1\leq |\a'|\leq |\a|} P_{\a,\a'}(\psi_{z',z}^{\bfr',\bfr})\,(\del^{\a'} f\circ (n_{z'}^{\bfr'})^{-1} )\circ \psi_{z',z}^{\bfr',\bfr}\  
$$
and thus 
$$
\del_{z,\bfr}^\a =  \sum_{1\leq |\a'|\leq |\a|} (P_{\a,\a'}(\psi_{z',z}^{\bfr',\bfr})\circ n_z^\bfr) \ \del_{z',\bfr'}^{\a'} =:\sum_{1\leq |\a'|\leq |\a|} f_{\a,\a'} \ \del_{z',\bfr'}^{\a'}
$$ 
where $P_{\a,\a'}(g)$ is a linear combination of terms of the form $\prod_{j=1}^{s} (\del^{l^j} g)^{k^j} $, where $1\leq s\leq |\a|$ and  
the $k^j$ and $l^j$ are $n$-multi-indices with $|k^j|>0$, $|l^j|>0$, $\sum_{j=1}^s |k^j| = |\a'|$ and $\sum_{j=1}^s|k^j||l^j|=|\a|$. In the case where $(M,\exp)$ has a $S_\sigma$-bounded geometry, for each $s,(k^j),(l^j)$, there is $K>0$ such that for any $\rx\in \R^n$,
$$
|\prod_{j=1}^{s} (\del^{l^j} \psi_{z',z}^{\bfr',\bfr})^{k^j}(\rx)|\leq K \langle \rx \rangle^{-\sigma \sum_{j=1}^s(|l^j|-1)|k^j|} = K \langle \rx \rangle^{-\sigma (|\a|-|\a'|)}  
$$
which gives the result. The case where $(M,\exp)$ has a $\O_M$-bounded geometry is similar.
\end{proof}

\begin{thm}\cite{Constantine}
\label{FaaCS} 
Let $f\in C^\infty(\R^p,\mathfrak{E})$ and $g\in C^\infty(\R^n,\R^p)$. Then for any $n$-multi-index $\nu\neq 0$,
$$
\del^\nu (f\circ g) = \sum_{1\leq |\la|\leq |\nu|} (\del^\la f)\circ g\  \sum_{s=1}^{|\nu|} \sum_{p_s(\nu,\la)}\nu! \prod_{j=1}^s \tfrac{1}{k^j!(l^j!)^{|k^j|}} (\del^{l^j} g)^{k^j}
$$
where $p_s(\nu,\la)$ is the set of $p$-multi-indices $k^j$ and $n$-multi-indices $l^j$ ($1\leq j\leq s$) such that $0\prec l^1 \prec \cdots\prec l^s$ ($l\prec l'$ being defined as ``$|l|<|l'|$ or $|l|=|l'|$ and $l<_Ll'$'' where $<_L$ is the strict lexicographical order), $|k^j|>0$, $\sum_{j=1}^s k^j = \la$ and $\sum_{j=1}^s |k^j| l^j= \nu$.
\end{thm}

Note that by Lemma \ref{B1-lem}, if $(M,\exp)$ satisfies $(S_\sigma 1)$ (resp. $(\O_M 1)$), then $(S_\sigma 2)$ (resp. $(\O_M 2)$) is equivalent to: for any $z',z''\in M$, there exists a frame $(z,\bfr)$ such that $
\del^\a_{z,\bfr} \tau_{z'}^{-1}\tau_{z''} (x) =\O(\langle x \rangle_{z,\bfr}^{-\sigma |\a|})$ (resp. $\O(\langle x \rangle_{z,\bfr}^{p_{\a}})$ for a $p_\a\geq 1$)
for any $n$-multi-index $\a$. 

As the following proposition shows, $S_\sigma$ or $\O_M$-bounded geometry properties can be transported by any diffeomorphism.

\begin{prop}
\label{ssigmDiff}
If $(M,\exp_M)$ has a $S_\sigma$ (resp. $\O_M$) bounded geometry, $N$ a smooth manifold and $\varphi: M\to N$  is a diffeomorphism, then $(N,\exp_N:=\varphi\circ \exp_M \circ \  d\varphi^{-1})$ has a $S_\sigma$ (resp. $\O_M$) bounded geometry.
\end{prop}
\begin{proof} Let us note $\psi_{z,z',N}^{\bfr,\bfr'}:=n_{z,N}^\bfr \circ (n_{z',N}^{\bfr'})^{-1}$ where $n_{z,N}^\bfr:= L_\bfr \circ \exp_{N,z}^{-1}$ and $(z,\bfr)$, $(z',\bfr')$ are two frames on $N$.
Since $\exp_{z',N}=\varphi \circ \exp_{M,\varphi^{-1}(z')}\circ (d\varphi^{-1})_{z'}$ and $\exp_{N,z}^{-1}= (d\varphi^{-1})_z^{-1}\circ \exp_{M,\varphi^{-1}(z)}^{-1}\circ\ \varphi^{-1}$, we obtain $\psi_{z,z',N}^{\bfr,\bfr'} = \psi_{\varphi^{-1}(z),\varphi^{-1}(z'),M}^{\bfr_z,\bfr'_{z'}}$ where $\bfr_z$ is the basis of $T_{\varphi^{-1}(z)}(M)$ such that $L_{\bfr_z}= L_\bfr\circ (d\varphi)_{\varphi^{-1}(z)}$. The result follows.
\end{proof}

The following technical lemma will be used for Fourier transform and the definition of rapidly decreasing section spaces over the tangent and cotangent bundle in section 3. It will also give the behaviour of symbols under coordinate change. 

\begin{lem}
\label{cchange} Let $(z,\bfr)$, $(z',\bfr')$ be given frames.

\noindent (i) We can express $\del_{z,\bfr}^{(\a,\b)}$ as an operator on $C^{\infty}(T^*M,\mathfrak{E})$ (resp. $C^{\infty}(TM,\mathfrak{E})$) , where $(\a,\b)$ is a $2n$-multi-index,  with the following finite sum:
$$
\del_{z,\bfr}^{(\a,\b)} = \underset{|\b'|\geq |\b|}{\sum_{0\leq |(\a',\b')|\leq |(\a,\b)|}}\, f_{\a,\b,\a',\b'}\,   \del_{z',\bfr'}^{(\a',\b')} 
$$
where the $f_{\a,\b,\a',\b'}$ are smooth functions on $T^*M$ (resp. $TM$) such that

(a) if $(M,\exp)$ has a $S_\sigma$-bounded geometry for a given $\sigma\in [0,1]$, there exists $C_{\a,\b}>0$ such that for any $(x,\th)\in T^*M$ (resp. $TM$),

\begin{equation}
\label{fabab}
|f_{\a,\b,\a',\b'}(x,\th)|\leq    C_{\a,\b}\langle x \rangle_{z,\bfr }^{\sigma(|\a'|-|\a|)}\langle \th \rangle_{z,\bfr,x}^{|\b'|-|\b|}.
\end{equation}

(b) if $(M,\exp)$ has a $\O_M$-bounded geometry, there exist $C_{\a,\b}>0$ and $q_{\a,\b}\geq 1$ such that for any $(x,\th)\in T^*M$ (resp. $TM$),

\begin{equation}
\label{fabab2}
|f_{\a,\b,\a',\b'}(x,\th)|\leq    C_{\a,\b}\langle x \rangle_{z,\bfr }^{q_{\a,\b}}\langle \th \rangle_{z,\bfr,x}^{|\b'|-|\b|}.
\end{equation}

\noindent (ii) We can express $\del_{z,\bfr}^{(\a,\b)}$ as an operator on $C^{\infty}(M\times M,\mathfrak{E})$,  with the following finite sum:
$$
\del_{z,\bfr}^{(\a,\b)} = \underset{0\leq |\b'|\leq |\b|}{\sum_{0\leq |\a'|\leq |\a|}}\, f_{\a,\b,\a',\b'}\,   \del_{z',\bfr'}^{(\a',\b')} 
$$
where the $f_{\a,\b,\a',\b'}$ are smooth functions on $M\times M$ such that 

(a) if $(M,\exp)$ has a $S_\sigma$-bounded geometry for a given $\sigma\in [0,1]$, there exists $C_{\a,\b}>0$ such that for any $(x,y)\in M\times M$,

\begin{equation}
\label{fabab3}
|f_{\a,\b,\a',\b'}(x,y)|\leq    C_{\a,\b}\,\langle x \rangle_{z,\bfr }^{\sigma(|\a'|-|\a|)}\,\langle y \rangle_{z,\bfr}^{\sigma(|\b'|-|\b|)}.
\end{equation}

(b) if $(M,\exp)$ has a $\O_M$-bounded geometry, there exist $C_{\a,\b}>0$ and $q_{\a},q_{\b}\geq 1$ such that for any $(x,y)\in M\times M$,

\begin{equation}
\label{fabab4}
|f_{\a,\b,\a',\b'}(x,y)|\leq    C_{\a,\b}\,\langle x \rangle_{z,\bfr}^{q_{\a}}\,\langle y \rangle_{z,\bfr}^{q_{\b}}.
\end{equation}

\end{lem}
\begin{proof} 
$(i)$ Suppose that $(M,\exp)$ has a $S_\sigma$-bounded geometry.
Let us note $\psi_*:=n_{z',*}^{\bfr'}\circ (n_{z,*}^{\bfr})^{-1}$ and $\psi_T:=n_{z',T}^{\bfr'}\circ (n_{z,T}^{\bfr})^{-1}$.
We have $\psi_*=(\psi_{z',z}^{\bfr',\bfr}\circ \pi_1,L)$ where $\pi_1$ is the projection from $\R^{2n}$ onto the first copy of $\R^n$ in $\R^{2n}$ and $L$ is the smooth map from $\R^{2n}$ to $\R^n$ defined as $L(\rx,\vth):=\, ^t(d\psi_{z',z}^{\bfr',\bfr})^{-1}_{\rx}(\vth)=\, ^t(d\psi_{z,z'}^{\bfr,\bfr'})_{\psi_{z',z}^{\bfr',\bfr}(\rx)}(\vth)$. Noting $(L_i)_{1\leq i\leq n}$ the components of $L$, we have $L_i(\rx,\vth)=\sum_{1\leq p\leq n} L_{i,p}(\rx)\,\vth_{p}$, where $L_{i,p}:=(\del_i \psi_{z,z'}^{\bfr,\bfr'})_p\circ \psi_{z',z}^{\bfr',\bfr}$. As a consequence, for $1\leq i\leq n$ and $\a,\b$, $n$-multi-indices such that $|(\a,\b)|>0$
\begin{align*}
&(\del^{(\a,\b)}\psi_*)_i=  \delta_{\b,0}(\del^{\a}\psi_{z',z}^{\bfr',\bfr})_i \circ \pi_1\, ,\qquad (\del^{(\a,\b)}\psi_*)_{n+i}=(\del^{(\a,\b)}L)_i\, , \\
&(\del^{(\a,\b)}L)_i(\rx,\vth) = \sum_{1\leq p\leq n} (\del^{\a} L_{i,p})(\rx)\  F_{\b,p}(\vth)\,,\\ 
&\del^{\a} L_{i,p} = \sum_{1\leq |\a'|\leq |\a|} P_{\a,\a'}(\psi_{z',z}^{\bfr',\bfr})\ ((\del^{\a'+e_i} \psi_{z,z'}^{\bfr,\bfr'})_p\circ \psi_{z',z}^{\bfr',\bfr})\quad \text{ if } \quad |\a|>0 \, ,
\end{align*}
where $F_{\b,p}(\vth)$ is equal to $\vth_p$ if $\b=0$, to $\delta_{p,r}$ if $\b=e_r$, and to 0 otherwise. We get from the proof of Lemma \ref{B1-lem} that (for $1\leq |\a'|\leq |\a|$) $P_{\a,\a'}(\psi_{z',z}^{\bfr',\bfr})(\rx)=\O(\langle \rx \rangle^{-\sigma(|\a|-|\a'|)})$. As a consequence, using (\ref{i-first}), we see that $\del^{\a} L_{i,p}(\rx)=\O(\langle \rx \rangle^{-\sigma|\a|})$. Thus, if $|\b|>1$, $\del^{(\a,\b)}\psi_* = 0$ and 
\begin{align*}
&\text{ if } \b=0\, , \quad \ (\del^{(\a,\b)}\psi_*)_i(\rx,\vth) = \O(\langle \rx \rangle^{-\sigma(|\a|-1)}) \quad  \text{ and }  \quad  (\del^{(\a,\b)}\psi_*)_{n+i}(\rx,\vth) = \O(\langle\rx\rangle^{-\sigma |\a|}\langle \vth \rangle )\, , \\
&\text{ if } |\b|=1\, , \quad (\del^{(\a,\b)}\psi_*)_i =0 \quad \hspace{3cm} \text{ and } \quad   (\del^{(\a,\b)}\psi_*)_{n+i}(\rx,\vth) = \O(\langle\rx\rangle^{-\sigma |\a|})\, .
\end{align*}
Similar results hold for $\psi_T$, the only difference is that we just have to take $\wt L:=(d\psi_{z',z}^{\bfr',\bfr})_{\rx}(\vth)$ instead of $L$.

\noindent We have for any $f\in C^\infty(T^*M,\mathfrak{E})$, 
$$
\del_{z,\bfr}^{\nu} (f) = \del^{\nu}(f\circ (n_{z,*}^\bfr)^{-1}) \circ n_{z,*}^\bfr = \del^{\nu}(f\circ (n_{z',*}^{\bfr'})^{-1} \circ \psi_*) \circ n_{z,*}^\bfr\, .
$$
Using again the Faa di Bruno formula in Theorem \ref{FaaCS}, we get 
$$
\del_{z,\bfr}^{\nu} =  \sum_{1\leq |\nu'|\leq |\nu|} (P_{\nu,\nu'}(\psi_*)\circ n_{z,*}^\bfr) \ \del_{z',\bfr'}^{\nu'} =:\sum_{1\leq |\nu'|\leq |\nu|} f_{\nu,\nu'} \ \del_{z',\bfr'}^{\nu'}
$$ 
where $P_{\nu,\nu'}(\psi_*)$ is a linear combination of terms of the form $\prod_{j=1}^{s} (\del^{l^j} \psi_*)^{k^j} $, where $1\leq s\leq |\nu|$,  
the $k^j$ and $l^j$ are $2n$-multi-indices with $|k^j|>0$, $|l^j|>0$, $\sum_{j=1}^s k^j = \nu'$ and $\sum_{j=1}^s|k^j|l^j=\nu$. 

Let us note $l^{j}=:(l^{j,1},l^{j,2})$, $k^{j}=:(k^{j,1},k^{j,2})$ where $l^{j,1},l^{j,2},k^{j,1},k^{j,2}$ are $n$-multi-indices. Thus,
 $$
(\del^{l^j}\psi_*)^{k^j}= \prod_{i=1}^n ((\del^{l^j}\psi_*)_i)^{k^{j,1}_i}\ ((\del^{l^j}\psi_*)_{n+i})^{k^{j,2}_i}
$$
and we get, for a given $s$, $(l^{j})$, $(k^j)$ such that $(\del^{l^j} \psi_*)^{k^j} \neq 0$ for all $1\leq j\leq s$,
\begin{align*}
&\text{ if } l^{j,2}=0\, ,  \qquad   (\del^{l^j}\psi_*)^{k^j} = \O(\langle \rx\rangle^{-\sigma(|l^{j}|-1)|k^{j}| - \sigma |k^{j,2}|}\langle \vth\rangle^{|k^{j,2}|})\, , \\ 
&\text{ if } |l^{j,2}|=1\, ,\qquad k^{j,1}=0 \text{ and } \    (\del^{l^j}\psi_*)^{k^j} = \O(\langle \rx \rangle^{-\sigma (|l^j|-1)|k^j|})\, . 
\end{align*}

Since $k^j\neq 0$ and $(\del^{l^j} \psi_*)^{k^j} \neq 0$, $l^{j,2}$ always satisfies $|l^{j,2}|\leq 1$.
By permutation on the $j$ indices, we can suppose that for $1\leq j\leq j_1-1$, we have $l^{j,2}=0$, for $j_1\leq j\leq s$, we have $|l^{j,2}|=1$, where $1\leq j_1 \leq s+1$. Thus,  
$$
\prod_{j=1}^s (\del^{l^j} \psi_*)^{k^j} = \O(\langle\rx \rangle^{-\sigma(\sum_{j=1}^{s}(|l^j|-1)|k^j| + \sum_{j=1}^{j_1-1}|k^{j,2}|)}\langle \vth \rangle^{\sum_{j=1}^{j_1-1}|k^{j,2}|})\, .
$$
Since, with $\nu=(\a,\b)$, $\nu'=(\a',\b')$, 
$$
\sum_{j=1}^{j_1-1} |k^{j,2}| = \sum_{j=1}^s |k^{j,2}| - \sum_{j=j_1}^{s} |k^{j,2}| = |\b'|-\sum_{j=j_1}^{s} |k^{j}||l^{j,2}| =|\b'|-|\b|\, , 
$$
(\ref{fabab}) follows. If we set $f_{0,0,0,0}:=1$ and $f_{\a,0,0,0}:=0$ if $\a\neq 0$, then for any $2n$-multi-index $(\a,\b)$,
$$
\del_{z,\bfr}^{(\a,\b)} = \underset{|\b'|\geq |\b|}{\sum_{0\leq |(\a',\b')|\leq |(\a,\b)|}}\, f_{\a,\b,\a',\b'}\,   \del_{z',\bfr'}^{(\a',\b')} 
$$
and the estimate (\ref{fabab}) holds for any $f_{\a,\b,\a',\b'}$. In the case of $\O_M$-bounded geometry, the proof is similar, and we obtain for a $r_\nu\geq 1$, $\prod_{j=1}^s (\del^{l^j} \psi_*)^{k^j} = \O(\langle\rx \rangle^{r_\nu}\langle \vth \rangle^{|\b'|-|\b|})$, which gives the result.

\noindent $(ii)$ Replacing $\psi_{*}$ by $\psi_{z',z,M^2}^{\bfr',\bfr}:=n_{z',M^2}^{\bfr'}\circ (n_{z,M^2}^\bfr)^{-1}$ in $(i)$, we obtain the result by similar arguments.
\end{proof}

\subsection{Basic function and distribution spaces}

We suppose in this section that $E$ is an hermitian vector bundle on the exponential manifold $(M,\exp)$. Recall that if $u\in C^\infty(M;E)$ (resp. $C^\infty_c(M;E)$) the Fr\'{e}chet space of smooth sections (resp. the $LF$-space of compactly supported smooth sections) of $E\to M$, we have for any $z\in M$, $u^z:=\tau_z^{-1} u \in C^\infty(M,E_z)$ (resp. $C^\infty_c(M,E_z)$). We define for any frame $(z,\bfr)$ on $M$,
$$
T_{z,\bfr}(u):= u^z \circ (n_z^\bfr)^{-1}.
$$
Thus, $T_{z,\bfr}$ sends sections of $E\to M$ to functions from $\R^n$ to $E_z$ and is in fact a topological isomorphism from $C^\infty(M;E)$ (resp. $C^\infty_c(M;E)$) onto $C^{\infty}(\R^n,E_z)$ (resp. $C^{\infty}_c(\R^n,E_z)$). 

In the following, a density (resp. a codensity) is a smooth section of the complex line bundle over $M$ defined by the disjoint union over $x\in M$ of the complex lines formed by the 1-twisted forms on $T_x M$ (resp. $T_x^*(M)$). Recall that a 1-twisted form on a $n$-dimensional vector space $V$ is a function on $F$ on $\Lambda_n V\backslash \{0\}$ such that 
$$
F(cv)= |c| F(v) \qquad \text{for all } v\in \Lambda_n V\backslash \{0\} \text{ and } c\in \R^*.
$$
For a given frame $(z,\bfr)$, let us note $|dx^{z,\bfr}|$ the density associated to the volume form on $M$: $dx^{z,\bfr}:=dx^{1,z,\bfr}\wedge \cdots \wedge dx^{n,z,\bfr}$ and $|\del_{z,\bfr}|$ the codensity defined as $|\del_{1,z,\bfr}\wedge \cdots \wedge \del_{n,z,\bfr}|$.

Any density (resp. codensity) is of the form $c|dx^{z,\bfr}|$ (resp. $c|\del_{z,\bfr}|$) where $c$ is a smooth function on $M$, and by definition is strictly positive if $c(x)>0$ for any $x\in M$.
 For a given strictly positive density $d\mu$, we also note by $d\mu$ its associated (positive, Borel--Radon, $\sigma$-finite) measure on $M$.  
This allows to define the following Banach spaces of (equivalence classes of) functions on $M$: $L^p(M,d\mu)$ ($1\leq p\leq \infty$). Actually, $L^{\infty}(M):=L^{\infty}(M,d\mu)$ does not depend on the chosen $d\mu$, since the null sets for $d\mu$ are exactly the null sets for any other strictly positive density $d\mu'$ on $M$.

For a given $z\in M$, we denote $L^p(M,E_z,d\mu)$ ($1\leq p<\infty$) and $L^\infty(M,E_z)$ the Bochner spaces on $M$ with values in $E_z$.
$E_z$ is a hermitian complex vector space, so we can identify $E_z$ with its antidual $E'_z$. There is a natural anti-isomorphism between $E'_z$ and the dual of $E_z$ but there is in general no canonical way to identify $E_z$ with its dual with a {\it linear} isomorphism. 
Thus, we shall use antiduals rather than duals in the following. However, $E_z$ is anti-isomorphic with its dual by complex conjugaison on $E'_z$. We shall note $\ol x$ the image under this anti-isomorphism of $x\in E_z$ and $\ol E_z$ the dual of $E_z$.


We denote 
$L^p(M;E,d\mu):=\set{\psi \text{ section of }E\to M\text{ such that }|\psi|^p\in L^1(M,d\mu)}/\sim_{a.e.}$ and 
$L^\infty(M;E):= \set{\psi \text{ section of }E\to M\text{ such that }|\psi|\in L^\infty(M) }/\sim_{a.e.}$
where $\sim_{a.e.}$ the standard ``almost everywhere'' equivalence relation.
Since the $\tau_{xy}$ maps are isometries, for any $z\in M$, the map $\psi \to \tau_z^{-1} \psi$ defines linear isometries: $L^p(M;E,d\mu)\simeq L^p(M,E_z,d\mu)$, and $L^\infty(M;E)\simeq L^\infty(M,E_z)$.
In particular, $L^p(M;E,d\mu)$ and $L^\infty(M;E)$ are Banach spaces and $L^2(M;E,d\mu)$ a Hilbert space.
Moreover, we can define for any $\psi \in L^1(M;E,d\mu)$ and $z\in M$ the following Bochner integral $\int \tau_z^{-1}\psi \in E_z$. We can canonically identify $L^\infty(M; E)$ as the antidual of $L^1(M;E,d\mu)$ and 
$L^2(M;E,d\mu)$ as its own antidual. 
The (strong) antiduals of $C_c^\infty(M;E)$ and $C^\infty(M;E)$ are noted respectively $\DD' (M;E)$ and $\E'(M;E)$. 

We define $G_{\sigma}(\R^p,\mathfrak{E})$ (resp. $S_\sigma(\R^p)$), where $\sg\in [0,1]$, as the space of smooth functions $g$ from $\R^{p}$ into $\mathfrak{E}$ (resp. $\R$) such that for any $p$-multi-index $\nu\neq 0$ (resp. any $p$-multi-index $\nu$), there exists $C_\nu>0$
such that $\norm{\del^{\nu} g (\rv)}\leq C_{\nu} \langle \rv \rangle^{-\sigma (|\nu|-1)}$ (resp. $|\del^{\nu} g (\rv)|\leq C_{\nu} \langle \rv \rangle^{-\sigma |\nu|}$) for any $\rv\in \R^{p}$. We note $\O_M(\R^p,\mathfrak{E})$ the space of smooth $\mathfrak{E}$-valued functions with polynomially bounded derivatives.
We use the shorcuts $G_\sigma(\R^p):=G_\sigma(\R^p,\R^p)$ and $\O_M(\R^p):=\O_M(\R^p,\R)$.

We have the following lemma which will give an equivalent formulation of $S_\sigma$ or $\O_M$-bounded geometry.

\begin{lem}
\label{Ssigma} 
(i) Let $f\in G_{\sigma}(\R^p,\mathfrak{E})$ (resp. $S_\sigma(\R^p)$) and $g\in G_\sigma(\R^{n},\R^p)$ such that, if $\sigma>0$, there exists $\eps>0$ such that $\langle g(\rv) \rangle \geq \eps \langle \rv \rangle$ for any $\rv\in \R^n$. Then $f\circ g \in G_\sigma(\R^n,\mathfrak{E})$ (resp. $S_\sigma(\R^n)$).

\noindent (ii) The set $G_\sigma^\times(\R^p)$ of diffeomorphisms $g$ on $\R^p$ such that $g$ and $g^{-1}$ are in $G_\sigma(\R^p)$ is a subgroup of $\mathrm{Diff}(\R^p)$ and contains $GL_p(\R)$ as a subgroup.  

\noindent (iii) We have $\O_M(\R^p,\mathfrak{E})\circ \O_M(\R^n,\R^p) \subseteq \O_M(\R^n,\mathfrak{E})$. In particular, the space $\O_M(\R^p,\R^p)$ is a monoid under the composition of functions. The set of inversible elements of the monoid $\O_M(\R^p,\R^p)$, noted $\O_M^\times(\R^p,\R^p)$, is a subgroup of $\mathrm{Diff}(\R^p)$ and contains $G_\sigma^\times(\R^p)$ as a subgroup. 

\noindent (iv)
$(M,\exp)$ has a $S_\sigma$ (resp. $\O_M$)-bounded geometry if and only if there exists a frame $(z_0,\bfr_0)$ such that for any frame $(z,\bfr)$, $\psi_{z_0,z}^{\bfr_0,\bfr}\in G_{\sigma}^\times(\R^n)$ (resp. $\O_M^\times(\R^n,\R^n)$).

\noindent (v) The set, noted $S_\sigma^\times(\R^p)$ (resp. $\O_M^\times(\R^p)$), of smooth functions $f:\R^p\to \R^*$ such that $f$ and $1/f$ are in $S_\sigma(\R^p)$ (resp. $\O_M(\R^p)$) is a commutative group under pointwise multiplication of functions. Moreover, $S_\sigma^\times(\R^p) \leq S_{\sigma'}^\times(\R^p) \leq \O_M^\times(\R^p)$ if $1\geq \sigma \geq \sigma'\geq 0$.  

\noindent (vi) If $g\in G_\sigma^\times (\R^p)$ (resp. $\O_M^\times(\R^p,\R^p)$) then its Jacobian determinant $J(g)$ is in $S_\sigma^\times(\R^p)$ (resp. $\O_M^\times(\R^p)$).
\end{lem}

\begin{proof} 
$(i)$ The Faa di Bruno formula yields for any $n$-multi-index $\nu\neq 0$,
\begin{equation}
\del^\nu (f\circ g) = \sum_{1\leq |\la|\leq |\nu|} (\del^\la f)\circ g\  P_{\nu,\la}(g) \label{delnufg}
\end{equation}
where $P_{\nu,\la}(g)$ is a linear combination (with coefficients independent on $f$ and $g$) of functions of the form $\prod_{j=1}^s (\del^{l^j} g)^{k^j}$ where $s\in \set{1,\cdots ,|\nu|}$.  
The $k^j$ are $p$-multi-indices and the $l^j$ are $n$-multi-indices (for $1\leq j\leq s$) such that $|k^j|>0$, $|l^j|>0$, $\sum_{j=1}^s k^j = \la$ and $\sum_{j=1}^s |k^j| l^j= \nu$. As a consequence, since $g\in G_\sigma(\R^{n},\R^p)$, for each $\nu,\la$ with $1\leq |\la|\leq |\nu|$ there exists $C_{\nu,\la}>0$ such that for any $\rv\in \R^n$,
\begin{equation}\label{pnulag1}
|P_{\nu,\la}(g) (\rv)|\leq C_{\nu,\la} \langle \rv \rangle^{-\sigma(|\nu|-|\la|)}\, . 
\end{equation}
Moreover, if $f\in G_\sigma(\R^p,\mathfrak{E})$ (resp. $S_\sigma(\R^p)$), there is $C'_\la>0$ such that for any $\rv\in \R^n$, $\norm{(\del^\la f) \circ g(\rv)}\leq C'_\la \langle \rv\rangle ^{-\sigma(|\la|-1)}$ (resp. $|(\del^\la f) \circ g(\rv)|\leq C'_\la \langle \rv\rangle ^{-\sigma|\la|}$). The result now follows from (\ref{delnufg}) and (\ref{pnulag1}).

\noindent $(ii)$ Let $f$ and $g$ in $G_\sigma^\times(\R^p)$. We have $\del_i g^{-1} = \O(1)$ for any $i\in \set{1,\cdots,p}$. Taylor--Lagrange inequality of order 1 entails that $\langle g^{-1}(v)\rangle =\O(\langle v\rangle)$ and thus there is $\eps>0$ such that $\langle g(\rv) \rangle \geq \eps \langle \rv \rangle$ for any $\rv\in \R^n$. With $(i)$, we get $f\circ g\in G_\sigma(\R^p)$. The same argument shows that $g^{-1}\circ f^{-1} \in G_\sigma(\R^p)$.

\noindent $(iii)$ Direct consequence of Theorem \ref{FaaCS}.

\noindent $(iv)$
The only if part is obvious. Suppose then that for any frame $(z,\bfr)$, $\psi_{z_0,z}^{\bfr_0,\bfr}\in G_{\sigma}^\times(\R^n)$ (resp. $\O_M^\times(\R^n,\R^n)$. Let $(z,\bfr)$, $(z',\bfr')$ be two frames. We have $\psi_{z,z'}^{\bfr,\bfr'}=\psi_{z,z_0}^{\bfr,\bfr_0}\circ \psi_{z_0,z'}^{\bfr_0,\bfr'}$. So, by $(ii)$ (resp. $(iii)$),  $\psi_{z,z'}^{\bfr,\bfr'}\in G_{\sigma}^\times(\R^n)$ (resp. $\O_M^\times(\R^n,\R^n)$), which yields the result. 

\noindent $(v)$ By Leibniz rule, the spaces $S_\sigma(\R^p)$ and $\O_M(\R^p)$ are $\R$-algebras under the pointwise product of functions. The result follows.

\noindent $(vi)$ Consequence of $(ii)$, $(iii)$, $1/J(g) = J(g^{-1})\circ g$ and the fact that $S_\sigma(\R^p)$ (resp. $\O_M(\R^p)$) is stable under the pointwise product of functions.
\end{proof}

Remark that for any $g \in G_\sigma^\times(\R^p)$, we have $g\asymp\Id_{\R^p}$.
The multiplication by a function in $\O_M^\times (\R^n)$ is a topological isomorphism from the Fr\'{e}chet space of rapidly decaying $E_z$-valued functions $\S(\R^n,E_z)$ onto itself. 
If we denote $J_{z,z'}^{\bfr,\bfr'}$ the Jacobian of $\psi_{z,z'}^{\bfr,\bfr'}$, then $1/{J_{z,z'}^{\bfr,\bfr'}}=J_{z',z}^{\bfr',\bfr}\circ \psi_{z,z'}^{\bfr,\bfr'}$ and $J_{z,z'}^{\bfr,\bfr'}\circ n_{z'}^{\bfr'}(x)= dx^{z,\bfr}/dx^{z',\bfr'} (x)= \det M_{z,x}^\bfr ( M_{z',x}^{\bfr'})^{-1}=\det( M_{z',x}^{\bfr'})^{-1}  M_{z,x}^\bfr$.
We deduce from Lemma \ref{Ssigma} $(vi)$ that if $(M,\exp)$ has a $S_{\sigma}$ (resp. $\O_M$) bounded geometry then $J_{z,z'}^{\bfr,\bfr'}$ is in $S^\times_{\sigma}(\R^n)$ (resp. $\O_M^\times (\R^n)$). 

\begin{defn} Any smooth function $f$ is in $S_\sigma$ (resp. $\O_M$) if for any frame $(z,\bfr)$, $f\circ (n_z^\bfr)^{-1}\in S_\sigma(\R^n)$ (resp. $\O_M(\R^n)$). Similarly, any smooth function $f$ is in $S_\sigma^\times$ (resp. $\O_M^\times$) if for any frame $(z,\bfr)$, $f\circ (n_z^\bfr)^{-1}\in S_\sigma^\times(\R^n)$ (resp. $\O_M^\times(\R^n)$).
\end{defn}

\begin{lem}
\label{redSOM}
If $(M,\exp)$ has a $S_\sigma$-bounded geometry then a smooth function $f$ on $M$ is in $S_\sigma$ (resp. $S_\sigma^\times$) if there exists a frame $(z,\bfr)$ such that $f\circ (n_z^\bfr)^{-1}\in S_\sigma(\R^n)$ (resp. $S_\sigma^\times(\R^n)$). Similarly, If $(M,\exp)$ has a $\O_M$-bounded geometry then $f$ is in $\O_M$ (resp. $\O_M^\times$) if there exists a frame $(z,\bfr)$ such that $f\circ (n_z^\bfr)^{-1}\in \O_M(\R^n)$ (resp. $\O_M^\times(\R^n)$). 
\end{lem}
\begin{proof} Let $(z',\bfr')$ be a frame such that $f\circ (n_{z'}^{\bfr'})^{-1}\in S_\sigma(\R^n)$, and let $(z,\bfr)$ be another frame. By Lemma \ref{B1-lem} $(ii)$, if $(M,\exp)$ has a $S_\sigma$-bounded geometry then for any $n$-multi-index $\a$, 
$$
\del^\a (f\circ (n_z^\bfr)^{-1}) = \sum_{0\leq |\a'|\leq |\a|} f_{\a,\a'}\circ (n_{z}^{\bfr})^{-1} \, (\del^{\a'} f\circ (n_{z'}^{\bfr'})^{-1}) \circ \psi_{z',z}^{\bfr',\bfr}
$$
where $(f_{\a,\a'}\circ (n_{z}^{\bfr})^{-1})(\rx) = \O(\langle \rx\rangle^{-\sigma(|\a|-|\a'|)})$. As a consequence $\del^\a (f\circ (n_z^\bfr)^{-1})(\rx)=\O(\langle \rx\rangle^{-\sigma|\a|})$ and the result follows. The case of $\O_M$ bounded geometry is similar.
\end{proof}

\begin{defn}
A smooth strictly positive density $d\mu$ is a $S_\sigma^\times$-density (resp. $\O_M^\times$-density) if for any frame $(z,\bfr)$, the unique smooth strictly positive function $f_{z,\bfr}$ such that $d\mu = f_{z,\bfr} |dx^{z,\bfr}|$ is in $S_\sigma^\times$ (resp. $\O_M^\times$). In this case, we shall note $\mu_{z,\bfr}$ the smooth stricly positive function in $S_\sigma^\times(\R^n)$ (resp. $\O_M^\times(\R^n)$) such that $d\mu =(\mu_{z,\bfr}\circ n_z^\bfr)\, |dx^{z,\bfr}|$. 
\end{defn}

We shall say that $(M,\exp,d\mu)$ has a $S_\sg$ (resp. $\O_M$) bounded geometry if $(M,\exp)$ has a $S_\sg$ (resp. $\O_M$) bounded geometry  and $d\mu$ is a $S^\times_\sg$(resp. $\O^\times_M$) density.

\begin{lem} If $(M,\exp)$ has a $S_\sigma$ (resp. $\O_M$) bounded geometry then any density of the form $ u\circ n_{z'}^{\bfr'} |dx^{z,\bfr}|$ where $u$ is a smooth strictly positive function in $S_\sigma^\times(\R^n)$ (resp. $\O_M(\R^n)$) and $(z,\bfr)$, $(z',\bfr')$ are frames, is a $S_\sigma^\times$-density (resp. $\O_M^\times$-density).
\end{lem}
\begin{proof} Let $(z'',\bfr'')$ be an arbitrary frame. Noting $d\mu:=u\circ n_{z'}^{\bfr'} |dx^{z,\bfr}|$, we get $d\mu = (u\circ n_{z'}^{\bfr'})| J_{z,z''}^{\bfr,\bfr''}|\circ n_{z''}^{\bfr''} |dx^{z'',\bfr''}|$.  We already saw that the function $J_{z,z''}^{\bfr,\bfr''}$ is in $S^\times_\sigma(\R^n)$ (resp. $\O_M^\times(\R^n))$. By Lemma \ref{redSOM}, $(u\circ n_{z'}^{\bfr'})( |J_{z,z''}^{\bfr,\bfr''}|\circ n_{z''}^{\bfr''})$  is in $S_\sigma^\times$ (resp. $\O_M^\times$).
\end{proof}

\begin{rem} By taking $u:=\rx\mapsto 1$ in the previous lemma, we see that for any exponential manifold $(M,\exp)$ with $S_\sigma$ (resp. $\O_M$) bounded geometry, we can define a canonical family of $S^\times_\sigma$-densities (resp. $\O_M^\times$-densities) on $M$: $\mathcal{D}:=(|dx^{z,\bfr}|)_{(z,\bfr)\in I}$ where $I$ is the set of frames on $M$. If the map $\exp$ is the exponential map associated to a pseudo-Riemannian metric $g$ on $M$, we can also define a canonical subfamily of $\mathcal{D}$ by $\mathcal{D}_g:=(|dx^{z}|)_{z\in M}$ where $|dx^z|:=|dx^{z,\bfr}|$ with $\bfr$ any orthonormal basis (in the sense $g_z(\bfr_i,\bfr_j)=\eta_i\delta_{i,j}$ where $\eta_i=1$ for $1\leq i \leq m$ and $\eta_i=-1$ for $i>m$, where $g$ has signature $(m,n-m)$) of $T_z(M)$ ($|dx^z|$ is then independent of $\bfr$). A priori, the Riemannian density does not belong to the canonical $M$-indexed family $\mathcal{D}_g$. 
\end{rem}

We shall need integrations over tangent and cotangent fibers and manifolds. We thus define $d\mu^*:=(\mu^{-1}_{z,\bfr}\circ n_z^\bfr)\,|\del_{z,\bfr}|$ the codensity associated to $d\mu$, where $\mu^{-1}_{z,\bfr}:=\tfrac{1}{\mu_{z,\bfr}}$ and 
$(z,\bfr)$ is a frame. Note that since $|\del_{z,\bfr}|/|\del_{z',\bfr'}|=|dx^{z',\bfr'}|/|dx^{z,\bfr}|= (\mu_{z,\bfr}\circ n_{z}^{\bfr})/(\mu_{z',\bfr'}\circ n_{z'}^{\bfr'})$, $d\mu^*$ is independent of $(z,\bfr)$.
For a given $x\in M$, the density on $T_x(M)$ associated to $d\mu$ is $d\mu_x:=(\mu_{z,\bfr}\circ n_z^\bfr(x))\, |dx^{z,\bfr}_x|$ and the associated density on $T_x^*(M)$ is $d\mu_x^*:=(\mu^{-1}_{z,\bfr}\circ n_z^\bfr(x))\,|{\del_{z,\bfr}}_x|$. For a function $f$ defined on $T_x(M)$ or $T^*_x(M)$, we have formally:
\begin{align*}
&\int_{T_x(M)} f(\xi)\, d\mu_x(\xi) = \mu_{z,\bfr}\circ n_z^\bfr(x)\int_{\R^n}f\circ (  M_{z,x}^{\bfr})^{-1}(\zeta)\, d\zeta \,, \\
&\int_{T_x^*(M)} f(\th)\, d\mu_x^*(\th) = \mu_{z,\bfr}^{-1}\circ n_z^\bfr(x)\int_{\R^n}f\circ (\wt M_{z,x}^{\bfr})^{-1}(\vth)\, d\vth\, ,
\end{align*}
and it is straightforward to check that these integrals are independent of the chosen frame $(z,\bfr)$.

\subsection{Schwartz spaces and operators}

\begin{assum} We suppose in this section and in section \ref{fouriersec} that $(M,\exp,E,d\mu)$, where $E$ is an hermitian vector bundle on $M$, has a $\O_M$-bounded geometry.
\end{assum}

The main consequence of the exponential structure is the possibility to define Schwartz functions on $M$.

\begin{defn} A section $u\in C^\infty(M,E)$ is rapidly decaying at infinity if for any $(z,\bfr)$, any $n$-multi-index $\a$ and $p\in \N$, there exists $K_{\a,p}>0$ such that the following estimate
\begin{equation}\label{SB-est}
\norm{\del_{z,\bfr}^{\a} u^z (x)}_{E_z}< K_{\a,p} \langle x\rangle_{z,\bfr}^{-p} 
\end{equation}
holds uniformly in $x\in M$. We note $\S(M,E)$ the space of such sections.
\end{defn}

With the hypothesis of $\O_M$-bounded geometry, we see that the requirement ``any $(z,\bfr)$'' can be reduced to a simple existence:

\begin{lem}
\label{indepzbS}
A section $u\in C^\infty(M,E)$ is in $\S(M,E)$ if and only if there exists a frame $(z,\bfr)$ such that (\ref{SB-est}) is valid. 
\end{lem}
\begin{proof} Suppose that (\ref{SB-est}) is valid for $(z',\bfr')$ and let $(z,\bfr)$ be another frame.  Thus, with Lemma \ref{B1-lem} $(ii)$ and  Leibniz rule,
\begin{align}\label{del-coord-ch}
\del_{z,\bfr}^\a u^z(x)  =  \sum_{0\leq |\a'|\leq |\a| }\, \sum_{\b\leq \a' }\ \ f_{\a,\a'} \,\tbinom{\a'}{\b}\,\del_{z',\bfr'}^{\a'-\b} (\tau_z^{-1}\tau_{z'})\ \del_{z',\bfr'}^{\b} u^{z'}(x).
\end{align}
Moreover $|f_{\a,\a'} \,\tbinom{\a'}{\b}\,\del_{z',\bfr'}^{\a'-\b} (\tau_z^{-1}\tau_{z'})| \leq C_{\a} \langle x \rangle_{z,\bfr}^{q_\a}$ for a $C_\a>0$ and a $q_\a\geq 1$. Now (\ref{SB-est}) and (\ref{i-first-om}) entail that for any $p\in \N$, there is $K>0$ such that
$\norm{\del_{z,\bfr}^\a u^z(x)}_{E_z} \leq K \langle x \rangle_{z,\bfr}^{-p}$. The result follows.
\end{proof}

\begin{rem} Let $u\in C^\infty(M,E)$ and $(z,\bfr)$ a frame. Then $u\in \S(M,E)$ if and only if $(\tau_z^{-1} u)\circ (n_{z}^{\bfr})^{-1} \in \S(\R^n,E_z)$. In other words, if $v\in \S(\R^n,E_z)$ then $\tau_z(v\circ n_z^\bfr) \in \S(M,E)$. 
\end{rem}

The following lemma shows that we can define canonical Fr\'{e}chet topologies on $\S(M,E)$.

\begin{lem} Let $(z,\bfr)$ a frame. Then 

\noindent (i) The following set of semi-norms: 
$$
q_{\a,p}(u):= \sup_{x\in M} \langle x \rangle_{z,\bfr}^{p} \norm{\del_{z,\bfr}^{\a} u^z (x)}_{E_z}\,.
$$
defines a locally convex metrizable topology $\mathcal{T}$ on $\S(M,E)$.

\noindent (ii) The application $T_{z,\bfr}$ is a topological isomorphism from the space $\S(M,E)$ onto $\S(\R^n,E_z)$.

\noindent (iii) The topology $\mathcal{T}$ is Fr\'{e}chet and independent of the chosen frame $(z,\bfr)$.
\end{lem}
\begin{proof}
$(i)$ and $(ii)$ are obvious.

\noindent $(iii)$ Since $T_{z,\bfr}$ is a topological isomorphism, $\mathcal{T}$ is complete. Following the arguments of the proof of Lemma \ref{indepzbS}, we see that there is $r\in \N^*$ such that for any $n$-multi-index $\a$ and $p\in \N$, there exist $C_{\a,p}>0$, $r_{\a,p}\in \N^*$, such for any $u\in \S(M,E)$,
$$
q_{\a,p}^{(z,\bfr)} (u) \leq C_{\a,p} \sum_{|\b|\leq |\a|} q^{(z',\bfr')}_{\b,r_{\a,p}}(u)\, .
$$
The independence on $(z,\bfr)$ follows.
\end{proof}

\begin{rem}
\label{BME}
If $(M,\exp,E,d\mu)$ has a $S_0$-bounded geometry, then it is possible to define the Fr\'{e}chet space of smooth sections with bounded derivatives $\B(M,E)$ by following the same procedure of $\S(M,E)$, with Lemma \ref{B1-lem}.
\end{rem}

Classical results of distribution theory \cite{Treves} and the previous topological isomorphisms $T_{z,\bfr}$ entail the following diagrams of continuous linear injections ($(M;E)$ ommitted and $1\leq p<\infty$):

\[\xymatrix{
    C_c^\infty \ar[r]\ar[d] & \S \ar[r]\ar[d]  & C^\infty \ar[d]\\
    \E' \ar[r] &  \S' \ar[r] & \D'\\
  } \hskip1cm
\xymatrix{
     & \B \ar[rr]    &&  L^\infty \ar[rd] & \\
    \S \ar[rr]\ar[ru] & & L^p(d\mu) \ar[rr] && \S' 
  \ .} \] 
  
The injections $\S\to \B\to L^\infty$ are valid in the case where $M$ has a $S_0$-bounded geometry. In the case of a general $\O_M$-bounded geometry, only the injection $\S\to L^\infty$ holds a priori.  
The injection from functions into distribution spaces is given here by $u\mapsto \langle u,\cdot \rangle$ where $\langle u,v\rangle :=\int (u | v)\, d\mu$. Note that the following continuous injections $\S \to \S' $ and $ \S\to L^{p}(d\mu) \to \S'$, $(1\leq p<\infty)$ have a dense image.

Using the same principles of the definition of $\S$ together with the $\O_M$-bounded geometry hypothesis and Lemma \ref{cchange} $(ii)$, we define the Fr\'{e}chet space $\S(M\times M, L(E))$ such that for any $(z,\bfr)$ the applications $T_{z,\bfr,M^2}:=K\mapsto  K^z \circ (n_{z,M^2}^\bfr)^{-1}$
are topological isomorphisms from $\S(M\times M , L(E))$ onto $\S(\R^{2n},L(E_z))$. Noting $j_{M^2}$ the continous dense injection from $\S(M\times M,L(E))$ into its antidual $\S'(M\times M,L(E))$ defined as $\langle j_{M^2}(K),K'\rangle= \int_{M\times M} \Tr(K(x,y) (K'(x,y))^*)\,d\mu\ox d\mu(x,y)$, we have the following commutative diagram, where $j$ is the classical continuous dense inclusion from $\S(\R^{2n},L(E_z))$ into its antidual, and $M_{\mu\otimes\mu}$ is the multiplication operator from $\S(\R^{2n},L(E_z))$ onto itself by the $\O_M^\times(\R^{2n})$ function $\mu_{z,\bfr}\otimes \mu_{z,\bfr}$:

\[\xymatrix{
    \S(M\times M,L(E)) \ar[rr]^{j_{M^2}} \ar[d]_{T_{z,\bfr,M^2}} &  \ar[r] &\S'(M\times M,L(E))  \\
    \S(\R^{2n},L(E_z)) \ar[r]_{M_{\mu\ox\mu}} &  \S(\R^{2n},L(E_z)) \ar[r]_j & \S'(\R^{2n},L(E_z))\ . \ar[u]^{T^*_{z,\bfr,M^2}}   \\
  }\]

\noindent Since $\S$ is nuclear, $L(\S,\S')\simeq \S'(M\times M,L(E))$ and $\S(M\times M, L(E)) \simeq \S\ \wox\  \ol\S$ where $\ol \S :=\S(M,\ol E)$.  
Thus, $\S'(M\times M,L(E))\simeq \S'\wox\ \ol\S'$, where $\ol \S'$ is the dual of $\S$ which is also the antidual of $\ol \S$. Note that the isomorphism $L(\S,\S')\simeq \S'(M\times M,L(E))$ is given by 
$$
\langle A_K (v) ,u\rangle = K(u\otimes \ol v) 
$$
where $A_K$ is operator associated to the kernel $K$, $u,v\in \S$, and $\ol v(y):= \ol{v(y)}$. Formally, 
$$
\langle A_K (v), u\rangle= \int_{M\times M}(K(x,y) v(y) | u(x)) d\mu\otimes d\mu (x,y)\, ,  \qquad (A_K v)(x)= \int_M K(x,y) v(y) d\mu(y).
$$  
Thus any continuous linear operator $A:\S \to \S'$ is uniquely determined by its kernel $K_A \in  \S'(M\times M , L(E))$.
The transfert of $A$ into the frame $(z,\bfr)$ is the operator $A_{z,\bfr}$ from $\S(\R^n,E_z)$ into $\S'(\R^{n},E_z)$ such that 
$$
\langle A_{z,\bfr}(v),u\rangle := \langle A(T_{z,\bfr}^{-1}(v)), T_{z,\bfr}^{-1}(u) \rangle.
$$
Thus, if $K_A$ is the kernel of $A$, we have $K_{A_{z,\bfr}}:= \wt T_{z,\bfr,M^2}(K_A)$ as the kernel of $A_{z,\bfr}$, where $\wt T_{z,\bfr,M^2}$ here is the inverse of the adjoint of $T_{z,\bfr,M^2}$. $\wt T_{z,\bfr,M^2}$ is a topological isomorphism from $\S'(M\times M , L(E))$ onto  $\S'(\R^{2n} , L(E_z))$.

\begin{defn} An operator $A\in L(\S,\S')$ is regular if $A$ and its adjoint $A^\dag$ send continously $\S$ into itself. An isotropic smoothing operator is an operator with kernel in $\S(M\times M,L(E))$. The space regular operators and the space of isotropic smoothing operators are respectively noted $\Re(\S)$ and $\Psi^{-\infty}$.
\end{defn}

Note that this definition of isotropic smoothing operators differs from the standard smoothing operators one where only local effects are taken into account, since in this case, a smoothing operator is just an operator with smooth kernel.  Clearly, $A$ is regular if and only if for any frame $(z,\bfr)$, $A_{z,\bfr}$ is regular as an operator from $\S(\R^{n},E_z)$ into $\S'(\R^{n},E_z)$. Remark that the space of regular operators  forms a $*$-algebra under composition and the space of isotropic smoothing operators $\Psi^{-\infty}$ is a $*$-ideal of this algebra.

Let us record the following important fact:

\begin{prop} Any isotropic smoothing operator extends (uniquely) as a Hilbert--Schmidt operator on $L^2(d\mu)$.
\end{prop}
\begin{proof} An isotropic smoothing operator $A$ (with kernel $K$) extends as a continous linear operator from $\S'$ to $\S$, and thus it also extends as a bounded operator on $L^2(d\mu)$. Let $(z,\bfr)$ be a frame. If $U$ is the unitary associated to the isomorphism $L^2(d\mu)$ onto $\H_{z,\bfr}:=L^2(\R^n,E_z,\mu_{z,\bfr}\, dx)$ we have $A = U^* A_{z,\bfr} U$ where $A_{z,\bfr}$ is a bounded operator on $\H_{z,\bfr}$ given by the kernel $K^z\circ (n_z^\bfr,n_z^\bfr)^{-1}$. Since this kernel is in $\H_{z,\bfr}\otimes \ol \H_{z,\bfr} = L^2(\R^{2n},E_z\otimes \ol E_z, (\mu_{z,\bfr}\, dx)^{\ox 2} )$, it follows that $A_{z,\bfr}$ is Hilbert--Schmidt on $\H_{z,\bfr}$, which gives the result. 
\end{proof}

\subsection{Fourier transform}\label{fouriersec}

Fourier transform is the fundamental element that will allow the passage from operators to their symbols. 
In our setting, it is natural to extend the classical Fourier transform on $\R^n$ to Schwartz spaces of rapidly decreasing sections 
on the tangent and cotangent bundles of $M$, and use the fibers $T_x(M)$, $T_x^*(M)$ as support of integration. 

\begin{defn} A smooth section $a \in C^\infty(T^*M,L(E))$ is in $\S(T^*M,L(E))$ if for any $(z,\bfr)$, any $2n$-multi-index $\nu$ and any $p\in \N$, there exists $K_{p,\nu}>0$ such that 
\begin{equation}
\label{STM-est}
\norm{\del_{z,\bfr}^{\nu} a^z(x,\th)}_{L(E_z)}\leq K_{p,\nu} \langle x,\th \rangle_{z,\bfr}^{-p}
\end{equation}
uniformly in $(x,\th)\in T^*M$. A similar definition is set for $\S(TM,L(E))$.
\end{defn}

Following the same technique as for the space $\S(M,E)$, using the coordinate invariance given by Lemma \ref{cchange} we obtain the
\begin{prop} 
(i)  A section $u\in C^\infty(T^*M,L(E))$ is in $\S(T^*M,L(E))$ if and only if there exists a frame $(z,\bfr)$ such that (\ref{STM-est}) is valid. A similar property holds for $\S(TM,L(E))$.

\noindent (ii) There is a Fr\'{e}chet topology on $\S(T^*M,L(E))$ such that each 
$$
T_{z,\bfr,*}: a \mapsto a^z \circ (n_{z,*}^{\bfr})^{-1}
$$
is a topological isomorphism from $\S(T^*M,L(E))$ onto $\S(\R^{2n},L(E_z))$. A similar property holds for $\S(TM,L(E))$ and the applications $T_{z,\bfr,T}:=a\mapsto a^z \circ (n_{z,T}^\bfr)^{-1}$.
\end{prop}
\begin{proof}
\noindent $(i,ii)$ 
Suppose that (\ref{STM-est}) is valid for $(z',\bfr')$ and $a\in C^\infty(T^*M,L(E))$ and let $(z,\bfr)$ another frame.  With Lemma \ref{cchange} and  Leibniz rule, noting $\nu=(\a,\b)$, $\nu'=(\a',\b')$, $\la=(\la^1,\la^2)$ and $\rho=(\rho^1,\rho^2)$, we get
\begin{align}\label{del-cchange}
\del_{z,\bfr}^{\nu} a^z  =  \underset{|\b'|\geq |\b| }{\sum_{0\leq |\nu'|\leq |\nu|}} \sum_{\rho\leq \la \leq \nu'}\ \ f_{\nu,\nu'} C_{\nu',\la,\rho}\,\del_{z',\bfr'}^{\a'-\la^1}(\tau_z^{-1}\tau_{z'})\, \del_{z',\bfr'}^{(\rho^1,\b')}(a^{z'}) \,\del_{z',\bfr'}^{\la^1-\rho^1}(\tau_{z'}^{-1}\tau_{z})\,
\end{align}
where $C_{\nu',\la,\rho}= \delta_{\b',\la^2}\delta_{\b',\rho^2} \tbinom{\nu'}{\la} \tbinom{\la}{\rho}$. Using now the fact that for any $\rx,\vth \in \R^n$,
$\langle \rx \rangle^{1/2} \langle \vth \rangle^{1/2} \leq \langle(\rx,\vth)\rangle \leq \langle  \rx \rangle \langle \vth \rangle$, and (\ref{i-first-om}), (\ref{i-bis-om}), we see that for any $2n$-multi-index $\nu$, and $p\in \N$, there is $r_{\nu,p}\in \N^*$ and $C_{\nu,p}>0$ such that $q_{\nu,p}^{(z,\bfr)}(a) \leq C_{\nu, p} \sum_{|\rho|\leq |\nu|} q^{(z',\bfr')}_{\rho,r_{\nu,p}}(a)$, where 
$$
q_{\nu,p}^{(z,\bfr)}(a):=  \sup_{(x,\th)\in T^*M} \langle x,\th \rangle_{z,\bfr}^{p} \norm{\del_{z,\bfr}^{\nu} a^z (x,\th)}_{L(E_z)}\, .
$$
The results follow, as in the case of $\S(M,E)$, by taking the topology given by the seminorms $q_{\nu,p}^{z,\bfr}$ for an arbitrary frame $(z,\bfr)$.
\end{proof}

\begin{rem} If $(M,\exp,E)$ has a $S_0$-bounded geometry, we saw in Remark \ref{BME} that a coordinate free (independent of the frame $(z,\bfr)$) definition of a space of smooth $E$-sections on $M$ with bounded derivatives is possible. However, a similar definition cannot be given in the same manner for $L(E)$-sections on $TM$ or $T^*M$ with bounded derivatives, due to the fact that the change of coordinates of Lemma \ref{cchange} impose an increasing power of $\langle\th \rangle$ (when $|\b'|>|\b|$). However, the independence over $(z,\bfr)$ would still hold for the space of smooth sections of $L(E)\to T^*M$ (resp. $TM$) with polynomially bounded derivatives.
\end{rem}

We note $\S'(T^*M,L(E))$ and $\S'(TM,L(E))$ the strong antiduals of $\S(T^*M,L(E))$ and $\S(TM,L(E))$, respectively.
We have the following continuous inclusion with dense image 
$$
j_{T^*M}:\S(T^*M,L(E))\to \S'(T^*M,L(E)) \qquad \text{\big(resp. } j_{TM}:\S(TM,L(E)) \to \S'(TM,L(E))\big)
$$
defined by
$$
\langle j_{T^*M}(a),b\rangle := \int_{TM^*} \Tr(ab^*) d\mu^*\qquad \text{\big(resp. }\langle j_{TM}(a),b\rangle := \int_{TM} \Tr(ab^*) d\mu^T \,\big)
$$
where $d\mu^*$ is the measure on $T^*M$ given by $d\mu^*(x,\th):= d\mu^*_x(\th)d\mu(x)$ and $d\mu^T$ is the measure on $TM$ given by $d\mu^T(x,\xi):= d\mu_x(\xi)d\mu(x)$. Note that for any $(z,\bfr)$, $d\mu^*(x,\th)= |{\del_{z,\bfr}}_x|(\th) |dx^{z,\bfr}|(x)$ (this is the Liouville measure on $T^*M$) and $d\mu^T(x,\th)= \mu^2_{z,\bfr}\circ n_z^\bfr(x)|dx^{z,\bfr}_x|(\xi) |dx^{z,\bfr}|(x)$. We have the following commutative diagram, where $M_{\mu^2}$ is the multiplication operator by the $\O_M^\times(\R^{2n})$ function $(\rx,\zeta)\mapsto \mu_{z,\bfr}^2(\rx)$,  
\[\xymatrix{
    \S(TM,L(E)) \ar[rr]^{j_{TM}} \ar[d]_{T_{z,\bfr,T}} &  \ar[r] &\S'(TM,L(E))  \\
    \S(\R^{2n},L(E_z)) \ar[r]_{M_{\mu^2}} &  \S(\R^{2n},L(E_z)) \ar[r]_j & \S'(\R^{2n},L(E_z)) \ar[u]_{T^*_{z,\bfr,T}}   \\
  }\]
and, in the case of $\S(T^*M,L(E))$ a similar diagram is valid if $M_{\mu^2}$ is replaced by the identity.
\begin{defn} The Fourier transform of $a\in \S(TM,L(E))$ is
$$
\F(a) : (x,\th)\mapsto \int_{T_x(M)} e^{-2\pi i \langle \th,\xi \rangle }\, a(x,\xi)\, d\mu_x(\xi)\, .
$$ 
\end{defn}

\begin{prop} $\F$ is a topological isomorphism from $\S(TM,L(E))$ onto $\S(T^*M,L(E))$ with inverse 
$$
\ol \F (a):= (x,\xi) \mapsto \int_{T^*_x(M)} e^{2\pi i \langle \th,\xi \rangle}\, a(x,\th)\, d\mu^*_x(\th)\, .
$$ 
The adjoint $\ol \F^*$ of $\ol \F$ coincides with $\F$ on $\S(TM,L(E))$, so we still note $\ol \F^*$ by $\F$ and $\F^*$ by $\ol\F$.
\end{prop}
\begin{proof} Let $(z,\bfr)$ be a frame. It is straightforward to check that the following diagram commutes
\[\xymatrix{
    S(TM,L(E)) \ar[r]^{\F} \ar[d]_{T_{z,\bfr,T}} & \S(T^*M,L(E)) \\
    \S(\R^{2n},L(E_z)) \ar[r]_{\F_{z,\bfr}} &  \S(\R^{2n},L(E_z))\ar[u]_{T_{z,\bfr,*}^{-1}}  \\
  }\]
where $\F_{z,\bfr}= \F_P \circ M_{\mu} = M_{\mu}\circ \F_P$, with $M_{\mu}$ the multiplication operator on $\S(\R^{2n},L(E_z))$ defined by $M_{\mu}(a):=(\rx,\zeta)\mapsto \mu_{z,\bfr}(\rx)\, a(\rx,\zeta)$ and $\F_P$ the partial Fourier transform on the space $\S(\R^{2n},L(E_z))$ (only the variables in the second copy of $\R^{n}$ in $\R^{2n}$ being Fourier transformed). It is clear that $\F_{z,\bfr}$ is a topological isomorphism from $\S(\R^{2n},L(E_z))$ onto itself with inverse $\F_{z,\bfr}^{-1}=M_{1/\mu}\circ \ol \F_P$.
The fact that $\ol \F^*$ coincides with $\F$ on $\S(TM,L(E))$ is a consequence of the following equality 
$$
\int_{TM} \Tr(a (\ol \F(b)) ^*)\,d\mu^T = \int_{T^*M}\Tr(\F(a)b^*)\,d\mu^*
$$
for any $a\in \S(TM,L(E))$ and $b\in \S(T^*M,L(E))$, that is a direct consequence of the Parseval formula for $\F_P$.
\end{proof}

\section{Linearization and symbol maps}

\subsection{Linearization and the $\Phi_{\la}$, $\Ups_t$ diffeomorphisms}

Recall that a linearization (Bokobza-Haggiag \cite{Bokobza}) on a smooth manifold $M$ is defined as a smooth map $\nu$ from $M\times M$ into $TM$ such that $\pi \circ \nu = \pi_1$, $\nu(x,x)=0$ for any $x\in M$ and $(d_y\nu)_{y=x}=\Id_{T_x M}$. Using this map, it is then possible by restricting $\nu$ on a small neighborhood of the diagonal of $M\times M$, to obtain a diffeomorphism onto a neighborhood of the zero section of $TM$ and obtain an isomorphism between symbols (with a local control of the x variables on compact) and pseudodifferential operators modulo smoothing ideals. These isomorphisms depend on the linearization, as shown in \cite[Proposition V.3]{Bokobza}. We follow here the same idea, with a global point of view, since we are interested in the behavior at infinity. We thus consider, on the exponential manifold $(M,\exp,E,d\mu)$ a fixed linearization $\ol \psi$ that comes from an (abstract) exponential map $\psi$ on $M$ (also called linearization map in the following), so that $\ol \psi(x,y) = \psi_x^{-1} y$, and $\psi_x$ is a diffeomorphism from $T_x M$ onto $M$, with $\psi_x (0)=x$, $(d\psi_x)_0=\Id_{T_{x} M}$. For example, $\psi$ may be the exponential map $\exp$.

Let $\la \in [0,1]$ and $\Phi_\la$ be the smooth map from $TM$ onto $M\times M$ defined by 
$$
\Phi_\la: (x,\xi) \mapsto \big(\psi_x(\la \xi), \psi_x(-(1-\la)\xi) \big)\, .
$$
\begin{assum}
\label{assumH}
We suppose from now on that whenever the parameters $\la$, $\la'$, are in $]0,1[$, it is implied that the linearization map $\psi$ satisfies for any $x,y\in M$ and $t\in \R$, $\psi_{x}(t\psi_x^{-1}(y))= \psi_{y}((1-t)\psi_y^{-1}(x))$. This hypothesis, called $(H_\psi)$ in the following, is automatically satisfied if the linearization is derived from a exponential map of a connection on the manifold. 
\end{assum}
A computation shows that $\Phi_\la$ is a diffeomorphism with the following inverse $\Phi_\la^{-1}:(x,y)\mapsto \a'_{yx}(1-\la)$ for $\la\neq 0$ and $\Phi_0^{-1}(x,y):\mapsto -\a'_{xy}(0)$, where $\a_{xy}(t):=\psi_x(t\psi_x^{-1}(y))$.
Noting $\Phi_\la^{-1}(x,y)=:(m_\la(x,y),\xi_\la(x,y))$, we see that $m_\la(x,y)=\a_{xy}(\la)$ and, if $\la\neq 0$, $\xi_{\la}(x,y)=\tfrac{1}{\la}\psi_{m_{\la}(x,y)}^{-1}(x)$, while $\xi_0(x,y)=-\psi_x^{-1}(y)$.
In all the following, we shall use the symbol $W$ (for Weyl) for the value $\la=\half$, so that $m_W:=m_{\half}$, $\Phi_W:=\Phi_\half$, and similar conventions for the other mathematical symbols containing $\la$.  
Note that $m_\la$ is a smooth function from $M\times M$ onto $M$, with $m_\la(x,x)=x$ for any $x\in M$. Moreover, for any $x,y\in M$, $m_{\la}(x,y)=m_{1-\la}(y,x)$, $m_W(x,y)=m_W(y,x)$ (the ``middle point'' of $x$ and $y$), $\xi_{\la}(x,y)=-\xi_{1-\la}(y,x)$, $\xi_W (x,y)=-\xi_W (y,x)$ and $x\mapsto \Phi_\la^{-1} (x,x)$ is the zero section of $TM\to M$. Noting $j$ the involution on $M\times M$ : $(x,y)\mapsto (y,x)$, we have $\Phi_\la = j\circ \Phi_{1-\la}\circ -\Id_{TM}$. 

For any $t\in [-1,1]$ (with the convention that if $(H_\psi)$ is not satisfied, we are restricted to $t\in \set{-1,0,1}$), we define, 
$$
\Ups_t :(x,\xi)\mapsto \big(\psi_x (t\xi), \tfrac{-1}{t} \psi^{-1}_{\psi_x(t\xi)}(x)\big)
$$
with the convention $\tfrac{-1}{t} \psi^{-1}_{\psi_x(t\xi)}(x):= \xi$ if $t=0$, so that $\Ups_0=\Id_{TM}$.
A computation shows that $\Ups_{t}^{-1}=\Ups_{-t}$. The $\Phi_\la$ and $\Ups_t$ diffeomorphisms are related by the following property: for any $\la,\la'\in [0,1]$, $\Phi_{\la}^{-1}\circ \Phi_{\la'} = \Ups_{\la'-\la}$. We will use the shorthand $\Ups_{t,T}(x,\xi):=\tfrac{-1}{t} \psi^{-1}_{\psi_x(t\xi)}(x)$, so that $\Ups_{t}=(\psi\circ\ t\Id_{TM}, \Ups_{t,T} ).$

\begin{rem}
\label{Hhyp}Note that $(H_\psi)$ entails that
$(\Ups_t)_{t\in \R}$ is a one parameter subgroup of $\Diff(TM)$. 
\end{rem}

\begin{rem}
\label{remTP}
Suppose that $\psi$ is the exponential map associated to a connection $\nabla$ on $TM$, and $\a_{x,\xi}$ the unique maximal geodesic such that $\a'_{x,\xi}(0)=(x,\xi).$ It is a standard result of differential geometry (see for instance \cite[Theorem 3.3, p.206]{Lang}) that for any $v:=(x,\eta)\in TM$, and $\xi\in T_x(M)$, there exists an unique curve $\beta_{v}^\xi:\R \to TM$ such that 
$\nabla_{\a'_{v}} \beta_{v}^\xi =0$, $\pi\circ \beta_{v}^{\xi} = \a_{v}$ (in other words, $\beta_{v}^\xi$ is $\a_{v}$-parallel lift of $\a_{v}$)
and $\beta_{v}^\xi(0) = (x,\xi)$. By definition of geodesics, $\beta^{\eta}_{x,\eta}=\a'_{x,\eta}$. Moreover, $\beta_{x,\eta}^{\xi}(1)\in T_{\psi_x^\eta}(M)$, so we can define the following linear isomorphism of tangent fibers:
$P_{x,\eta}\ :\ T_x(M) \to T_{\psi_x^\eta}(M),   \quad \xi\mapsto \beta_{x,\eta}^\xi(1) \, .$
Note that $P_{x,\eta}^{-1}= P_{\psi_x^\eta,\psi_{\psi_x^\eta}^{-1}(x)}= P_{-\Ups_{1}(x,\eta)} = P_{\Ups_{-1}(x,-\eta)}$. The $P_{x,\xi}$ are the parallel transport maps along geodesics on the tangent bundle. These maps are related to the $\Ups_t$ diffeomorphisms, since a computation shows that for any $(x,\eta)\in TM$ and $t\in \R$,  $P_{x,t\eta}(\eta) = \Ups_{t,T}(x,\eta)$.
\end{rem}

If $(z,\bfr)$ is a frame, we define $\Phi_{\la,z,\bfr}:= n_{z,M^2}^{\bfr}\circ \Phi_\la \circ (n_{z,T}^\bfr)^{-1}$ and we denote $J_{\la,z,\bfr}$ its Jacobian. We also define $\Ups_{t,z,\bfr}=n_{z,T}^{\bfr}\circ \Ups_{t}\circ (n_{z,T}^{\bfr})^{-1}$ and the smooth maps from $\R^{2n}$ to $\R^{n}$: 
\begin{align*}
&\psi_z^\bfr : (\rx,\zeta) \mapsto n_z^\bfr \circ \psi \circ\, (n_{z,T}^\bfr)^{-1}(\rx,\zeta)\, ,\\
&\ol{\psi_z^\bfr}: (\rx,\ry) \mapsto  M_{z,(n_z^\bfr)^{-1}(\rx)}^\bfr\circ \psi^{-1}_{(n_z^\bfr)^{-1}(\rx)}\circ (n_{z}^\bfr)^{-1}(\ry).
\end{align*}
Noting $\psi_{z,\rx}^{\bfr}(\zeta):=\psi_{z}^\bfr(\rx,\zeta)$ and $\ol{\psi_{z,\rx}^{\bfr}}(\ry):=\ol{\psi_{z}^\bfr}(\rx,\ry)$, we have $(\psi_{z,\rx}^{\bfr})^{-1}=\ol{\psi_{z,\rx}^{\bfr}}$.
A computation shows that for any $(\rx,\zeta,\ry)\in \R^{3n}$,
\begin{equation}\label{Phi_la}
\Phi_{\la,z,\bfr}(\rx,\zeta) = (\psi_{z}^\bfr(\rx,\la\zeta),\psi_z^\bfr(\rx,-(1-\la)\zeta))  \, , \qquad
\Phi_{\la,z,\bfr}^{-1}(\rx,\ry) = (m_{\la,z,\bfr}(\rx,\ry),\xi_{\la,z,\bfr}(\rx,\ry))\, 
\end{equation}
where we defined the following functions: $m_{\la,z,\bfr}(\rx,\ry):= \psi_{z}^\bfr(\rx,\la\,\ol{\psi_{z}^\bfr}(\rx,\ry))$, $\xi_{0,z,\bfr}:= -\ol{\psi_z^\bfr}$ and for $\la\neq 0$, $\xi_{\la,z,\bfr}(\rx,\ry):= \tfrac{1}{\la} \ol{\psi_z^\bfr}(m_{\la,z,\bfr}(\rx,\ry),\rx)$. We also obtain for $t\in [-1,1]$, $(\rx,\zeta)\in \R^{2n}$, 
\begin{equation}\label{Ups_t}
\Ups_{t,z,\bfr}(\rx,\zeta) = \big(\psi_{z}^\bfr(\rx,t\zeta),\tfrac{-1}{t} \ol{\psi_z^\bfr}(\psi_z^\bfr(\rx,t\zeta),\rx)) =:  (\psi_{z}^\bfr(\rx,t\zeta),\Ups^{z,\bfr}_{t,T}(\rx,\zeta) \big)  \, ,
\end{equation}
and $\Ups_{0,z,\bfr}=\Id_{\R^{2n}}$. Note that $\Ups_{t,z,\bfr}(\rx,0)= (\rx,0)$ for any $\rx\in \R^n$ and $\Ups_{t,T}^{z,\bfr}=\tfrac{1}{t}\Ups^{z,\bfr}_{1,T}\circ I_{1,t}$ where $I_{r,r'}$ is the diagonal matrix with coefficients $I_{ii} = r$ for $i\leq n$ for $1\leq i\leq n$ and $I_{ii}=r'$ for $n+1\leq i \leq 2n$.

\subsection{$\O_M$-linearizations}

We intent to use the linearization to define topological isomorphisms between rapidly decaying section on $TM$ and $M\times M$. We thus need a control at infinity over the derivatives of the linearization $\psi$.

We note $\tau^{z,\bfr}=\tau^z\circ (n_{z,M^2}^\bfr)^{-1} \in C^\infty(\R^{2n},L(E_z))$. Remark that for any $(\rx,\ry)\in \R^{2n}$, $\tau^{z,\bfr}(\rx,\ry)$ is an unitary operator on $E_z$. We will also need the following functions parametrized by $t\in \R$: $\tau_t(x,\eta):=\tau_x(\psi_x(t\eta))$ for any $(x,\eta)\in TM$ and $\tau_{t}^{z,\bfr}(\rx,\zeta):=\tau^{z,\bfr}(\rx,\psi_z^\bfr(\rx,t\zeta))$. 

\begin{defn} A linearization $\psi$ on the exponential manifold $(M,\exp,E,d\mu)$ is said to be a $\O_M$-linearization if for any frame $(z,\bfr)$ the functions $\psi_z^\bfr$ and $\ol{\psi_z^\bfr}$ are  in $\O_M(\R^{2n},\R^n)$ and the functions $\tau_{1}^{z,\bfr}$ and $(\tau_1^{z,\bfr})^{-1}$ are in $\O_M(\R^{2n},L(E_z))$. We will say that $(M,\exp,E,d\mu, \psi)$ has a $\O_M$-bounded geometry, if it the case of $(M,\exp,E,d\mu)$ and $\psi$ is a $\O_M$-linearization.
\end{defn}

\begin{lem} 
\label{OMlem}
Suppose that $\psi$ is a $\O_M$-linearization. Then for any frame $(z,\bfr)$, $\la \in [0,1]$ and $t\in [-1,1]$,

\noindent (i) $\Phi_{\la,z,\bfr} \in \O_M^\times(\R^{2n},\R^{2n})$ and $J_{\la,z,\bfr} \in \O_M^\times(\R^{2n})$,

\noindent (ii) $\Ups_{t,z,\bfr} \in \O_M^\times(\R^{2n},\R^{2n})$ and $J(\Ups_{t,z,\bfr}) \in \O_M^\times(\R^{2n})$,

\noindent (iii) $\tau_t^{z,\bfr}$ and $(\tau_t^{z,\bfr})^{-1}$ are in $\O_M(\R^{2n},L(E_z))$.
\end{lem}
\begin{proof}
\noindent $(i)$ By (\ref{Phi_la}), we have $\Phi_{\la,z,\bfr} = (\psi_{z}^\bfr\circ I_{1,\la},\psi_z^\bfr\circ I_{1,\la-1})$ and $\Phi_{\la,z,\bfr}^{-1}=(m_{\la,z,\bfr},\xi_{\la,z,\bfr}) $ where $m_{\la,z,\bfr}=\psi_z^\bfr\circ I_{1,\la}\circ (\pi_1,\ol{\psi_z^\bfr})$ and 
if $\la \neq 0$, $\xi_{\la}=\tfrac{1}{\la} \ol{\psi_z^\bfr}\circ (m_{\la,z,\bfr},\pi_1)$, while $\xi_{0,z,\bfr}=-\ol{\psi_z^\bfr}$. Thus, the result is a consequence Lemma \ref{Ssigma} $(iii)$ and $(vi)$.

\noindent $(ii)$ By (\ref{Ups_t}), we have for $t\neq 0$, $\Ups_{t,z,\bfr}=(\psi_{z}^\bfr\circ I_{1,t}, \tfrac{-1}{t}\ol{\psi_z^\bfr}\circ (\psi_z^\bfr\circ I_{1,t},\pi_1))$. The result follows again from Lemma \ref{Ssigma} $(iii)$ and $(vi)$. 

\noindent $(iii)$ We have $\tau_{t}^{z,\bfr}=\tau_1^{z,\bfr}\circ I_{1,t}$ and $(\tau_{t}^{z,\bfr})^{-1}=(\tau_1^{z,\bfr})^{-1}\circ I_{1,t}$ so the result follows from Lemma \ref{Ssigma} $(iii)$.  
\end{proof}

The following lemma shows that we can obtain topological isomorphisms on spaces of rapidly decaying functions from the functions $\tau_t$ and $\Phi_{\la}$.

\begin{lem}
\label{Stopolisom}
Let $p\in \N^*$, $\tau\in \O_M^\times(\R^p,GL(E_z))$ and $\Phi \in \O_M^\times (\R^p,\R^p)$. Then the maps $L_{\tau}:= u\mapsto \tau u$, $R_{\tau}:= u\mapsto u\tau$ and $C_\Phi:= u\mapsto u\circ \Phi$ are topological isomorphisms of $\S(\R^p, L(E_z))$.
\end{lem}
\begin{proof} Since $L_{\tau}^{-1}= L_{\tau^{-1}}$, $R_{\tau}^{-1}=R_{\tau^{-1}}$ and $C_\Phi^{-1}=C_{\Phi^{-1}}$, we only need to check the continuity of $L_{\tau}$, $R_{\tau}$ and $C_\Phi$. The continuity of $L_{\tau}$ and $R_{\tau}$ is a direct application of Leibniz formula. Let $\nu$ be a $p$-multi-index and $r\in \N$. Theorem \ref{FaaCS} implies that for any $u\in \S(\R^p,L(E_z))$, 
$$
q_{\nu,N}(u\circ \Phi) \leq \sum_{|\la|\leq |\nu|} \sup_{\rx\in \R^p} \langle \rx \rangle ^{N}|P_{\nu,\la}(\Phi)(\rx)| \norm{(\del^\la u)\circ\Phi (\rx) }_{L(E_z)}
$$
where the functions $P_{\nu,\la}(\Phi)$ are such that $|P_{\nu,\la}(\Phi)(\rx)| \leq C_\nu \langle \rx\rangle ^{q_\nu}$ for a $q_\nu \in \N^*$ and a $C_\nu>0$. Since $\langle \Phi^{-1} (\rx) \rangle \leq C \langle \rx\rangle ^r$ for a $r\in \N^*$ and a $C>0$, we see that there is $C'_\nu>0$ such that $q_{\nu,N}(u\circ \Phi) \leq C'_\nu \sum_{|\la|\leq |\nu|} q_{\la,(q_\nu+N)r}(u)$, which gives the result. 
\end{proof}

\begin{lem} If $(M,\exp,E,d\mu)$ has a $\O_M$-bounded geometry and $\psi$ is a linearization such that there exists $(z_0,\bfr_0)$ such that the functions $\psi_{z_0}^{\bfr_0}$, $\ol \psi_{z_0}^{\bfr_0}$ are in $\O_M(\R^{2n},\R^n)$ and $\tau_1^{z_0,\bfr_0}$, $(\tau_{1}^{z_0,\bfr_0})^{-1}$ are in $\O_M(\R^{2n},L(E_{z_0}))$, then $\psi$ is a $\O_M$-linearization.
\end{lem}
\begin{proof} The result is a direct consequence of the formulas $\psi_{z}^\bfr=\psi_{z,z_0}^{\bfr,\bfr_0}\circ \psi_{z_0}^{\bfr_0}\circ \psi_{z_0,z,T}^{\bfr_0,\bfr}$, $\ol\psi_{z,\rx}^\bfr(\ry)=(d\psi_{z_0,z}^{\bfr_0,\bfr})^{-1}_\rx \, \ol\psi_{z_0}^{\bfr_0}\circ \psi_{z_0,z,M^2}^{\bfr_0,\bfr}(\rx,\ry)$ and $\tau^{z,\bfr}=(\tau_z^{-1}\tau_{z_0})\circ \pi_2 \circ (n_{z,M^2}^{\bfr})^{-1}\,\tau^{z_0,\bfr_0}\circ \psi_{z_0,z,M^2}^{\bfr_0,\bfr} \, 
(\tau_{z_0}^{-1}\tau_{z})\circ \pi_1 \circ (n_{z,M^2}^{\bfr})^{-1}$.
\end{proof}

\subsection{Symbol maps and $\la$-quantization} \label{quantisec}

\begin{assum} We suppose in this section and in section \ref{moyalsection} that $(M,\exp,E,d\mu, \psi)$ has a $\O_M$-bounded geometry.
\end{assum}

The operator $\F$ is a topological isomorphism from $\S'(TM,L(E))$ onto $\S'(T^*M,L(E))$.
We shall now introduce a topological isomorphism between $\S'(M\times M,L(E))$ and $\S'(TM,L(E))$. 
We define the linear application $\Ga_\la$ from $C^\infty(M\times M,L(E))$ into $C^\infty(TM,L(E)))$:
$$
\Ga_\la(K): v \mapsto K^{\pi(v)}\circ \Phi_\la(v)\, .
$$
 As a consequence, $\Ga_\la(K)= \tau^{-1}_{\la}\,(K\circ \Phi_\la)\,\tau_{\la-1}$ and $\Ga_\la^{-1}(a)=(\tau_{\la}\,a\,\tau_{\la-1}^{-1})\circ \Phi_\la^{-1}$. For a given frame $(z,\bfr)$, we denote $\Ga_{\la,z,\bfr}:=T_{z,\bfr,T}\circ \Ga_\la \circ T_{z,\bfr,M^2}^{-1}$. A computation shows that for any smooth function $u\in C^{\infty}(\R^{2n},L(E_z))$, $\Ga_{\la,z,\bfr}(u)=(\tau_{\la}^{z,\bfr})^{-1}(u\circ \Phi_{\la,z,\bfr}) \tau_{\la-1}^{z,\bfr}$.

Let us define the smooth strictly positive functions on $\R^{2n}$ and $M\times M$ respectively:
\begin{equation}
\label{mu_la_def}
\mu_{\la,z,\bfr}(\rx,\ry):= \tfrac{\mu_{z,\bfr}(\rx)\mu_{z,\bfr}(\ry)}{\mu_{z,\bfr}^2(m_{\la,z,\bfr}(\rx,\ry))}\, |J_{\la,z,\bfr}|\circ \Phi_{\la,z,\bfr}^{-1}(\rx,\ry) \qquad \mu_{\la}:=\mu_{\la,z,\bfr}\circ (n_{z}^\bfr,n_z^\bfr).
\end{equation}
It is straithtforward to check that $\mu_{\la}$ is indeed independent of $(z,\bfr)$. Note that $\mu_{1-\la}(x,y)=\mu_\la(y,x)$.
Since $\mu_{\la,z,\bfr}\in \O_M^\times(\R^{2n})$, the operator of multiplication $M_{\mu_{\la}}$ is a topological isomorphism on $\S(M\times M,L(E))$. Note also that $\Ga_\la \circ M_{\mu_\la} = M_{\mu_\la\circ \Phi_\la}\circ \Ga_\la$.

\begin{prop} $\Ga_\la$ is a topological isomorphism from $\S(M\times M,L(E))$ onto $\S(TM,L(E))$. Moreover,
$\wt \Ga_\la \circ j_{M^2}= j_{TM}\circ \Ga_\la \circ M_{\mu_\la}$, where $\wt \Ga_\la:={\Ga_\la^{-1}}^*$.
\end{prop}
\begin{proof} Let $(z,\bfr)$ be a frame. It suffices to prove that $\Ga_{\la,z,\bfr}$
is a topological isomorphism from $\S(\R^{2n},L(E_z))$ onto itself. Since $\Ga_{\la,z,\bfr}=L_{(\tau_\la^{z,\bfr})^{-1}}\circ R_{\tau_{\la-1}^{z,\bfr}}\circ C_{\Phi_{\la,z,\bfr}}$, the result follows from Lemma \ref{Stopolisom} and Lemma \ref{OMlem} $(i)$ and $(iii)$.
Let $u,v\in \S(\R^{2n},L(E_z))$. We have (with $j$ the canonical inclusion from $\S(\R^{2n},L(E_z))$ into $\S'(\R^{2n},L(E_z))$: 
\begin{align*}
(\wt\Ga_{\la,z,\bfr}\circ j (u))(v)&= \int_{\R^{2n}} \Tr\big(u(\rx,\ry)(\Ga_{\la,z,\bfr}^{-1}(v)(\rx,\ry))^*\big)\,d\rx\, d\ry \,  \\
&=\int_{\R^{2n}} \Tr\big((\tau_{\la}^{z,\bfr})^{-1}\circ \Phi_{\la,z,\bfr}^{-1}(\rx,\ry)\  u(\rx,\ry)\  \tau_{\la-1}^{z,\bfr}\circ \Phi_{\la,z,\bfr}^{-1}(\rx,\ry)\ \\
&\hspace{3cm}  v^*\circ \Phi_{\la,z,\bfr}^{-1}(\rx,\ry)\big)\, d\rx\, d\ry \\ 
&=\int_{\R^{2n}} \Tr(\Ga_{\la,z,\bfr}(u)(m,\zeta) v^*(m,\zeta)) |J_{\la,z,\bfr}|(m,\zeta)\, dm\, d\zeta \\
&= (j\circ M_{|J_{\la,z,\bfr}|}\circ \Ga_{\la,z,\bfr}(u)) (v)
\end{align*} where we used the following change of variables $(m,\zeta):=\Phi_{\la,z,\bfr}^{-1}(\rx,\ry)$.
Thus, we have $\wt\Ga_{\la,z,\bfr}\circ j = j\circ M_{|J_{\la,z,\bfr}|}\circ \Ga_{\la,z,\bfr}$. The relation $\wt \Ga_\la \circ j_{M^2}= j_{TM}\circ \Ga_\la \circ M_{\mu_\la}$ now follows since $M_{|J_{\la,z,\bfr}|}\circ \Ga_{\la,z,\bfr} = \Ga_{\la,z,\bfr}\circ M_{|J_{\la,z,\bfr}|\circ \Phi_{\la,z,\bfr}^{-1}}$, $T_{z,\bfr,T}^*\circ j\circ M_{\mu^2_{z,\bfr}}=j_{TM}\circ T_{z,\bfr,T}^{-1}$ and $T_{z,\bfr,M^2}^*\circ j \circ M_{\mu_{z,\bfr}\ox \mu_{z,\bfr}}=j_{M^2}\circ T_{z,\bfr,M^2}^{-1}$. 
\end{proof}

As a consequence, $\wt \Ga_\la$ is a topological isomorphism from the space of tempered distributional $L(E)$-sections on $M\times M$, $\S'(M\times M,L(E))$ onto $\S'(TM,L(E))$ and when restricted (in the sense of the previous continous inclusions) to $\S(M\times M,L(E))$, is equal to $\Ga_\la \circ M_{\mu_\la}^{-1}$, so provides a topological isomorphism from $\S(M\times M,L(E))$ onto $\S(TM,L(E))$. Fourier transform coupled with $\wt\Ga_\la$ lead us to the following natural isomorphism from $\S'(M\times M,L(E))$ onto $\S'(T^*M,L(E))$. 

\begin{defn} Let $\la\in [0,1]$. The $\la$-symbol map is the topological isomorphism from $\S'(M\times M,L(E))$ onto $\S'(T^*M,L(E))$: $\sigma_\la:=\F\circ \wt \Ga_\la$. The $\la$-quantization map is the inverse of $\sigma_\la$, noted $\mathfrak{Op}_\la$.
\end{defn}

Thus, the data of a tempered distributional section on the cotangent bundle (i.e. an element of $\S'(T^*M,L(E))$) determines in an unique way (for a given $\la$), an operator continuous from $\S$ to $\S'$, and vice versa. Remark that $\sigma_\la \circ j_{M^2}= j_{T^*M}\circ \F\circ \Ga_\la\circ M_{\mu_\la}$ and $\Op_\la\circ j_{T^*M}=j_{M^2}\circ M_{1/\mu_{\la}}\circ \Ga_\la^{-1}\circ \ol\F$. If $(z,\bfr)$ is a frame then, noting $\Op_{\la,z,\bfr}:= \wt T_{z,\bfr,M^2} \circ \Op_\la\circ \wt T_{z,\bfr,*}^{-1}$, we obtain $\Op_{\la,z,\bfr}=\Ga_{\la,z,\bfr}^*\circ M_{\mu_{z,\bfr}}^*\circ \F_P^*$ so that for any $u\in \S(\R^{2n},L(E_z))$ and $b\in \O_M(\R^{2n},L(E_z))$, 
\begin{equation}
\langle \Op_{\la,z,\bfr}(b),u\rangle= \int_{\R^{3n}} e^{2\pi i \langle \zeta,\vth\rangle} \Tr\big(\mu b(\rx,\vth) (\Ga_{\la,z,\bfr}(u))^*(\rx,\zeta)\big)\,d\zeta\,d\vth\,d\rx\, .\label{eq-oplazb}
\end{equation}
where $\mu b:(\rx,\vth)\mapsto  \mu_{z,\bfr}(\rx)\,b(\rx,\vth)$.

\subsection{Moyal product} \label{moyalsection}

The applications $\Op_{0}$, $\Op_{1}$, $\mathfrak{Op}_W:=\Op_\half$ are respectively the normal, antinormal and Weyl quantization maps. Remark that for any $T\in \S'(T^*M,L(E))$, $\mathfrak{Op}_\la(T^*) = (\mathfrak{Op}_{1-\la}(T))^\dag$. In particular
$$
\mathfrak{Op}_0(T^*) = (\mathfrak{Op}_1(T))^\dag\, ,\qquad \mathfrak{Op}_W(T^*) = (\mathfrak{Op}_W(T))^\dag 
$$
where $\dag$ is the topological isomorphism of $\S'(M\times M,L(E))$ defined as $\langle K^\dag, u\rangle:=\ol{\langle K , u^*\circ j\rangle}$ with $j$ the diffeomorphism on $M\times M : (x,y)\mapsto (y,x)$ and $u\in \S(M\times M,L(E))$. The kernel of the adjoint $A^\dag$ of any operator $A\in L(\S,\S')$ is $(K_A)^\dag$.
As a consequence, $\sigma_\la$ is a linear topological isomorphism (and a $*$-isomorphism in the case of the Weyl quantization) from the algebra $\Re(\S)=L(S,S)\cap L(S',S')$ of regular operators onto its image $\mathfrak{M}_\la:=\sigma_\la(\Re(\S))$. We can transport the operator composition in the world of functions, by defining the $\la$-product on $\mathfrak{M}_\la$ as
$$
T\circ_\la T':= \sigma_\la(\mathfrak{Op}_\la(T)\,\mathfrak{Op}_\la(T'))
$$
so that $\mathfrak{M}_\la$ forms an algebra, and $\mathfrak{M}_\la^*=\mathfrak{M}_{1-\la}$. In the case of $\la=\half$, we recover the Moyal $*$-algebra $\mathfrak{M}_W$ and the Moyal product $\circ_W$. The space $\Psi^{-\infty}(M)\simeq \S(M\times M,L(E))$ of isotropic smoothing operators being an $*$-ideal of $\Re(\S)$, the space $\S(T^*M,L(E))=\sigma_{\la}(\Psi^{-\infty}(M))$ forms an ideal of $\mathfrak{M}_\la$. 
Since we will focus on the pseudodifferential calculus over $M$, we shall not investigate in this chapter the full analysis of the Moyal product over $T^*M$. Note however the following property on $\S(T^*M):=\S(T^*M,L(M\times \C)) \simeq \S(T^*M,\C)$:

\begin{prop}
\label{la-product}
$(\S(T^*M),\circ_\la)$ is a (noncommutative, nonunital) Fr\'{e}chet algebra. Moreover,
$$ 
a\circ_\la b\,(x,\eta) = \int_{T_x(M) \times M} d\mu_{x}(\xi)d\mu(y)\int_{V^\la_{x,\xi,y}} d\mu_{x,\xi,y}^*(\th,\th')\, g^\la_{x,\xi,y}\,e^{2\pi i \om^\la_{x,\xi,y}(\eta,\th,\th')} a(y^\la_{x,\xi},\th)\,b(y^{1-\la}_{x,-\xi},\th')
$$
where $y_{x,\xi}^{\la}:=m_\la(\psi_x^{\la\xi},z)$, $\ol y_{x,\xi}^{\la}:=\xi_{\la}(\psi_x^{\la\xi},z)$ and 
\begin{align*}
&V_{x,\xi,y}^\la:= T^*_{y_{x,\xi}^\la}(M)\times T^*_{y^{1-\la}_{x,-\xi}}(M)\, ,\qquad d\mu_{x,\xi,y}^*(\th,\th'):=d\mu^*_{y^\la_{x,\xi}}(\th)\,d\mu^*_{y^{1-\la}_{x,-\xi}}(\th')\, , \\
&g_{x,\xi,y}^\la:= \tfrac{\mu_\la(\psi_x^{\la \xi},\psi_x^{(\la-1)\xi})}{\mu_\la(\psi_x^{\la\xi},y)\,\mu_\la(y,\psi_x^{(\la-1)\xi})}\, , \\
&\om_{x,\xi,y}^\la(\eta,\th,\th'):=\langle \th,\ol y_{x,\xi}^\la\rangle -  \langle \th', \ol y_{x,-\xi}^{1-\la}\rangle - \langle \eta ,\xi \rangle \, .
\end{align*}

\end{prop}
\begin{proof} The product $a\circ_\la b$ on $\S(T^*M)$ is obtained by computation of $\F\circ \Ga_\la\circ M_{\mu_\la}\circ \big( (M_{\mu_\la}^{-1}\circ \Ga_{\la}^{-1}\circ \ol \F(a))\circ_V (M_{\mu_\la}^{-1}\circ \Ga_\la^{-1}\circ \ol \F(b))\big)$, where $\circ_V$ is the Volterra product of kernels. Since $\sigma_\la$ is a topological isomorphism between $\S(M^2)$ and $\S(T^*M)$, the continuity of the Moyal product is equivalent to the continuity of  $\circ_V$, which is equivalent to the continuity of the following product on $\S(\R^{2n})$: 
$$
K\cdot K' (\rx,\ry):= \int_{\R^{n}} K(\rx,\rt)K(\rt,\ry) \mu_{z,\bfr}(\rt) d\rt.
$$
The continuity of this product is obtained by the following estimates
$$
q_{p,(\a,\b)}(K\cdot K') \leq C \, q_{2(p+r),(\a,0)}(K)\,q_{p,(0,\b)}(K'),\qquad  q_{p,\nu}(K):=\sup_{(\rx,\ry)\in \R^{2n}} \langle(\rx,\ry) \rangle^p |\del^\nu K (\rx,\ry)|
$$
where $|\mu_{z,\bfr}(\rt)|\leq C_1 \langle \rt \rangle^{r-n-1}$ and $C:=C_1\int_{\R^n} \langle\rt \rangle^{-(n+1)} d\rt$. 
\end{proof}

\begin{rem} $(\S(T^*M),\circ_W)$ is a $*$-algebra since $(a\circ _W b)^*= b^* \circ_W a^*$ for any $a,b\in \S(T^* M)$. We can also construct another $*$-algebra on $\S(T^*M)$ with the product $a \star b := \half( a\circ_0 b + a\circ_1 b)$. This proves that when $(H_\psi)$ (see Assumption \ref{assumH}) is not satisfied (so that no middle point exist in the classical world) we can still have a canonical star-product on $\S(T^*M)$ which satisfies $(a\star b)^*= b^* \star a^*$.
\end{rem}

\section{Symbol calculus of pseudodifferential operators}

\subsection{Symbols}

\begin{assum} Let $\sigma\in [0,1]$. We suppose in this section that $(M,\exp,E)$ has a $S_{\sigma}$-bounded geometry.   
\end{assum}

The algebra $\Re(\S)$ and $\Psi^{-\infty}$ are respectively too big and too small to develop a satisfactory pseudodifferential calculus that allows an efficient utilization of symbol maps. We shall in this section define some spaces of symbols that will be used to define later special algebras of pseudodifferential operators that lie between $\Re(\S)$ and $\Psi^{-\infty}$.

\begin{defn} A symbol of degree $(l,m)\in \R^2$ of type $\sigma$, on $M$ is a smooth section $a\in C^\infty(T^*M,L(E))$ such that for any $(z,\bfr)$ and any $n$-multi-indices $\a$, $\beta$, there exists $K>0$ such that
\begin{equation}\label{symbolIneq}
\norm{ \del^{(\a,\b)}_{z,\bfr} a^{z} (x,\th)}_{L(E_z)}\leq K \langle x\rangle^{\sigma(l- |\a|)}_{z,\bfr}\,\langle \th\rangle^{m-|\b|}_{z,\bfr,x}
\end{equation}
is valid for all $(x,\th)\in T^*M$. The space of symbols of degree $(l,m)$ and type $\sigma$ is noted $S^{l,m}_\sigma$. 
\end{defn}

Remark that $S^{l,m}_0$ is independent of $l$, so we denote this space $S^m_0$. We note $S^{-\infty}_\sigma:=\cap_{l,m} S^{l,m}_\sigma$ and in the case $\sigma>0$, we define $S^{-\infty}:= S^{-\infty}_\sigma=\S(T^*M,L(E))$ (it is independent of $\sigma>0$). We set $S^\infty_\sigma := \cup_{l,m} S^{l,m}_\sigma$. We define similarly $S^{l,m}_{\sg,z}:=S^{l,m}_\sigma(\R^{2n},L(E_z))$, without reference to a frame.  

Since $M$ has a $S_{\sigma}$-bounded geometry, we get the following coordinate independence of the previous definition:

\begin{prop} Let $a\in C^\infty(T^*M,L(E))$. Then $a \in S^{l,m}_\sigma$ if and only if there exists a frame $(z,\bfr)$ such that 
$a$ satisfies (\ref{symbolIneq}).
\end{prop}
\begin{proof} 
Suppose that (\ref{symbolIneq}) is satisfied for $(z',\bfr')$ and let ($z,\bfr$) be another frame. For $(x,\th)\in T^*M$ and $\a,\b$ two $n$-multi-indices with $\nu=(\a,\b)\neq 0$, we get from Equation (\ref{del-cchange}) and Lemma \ref{cchange},
\begin{align*} 
\norm{\del_{z,\bfr}^{\nu} a^z(x,\th)}_{L(E_z)}  \leq K \sum_{\a',\b'} \sum_{\rho\leq \la \leq \nu'}\ 
\langle x\rangle_{z,\bfr}^{\sigma(|\a'|-|\a|)}\langle\th \rangle_{z,\bfr,x}^{|\b'|-|\b|} \langle x \rangle_{z',\bfr'}^{\sigma(|\la^1|-|\a'|)}\\
\hspace{3cm} \times \langle x\rangle_{z',\bfr'}^{\sigma(l-|\rho^1|)}\langle \th \rangle_{z',\bfr',x}^{m-|\b'|}\langle x \rangle_{z',\bfr'}^{\sigma(|\rho^1|-|\la^1|)}.
\end{align*} 
Using (\ref{i-first}), (\ref{i-bis}) and the fact that $|\a|\geq |\rho^1|$, we get the result.
\end{proof} 

\begin{cly} If $a\in C^\infty(T^*M,L(E))$, then $a\in S^{l,m}_\sigma$ if and only if for any $(z,\bfr)$, $a^z\circ (n^\bfr_{z,*})^{-1} \in S^{l,m}_\sigma(\R^{2n},L(E_z))$, or equivalently, there exists $(z,\bfr)$ such that $a^z\circ (n^\bfr_{z,*})^{-1} \in S^{l,m}_\sigma(\R^{2n},L(E_z))$.
\end{cly}

We see that $S^{l,m}_\sigma \cdot S^{l',m'}_{\sigma} \subseteq S^{l+l',m+m'}_{\sigma}$ where $\cdot$ is the composition of sections induced by the matrix product on the fibers of $L(E)$. Moreover, $S^{l,m}_\sigma \subseteq S^{l',m'}_{\sigma}$ for $m\leq m'$ and $l\leq l'$. Thus, $S^{\infty}_\sigma$ is a $*$-algebra, which is bigraduated for $\sigma>0$ and graduated for $\sigma=0$. Remark also that $S^{-\infty}\cdot S^{m}_0$ and $S^{m}_0\cdot S^{-\infty}$ are included in $S^{-\infty}$. Note that if $f\in S^{l,m}_\sigma(T^*M)$ (a symbol where $M$ has its trivial bundle $M\times \C$), then $a_f(x,\th):=f(x,\th)I_{L(E_x)}$ defines a symbol in $S^{l,m}_\sigma$. Such symbols will be called scalar symbols.
Note also that if $a\in S^{l,m}_\sigma$, then $\del_{z,\bfr}^{(\a,\b)}a:=(\tau_z\circ \pi )(\del^{(\a,\b)}_{z,\bfr} a^z)(\tau_z^{-1}\circ\pi)\in S^{l-|\a|,m-|\b|}_{\sigma}$.

If $f\in S_\sigma(\R^n)$ then $(\rx,\vth)\mapsto f(\rx) \Id_{L(E_z)} \in S_\sigma^{0,0}(\R^n,L(E_z))$. In particular  $(\rx,\vth)\mapsto \mu_{z,\bfr}^{\pm 1}(\rx) \Id_{L(E_z)} \in S_\sigma^{0,0}(\R^n,L(E_z))$ if $d\mu$ is a $S_\sigma^\times$-density.

\begin{rem} We note $PS^{l,m}_\sigma(\R^{2n},L(E_z))$ the subspace of $S^{l,m}_\sigma(\R^{2n},L(E_z))$ consisting of functions of the form $\sum_{1\leq i\leq (\dim E_z)^2} P_i e_i$ where $(e_i)$ is a linear basis of $L(E_z)$ and $P_i$ are of the form $\sum_{\beta} c_{i,\b}(\rx)\vth^{\b} $ (finite sum over the $n$-multi-indices $\b$), where for any $i,\b$, $\del^\a c_{i,\b}(\rx)=\O(\langle \rx\rangle^{\sigma(l-|\a|)})$ for any $n$-multi-indices $\a$, and $m=\max_i \deg_\vth P_i$. We check that this definition is independent of the chosen basis $(e_i)$. 

We call polynomial symbol of order $l,m$ and type $\sigma$ any section of the form  $(\tau_{z}\circ \pi) (P\circ n_{z,*}^{\bfr})(\tau_z^{-1}\circ\pi)$ where $P\in PS^{l,m}_\sigma(\R^{2n},L(E_z))$ and $(z,\bfr)$ is a frame. This definition is independent of $(z,\bfr)$. We note $PS^{l,m}_\sigma$ the subspace of $S^{l,m}_\sigma$ consisting of polynomial symbols of order $l,m$ and type $\sigma$. Remark that the section $I:(x,\th)\mapsto I_{L(E_x)}$ 
is in $PS^{0,0}_1$.
\end{rem}

We now topologize the symbol spaces:

\begin{lem}
\label{toposymbol}
The following semi-norms on $S_\sigma^{l,m}$, for $N\in \N$, 
$$
q_{(\a,\b)}(a):= \sup_{(x,\th)\in T^*M} \langle x \rangle_{z,\bfr}^{\sigma(|\a|-l)} \langle \th \rangle_{z,\bfr,x}^{|\b|-m}\norm{\del^{(\a,\b)}_{z,\bfr}a^z(x,\th)}_{L(E_z)}
$$
determine a Fr\'{e}chet topology on $S_\sigma^{l,m}$, which is independent of $(z,\bfr)$.
The applications $T_{z,\bfr,*}$ are topological isomorphisms from $S_\sigma^{l,m}$ onto $S_\sigma^{l,m}(\R^{2n},L(E_z))$. The following inclusions are continous for these topologies: $S^{l,m}_\sigma \cdot S^{l',m'}_{\sigma} \subseteq S^{l+l',m+m'}_{\sigma}$, $S^{l,m}_\sigma \subseteq S^{l',m'}_{\sigma}$ ($m\leq m'$ and $l\leq l'$) and $S^{-\infty}_\sigma\subseteq S^{l,m}_\sigma$. Moreover, the last inclusion is dense when  $S^{l,m}_\sigma$ has the topology of $S^{l',m'}_\sigma$ for $m< m'$ and $l<l'$.
\end{lem}
\begin{proof} The independence of the topology for $(z,\bfr)$ follows from the easily checked estimate for any $(\a,\b)$,
$$
q_{(\a,\b)}^{(z,\bfr)}(a) \leq K_{\a,\b} \underset{|\b'|\geq |\b|\,,\ga\leq \a'}{\sum_{0\leq |(\a',\b')|\leq |(\a,\b)|}} q_{(\ga,\b')}^{(z',\bfr')}(a) .
$$ 
where $K_{\a,\b}>0$. By construction the applications $T_{z,\bfr,*}$ are clearly topological isomorphisms from $S_\sigma^{l,m}$ onto $S_\sigma^{l,m}(\R^{2n},L(E_z))$. The continuity of $S^{l,m}_\sigma \cdot S^{l',m'}_{\sigma} \subseteq S^{l+l',m+m'}_{\sigma}$, $S^{l,m}_\sigma \subseteq S^{l',m'}_{\sigma}$ ($m\leq m'$ and $l\leq l'$) and $S^{-\infty}_\sigma\subseteq S^{l,m}_\sigma$ are straightforward. Following \cite{Melrose}, to prove the density result, we shall prove the stronger property: for any $a\in S^{l,m}_{\sigma}(\R^{2n},L(E_z))$ the sequence 
$$
a_p(\rx,\vth):= (\rho(\rx/p))^{1-\delta_{\sigma,0}}\, \rho(\vth/p)\, a(\rx,\vth) 
$$
converges to $a$ for the topology of $S^{l',m'}_{\sigma}(\R^{2n},L(E_z))$ where $m'>m$ and $l'>l$. Here $\rho\in C^\infty_c(\R^{n},[0,1])$ with $\rho=1$ on $B(0,1)$ and $\rho=0$ on $\R^{n}\backslash B(0,2)$. First, it is clear that $a_p \in S^{-\infty}_\sigma(\R^{2n},L(E_z))$. Noting $R_p(\rx,\vth):=  \langle \rx \rangle^{\sigma(|\a|-l')} \langle \vth \rangle^{|\b|-m'}\norm{\del^{(\a,\b)}(a-a_p)(\rx,\vth)}_{L(E_z)}$ for a given $2n$-multi-index $\nu:=(\a,\b)$, we get with Leibniz rule, for a $K>0$ (by convention $\nu'<\nu$ if and only if $\nu'\leq \nu$ and $\nu'\neq \nu$):
$$
\tfrac{1}{K}\,R_p (\rx,\vth)\leq \Delta_p(\rx,\vth) \langle \rx \rangle^{\sigma(l-l')} \langle \vth \rangle^{m-m'}+ \sum_{\nu'< \nu}   |\del^{\nu-\nu'}\Delta_p(\rx,\vth)|\langle \rx \rangle^{\sigma(l-l'+|\a|-|\a'|)} \langle \vth \rangle^{m-m'+|\b|-|\b'|}
$$
where $\Delta_p(\rx,\vth):=1-(\rho(\rx/p))^{1-\delta_{\sigma,0}} \rho(\vth/p)$. Suppose that $\sigma=0$. In that case, 
$|\Delta_p(\rx,\vth)|\leq 1_{[p,+\infty[}(\vth)$ and if $\nu'<\nu$, 
\begin{equation}
\label{Deltasigma0}
|\del^{\nu-\nu'}\Delta_p(\rx,\vth)|\leq \delta_{\a,\a'}\, K_{\b}\, p^{-|\b|+|\b'|}\, 1_{[p,2p]}(\vth)
\end{equation}
where $1_{[r,r']}$ is the characteristic function of the annulus $A_{r,r'}:=\set{\vth\in \R^{n} \ : \ r\leq \norm{\vth} \leq r'}$ and $K_{\b}:=\sup_{\b'<\b}\norm{\del^{\b-\b'} \rho}_\infty$. As a consequence, for $K'>0$,
$$
\tfrac{1}{K}\,R_p (\rx,\vth)\leq \langle p \rangle^{m-m'} +K_\b \sum_{\nu'< \nu} \delta_{\a,\a'}\, 1_{[p,2p]}(\vth)\ p^{-|\b|+|\b'|}\, \langle \vth \rangle^{m-m'+|\b|-|\b'|} \leq K' \langle p \rangle^{m-m'}
$$ 
and the result follows. Suppose now $\sigma\neq 0$. In that case $|\Delta_p(\rx,\vth)|\leq 1_{F_p}(\rx,\vth)$ where $F_p:=\R^{2n}-B(0,p)^2$ and if $\nu'<\nu$, for a constant $K_\nu>0$
\begin{equation}
\label{Deltasigma1}
|\del^{\nu-\nu'}\Delta_p(\rx,\vth)|\leq K_\nu\, 1_{[\sgn(\a-\a')p,2p]}(\rx)\  1_{[\sgn(\b-\b')p,2p]}(\vth)\, p^{-|\nu|+|\nu'|}\, .
\end{equation}
As a consequence, for $K',K''>0$, and with $r:=\max \{m-m',\sigma(l-l')\}<0$,
$$
\tfrac{1}{K}\,R_p (\rx,\vth)\leq \langle p \rangle^{r} +K'\sum_{\nu'< \nu}  1_{[\sgn(\a-\a')p,2p]}(\rx)\  1_{[\sgn(\b-\b')p,2p]}(\vth)\, \langle \rx \rangle^{\sigma(l-l')}\langle \vth \rangle^{m-m'} \leq K'' \langle p \rangle^{r}
$$ 
and the result follows.
\end{proof}

Note that $S^{-\infty}:=\cap_{l,m} S^{-\infty}_{\sigma>0}=\S(T^*M,L(E))$ and the equality is also valid for the topologies.
The following lemma shows that the symbols of $S^{l,m}_\sigma$ are tempered distributional sections on $T^*M$.

\begin{lem}
\label{slmdistr}
The application 
$j_{T^*M}$ is injective and continuous from $S^{l,m}_\sigma$ into $\S'(T^*M,L(E))$.
\end{lem}
\begin{proof} Since we have the following commutative diagram
\[\xymatrix{
    S^{l,m}_\sigma \ar[rr]^{j_{T^*M}} \ar[d]_{T_{z,\bfr,*}} &&  \S'(T^*M,L(E))  \\
    S^{l,m}_\sigma(\R^{2n},L(E_z))  \ar[r]_i  &\O_M(\R^{2n},L(E_z)) \ar[r]_j & \S'(\R^{2n},L(E_z)) \ar[u]_{T^*_{z,\bfr,*}}   \\
  }\]
where $T^*_{z,\bfr,*}$ is the adjoint of $T_{z,\bfr,*}$ on $\S(T^*M,L(E))$ and $\O_M(\R^{2n},L(E_z))$ is the locally convex complete Hausdorff space of $L(E_z)$-valued functions on $\R^{2n}$ with polynomially bounded derivatives, it is sufficient to check that the natural injection $i$ is continuous from $S^{l,m}_\sigma(\R^{2n},L(E_z))$ into $\O_M(\R^{2n},L(E_z))$. This is obtained by the following estimate, for any $\varphi\in \S(\R^{2n})$ and $\nu=(\a,\b)$ $2n$-multi-index,
$$
\sup_{(\rx,\vth)\in \R^{2n}} \norm{ \varphi\, \del^{\nu} a (\rx,\vth)}_{L(E_z)} \leq K_{\varphi,\nu}\, q_{\nu}(a) 
$$
where $K_{\varphi,\nu}:=\sup_{(\rx,\vth)\in \R^{2n}} |\varphi(\rx,\vth)\langle \rx \rangle^{\sigma(l-|\a|)} \langle \vth \rangle^{m-|\b|}|$.
\end{proof}

\begin{defn} Let $(a_j)_{j\in \N^*}$ be a sequence in $S^{l_j,m_j}_\sigma$ where $(l_j)$ and $(m_j)$ are real strictly decreasing sequences such that $\lim_{j\to \infty} l_j= \lim_{j\to \infty} m_j = -\infty$. We say that $a$ is an asymptotic expansion of $(a_j)_{j\in \N^*}$ and we denote
$$
a \sim \sum_{j=1}^\infty a_j
$$
if $a\in C^{\infty}(T^*M,L(E))$ is such that $a - \sum_{j=1}^{k-1} a_j \in S^{l_k,m_k}_{\sigma}$ for any $k\in \N$ with $k\geq 2$. In particular, we have $a\in S^{l_1,m_1}_\sigma$.
\end{defn}

We need asymptotic summation of symbols modulo $S^{-\infty}_\sigma$. The following result of asymptotic completeness is based on a classical method \cite{Shubin} of approximation of series by weightening summands $a_j(\rx,\th)$ with functions which ``cut'' a neighborhood of zero in the domain of $x$ (if $\sigma\neq0$) and $\th$. The idea is that the part we cut is bigger and bigger when $j\to\infty$ so that convergence occurs. 

\begin{lem}
\label{asympt}
Let $(a_j)_{j\in \N^*}$ be a sequence in $S^{l_j,m_j}_\sigma$ where $(l_j)$ and $(m_j)$ are real strictly decreasing sequences such that $\lim_{j\to \infty} l_j= \lim_{j\to \infty} m_j = -\infty$. Then 

\noindent (i) There exists $a\in S^{l_1,m_1}_\sigma$ such that $a \sim \sum_{j=1}^\infty a_j$.

\noindent (ii) If another $a'$ satisfies  $a' \sim \sum_{j=1}^\infty a_j$, then $a-a'\in S^{-\infty}_\sigma$.
\end{lem}
\begin{proof} $(ii)$ is obvious. Let us prove $(i)$ for a sequence $(a_j)_{j\in \N^*}$ in $S^{l_j,m_j}_\sigma(\R^{2n},L(E_z))$ and with $a \sim \sum_{j=1}^\infty a_j \in S^{l_1,m_1}_\sigma(\R^{2n},L(E_z))$. The result will then follows for a sequence $(b_j)$ in $S^{l,m}_\sigma$ by taking $b:=T_{z,\bfr,*}^{-1}(a)$ where $a_j:=T_{z,\bfr,*}(b_j)$. Define
$$
a'_j(\rx,\vth):= \Delta_{p_j}(\rx,\vth)\,a_j(\rx,\vth)
$$
where $\Delta_{p_j}$ is defined in the proof of Lemma \ref{toposymbol} and $(p_j)$ is a real sequence in $[1,+\infty[$. For any $j\in \N$, $a'_j-a_j\in S^{-\infty}_\sigma(\R^{2n},L(E_z))$. Thus, the result will follow if we prove that for a specified sequence $(p_j)$ and for any $N\geq 0$, there exists $k_0(N)\geq 2$ such that for any $k\geq k_0(N)$, 
\begin{equation}
\label{asympteq1}
\sum_{j=k+1}^\infty q_{N,l_k,m_k}(a'_j)<\infty 
\end{equation}
where $q_{N,l_k,m_k}:=\sup_{|\nu|\leq N} q_{\nu,l_k,m_k}$, and $q_{\nu,l_k,m_k}$ are the semi-norms of $S^{l_k,m_k}_\sigma(\R^{2n},L(E_z))$. Indeed, with $\norm{\del^\nu a'_j}_{\infty}\leq q_{|\nu|,l_k,m_k}(a'_j)$ for $k\geq k_1(\nu)$, $a':=\sum_{j=1}^\infty a'_j$ is a well defined smooth function and we have then $a'-\sum_{j=1}^{k-1} a_j \in S^{l_k,m_k}_{\sigma}(\R^{2n},L(E_z))$.
Using Leibniz rule, we see that for any $2n$-multi-index $\nu:=(\a,\b)$, and any $j\in \N^*$, there is $K_{\nu,j}>0$ such that 
\begin{align*}
\tfrac{1}{K_{\nu,j}}\norm{\del^{\nu}a'_j(\rx,\vth)}_{L(E_z)}&\leq \Delta_p(\rx,\vth) \langle \rx \rangle^{\sigma (l_j-|\a|)} \langle \vth \rangle^{m_j-|\b|}\\
&\hspace{2cm}+\sum_{\nu' <\nu}   |\del^{\nu-\nu'}\Delta_p(\rx,\vth)|\langle \rx \rangle^{\sigma(l_j-|\a'|)} \langle \vth \rangle^{m_j-|\b'|}.
\end{align*}
Let us suppose that $\sigma=0$. The estimate (\ref{Deltasigma0}) yields for any $N\geq 0, k\geq 2$, $j\geq k+1$,
$$
q_{N,l_k,m_k}(a'_j)\leq K_{N,j} \langle p_j \rangle^{m_j-m_{j-1}}  
$$
for a constant $K_{N,j}>0$. If we now fix $p_j$ as $p_j = (2^j \sup_{N\leq j} \set{K_{N,j},1})^{1/(m_{j-1}-m_j)}$, then we see that for any $N\geq 0$, $k\geq N+2$, $j\geq k+1$, we have $q_{N,l_k,m_k}(a'_j)\leq 2^{-j}$ and (\ref{asympteq1}) is satisfied. Suppose now $\sigma\neq 0$. The estimate (\ref{Deltasigma1}) yields for any $N\geq 0, k\geq 2$, $j\geq k+1$,
$$
q_{N,l_k,m_k}(a'_j)\leq K'_{N,j} \langle p_j \rangle^{r_j}  
$$
for a constant $K'_{N,j}>0$ and with $r_j:=\max \{m_j-m'_{j-1},\sigma(l_j-l'_{j-1})\}<0$. If we now fix $p_j$ as $p_j = (2^j \sup_{N\leq j} \set{K'_{N,j},1})^{-r_j^{-1}}$, then we see that for any $N\geq 0$, $k\geq N+2$, (\ref{asympteq1}) is satisfied as for the case $\sigma=0$.
\end{proof}

\subsection{Amplitudes and associated operators on $\S(\R^n,E_z)$}

We shall see in this section amplitudes as generalizations of symbols of the type $S^{l,m}_{\sigma,z}:=S^{l,m}_\sigma(\R^{2n},L(E_z))$ where $z\in M$ is fixed. For each amplitude, a continuous operator from $\S(\R^n,E_z)$ into itself will be defined. Here the spaces $L(E_z)$ and $E_z$ can simply be considered as $\M_n(\C)$ and $\C^n$. The results in this section will be important for pseudodifferential operators on $M$ in the next section.

\begin{defn} An amplitude of order $l,w,m$ and type $\sigma\in[0,1]$, $\kappa\geq 0$, is a smooth function $a\in C^\infty(\R^{3n},L(E_z))$ such that for any $3n$-multi-index $\nu=(\a,\b,\ga)$, there exists $C_\nu>0$ such that
\begin{equation}
\label{amplest}
\norm{\del^{(\a,\b,\ga)} a (\rx,\zeta,\vth)}_{L(E_z)} \leq C_\nu\, \langle \rx\rangle^{\sigma(l-|\a+\b|)}\, \langle\zeta\rangle^{w+\kappa|\a+\b|}\, \langle \vth \rangle^{m-|\ga|}
\end{equation}
for any $(\rx,\zeta,\vth)\in \R^{3n}$. We note $\Pi_{\sigma,\kappa,z}^{l,w,m}:=\Pi_{\sigma,\kappa}^{l,w,m}(\R^{3n},L(E_z))$ the space of amplitudes of order $l,w,m$ and type $\sigma,\kappa$.
\end{defn} 

Remark that $\Pi^{l,w,m}_{0,\kappa,z}$ is independent of $l$, we denote this space $\Pi^{0,w,m}_{0,\kappa,z}$. We note $\Pi^{-\infty,w}_{\sigma,\kappa,z}:=\cap_{l,m} \Pi^{l,w,m}_{\sigma,\kappa,z}$. We set $\Pi^\infty_{\sigma,\kappa,z} := \cup_{l,w,m} \Pi^{l,w,m}_{\sigma,\kappa,z}$ and $\Pi_{\sg,z}^{-\infty}:= \cap_{l,m} \cup_{w,\ka} \Pi_{\sg,\ka,z}^{l,w,m}$. 
We see that $\Pi^{l,w,m}_{\sigma,\kappa,z} \cdot \Pi^{l',w',m'}_{\sigma,\kappa,z} \subseteq \Pi^{l+l',w+w',m+m'}_{\sigma,\kappa,z}$ and $\Pi^{l,w,m}_{\sigma,\kappa,z} \subseteq \Pi^{l',w',m'}_{\sigma,\kappa, z}$ for $m\leq m'$, $w\leq w'$, and $l\leq l'$. Thus, $\Pi^{\infty}_{\sigma,\kappa,z}$ is a $*$-algebra, which is trigraduated for $\sigma>0$ and bigraduated for $\sigma=0$. 
Note also that if $a\in \Pi^{l,w,m}_{\sigma,\kappa, z}$, then $\del^{(\a,\b,\ga)}a \in \Pi^{l-|\a+\b|,w+\kappa|\a+\b|,m-|\ga|}_{\sigma,\kappa,z}$. 

Amplitudes and symbols in $S^{l,m}_{\sigma, z}$ are related by the following lemma:

\begin{lem} 
\label{Amplisymb}
(i) For any $a\in \Pi^{l,w,m}_{\sigma,\kappa,z}$ we have $a_{\zeta=0}:=(\rx,\vth)\mapsto a(\rx,0,\vth)$ in $S^{l,m}_{\sigma,z}$. 

\noindent (ii) For any $s\in S^{l,m}_{\sigma,z}$, the function $(\rx,\zeta,\vth)\mapsto s(\rx,\vth)$ is in $\Pi^{l,0,m}_{\sigma,0,z}$.

\noindent (iii) For any $f\in S_{\sigma}(\R^n)$, the function $(\rx,\zeta,\vth)\mapsto f(\rx)\Id_{L(E_z)}$ is in $\Pi_{\sigma,0,z}^{0,0,0}$.
\end{lem}
\begin{proof}
$(i)$ follows from the fact that $\del^\nu (a \circ P) = (\del^{P(\nu)} a) \circ P$ where $P(\rx,\vth):=(\rx,0,\vth)$.

\noindent $(ii)$ Noting $Q(\rx,\zeta,\vth):=(\rx,\vth)$, the result follows from $\del^{\a,\b,\ga} (s\circ Q) = \delta_{\b,0} (\del^{\a,\ga} s)\circ Q$.

\noindent $(iii)$ follows from $(ii)$ and the fact that $(\rx,\vth)\mapsto f(\rx)\Id_{L(E_z)} \in S^{0,0}_{\sigma,z}$. 
\end{proof}

As the spaces of symbols, the $\Pi_{\sigma,\kappa,z}^{l,w,m}$ are naturally Fr\'{e}chet spaces:

\begin{lem}
\label{topoampli}
The following semi-norms on $\Pi_{\sigma,\kappa,z}^{l,w,m}$: 
$$
q_{(\a,\b,\ga)}^{l,w,m}(a):= \sup_{(\rx,\zeta,\vth)\in \R^{3n}} \langle \rx \rangle^{\sigma(|\a+\b|-l)} \langle \zeta\rangle^{-w-\kappa|\a+\b|}\langle \vth \rangle^{|\ga|-m}\norm{\del^{(\a,\b,\ga)} a (\rx,\zeta,\vth)}_{L(E_z)}
$$
determine a Fr\'{e}chet topology on $\Pi_{\sigma,\kappa,z}^{l,w,m}$. The following inclusions are continous for these topologies: $\Pi^{l,w,m}_{\sigma,\kappa,z} \cdot \Pi^{l',w',m'}_{\sigma,\kappa,z} \subseteq \Pi^{l+l',w+w',m+m'}_{\sigma,\kappa,z}$, $\Pi^{l,w,m}_{\sigma,\kappa,z} \subseteq \Pi^{l',w',m'}_{\sigma,\kappa,z}$ ($m\leq m'$, $w\leq w'$ and $l\leq l'$) and $\Pi^{-\infty,w}_{\sigma,\kappa,z}\subseteq \Pi^{l,w,m}_{\sigma,\kappa,z}$. Moreover, the last inclusion is dense when  $\Pi^{l,w,m}_{\sigma,\kappa,z}$ has the topology of $\Pi^{l',w,m'}_{\sigma,\kappa,z}$ for $m< m'$ and $l<l'$.
\end{lem}
\begin{proof} The continuity results are straightforward. For the density result, we prove as in Lemma \ref{toposymbol}, that for any $a\in \Pi^{l,w,m}_{\sigma,\kappa,z}$ the sequence
$$
a_p(\rx,\zeta,\vth):= (\rho(\rx/p))^{1-\delta_{\sigma,0}}\, \rho(\vth/p)\, a(\rx,\zeta,\vth) =:(1-\Delta_{p}(\rx,\vth))\, a(\rx,\zeta,\vth)
$$
converges to $a$ for the topology of $\Pi^{l',w,m'}_{\sigma,\kappa}(\R^{2n},L(E_z))$ where $m'>m$ and $l'>l$. First note that the application $(\rx,\zeta,\vth)\mapsto  (\rho(\rx/p))^{1-\delta_{\sigma,0}}\, \rho(\vth/p)\, \, \Id_{L(E_z)}$ is an amplitude in $\Pi^{-\infty,0}_{\sigma,0,z}$. Thus, $(a_p)_{p\in \N^*}$ is a sequence in  $\Pi_{\sigma,\kappa,z}^{-\infty,w}$. We define the function $R_p$ such that $q_{(\a,\b,\ga)}^{l',w,m'}(a-a_p)=\sup_{(\rx,\zeta,\vth)\in \R^{3n}}R_p(\rx,\zeta,\vth)$, where $m'>m$ and $l'>l$. For a given $3n$-multi-index $\nu:=(\a,\b,\ga)$, we get with Leibniz rule, for a $K>0$,
\begin{align*}
\tfrac{1}{K}\,R_p (\rx,\zeta,\vth)\leq &\ \Delta_p(\rx,\vth)\, \langle \rx \rangle^{\sigma(l-l')}\, \langle \vth \rangle^{m-m'} + \sum_{\nu'< \nu}   |\del^{\nu-\nu'}\Delta_p(\rx,\vth)|\\
&\hspace{1cm}\times\langle \rx \rangle^{\sigma(l-l'+|\a+\b|-|\a'+\b'|)}\langle \zeta \rangle^{\kappa(|\a'+\b'|-|\a+\b|)} \langle \vth \rangle^{m-m'+|\ga|-|\ga'|}\, .
\end{align*}
Suppose that $\sigma=0$. In that case, 
$|\Delta_p(\rx,\vth)|\leq 1_{[p,+\infty[}(\vth)$ and if $\nu'<\nu$, 
\begin{equation*}
|\del^{\nu-\nu'}\Delta_p(\rx,\vth)|\leq \delta_{\a,\a'}\, \delta_{\b,\b'}\,K_{\ga}\, p^{-|\ga|+|\ga'|}\, 1_{[p,2p]}(\vth)\,.
\end{equation*}
As a consequence we find
$R_p (\rx,\zeta,\vth) = \O_{p\to \infty} (\langle p \rangle^{m-m'})$, as in Lemma \ref{toposymbol}. 
Suppose now $\sigma\neq 0$. In that case $|\Delta_p(\rx,\vth)|\leq 1_{F_p}(\rx,\vth)$ where $F_p:=\R^{2n}-B_{n}(0,p)\times B_n(0,p)$ and if $\nu'<\nu$, for a constant $K_\nu>0$
\begin{equation*}
|\del^{\nu-\nu'}\Delta_p(\rx,\vth)|\leq \delta_{\b-\b',0} K_\nu\, 1_{[\sgn(\a-\a')p,2p]}(\rx)\  1_{[\sgn(\ga-\ga')p,2p]}(\vth)\, p^{-|\nu|+|\nu'|}\, .
\end{equation*}
As a consequence, we find
$R_p (\rx,\zeta,\vth) = \O_{p\to \infty} (\langle p \rangle^{r})$ where $r:=\max \{m-m',\sigma(l-l')\}<0$ and the result follows.
\end{proof}

We shall note $\Delta_\zeta$ the differential operator $\sum_{i=1}^n \del_{\zeta_i}^2$. 
The following formula is valid for any $\vth,\zeta\in \R^n$ and $p\in \N$,
\begin{equation}
\label{e2pii}
\langle \vth\rangle^{2p}e^{2\pi i \langle \vth,\zeta\rangle} = (1-(2\pi)^{-2} \Delta_{\zeta})^{p}\,e^{2\pi i \langle \vth,\zeta\rangle}=: L_{\zeta}^p \,e^{2\pi i \langle \vth,\zeta\rangle}\, .
\end{equation} 
A computation shows that $(1-(2\pi)^{-2} \Delta_{\zeta})^{p} = \sum_{0\leq |\b|\leq p} c_{p,\b}\,\del^{2\b}_\zeta$, where the summation is on $n$-multi-indices $\b$ and $c_{p,\b}:= \tbinom{p}{|\b|}(-1)^{|\b|}(2\pi)^{-2|\b|}\b!$. 
We shall also use the following useful formula valid for any $\vth\in \R^n$, $\zeta\in \R^{n}\backslash\set{0}$ and $p\in \N$,
\begin{equation}
\label{Mformula}
e^{2\pi  i \langle \vth,\zeta\rangle} = \sum_{|\b|=p} \la_\b\,\tfrac{\zeta^\b}{\norm{\zeta}^{2p}}\, \del^{\b}_\vth \,e^{2\pi  i \langle \vth,\zeta\rangle} =: M_{\vth}^{p,\zeta} \,e^{2\pi  i \langle \vth,\zeta\rangle} 
\end{equation}  
where $\la_\b :=\b!(2\pi)^{-|\b|}i^{|\b|}$. We define $^t M_{\vth}^{p,\zeta}:= \sum_{|\b|=p} \la_\b (-1)^p \tfrac{\zeta^\b}{\norm{\zeta}^{2p}} \del^{\b}_\vth$. 

\begin{defn}
\label{OFZ}
We note $\O_{f,z}$, where $f_1,f_2,f_3:\N^{3n}\to \R$, and $f:=(f_1,f_2,f_3)$, the space of smooth functions in $C^{\infty}(\R^{3n},L(E_z))$ such that for any $3n$-multi-index $\nu=(\a,\b,\ga)$, there is $C_\nu>0$ such that 
$$
\norm{\del^\nu a (\rx,\zeta,\vth)}_{L(E_z)} \leq C_\nu \langle \rx \rangle^{f_1(\nu)}\langle \zeta\rangle ^{f_2(\nu)} \langle \vth\rangle ^{f_3(\nu)}
$$
uniformly in $(\rx,\zeta,\vth)\in \R^{3n}$. 
\end{defn}

The vector space $\O_{f,z}$ has a natural family of seminorms $q_{\nu}^{f}$ given by the best constants $C_\nu$ in the previous estimate. With this family, $\O_{f,z}$ is a Fr\'{e}chet space. Obviously, amplitudes in $\Pi_{\sg,\ka,z}^{l,w,m}$ form an $\O_{f,z}$ space where $f_1(\nu):=\sg(l-|\a+\b|)$, $f_2(\nu):=w+\kappa|\a+\b|$ and $f_3(\nu):=m-|\ga|$. For a given triple $f:=(f_1,f_2,f_3)$ and $\rho\in \R$, we will note $f_{3,\rho,\a,\ga}:=\sup_{\b} f_3(\a,\b,\ga)-\rho|\b|$, $f_{2,\rho,\a,\b}:=\sup_{\ga} f_2(\a,\b,\ga)-\rho|\ga|$ and $f_{1,\rho,\a,\b}:=\sup_{\ga} f_1(\a,\b,\ga)-\rho|\ga|$.

\begin{prop}
\label{ampliOP} Let $\Ga$ a continuous linear operator on the space $\S(\R^{2n},L(E_z))$, and $f:=(f_1,f_2,f_3)$ a triple such that there exists $\rho<1$ such that $f_{3,\rho,0,0}<\infty$.

\noindent (i) For any function $a\in \O_{f,z}$ the following antilinear form on $\S(\R^{2n},L(E_z))$ 
$$
\langle \Op_{\Ga}(a) , u\rangle :=  \int_{\R^{3n}} e^{2\pi i \langle \vth,\zeta \rangle } \Tr(a(\rx,\zeta,\vth)\, \Ga(u)^*(\rx,\zeta)) \, d\zeta\,d\vth\,  d\rx
$$
is in $\S'(\R^{2n},L(E_z))$.

\noindent (ii) For any given $u\in \S(\R^{2n},L(E_z))$, the linear form $L_{u,\Ga}:= a\mapsto \langle \Op_{\Ga}(a) , u\rangle$ is continuous on $\O_{f,z}$. In particular $L_{u,\Ga}$ is continuous on any amplitude space $\Pi_{\sigma,\kappa,z}^{l,w,m}$. 
\end{prop}
\begin{proof} $(i)$ We have $\Op_{\Ga}(a) = I(a) \circ \Ga$, where $I(a)$ is the antilinear form on  
$\S(\R^{2n},L(E_z))$:
$$
\langle I(a) , u\rangle :=  \int_{\R^{3n}} e^{2\pi i \langle \vth,\zeta \rangle } \Tr(a(\rx,\zeta,\vth)\, u^*(\rx,\zeta)) \,  d\zeta\, d\vth\, d\rx\,.
$$
We shall prove that $I(a)\in \S'(\R^{2n},L(E_z))$, which will give the result. Let $u\in \S(\R^{2n},L(E_z))$ and let us fix for now $\rx$ and $\vth \in \R^n$. We can check that the map $\zeta\mapsto a(\rx,\zeta,\vth)\, u^*(\rx,\zeta)$ is in $\S(\R^n, L(E_z))$. As a consequence, with (\ref{e2pii}) 
and integration by parts, we get with $R(\rx,\vth):= \int_{\R^n} e^{2\pi i \langle \vth,\zeta \rangle } a(\rx,\zeta,\vth)\,u^*(\rx,\zeta) \,d\zeta$,
\begin{align*}
 R(\rx,\vth)&= \int_{\R^n} e^{2\pi i \langle \vth,\zeta \rangle }\langle \vth\rangle^{-2p}  (1-(2\pi)^{-2} \Delta_{\zeta})^{p}\, a(\rx,\zeta,\vth)\, u^*(\rx,\zeta)\, d\zeta \\
 &= \sum_{0\leq |\beta| \leq p}\,\sum_{\b'\leq 2\b} c_{p,\b}\,\tbinom{2\b}{\b'}\langle \vth\rangle^{-2p}\int_{\R^n} e^{2\pi i \langle \vth,\zeta \rangle }(\del^{(0,\b',0)}\, a(\rx,\zeta,\vth))\,(\del^{(0,2\b-\b')} u^*(\rx,\zeta))\, d\zeta\, .
\end{align*}
Thus, for any $\rx,\vth\in \R^n$, we get by fixing $p$ such that $2(\rho-1)p+f_{3,\rho,0,0}\leq -2n$ (this is possible since $\rho<1$) that for any $N\in \N$,
$$
\norm{R(\rx,\vth)}_{L(E_z)} \leq C_p \langle \vth\rangle^{-2n} \int_{\R^{n}}\langle \rx,\zeta \rangle^{-N+r_p}\, d\zeta \sum_{0\leq |\beta|\leq p}\sum_{\b'\leq 2\b} q_{0,\b',0}^{f}(a)\, q_{N,(0,2\b-\b')}(u)\,   
$$
for a $C_p>0$, where $r_p:=\max_{|\b'|\leq 2p} |f_1(0,\b',0)| + |f_2(0,\b',0)|$.  If we now fix $N$ such that $-N+r_p\leq -4n$, we see, using the inequality $\langle \rx,\zeta \rangle^{-2} \leq \langle \rx \rangle^{-1} \langle \zeta\rangle^{-1}$, that there is $C_{\rho,f}>0$ such that 
\begin{equation}
\label{Iaueq}
|\langle I(a), u\rangle |\leq C_{\rho,f} \sum_{0\leq |\beta|\leq p}\sum_{\b'\leq 2\b}  q_{0,\b',0}^{f}(a)\,q_{N,(0,2\b-\b')}(u)
\end{equation}
which yields the result. 

\noindent $(ii)$ The continuity of $L_{u,\Ga}$ on $\O_{f,z}$ follows directly from (\ref{Iaueq}) since $L_{u,\Ga}(a)=\langle I(a),\Ga(u)\rangle$. Since $\Pi^{l,w,m}_{\sigma,\kappa,z}= \O_{f,z}$ for a triple $f=(f_1,f_2,f_3)$ such that $f_{3,0,0,0}<\infty$, $L_{u,\Ga}$ is continous on any amplitude space.
\end{proof}

For any amplitude $a$, we will also note $\Op_{\Ga}(a)$ the continous linear map from $\S(\R^{n},E_z)$ into $\S'(\R^n,E_z)$, associated to the tempered distribution $u\mapsto \langle \Op_{\Ga}(a),u\rangle$. 

\begin{rem} If $(M,\exp,E,d\mu,\psi)$ has a $\O_M$-bounded geometry, we saw that for any frame $(z,\bfr)$ and $\la\in [0,1]$, the $\Ga_{\la,z,\bfr}$ maps are topological isomorphisms on $\S'(\R^{2n},L(E_z))$. Thus, Lemma \ref{ampliOP} implies that for a given $a\in \Pi_{\sigma,\kappa,z}^{l,w,m}$, we can define a family indexed by $\la\in [0,1]$ of operators $\Op_{\Ga_{\la,z,\bfr}}(a)$ which are continous from  $\S(\R^{n},E_z)$ into $\S'(\R^n,E_z)$.
\end{rem}

\begin{rem}
\label{Oplien} Suppose that $(M,\exp,E,d\mu)$ has a $\S_{\sigma}$ bounded geometry and that $\psi$ is a $\O_M$-linearization. We deduce from (\ref{eq-oplazb}) that if $s$ is a symbol in $S^{l,m}_\sigma$ and $\la\in [0,1]$, we have $(\Op_\la(s))_{z,\bfr}=\Op_{\Ga_{\la,z,\bfr}}(\mu s_{z,\bfr})$ where $(z,\bfr)$ is a frame, $s_{z,\bfr}:=T_{z,\bfr,*}(s)$ and $\mu s_{z,\bfr}:=(\rx,\zeta,\vth)\mapsto \mu_{z,\bfr}(\rx)\, s_{z,\bfr}(\rx,\vth) \in \Pi^{l,0,m}_{\sigma,0,z}$. We will also note $\mu^{-1} s_{z,\bfr}(\rx,\zeta,\vth):= \mu_{z,\bfr}^{-1}(\rx) s_{z,\bfr}(\rx,\vth) \in \Pi^{l,0,m}_{\sigma,0,z}$.
\end{rem}

We now establish a sufficient condition on $\Ga$ and $a$ in order to have $\Op_\Ga(a)$ stable (and continuous) on $\S(\R^n,E_z)$.
The result will be used to establish regularity of pseudodifferential operators.

\begin{lem}
\label{amplContinu} Let $\Ga$ be a continuous linear operator on $\S(\R^{2n},L(E_z))$ of the form $\Ga= L_{\tau_1}\circ R_{\tau_2}\circ C_{\Phi}$, where $\tau_i\in \O_M(\R^{2n},L(E_z))$ (for $1\leq i\leq 2$), and $\Phi:=(\pi_1,\psi)\in C^{\infty}(\R^{2n},\R^{2n})$ is such that $\psi\in \O_M(\R^{2n},\R^n)$ and there exist $c,\eps,r >0$, such that for any $(\rx,\zeta)\in \R^{2n}$, $\langle \psi(\rx,\zeta)\rangle\geq c \langle \rx \rangle^{\eps} \langle \zeta \rangle^{-r}$ and for any $\rx\in \R^n$, there is $c_\rx >0$ such that $\langle\psi(\rx,\zeta)\rangle\geq c_\rx \langle \zeta \rangle^{\eps}$ uniformly in $\zeta\in \R^n$.

Suppose that $f=(f_1,f_2,f_3)$ is such that there exist $(\rho_1,\rho_2,\rho_3)\in \R^3$ such that $\rho_3<1$, $(r/\eps)\rho_1+\rho_2<1$ and for any $2n$-multi-index $\mu$, $f_{1,\rho_1,\mu}<\infty$, $f_{2,\rho_2,\mu}<\infty$, $f_{3,\rho_3,\mu}<\infty$ and for any $n$-multi-index $\a$ $f_{3,\rho_3,\a}:=\sup_{\ga}f_{3,\rho_3,\a,\ga}<\infty$. Then for any function $a\in \O_{f,z}$, the operator $\mathfrak{Op}_{\Ga}(a)$ is continuous from $\S(\R^{n},E_z)$ into itself. In particular, this is the case for any amplitude $a\in \Pi^{l,w,m}_{\sigma,\kappa,z}$.
\end{lem}
\begin{proof} Let $u,v\in \S(\R^{n},E_z)$. By definition, $\langle\mathfrak{Op}_{\Ga}(a) (v) ,u\rangle = \Op_{\Ga}(a)(u\otimes \ol v)$ and $\Ga(K)= \tau_1\, (K\circ \Phi)\, \tau_2$. Noting $a'(\rx,\zeta,\vth):=\tau_1^*(\rx,\zeta)\,a(\rx,\zeta,\vth)\,\tau_2^*(\rx,\zeta)$, we obtain
\begin{align*}
\langle\mathfrak{Op}_{\Ga}(a) (v) ,u\rangle&:=  \int_{\R^{3n}} e^{2\pi i \langle \vth,\zeta \rangle } \big(\,a'(\rx,\zeta,\vth)\, \,v(\psi(\rx,\zeta))\, \big| \,u(\rx)\,\big) \,  d\zeta\,d\vth\,d\rx \, \\
& =  \int_{\R^{n}} \big(\,g(\rx) \big| \,u(\rx)\,\big) \,  d\rx\ 
\end{align*}
where $g(\rx):=\int_{\R^{2n}} e^{2\pi i \langle \vth,\zeta \rangle } \,a'(\rx,\zeta,\vth)\, \,v\circ \psi(\rx,\zeta)\,d\zeta\,d\vth$.

A computation with the Faa di Bruno formula shows that for any $2n$-multi-index $\nu$, any $N\in \N$ and any $\rx\in \R^n$ there is $C_{\rx,N,\nu}>0$ such that $\norm{\del^\nu (v\circ \psi) (\rx,\zeta)}_{E_z}\leq C_{\rx,N,\nu}\langle \zeta\rangle^{-N}$ uniformly in $\zeta\in \R^n$. As a consequence, the map $\zeta \mapsto \del^{\a',0}a'(\rx,\zeta,\vth)\,\del^{\a-\a'}(v\circ \psi)(\rx,\zeta)$ is in $\S(\R^n,E_z)$. We can thus successively integrate by parts in $g(\rx)$ so that for any $p\in \N^*$,
\begin{align*}
&g(\rx) =\int_{\R^{2n}}e^{2\pi i \langle \vth,\zeta \rangle } \langle\vth \rangle^{-2p}L_{\zeta}^p (a'(v\circ \psi))(\rx,\zeta,\vth)\,d\zeta\,d\vth \, .
\end{align*}
By taking $p$ such that $(\rho_3-1)2p +c_0 \leq -2n$ where $c_\a:=\sup_{\a'\leq \a} f_{3,\rho_3,\a'}$, we see that the previous integrand is absolutely integrable, and we can permute the order of integrations $d\zeta d\vth\to d\vth d\zeta$. Since all the successive $\vth$-derivatives of $\langle\vth \rangle^{-2p}L_{\zeta}^p (a'(v\circ \psi))(\rx,\zeta,\vth)$ converge to 0 when $\langle \vth\rangle$ goes to infinity, we can then integrate by parts in $\vth$ so that for any $q\in \N$ and $p\geq p_0$
$$
g(\rx) =\int_{\R^{2n}}e^{2\pi i \langle \vth,\zeta \rangle }\langle \zeta\rangle^{-2q} L_{\vth}^q(\langle\vth \rangle^{-2p}L_{\zeta}^p (a'(v\circ \psi)))(\rx,\zeta,\vth)\,d\zeta\,d\vth \, .
$$
Noting $h_{p,q}$ the previous integrand, we see that for any $n$-multi-index $\a$,  $\del^\a h_{p,q}$ is a linear combination of terms of the form
$$
e^{2\pi i\langle \vth,\zeta\rangle} \langle \zeta\rangle^{-2q} \langle \vth\rangle^{-2p-|\ga-\ga'|} \del^{\a',\b',\ga'} a' \del^{\a-\a',\b-\b'} v\circ \psi
$$
where $|\ga|\leq 2p$, $\ga'\leq \ga$, $|\b|\leq 2q$, $\b'\leq \b$ and $\a'\leq \a$. A computation with the Faa di Bruno formula shows that for any $2n$-multi-index $\nu$ there is $r_{\nu}\in\N^*$ such that for any $N>0$, there is $C_{\nu,N}>0$ such that for any $w\in \S(\R^n,E_z)$ and any $(\rx,\zeta)\in \R^{2n}$, $\norm{\del^\nu (w\circ \psi)(\rx,\zeta)}_{E_z}\leq C_{\nu,N}\langle\rx,\zeta\rangle^{r_\nu-N}\langle \zeta\rangle^{r_\nu+(r/\eps)N}\sum_{|\nu'|\leq |\nu|} q_{[N/\eps]+1,\nu'}(w)$. Moreover, we check that there is $K_{\a,p}>0$ such that 
\begin{equation*}
\norm{\del^{(\a',\b',\ga')} a'(\rx,\zeta,\vth)}_{L(E_z)} \leq C_{\a,p,q} \langle \rx\rangle^{K_{\a,p}+ \rho_1 2q}\langle \zeta\rangle^{K_{\a,p}+\rho_2 2q} \langle \vth\rangle^{c_\a +\rho_3 2p} \, .
\end{equation*}
As a consequence, we get the estimate
$$
\norm{\del^\a h_{p,q}} \leq C_{\a,p,q,N} \langle\rx \rangle^{K'_{\a,p}+\rho_1 2q -N} \langle \zeta\rangle^{K'_{\a,p}+(\rho_2-1)2q+(r/\eps)N} \langle \vth\rangle^{c_\a+(\rho_3-1)2p} \sum_{|\nu'|\leq |\nu|} q_{[N/\eps]+1,\nu'}(v)\, .
$$
or equivalently, replacing $K'_{\a,p}+\rho_1 2q -N$ by $-N$, 
\begin{align*}
&\norm{\del^\a h_{p,q}} \leq C_{\a,p,q,N} \langle\rx \rangle^{-N} \langle \zeta\rangle^{K''_{\a,p}+(\rho_2-1+(r/\eps)\rho_1)2q+(r/\eps)N} \langle \vth\rangle^{c_\a+(\rho_3-1)2p}\\ 
&\hspace{3cm} \sum_{|\nu'|\leq |\nu|} q_{[N+K'_{\a,p}+\rho_1 2q/\eps]+1,\nu'}(v)\, .
\end{align*}
Fixing now, for a given $N$, $p$ such that $(\rho_3-1)2p+ c_\a \leq -2n$ and $q$ such that $K''_{\a,p}+(\rho_2-1+(r/\eps)\rho_1)2q+(r/\eps)N \leq -2n$, we obtain the result.  
\end{proof}

The following lemma gives a characterization of smoothing kernels in the cases $\sigma=0$ and $\sigma\neq 0$. If $s$ is in a space of symbols and $\Ga$ is a continuous linear map on $\S(\R^{2n},L(E_z))$, we will note $\Op_\Ga(s):=\Op_\Ga((\rx,\zeta,\vth)\mapsto s(\rx,\vth))$. We shall use the Fr\'{e}chet space $\O_{\sg,f,z}^{l,m}$ of smooth functions $a$ in $C^{\infty}(\R^{3n},L(E_z))$ such that for any $\nu:=(\mu,\ga)\in \N^{2n}\times \N^n$ 
$$
\norm{\del^\nu a (\rx,\zeta,\vth)}_{L(E_z)}\leq C_\nu \langle \rx\rangle^{\sg(l+f_1(\mu))} \langle \zeta\rangle^{f_2(\nu)} \langle \vth\rangle^{m+f_3(\mu)}\, .
$$
We will note $\O_{0,f,z}^{l,m}=:\O_{f_2,f_3,z}^{m}$.
Clearly, $\Op_{\Ga}(a)$ (see Lemma \ref{ampliOP}) is defined as an antilinear form on $\S(\R^{2n},L(E_z))$ whenever $a\in \O_{f,z}^{l,m}$ with $m+f_3(0)< -n$.
We note $F$ the set of functions $f_2:\N^{3n}\to \R$ such that there is $\rho<1$ such that for any $(\a,\b)\in \N^{2n}$ $f_{2,\rho,\a,\b}:=\sup_{\ga} f_2(\a,\b,\ga)-\rho|\ga|<\infty$.

\begin{lem}
\label{noyauReste}
Let $K\in \S'(\R^{2n},L(E_z))$, and $\Ga$ a topological isomorphim on $\S(\R^{2n},L(E_z))$ of the form $\Ga=L_{\tau_1}\circ R_{\tau_2}\circ C_\Phi$ with $\tau_1,\tau_2\in \O_M^\times(\R^{2n},GL(E_z))$, $\Phi\in\O^\times_M(\R^{2n},\R^{2n})$. Then 

\noindent (i) Case $\sigma=0$. The following are equivalent:
 
(i-1) There is $f_3:\N^{2n}\to \R$ such that for any $m\leq-f_3(0) -2n$, there exist $f_{2,m}\in F$, $a_m \in \O_{f_{2,m},f_3,z}^{m}$ such that $K=\Op_{\Ga}(a_m)$.

(i-2) $K \in C^\infty(\R^{2n},L(E_z))$ and for any $2n$-multi-index $\nu$, $N\in \N$, there is $C_{\nu,N}>0$ such that for any $(\rx,\zeta)\in \R^{2n}$, $\norm{\del^{\nu} K_{\Ga} (\rx,\zeta)}_{L(E_z)} \leq C_{\nu,N} \langle \zeta \rangle^{-N}$, where $K_{\Ga}:=K\circ \Ga=\wt\tau_1\,K\circ \Phi\,\wt \tau_2\,|J(\Phi)|$. 

(i-3) There is $s\in S^{-\infty}_{0,z}$ such that $K=\Op_{\Ga}(s)$.

\noindent (ii) Case $\sigma>0$. The following are equivalent:

(ii-1) There is $f_1,f_3:\N^{2n}\to \R$ such that for any  $m\leq -f_3(0)-2n$, there exist $f_{2,m}\in F$ and $a_{m} \in \O_{\sg,f_1,f_{2,m},f_3,z}^{m,m}$ such that $K=\Op_\Ga(a_{m})$.

(ii-2) $K\in \S(\R^{2n},L(E_z))$.

(ii-3) There is $s\in S^{-\infty}_z$ such that $K=\Op_{\Ga}(s)$.
\end{lem}
\begin{proof}

$(i)$ The implication \emph{(i-3)} $\Rightarrow$ \emph{(i-1)} is trivial. We will prove \emph{(i-1)} $\Rightarrow$ \emph{(i-2)} $\Rightarrow$ \emph{(i-3)}. Suppose \emph{(i-1)}. Thus, for any $m\leq -2n-f_3(0)$, there is $f_{2,m}\in F$, $a_m \in \O_{f_{2,m},f_3,z}^{m}$ such that for any $u\in \S(\R^{2n},L(E_z))$, 
$$
\langle K\circ \Ga^{-1},u\rangle = \int_{\R^{3n}} e^{2\pi i \langle\vth ,\zeta\rangle} \Tr\big(a_m(\rx,\zeta,\vth)\, u^*(\rx,\zeta) \big) \, d\zeta\,d\vth\,d\rx\, .
$$
Since $m\leq -2n-f_3(0)$, the preceding integral is absolutely convergent and we can permute the order of integration. As a consequence, we get 
$\langle K\circ \Ga^{-1},u\rangle = \int_{\R^{2n}} \Tr\big(U_m(\rx,\zeta)\, u^*(\rx,\zeta) \big) \, d\zeta\,d\rx$ where $U_m(\rx,\zeta):= \int_{\R^{n}} e^{2\pi i \langle \vth,\zeta \rangle}\,a_m(\rx,\zeta,\vth)\,d\vth$, we check easily that $U_m$ is a continous function on $\R^{2n}$, so we deduce that $U_m=:U$ is independent of $m$ and $K\circ \Ga^{-1}$ is a distribution which is continous function equal to $U$. Noting $b_m:= e^{2\pi i \langle \vth,\zeta \rangle}\,a_m(\rx,\zeta,\vth)$ we see that for any $2n$-multi-index $\mu:=(\a,\b)$, $\del^\mu_{\rx,\zeta} b_m = e^{2\pi i\langle \vth,\zeta\rangle}\sum_{\b'\leq \b}\tbinom{\b}{\b'} (2\pi i \vth)^{\b-\b'} \del^{\a,\b',0}a_m$ and we have then the estimates
$$
\norm{\del^{\mu} b_m } \leq C_{\mu,m}  \langle \zeta\rangle^{\sup_{\b'\leq \b} f_{2,m}(\a,\b',0)} \langle \vth\rangle^{m+c_\mu}
$$
where $c_\mu=\sup_{\b'\leq \b} f_{3}(\a,\b')+|\b|$. Defining $m_{\mu}:= -2n-\sup_{|\mu'|\leq |\mu|} c_{\mu'}$, we see that $U$ is smooth and 
$$
\del^\mu U = \int_{\R^{2n}} \del^\mu b_{m_\mu} d\vth = \sum_{\b'\leq \b}\tbinom{\b}{\b'}(2\pi i)^{|\b-\b'|} \int_{\R^{n}} e^{2\pi i\langle \vth,\zeta\rangle} \vth^{\b-\b'} \del^{\a,\b',0} a_{m_{\mu}}(\rx,\zeta,\vth) \, d\vth\, .
$$
All the $\vth$-derivatives of $\vth\mapsto \vth^{\b-\b'} \del^{\a,\b',0} a_{m_{\mu}}(\rx,\zeta,\vth)$ converge to zero when $\norm{\vth}\to \infty$ so we can we integrate by parts in $\vth$ so that for any $p\in \N$:
$$
\del^\mu U= \sum_{\b'\leq \b}\tbinom{\b}{\b'}(2\pi i)^{|\b-\b'|}  \int_{\R^{n}} e^{2\pi i\langle \vth,\zeta\rangle} \langle \zeta\rangle^{-2p} L_{\vth}^p\big( \vth^{\b-\b'} \del^{\a,\b',0} a_{m_{\mu}}\big)(\rx,\zeta,\vth) \, d\vth\, .
$$
Since $a_{m_{\mu}} \in \O_{f_{2,m_{\mu}},f_3,z}^{m_\mu}$ and $f_{2,m_\mu,\rho_\mu,\la}<\infty$ for a $\rho_\mu<1$, we see that the integrand $h_{p}$ of the previous integral satisfies the estimate
$$
\norm{h_{p}(\rx,\zeta,\vth)} \leq C_{p,\mu} \langle \zeta\rangle^{-2p+\sup_{\b'\leq \b} f_{2,m_{\mu},\rho_\mu,\a,\b'}+2p\rho_\mu} \langle \vth \rangle^{-2n}\, .
$$
Given $N>0$ and fixing $p$ such that $(\rho_\mu-1)2p+\sup_{\b'\leq \b} f_{2,m_{\mu},\rho_\mu,\a,\b'}\leq -N$, we finally obtain that $K\circ \Ga^{-1}=U$ is smooth and satisfies for any $\mu \in \N^{2n}$ and $N>0$, $\norm{\del^\mu K\circ \Ga^{-1} (\rx,\zeta)}_{L(E_z)}\leq C_{\mu,N} \langle \zeta\rangle^{-N}$. We also have for any $u\in \S(\R^{2n},L(E_z))$, $\langle K,u\rangle = \langle U,\Ga(u)\rangle = \int_{\R^{2n}} \Tr(U'(\rx,\zeta) u^* \circ \Phi(\rx,\zeta))d\rx\,d\zeta$ where $U'(\rx,\zeta):=\tau_1^*(\rx,\zeta)U(\rx,\zeta) \tau_2^*(\rx,\zeta)$. 
Using the change of variables provided by the diffeomorphism $\Phi$, we get $\langle K,u\rangle= \int_{\R^{2n}} \Tr(K(\rx,\ry)\,u^*(\rx,\ry))\,d\rx\,d\ry$ where $K(\rx,\ry):= (|J(\Phi^{-1})|(\rx,\ry)) U'\circ \Phi^{-1}(\rx,\ry)$. The result follows. 

\noindent Suppose now \emph{(i-2)}. It is not difficult to see that $\F_P$ sends $S^{-\infty}_{0,z}$ (seen as a subspace of $\S'(\R^{2n},L(E_z))$) into $S^{-\infty}_{0,z}$. In particular, we have $s:=\F_P(K_\Ga) \in S^{-\infty}_{0,z}$. A computation shows that 
$\langle K,u\rangle = \langle \Op_\Ga(s),u\rangle$ for any $u\in \S(\R^{2n},L(E_z))$.

\noindent $(ii)$ Suppose \emph{(i-1)}. Following the proof of $(i)$, we see that it is sufficient to prove that $U$ is in $\S(\R^{2n},L(E_z))$, where $U(\rx,\zeta):=\int_{\R^n} e^{2\pi i \langle \vth,\zeta\rangle} a_{m}(\rx,\zeta,\vth)\, d\vth$ (independent of $m$). Let us fix $N>0$.
For any $2n$-multi-index $\mu:=(\a,\b)$, $\del^\mu_{\rx,\zeta} b_m = e^{2\pi i\langle \vth,\zeta\rangle}\sum_{\b'\leq \b}\tbinom{\b}{\b'} (2\pi i \vth)^{\b-\b'} \del^{\a,\b',0}a_m$ and we have the estimates
$$
\norm{\del^{\mu} b_m } \leq C_{\mu,m} \langle\rx\rangle^{\sg m+\sg d_\mu} \langle \zeta\rangle^{\sup_{\b'\leq \b} f_{2,m}(\a,\b',0)} \langle \vth\rangle^{m+c_\mu}
$$
where $c_\mu=\sup_{\b'\leq \b} f_{3}(\a,\b')+|\b|$ and $d_\mu:=\sup_{\b'\leq \b}f_1(\a,\b')$. Defining 
$$
m_{\mu,N}:= \min\{-2n-\sup_{|\mu'|\leq |\mu|} c_{\mu'},-N/\sg-\sup_{|\mu'|\leq |\mu|}d_{\mu'}\}
$$ 
we see that $U$ is smooth and 
$$
\del^\mu U = \int_{\R^{2n}} \del^\mu b_{m_{\mu,N}} d\vth = \sum_{\b'\leq \b}\tbinom{\b}{\b'}(2\pi i)^{|\b-\b'|} \int_{\R^{n}} e^{2\pi i\langle \vth,\zeta\rangle} \vth^{\b-\b'} \del^{\a,\b',0} a_{m_{\mu,N}}(\rx,\zeta,\vth) \, d\vth\, .
$$
All the $\vth$-derivatives of $\vth\mapsto \vth^{\b-\b'} \del^{\a,\b',0} a_{m_{\mu,N}}(\rx,\zeta,\vth)$ converge to zero when $\norm{\vth}\to \infty$ so we can we integrate by parts in $\vth$ so that for any $p\in \N$:
$$
\del^\mu U= \sum_{\b'\leq \b}\tbinom{\b}{\b'}(2\pi i)^{|\b-\b'|}  \int_{\R^{n}} e^{2\pi i\langle \vth,\zeta\rangle} \langle \zeta\rangle^{-2p} L_{\vth}^p\big( \vth^{\b-\b'} \del^{\a,\b',0} a_{m_{\mu,N}}\big)(\rx,\zeta,\vth) \, d\vth\, .
$$
Since $a_{m_{\mu,N}} \in \O_{\sg,f_1,f_{2,m_{\mu,N}},f_3,z}^{m_{\mu,N},m_{\mu,N}}$ and $f_{2,m_{\mu,N},\rho_{\mu,N},\la}<\infty$ for a $\rho_{\mu,N}<1$, we see that the integrand $h_{p}$ of the previous integral satisfies the estimate
$$
\norm{h_{p}(\rx,\zeta,\vth)} \leq C_{p,\mu,N} \langle \rx\rangle^{-N} \langle \zeta\rangle^{-2p+\sup_{\b'\leq \b} f_{2,m_{\mu,N},\rho_{\mu,N},\a,\b'}+2p\rho_{\mu,N}} \langle \vth \rangle^{-2n}\, .
$$
Fixing $p$ such that $(\rho_{\mu,N}-1)2p+\sup_{\b'\leq \b} f_{2,m_{\mu,N},\rho_{\mu,N},\a,\b'}\leq -N$, we finally obtain the following estimate $\norm{\del^\mu U}_{L(E_z)}\leq C_{\mu,N} \langle \rx\rangle^{-N}\langle \zeta\rangle^{-N}$, which yields \emph{(i-2)}. The other implications are straightforward.
\end{proof}

\begin{cly}
\label{correste} Same hypothesis. We have (for $\sg=0$ or $\sg> 0$), $\Op_{\Ga}(S^{-\infty}_{\sigma,z})=\cap_{l,m}\cup_{w,\ka} \Op_{\Ga}(\Pi_{\sg,\ka,z}^{l,w,m})=\Op_{\Ga}(\Pi_{\sg,z}^{-\infty})$.
\end{cly}

\begin{lem}
\label{amplireste}
Let $u \in S(\R^{2n},L(E_z))$ and $\b$ a $n$-multi-index.

\noindent (i) For any triple $f:=(f_1,f_2,f_3)$ such that there exists $\rho<1$ such that for any $2n$-multi-index $(\a,\ga)$, $f_{3,\rho,\a,\ga}<\infty$, 
the following linear forms are continuous on $\O_{f,z}$
\begin{align*}
&R_{\b,u} : a\mapsto \int_{\R^{3n}} \zeta^\b e^{2\pi i \langle \vth,\zeta\rangle} \Tr(a(\rx,\zeta,\vth)\,u(\rx,\zeta))\,d\zeta\,d\vth\,d\rx\, ,\\
&S_{\b,u} : a\mapsto  (i/2\pi)^{|\b|}\int_{\R^{3n}} e^{2\pi i\langle \vth,\zeta\rangle} \Tr(\del_{\vth}^\b a(\rx,\zeta,\vth)\,u(\rx,\zeta))\,d\zeta\,d\vth\,d\rx\, .
\end{align*}
\noindent (ii) $R_{\b,u}=S_{\b,u}$ on any $\Pi^{l,w,m}_{\sigma,\kappa,z}$ space.
\end{lem}

\begin{proof} $(i)$ The continuity of $R_{\b,u}$ is a direct consequence of Proposition \ref{ampliOP} since $R_{\b,u}=L_{u_\b,\Id}$ where $u_\b(\rx,\zeta):=\zeta^\b u(\rx,\zeta)$. Suppose that $\nu_0$ is a $3n$-multi-index, we denote $f^{\nu_0}:=\nu\mapsto f(\nu+\nu_0)$. A computation shows for any $\rho$, and $n$-multi-indices $\a,\ga$, $f_{3,\rho,\a,\ga}^{\nu_0} \leq f_{3,\rho,\a+\a_0,\ga+\ga_0} + \rho|\b_0|$. Thus if there is $\rho<1$ such that for any $2n$-multi-index $(\a,\ga)$, $f_{3,\rho,\a,\ga}<\infty$, then for any $2n$-multi-index $(\a,\ga)$, $f^{\nu_0}_{3,\rho,\a,\ga}<\infty$. If $a \in \O_{f,z}$ then $\del^{\nu_0} a \in \O_{f^{\nu_0},z}$ and the linear map $a\mapsto \del^{\nu_0} a$ is continuous.  As a consequence, since $S_{\b,u}=L_{u,\Id}\circ D_\b$, where $D_\b:=(i/2\pi)^{\b}\del^\b_{\vth}$, the continuity of $S_{\b,u}$ on $\O_{f,z}$ follows from Proposition \ref{ampliOP}.

\noindent $(ii)$ The equality is easily obtained on $\Pi_{\sigma,\kappa,z}^{-\infty,w}$ by an integration by parts in $\vth$ and permutations of the order of integration $d\zeta d\vth \to d\vth d\zeta$ in $R_{\b,u}(a)$ (authorized for $a\in \Pi_{\sigma,\kappa,z}^{-\infty,w}$). The result now follows from $(i)$ and the density result of Lemma \ref{topoampli}.
\end{proof}

If $N\geq 1$ and $\b,\ga,$ $n$-multi-indices, we denote for any amplitude $a\in \Pi^{l,w,m}_{\sigma,\kappa,z}$, the smooth function $a_{\b,\ga,N}$ as $a_{\b,\ga,N}(\rx,\zeta,\vth):=\int_{0}^1 (1-t)^{N} (\del^{(0,\b,\ga)}a) (\rx,t\zeta,\vth) \,dt$.
It is straightforward to check that the linear map $a\mapsto a_{\b,\ga,N}$ is continuous from $\Pi_{\sg,\ka,z}^{l,w,m}$ into $\Pi_{\sg,\ka,z}^{l-|\b|,|w|+\ka|\b|,m-|\ga|}$.

The following lemma shows that $\la$-quantization of amplitudes and symbols yields the same operators. This result of ``reduction'' of amplitudes to symbols will be important for Theorem \ref{lambdainv} and thus, for a $\la$-invariant definition of pseudodifferential operators.  

\begin{lem}
\label{reduction} 
\noindent (i) For any $a\in \Pi_{\sigma,\kappa,z}^{l,w,m}$, $(\del^{0,\b,\b} a)_{\zeta=0} \in S^{l-|\b|,m-|\b|}_{\sigma,z}$ for any $n$-multi-index $\b$.  

\noindent (ii) Let $\Ga$ be as in Lemma \ref{noyauReste} and let $a\in \Pi_{\sigma,\kappa,z}^{l,w,m}$. Then for any symbol $s\in S^{l,m}_{\sigma,z}$ such that $s\sim \sum_{\b} \tfrac{(i/2\pi)^{|\b|}}{\b!} (\del^{0,\b,\b} a)_{\zeta=0}$, there is $r \in S^{-\infty}_{\sigma,z}$ such that $\Op_{\Ga}(a)=\Op_{\Ga}(s+r)$. In particular there exists an unique symbol $s(a)\in S^{l,m}_{\sg,z}$ such that $\Op_{\Ga}(a)=\Op_{\Ga}(s(a))$.  Moreover, we have $s(a)\sim \sum_{\b} \tfrac{(i/2\pi)^{|\b|}}{\b!} (\del^{0,\b,\b} a)_{\zeta=0}$.

\noindent (iii) Suppose that $(M,\exp,E,d\mu )$ has a $S_\sigma$-bounded geometry and $\psi$ is a $\O_M$-linearization. 
Let $a\in \Pi_{\sigma,\kappa,z}^{l,w,m}$, $\la\in [0,1]$ and $(z,\bfr)$ be given a frame. Then there exists an unique symbol $s_\la(a)\in S^{l,m}_{\sigma}$ such that $\Op_{\Ga_{\la,z,\bfr}}(a)=(\Op_{\la}(s_\la(a))_{z,\bfr}$. Moreover, we have $T_{z,\bfr,*}(s_\la(a))\sim \sum_{\b} \tfrac{(i/2\pi)^{|\b|}}{\b!} \mu^{-1} (\del^{0,\b,\b} a)_{\zeta=0}$.
\end{lem}
\begin{proof}
$(i)$ is a direct consequence of Lemma \ref{Amplisymb} $(i)$.

\noindent $(ii)$ Using a Taylor expansion of $a$ at $\zeta=0$, we find that for any $u\in \S(\R^{2n},L(E_z))$, $N\in \N^*$, $\langle \Op_{\Ga}(a),u\rangle=\sum_{0\leq |\b|\leq N} I_{\b} + \sum_{|\b|=N+1}\tfrac{N+1}{\b!} R_{\b,N}$ where 
\begin{align*}
&I_\b:= \int_{\R^{3n}} \zeta^\b e^{2\pi i \langle\vth,\zeta \rangle } \Tr\big(\tfrac{1}{\b!}(\del^{(0,\b,0)}a)_{\zeta=0}(\rx,\vth) \Ga(u)^*(\rx,\zeta)\big) d\zeta\,d\vth\,d\rx \, ,\\
&R_{\b,N}:= \int_{\R^{3n}} \zeta^\b e^{2\pi i \langle\vth,\zeta \rangle } \Tr\big(a_{\b,0,N}(\rx,\zeta,\vth)\, \Ga(u)^*(\rx,\zeta)\big) d\zeta\,d\vth\,d\rx \, .
\end{align*}
We get from Lemma \ref{amplireste} $(ii)$,
$$
I_\b= \int_{\R^{3n}} e^{2\pi i \langle\vth,\zeta \rangle } \Tr\big(\tfrac{(i/2\pi)^{|\b|}}{\b!}(\del^{(0,\b,\b)}a)_{\zeta=0}(\rx,\vth) \Ga(u)^*(\rx,\zeta)\big) d\zeta\,d\vth\,d\rx\, .
$$
Let $s\in S^{l,m}_{\sigma,z}$ be a symbol such that $s\sim \sum_{\b} \tfrac{(i/2\pi)^{|\b|}}{\b!} (\del^{0,\b,\b} a)_{\zeta=0}$. Then noting $s_N:=s-\sum_{|\b |\leq N} \tfrac{(i/2\pi)^{|\b|}}{\b!} (\del^{0,\b,\b} a)_{\zeta=0}\in S^{l-(N+1),m-(N+1)}_{\sigma,z}$, we find with Lemma \ref{amplireste} $(ii)$ that  $\Op_{\Ga}(a-s) = \Op_{\Ga}(r_N)$ where 
$$
r_N:= \sum_{|\b|=N+1} \tfrac{(N+1)(i/2\pi)^{N+1}}{\b!} a_{\b,\b,N} -s_N\, .
$$
We check that $r_N\in \Pi_{\sigma,\kappa,z}^{l-(N+1),w_{N},m-(N+1)}$ where $w_N=|w|+\kappa(N+1)$. Corollary \ref{correste} applied to $\Op_{\Ga}(a-s)$ now implies that there is $r \in S^{-\infty}_{\sigma,z}$ such that $\Op_{\Ga}(a)=\Op_{\Ga}(s+r)$. As a consequence, there exists $s(a)\in S^{l,m}_{\sg,z}$ such that $\Op_{\Ga}(a)=(\Op_{\Ga}(s(a))$. The unicity is a direct consequence of the fact that $\Op_{\Ga}=\Ga^*\circ \F_P^*$ on $\S'(\R^{2n},L(E_z))$. 

\noindent $(iii)$ Direct consequence of $(ii)$ and that fact that $(\Op_{\la}(s))_{z,\bfr}=\Op_{\Ga_{\la,z,\bfr}}(\mu_{z,\bfr} s_{z,\bfr})$.
\end{proof}

\subsection{$S_\sigma$-linearizations}\label{Ssigsec}

In order to have a full symbol-operator isomorphism, a polynomial control at infinity on the linearization is not enough. As we shall see, a stronger, ``amplitude-like'' control on the $\psi_{z}^{\bfr}$ maps and a local equivalent of the $P_{x,\xi}$ parallel transport linear isomorphisms (see Remark \ref{remTP}) appears to be crucial for pseudodifferential calculus on $(M,\exp,E)$ and the $\la$-invariance (see Theorem \ref{lambdainv}).

We define $H_{\sigma,\kappa}^w(\mathfrak{E})$ (resp. $E_{\sigma,\kappa}^w(\mathfrak{E})$), where $w\in\R$, $\sigma\in [0,1]$ and $\kappa\geq 0$, as the space of smooth functions $g$ from $\R^{2n}$ into $\mathfrak{E}$ such that for any $2n$-multi-index $\nu$, there exists $C_\nu>0$
such that for any $(\rx,\zeta)\in \R^{2n}$, $\norm{\del^{\nu} g (\rx,\zeta)}\leq C_{\nu} \langle \rx\rangle^{-\sigma (|\nu|-1)}\langle \zeta \rangle^{w+\kappa(|\nu|-1)} $  (if $\nu\neq0$) (resp. $\norm{\del^{\nu} g (\rx,\zeta)}\leq C_{\nu} \langle \rx\rangle^{-\sigma |\nu|}\langle \zeta \rangle^{w+\kappa|\nu|})$. We note $H_{\sigma,\kappa}(\mathfrak{E}):=\cup_{w\in \R} H_{\sigma,\kappa}^w(\mathfrak{E})$, $H_{\sigma}(\mathfrak{E}):=\cup_{\ka\geq 0 } H_{\sigma,\kappa}(\mathfrak{E})$, $E_{\sigma,\kappa}(\mathfrak{E})=\cup_{w\in \R} E_{\sigma,\kappa}^w(\mathfrak{E})$ and $E_{\sigma}(\mathfrak{E})=\cup_{\ka\geq 0} E_{\sigma,\kappa}(\mathfrak{E})$.
Remark that by Leibniz rule, $E_{\sigma,\kappa}(\R)$ and $E_{\sigma,\kappa}(\mathcal{M}_{p}(\R))$ are $\R$-algebras (graduated by the parameter $w$) while $E_{\sigma,\kappa,z}:=E_{\sigma,\kappa}(L(E_z))$ is a $\C$-algebra (under pointwise matricial product). Thus, if $P\in E_{\sigma,\kappa}(\M_p(\R))$, then $\det P \in E_{\sg,\ka}(\R)$. Note also that $f\in H_{\sg,\ka}(\mathfrak{E})$ if and only if for any $i\in \set{1,\cdots,2n}$, $\del_i f \in E_{\sg,\ka}(\mathfrak{E})$. In particular, $f\in H_{\sg,\ka}(\R^p)$ if and only if $df:=(\rx,\zeta)\mapsto (df)_{\rx,\zeta}$ is in $E_{\sg,\ka}(\M_{p,2n}(\R))$. As a consequence, if $f\in H_{\sg,\ka}(\R^{2n})$, its Jacobian determinant $J(g)$ is in $E_{\sg,\ka}(\R)$. Note that any function in $E_{\sg,\ka}^0(\mathfrak{E})$ is bounded and if $f\in H_{\sg,\ka}^0(\mathfrak{E})$ then there is $C>0$ such that $\norm {f(\rx,\zeta)}_{\mathfrak{E}} \leq C \langle \rx,\zeta\rangle$ for any $(\rx,\zeta)\in \R^{2n}$.
The following lemma will give us the behaviour of the  $E_{\sg,\ka}$ and $H_{\sg,\ka}$ spaces under composition.

\begin{lem}
\label{lemGsigma}
(i) Let $f\in H_{\sigma,\kappa}^{w'}(\mathfrak{E})$ (resp. $E^{w'}_{\sigma,\kappa}(\mathfrak{E})$) and $g\in H^w_{\sigma,\kappa}(\R^{2n})$ such that there exists $C,c>0$, $r\geq 0$, such that $\langle g_1(\rx,\zeta) \rangle \geq c \langle \rx\rangle \langle \zeta \rangle^{-r}$ (if $\sg\neq 0$) and $\langle g_2(\rx,\zeta)\rangle \leq C\langle \zeta \rangle$ for any $(\rx,\zeta)\in \R^{2n}$, where $g=:(g_1,g_2)$. Then $f\circ g \in H_{\sg,\ka+|w|+r\sg}^{|w|+|w'|}(\mathfrak{E})$ (resp. $E_{\sg,\ka+|w|+r\sg}^{|w'|}(\mathfrak{E})$).

\noindent (ii) If $P\in E^w_{\sg,\ka}(\M_{n}(\R))$, then $(\rx,\zeta)\mapsto P_{\rx,\zeta}(\zeta) \in H^{w+\ka+1}_{\sg,\ka}(\R^n)$.

\noindent (iii) Let $f\in G_{\sigma}(\R^{n},\mathfrak{E})$  and $g\in H^w_{\sigma,\kappa}(\R^n)$ such that there exists $c>0$, $r\geq 0$, such that, if $\sg\neq 0$, $\langle g(\rx,\zeta) \rangle \geq c \langle \rx\rangle \langle \zeta \rangle^{-r}$ for any $(\rx,\zeta)\in \R^{2n}$. Then $f\circ g\in H_{\sg,\max\set{r\sg,\ka}+|w|}^{|w|}(\mathfrak{E})$. Moreover, if $f\in G_{\sigma}(\R^n,\R^p)$, then $df\circ g \in E_{\sg,\max\set{r\sg,\ka}+|w|}^{0}(\M_{p,n}(\R))$.
\end{lem}
\begin{proof}$(i)$ The Faa di Bruno formula yields for any $2n$-multi-index $\nu\neq 0$,
\begin{equation}
\label{eqFdB}
\del^\nu (f\circ g) = \sum_{1\leq |\la|\leq |\nu|} (\del^\la f)\circ g\  P_{\nu,\la}(g)
\end{equation}
where $P_{\nu,\la}(g)$ is a linear combination (with coefficients independent of $f$ and $g$) of functions of the form $\prod_{j=1}^s (\del^{l^j} g)^{k^j}$ where $s\in \set{1,\cdots ,|\nu|}$.  
The $k^j$ and $l^j$ are $2n$-multi-indices (for $1\leq j\leq s$) such that $|k^j|>0$, $|l^j|>0$, $\sum_{j=1}^s k^j = \la$ and $\sum_{j=1}^s |k^j| l^j= \nu$. As a consequence, since $g\in H_{\sigma,\kappa}^w(\R^{2n})$, we see that for each $\nu,\la$ with $1\leq |\la|\leq |\nu|$ there exists $C_{\nu,\la}>0$ such that for any $(\rx,\zeta)\in \R^n$,
\begin{equation}\label{pnulag}
|P_{\nu,\la}(g) (\rx,\zeta)|\leq C_{\nu,\la} \langle \rx\rangle^{-\sigma(|\nu|-|\la|)} \langle\zeta \rangle^{w|\la|+\kappa(|\nu|-|\la|)} \, . 
\end{equation}
Moreover, since $f\in H_{\sigma,\kappa}^{w'}(\R^{2n})$ (resp. $E_{\sigma,\kappa}^{w'}(\R^{2n})$), there is $C'_\la>0$ such that for any $(\rx,\zeta)\in \R^{2n}$, the estimate $\norm{(\del^\la f) \circ g(\rx,\zeta)}\leq C'_\la \langle \rx\rangle ^{-\sigma(|\la|-1)} \langle \zeta \rangle^{|w'|+(\kappa+r\sg)(|\la|-1)}$ (resp.  
 $\norm{(\del^\la f) \circ g(\rx,\zeta)}\leq C'_\la \langle \rx\rangle ^{-\sigma|\la|} \langle \zeta \rangle^{|w'|+(\kappa+r\sg)|\la|}$) is valid. We deduce then from (\ref{eqFdB}) and (\ref{pnulag}) that $f\circ g$ belongs to $H_{\sg,\ka+|w|+r\sg}^{w+|w'|}(\mathfrak{E})$ (resp. $E_{\sg,\ka+|w|+r\sg}^{|w'|}(\mathfrak{E})$).
 
\noindent $(ii)$ We note $P^{i,j}_{\rx,\zeta}$ the matrix entries of $P_{\rx,\zeta}$. Each component $(f^{i})_{1\leq i \leq n}$ of the map $f:=(\rx,\zeta)\mapsto P_{\rx,\zeta}(\zeta)$ is of the form $f^i= \sum_{j=1}^n P^{i,j}\, \zeta_j$. It is straightforward to check that the applications $(\rx,\zeta)\mapsto \zeta_j$ satify for any $\nu\in \N^{2n}$, $\del^\nu \zeta_j = \O(\langle \zeta\rangle^{1-|\nu|}\langle \rx\rangle^{\sg(1-|\nu|)})$. The result now follows from an application of the Leibniz rule.

 \noindent $(iii)$ Following the proof of $(i)$, (\ref{pnulag}) is still valid, this time with $\la$ as $n$-multi-indices and $\nu$ as $2n$-multi-indices with $1\leq |\la|\leq |\nu|$. Using the fact that $\langle g(\rx,\zeta) \rangle \geq c \langle \rx\rangle \langle \zeta \rangle^{-r}$ for any $(\rx,\zeta)\in \R^{2n}$, we obtain the following estimate
 $$
 \norm{(\del^\la f) \circ g(\rx,\zeta)}\leq C'_\la \langle \rx\rangle ^{-\sigma(|\la|-1)} \langle \zeta \rangle^{r\sigma(|\la|-1)}\leq C'_\la \langle \rx\rangle ^{-\sigma(|\la|-1)} \langle \zeta \rangle^{\max\set{r\sg,\ka}(|\la|-1)}
 $$
which, with (\ref{pnulag}) and (\ref{eqFdB}), yields $f\circ g$ belongs to $H_{\sg,\max\set{r\sg,\ka}+|w|}^{|w|}(\mathfrak{E})$. The fact that $df\circ g$ is in $E_{\sg,\max\set{r\sg,\ka}+|w|}^{0}(\M_{p,n}(\R))$ when $f\in G_{\sg}(\R^n,\R^p)$ is based on the same argument.
\end{proof}

The $H_{\sg,\ka}$ and $E_{\sg,\ka}$ spaces are related to the symbol and amplitude spaces by the following lemma.

\begin{lem}
\label{HEamp}
\noindent (i) If $f\in E^w_{\sg,\ka,z}$, then $(\rx,\zeta,\vth)\mapsto f(\rx,\zeta)$ is in $\Pi^{0,w,0}_{\sg,\ka,z}$.

\noindent (ii) Let $s\in S^{l,m}_{\sigma,z}$, $m\in H^w_{\sg,\ka}(\R^n)$ such that there exist $C,c,r>0$ such that, if $\sg\neq 0$, for any $(\rx,\zeta)\in \R^{2n}$, $ c \langle \rx\rangle \langle \zeta \rangle^{-r}\leq \langle m(\rx,\zeta)\rangle \leq C \langle \rx \rangle \langle \zeta\rangle^{r}$, and $P\in E^{0}_{\sg,\ka}(\M_n(\R))$ such that for any $(\rx,\zeta,\vth)\in \R^{3n}$, $\langle P_{\rx,\zeta}(\vth)\rangle \geq c \langle \vth \rangle$. Then $(\rx,\zeta,\vth)\mapsto s(m(\rx,\zeta),P_{\rx,\zeta}(\vth))$ is in $\Pi^{l,\sigma r|l|,m}_{\sigma,\ka+|\sg r-\ka+w|,z}$.

\noindent (iii) If $s\in S_\sigma(\R^n)$, $m\in H^w_{\sg,\ka}(\R^n)$ such that, if $\sg\neq 0$, there exists $c,r>0$ such that for any $(\rx,\zeta)\in \R^{2n}$ $\langle m(\rx,\zeta)\rangle \geq c \langle \rx\rangle \langle \zeta \rangle^{-r}$, then $(\rx,\zeta,\vth)\mapsto s(m(\rx,\zeta)) \Id_{L(E_z)}$ is in $\Pi_{\sg,\ka+|\sg r-\ka+w|,z}^{0,0,0}$.

\noindent (iv) If $a \in \Pi^{l,w,m}_{\sg,\ka,z}$ and $P\in E^{0}_{\sg,\ka}(\M_n(\R))$ is such that there is $c> 0$ such that for any $(\rx,\zeta,\vth)\in \R^{3n}$, $\langle P_{\rx,\zeta}(\vth)\rangle \geq c \langle \vth \rangle$, then $a_P:(\rx,\zeta,\vth)\mapsto a(\rx,\zeta,P_{\rx,\zeta}(\vth)) \in \Pi^{l,w,m}_{\sg,\ka,z}$.
\end{lem}
\begin{proof}
\noindent $(i)$ is straightforward.

\noindent $(ii)$ Let us note $g(\rx,\zeta,\vth):=(m(\rx,\zeta),P_{\rx,\zeta}(\vth))$. For any $i,j\in \set{1,\cdots, n}$, we denote $P_{\rx,\zeta}^{i,j}$ the $(i,j)$ matrix entry of $P_{\rx,\zeta}$. Since $P\in E^{0}_{\sg,\ka}(\mathcal{M}_n(\R))$, we have $P_{\cdot,\cdot}^{i,j}\in E^{0}_{\sg,\ka}(\R)$. Faa di Bruno formula in Theorem \ref{FaaCS} yields for any $\nu\neq0$ 
\begin{equation}
\label{delnusg}
\del^{\nu} (s\circ g)  =  \sum_{1\leq |\la|\leq |\nu|} (P_{\nu,\la}(g)) \ (\del^{\la}s)\circ g  
\end{equation}
where $P_{\nu,\la}(g)$ is a linear combination of terms of the form $\prod_{j=1}^{s} (\del^{l^j} g)^{k^j} $, where $1\leq s\leq |\nu|$,  
the $k^j$ (resp. $l^j$) are $2n$-multi-indices (resp. $3n$-multi-indices) with $|k^j|>0$, $|l^j|>0$, $\sum_{j=1}^s k^j = \la$ and $\sum_{j=1}^s|k^j|l^j=\nu$. 
Let us note $l^{j}=:(l^{j,1},l^{j,2},l^{j,3})$, $k^{j}=:(k^{j,1},k^{j,2})$ where $l^{j,1},l^{j,2},l^{j,3},k^{j,1},k^{j,2}$ are $n$-multi-indices. We have, noting $Q(\rx,\zeta,\vth):=(\rx,\zeta)$,
$$
(\del^{l^j}g)^{k^j}= \prod_{i=1}^n (\delta_{l^{j,3},0}(\del^{(l^{j,1},l^{j,2})} m)_i\circ Q )^{k^{j,1}_i}\ \prod_{i=1}^n \big(\sum_{k=1}^n \del^{(l^{j,1},l^{j,2})}P^{i,k}_{\cdot,\cdot}\ \del^{l^{j,3}} \vth_{k}\big)^{k^{j,2}_i}
$$
and we get, for a given $s$, $(l^{j})$, $(k^j)$ such that $(\del^{l^j} g)^{k^j} \neq 0$ for all $1\leq j\leq s$,
\begin{align*}
&\text{ if } l^{j,3}=0\, ,  \qquad   (\del^{l^j}g)^{k^j} = \O(\langle \rx\rangle^{-\sigma|l^j||k^{j}|+\sigma|k^{j,1}|}\langle \zeta\rangle^{\ka|l^j||k^{j}|-\ka|k^{j,1}|+w|k^{j,1}|}\langle \vth\rangle^{|k^{j,2}|})\, , \\ 
&\text{ if } |l^{j,3}|=1\, ,\qquad k^{j,1}=0 \text{ and } \    (\del^{l^j}g)^{k^j} = \O(\langle \rx \rangle^{-\sigma |l^j||k^{j}|+\sg|k^j|}\langle \zeta\rangle^{\ka |l^j||k^j|-\ka|k^j|})\, .
\end{align*}

The case $|l^{j,3}|>1$ is excluded since $k^{j}\neq 0$  and $(\del^{l^j}g)^{k^j} \neq 0$.
By permutation on the $j$ indices, we can suppose as in the proof of Lemma \ref{cchange} that for $1\leq j\leq j_1-1$, we have $l^{j,3}=0$ and for $j_1\leq j\leq s$, we have $|l^{j,3}|=1$, where $1\leq j_1\leq s+1$. Thus, we get 
\begin{align*}
&\prod_{j=1}^s (\del^{l^j} g)^{k^j} = \O(\langle\rx \rangle^{-\sg(\sum_{j=1}^{s}(|l^j|-1)|k^j| + \sum_{j=1}^{j_1-1}|k^{j,2}|)}\\
&\hspace{3cm}\times\langle \zeta\rangle^{w\sum_{j=1}^{s}|k^{j,1}|+\kappa(\sum_{j=1}^{s}(|l^j|-1)|k^j| + \sum_{j=1}^{j_1-1}|k^{j,2}|)}\langle \vth \rangle^{\sum_{j=1}^{j_1-1}|k^{j,2}|})\, .
\end{align*}
We check that $\sum_{j=1}^{j_1-1} |k^{j,2}|=|\la^2|-|\ga|$ and $\sum_{j=1}^{s}(|l^j|-1)|k^j| = |\nu|-|\la|$ where $\la=(\la^1,\la^2)$ and $\nu=(\a,\b,\ga)$. As a consequence, 
\begin{equation}
\label{pnu2}
P_{\nu,\la}(g) = \O(\langle \rx\rangle^{-\sigma(|\a+\b|-|\la^1|)}\langle \zeta\rangle^{w|\la^1|+\ka(|\a+\b|-|\la^1|)} \langle \vth \rangle^{|\la^2|-|\ga|}) \, .
\end{equation}
Since there exist $C,c>0$ such that for any $(\rx,\zeta)\in \R^{2n}$ $\langle m(\rx,\zeta)\rangle \leq C \langle \rx\rangle \langle\zeta\rangle^{r}$ and $\langle m(\rx,\zeta)\rangle \geq c \langle \rx\rangle \langle \zeta\rangle^{-r}$, we see that there is $K_\nu>0$ such that for any $1\leq |\la|\leq |\nu|$ and any $(\rx,\zeta)\in \R^{2n}$,
$\langle m(\rx,\zeta)\rangle^{\sigma(l-|\la^1|)}\leq K_\nu \langle \rx \rangle^{\sg(l-|\la^1|)} \langle \zeta \rangle^{\sg r|l|+\sg r|\la^1|}$. As a consequence, we see that there is $C_\nu>0$ such that for any $1\leq |\la|\leq |\nu|$ and any $(\rx,\zeta,\vth)\in \R^{3n}$,
$$
\norm{(\del^\la s)\circ g (\rx,\zeta,\vth)}_{L(E_z)} \leq C_\nu \langle \rx \rangle^{\sg(l-|\la^1|)} \langle \zeta \rangle^{\sg r|l|+\sg r|\la^1|} \langle \vth \rangle^{m-|\la^2|}.
$$
Thus, since we can reduce the sum in (\ref{delnusg}) to $2n$-multi-indices $\la$ such that $|\la^2|\geq |\ga|$ (and thus $|\la^1|\leq |\a+\b|$), we obtain the result from (\ref{pnu2}) and a straightforward verification of the case $\nu=0$. 

\noindent $(iii)$ is obtain exactly as $(ii)$ (with $P_{\rx,\zeta}=\Id$), since $(\rx,\zeta)\mapsto \mu_{z,\bfr}(\rx)\Id_{L(E_z)} \in S^{0,0}_{\sigma,z}$. The hypothesis $ m(\rx,\zeta) = \O (\langle \rx\rangle \langle \zeta\rangle^r)$ is not necessary since $l=0$ here.

\noindent $(iv)$ We have, noting $g(\rx,\zeta,\vth):=(\rx,\zeta,P_{\rx,\zeta}(\vth))$, for any $3n$-multi-indices $\nu\neq 0$, $1\leq |\nu'|\leq |\nu|$, 
$P_{\nu,\nu'}(g)$ as a linear combination of terms of the form $\prod_{j=1}^s (\del^{l^j}g)^{k^j}$, with $\sum_{j=1}^{s}|k^j|l^j=\nu$ and $\sum_{j=1}^s k^j=\nu'$, noting $k^j=(k^{j,1},k^{j,2}), l^{j}=(l^{j,1},l^{j,2})$, where $k^{j,1}$ and $l^{j,1}$ are $2n$-multi-indices, we get, following the proof of $(ii)$, 
$$
P_{\nu,\nu'}(g)=\O(\langle \rx\rangle^{-\sigma(|\a+\b|-|\a'+\b'|)}\langle \zeta\rangle^{\ka(|\a+\b|-|\a'+\b'|)} \langle \vth \rangle^{|\ga'|-|\ga|})\, .
$$
Since $P_{\rx,\zeta}=\O(1)$ and $\langle P_{\rx,\zeta}(\vth)\rangle \geq \eps\langle \vth\rangle$ we get the result.
\end{proof}

\begin{defn}
\label{sprime} Let $\sigma \in [0,1]$ and $\psi$ a linearization on $(M,\exp,E,d \mu)$. We say that $\psi$ is a $S_\sg$-linearization if for any frame $(z,\bfr)$, 
there is $\ka_{z,\bfr}\geq 0$ such that

\noindent $(i)$ $\psi_z^\bfr \in H_{\sg,\ka_{z,\bfr}}(\R^n)$ with ${\psi_z^\bfr}(\rx,\zeta)=\O(\langle \rx\rangle \langle \zeta\rangle^{r})$ for a $r\geq 1$ and $\ol{\psi_z^\bfr} \in \O_M(\R^{2n},\R^n)$ ,

\noindent $(ii)$ there is $P^{z,\bfr}\in C^\infty(\R^{2n}, GL_n(\R))$ such that $P^{z,\bfr}$ and  $(P^{z,\bfr})^{-1}$ are in $E^0_{\sg,\ka_{z,\bfr}}(\M_n(\R))$, and for any $(\rx,\zeta)\in \R^{2n}$, $P_{\rx,\zeta}^{z,\bfr}(\zeta)=\Ups_{1,T}^{z,\bfr}(\rx,\zeta)$ and $P_{\rx,0}^{z,\bfr}=\Id_{\R^n}$. 

\noindent $(iii)$ $\tau_{1}^{z,\bfr}$ and $(\tau_1^{z,\bfr})^{-1}$ are in $E^0_{\sg,\ka_{z,\bfr}}(L(E_z))$.

\noindent We shall say that the combo $(M,\exp,E,d\mu,\psi)$ has a $S_{\sigma}$-bounded geometry if this is the case of $(M,\exp,E,d\mu)$ and $\psi$ is a $S_\sg$-linearization. 
\end{defn}

It is clear that a $S_{\sg}$-linearization is also a $\O_M$-linearization. Moreover, in case of $S_\sg$-bounded geometry, we check the properties $(i)$, $(ii)$ and $(iii)$ in just one frame:

\begin{lem} If $(M,\exp,E,d\mu)$ has a $S_\sg$-bounded geometry and $\psi$ is a linearization such that there exists $(z_0,\bfr_0)$, $\ka_{z_0,\bfr_0}\geq 0$, such that the functions $\psi_{z_0}^{\bfr_0}$, ${\ol \psi}_{z_0}^{\bfr_0}$ satisfy $(i)$, $(ii)$ and $(iii)$, then $\psi$ is a $S_\sg$-linearization.
\end{lem}
\begin{proof} This follows from applications of Lemma \ref{lemGsigma}.
\end{proof}

\begin{rem} The condition $(ii)$ in Definition \ref{sprime} encodes an abstract parallel transport isomorphisms in normal coordinates. Indeed, in the case where the linearization $\psi$ is derived from a connection on $M$, the $GL_n(\R)$-valued smooth functions on $\R^{2n}$: $P^{z,\bfr}:=(\rx,\zeta)\mapsto  M^\bfr_{z,\exp\circ (n_{z,T}^\bfr)^{-1}(\rx,\zeta)} P_{(n_{z,T}^\bfr)^{-1}(\rx,\zeta)} ( M^\bfr_{z,(n_{z}^\bfr)^{-1}(\rx)})^{-1}$ where the applications $P_{x,\xi}$ are the parallel transport isomorphisms on the tangent bundle (see Remark \ref{remTP}), satisfy for any $(\rx,\zeta)\in \R^{2n}$, $P_{\rx,\zeta}^{z,\bfr}(\zeta)=\Ups_{1,T}^{z,\bfr}(\rx,\zeta)$ and $P_{\rx,0}^{z,\bfr}=\Id_{\R^n}$. Thus, in this case, $(ii)$ is satisfied if  $P^{z,\bfr}$ and  $(P^{z,\bfr})^{-1}$ are in $E_{\sg,\ka_{z,\bfr}}^0(\M_n(\R))$ for a $\ka_{z,\bfr}\geq 0$.
\end{rem}

Remark that for any $t\in \R$ and $(\rx,\zeta)\in \R^{2n}$, if $P^{z,\bfr}\in C^\infty(\R^{2n}, GL_n(\R))$ satisfies $(ii)$, then $P^{z,\bfr}_{\rx,t\zeta}(\zeta)= \Ups_{t,T}^{z,\bfr}(\rx,\zeta)$. We shall note $P^{z,\bfr}_t:=(\rx,\zeta)\mapsto P^{z,\bfr}_{\rx,t\zeta}$, so that $P_1^{z,\bfr}=P^{z,\bfr}$ and $P_0^{z,\bfr}=\Id_{\R^n}$. Thus, $\Ups_{t,z,\bfr}(\rx,\zeta)=(\psi_{z}^\bfr(\rx,t\zeta),P^{z,\bfr}_{t,\rx,\zeta}(\zeta))$ and we define the following diffeomorphism on $\R^{3n}$, 
\begin{equation}
\label{xidef} \Xi_{t,z,\bfr}:=(\rx,\zeta,\vth)\mapsto (\Ups_{t,z,\bfr}(\rx,\zeta),\wt P_{t,\rx,\zeta}^{z,\bfr}(\vth))\, .
\end{equation}
We also define the $\R^{2n}$-valued function $\wh\Xi_{t,z,\bfr}:(\rx,\zeta,\vth)\mapsto (\psi_z^\bfr(\rx,t\zeta),\wt P_{t,\rx,\zeta}^{z,\bfr}(\vth))$.
We check that $J(\Xi_{t,z,\bfr}) =J(\Ups_{t,z,\bfr})\, (\det (P_{t}^{z,\bfr})^{-1})$ and $J(\Xi_{t,z,\bfr}^{-1})=J(\Ups_{-t,z,\bfr})\, (\det (P_{t}^{z,\bfr}\circ \Ups_{-t,z,\bfr}))$. Note also that for any $(\rx,\ry)\in \R^{2n}$, $\ol{\psi_{z}^\bfr}(\ry,\rx)= -P_{\rx,\ol{\psi_{z}^\bfr}(\rx,\ry)}^{z,\bfr}(\ol{\psi_{z}^\bfr}(\rx,\ry))$.
\begin{lem}
\label{lem-Phi-la}
Let $(z,\bfr)$ be a given frame, $\la, \la'\in[0,1]$ and $t\in [-1,1]$. Suppose also that $(M,\exp,E, d\mu ,\psi)$ has a $S_\sigma$-bounded geometry. Then

\noindent (i) $P_{t}^{z,\bfr}$, $(P_{t}^{z,\bfr})^{-1}$ are in $E_{\sigma,\kappa_{z,\bfr}}^0(\mathcal{M}_n(\R))$, and $\tau_{t}^{z,\bfr}$, $(\tau_{t}^{z,\bfr})^{-1}$ are in $E_{\sigma,\kappa_{z,\bfr}}^0(L(E_z))$.

\noindent (ii) $m_{t}^{z,\bfr}:=\psi_z^\bfr\circ I_{1,t} \in H_{\sigma,\kappa_{z,\bfr}}(\R^n)$ and there is $c>0$, $r\geq 1$ such that for any $(\rx,\zeta)\in \R^{2n}$, $\langle m_{t}^{z,\bfr} (\rx,\zeta)\rangle \geq c \langle \rx \rangle \langle \zeta\rangle^{-r}$. 

\noindent (iii) There is $c,\eps>0$ such that for any $(\rx,\zeta)\in \R^{2n}$, $\langle \psi_z^\bfr(\rx,\zeta) \rangle \geq c \langle \zeta \rangle^{\eps} \langle \rx \rangle^{-1}$.  

\noindent (iv) $\Phi_{\la,z,\bfr}\in H_{\sigma,\kappa_{z,\bfr}}(\R^{2n})$. In particular $J_{\la,z,\bfr} \in E_{\sigma,\kappa_{z,\bfr}}(\R)$.

\noindent (v) $\Ups_{t,z,\bfr} \in H_{\sigma,\kappa_{z,\bfr}}(\R^{2n})$. In particular $J({\Ups_{t,z,\bfr}}) \in E_{\sigma,\kappa_{z,\bfr}}(\R)$. Moreover, there is $C>0$ such that $\langle (\Ups_{t,T}^{z,\bfr})(\rx,\zeta)\rangle \leq C\langle \zeta \rangle$ for any $(\rx,\zeta)\in \R^{2n}$.

\noindent (vi) $J(\Xi_{t,z,\bfr})$ and $J(\Xi_{t,z,\bfr}^{-1})$ are in $E_{\sigma,\kappa_{z,\bfr}}(\R)$.
\end{lem}
\begin{proof}
$(i)$ The case $t=0$ is obvious. Suppose $t\neq 0$. Since $P_{t}^{z,\bfr}=P^{z,\bfr}\circ I_{1,t}$ and $I_{1,t}\in H_{\sg,\ka_{z,\bfr}}^0$ the result follows from Lemma \ref{lemGsigma} $(i)$. The same argument is applied to $(P_{t}^{z,\bfr})^{-1}$, $\tau_{t}^{z,\bfr}$ and $(\tau_{t}^{z,\bfr})^{-1}$.

\noindent $(ii)$ We shall use the shorthand $m_t:=m_t^{z,\bfr}$. In the case $t=0$, $m_0= \pi_1$, so we obtain the result. Suppose $t\neq 0$. In that case Lemma \ref{lemGsigma} $(i)$ entails that $m_t\in H_{\sg,\ka_{z,\bfr}}(\R^n)$. Since $\Ups_{t,z,\bfr}=(m_{t},\Ups_{t,T}^{z,\bfr})$, we see that $\langle \Ups_{t,z,\bfr}(\rx,\zeta)\rangle=\O(\langle \rx\rangle \langle \zeta\rangle^r)$ for a $r\geq 1$. Thus, there is $C>0$ such that for any $(\rx,\zeta)\in \R^{2n}$, we have $\langle m_t(\rx,\zeta)\rangle \langle P_{t,\rx,\zeta}(\zeta)\rangle^r \geq C \langle \rx,\zeta\rangle$. Since there is $K>0$ such that for any $(\rx,\zeta)\in \R^{2n}$, $\langle P_{t,\rx,\zeta}^{z,\bfr}(\zeta)\rangle \leq K \langle \zeta\rangle$, we obtain the desired estimate.

\noindent $(iii)$ $V:=(\pi_1,\psi_{z}^{\bfr})$ is a diffeomorphism on $\R^{2n}$ with inverse $V^{-1}=(\pi_1,\ol{\psi_z^\bfr})$. Since $\ol{\psi_{z}^\bfr}=\O(\langle \rx,\ry\rangle^{r})$ for a $r\geq 1$ by hypothesis, we see that there is $c>0$ such that 
$\langle \rx,\psi_z^\bfr(\rx,\zeta)\rangle\geq c \langle \rx,\zeta\rangle$ for any $(\rx,\zeta)\in \R^{2n}$. This yields the result.

\noindent $(iv)$ Direct consequence of $(ii)$ and the fact that $\Phi_{\la,z,\bfr}=(m_{\la},m_{\la-1})$.

\noindent $(v)$ follows from a straithforward application of $(ii)$, Lemma \ref{lemGsigma} $(ii)$ and the fact that for any $(\rx,\zeta)\in \R^{2n}$, $\Ups_{t,z,\bfr}(\rx,\zeta)=(m_{t}(\rx,\zeta),P_{t,\rx,\zeta}^{z,\bfr}(\zeta))$.

\noindent $(vi)$ By $(i)$, $(v)$ and Lemma \ref{lemGsigma} $(i)$, $P_{t}^{z,\bfr}\circ \Ups_{-t,z,\bfr} \in E^0_{\sg,\ka}(\M_n(\R))$. Thus the result follows from $(i)$, $(v)$, and the formulas $J(\Xi_{t,z,\bfr}) =J(\Ups_{t,z,\bfr})\, (\det (P_{t}^{z,\bfr})^{-1})$ and $J(\Xi_{t,z,\bfr}^{-1})=J(\Ups_{-t,z,\bfr})\, (\det (P_{t}^{z,\bfr}\circ \Ups_{-t,z,\bfr}))$.
\end{proof}

\subsection{Pseudodifferential operators}\label{pdosection}

\begin{assum} We suppose in this section and until section \ref{exsec} that $(M,\exp,E,d\mu,\psi)$ has a $S_{\sigma}$-bounded geometry.
\end{assum}

\begin{defn} A pseudodifferential operator of order $l,m$ and type $\sigma$ is an element of $\Psi_\sigma^{l,m}:=\mathfrak{Op}_\la(S^{l,m}_\sigma)$, where $\la\in [0,1]$.
\end{defn}
By Lemma \ref{slmdistr}, $S^{l,m}_{\sigma}$ can be seen as included in $\S'(T^*M, L(E))$, so $\mathfrak{Op}_\la(S^{l,m}_\sigma)$ is well defined. 
The following theorem shows that it does not depend on $\la$, and thus justify the notation $\Psi_\sigma^{l,m}$.  We note  $\tau_{R}^{\la,\la'}:=(\tau_{\la}^{z,\bfr})^{-1}\circ \Ups_{\la'-\la,z,\bfr}\,\tau_{\la'}^{z,\bfr}$ and $\tau_{L}^{\la,\la'}:=(\tau_{\la'-1}^{z,\bfr})^{-1}\tau_{\la-1}^{z,\bfr}\circ \Ups_{\la'-\la}^{z,\bfr}$. If $\psi=\exp$, we have  $\tau_{R}^{\la,\la'}=\tau_{R,\la'-\la}$ and $\tau_{L}^{\la,\la'}=(\tau_{L,\la'-\la})^{-1}$ where $\tau_{L,t}:= \tau_{t}^{z,\bfr}$ if $t\neq 1$ and $\tau_{L,t}:=(\tau_{-1}^{z,\bfr})^{-1}\circ \Ups_{1,z,\bfr}$ if $t=1$, and $\tau_{R,t}:=\tau_{t}^{z,\bfr}$ if $t\neq -1$ and $\tau_{R,t}:=(\tau_{1}^{z,\bfr})^{-1}\circ \Ups_{-1,z,\bfr}$ if $t=-1$.

\begin{thm} 
\label{lambdainv}
Let $\la,\la'\in [0,1]$ and $K= \mathfrak{Op}_\la(a)$, with $a\in S^{l,m}_\sigma$.  Then there exists (an unique) $a'\in S^{l,m}_\sigma$ such that  $K=\mathfrak{Op}_{\la'}(a')$. Moreover, for any frame $(z,\bfr)$, 
$$
a'_{z,\bfr}\sim \sum_{\b} \tfrac{(i/2\pi)^{|\b|}}{\b!} \big(\del^{(0,\b,\b)} \tau_{L}^{\la,\la'} a_{\la'-\la}^{z,\bfr} \tau_{R}^{\la,\la'}\big)_{\zeta=0} 
$$
where $a_{z,\bfr}:=T_{z,\bfr,*}(a)$, $a'_{z,\bfr}:=T_{z,\bfr,*}(a')$, and $a_{t}^{z,\bfr}$ is the amplitude defined for any $t\in [-1,1]$ as
\begin{align*}
a_{t}^{z,\bfr}(\rx,\zeta,\vth):= \tfrac{\mu_{z,\bfr}(m_{t}^{z,\bfr}(\rx,\zeta))}{\mu_{z,\bfr}(\rx)}|J\Xi_{t,z,\bfr}(\rx,\zeta)|\,(a_{z,\bfr}\circ \wh\Xi_{t,z,\bfr}(\rx,\zeta,\vth))\, .
\end{align*}
\end{thm}
\begin{proof} 
Let us fix a frame $(z,\bfr)$ and note $a_{z,\bfr}:=T_{z,\bfr,*}(a)$. We saw in Remark \ref{Oplien} that $\Op_{\la}(a)_{z,\bfr}=\Op_{\Ga_{\la,z,\bfr}}(\mu a_{z,\bfr}))$. Thus, for any $u\in \S(M\times M, L(E))$, we have with $u_{z,\bfr}:=T_{z,\bfr,M^2}(u)\in \S(\R^{2n},L(E_z))$, 
\begin{align*}
\langle K,u\rangle =\int_{\R^{3n}}e^{2\pi i\langle \vth,\zeta\rangle}\Tr\big(\mu a_{z,\bfr}(\rx,\vth)\, (\Ga_{\la,z,\bfr}(u_{z,\bfr})(\rx,\zeta))^*\big)\, \,d\zeta\,d\vth\,d\rx\, .
\end{align*}
Suppose that $m\leq -2n$ so that the integral is absolutely convergent. We now proceed to the global change of variables provided by the diffeomorphism $\Xi_{\la'-\la}^{z,\bfr}$ of $\R^{3n}$ ($\Xi_{t,z,\bfr}$ is defined at (\ref{xidef})). We get $\langle K,u \rangle = \langle \Op_{\la',z,\bfr}(\mu  \tau_{L}^{\la,\la'} a_{\la'-\la}^{z,\bfr} \tau_{R}^{\la,\la'}),u_{z,\bfr}\rangle$. We check with Lemmas \ref{lem-Phi-la} and \ref{HEamp} that $ \tau_{L}^{\la,\la'} a_{\la'-\la}^{z,\bfr} \tau_{R}^{\la,\la'}$ is an amplitude in $\Pi_{\sigma,\ka,z}^{l,w,m}$ for a $\ka \geq 0$ and a $w\in \R$. We also see that the linear map $a_{z,\bfr}\mapsto \mu  \tau_{L}^{\la,\la'} a_{\la'-\la}^{z,\bfr} \tau_{R}^{\la,\la'}$ is continuous on $S_{\sg,z}^{l,m}$, which yields, using Proposition \ref{ampliOP} $(ii)$ and the density result of Lemma \ref{toposymbol}, the equality $\langle K,u \rangle = \langle \Op_{\la',z,\bfr}(\mu  \tau_{L}^{\la,\la'} a_{\la'-\la}^{z,\bfr} \tau_{R}^{\la,\la'}),u_{z,\bfr}\rangle$, for any order $m$ of the symbol $a$. 
The result now follows from Lemma \ref{reduction} $(iii)$.
\end{proof}

\begin{prop}
\label{pdoadjoint}
For each $\la\in [0,1]$ and $l,m\in \R$, $\sigma_{\la}$ is a linear isomophism from $\Psi_\sigma^{l,m}$ onto $S^{l,m}_\sigma$ and $\sigma_\la(A^\dag)= (\sigma_{1-\la}(A))^*$ for any $A\in \Psi^{l,m}_\sigma$. In particular a pseudodifferential operator $A$ is formally selfadjoint (i.e $A=A^\dag$ as operators on $\S$) if and only if its Weyl symbol $\sigma_W(A)$ is selfadjoint (as a $L(E)\to T^*M$ section).
\end{prop}
\begin{proof} The fact that $\sigma_{\la}$ is a linear isomophism from $\Psi_\sigma^{l,m}$ onto $S^{l,m}_\sigma$ is a consequence Theorem \ref{lambdainv} and the fact that $\sigma_{\la}$ is a topological isomorphism from $\S'(M\times M,L(E))$ onto $\S'(T^*M,L(E))$. We check that for any $T\in \S'(T^*M,L(E))$, $\Op_{\la}(T)^\dag= \Op_{1-\la}(T^*)$ which is a direct consequence of the fact that $\Phi_{\la}(x,-\xi)=j\circ\Phi_{1-\la}(x,\xi)$ where $j(x,y)=(y,x)$. 
\end{proof}

\begin{prop}
\label{regularity}
Any operator in $\Psi_\sigma^{l,m}$ is regular. Moreover, for any $A\in \Psi_\sigma^{l,m}$ and $v\in \S$, we have 
$$
A(v)\, :x\mapsto  \int_{T_x^*(M)}d\mu_x^*(\th)\int_{T_x(M)}d\mu_x(\xi)\  e^{2\pi i \langle \th,\xi\rangle}\, \sigma_0(A)(x,\th) \,\tau_{-1}^{-1}(x,\xi)\,v(\psi_x^{-\xi})\,.
$$ 
\end{prop} 
\begin{proof} Let $A\in \Psi_{\sigma}^{l,m}$ and $a:=\sigma_0(A)$. Thus, for any frame $(z,\bfr)$, $A_{z,\bfr}=\Op_{\Ga_{0,z,\bfr}}(\mu a_{z,\bfr})$ so by Lemmas \ref{amplContinu}, \ref{lem-Phi-la} $(ii)$ and $(iii)$, $A_{z,\bfr}$ is continuous from $\S(\R^n,E_z)$ into itself. By Proposition \ref{pdoadjoint}, $A^\dag$ is a pseudodifferential operator in $\Psi_{\sigma}^{l,m}$, so we also obtain $(A^\dag)_{z,\bfr}$ continuous from $\S(\R^n,E_z)$ into itself. The result follows.
\end{proof}

\subsection{Link with standard pseudodifferential calculus on $\R^n$ and $L^2$-continuity}
\label{linkstd}

We suppose in this section that $E$ is the scalar bundle. If $A\in \Psi_{\sg}$, then $A_{z,\bfr}$ belongs to the space, noted $\Psi_{\sg,\psi}$, of regular operators $B$ on $\S(\R^n)$, of the form 
$$
B(v) (\rx) = \int_{\R^{2n}} e^{2\pi i \langle \vth,\zeta\rangle} a(\rx,\vth) v(\psi_{z}^\bfr(\rx,-\zeta)) d\zeta d\vth 
$$
where $a\in S^{\infty}_{\sg}(\R^{2n})$. We study in this section a sufficient condition on $\psi$, such that this space $\Psi_{\sg,\psi}$ is in fact equal to the usual algebra $\Psi_{\sg,std}$ pseudodifferential operators on $\R^n$ with the standard linearization $\psi(x,\zeta)= x+\zeta$. Here $\Psi_{0,std}$ corresponds to the Hörmander calculus \cite{Hormander} on $\R^n$ and $\Psi_{1,std}$ is the $SG$-calculus on $\R^n$.

We will note $\psi:=\psi_z^\bfr$, $V_\rx(\zeta):=-\psi(\rx,-\zeta)+\rx$, $M_{\rx,\zeta}:= [\int_0^1 \del_{j}(V_x^{-1})^{i} (t\zeta) dt]_{i,j}$ and $N_{\rx,\zeta}:=[\int_0^1 \del_{j}V_x^{i} (t\zeta) dt]_{i,j}$. We consider the following hypothesis, noted $(H_V)$:

\noindent  $(i)$ there is $\eps,\delta,\eta>0$ such that for any $(\rx,\zeta)\in \R^{2n}$ with $\norm{\zeta}\leq \eps \langle \rx\rangle^{\sg \eta}$, we have $\det M_{\rx,\zeta} \geq \delta$ and $\det N_{\rx,\zeta}\geq \delta$,

\noindent $(ii)$ the functions $(dV_\rx)_{\rx,\zeta}$ and $(dV_\rx^{-1})_{\rx,\zeta}$ are in $E_\sg^0(\M_n(\R))$.

\begin{prop}
\label{propreducRn}
If the hypothesis $(H_V)$ holds, we have $\Psi_{\sg,\psi} = \Psi_{\sg,std}$.
\end{prop}

We set $\chi_{\eps,\eta}(\rx,\zeta):= b(\tfrac{\norm{\zeta}^2}{\eps^{2}\langle \rx\rangle^{2\sg \eta}})$ where $b\in C^{\infty}_c(\R,[0,1])$ is such that $b=0$ on $\R\backslash]-1,1[$ and $b=1$ on $[-1/4,1/4]$. 

\begin{lem}
\label{H1H2cons}
Suppose $(H_V)$.
If $a \in S^{l,m}_\sg(\R^{2n})$, then the application $$
a_{\chi,M}:(\rx,\zeta,\vth)\mapsto \chi_{\eps,\eta}(\rx,\zeta) a(\rx,\wt M_{\rx,\zeta} \vth) |J(V_\rx^{-1}|(\zeta)\,(\det M_{\rx,\zeta} )^{-1}$$
is an amplitude in $\cup_{k,w}$ $\Pi_{\sg,\ka,z}^{l,w,m}(\R^{3n})$. Similarly, 
$$a_{\chi,N}:(\rx,\zeta,\vth)\mapsto \chi_{\eps,\eta}(\rx,\zeta) a(\rx,\wt N_{\rx,\zeta} \vth) |J(V_\rx)|(\zeta)\,(\det N_{\rx,\zeta} )^{-1}$$ is in $\bigcup_{k,w}\Pi_{\sg,\ka,z}^{l,w,m}(\R^{3n})$.
\end{lem}
\begin{proof}
The result follows from Lemma \ref{HEamp} $(ii)$ and applications of Proposition \ref{inverse}.
\end{proof}

\begin{proof}[Proof of Proposition \ref{propreducRn}] Suppose that $a\in S^{l,m}_{\sg}(\R^{2n})$ and define $A$ as the operator in $\Psi_{\sg,\psi}$ with normal symbol $a$. We obtain for any $v\in \S(\R^{2n})$
$$
 A(v)(\rx):=\int_{\R^{2n}} e^{2\pi i \langle \vth,\zeta\rangle} a(\rx,\vth) v(\psi(\rx,-\zeta)) d\zeta d\vth \, .
$$
We suppose first that $a\in S^{-\infty}_{\sg}(\R^{2n})$. We have after a change of variable, and cutting the integral in two parts 
$A(v)(\rx) = A_1(v)(\rx) + A_2(v)(\rx)$ where 
\begin{align*}
&A_1(v)(\rx) = \int_{\R^{2n}}  e^{2\pi i \langle \vth,M_{\rx,\zeta}(\zeta)\rangle} \chi_{\eps,\eta}(\rx,\zeta) a(\rx,\vth) |J(V_\rx^{-1})|(\zeta) v(\rx-\zeta) d\zeta d\vth\, , \\
&A_2(v)(\rx) = \int_{\R^{2n}}  e^{2\pi i \langle \vth,V_\rx^{-1}(\zeta)\rangle} (1-\chi_{\eps,\eta})(\rx,\zeta) a(\rx,\vth) |J(V_\rx^{-1})|(\zeta) v(\rx-\zeta) d\zeta d\vth \, .
\end{align*}
In $A_1$, we permute the integrations $d\zeta$ and $d\vth$ and proceed to a change of the variable $\vth$, while in $A_2$ we integrate by parts in $\vth$ using formula (\ref{Mformula}) so that for any $p\in \N$,
\begin{align*}
&A_1(v)(\rx) = \int_{\R^{2n}}  e^{2\pi i \langle \vth,\zeta \rangle}  a_{\chi,M}(\rx,\zeta,\vth) v(\rx-\zeta) d\zeta d\vth\,  ,\\
&A_2(v)(\rx) = \int_{\R^{2n}}   e^{2\pi i \langle \vth,V_\rx^{-1}(\zeta)\rangle} (1-\chi_{\eps,\eta})(\rx,\zeta) \, ^t M_{\vth}^{p,V_\rx^{-1}(\zeta)}(a) |J(V_\rx^{-1})|(\zeta) v(\rx-\zeta) d\zeta d\vth \, .
\end{align*}
As a consequence with Lemma \ref{H1H2cons}, and with the density of $S^{-\infty}_{\sg}(\R^{2n})$ in $S^{l,m}_{\sg}(\R^{2n})$, we see that $A$ is the sum of two pseudodifferential operators in $\Psi_{\sg,std}$: $A= A_\chi + R$ where $R\in\Psi_{\sg,std}^{-\infty}$ and $A_\chi$ has $a_{\chi,M}$ as (standard) amplitude.
The implication in the other sense is similar.
\end{proof}

\begin{rem} 
In the case of pseudodifferential operator with local compact control over the $x$ variable and with $\psi$ coming from a connection, by cutting-off in the $\zeta$-variable or in other words taking $y:=\psi(x,-\zeta)$ and $x$ sufficiently close to each other, we have in fact $\Psi_{\sg,\psi}$ equal to $\Psi_{\sg,std}$ modulo smoothing elements (see \cite{Shara1}).
\end{rem}

As a consequence, we see that if the hypothesis $(H_V)$ is satisfied for a frame $(z,\bfr)$, then $\Psi_{\sg,\psi}(=\Psi_{\sg,std})$ is stable under composition of operators and the symbol composition formula is then given by a quadruple asympotic summation modulo smoothing symbols. 

We will show in the next section that we can also obtain stability under composition directly, without using a reduction to the standard calculus on $\R^n$. We shall obtain with this method a simpler symbol composition formula on $\Psi_{\sg,\psi}$, analog to the usual one on $\Psi_{\sg,std}$.

As a direct consequence of the previous proposition, we have the following $L^2$-continuity result for pseudodifferential operators on $M$.

\begin{prop} 
\label{L2cont}
If $(H_V)$ is satisfied for the function $V_{\rx}^{-1}$ in a frame $(z,\bfr)$, then any pseudodifferential operators on $M$ of order $(0,0)$ extends as a bounded operator on $L^2(M,d\mu)$.
\end{prop}
\begin{proof} Since $(H_V)$ is satisfied for $V_{\rx}^{-1}$, the proof of the previous proposition entails that $\Psi_{\sg,\psi}^{0,0} \subseteq \Psi_{\sg,std}^{0,0}$, so the result follows from the $L^2$-continuity of standard pseudodifferential operators \cite{Hormander}.
\end{proof}

\subsection{Composition of pseudodifferential operators}
\label{composec}
The goal of this section is to prove that pseudodifferential operators of $\Psi_\sigma^\infty$ are stable under composition without using the hypothesis of the previous section, and to obtain an adapated symbol composition formula. We shall adapt to our situation a technique used for Fourier integral operators in Coriasco \cite{Coriasco}, Ruzhansky and Sugimoto \cite{Ruzhansky3,Ruzhansky}.

Let us note for $(x,\xi)\in TM$ and $\xi'\in T_{\psi_x^{-\xi}}(M)$,
$\psi_{x,\xi,\xi'}:=\psi_{\psi_x^{-\xi}}^{-\xi'}$, $r_x(\xi,\xi'):= \psi_x^{-1}(\psi_{x,\xi,\xi'})$ and $q_x(\xi,\xi'):=\psi^{-1}_{\psi_{x,\xi,\xi'}}(\psi_x^{-\xi})$.
 We define $V_x$ the $2n$ dimensional smooth manifold as $V_x:=\set{(\xi,\xi')\in T_x(M)\times \cup_{y\in M}T_y(M) \ | \ \xi'\in T_{\psi_x^{-\xi}}(M)}$. Each $V_x$ manifold is diffeomorphic to $\R^{2n}$ via the map, defined for any fixed frame $(z,\bfr)$, $n_{z,V_x}^\bfr(\xi,\xi'):=( M_{z,x}^\bfr(\xi), M_{z,\psi_x^{-\xi}}^{\bfr}(\xi'))$, and has a canonical involutive diffeomorphism $R_x$ defined as 
$$
R_x : (\xi,\xi')\mapsto (r_{x}(\xi,\xi'), q_x(\xi,\xi')) \, .
$$ 
In all the following we fix a frame $(z,\bfr)$, and note also $\psi$ the function $m_{-1}^{z,\bfr}$. We note $\rx^{\zeta,\zeta'}:=\psi(\psi(\rx,\zeta),\zeta')$. For each $\rx\in \R^n$, $R_\rx:=n_{z,V_{(n_z^\bfr)^{-1}(\rx)}}^{\bfr}\circ R_{(n_z^\bfr)^{-1}(\rx)} \circ (n_{z,V_{(n_z^\bfr)^{-1}(\rx)}}^{\bfr})^{-1}$ is a diffeomorphism on $\R^{2n}$, and we define $R_{\rx}=:(r_\rx,q_\rx)$, $r=r^{z,\bfr}:=(\rx,\zeta,\zeta')\mapsto r_\rx(\zeta,\zeta')$ and $q=q^{z,\bfr}:=(\rx,\zeta,\zeta')\mapsto q_{\rx}(\zeta,\zeta')$.
Remark that $r_{\rx}(\zeta,\zeta')=-\ol{\psi_z^\bfr}(\rx,\rx^{\zeta,\zeta'})=:\ol\psi_\rx \circ \psi_{\psi_\rx(\zeta)}(\zeta')$ and $q_{\rx}(\zeta,\zeta')=-P^{z,\bfr}_{-1,\psi(\rx,\zeta),\zeta'} (\zeta')$. The map $r_{\rx,\zeta}:\zeta'\mapsto r_{\rx}(\zeta,\zeta')$ is a diffeomorphism on $\R^n$ such that $r_{\rx,\zeta}^{-1}=r_{\psi_\rx(\zeta),\ol\psi_{\psi(\rx,\zeta)}(\rx)}$ so that $(dr_{\rx,\zeta})_{\zeta'}^{-1}=(dr_{{\psi_\rx(\zeta),\ol\psi_{\psi(\rx,\zeta)}(\rx)}})_{r_{\rx,\zeta}(\zeta')}$. We will use the shorthand $\tau:=(\tau_{-1}^{z,\bfr})^{-1}$.

We note $s(\rx,\zeta,\zeta'):=r(\rx,\zeta,\zeta')-\zeta$. We have $s(\rx,\zeta,\zeta')=s_{\rx,\zeta}(\zeta')$ where $s_{\rx,\zeta}=T_{-\zeta}\circ \ol\psi_\rx \circ \psi_{\psi_{\rx}(\zeta)}$ is a diffeomorphism on $\R^n$ such that $s_{\rx,\zeta}(0)=0$.
We also define 
$$
\varphi_{\rx,\zeta}(\zeta'):=r_{\rx,\zeta}(\zeta')-\zeta-(dr_{\rx,\zeta})_0(\zeta')
$$
so that $\varphi_{\rx,\zeta}(0)=0$ and $(d\varphi_{\rx,\zeta})_0=0$, and 
$$
V(\rx,\zeta,\zeta'):= (dr_{\rx,\zeta})_{\zeta'}
$$
as a smooth function from $\R^{3n}$ into $\M_n(\R)$. We shall note $(\rx,\zeta) \mapsto L_{\rx,\zeta}:=-\,^t(dr_{\rx,\zeta})_0$.

We define $ \O_{\sigma,\kappa,\eps_0,\eps_1,c}^{l,w_0,w_1}(\mathfrak{E})$, where $c\in \N$, $l\in \R$, $w:=(w_0,w_1)\in\R^2_+$, $\eps:=(\eps_0,\eps_1)$, $\eps_0\geq 0$, $\eps_1>0$, $\sigma\in [0,1]$ and $\kappa\geq 0$, as the space of smooth functions $g$ from $\R^{3n}$ into $\mathfrak{E}$ such that for any $3n$-multi-index $\nu=(\mu,\ga)\in \N^{2n}\times \N^n$, there exists $C_\nu>0$
such that for any $(\rx,\zeta,\zeta')\in \R^{3n}$, $\norm{\del^{\nu} g (\rx,\zeta,\zeta')}\leq C_{\nu} \langle \rx\rangle^{\sigma(l- |\mu|-\eps_1 |\ga|_c)}\langle \zeta\rangle^{w_0+\ka|\mu|+\eps_0|\ga|} \langle\zeta' \rangle^{w_1+\kappa|\nu|}$. Here, we denoted $|\ga|_c:= 0$ if $|\ga|<c$ and $|\ga|_c:=|\ga|-c$ if $|\ga|\geq c$. We note $ \O_{\sg,\ka,\eps}(\mathfrak{E}):=\cup_{c,l,w} \O_{\sigma,\kappa,\eps,c}^{l,w}(\mathfrak{E})$. We check that for any multi-indices $\ga,\ga'$ and $c,c'\in \N$, $|\ga|_c+|\ga'|_c\geq |\ga+\ga'|_{c+c'}$, and $|\ga+\ga'|_c\geq |\ga|_c+|\ga'|_c$.
Thus, $ \O_{\sg,\ka,\eps}(\R)$, $\O_{\sg,\ka,\eps}(\mathcal{M}_{p}(\R))$ and $\O_{\sg,\ka,\eps,z}:=\O_{\sg,\ka,\eps}(L(E_z))$ are algebras (graduated by the parameters $c$, $l$, $w_0$ and $w_1$) and $\del^\nu \O_{\sg,\ka,\eps,c}^{l,w}(\mathfrak{E})\subseteq \O_{\sg,\ka,\eps,c}^{l-|\mu|-\eps_1|\ga|_c,w_0+\ka|\mu|+\eps_0|\ga|,w_1+\ka|\nu|}(\mathfrak{E})$.
If $f\in \O_{\sg,\ka,\eps,c}^{0,w}(\mathfrak{E})$, then $(\rx,\zeta)\mapsto f(\rx,\zeta,0)\in E^{w_0}_{\sg,\ka}(\mathfrak{E})$, and if $f\in \O_{\sg,\ka,\eps,c,z}^{l,w}$, then $(\rx,\zeta,\vth)\mapsto f(\rx,\zeta,0) \in \Pi_{\sg,\ka,z}^{l,w_0,0}$. Remark that any monomial of the form $(\rx,\zeta,\zeta')\mapsto \zeta'^\b$ where $\b \in \N^n$, is in $\O_{\sg,\ka,\eps,|\b|}^{0,0,|\b|}(\R)$ for any $\ka\geq 0$ and $\eps_0\geq 0$, $\eps_1>0$.

In the definition of $S'_\sg$ bounded geometry, we only require a polynomial control over the $\ol\psi_{z}^\bfr$ functions. It appears that for the theorem of composition, a stronger control over these functions is important. We thus introduce the following:

\begin{defn}
\label{Csigma}
We shall say that $(C_\sg)$ is satisfied if there is a frame $(z,\bfr)$, $(\ka_v, w_v)\in \R^2_+$ with $\ka_v\geq 1$, and $\eps_v\in ]0,1[$, such that 
\begin{align}
\label{CHyp1}
V \in \O_{\sg,\ka_v,\eps_v,\eps_v,0}^{0,0,w_v}(\M_n(\R))\,,  \quad\, \text{and}\quad (d\psi_{z,\rx}^\bfr)_\zeta ,\,(d\ol\psi_{z,\rx}^\bfr)_\ry = \O(1)\, . 
\end{align}
\end{defn}
\noindent In particular $(C_\sg)$ entails that $(dr_{\rx,\zeta})_{0}$ and thus $L$ are in $E_{\sg,\ka_v}^0(\M_n(\R))$. 

We note $\RR_{\sg,\ka,\eps_1}^{w_0,w_1}(\mathfrak{E})$ ($\eps_1>0$) as the space of smooth functions $g$ such that for any nonzero $\nu=(\mu,\ga)\in \N^{2n}\times \N^n$, $\del^\nu g  = \O(\langle \rx\rangle^{\sg(1-|\mu|-\eps_1|\ga|)}\langle \zeta\rangle^{w_0 + \ka(|\nu|-1)} \langle \zeta'\rangle^{w_1+\ka(|\nu|-1)})$. It follows from $(C_\sg)$ that $r\in \cup_{w_0,w_1}\RR_{\sg,\ka_v,\eps_v/2}^{w_0,w_1}(\R^n)$.

The following lemma will give us the link between the the $\O$, $\RR$, $H$, $E$ spaces and the behaviour under composition.

\begin{lem}
\label{htilde}

$(i)$ Let $f\in H_{\sg,\ka}^w(\mathfrak{E})$ (resp. $E_{\sg,\ka}^w(\mathfrak{E})$) and $g\in  \RR_{\sg,\ka,\eps_1}^{w_0,w_1}(\R^{2n})$ such that $g_2(\rx,\zeta,\zeta')=\O(\langle\zeta \rangle^{k_2}\langle \zeta'\rangle^{k'_2})$ for a $(k_2,k'_2)\in \R^2_+$ and, if $\sg\neq 0$, $\langle g_1(\rx,\zeta,\zeta')\rangle\geq c\langle \rx\rangle \langle \zeta\rangle^{-k_1}\langle \zeta'\rangle^{-k_1'}$, for a $(k_1,k'_1)\in \R^2_+$ and $c>0$. Then, $f\circ g\in \RR_{\sg,\ka_H,\eps_1}^{w_0+k_2w,w_1+k'_2 w}(\mathfrak{E})$ (resp. $\O_{\sg,\ka_E,\ka_E,\eps_1,0}^{0,k_2 w,k'_2 w}(\mathfrak{E})$) where $\ka_H:=\ka+\max\set{|w_0+k_1\sg +k_2\ka|,|w_1+k'_1\sg +k'_2\ka|}$ and $\ka_E:= \ka+\max\set{|w_0+k_1\sg+(k_2-1)\ka|,|w_1+k'_1\sg+(k'_2-1)\ka|}$.

\noindent $(ii)$ $(\rx,\zeta,\zeta')\mapsto (\psi(\rx,\zeta),\zeta') \in \RR_{\sg,\ka_\psi,1}^{w_\psi,0}(\R^{2n})$ and $(\rx,\zeta,\zeta')\mapsto \rx^{\zeta,\zeta'} \in \RR_{\sg,\ka_\psi,1}$ for a $(\ka_\psi,w_\psi)\in \R^2_+$.

\noindent $(iii)$ The functions $q$, $(\rx,\zeta,\zeta')\mapsto (P^{z,\bfr}_{-1,\psi(\rx,\zeta),\zeta'})^{-1}$ and $(\rx,\zeta,\zeta')\mapsto \det (P^{z,\bfr}_{-1,\psi(\rx,\zeta),\zeta'})^{-1}$ are respectively in $\RR_{\sg,\ka_q,1}(\R^n)$, $\O^{0,0,0}_{\sg,\ka_q,\ka_q,1,0}(\M_n(\R))$, and $\O^{0,0,0}_{\sg,\ka_q,\ka_q,1,0}(\R)$, for a $\ka_q\geq 0$.
Moreover, there exists $C>0$ such that for any $(\rx,\zeta,\zeta')\in \R^{3n}$, $\norm{q_\rx(\zeta,\zeta')}\leq C\langle \zeta' \rangle$.

\noindent $(iv)$ $(\rx,\zeta,\zeta')\mapsto \tau(\rx^{\zeta,\zeta'},q_\rx(\zeta,\zeta'))$ is in $\O^{0,0,0}_{\sg,\ka_\tau,\ka_\tau,1,0,z}$ for a $\ka_\tau\geq 0$.
\end{lem}
\begin{proof}
$(i)$ If $\nu=(\a,\b,\ga)\neq 0$ is a $3n$-multi-index, we have $\del^\nu f\circ g = \sum_{1\leq |\nu'|\leq |\nu|} P_{\nu,\nu'}(g) (\del^{\nu'}f)\circ g$, with $P_{\nu,\nu'}(g)$ a linear combination of terms of the form $\prod_{j=1}^s (\del^{l^j}g)^{k^j}$, with $1\leq s \leq |\nu|$, $\sum_{1}^s l^j |k^j|=\nu$, $\sum_1^s k^j=\nu'$.
As a consequence, we get the following estimate for any $1\leq |\nu|\leq |\nu'|$, $P_{\nu,\nu'}(g)=\O(\langle \rx\rangle^{\sg(|\nu'|-|\mu|-\eps_1|\ga|)} \langle \zeta\rangle^{w_0|\nu'|+\ka(|\nu|-|\nu'|)}\langle \zeta'\rangle^{w_1|\nu'|+\ka(|\nu|-|\nu'|)})$.
Moreover, for any $1\leq |\nu'|\leq |\nu|$, there is $C_\nu>0$ such that for any $(\rx,\zeta,\zeta')\in \R^{3n}$, the following estimate is valid $\norm{(\del^{\nu'}f) \circ g(\rx,\zeta,\zeta')}\leq C_\nu \langle\rx\rangle^{-\sg(|\nu'|-1)} \langle \zeta\rangle^{(k_1\sg+k_2\ka)(|\nu'|-1)+k_2w}\langle \zeta'\rangle^{(k_1'\sg+k'_2\ka)(|\nu'|-1)+k'_2w}$ (resp. $\norm{(\del^{\nu'}f) \circ g(\rx,\zeta,\zeta')}\leq C_\nu \langle\rx\rangle^{-\sg|\nu'|} \langle \zeta\rangle^{(k_1\sg+k_2\ka)|\nu'|+k_2w}\langle \zeta'\rangle^{(k_1'\sg+k'_2\ka)|\nu'|+k'_2w}$). The result follows.

\noindent $(ii)$ By hypothesis, $\psi\in H^{w_\psi}_{\sg,\ka_\psi}$. We deduce that $(\rx,\zeta,\zeta')\mapsto \psi(\rx,\zeta)\in \RR^{w_\psi,0}_{\sg,\ka_\psi,1}$ and the first statement now follows from $(\rx,\zeta,\zeta')\mapsto \zeta'\in \RR^{0,0}_{\sg,\ka_\psi,1}$. The second statement follows from $(i)$.

\noindent $(iii)$ Since $q_{\rx}(\zeta,\zeta')=-P^{z,\bfr}_{-1,\psi(\rx,\zeta),\zeta'}(\zeta')$, the fact that $q_\rx\in  \RR_{\sg,\ka_q,1}(\R^n)$ for a $\ka_q\geq 0$  is a consequence of $(i)$, $(ii)$ and Lemma \ref{lemGsigma} $(iii)$. We also have by $(i)$ and $(ii)$, $(P^{z,\bfr}_{-1,\psi(\rx,\zeta),\zeta'})^{-1} \in  \O^{0,0,0}_{\sg,\ka_q,\ka_q,1,0}(\M_n(\R))$.

\noindent $(iv)$ Since $\tau\in E^0_{\sg,\ka}(L(E_z))$ for a $\ka\geq 0$,
the result follows $(i)$, $(ii)$, $(iii)$ and the estimate $\langle \rx^{\zeta,\zeta'}\rangle \geq c \langle \rx\rangle \langle \zeta\rangle^{-k} \langle \zeta' \rangle^{-k}$ for $c,k>0$.
\end{proof}

\begin{lem}
\label{Csgr}
Suppose $(C_\sg)$. Then 

\noindent (i) $s,\varphi \in \O_{\sg,\ka_v,\eps_v,\eps_v,1}^{0,0,w_s}(\R^n)$ and $\varphi\in \O_{\sg,\ka_v,\eps_v,\eps_v,2}^{-\eps_v,\eps_v,w_\varphi}(\R^n)$ where $w_s:=w_v+1$ and $w_\varphi:=2+w_v+\ka_v$.

\noindent (ii) $V=(dr_{\rx,\zeta})_{\zeta'}$  and  $(dr_{\rx,\zeta})^{-1}_{\zeta'}$ are bounded on $\R^{3n}$.

\noindent (iii) The function $J(R):(\rx,\zeta,\zeta')\mapsto J(R_\rx)(\zeta,\zeta')$ is in $\cup_{\ka,w_0,w_1,\eps_0,\eps_1}\O_{\sg,\ka,\eps_0,\eps_1,0}^{0,w_0,w_1}(\R)$ and $(\rx,\zeta,\zeta')\mapsto \tau (\rx,r_\rx(\zeta,\zeta'))$ is in $\O_{\sg,\ka_\tau,\ka_\tau,\eps_v/2,0,z}^{0,0,0}$ for $\ka_\tau\geq 0$.
\end{lem}

\begin{proof}
\noindent $(i)$ We have $s_{\rx,\zeta}(\zeta')=\sum_{i=1}^n\zeta'_i \int_0^1 \del_{\zeta'_i} r_{\rx,\zeta}(t\zeta')\,dt$. Since $V \in \O_{\sg,\ka_v,\eps_v,0}^{0,0,w_v}(\M_n(\R))$ each function $(\rx,\zeta,\zeta')\mapsto\int_0^1 \del_{\zeta'_i} r_{\rx,\zeta}(t\zeta')\,dt$ is in $\O_{\sg,\ka_v,\eps_v,0}^{0,0,w_v}(\R^n)$ and thus, since $(\rx,\zeta,\zeta')\mapsto \zeta'_i \in \O_{\sg,\ka_v,\eps_v,1}^{0,0,1}(\R)$, we see that $s \in \O_{\sg,\ka_v,\eps_v,1}^{0,0,w_s}(\R^n)$. We have also $\varphi_{\rx,\zeta}(\zeta')=\sum_{|\b|=2} \tfrac{2}{\b!} (\zeta')^\b \int_{0}^1(1-t)\, \del_{\zeta'}^\b r_{\rx,\zeta}(t\zeta')\,dt$ and each function $(\rx,\zeta,\zeta')\mapsto\int_0^1 (1-t)\, \del_{\zeta'}^\b r_{\rx,\zeta}(t\zeta')\,dt$ is in $\O_{\sg,\ka_v,\eps_v,0}^{-\eps_v,\eps_v,w_v+\ka_v}(\R^n)$. With $(\rx,\zeta,\zeta')\mapsto (\zeta')^\b \in \O_{\sg,\ka_v,\eps_v,2}^{0,0,2}(\R)$, we get $\varphi \in \O_{\sg,\ka_v,\eps_v,2}^{-\eps_v,\eps_v,w_\varphi}(\R^n)$.

\noindent $(ii)$ Direct consequence of $(C_\sg)$ and the following equalities for any $(\rx,\zeta,\zeta')\in \R^{3n}$, $(dr_{\rx,\zeta})_{\zeta'} =(d\ol\psi_{\rx})_{\rx^{\zeta,\zeta'}}(d\psi_{\psi_\rx(\zeta)})_{\zeta'}$ and $(dr_{\rx,\zeta})^{-1}_{\zeta'}=(d\ol\psi_{\psi_\rx(\zeta)})_{\rx^{\zeta,\zeta'}}(d\psi_\rx)_{r_{\rx,\zeta}(\zeta')}$. 

\noindent $(iii)$ The first statement follows from Lemma \ref{htilde} $(ii)$. The second statement follows from Lemma \ref{htilde} $(i)$ and the estimate $r_{\rx}(\zeta,\zeta')=\O(\langle \zeta \rangle \langle \zeta'\rangle^{w_v})$.
\end{proof}

We shall use a generalization to four variables of the $\Pi_{\sg,\ka,z}^{l,w,m}$ spaces of amplitude. We define $\wt \Pi_{\sg,\ka,\eps_1,z}^{l,w_0,w_1,m}$ ($0<\eps_1\leq 1$) as the space of smooth functions $a\in C^{\infty}(\R^{4n},L(E_z))$ such that for any $4n$-multi-index $(\nu,\delta)\in \N^{3n}\times \N^n$, (with $\nu=(\mu,\ga)\in \N^{2n}\times \N^n$) there is $C_{\nu,\delta}>0$  such that for any $(\rx,\zeta,\zeta',\vth)\in \R^{4n}$, 
$$
\norm{\del^{(\nu,\delta)} a (\rx,\zeta,\zeta',\vth)}_{L(E_z)}\leq C_{\nu,\delta} \langle \rx\rangle^{\sg(l-|\mu|-\eps_1|\ga|)} \langle \zeta \rangle^{w_0+\ka|\nu|} \langle \zeta'\rangle^{w_1+\ka|\nu|} \langle \vth\rangle^{m-|\delta|}\, .
$$ 
These spaces have natural Fr\'{e}chet topologies and form a graded topological algebra under pointwise composition.

\begin{lem}
\label{amptilde}
$(i)$ If $a \in \wt \Pi_{\sg,\ka,\eps_1,z}^{l,w_0,w_1,m}$, then $a_{\zeta'=0}:(\rx,\zeta,\vth)\mapsto a(\rx,\zeta,0,\vth)$ is in $\Pi_{\sg,\ka,z}^{l,w_0,m}$.
 
\noindent $(ii)$ If $h\in  \O_{\sg,\ka,\eps_0,\eps_1,0,z}^{l,w_0,w_1}$, then $(\rx,\zeta,\zeta',\vth)\mapsto h(\rx,\zeta,\zeta')$ is in $\wt\Pi_{\sg,\max\set{\ka,\eps_0},\eps_1,z}^{l,w_0,w_1,0}$.

\noindent $(iii)$ There is $\ka_\Xi,k_1 \geq 0$ such that for any $b\in S^{l,m}_{\sg,z}$, the application $b\circ \wt \Xi$, where $\wt \Xi(\rx,\zeta,\zeta',\vth):=(\rx^{\zeta,\zeta'},-\wt P_{-1,\psi(\rx,\zeta),\zeta'}^{z,\bfr}(\vth))$, is in $\wt\Pi_{\sg,\ka_{\Xi},1,z}^{l,\sg k_1|l|,\sg k_1|l|,m}$.

\end{lem}
\begin{proof}
$(i)$ and $(ii)$ are direct.

\noindent $(iii)$ If $\mu=(\nu,\delta)\neq 0$ is a $4n$-multi-index, we have $\del^{\mu}(b\circ \wt \Xi)=\sum_{1\leq |\mu'|\leq |\mu| }P_{\mu,\mu'}(\wt \Xi) \, (\del^{\mu'} b )\circ \wt \Xi$ with $P_{\mu,\mu'}(\wt \Xi)$ a linear combination of terms of the form $\prod_{j=1}^s (\del^{l^j}\wt \Xi)^{k^j}$, with $1\leq s \leq |\mu|$, $l^j=(l^{j,1},l^{j,2})\in \N^{3n}\times \N^n$, $k^{j}=(k^{j,1},k^{j,2})\in \N^{n}\times \N^n$, such that $l^{j,2}=0$ for $1\leq j\leq j_1 \leq s$, and $\sum_{1}^s l^j |k^j|=\mu$, $\sum_1^s k^j=\mu'$. We have
$$
(\del^{l^j}\wt\Xi)^{k^j}= \prod_{i=1}^n (\delta_{l^{j,2},0}(\del^{l^{j,1}} \rx^{\zeta,\zeta'})_i )^{k^{j,1}_i}\ \prod_{i=1}^n \big(\sum_{k=1}^n \del^{l^{j,1}}P^{i,k}\ \del^{l^{j,2}} \vth_{k}\big)^{k^{j,2}_i}
$$
where $P^{i,k}$ are the matrix entries of $-\wt P_{-1,\psi(\rx,\zeta),\zeta'}^{z,\bfr}$. By Lemma \ref{htilde} $(ii)$ and $(iii)$, $\rx^{\zeta,\zeta'} \in \RR_{\sg,\ka_\psi,1}^{w_0,w_1}(\R^n)$ and the $P^{i,k}$ are in $\O^{0,0,0}_{\sg,\ka_\psi,\ka_\psi,1,0}(\R)$ for a $(\ka_\psi,w_0,w_1)\in \R^3_+$. We obtain thus the following estimate 
$$
|P_{\mu,\mu'}(\wt \Xi)(\rx,\zeta,\zeta',\vth) |\leq C_\mu \langle \rx\rangle ^{-\sg(|\nu|-|\a'|)} \langle \zeta \rangle^{w_0|\a'|+\ka_\psi(|\nu|-|\a'|)}\langle \zeta'\rangle^{w_1|\a'|+\ka_\psi(|\nu|-|\a'|)} \langle \vth\rangle^{|\b'|-|\delta|}
$$
with $\mu'=:(\a',\b')$. Since $b\in S^{l,m}_{\sg,z}$ we also have the estimate 
$$
\norm{(\del^{\mu'} b) \circ \wt \Xi (\rx,\zeta,\zeta')} \leq C'_{\mu} \langle \rx^{\zeta,\zeta'} \rangle^{\sg(l-|\a'|)} \langle \vth \rangle^{m-|\b'|}  
$$
so the result follows now from the estimate $\langle \rx^{\zeta,\zeta'} \rangle^{\sg(l-|\a'|)} =\O( \langle \rx\rangle^{\sg(l-|\a'|)} (\langle \zeta\rangle \langle \zeta'\rangle)^{\sg k_1 |l|+\sg k_1 |\a'|})$, with $\ka_{\Xi}:=\ka_{\psi}+\max\set{|w_0+\sg k_1-\ka_\psi|,|w_1+\sg k_1 -\ka_\psi|}$.
\end{proof}

\begin{lem}
\label{derivphase}
Let $s\in C^\infty(\R^p,\R^n)$. Then for any $p+n$-multi-index $\nu=(\a,\b)\neq 0$, we have 
$$
\del_{\rx,\vth}^\nu\, e^{i\langle \vth,s(\rx)\rangle } = P_{\nu}(\rx,\vth)\, e^{i\langle \vth,s(\rx)\rangle }
$$
where $P_{\nu}$ is of the form $\sum_{|\ga|\leq |\a|} \vth^\ga \, T_{\nu,\ga}(\rx)$, and $T_{\nu,\ga}$ is a linear combination of terms of the form $\prod_{j=1}^m (\del^{l^j} s)^{\mu^j}$ where $1\leq m\leq |\nu|$, $(l^j)$ are $p$-multi-indices  and $(\mu^j)$ are $n$-multi-indices. Moreover, they satisfy $|\mu^j|>0$, $\sum_{j=1}^m |\mu^j|=|\ga|+|\b|$, $\sum_{j=1}^{m} |\mu^j||l^j|=|\a|$ and if $|\b|=0$, then $|l^j|>0$ and $|\ga|>0$.
\end{lem}
\begin{proof}
We note $g(\rx,\vth):=\langle \vth,s(\rx)\rangle$. By Theorem \ref{FaaCS}, we get the following equality for any $\nu\neq 0$, $\del_{\rx,\vth}^\nu\, e^{i\langle \vth,s(\rx)\rangle } = P_{\nu}(\rx,\vth) e^{i\langle \vth,s(\rx)\rangle}$ where $P_{\nu}(\rx,\vth)=\sum_{1\leq k\leq |\nu|} P_{\nu,k}(g)$ and $P_{\nu,k}$ is a linear combination of terms of the form $\prod_{j=1}^{m} (\del^{l^j} g)^{k^j}$ such that $|l^{j}|>0$, $k^j>0$, $\sum_{1}^m k^{j}= k$ and $\sum_1^m k^j l^j = \nu$. If we suppose that the term $\prod_{j=1}^{m} (\del^{l^j} g)^{k^j}$ is non-zero, then $|l^j|\leq 1$ and if we define $j_1$ such that for any $1\leq j \leq j_1$, $l^{j,2}=0$, we obtain, noting $l^{j}=(l^{j,1},l^{j,2})$,
\begin{align*}
\prod_{j=1}^{m} (\del^{l^j} g)^{k^j}&= \prod_{j=1}^{j_1} \langle \vth, \del^{l^{j,1}} s\rangle^{k^j} \prod_{j=j_1+1}^m (\del^{l^{j,1}} s^{q_j})^{k^j} \\
&= \sum_{|\ga^j|=k^j,\,1\leq j\leq j_1} \ga^1!\cdots \ga^j!\ \vth^{\sum_1^{j_1} \ga^j} \prod_{j=1}^{j_1} (\del^{l^{j,1}}s)^{\ga^j} \prod_{j=j_1+1}^m (\del^{l^{j,1}} s^{q_j})^{k^j} \, .
\end{align*}
Thus, we have $P_{\nu,k}=\sum_{|\ga|=k-|\b|} \vth^\ga \, T_{\nu,\ga,k}(\rx)$ where $T_{\nu,\ga,k}$ is a linear combination of terms of the form $\prod_{j=1}^{j_1} (\del^{l^{j,1}} s)^{\mu^{j}} \prod_{j=j_1+1}^m (\del^{l^{j,1} } s^{q_j})^{k^j}$, where $1\leq q_j\leq n$, $1\leq j\leq m \leq |\nu|$, $1\leq j_1\leq m$, $l^{j,1}\in \N^p$, $k^j\in \N^*$, $\la^j\in \N^n$ are such that $\sum_{1}^m k^j=k$, $\sum_{1}^{j_1} |\la^j||l^{j,1}| + \sum_{j_1
+1}^{m} k^j |l^{j,1}+1|=|\nu|$ and $\sum_{j_1+1}^m k^j = |\b|$. The result follows.
\end{proof}

\begin{lem}
\label{phasecontrol}
Suppose that $(C_\sg)$ is satisfied. Then

\noindent (i) Representing by $\ru$ the letter $s$ or $\varphi$, for any $3n$-multi-index $\nu=(\mu,\ga)\in \N^{2n}\times \N^n$, we have the equality $\del^\nu_{\rx,\zeta,\vth} e^{2\pi  i \langle \vth,\ru_{\rx,\zeta}(\zeta')\rangle} = (\sum_{|\om|\leq |\mu|} \vth^\om T_{\nu,\om,\ru}(\rx,\zeta,\zeta'))\, e^{2\pi  i \langle \vth,\ru_{\rx,\zeta}(\zeta')\rangle}$ where each term $T_{\nu,\om,s}\in \O_{\sg,\ka_v,\eps_v,\eps_v,|\om+\ga|}^{-|\mu|,\ka_v|\mu|,w_s|\om+\ga|+\ka_v|\mu|}(\R)$ and $T_{\nu,\om,\varphi}\in \O_{\sg,\ka_v,\eps_v,\eps_v,2|\om+\ga|}^{-|\mu|-\eps_v|\om+\ga|,\eps_v|\om+\ga|+\ka_v|\mu|,w_\varphi|\om+\ga|+\ka_v|\mu|}(\R)$. In particular, it satisfies the following estimate valid for any $(\rx,\zeta,\zeta')\in \R^{3n}$, and any $n$-multi-index $\rho$,
\begin{align*}
&|\del^\rho_{\zeta'} T_{\nu,\om,s}(\rx,\zeta,\zeta') |\leq C_{\nu,\om,\rho}\langle \rx\rangle^{-\sg(|\mu|+\eps_v |\rho|_{|\om+\ga|})}\langle \zeta \rangle^{\ka_v|\mu|+\eps_v |\rho|} \langle \zeta'\rangle^{w_s|\om+\ga|+\ka_v(|\mu|+|\rho|)}\, , \\
& |\del^\rho_{\zeta'} T_{\nu,\om,\varphi }(\rx,\zeta,\zeta') |\leq C_{\nu,\om,\rho}\langle \rx\rangle^{-\sg(|\mu|+(\eps_v/2)|\rho|)}\langle \zeta \rangle^{\eps_v|\om+\ga|+\ka_v|\mu|+\eps_v|\rho|} \langle \zeta'\rangle^{w_\varphi|\om+\ga|+\ka_v(|\mu|+|\rho|)}\, .
\end{align*}

\noindent (ii) For any $n$-multi-index $\b$, we have $\del^{\b}_{\zeta'} e^{2\pi i \langle \vth, \varphi_{\rx,\zeta}(\zeta')\rangle} = P_{\b,\varphi}(\rx,\zeta,\zeta',\vth) e^{2\pi i \langle \vth, \varphi_{\rx,\zeta}(\zeta')\rangle}$ where $P_{\b,\varphi}(\rx,\zeta,\zeta',\vth)$ is a linear combination of terms of the form $\vth^\om \zeta'^\la t_{\om,\la}(\rx,\zeta,\zeta')$ where $\om$ and $\la$ are $n$-multi-indices satifying $|\om|\leq |\b|$, $(2|\om|-|\b|)_+\leq |\la|\leq |\om|$, and $t_{\om,\la}$ are functions in $\O_{\sg,\ka_v,\eps_v,\eps_v,|\b|}^{-\eps_v|\b|/2,2\eps_v,w'_s|\b|}(\R)$. In particular they are estimated by 
$$
t_{\om,\la}(\rx,\zeta,\zeta') = \O(\langle \rx\rangle^{-\sg\eps_v|\b|/2} \langle \zeta\rangle^{2\eps_v|\b|} \langle \zeta'\rangle^{w'_s|\b|})
$$
where  $w'_s:=w_s+2\ka_v$.
Moreover, $(\rx,\zeta,\vth)\mapsto P_{\b,\varphi}(\rx,\zeta,0,\vth)\, 1_{L(E_z)}\in  \Pi_{\sg,\ka_v,z}^{-\eps_v|\b|/2,\eps_v|\b|,|\b|/2}$.

\noindent (iii) If $\b\in \N^n$ and $f\in \wt \Pi_{\sg,\ka,\eps_1,z}^{l,w_0,w_1,m}$ then the function 
$$
f_{\b,\varphi}: (\rx,\zeta,\vth)\mapsto \del^{\b}_{\zeta'}\big(e^{2\pi i \langle \vth,\varphi_{\rx,\zeta}(\zeta')\rangle} \del^{0,0,0,\b}f(\rx,\zeta,\zeta',L_{\rx,\zeta}(\vth))\big)_{\zeta'=0}$$
belongs to $\Pi_{\sg,\ka_1,z}^{l-\eps'_1 |\b|,w_0+\ka_2|\b|,m-|\b|/2}$, where $\eps'_1:=\min\set {\eps_1/2,\eps_v/2}>0$, $\ka_1:=\max\set{\ka_v,\ka}$, $\ka_2:=\ka+|\eps_v-\ka|$, and the 
application $f\mapsto f_{\b,\varphi}$ is continuous.

\end{lem}

\begin{proof} $(i)$ By Lemma \ref{derivphase}, if $\nu\neq 0$, we have the following equality, valid for any $(\rx,\zeta,\zeta',\vth)\in {\R^{4n}}$, $\del^\nu_{\rx,\zeta,\vth} e^{2\pi  i \langle \vth,\ru_{\rx,\zeta}(\zeta')\rangle} = (\sum_{|\om|\leq |\mu|} \vth^\om T_{\nu,\om,\ru}(\rx,\zeta,\zeta'))\, e^{2\pi  i \langle \vth,\ru_{\rx,\zeta}(\zeta')\rangle}$ where $T_{\nu,\om,\ru}$
is a linear combination of terms of the form $\prod_{j=1}^m( \del_{\rx,\zeta}^{l^j} \ru )^{\mu^j}$ with $1\leq m \leq |\nu|$, $\mu^j\neq 0$, $\sum_{j=1}^m |\mu^j|= |\om+\ga|$ and $\sum_{j=1}^{m} |\mu^j||l^j|=|\mu|$. Since by Lemma \ref{Csgr} $(i)$, $s\in \O_{\sg,\ka_v,\eps_v,\eps_v,1}^{0,0,w_s}(\R^n)$, it is straightforward to check that $T_{\nu,\om,s}\in \O_{\sg,\ka_v,\eps_v,\eps_v,|\om+\ga|}^{-|\mu|,\ka_v|\mu|,w_s|\om+\ga|+\ka_v|\mu|}(\R)$. Moreover, since $\varphi\in \O_{\sg,\ka_v,\eps_v,\eps_v,2}^{-\eps_v,\eps_v,w_\varphi}(\R^n)$, we get $T_{\nu,\om,\varphi}\in \O_{\sg,\ka_v,\eps_v,\eps_v,2|\om+\ga|}^{-|\mu|-\eps_v|\om+\ga|,\eps_v|\om+\ga|+\ka_v|\mu|,w_\varphi|\om+\ga|+\ka_v|\mu|}(\R)$. The first estimate is direct and the second estimate follows from the inequality $|\om+\ga|+|\rho|_{2|\om+\ga|} \geq |\rho|/2$.

\noindent $(ii)$  By Lemma \ref{derivphase}, if $\b\neq 0$, we have for any $(\rx,\zeta,\zeta',\vth)\in {\R^{4n}}$, the following relation $\del^\b_{\zeta'} e^{2\pi  i \langle \vth,\varphi_{\rx,\zeta}(\zeta')\rangle} = (\sum_{1\leq|\om|\leq |\b|} \vth^\om T_{\b,\om,\varphi}(\rx,\zeta,\zeta')) e^{2\pi  i \langle \vth,\varphi_{\rx,\zeta}(\zeta')\rangle}$ where $T_{\b,\om,\varphi}$
is a linear combination of terms of the form $\prod_{j=1}^m( \del^{l^j} \varphi_{\rx,\zeta} )^{\mu^j}$ with $1\leq m \leq |\b|$, $\mu^j\neq 0$, $l^j\neq 0$, $\sum_{j=1}^m |\mu^j|= |\om|$ and $\sum_{j=1}^{m} |\mu^j||l^j|=|\b|$. Let us reorder the $l^j$ indices so that for any $1\leq j\leq j_1$, $|l^j|=1$ and for any  $j\geq j_1+1$, $|l^j|>1$, where $j_1 \in \set{0,\cdots m}$. Thus $\prod_{j=1}^m( \del^{l^j} \varphi_{\rx,\zeta} )^{\mu^j} = \prod_{j=1}^{j_1}( \del^{l^j} \varphi_{\rx,\zeta} )^{\mu^j} \prod_{j\geq j_1+1} ( \del^{l^j} \varphi_{\rx,\zeta} )^{\mu^j}$ and with a Taylor expansion at order 1 of $\del^{l^j}\varphi_{\rx,\zeta}$ in $\zeta'$ around 0 when $1\leq j\leq j_1$, we get $\del^{l^j}\varphi_{\rx,\zeta} = \sum_{1\leq i\leq n} \zeta'_i t_{i,j}^k$ where $t_{i,j}^k = \int_{0}^1 \del_{\zeta'}^{e_i+l^j} \varphi_{\rx,\zeta}(t\zeta')dt$. Thus, using the fact that $\varphi\in \O_{\sg,\ka_v,\eps_v,\eps_v,1}^{0,0,w_s}(\R^n)$, we see that $\prod_{j=1}^{j_1} (\del^{l^j}\varphi_{\rx,\zeta} )^{\mu^j}$ is a linear combination of terms of the form $\zeta'^\la V_\la$ where $|\la|=\sum_{j=1}^{j_1}|\mu^j|$ 
and 
$$
V_\la =\O( \langle \rx \rangle^{-\sg \eps_v \sum_1^{j_1}|l^j||\mu^j|} \langle \zeta\rangle^{\eps_v|\la| + \eps_v\sum_1^{j_1}|\mu^j||l^j|} \langle \zeta'\rangle^{(k_v +w_s)|\la|+ \ka_v\sum_{1}^{j_1}|l^j||\mu^j|}).
$$
As a consequence, we see that $\prod_{j=1}^{m} (\del^{l^j}\varphi_{\rx,\zeta} )^{\mu^j}$ is a linear combination of terms of the form $\zeta'^\la W_\la$ where $|\la|=\sum_{j=1}^{j_1}|\mu^j|$ and
$$
W_\la =\O( \langle \rx \rangle^{-\sg \eps_v (|\b|-v)} \langle \zeta\rangle^{2\eps_v|\b|} \langle \zeta'\rangle^{w'_s|\b|})
$$
where $v:=\sum_{j=j_1+1}^{m} |\mu^j|=|\om|-|\la|$. The first statement now follows from the inequality $2v\leq |\b|-|\la|$.

Since $\varphi_{\rx,\zeta}(0)=0$ and $(d\varphi_{\rx,\zeta})_0=0$, $P_{\b,\varphi}(\rx,\zeta,0,\vth)$ is a linear combination of terms of the form $\vth^\om \prod_{j=1}^m (\del^{0,0,l^j} \varphi(\rx,\zeta,0))^{\mu^j}$ with  
$1\leq |\om|\leq |\b|/2$, $1\leq m \leq |\b|$, $\mu^j\neq 0$, $|l^j|\geq 2$, $\sum_{j=1}^m |\mu^j|= |\om|$ and $\sum_{j=1}^{m} |\mu^j||l^j|=|\b|$.
We check with Lemma \ref{Csgr} $(i)$ that any function of the form $\prod_{j=1}^m (\del^{0,0,l^j} \varphi(\rx,\zeta,\zeta'))^{\mu^j}$ is in $\O_{\sg,\ka_v,\eps_v,|\b|/2}^{-\eps_v|\b|/2,\eps_v|\b|,(w_s/2+\ka_v)|\b|}(\R)$, and thus, $(\rx,\zeta,\vth)\mapsto \prod_{j=1}^m (\del^{0,0,l^j} \varphi(\rx,\zeta,0))^{\mu^j} \,1_{L(E_z)} \in \Pi_{\sg,\ka_v,z}^{-\eps_v |\b|/2,\eps_v|\b|,0}$. Since $(\rx,\zeta,\vth)\mapsto \vth^\om\, 1_{L(E_z)}\in \Pi_{\sg,\ka_v,z}^{0,0,|\b|/2}$ we obtain $(\rx,\zeta,\vth)\mapsto P_{\b,\varphi}(\rx,\zeta,0,\vth)\,1_{L(E_z)}\in \Pi_{\sg,\ka_v,z}^{-\eps|\b|/2,\eps_v|\b|,|\b|/2}$.

\noindent $(iii)$ We have
\begin{align*}
f_{\b,\varphi} (\rx,\zeta,\vth)&= \sum_{\b'\leq \b}\tbinom{\b}{\b'}\del^{\b'}_{\zeta'}(e^{2\pi i \langle \vth,\varphi_{\rx,\zeta}(\zeta')\rangle})_{\zeta'=0}\, \del^{0,0,\b-\b',\b}f(\rx,\zeta,0,L_{\rx,\zeta}(\vth))\\
&=\sum_{\b'\leq \b}\tbinom{\b}{\b'}P_{\b',\varphi}(\rx,\zeta,0,\vth)\, \del^{0,0,\b-\b',\b}f(\rx,\zeta,0,L_{\rx,\zeta}(\vth))\, .
\end{align*}
Since $(\rx,\zeta)\mapsto L_{\rx,\zeta}\in E^0_{\sg,\ka_v}(\M_n(\R))$ and $L_{\rx,\zeta}^{-1}=\O(1)$, we deduce from Lemma \ref{amptilde} $(i)$ and Lemma \ref{HEamp} $(iv)$ that $(\rx,\zeta,\vth)\mapsto \del^{0,0,\b-\b',\b}f(\rx,\zeta,0,L_{\rx,\zeta}(\vth))$ belongs to the amplitude space $\Pi_{\sg,\max\set{\ka,\ka_v},z}^{l-\eps_1|\b-\b'|,w_0+\ka|\b-\b'|,m-|\b|}$. The result now follows from $(ii)$.
\end{proof}

We now introduce two parametrized cut-off functions that will be used later. Let $b \in C^{\infty}_c(\R,[0,1])$ such that $b=1$ on $[-1/4,1/4]$ and $b=0$ on $\R\backslash]-1,1[$. We define for $\eps,\delta,\eta_1,\eta_2 >0$ with $\eps,\delta<1,$
\begin{align*}
&\chi_{\eps}(\vth,\vth'):=b(\tfrac{\norm{\vth'}^2}{\eps^2\langle \vth\rangle^{2}})\, ,\\
&\chi_{\delta,\eta}(\rx,\zeta,\zeta') := b(\tfrac{\norm{\zeta'}^2}{\delta^2\langle \rx\rangle^{2\sg \eta_1} \langle \zeta\rangle^{-2\eta_2}}).
\end{align*}

\begin{lem}
\label{lemchi} The cut-off functions $\chi_{\eps}$ and $\chi_{\delta,\eta}$ are repectively in the spaces $C^{\infty}(\R^{2n},[0,1])$ and $C^{\infty}(\R^{3n},[0,1])$ and satisfy: 

\noindent $(i)$ For any $(\rx,\zeta,\zeta')\in \R^{3n}$, if $\norm{\zeta'}\leq \half \delta \langle \rx\rangle^{\sg\eta_1}\langle \zeta\rangle^{-\eta_2}$, then $\chi_{\delta,\eta}(\rx,\zeta,\zeta')=1$, and if $\norm{\zeta'}\geq \delta\langle  \rx\rangle^{\sg\eta_1}\langle \zeta\rangle^{-\eta_2}$, then $\chi_{\delta,\eta}(\rx,\zeta,\zeta')=0$. In particular, for any $(\rx,\zeta)\in \R^{2n}$, $\chi_{\delta,\eta}(\rx,\zeta,0)=1$ and for any $3n$-multi-index $\nu\neq 0$, $(\del^\nu \chi_{\delta,\eta} )(\rx,\zeta,0)=0$.

\noindent $(ii)$ For any $(\vth,\vth')\in \R^{2n}$, if $\norm{\vth'}\leq \half \eps \langle \vth\rangle$, then $\chi_\eps(\vth,\vth')=1$, and if $\norm{\vth'}\geq \eps\langle \vth\rangle$, then $\chi_\eps(\vth,\vth')=0$. In particular, for any $\vth\in \R^n$, $\chi_\eps(\vth,0)=1$ and for any $2n$-multi-index $\nu\neq 0$, $(\del^\nu \chi_\eps)(\vth,0)=0$.

\noindent $(iii)$ For any $3n$-muti-index $\nu=(\a,\b,\ga)$, we have $\del^{\nu} \chi_{\delta,\eta}(\rx,\zeta,\zeta')=\O(\langle \rx\rangle^{-|\a|} \langle \zeta\rangle^{-\b}\langle \zeta'\rangle^{-|\ga|})$, and $\del^{\nu} \chi_{\delta,\eta}(\rx,\zeta,\zeta')=\O(\langle \rx\rangle^{-\sg|\nu|} \langle \zeta\rangle^{(-1+\eta_2/\eta_1)|\b|+(\eta_2/\eta_1)|\ga|}\langle \zeta'\rangle^{(\eta_1^{-1}-1)|\ga|+\eta_1^{-1}|\b|})$. In particular, the function $\chi_{\delta,\eta}$ is in $\O_{\sg,\ka'_\eta,\ka'_\eta,1,0}^{0,0,0}(\R)$ for a $\ka'_\eta>0$.

\noindent $(iv)$ For any $2n$-muti-index $\nu$, $\del^{\nu} \chi_\eps(\vth,\vth') =\O(\langle \vth\rangle^{-|\nu|} )$ and $\del^{\nu} \chi_\eps(\vth,\vth')=\O(\langle \vth'\rangle^{-|\nu|}).$

\end{lem}
\begin{proof}  $(i)$ and $(ii)$ are straightforward. For any $\nu\neq 0$, 
$\del^\nu \chi_{\delta,\eta} = \sum_{1\leq \nu'\leq |\nu|} P_{\nu,\nu'}(g)\, (\del^{\nu'}b)\,\circ\, g$ where $g(\rx,\zeta,\zeta'):= \tfrac{\norm{\zeta}^2}{\delta^2\langle \rx\rangle^{2\sg\eta_1}\langle \zeta\rangle^{-2\eta_2}}$. We obtain from a direct computation the estimate $P_{\nu,\nu'}(g)=\O(\langle \rx\rangle^{-2\sg\eta_1\nu'-|\a|}\langle \zeta\rangle^{2\eta_2\nu'-|\b|} \langle \zeta'\rangle^{2\nu'-|\ga|})$.
Since for any $\nu\in \N$, we have $\del^{\nu'} b=\O(1)$ we obtain $\del^{\nu} \chi_{\delta,\eta}=\O(\langle \rx\rangle^{-|\a|} \langle \zeta\rangle^{-\b}\langle \zeta'\rangle^{-|\ga|}1_{D_\delta})$ where $D_\delta$ is the set of triples $(\rx,\zeta,\zeta')$ satifying the inequalities $\delta/2\leq \langle  \zeta'\rangle \langle \rx\rangle^{-\sg\eta_1}\langle \zeta\rangle^{\eta_2}\leq \sqrt{2}$. The estimates of $(iii)$ follow. The proof of $(iv)$ is similar.
\end{proof}

We will use in the following lemma the space $\O_{\ka}^{t_0,t_1,j}$ ($\ka\geq 0$, $j\in \N$, $(t_0,t_1)\in \R_+^2$) of functions $f\in C^{\infty}(\R^{4n},\C)$ such that for any $\a\in \N^n$,
there is $C_\a>0$ such that for any $(\rx,\zeta,\zeta',\vth)\in \R^{4n}$,  $|\del^\a_{\zeta'} f (\rx,\zeta,\zeta',\vth)| \leq C_\a \langle \zeta\rangle^{t_0+\ka|\a|} \langle \zeta' \rangle^{t_1+\ka|\a|} \langle \vth\rangle^{-2j}$. Clearly, $\O_{\ka}^{t_0,t_1,j}\O_{\ka}^{t'_0,t'_1,j'}\subseteq \O_{\ka}^{t_0+t'_0,t_1+t'_1,j+j'}$ and $\del_{\zeta'}^\a \O_{\ka}^{t_0,t_1,j}\subseteq \O_\ka^{t_0+\ka|\a|,t_1+\ka|\a|,j}$. 

\begin{lem}
\label{lem-hL} .
Defining $h(\rx,\zeta,\zeta',\vth):= \big(1+\norm{^t(ds_{\rx,\zeta})_{\zeta'}(\vth)}^2-(i/2\pi)\langle \vth,(\Delta s_{\rx,\zeta})(\zeta')\rangle\big)^{-1}$, we have the following relation, valid for any $(\rx,\zeta,\zeta',\vth)\in \R^{4n}$, $p\in \N$,
\begin{align*}
e^{2\pi i \langle \vth,s_{\rx,\zeta}(\zeta')\rangle } = (h(\rx,\zeta,\zeta',\vth)\, L_{\zeta'})^p e^{2\pi i \langle \vth,s_{\rx,\zeta}(\zeta')\rangle}\, 
\end{align*}
where $L_{\zeta'}:=1-(2\pi)^{-2}\Delta_{\zeta'}$. Moreover, if $(C_\sg)$ holds, there is $\ka_L \geq 0$ such that for any $p\in \N$, there is $N_p\in \N^*$, $(h^p_k)_{1\leq k\leq N_p}$ functions in $\O_{\ka_L}^{2p\ka_L,2p\ka_L,p}$, $(\b^{k,p})_{1\leq k\leq N_p}$ $n$-multi-indices satisfying $|\b^{k,p}|\leq 2p$, such that  $(L_{\zeta'}\,h)^p = \sum_{k=1}^{N_p} h_k^{p}\,\del_{\zeta'}^{\b^{k,p}}$. 
\end{lem}
\begin{proof} We obtain $L_{\zeta'} e^{2\pi i \langle \vth,s_{\rx,\zeta}(\zeta')\rangle}= (1/h) e^{2\pi i \langle \vth,s_{\rx,\zeta}(\zeta')\rangle}$ through a direct computation. Let us show the remaining statement by induction on $p$. Note that by Lemma \ref{Csgr} $(ii)$, we have $|1/h|\geq c  \langle\vth\rangle^{2}$ for a $c>0$ and we check that $1/h \in \wt\Pi_{\sg,\ka_v,\eps_v,z}^{0,\eps_v,w'_v,2}$ where $w'_v=\max\set{2w_v,w_v+\ka_v}$. With a reccurence or using Proposition \ref{inverse}, we check that $h\in \O_{\ka_L}^{0,0,1}$ where $\ka_L:=\max\set{2\eps_v,w'_v+\ka_v}$. The property is obviously true for $p=0$.  Suppose now that the property is true for $p\geq 0$, so that $(L_{\zeta'}\,h)^p =\sum_{k=1}^{N_p} h_k^{p}\,\del_{\zeta'}^{\b^{k,p}}$ with $N_p\in \N^*$, $(h^p_k)_{1\leq k\leq N_p}$ functions in $\O_{\ka_L}^{2p\ka_L,2p\ka_L,p}$ and $(\b^{k,p})_{1\leq k\leq N_p}$ $n$-multi-indices satisfying $|\b^{k,p}|\leq 2p$.  We also have 
\begin{align*}
&(L_{\zeta'} h)^{p+1} = (L_{\zeta'} h) \sum_{k=1}^{N_p} h_k^{p}\,\del_{\zeta'}^{\b^{k,p}} = \sum_{k=1}^{N_p} hh^{p}_k \del_{\zeta'}^{\b^{k,p}} -(2\pi)^{-2}\big(\Delta_{\zeta'}(h h_k^p) \del_{\zeta'}^{\b^{k,p}} \\
&\hspace{4cm}+2 \sum_{i=1}^n \del_{\zeta'_i} (hh^p_k)\del^{\b^{k,p}+e_i}_{\zeta'} +hh^{p}_k\Delta_{\zeta'}\del^{\b^{k,p}}_{\zeta'}\big)
\end{align*}
so the property holds for $p+1$.
\end{proof}

We note $\S_{\sg,c}(\R^{3n},L(E_z))$ the space of smooth functions $f$ such that for any $N\in \N^*$ and $\nu=(\mu,\ga)\in \N^{2n}\times \N^n$, $\del^\nu f (\rx,\zeta,\vth)=\O(\langle \rx\rangle^{-\sg N}\langle \zeta\rangle^{c_0+c_1 N +c_2 |\mu|}\langle  \vth \rangle^{-N})$. It follows from Lemma \ref{noyauReste} that if $f\in \S_{\sg,c}(\R^{3n},L(E_z))$, then $\Op_\Ga(f) \in \Op_\Ga(S^{-\infty}_{\sg,z})$. Here and thereafter $\Ga$ satisfies the hypothesis of Lemma \ref{noyauReste}.

\begin{lem}
\label{lemnegamp} Assume that $(C_\sg)$ holds.

\noindent (i) For any $l,w_0,w_1,m,\ka$, $S_{m,w_1}(\wt \Pi_{\sg,\ka,\eps_1,z}^{l,w_0,w_1,,m})\subseteq \S_{\sg,c}(\R^{3n},L(E_z))$ for a triple $c:=(c_0,c_1,c_2)$ and the linear map $S_{m,w_1}:f\mapsto S_{m,w_1}(f)$ is continuous, where
$$
S_{m,w_1}(f):(\rx,\zeta,\vth)\mapsto \int_{\R^{2n}} e^{2\pi i (\langle \vth',\zeta'\rangle+ \langle \vth,s_{\rx,\zeta}(\zeta')\rangle)}\, ^t M_{\vth'}^{p_{m,w_1},\zeta'}(f)(\rx,\zeta,\zeta',\vth') (1-\chi_{\delta,\eta})(\rx,\zeta,\zeta')\, d\vth'\,d\zeta'
$$
and $p_{m,w_1}:=\max \set{m+2n,[|w_1|]+1+2n}$.

\noindent (ii) For any $u\in \S(\R^{2n},L(E_z))$, the linear application $f\mapsto \langle\Op_\Ga S_{m,w_1}(f),u\rangle$ 
is continuous.
\end{lem}
\begin{proof} We fix $N\in \N^*$. First note that $S_{m,w_1}(f)$ is well-defined since for any $(\rx,\zeta)\in \R^{2n}$, there is $C_{\rx,\zeta}>0$ such that $\norm{\, ^t M_{\vth'}^{p_{m,w_1},\zeta'}(f)(\rx,\zeta,\zeta',\vth')(1-\chi_{\delta,\eta})(\rx,\zeta,\zeta')}\leq C_{\rx,\zeta} \langle \vth'\rangle^{-2n}\langle \zeta'\rangle^{-2n}$. Since for any $n$-multi-index $\delta$, $\del^\delta_{\vth'}\, ^t M_{\vth'}^{p_{m,w_1},\zeta'}(f)$ decrease to zero with $\vth'$, we can successively integrate by parts with (\ref{Mformula}), which is valid since $1-\chi_{\delta,\eta}$ assures that $\norm{\zeta'}\geq \half\delta$ on the domain of integration. We obtain thus for any $q\in \N^*$,  
$$
S_{m,w_1}(f):(\rx,\zeta,\vth)\mapsto \int_{\R^{2n}} e^{2\pi i (\langle \vth',\zeta'\rangle+ \langle \vth,s_{\rx,\zeta}(\zeta')\rangle)}\, ^t M_{\vth'}^{p_{m,w_1}+q,\zeta'}(f) (1-\chi_{\delta,\eta})\, d\vth'\,d\zeta'\, .
$$
We note $f_q$ the integrand of the previous integral. If $\nu=(\a,\b,\ga)=(\mu,\ga)$ is a $3n$-multi-index, we see with Lemma \ref{derivphase} that
\begin{align*}
&\del^\nu_{\rx,\zeta,\vth} f_q = e^{2\pi i \langle \vth',\zeta'\rangle}\sum_{\mu'\leq \mu} \tbinom{\mu}{\mu'} e^{2\pi i \langle \vth,s_{\rx,\zeta}(\zeta')\rangle} \sum_{|\om|\leq |\mu'|} \vth^\om T_{\nu',\om,s}(\rx,\zeta,\zeta')\\ 
&\hspace{2cm} \sum_{|\wt\delta|=p_{m,w}+q} \la_\delta (-1)^{|\wt\delta|} \tfrac{\zeta'^{\wt\delta}}{\norm{\zeta'}^{2(p_{m,w_1}+q)}} \del^{\mu-\mu'}_{\rx,\zeta} \del^{\wt\delta}_{\vth'} (f (1-\chi_{\delta,\eta}))\, .
\end{align*}
By Lemma \ref{lemchi} $(iii)$, $(\rx,\zeta,\zeta',\vth')\mapsto \chi_{\delta,\eta}(\rx,\zeta,\zeta') \, 1_{L(E_z)}$ is in $\wt \Pi_{\sg,\ka'_\eta,1,z}^{0,0,0,0}$, so the multiplication operator $f\mapsto f(1-\chi_{\delta,\eta})$ is continuous from $\wt \Pi_{\sg,\ka,\eps_1,z}^{l,w_0,w_1,m}$ into $\wt \Pi_{\sg,\ka_\eta,\eps_1,z}^{l,w_0,w_1,m}$, where $\ka_\eta=\max\set{\ka,\ka'_\eta}$.
Since $\norm{\zeta'}\geq \delta/2$ in the support of $f(1-\chi_{\delta,\eta})$, we get from Lemma \ref{phasecontrol} $(i)$ the following estimates, where $\ka''_\eta:=\ka_v+w_s+\ka_\eta$,
\begin{align*}
&\norm{\del^\nu_{\rx,\zeta,\vth} f_q(\rx,\zeta,\vth,\zeta',\vth')} \leq C_{\nu,q} \langle \vth\rangle^{|\mu|} \langle \vth'\rangle^{m-p_{m,w_1}-q} \sum_{\mu'\leq \mu}\langle \zeta\rangle^{\ka_v|\mu'|+w_0+\ka_\eta|\mu|}\\&\hspace{5cm}\times \langle \zeta'\rangle^{w_1+(\ka_v+w_s)|\mu'|+\ka_\eta|\mu|-(p_{m,w_1}+q) +w_s |\ga|} \langle\rx \rangle^{\sg|l|}\\
&\hspace{4cm} \leq C'_{\nu,q}  \langle \rx\rangle^{\sg|l|} \langle \zeta\rangle^{w_0+\ka''_\eta|\mu|} \langle \vth\rangle^{|\mu|} \langle \vth'\rangle^{m-p_{m,w_1}-q}  \langle \zeta'\rangle^{w_1+\ka''_\eta|\nu|-p_{m,w_1}-q}\, .
\end{align*}
If $k\in \N^*$, and if we set $q:=q_k$ such that $w_1+\ka''_\eta k-p_{m,w_1}-q_k\leq -2n$, we see by applying the theorem of derivation under the integral sign that $S_{m,w}(f)$ is smooth and for any $3n$-multi-index $\nu=(\a,\b,\ga)$ and $q\in \N^*$, after integrations by parts in $\vth'$, with $\nu':=(\mu',\ga)$,
\begin{align*}
&\del^\nu S_{m,w_1}(f)(\rx,\zeta,\vth)=\sum_{\mu'\leq \mu}\sum_{|\om|\leq |\mu'|} \tbinom{\mu}{\mu'} \vth^\om \int_{\R^{2n}} e^{2\pi i (\langle \vth',\zeta'\rangle + \langle \vth,s_{\rx,\zeta}(\zeta')\rangle)} T_{\nu',\om,s}(\rx,\zeta,\zeta')\\
&\hspace{5cm} ^tM_{\vth'}^{p_{m,w_1}+q_{|\nu|}+q,\zeta'} \del^{\mu-\mu'}_{\rx,\zeta} (f (1-\chi_{\delta,\eta}))\, d\vth'\,d\zeta'\, .
\end{align*}
We note $g_q(\rx,\zeta,\zeta',\vth'):=e^{2\pi i \langle \vth',\zeta'\rangle}T_{\nu',\om,s}(\rx,\zeta,\zeta') ^tM_{\vth'}^{p_{m,w_1}+q_{|\nu|}+q,\zeta'} \del^{\mu-\mu'}_{\rx,\zeta} (f (1-\chi_{\delta,\eta}))$. Using now Lemma \ref{lem-hL}, we get the estimates for any $p\in \N$,
\begin{align*}
&\norm{(L_{\zeta'}h)^p g_q(\rx,\zeta,\zeta',\vth')}\leq C_p  \langle \zeta'\rangle^{2p\ka_L}  \langle \zeta\rangle^{2p\ka_L}\langle\vth\rangle^{-2p} \sum_{k=1}^{N_p} \norm{\del_{\zeta'}^{\b^{k,p}} g_q(\rx,\zeta,\zeta',\vth')} \, .
\end{align*}
Thus, with
Lemma \ref{phasecontrol} $(i)$, we obtain with $k_1:=w_s+\ka_v+\ka_\eta+\ka_L$,
\begin{align*}
&\norm{(L_{\zeta'}h)^p g_q(\rx,\zeta,\zeta',\vth')}\leq C'_p \langle \rx\rangle^{\sg |l|} \langle \zeta'\rangle^{w_1+(2p+|\nu|)k_1-p_{m,w_1}-q_{|\nu|}-q} \langle \vth\rangle^{-2p}\\ &\langle \vth'\rangle^{2p +m-p_{m,w_1}-q_{|\nu|}-q} 
\langle \zeta\rangle^{(2p+|\mu|)k_1 +w_0}\sum_{|\wt \b|\leq 2p}\sum_{\mu'\leq \mu} \sum_{|\wt\delta|=p_{m,w_1}+q_{|\nu|}+q} q_{\mu',\wt\b,\wt\delta}(f(1-\chi_{\delta,\eta}))\, 1_{D}(\rx,\zeta,\zeta')
\end{align*}
where $D:=\set{(\rx,\zeta,\zeta')\in \R^{2n}\ | \ \norm{\zeta'}\geq \half \delta \langle \rx\rangle^{\sg\eta_1}\langle\zeta\rangle^{-\eta_2}}$.
If we now fix $p$ such that $-N-2\leq -2p+|\mu|\leq -N$, we see that by taking $q$ such that $A_q\leq -N/\eta_1 -|l|/\eta_1$ where $A_q:=w_1+(2p+|\nu|)k_1-p_{m,w_1}-q_{|\nu|}-q+2n$, and $2p+m-p_{m,w_1}-q_{|\nu|}-q\leq -2n$, we can successively integrate by parts in $\zeta'$ ($p$ times) using the formula of Lemma \ref{lem-hL}. We obtain then the estimate for given constants $c_0, c_1, c_2 >0$, 
\begin{align*}
&\norm{\del^\nu S_{m,w_1}(f)(\rx,\zeta,\vth)}\leq C_{\nu,N} \langle \rx\rangle^{-\sg N}\langle \zeta\rangle^{c_0+ c_1 N +c_2 |\mu|} \langle \vth\rangle^{-N}\\&\hspace{4cm}\sum_{|\wt \b|\leq 2p}\,
\sum_{\mu'\leq \mu}\, \sum_{|\wt \delta|=p_{m,w_1}+q_{|\nu|}+q} q_{\mu',\wt\b,\wt \delta}(f(1-\chi_{\delta,\eta}))
\end{align*}
which yields the result.

\noindent $(ii)$ This statement follows from $(i)$ and Lemma \ref{ampliOP} $(ii)$.
\end{proof}

\begin{lem}
\label{Pi-amp}
Suppose $(C_\sg)$. 

\noindent (i) Defining for any $f\in \wt \Pi_{\sg,\ka,\eps_1,z}^{l,w_0,w_1,m}$,
$$
\Pi(f):(\rx,\zeta,\vth)\mapsto \int_{\R^{2n}} e^{2\pi i (\langle \vth',\zeta'\rangle+\langle \vth,\varphi_{\rx,\zeta}(\zeta')\rangle)} f(\rx,\zeta,\zeta',\vth'+L_{\rx,\zeta}(\vth))\,\chi_{\delta,\eta}(\rx,\zeta,\zeta')\,d\zeta'\,d\vth'\, ,
$$
there is $\delta,\eta,$ such that for any $N\geq |m|$, we have $\Pi(f)=\Pi_N(f) +\Pi_{R,N}(f)$ where 
$
 \Pi_{N}(f) = \sum_{0\leq |\b|\leq N} \tfrac{(i/2\pi)^{|\b|}}{\b!} f_{\b,\varphi}
$
and there is such that $\Pi_{R,N}(f)$ satisfies the estimates for any $3n$-multi-index $\nu=(\mu,\ga)\in \N^{2n}\times \N^n$,
$$
\del^\nu \Pi_{R,N}(f) = \O( \langle \rx\rangle^{\sg(l-\eps'_1(N+1))} \langle \zeta\rangle^{k_0+k_1(N+1+|\mu|) +\eps_v|\ga|} \langle \vth\rangle^{m+|\mu|-(N+1)/2 +n})
$$
where $\eps'_1,k_0,k_1>0$.

\noindent (ii) We have for any $3n$-multi-index $\nu=(\mu,\ga)\in \N^{2n}\times \N^n$, 
$$
\del^\nu \Pi(f) = \O( \langle \rx\rangle^{\sg l} \langle \zeta\rangle^{k'_0+k'_1|\mu| +\eps_v|\ga|} \langle \vth\rangle^{m})
$$
where $k'_0,k'_1>0$. In particular, for any $u\in \S(\R^{2n},L(E_z))$, the linear application $f\mapsto \langle\Op_\Ga\Pi(f),u\rangle$ 
is continuous.
\end{lem}
\begin{proof} 
$(i)$
We proceed to a Taylor expansion of $\wt f(\rx,\zeta,\zeta',\vth',\vth):=f(\rx,\zeta,\zeta',\vth'+L_{\rx,\zeta}(\vth))$ in $\vth'$ around zero at order $N\in \N^*$, so that 
$$
\Pi(f)= \sum_{0\leq |\b|\leq N} \tfrac{1}{\b!} I_{\b}(f) + \sum_{|\b|=N+1} \tfrac{N+1}{\b!} R_{\b,N}(f)=:\Pi_N(f)+\Pi_{R,N}(f)
$$
where 
\begin{align*}
&I_\b(f)= \int_{\R^{2n}} \vth'^\b e^{2\pi i (\langle \vth',\zeta'\rangle+\langle \vth,\varphi_{\rx,\zeta}(\zeta')\rangle)} \del^{0,0,0,\b}f(\rx,\zeta,\zeta',L_{\rx,\zeta}(\vth))\,\chi_{\delta,\eta}(\rx,\zeta,\zeta')\,d\zeta'\,d\vth'\, ,\\
&R_{\b,N}(f)= \int_{\R^{2n}} \vth'^\b e^{2\pi i (\langle \vth',\zeta'\rangle+\langle \vth,\varphi_{\rx,\zeta}(\zeta')\rangle)} r_{\b,N,f}(\rx,\zeta,\zeta',\vth',\vth)\,d\zeta'\,d\vth'\, ,
\end{align*}
and $r_{\b,N,f}:=\int_{0}^1(1-t)^N\del^{0,0,0,\b} f_\chi(\rx,\zeta,\zeta',t\vth'+L_{\rx,\zeta}(\vth))\,dt$, $f_\chi:=f\chi_{\delta,\eta}\in\wt \Pi_{\sg,\ka_\eta,z}^{l,w_0,w_1,m}$.
By integration by parts in $\zeta'$ in the integrals $I_\b(f)$, we get
$$
\Pi_{N}(f)= \sum_{0\leq |\b|\leq N} \tfrac{(i/2\pi)^{|\b|}}{\b!} \del^{\b}_{\zeta'}\big(e^{2\pi i \langle \vth,\varphi_{\rx,\zeta}(\zeta')\rangle} \del^{0,0,0,\b}f(\rx,\zeta,\zeta',L_{\rx,\zeta}(\vth))\big)_{\zeta'=0}=\sum_{0\leq |\b|\leq N} \tfrac{(i/2\pi)^{|\b|}}{\b!} f_{\b,\varphi}\, .
$$
Using integration by parts in $\zeta'$, we obtain $R_{\b,N,f}=(i/2\pi)^{|\b|} I_f$, where for any $p\in \N$, 
\begin{align*}
&I_f(\rx,\zeta,\vth):= \int_{\R^{2n}} e^{2\pi i \langle\vth',\zeta' \rangle} \del_{\zeta'}^\b G(\rx,\zeta,\zeta',\vth',\vth)\,d\zeta'\,d\vth'\,, \\
&G(\rx,\zeta,\zeta',\vth',\vth):= e^{2\pi i \langle \vth,\varphi_{\rx,\zeta}(\zeta')\rangle} r_{\b,N,f}(\rx,\zeta,\zeta',\vth',\vth) \, .
\end{align*}
Using integration by parts in $\zeta'$ and $e^{2\pi i \langle\vth',\zeta' \rangle} =\langle \vth'\rangle^{-2p}L_{\zeta'}^p e^{2\pi i \langle\vth',\zeta' \rangle}$, we check that $I_f$ is smooth on $\R^{3n}$ and if $\nu$ is a $3n$-multi-index, we see that $\del^\nu I_f$ is a linear combination of terms of the form
$$
J_f:= \vth^{\wt\om} \int_{\R^{2n}}  e^{2\pi i (\langle \vth',\zeta'\rangle+\langle \vth,\varphi_{\rx,\zeta}(\zeta')\rangle)} \del^{\b^1}_{\zeta'} T_{\nu',\wt \om,\varphi} P_{\b^2,\varphi} \del^{\nu-\nu'}_{\rx,\zeta,\vth} \del^{\b^3}_{\zeta'} r_{\b,N,f}\, d\zeta'\, d\vth'
$$
where $|\wt \om|\leq |\mu'|$, $\nu'\leq \nu$, $\sum \b^i =\b$, $|\b|=N+1$. We now cut the integral $J_f$ in two parts $J_{\chi}+ J_{1-\chi}$, where the cut-off function $\chi_{\eps}(\vth,\vth')$ appears in $J_{\chi}$.

\noindent \emph{Analysis of $J_{\chi}$}

Using Lemma \ref{phasecontrol} $(ii)$ and integration by parts in $\zeta'$, we see that $J_\chi$ is a linear combination of terms of the form 
$$
J_{\chi,\om} = \vth^{\wt\om}\vth^\om \int_{\R^{2n}}e^{2\pi i (\langle \vth',\zeta'\rangle+\langle \vth,\varphi_{\rx,\zeta}(\zeta')\rangle)}\langle \zeta'\rangle^{-2p}
t_{\om,\la}\, \del_\zeta'^{\b^1} T_{\nu',\wt \om,\varphi} \, \del_{\vth'}^{\la'}\del_{\rx,\zeta,\vth}^{\nu-\nu'} \del^{\b^3} r_{\b,N,f} \, \del^{\la+\rho-\la'}\chi_\eps\, d\zeta'\,d\vth' 
$$ 
where $p\in \N$, $|\rho|\leq 2p$, $|\om|\leq |\b^2|$, $(2|\om|-|\b^2|)_+\leq |\la|\leq |\om|$, $\la'\leq \la+\rho$. We now fix $\eps$ such that $\eps<c/2$ where $c$ is a constant such that $c\langle \vth\rangle\leq \langle L_{\rx,\zeta}(\vth)\rangle$. Thus, in the domain of integration of $J_{\chi,\om}$, we have for any $t\in [0,1]$, $\langle t\vth' + L_{\rx,\zeta}(\vth)\rangle \geq c_1 \langle \vth \rangle$ for a $c_1>0$.
As a consequence, we obtain the following estimate:
\begin{align*}
&\norm{\del_{\vth'}^{\la'}\del_{\rx,\zeta,\vth}^{\nu-\nu'} \del^{\b^3}_{\zeta'} r_{\b,N,f}}\leq C\langle \rx\rangle^{\sg(l-\eps_1|\b^3|)} \langle \zeta\rangle^{(\ka_v+\ka_\eta)|\mu-\mu'|+w_0+\ka_\eta|\b^3|}\\
&\hspace{4cm} \langle \zeta'\rangle^{w_1+\ka_\eta(|\mu-\mu'|+|\b^3|)} \langle \vth\rangle^{|\mu-\mu'|+m-|\b|-|\la'|}\, .
\end{align*}
We also deduce from Lemma \ref{phasecontrol} the estimate
$$
|t_{\om,\la}\, \del_\zeta'^{\b^1} T_{\nu',\wt \om,\varphi} |\leq C'\langle \rx\rangle^{-\sg(|\mu'|+(\eps/2)|\b^1+\b^2|)} \langle \zeta\rangle^{2\eps_v|\b^1+\b^2|+(\ka_v+\eps_v)|\mu|+\eps_v|\ga|} \langle \zeta'\rangle^{c_1(N+1)+c_2|\nu|}\, .
$$
As a consequence, by taking $p$ sufficiently big, the integrand $j(\rx,\zeta,\zeta',\vth,\vth')$ of $J_{\chi,\om}$ satisfies the estimate, for a $\eps'_1>0$ and a $k_1>0$,
$$
\norm{j}\leq C'' \langle \rx\rangle^{\sg(l-\eps'_1(N+1))} \langle \zeta\rangle^{w_0+k_1(N+1+|\mu|) +\eps_v|\ga|} \langle \zeta'\rangle^{-2n} \langle \vth\rangle^{m+|\mu|-(N+1)/2}\,  1_{D_\eps}(\vth,\vth')
$$
where $D_\eps$ is the set of $(\vth,\vth')$ in $\R^{2n}$ such that $\norm{\vth'}\leq \eps \langle \vth\rangle$.
We deduce finally that for any $\nu \in \N^{3n}$, 
$$
J_{\chi} = \O( \langle \rx\rangle^{\sg(l-\eps'_1(N+1))} \langle \zeta\rangle^{w_0+k_1(N+1+|\mu|) +\eps_v|\ga|} \langle \vth\rangle^{m+|\mu|-(N+1)/2 +n})\, .
$$

\noindent \emph{Analysis of $J_{1-\chi}$}

We set $\om:=\langle \zeta',\vth'\rangle+\langle \vth,\varphi_{\rx,\zeta}(\zeta')\rangle$. 
By Lemma \ref{Csgr} $(i)$, we have $\sum_i\norm{\del_{\zeta'_i} \varphi_{\rx,\zeta}(\zeta')} \leq C\langle\rx \rangle^{-\sg\eps_v } \langle \zeta\rangle^{c_1} \langle \zeta'\rangle^{c_2})$ for $C,c_1,c_2>0$. The presence of $\chi_{\delta,\eta}$ in the integrand of $J_{1-\chi}$ allows to use the estimate $\langle \zeta'\rangle\leq \sqrt{2} \delta \langle \rx\rangle^{\sg\eta_1} \langle \zeta\rangle^{-\eta_2}$, so that $\sum_i\norm{\del_{\zeta'_i} \varphi_{\rx,\zeta}(\zeta')} \leq C\, 2^{c_2/2}\, \delta^{c_2}$
by taking $\eta_1\leq \eps_v/c_2$ and $\eta_2\geq c_1/c_2$.
As a consequence, we obtain the following estimate in the domain of integration of $J_{1-\chi}$,
$$
|\nabla_{\zeta'} \om|^2 \geq \norm{\vth'}^2 (1- \tfrac{4}{\eps}\,C\, 2^{c_2/2}\, \delta^{c_2})\, . 
$$
We now fix $\delta$ such that $\tfrac{4}{\eps}\,C\, 2^{c_2/2}\, \delta^{c_2}<1$ so that there is $k>0$ such that $|\nabla_{\zeta'} \om| \geq k \norm{\vth'}$. Noting $U_{\zeta'}:=(2\pi i |\nabla_{\zeta'} \om|^{2})^{-1}\sum_i (\del_{\zeta'_i} \om ) \del_{\zeta'_i}$ we have (see for instance \cite{Ruzhansky}) $U_{\zeta'} e^{2\pi i \om} = e^{2\pi i \om}$ and 
$$
(^t U_{\zeta'})^r = \tfrac{1}{ |\nabla_{\zeta'} \om|^{4r} } \sum_{|\rho|\leq r} P^{\om}_{\rho,r} \del^\rho_{\zeta'}
$$
where $P_{\rho,r}^\om$ is a linear combination of terms of the form $(\nabla_{\zeta'} \om)^\pi \del^{\delta^1}_{\zeta'} \om \cdots\del^{\delta^r}_{\zeta'} \om$, with $|\pi|=2r$, $|\delta^{i}|>0$ and $\sum_{j=1}^r |\delta^j| + |\rho| =2r$. We thus obtain after integration by parts in $\zeta'$, for any $r\in \N^*$, that $J_{1-\chi}$ is a linear combination of integrals of the form
$$
 \vth^{\wt\om+\wh \om} \int_{\R^{2n}}  e^{2\pi i \om} (^t U_{\zeta'})^r\big(\del^{\b^1}_{\zeta'} T_{\nu',\wt \om,\varphi} P_{\wh\om,\b^2,\varphi} \del^{\nu-\nu'}_{\rx,\zeta,\vth} \del^{\b^3}_{\zeta'} r_{\b,N,f}\,\big)(1-\chi_\eps) d\zeta'\, d\vth'
$$
where $|\wh \om|\leq|\b^2|$. We noted $P_{\b^2,\varphi}=:\sum_{\wh \om} P_{\wh \om,\b^2,\varphi} \vth^{\wh \om}$. By Lemma \ref{phasecontrol} $(ii)$, we see that $P_{\wh \om,\b^2,\varphi} \in \O_{\sg,\ka_v,\eps_v,\eps_v,2|\b^2|}^{-\eps_v|\b^2|/2,2\eps_v|\b^2|,(w'_s+1)|\b^2|}$. Let us note $\wt T:=\del^{\b^1}_{\zeta'} T_{\nu',\wt \om,\varphi} P_{\wh\om,\b^2,\varphi}$. Lemma \ref{phasecontrol} $(i)$ yields $\wt T \in \O_{\sg,\ka_v,\eps_v,\eps_v,2(|\nu|+N)}^{-(\eps_v/2)|\b^1+\b^2|,c_0(|\mu|+N)+\eps_v|\ga|,c_0(|\nu|+N)}(\R)$ for a constant $c_0>0$.
With our choice of the parameters $\eta_1$ and $\eta_2$, we also have the following estimate, valid in the domain of integration of $J_{1-\chi}$, 
$$
\del^{\la}_{\vth'} \del^{\ga+e_i}_{\zeta'} \om  = \O\big( \langle \zeta\rangle^{\eps_v|\ga|} \langle \zeta'\rangle^{\ka_v|\ga|} \langle \vth'\rangle^{1-|\la|} \big)\, .
$$
In particular, noting $\O_{\ka_v}^{l,m}$ the space of smooth functions $f$ such that for any $n$-multi-indices $\la,\ga$, $\del^{\la}_{\vth'}\del^{\ga}_{\zeta'} f =\O\big( (\langle \zeta\rangle\langle \zeta'\rangle)^{l+\ka_v|\ga|} \langle \vth'\rangle^{m}\big)$, we see that $|\nabla_{\zeta'} \om|^2 \in \O_{\ka_v}^{0,2}$, and for any $\la\in \N^n$, $\del^\la_{\vth'} |\nabla_{\zeta'}\om|^{-4r}$ = $\O(\langle \vth'\rangle^{-4r})$. Moreover, each term $P_{\rho,r}^\om$ is in $\O_{\ka_v}^{\ka_v r , 3r}$ so that finally, for any $\la\in \N^n$ 
$$
\del^{\la}_{\vth'} \tfrac{P^{\om}_{\rho,r}}{ |\nabla_{\zeta'} \om|^{4r} }   = \O\big( (\langle \zeta\rangle\langle \zeta'\rangle)^{\ka_v r } \langle \vth'\rangle^{-r}\big)\, .
$$

We easily check that if $r\geq 2n$, then $h:=(^t U_{\zeta'})^r\big(\del^{\b^1}_{\zeta'} \wt T \del^{\nu-\nu'}_{\rx,\zeta,\vth} \del^{\b^3}_{\zeta'} r_{\b,N,f}\,\big)(1-\chi_\eps) $ satisfies the estimates for any $q\in \N$, $\norm{L_{\vth'}^q h }\leq C_{\rx,\zeta,\zeta',\vth,q} \langle \vth'\rangle^{-2n}$. As a consequence, we can permute the integration $d\zeta'd\vth'\to d\vth' d\zeta'$ and successively integrate by parts in $\vth'$, so that finally $J_{1-\chi}$ is a linear combination of terms of the form
$$
 \vth^{\wt\om+\wh \om} \int_{\R^{2n}}  e^{2\pi i \om} \langle \zeta'\rangle^{-2q} \del_{\vth'}^{\la^1}\tfrac{P^{\om}_{\rho,r}}{ |\nabla_{\zeta'} \om|^{4r} } \, \del^{\rho^1}_{\zeta'}\wt T\, \del_{\vth'}^{\la^2}\del^{\nu-\nu'}_{\rx,\zeta,\vth} \del^{\b^3+\rho^2}_{\zeta'} r_{\b,N,f}\,\del_{\vth'}^{\la^3}(1-\chi_\eps) d\vth'\,d\zeta' 
$$
where $\sum_{i}\la^i = \la$, $|\la|\leq 2q$, $\sum_i \rho^i = \rho$, $|\rho|\leq r$. We also have the following estimate for $c'_0, c'_1>0$,
$$
\del_{\vth'}^{\la^2}\del^{\nu-\nu'}_{\rx,\zeta,\vth} \del^{\b^3+\rho^2}_{\zeta'} r_{\b,N,f}= \O\big(\langle \rx\rangle^{\sg(l-|\b^3|)}(\langle \zeta\rangle\langle \zeta\rangle)^{c'_0+c'_1(|\mu-\mu'|+|\b^3|+|\rho^2|)}\big)\, .
$$
With Lemma \ref{lemchi} $(iv)$ we now see that the integrand $j'$ of the previous integral is estimated by
$$
\norm{j'} \leq C \langle \vth' \rangle^{-r+|\mu|+N+1}\langle \rx\rangle^{\sg(l-\eps'_1(N+1))} \langle \zeta\rangle^{k_0+k_1 N+k_2 r +k_3|\mu|+\eps_v |\ga|}\langle \zeta'\rangle^{-2q +k_0+k_1N +k_2r+k_3|\nu|} 
$$
for constants $k_0,k_1,k_2,k_3>0$. If we now fix $r\geq 2n$ such that $-r+|\mu|+N+1 +2n = m+|\mu|-(N+1) +n$, and $q$ such that $-2q +k_0+k_1N +k_2r+k_3|\nu|\leq -2n$
we finally obtain the estimate
$\nu \in \N^{3n}$, 
$$
J_{1-\chi} = \O( \langle \rx\rangle^{\sg(l-\eps'_1(N+1))} \langle \zeta\rangle^{k'_0+k'_1(N+1+|\mu|) +\eps_v|\ga|} \langle \vth\rangle^{m+|\mu|-(N+1) +n})\, .
$$
The result follows now from this estimate and the one obtained for $J_{\chi}$.

\noindent $(ii)$ The estimate is obtained by applying $(i)$ and $N+1= \max\set{ 2(n+|\mu|),|m|}$. The second statement is then a consequence of Lemma \ref{ampliOP} $(ii)$.
\end{proof}

\begin{thm} 
\label{compo} If $(C_\sg)$ holds, 
$\Psi_{\sigma}^\infty$ is a $*$-subalgebra of $\Re(\S)$. Moreover, if $A \in \Psi_{\sigma}^{l',m'}$ and $B\in \Psi_{\sigma}^{l,m}$, then $AB\in \Psi_{\sigma}^{l+l',m+m'}$ with the following asymptotic expansion of the normal symbol of $AB$, in a frame $(z,\bfr)$:
$$
\sigma_{0}(AB)_{z,\bfr} \sim \sum_{\b,\ga \in \N^n} c_\b c_\ga \del_{\zeta,\vth}^{\ga,\ga}\big( a(\rx,\vth)\del^{\b}_{\zeta'}\big(e^{2\pi i \langle \vth,\varphi_{\rx,\zeta}(\zeta')\rangle} (\del^{\b}_{\vth'} f_b)(\rx,\zeta,\zeta',L_{\rx,\zeta}(\vth))\big)_{\zeta'=0} \tau^{-1}_{\rx,\zeta} \big)_{\zeta=0}
$$
where  $a:=\sigma_0(A)_{z,\bfr}$, $b:=\sigma_0(B)_{z,\bfr}$, $c_\b:=(i/2\pi)^{|\b|}/\b!$ and 
$$
f_b(\rx,\zeta,\zeta',\vth'):=\tau_{\rx,r_{\rx,\zeta}(\zeta')}\,b\circ \wt \Xi(\rx,\zeta,\zeta',\vth')\,\tau_{\rx^{\zeta,\zeta'},q_{\rx,\zeta}(\zeta')}\,|J(R)|(\rx,\zeta,\zeta')\,|\det (P^{z,\bfr}_{-1,\psi(\rx,\zeta),\zeta'})^{-1}|\, .
$$

\end{thm} 
\begin{proof} We fix a frame $(z,\bfr)$. We note $K_{AB}$ the kernel of the operator $AB$. As a consequence of Proposition \ref{regularity} we have for any $u,v\in \S(\R^n,E_z)$, $\langle(K_{AB})_{z,\bfr},u\ox \ol v \rangle=\big( A_{z,\bfr} (\mu^{-1} B_{z,\bfr}(v))| u\big)$. We shall note $g:=A_{z,\bfr} (\mu^{-1} B_{z,\bfr}(v))$. A computation shows that for any $\rx\in \R^n$, 
$g(\rx)= \int_{\R^n} \mu a(\rx,\vth) \,\wt b(\rx,\vth)\, d\vth$, and
$$
\wt b(\rx,\vth):= \int_{\R^{3n}} e^{2\pi i (\langle \vth,\zeta\rangle +\langle \vth',\zeta'\rangle)} \tau_{\rx,\zeta} b(\psi(\rx,\zeta),\vth')\tau_{\psi(\rx,\zeta),\zeta'} v(\rx^{\zeta,\zeta'})\,d\zeta'\,d\vth'\,d\zeta\, .
$$
We suppose at first that $b\in S^{l,-2n}_{\sg,z}$. Since $\zeta'\mapsto v(\rx^{\zeta,\zeta'})\in \S(\R^n,E_z)$, we can permute the order integration $d\zeta' d\vth' \mapsto d\vth'\,d\zeta'$ in $\wt b(\rx,\vth)$. Thus, after integrations by parts in $\vth'$, we get
for any $p\in \N^*$, 
$$
\wt b(\rx,\vth)=\int_{\R^{2n}} e^{2\pi i \langle \vth,\zeta\rangle} \tau_{\rx,\zeta}\, \big(\int_{\R^n} e^{2\pi i \langle \vth',\zeta'\rangle}\langle \zeta'\rangle^{-2p}(L_{\vth'}^p b)(\psi(\rx,\zeta),\vth')\,d\vth'\big)\, \tau_{\psi(\rx,\zeta),\zeta'}\,v(\rx^{\zeta,\zeta'})\,d\zeta'\,d\zeta\, .
$$
With the estimate $\langle \rx^{\zeta,\zeta'}\rangle \geq c \langle \zeta\rangle\langle \rx\rangle^{-1} \langle \zeta'\rangle^{-1}$ for a $c>0$, we see that for any $N\in \N$, $\norm{v(\rx^{\zeta,\zeta'})}\leq c_N q_{0,N}(v)\langle \rx\rangle^{N} \langle \zeta'\rangle^N\langle \zeta\rangle^{-N}$. 
As a consequence, we get the following estimates for the integrands $b_p$ of $\wt b(\rx,\vth)$: for any $\rx,\zeta,\zeta',\vth,\vth'$, any $p\in \N^*$ and any $N\in \N^*$, $\norm{b_p(\rx,\zeta,\zeta',\vth,\vth')}\leq C_{p,N} \langle \zeta'\rangle^{N-2p} \langle \rx\rangle^{\sg|l|+N} \langle \zeta\rangle^{\sg|l|-N}\langle \vth'\rangle^{-2n}$. Taking $N$ such that $\sg|l|-N\leq -2n$ and then taking $p$ such that $N-2p \leq -2n$, we see that $(\vth',\zeta',\zeta)\mapsto b_p(\rx,\zeta,\zeta',\vth',\vth)$ is absolutely integrable and we can thus apply the following change of variable $(\zeta,\zeta',\vth')\mapsto (R_\rx(\zeta,\zeta'),\vth')$ to $\wt b(\rx,\vth)$. After reversing the integration by parts in $\vth'$ and applying the change of variable $\vth'=-\wt P_{-1,\psi(\rx,\zeta),\zeta'}^{z,\bfr}(\vth'')$, we get  
$$
\wt b(\rx,\vth)= \int_{\R^{3n}} e^{2\pi i( \langle \vth,r_{\rx,\zeta}(\zeta')\rangle+\langle \vth',\zeta'\rangle)} f_b(\rx,\zeta,\zeta',\vth')\,v(\psi(\rx,\zeta))\,d\vth'\,d\zeta'\,d\zeta\, .
$$
By Lemma \ref{amptilde} $(ii)$ and $(iii)$, Lemma \ref{htilde} $(iii)$ and $(iv)$ and Lemma \ref{Csgr} $(iii)$, we see that $f_b\in \wt \Pi^{l,w_l,w_l,m}_{\sg,\ka,\eps_1,z}$ for a $(w_l,\ka)\in \R^2_+$ and $\eps_1>0$, and the linear application $b\mapsto f_b$ is continuous on any symbol space $S^{l,m}_{\sg,z}$ into $\wt \Pi^{l,w_l,w_l,m}_{\sg,\ka,\eps_1,z}$. We have $g(\rx)= \int_{\R^n}e^{2\pi i \langle \zeta,\vth\rangle} \mu a(\rx,\vth) \, c_b(\rx,\zeta,\vth) v(\psi(\rx,\zeta)\, d\zeta\,d\vth$ and $\langle (K_{AB})_{z,\bfr},u\ox\ol v\rangle = \langle \Op_{\Ga_{0,z,\bfr}}(d_b),u\ox \ol v\rangle$ where $d_b(\rx,\zeta,\vth):=\mu a(\rx,\vth)\,c_b(\rx,\zeta,\vth)\,\tau^{-1}(\rx,\zeta)$ and 
$$
c_b(\rx,\zeta,\vth):= \int_{\R^{2n}} e^{2\pi i( \langle \vth,s_{\rx,\zeta}(\zeta')\rangle+\langle \vth',\zeta'\rangle)} f_b(\rx,\zeta,\zeta',\vth')\,d\vth'\,d\zeta'\, .
$$ 
Using now the cut-off function $(\rx,\zeta,\zeta')\mapsto \chi_{\delta,\eta}(\rx,\zeta,\zeta')$ we see that 
\begin{align*}
c_b(\rx,\zeta,\vth)&= \Pi(f_b)(\rx,\zeta,\vth) + S_{m,w_l}(f_b)(\rx,\zeta,\vth)\, .
\end{align*}
For this equality, we used the formula of Lemma \ref{Mformula} and integration by parts in $\vth'$ in the integral $\int_{\R^{2n}} e^{2\pi i( \langle \vth,s_{\rx,\zeta}(\zeta')\rangle+\langle \vth',\zeta'\rangle)} f_b(\rx,\zeta,\zeta',\vth')(1-\chi_{\delta,\eta}(\rx,\zeta,\zeta'))\,d\vth'\,d\zeta'\,$, which are authorized since $b\in S^{l,-2n}_{\sg,z}$ by hypothesis. In $\int_{\R^{2n}} e^{2\pi i( \langle \vth,s_{\rx,\zeta}(\zeta')\rangle+\langle \vth',\zeta'\rangle)} f_b(\rx,\zeta,\zeta',\vth')\chi_{\delta,\eta}(\rx,\zeta,\zeta')\,d\vth'\,d\zeta'\,$, we translated the $\vth'$ variable by $-L_{\rx,\zeta}(\vth)$ and permuted the order of integration $d\vth'\,d\zeta'\to d\zeta' \,d\vth'$, which is legal since  $b\in S^{l,-2n}_{\sg,z}$ and $\zeta'\mapsto \chi(\rx,\zeta,\zeta')$ is of compact support. We deduce from Lemma \ref{lemnegamp} $(ii)$ and Lemma \ref{Pi-amp} $(ii)$ that $b\mapsto \langle \Op_{\Ga_{0,z,\bfr}}(d_b),u\ox \ol v\rangle$ is continuous on $S^{l,m}_{\sg,z}$, and thus, by the density result of Lemma \ref{toposymbol}, we have the equality $\langle (K_{AB})_{z,\bfr},u\ox\ol v\rangle = \langle \Op_{\Ga_{0,z,\bfr}}(d_b),u\ox \ol v\rangle$ even when the hypothesis $b\in S^{l,-2n}_{\sg,z}$ does not hold.

Let us recall the linear map $s:a\mapsto s(a)$ given in Lemma \ref{reduction} $(ii)$ (for $\Ga=\Ga_{0,z,\bfr}$) which is such that $\Op_{\Ga_{0,z,\bfr}}(f)=\Op_{\Ga_0,z,\bfr}(s(f))$ for any $f\in \Pi_{\sg,\ka,z}^{l,w,m}$. We define $f_{a,b,\b}:= \mu a (f_b)_{\b,\varphi} \tau^{-1}$, 
$r_N:= \mu a \Pi_{R,N}(f_b)\tau^{-1}$, $s_0:= \mu a S_{m,w_l}(f_b) \tau^{-1}$.
We now consider a symbol $s_{a,b}$ such that 
$$
s_{a,b}\sim \sum_{\b\in \N^n} \tfrac{(i/2\pi)^{|\b|}}{\b!} s\big( f_{a,b,\b}\big)\, .
$$
Such a symbol exists since by Lemma \ref{phasecontrol} $(iii)$, $s(f_{a,b,\b})\in S^{l+l'-\eps'_1|\b|,m+m'-|\b|/2}_{\sg,z}$.
By Lemma \ref{Pi-amp} $(i)$, we have for any $N\geq |m|$, $u_N:=s( \mu a \Pi_N(f_b) \tau^{-1}) -s_{a,b} \in S^{l+l'-\eps'_1(N+1),m+m'-(N+1)/2}_{\sg,z}$. Thus, noting $S_0:=\Op_{\Ga_{0,z,\bfr}}(s_0)$, which is in $\Op_{\Ga_{0,z,\bfr}}(S^{-\infty}_{\sg,z})$ by Lemma \ref{lemnegamp}, $R_N:= \Op_{\Ga_{0,z,\bfr}}(r_N)$ and $U_N:=\Op_{\Ga_{0,z,\bfr}}(u_N)$  we have 
\begin{align*}
(K_{AB})_{z,\bfr}=\Op_{\Ga_{0,z,\bfr}}(d_b) &= \Op_{\Ga_{0,z,\bfr}}(s( \mu a \Pi_N(f_b) \tau^{-1}) ) + R_N + S_0 \\
&= \Op_{\Ga_{0,z,\bfr}}(s_{a,b}) + U_N + R_N +S_0\, .
\end{align*}
Lemma \ref{noyauReste} and Lemma \ref{Pi-amp} $(i)$ now implies that the kernel $U_N$ + $R_N$ (which independent of $N$) is in $\Op_{\Ga_{0,z,\bfr}}(S^{-\infty}_{\sg,z})$. As a consequence, $(K_{AB})_{z,\bfr}=\Op_{\Ga_{0,z,\bfr}}(s_{a,b} + r)$ where $r\in S^{-\infty}_{\sg,z}$ and the symbol product asymptotic formula is entailed by Lemma \ref{reduction} $(ii)$.
\end{proof}

\section{Examples} \label{exsec}

In order to be able to apply the previous results about the pseudodifferential and symbolic calculi on some concrete cases, we shall see in this section examples of exponential manifolds and associated linearizations that satisfy the hypothesis $S_\sigma$-bounded geometry. 
The Euclidean space $\R^n$ seen as exponential manifold, has its own exponential map $\psi:=\exp (x,\xi)\mapsto x+\xi$ as a $S_1$-linearization, leading to the usual pseudodifferential $SG$ calculus (if $\sg=1$) or standard (if $\sg=0$) pseudodifferential calulus on $\R^n$.
However, we can define other kinds of linearization, leading to new kind of pseudodifferential and symbol calculi, with a non-bilinear linearization map. We will see in particular that we can construct on the flat $\R^n$, a family of $S_\sg$-linearizations that generalize the case of the flat euclidian geometry, and we obtain a extension of the normal ($\la=0$) and antinormal ($\la=1$) quantization on $\R^n$. 

We will also prove that the 2-dimensional hyperbolic space, which is a Cartan--Hadamard manifold (and thus an exponential Riemannian manifold) has $S_1$-bounded geometry. This allows to define a global Fourier transform, Schwartz spaces $\S(\HH)$, $\S(T^*\HH)$, $\S(T\HH)$, $\B(\HH)$ and the space of symbols $S_1^{l,m}(T^*\HH)$. As a consequence, we can define in an intrinsic way a global complete pseudodifferential calculus on $\HH$, if one chose a fixed $S_\sg$-linearization $\psi$ on $T\HH$. There are many possible linearizations, for instance one can take $\psi$ such that in a frame $(z,\bfr)$ $\psi_{z}^\bfr$ is the standard linearization $x+\xi$ of $\R^{n}$.

\subsection{A family of $S_\sg$-linearizations on the euclidean space}

Recall that $G_\sigma^\times(\R^n)$ ($0\leq \sigma \leq 1$) is defined as the subgroup of diffeomorphisms $s$ on $\R^n$ such that for any $n$-multi-index $\a\neq 0$, there are $C_\a$, $C'_\a >0$, such that for any $\rx\in \R^n$, $\norm{\del^\a s (\rx)} \leq C_\a \langle \rx \rangle^{\sigma(1-|\a|)}$ and $\norm{\del^\a s^{-1} (\rx)} \leq C'_\a \langle \rx \rangle^{\sigma(1-|\a|)}$. $G_\sigma^\times (\R^n)$ contains $GL_n(\R)$ and the translations $T_v:=w\mapsto v+w$.

We fix $\eta\in ]0,1[$ such that for any matrix $A\in \M_n(\R)$ such that $\norm{A}_1\leq  \eta$, we have $\det (I_n+A) \geq \half$, where $\norm{A}_1:= \max_{i,j}  |A_{i,j}|$. Taking now $h\in G_0(\R^n,\R^n)$ such that for any $1\leq i,j\leq n$, $|\del_j h^i|\leq \eta/16$, and 
$g(x):= h(x) - h(0) - dh_0 (x)$
we see that $s:=\Id +g$ is a diffeomorphism on $\R^n$ which belongs to $G_0^\times (\R^n)$, satisfying $s(0)=0$ and $ds_0 =\Id$.
 
\noindent We set, for $\sg\in [0,1]$,
$$
\psi(x,\xi):=x+\xi + \langle x\rangle^\sg g( \tfrac{\xi}{ \langle x\rangle^\sg} )= x+ \langle x\rangle^\sg s (  \tfrac{\xi}{ \langle x\rangle^\sg} ).
$$
We obtain the following

\begin{prop} $(\R^n,+,d\la,\psi)$ has a $S_\sg$-bounded geometry and satisfies $(C_\sg)$ (see Definition \ref{Csigma}).
\end{prop}
\begin{proof} A computation shows that $\psi \in H_\sg(\R^n)$ and $\psi(\rx,\zeta)=\O(\langle \rx \rangle\langle \xi\rangle)$. We have $\ol \psi(x,y)= \langle x\rangle^\sg s^{-1}(\tfrac{y-x}{\langle x\rangle^\sg})$, and thus $\ol \psi \in \O_M(\R^{2n},\R^n)$. Noting $\wh g:=g\circ (g+\Id)^{-1}\circ -\Id \in G_0(\R^n)$, we also have 
\begin{align*}
\Ups_{1,T}(x,\xi)&= \xi+\langle x\rangle^\sg g(\tfrac{\xi}{\langle x\rangle^\sg}) + \langle \psi(x,\xi)\rangle^\sg \wh g \big(\langle \psi(x,\xi)\rangle^{-\sg}\langle x\rangle^\sg s(\tfrac{\xi}{\langle x\rangle^\sg})\big)\\
&= (\Id + V_{x,\xi} + W_{x,\xi} ) (\xi) 
\end{align*}
where $V_{x,\xi}:= [\int_0^1 \del_j v_x^i (t\xi)dt]_{i,j}$, $W_{x,\xi}:= [\int_0^1 \del_j w_{x,\xi}^i (t\xi)dt]_{i,j}$, and 
$v_x:= M_x \circ g\circ M_x^{-1}$, $w_{x,\xi} := M_{\psi(x,\xi)}\circ  \wh g \circ M_{\psi(x,\xi)}^{-1}\circ M_x \circ s \circ M_x^{-1}$, $M_x$ being the multiplication by $\langle x\rangle^\sg$. We get $dv_x = dg\circ M_x^{-1}$ and $dw_{x,\xi}=d\wh g\circ (M_{\psi(x,\xi)}^{-1}\circ M_x\circ s\circ M_x^{-1}) \, ds\circ M_x^{-1}$.   
and thus, after computations we check that $V_{x,\xi}$ and $W_{x,\xi}$ are in $E_\sg^0$. Moreover, we have $\norm{V_{x,\xi}}_1 \leq \eta/2$ and $\norm{W_{x,\xi}}_1\leq \eta/2$, which proves that $P_{x,\xi}:=\Id + V_{x,\xi} + W_{x,\xi}$ is invertible with $\det P_{x,\xi}\geq \half$. As a consequence its inverse $P_{x,\xi}^{-1}=(\det P_{x,\xi})^{-1} \, ^t \cof(P_{x,\xi})$ is also in $E_\sg^0$. We deduce then that $(\R^n,+,d\la,\psi)$ has a $S_\sg$-bounded geometry. With $r(x,\xi,\xi')=-\ol\psi(x,\psi(\psi(x,-\xi),-\xi'))$, we get
$$
 r(x,\xi,\xi') =-\langle x\rangle^\sg s^{-1}\big( s(\tfrac{-\xi}{\langle x\rangle^\sg})+ \tfrac{\langle \psi(x,-\xi)\rangle^\sg}{\langle x\rangle^\sg} s(\tfrac{-\xi'}{\langle \psi(x,-\xi)\rangle^\sg}) \big)\, .
$$
so that $(dr_{x,\xi})_{\xi'} = (d s^{-1}\circ w )\,( ds \circ u)$ where $w(x,\xi,\xi'):= s(\tfrac{-\xi}{\langle x\rangle^\sg})+ v(x,\xi,\xi')$, $v(x,\xi,\xi):= \tfrac{\langle \psi(x,-\xi)\rangle^\sg}{\langle x\rangle^\sg} s(\tfrac{-\xi'}{\langle \psi(x,-\xi)\rangle^\sg})$, $u(x,\xi,\xi'):= -\tfrac{\xi'}{\langle \psi(x,-\xi)\rangle^\sg}$. 
We check that $v$ satisfies 
$$
\del^{(\mu,\ga)} v =\O\big(\langle \psi(x,-\xi)\rangle^{-\sg |\ga|}\langle x\rangle^{-\sg( |\mu|+1)}\langle \zeta\rangle^{\ka_1|\mu|}\langle \zeta'\rangle^{|\mu|+1}\big).
$$  
It follows from Peetre's inequality that for any $\eps \in [0,1]$ and $x,y\in \R^n$, $\langle x + y\rangle \geq 2^{-\eps/2} \tfrac{\langle x\rangle^\eps}{\langle y\rangle^\eps}$, which implies that $\langle \psi(x,-\xi)\rangle^\sg =\O\big( \langle x\rangle ^{-\sg \eps} \langle \xi\rangle^{\sg\eps}\big)$. As a consequence we get the estimates
\begin{align*}
&\del^{(\mu,\ga)} w = \O\big(\langle x\rangle^{-\sg (1+|\mu|+\eps|\ga|)}\langle \zeta\rangle^{\ka_1|\mu|+ \eps|\ga| + \delta_{\ga,0}}\langle \zeta'\rangle^{|\mu|+1}\big)\, , \\ 
&\del^{(\mu,\ga)} u = \O\big(\langle x\rangle^{-\sg (|\mu|+\eps|\ga|)}\langle \zeta\rangle^{\ka_1|\mu|+ \eps|\ga|}\langle \zeta'\rangle^{1-|\ga|}\big)\, .
\end{align*}
We deduce from this that $(C_\sg)$ is satisfied.
\end{proof}

We also check that the hypothesis $(H_V)$ of section \ref{linkstd} is satisfied so that 
the previous pseudodifferential calculus (for $\la\in \set{0,1}$) is then valid on  $(\R^n,+,d\la,\psi)$, and proves in particular the space of operators of the form
$$
A (v) (x) =\int_{\R^{2n}} e^{2\pi i \langle \th,\xi\rangle } a(x,\th) v(\psi(x,-\xi)) \, d\xi \,d\th = \int_{\R^{2n}} e^{-2\pi i \langle \th,\ol\psi_{x}(y)\rangle } a(x,\th) v(y) |J(\ol\psi_x)|(y)\, dy \,d\th
$$  
where $a\in S^{\infty}_{\sg}(\R^{2n})$, is equal to the standard algebra of pseudodifferential operators $\R^n$. However, since $(C_\sg)$ is satisfied, we have now at our disposal a new symbol composition formula given by Theorem \ref{compo}, adapted to the new linearization $\psi$.

\subsection{$S_1$-geometry of the Hyperbolic plane}

The (hyperboloid model of the) 2-dimensional hyperbolic space is defined as the submanifold $\HH:=\set{x=(x_1,x_2,x_3)\in \R^3 \ : \ x_1^2+x_2^2-x_3^2=-1 \ \text{and}\ x_3>0}$ of the $(2,1)$-Minkowski space $\R^{2,1}$ with the bilinear symmetric form $\langle v,w\rangle_{2,1}=v_1 w_1+v_2 w_2-v_3 w_3$. The induced metric on $\HH$: $ds^2= (dx_1)^2+(dx_2)^2-(dx_3)^2$ is Riemannian and it is known that $\HH$ is a symmetric Cartan--Hadamard manifold with constant negative sectional curvature (equal to $-1$). The map $\varphi : \R^2\to \HH$ given by 
$$
\varphi(x,y):= (\sinh x,\,\cosh x \sinh y,\,\cosh x \cosh y )
$$
is a diffeomorphism with inverse $\varphi^{-1}(x_1,x_2,x_3)=(\argsh x_1,\,\argsh(\tfrac{x_2}{\cosh(\argsh x_1 )}))$. 
 As a consequence we can construct another model of the hyperbolic space, noted $R^2$ with domain $\R^2$ and metric obtained by pulling back the metric on $\HH$ onto $\R^2$. A computation shows that this metric is $ds^2:= (dx)^2+ \cosh^2 x \,(dy)^2$. We will note $\norm{\cdot}_{p}$ the norm on $T_p R^2\simeq \R^2$ given by this metric, where $p$ is a point in $\R^2$, and $\norm{\cdot}$ is the Euclidian norm. The geodesic equation on $R^2$ leads to the following system of ordinary differential equations:
\begin{align}
&x'' - \cosh x\,\sinh x\,(y')^2=0 \, \nonumber,\\
&y'' + 2 \tanh x \,x'\,y' = 0 \, \label{geodeq}. 
\end{align}
For each $p=(x,y)\in \R^2$ and $v\in \R^2$ such that $\norm{v}_p=1$ there exists an unique solution on $\R$ $\ga_{p,v}=(x(t),y(t))$ of (\ref{geodeq}) such that $\ga_{p,v}(0)=p$ and $\ga'_{p,v}(0)=v$. 

At each point $p=(x,y)\in \R^2$, we can define the ellipse of unit vectors centered at $0$ in $T_p R^2\simeq \R^2$ with equation $X^2+(\cosh^2 x) \, Y^2=1$. The polar equation of this ellipse is
$$
e_{p}(\th):= \tfrac{1}{\sqrt{1+\sinh^2 x \, \sin^2 \th}\, }\, .
$$
Thus, any tangent vector $v\in T_p R^2$ with decompostion $v=\norm{v}(\cos \th,\sin \th)$ also admits the following polar decomposition $v=\norm{v}_p(\cos_p \th,\sin_p \th)$ where $\cos_{p} \th :=e_{p}(\th)\,\cos \th$ and $\sin_{p} \th:=e_{p}(\th) \,\sin \th$. Remark that $e_p$, $\cos_p$, $\sin_p$ and $\norm{\cdot}_p$ are in fact independent of the second coordinate $y$ of $p$. We shall therefore also use the notations $e_x:=e_{(x,y)}$ and similarly for $\cos_x$, $\sin_x$ and $\norm{\cdot}_x$. Note that for any vector $v:=\norm{v}(\cos \th,\sin \th)$, we have $\norm{v}_x = \norm{v} / e_x(\th)$.

If $p\in \HH$ and $v\in \R^{2,1}$ are such that $\langle p,v\rangle_{2,1} =0$ and $\langle v,v\rangle_{2,1}=1$, then the unique geodesic $\a_{p,v}$ on $\HH$ such that $\a_{p,v}(0)=p$ and $\a'_{p,v}(0)=v$ is $\a_{p,v}(t)=\cosh t \, p + \sinh t\, v$ (see for instance \cite[p.195]{Jost}). As a consequence, the geodesics $\ga_{p,v}$ on the $R^2$ hyperbolic space can be obtained by pushing forward the $\a_{p,v}$ geodesics with the diffeomorphic isometry $\varphi$. We check after tedious calculations that for any given $p=(x,y)\in \R^2$ and $\th\in \R$, the following curve
\begin{align}
&\ga_{p,\th}^1(t)= \argsh\big( \cosh t \, \sinh x + \sinh t\, \cosh x \, \cos_x \th \big) \, \nonumber ,\\
&\ga_{p,\th}^2(t)=  \argsh \big ( \tfrac{\cosh t \, \cosh x\, \sinh y + \sinh t\,( \sinh x \, \sinh y\, \cos_x \th +\cosh x \, \cosh y\, \sin_x \th  )}{\cosh\big( \argsh ( \cosh t \, \sinh x + \sinh t\, \cosh x \, \cos_x \th ) \big) } \big)\, \label{exponential} ,
\end{align}
where $t\in \R$, is the unique maximal solution of the geodesic system (\ref{geodeq}) satisfying the initial conditions: $\ga_{p,\th}(0)= p$ and $\ga_{p,\th}'(0) = (\cos_x(\th),\sin_x(\th))$. An explicit formula for the exponential map at any point can therefore be obtained, since we have $\exp_p(v)=\ga_{p,\th}(\norm{v}_x)$ where $v\in T_p R^2-\set{0}$ and $\th \in \R$ such that $v=\norm{v}(\cos \th,\sin \th)$. The main interest of this hyperbolic model with domain equal to $\R^2$ is that it is possible to find explicitely the logarithmic map (the inverse of the exponential map) at any point. We find, after an elementary but long computation, the following inverse, for any $p=(x,y)$ and $p'=(x',y')\in \R^2$,
\begin{align}
&\exp_p^{-1}(p') = \tfrac{\argch f_p(p')}{\sqrt{(f_p(p'))^2-1}}\,\twobyone{-g_p(p')}{\cosh x' \,\sech x\,\sinh(y'-y) } \, \label{logarithm} , \\
&f_p(p'):=\cosh(x')\cosh(y'-y)\cosh(x)-\sinh(x')\sinh(x)\, \nonumber ,\\
&g_p(p'):= \cosh(x')\cosh(y'-y)\sinh(x)-\sinh(x')\cosh(x)\, \nonumber .
\end{align} 
We have $\norm{\exp_p^{-1}(p')}_p=\argch f_p(p')$ which is the geodesic distance between two arbitrary points $p, p'$ in the $R^2$ hyperbolic model. The goal of this section is to prove the following result.

\begin{thm}
\label{R2mainthm}
$\HH$ has a $S_1$-bounded geometry. 
\end{thm}

We note $\R^2_C:=\R^2\backslash]-\infty,0]\times \set{0}$ and $\R^2_P:= ]0,+\infty[\times ]-\pi,\pi[$. For any $x\in \R$, the map $\chi_x:\R^{2}_C\to \R^2_P$ given by $\chi_x(v_1,v_2):=(\norm{v}_x, \arctan (v_1,v_2))$ where $\arctan(v_1,v_2)$ is the unique element $\th$ of $]-\pi,\pi[$ such that $v_1+i v_2 = \norm{v} \exp(i \th)$, is a diffeomorphism with inverse $\chi_x^{-1}(r,\th)= (r \cos_x \th,r\sin_x \th).$

\begin{lem}
\label{passage}
Let $x\in \R$ and $f\in C^\infty(\R^2,\R)$ such that $ f\circ \chi_x^{-1} \in C^\infty(\R^2_P,\R)$ satisfies for any $(\a,\b)\in \N^2\backslash\set{(0,0)}$, and $(r,\th)\in \R_P^2$, $|\del^{\a,\b} f\circ \chi_x^{-1}(r,\th) | \leq C_{\a,\b} \langle r\rangle^{1-\a}$
where $C_{\a,\b}>0$. Then $f\in G_{1}(\R^2,\R)$.
\end{lem}
\begin{proof} By Theorem \ref{FaaCS}, for any $(\a,\b)\in \N^2\backslash\set{(0,0)}$, $
\del^{\a,\b}f=\sum_{1\leq |(\a',\b')|\leq |(\a,\b)|} (\del^{\a',\b'}f\circ \chi_x^{-1})\circ \chi_x \  P_{\a,\b,\a',\b'}(\chi_x)
$ on $\R^2_C$, where $P_{\a,\b,\a',\b'}(\chi_x)$ is a linear combination of functions of the form $\prod_{j=1}^s (\del^{l^j} \chi_x)^{k^j}$ where $s\in \set{1,\cdots ,\a+\b}$.  
The $k^j$ and $l^j$ are $2$-multi-indices (for $1\leq j\leq s$) such that $|k^j|>0$, $\sum_{j=1}^s k^j = (\a',\b')$ and $\sum_{j=1}^s |k^j| l^j= (\a,\b)$. By definition, $\chi_x(v)=(\chi_x^1(v),\chi_x^2(v))=(\norm{v}_x,\arctan(v_1,v_2))$. It is straightforward to check that for any 2-multi-index $\nu$, $|\del^\nu \chi_x^1 (v) |\leq C_\nu \langle v \rangle^{1-|\nu|}$ and $|\del^\nu \chi_x^2 (v)| \leq C'_\nu \langle v\rangle^{-|\nu|}$ on $\R^2_C$.  
As a consequence, for each $\a,\b,\a',\b'$ with $1\leq \a'+\b'\leq \a+\b$ there exists $C_{\a,\b,\a',\b'}>0$ such that for any $v\in \R^2_C$,
\begin{equation*}
|P_{\a,\b,\a',\b'}(\chi_x) (v)|\leq C_{\a,\b,\a',\b'}  \langle v \rangle^{\a' -(\a+\b)}\, . 
\end{equation*}
Moreover, by hypothesis, there is $C_{\a',\b'}>0$ such that for any $v\in \R^2_C$, $|(\del^{\a',\b'}f\circ \chi_x^{-1})\circ \chi_x(v)|\leq C_{\a',\b'} \langle v \rangle ^{1-\a'}$. This gives $f\in G_{1}(\R^2_C,\R)$. The extension to $G_1(\R^2,\R)$ is a direct consequence of the smoothness of $f$ on $\R^2$ and the fact that $\R^2_C$ is dense in $\R^2$. 
\end{proof}

We shall use the following proposition, which gives a formal expression of the successive derivatives of the inverse (and its real powers) of a smooth function.

\begin{prop}
\label{inverse} Let $s>0$ be given. For any nonzero $n$-multi-index ($n\in \N^*$) $\a$, 
there exist a finite nonempty set $J_\a$, nonzero real numbers $(\la_{s,\a,p})_{p\in J_{\a}}$ 
and $n$-multi-indices $\b^{\a,p,j}$ (with $p\in J_\a$, $1\leq j\leq |\a|$) such that 

\noindent - for any $p\in J_\a$, $\sum_{1\leq j\leq |\a|} \b^{\a,p,j}=\a$,

\noindent - for any smooth function $f\in C^\infty(\R^n,\R_+^*)$,
$$
\del^\a \tfrac{1}{f^s}  = \tfrac{1}{f^{|\a|+s}}\, \sum_{p\in J_\a} \la_{s,\a,p}\, \prod_{j=1}^{|\a|} \del^{\b^{\a,p,j}} f\, .
$$ 
\end{prop}
\begin{proof} The result is true for the case $|\a|=1$. Suppose then that the result holds for any $n$-multi-index $\a$ such that $|\a|=k$, where $k\in \N^*$ and let $\a'$ be a $n$-multi-index such that $|\a'|=k+1$. Let $i$ be the smallest element of $\set{1,\cdots,n}$ such that $\a'_i\geq 1$, and set 
$\a:=(\a'_1,\cdots,\a'_{i-1},\a'_i-1,\a'_{i+1},\cdots,\a'_n)$. Thus for any $f\in C^\infty(\R^n,\R^*_+)$, $\del^{\a'}\tfrac{1}{f^s}=\del_i \del^\a \tfrac{1}{f^s}$. Since $|\a|=k$, there exist a finite nonempty set $J_\a$, nonzero real numbers $(\la_{s,\a,p})_{p\in J_{\a}}$ 
and $n$-multi-indices $\b^{\a,p,j}$ (with $p\in J_\a$, $1\leq j\leq |\a|$) such that 
for any $p\in J_\a$, $\sum_{1\leq j\leq |\a|} \b^{\a,p,j}=\a$,
and such that for any $f\in C^\infty(\R^n,\R^*_+)$,
$\del^\a \tfrac{1}{f^s}  = \tfrac{1}{f^{|\a|+s}}\, \sum_{p\in J_\a} \la_{s,\a,p}\, \prod_{j=1}^{|\a|} \del^{\b^{\a,p,j}} f$. 
As a consequence, with the formula $\del_i \prod_{j=1}^{|\a|} g_j = \sum_{q=1}^{|\a|} \prod_{j=1}^{|\a|} \del^{\delta_{q,j}e_i} g_{j}$, we obtain for any $f\in C^\infty(\R^n,\R^*_+)$, 
$$
\del^{\a'} \tfrac{1}{f^s} = \tfrac{1}{f^{|\a'|+s}}\big( \sum_{p\in J_{\a}}-(|\a|+s)\la_{s,\a,p} (\prod_{j=1}^{|\a|} \del^{\b^{\a,p,j}} f)\del_i f + \sum_{(p,q)\in J_\a\times \N_{|\a|}}\la_{s,\a,p} (\prod_{j=1}^{|\a|} \del^{\delta_{q,j}e_i + \b^{\a,p,j}} f)f \big)\, .
$$
Thus, if we take $J_{\a'}=J_{\a}\coprod (J_{\a}\times \N_{|\a|})$, $\la_{s,\a',\wt p}:= -(s+|\a|)\la_{s,\a,p}$ if $\wt p =p \in J_{\a}$, $\la_{s,\a',\wt p}:= \la_{s,\a,p}$ if $\wt p = (p,q)\in J_{\a}\times \N_{|\a|}$, $\b^{\a',\wt p,j}:= \b^{\a,p,j}$ if $\wt p =p\in J_{\a}$ and $1\leq j \leq |\a|$, $\b^{\a',\wt p,j}:= e_i$ if $\wt p =p\in J_{\a}$ and $j= |\a|+1=|\a'|$, $\b^{\a',\wt p,j}:= \delta_{q,j}e_i+\b^{\a,p,j}$ if $\wt p =(p,q)\in J_{\a}\times \N_{|\a|}$ and $1\leq j \leq |\a|$ and $\b^{\a',\wt p,j}:= 0$ if $\wt p =(p,q)\in J_{\a}\times \N_{|\a|}$ and $ j = |\a|+1=|\a'|$, the result now holds for $\a'$.
\end{proof}

In the following we set the convention $J_{0}:=\set{1}$, $\la_{s,0,1}:=1$ and $\prod_{j=1}^0:=1$, so that the formula giving $\del^\a \tfrac{1}{f^s}$ in the previous lemma is still valid when $\a=0$. When $s\in \N^*$, the result is also valid for complex valued nowhere zero smooth functions.

We note $H_P$ the space of $C^\infty(\R^2_P,\R)$ functions of the form $(r,\th)\mapsto a(\th) \cosh r + b(\th) \sinh r$ where $a,b \in \mathcal{B}(\R)$, and $A_{P,k}$ the space of functions $f\in C^\infty(\R^2_P,\R)$ such that for any 2-multi-index $(\a,\b)$ with $\a\leq k \in \N$, there is $C_{\a,\b}>0$ such that for any $(r,\th)\in \R^2_P$, $|\del^{\a,\b} f (r,\th) |\leq C_{\a,\b} \langle r\rangle^{k-\a}$, and also such that for any 2-multi-index $(\a,\b)$ with $\a\geq k+1$, there is $C'_{\a,\b}>0$ such that for any $(r,\th)\in \R^2_P$,  $|\del^{\a,\b} f (r,\th)|\leq C'_{\a,\b} e^{-2r}$. Clearly, $A_{P,k}\subset S_{P,k}$ where $S_{P,k}$ is the space of functions $f\in C^\infty(\R^2_P,\R)$ such that for any 2-multi-index $(\a,\b)$, there is $C_{\a,\b}>0$ such that for any $(r,\th)\in \R^2_P$, $|\del^{\a,\b} f (r,\th) |\leq C_{\a,\b} \langle r\rangle^{k-\a}$. By Leibniz rule, $S_{P,k}S_{P,k'}\subseteq S_{P,k+k'}$. We note $N_P$ the space of functions $f\in C^\infty(\R^2_P,\R)$ such that
for any 2-multi-index $(\a,\b)$ there is $C_{\a,\b}>0$ such that for any $(r,\th)\in \R^2_P$,  $|\del^{\a,\b} f (r,\th)|\leq C_{\a,\b} e^{-2r}$. If $r_0>0$ we define the spaces $H_{P,r_0}$, $A_{P,k,r_0}$, $S_{P,k,r_0}$ and $N_{P,r_0}$ exactly as before, except that we now replace the domain $\R^2_P$ by $\R^2_{P,r_0}:=]r_0,+\infty[\times ]-\pi,\pi[$.

\begin{lem}
\label{techno}
Let $f,g,h,w\in H_{P,r_0}$ where $r_0>0$, such that there is $\eps>0$, $C>1$ such that for any $(r,\th)\in \R^2_{P,r_0}$, $f\geq C$, $f \geq \eps\, e^{r}$ and $h^2+g^2 \geq \eps\, e^{2r}$.

\noindent (i) The functions $\tfrac{w}{(h^2+g^2)^{3/2}}$, $\tfrac{w}{(f^2-1)^{3/2}}$ and any function of the form $(r,\th)\mapsto \tfrac{\sum_{k=-4}^{4} b_k(\th)e^{kr}}{((h^2+g^2)(1+h^2+g^2))^{3/2}}$, where $b_k\in \B(\R)$, are in $N_{P,r_0}$.

\noindent (ii) The functions $\argch \sqrt{1+h^2+g^2}$ and $\argch f$ are in $A_{P,1,r_0}$.

\noindent (iii) The functions $\tfrac{w}{\sqrt{h^2+g^2}}$ and $\tfrac{w}{\sqrt{f^2-1}}$ are in  $A_{P,0,r_0}$.

\end{lem}

\begin{proof}
$(i)$ We give a proof for $\tfrac{w}{(h^2+g^2)^{3/2}}$. The other cases are similar. By Proposition \ref{inverse} and Leibniz rule, we have for any 2-multi-index $\nu$, 
$$
\del^{\nu}\tfrac{w}{(h^2+g^2)^{3/2}} = \sum_{\nu'\leq \nu } \tbinom{\nu}{\nu'} \tfrac{\del^{\nu-\nu'} w }{(h^2+g^2)^{3/2+|\nu'|}} \sum_{p\in J_{\nu'}} \la_{3/2,\nu',p} \prod_{j=1}^{|\nu'|} \del^{\b^{\nu',p,j}} (h^2+g^2)\, .
$$
Note that we have for any 2-multi-index $\nu$, $\del^\nu (h^2+g^2) = \O(e^{2r})$ and $\del^\nu w = \O(e^r)$. The result follows.

\noindent $(ii)$ By $(i)$, since $\del_r^2 \argch \sqrt{1+h^2+g^2}$ is of the form $(r,\th)\mapsto \tfrac{\sum_{k=-4}^{4} b_k(\th)e^{kr}}{((h^2+g^2)(1+h^2+g^2))^{3/2}}$ where $b_k\in \B(\R)$, and $\del_r^2 \argch f$ is of the form $\tfrac{w}{(f^2-1)^{3/2}}$ where $w \in H_{P,r_0}$, we only need to check that for $0\leq \a \leq 1$, and $\b\in \N, \del^{\a,\b} \argch \sqrt{1+h^2+g^2} = \O(\langle r \rangle^{1-\a})$ and  $\del^{\a,\b} \argch f= \O(\langle r \rangle^{1-\a})$. Since $\del_{r} \argch \sqrt{1+h^2+g^2} = \tfrac{(\del_r h)h+(\del_r g)g}{\sqrt{(h^2+g^2)(1+h^2+g^2)}}$, $\del_{r} \argch f = \tfrac{\del_r f}{\sqrt{f^2-1}}$, $\del_{\th} \argch \sqrt{1+h^2+g^2} = \tfrac{(\del_\th h)h+(\del_\th g)g}{\sqrt{(h^2+g^2)(1+h^2+g^2)}}$ and $\del_{\th} \argch f = \tfrac{\del_\th f}{\sqrt{f^2-1}}$, the result follows from an application of Proposition \ref{inverse}.

\noindent $(iii)$ By $(i)$, since $\del_r \tfrac{w}{\sqrt{h^2+g^2}}$ is of the form $\tfrac{w_1}{(h^2+g^2)^{3/2}}$ where $w_1 \in H_{P,r_0}$, and $\del_r \tfrac{w}{\sqrt{f^2-1}}$ is of the form $\tfrac{w_2}{(f^2-1)^{3/2}}$ where $w_2 \in H_{P,r_0}$, we only need to check that for $\b\in \N, \del^{0,\b}\tfrac{w}{\sqrt{h^2+g^2}} = \O(1)$ and  $\del^{0,\b} \tfrac{w}{\sqrt{f^2-1}}= \O(1)$. This is a direct consequence of Proposition \ref{inverse}.
\end{proof}

\begin{proof}[Proof of Theorem \ref{R2mainthm}] By Lemma \ref{Ssigma} $(iii)$ and Proposition \ref{ssigmDiff}, it is sufficient to prove that for any $p:=(x,y)\in \R^2\backslash\set{0}$, $\exp_{p}^{-1}\circ \exp_0$ and $\exp_{0}^{-1}\circ \exp_p$ are in $G_1(\R^2)$. A computation based on (\ref{exponential}) and (\ref{logarithm}) shows that on $\R^2_P$,
\begin{align*}
&\exp_{p}^{-1}\circ \exp_0 \circ \chi_{0}^{-1}= (\argch f )\big (\tfrac{w_1}{\sqrt{f^2-1}}  ,  \tfrac{w_2}{\sqrt{f^2-1}}  \big )\, ,\\
&\exp_{0}^{-1}\circ \exp_{p}\circ \chi_x^{-1} = (\argch \sqrt{1+h^2+g^2} )\big(\tfrac{h}{\sqrt{h^2+g^2}} , \tfrac{g}{\sqrt{h^2+g^2}} \big)\, ,
\end{align*}
\noindent where 
\begin{align*}
&f(r,\th):= \cosh r \cosh y \cosh x - \sinh r (\sinh x \cos \th + \sinh y \cosh x \sin \th)\, ,\\
&w_1(r,\th):= -\cosh r \cosh y \sinh x + \sinh r (\cosh x \cos \th + \sinh y \sinh x \sin \th)\, ,\\
&w_2(r,\th):= - \cosh r \sinh y \sech x +\sinh r \sin \th \cosh y \sech x \, ,\\
&h(r,\th):= \cosh r \sinh x + \sinh r \cosh x \cos_x \th \, , \\
&g(r,\th):= \cosh r \cosh x \sinh y + \sinh r (\sinh x \sinh y \cos_x \th + \cosh x \cosh y \sin_x \th)\, .
\end{align*}
All these functions belong to $H_P$ and $f\geq 1$. Note that $f(r,\th)=1$ if any only if $\exp_0\chi_0^{-1}(r,\th)=p$, in which case $\exp_{p}^{-1}\circ \exp_0 \circ \chi_{0}^{-1}(r,\th)=0$,  so that $\exp_{p}^{-1}\circ \exp_0 \circ \chi_{0}^{-1}$ is well defined as a smooth function on the whole $\R^2_{P}$. The same argument holds for $\exp_{0}^{-1}\circ \exp_{p}\circ \chi_x^{-1}$. We check that 
$$
\half(\cosh x \cosh y - \sqrt{\cosh^2 x \cosh^2 y -1}) e^{r}\leq  f(r,\th)\leq \cosh r \ e^{\argch (\cosh x \cosh y)} 
$$
so that by defining $r_0:= \log 2/\eps$ where $0<\eps< \min\set {1, \half(\cosh x \cosh y - \sqrt{\cosh^2 x \cosh^2 y -1}) }$ we have for any $(r,\th)\in \R^2_{P,r_0}$, $f(r,\th)\geq \eps e^r\geq 2$. Note also that for any $v\in \R^2_C$, we have $\argch f(\chi_0(v)) = \norm{\exp_p^{-1}\circ \exp_0 (v)}_p$ and 
$$
\argch \sqrt{1+h^2(\chi_x(v))+g^2(\chi_x(v))}= \norm{\exp_0^{-1}\circ \exp_p(v)}_0\, .
$$
The first equality entails (since $\exp_p^{-1}\circ \exp_0 (\R^2_C)$ is a dense open subset of $\R^2$) that for any $v$ in $\R^2$, $\cosh \norm{v}_p \leq \cosh \norm{\exp_0^{-1}\circ \exp_p (v)}_0 e^{\argch (\cosh x\cosh y)}$. We then obtain for any $(r,\th)\in \R^2_P$, $\sqrt{1+h^2+g^2}\geq \cosh r \ e^{-\argch (\cosh x \cosh y)}$. In particular, defining 
$$r'_0:=\argch (\sqrt{2}\exp (\argch (\cosh x \cosh y))),$$ we get for any $r\geq r'_0$, the following estimate $h^2+g^2\geq \tfrac{1}{8} e^{-2\argch(\cosh x\cosh y)} e^{2r} $. If we now apply Lemma \ref{techno} for the space $H_{P,r''_0}$ where $r''_0:= \max\set{r_0,r'_0}$, we see that $\exp_{p}^{-1}\circ \exp_0 \circ \chi_{0}^{-1}$ and $\exp_{0}^{-1}\circ \exp_{p}\circ \chi_x^{-1}$ are in $S_{P,1}$. The result then follows from Lemma \ref{passage}.
\end{proof}

{\renewcommand{\thechapter}{}\renewcommand{\chaptername}{}
\addtocounter{chapter}{-1}
\chapter{Conclusion}\markboth{\sl CONCLUSION}
{\sl CONCLUSION}}

The Chamseddine--Connes spectral action is a fundamental object that lies at the interface of noncommutative geometry and its applications in particle physics. 
In this thesis, we studied the spectral action in noncommutative spaces such as the noncommutative torus and the quantum group $\SUq$. We have also been interested in some mathematical questions about commutative geometries (compact spin manifolds) and deformation quantization on manifolds with linearization. 
In all these works, pseudodifferential calculus (on abstract spectral triples or on manifolds) has played a crucial role. It is with the help of the fundamental notion of pseudodifferential operators that the spectral actions on the noncommutative torus and on $\SUq$ have been computed. 

\medskip

In chapter 2, we presented a computation of the spectral action on
the $n$-noncommutative torus $\big(\Coo(\T^n_\Th),\H,\DD\big)$, which is a simple spectral triple, even when pseudodifferential operators take into account the $J\A J^{-1}$ operators.
The spectral action in dimension 4 shows that a new noncommutative Yang-Mills term appears: $\tau(F_{\mu\nu}F^{\mu\nu})$, where $F_{\mu\nu}:=\delta_\mu(A_\nu)-\delta_\nu(A_\mu)-[A_\mu,A_\nu]$, which extends the usual commutative $F_{\mu\nu}:=\delta_\mu(A_\nu)-\delta_\nu(A_\mu)$.
A number theoretical condition on the deformation matrix $\Th$, related to Diophantine approximation theory, is crucial in the computation of
the spectral action, when the perturbation $D\to D+A +\eps J A J^{-1}$ is considered.
An interesting question remains: is this Diophantine condition really necessary to obtain this spectral action, with exactly the same Yang-Mills term $\tau(F_{\mu\nu}F^{\mu\nu})$? We conjecture that it is, as suggested by some heuristic arguments (Remark \ref{remdiophantine}).

\medskip

We presented in chapter 3 the computation of the spectral action on the spectral triple of \cite{DLSSV} associated to the quantum group $\SUq$.
Once again, we took into account the real structure $J$ and used the pseudodifferential techniques previously described in sections 1.2 and 1.3. However, contrarily to the case of the noncommutative torus, some remarkable new phenomena appear in this $q$-deformed noncommutative space. First, the dimension spectrum of $SU_q(2)$ is bounded below, which implies that there is only a finite number of terms in the spectral action expansion. Moreover, in this space, tadpoles do exist whereas they vanish on the noncommutative torus. We also found that the limit $q\to 1$ (which corresponds to a limit from the quantum 3-sphere towards the commutative 3-sphere) of the spectral action 
does not exist automatically, and when 
it exists, such limit does not lead to the associated action on the
commutative sphere $\mathbb{S}^3$. All these facts show that there is a ``wall'' between $q$-deformed geometries and the commutative world, that is nonexistent in the $\Th$-deformation of tori or Moyal planes.
Naturally, it would be interesting to investigate
other related cases like the quantum projective plane \cite{DDL}, Podle\'s spheres
\cite{D'AndreaDabrowski, D'ADLV} or the Euclidean quantum spheres
\cite{Landi, Dabrowski}, especially the 4-sphere \cite{D'ADL}.

\medskip

In the fourth chapter, we investigated possible cancellations of terms in the Chamseddine--Connes spectral action formula in commutative Riemannian spectral triples. We showed in particular that there
are no tadpoles, i.e. terms of the form $\ncint A\DD^{-1}$ are zero. More generally, the tadpoles are the $A$-linear terms in the internal fluctuation of the spectral action $S(\DD_A,\Phi,\Lambda)$. We considered, after Chamseddine and Connes \cite{CC2}, the case of a chiral boundary condition on the Dirac operator in a spin manifold with boundary. 
It turns out that, in this case, there are no tadpoles up to order 5. 
However, this approach is based on explicit computations of first heat kernel coefficients associated to a mixed boundary condition. We expect that there are no tadpoles at any order for a chiral boundary condition and shall investigate this question, using a spectral triple approach, in a subsequent paper \cite{avenir}. 

\medskip

We have seen, in chapter 5, certain hypotheses on the geometry ($S_\sg$ or $\O_M$-bounded geometry) of a manifold with linearization that allow a coordinate free definition of most of the topological vector spaces needed for Fourier analysis and global complete symbol calculus with uniform and decaying control over the $x$ variable. Given a linearization on the manifold with some properties of control at infinity, we constructed symbol maps and $\la$-quantization, explicit Moyal star-products on the cotangent bundle, and classes of pseudodifferential operators. We proved a result of stability under composition, and an associated symbol product asymptotic formula under a hypothesis $(C_\sg)$ of control at infinity of the linearization. The calculus presented here is a generalization of the standard and $SG$ symbol calculi over the Euclidean space $\R^n$ and can be applied to the hyperbolic 2-space since, as proven in section 5.2, it has a $S_1$-bounded geometry.  $L^2$-continuity of pseudodifferential operators of order $(0,0)$ has been established in section \ref{linkstd} under the hypothesis $(H_V)$. We do not know however if this result still holds without this hypothesis.
The full analysis of the obtained Moyal product on $\S(T^*M)$ and spectral properties of pseudodifferential operators in $\Psi_{\sg}^{l,m}$ remain to be studied. The main goal would be the construction of noncommutative noncompact spectral triples based on the algebra $(\S(T^*M),\circ_W)$, which could extend the spectral triple described in \cite{Gayral}.
Moreover, extension and connection of the symbol calculus presented here could be made with Fourier integral operators \cite{Coriasco,Ruzhansky3,Ruzhansky2}, regularized traces \cite{Paycha} and Gelfand--Shilov spaces \cite{Cappiello}.

\medskip

Finally, the spectral actions that we computed are classical, and quantization through functional integration of these actions still remains an open and challenging problem.

{\renewcommand{\thechapter}{}\renewcommand{\chaptername}{}
\addtocounter{chapter}{-1}
\chapter{Résumé de la thèse en français}\markboth{\sl Résumé de la thèse}
{\sl Résumé de la thèse}

Cette thèse constitue un recueil des travaux de recherche que j'ai effectués ces trois dernières années en collaboration avec Driss Essouabri, Bruno Iochum et Andrzej Sitarz.

\bigskip

La géométrie non commutative est un vaste domaine des mathématiques dont l'objet est
la généralisation de l'ensemble des concepts apparaissant en géométrie
classique. Plus particulièrement, à l'aide d'un formalisme issu de l'analyse
fonctionnelle, de la théorie des algèbres d'opérateurs, de la théorie spectrale et de
la géométrie spinorielle, la géométrie non commutative généralise notamment les
concepts d'espace topologique localement compact, de variété riemannienne orientée
compacte à spin et les calculs différentiels et intégraux de la géométrie
différentielle. Au-delà de l'intérêt purement
mathématique de la géométrie non commutative, il existe des motivations physiques
profondes qui poussent les physiciens théoriciens à utiliser ces concepts 
pour décrire les éléments fondamentaux de la physique
(l'espace-temps et les champs). Plus spécifiquement, la géométrie non commutative
apparaît comme un cadre mathématique particulièrement adapté à la formulation des
concepts quantiques et des processus de quantification. 

Il est possible de considérer que la géométrie non commutative (ou tout au moins sa
composante topologique) est née lorsqu'a été établi le théorème suivant (premier
théorème de Gelfand--Naimark) : toute $C^*$-algèbre
commutative unifère est isométriquement isomorphe à la $C^*$-algèbre des fonctions
continues sur un compact, à savoir l'espace des caractères sur l'algèbre.  Étant donné que toute l'information topologique
d'un espace est contenue dans l'ensemble des fonctions continues sur cet
espace, on peut constater que la notion de $C^*$-algèbre
unifère permet de généraliser la notion d'espace topologique compact.
La généralisation non commutative nous fait donc changer de point de vue : ce n'est
plus l'ensemble des points (l'espace topologique) qui est fondamental, mais plutôt
l'ensemble des fonctions sur l'ensemble des points. 

Ce théorème de Gelfand--Naimark a été le point de départ de la géométrie non
commutative. À partir de ce résultat fondamental, il a été possible d'étendre la
généralisation au-delà des concepts purement topologiques et de construire une
véritable géométrie riemannienne non commutative avec ses propres versions non
commutatives des notions classiques de calcul différentiel, de calcul intégral, de
fibré vectoriel, de mesure, de variétés riemanniennes à spin, etc. Ce travail
colossal a été réalisé principalement par Alain Connes
\cite{ConnesTorus,NCDG,Book,Cgeom,ConnesReality}.

\medskip

 Les deux aspects de la géométrie non commutative
(non-commutativité et "perte" de la notion de point) ne sont pas sans rappeler la
structure fondamentale de la physique quantique. En effet, en physique quantique, les
observables ne commutent pas forcément et l'évaluation $f(x)$ d'une observable $f$ en un point $x$, n'est
pas définie. En revanche, la notion d'observable existe toujours et $x(f)$ a un sens,
pourvu que $x$ soit un état, c'est à dire l'équivalent non commutatif du caractère (du
point) pouvant évaluer les observables $f$. 

\medskip

 Cette ressemblance frappante entre la
géométrie non commutative et la structure de la physique quantique n'est pas anodine
et constitue même la source principale qui a motivé le développement de la géométrie
non commutative et son application en physique. 
En particulier, la géométrie non commutative
fournit les concepts permettant d'appliquer l'idée fondamentale de la
non-commutativité des observables à l'espace-temps lui-même, donnant par là même une
de ses motivations fondamentales à la physique. 

\medskip

Les deux théories des interactions fondamentales, à savoir la théorie quantique des champs
(modèle standard) pour les interactions électrofaibles et fortes, et la relativité
générale pour l'interaction gravitationnelle, n'utilisent pas le même formalisme (la
première est quantique, la seconde est classique) et ne voient pas l'espace-temps de
la même façon (celui-ci est fixe et minkowskien pour la première, et dynamique pour la
seconde). Ces différences fondamentales ne sont pas gênantes pour l'étude des
phénomènes favorisant l'interaction gravitationnelle devant les autres ou
réciproquement, car ces théories ont chacune un grand succès expérimental dans leur
domaine d'application. Cependant, pour l'étude des phénomènes mettant manifestement en
jeu toutes les interactions (objets compacts, trous noirs, big-bang), il est
nécessaire de rendre compatible ces deux théories, et de les réunir sous un même
formalisme. L'idée généralement poursuivie par les théoriciens consiste à
généraliser le formalisme de la théorie quantique des champs à la gravitation,
autrement dit, réaliser une gravitation quantique. La poursuite de cet objectif s'est
développée à travers plusieurs approches différentes, dont notamment la théorie des
cordes, qui utilise une augmentation du nombre de dimensions, dont certaines sont
compactifiées, et la théorie de la gravité quantique à boucle, qui utilise une
structure en "spin foam" pour l'espace-temps sans utiliser de métrique
spatio-temporelle en "background" comme le fait la théorie des cordes. Aucune de ces
théories n'a reçu de confirmation expérimentale, les prédictions théoriques étant
elles-mêmes difficiles.

\medskip

L'approche suggérée par la géométrie non
commutative consiste à utiliser une généralisation non commutative de la variété
lorentzienne modélisant l'espace-temps. En introduisant la non-commutativité au niveau
même de la structure de l'espace-temps, cette approche permet d'appréhender
l'impossibilité de la continuité de l'espace-temps suggérée par la mécanique quantique
et la limite intrinsèque que constitue la longueur de Planck $l_p=\sqrt{\frac{{\cal
G}\hbar}{c^3}} \approx 10^{-33}$ cm. Cette approche a notamment permis d'unifier, au
niveau classique, les trois interactions du modèle standard avec la gravitation, et
d'interpréter géométriquement le mécanisme de Higgs en physique des particules. 
L'objet central dans l'interface entre la GNC et la physique fondamentale
est celui de triplet spectral, g\'en\'eralisation non commutative de la notion
de vari\'et\'e riemannienne \`a spin, point de d\'epart naturel pour l'\'elaboration de
th\'eories physiques.
En considérant un triplet spectral
dit ``presque commutatif'', c'est-à-dire un produit d'un triplet spectral commutatif
(variété riemannienne compacte, modélisant l'espace-temps ``continu'') avec un triplet
spectral de dimension nulle (une algèbre matricielle), on peut, à l'aide d'une fonctionnelle d'action
particulière sur le produit obtenu, appelée action spectrale de Chamseddine--Connes, retrouver le modèle standard et la relativité générale. Plus précisément,
l'action spectrale $S=\Tr \Phi(\DD/\Lambda)$ permet
d'unifier au niveau classique les interactions \'electro-faible, forte et gravitationnelle \cite{Carminati,CC,CC1,CC2,CCM,Connesaction,CGravity,ConnesMarcolli,Schucker,Jureit1,Jureit2,Stephan}.
Elle est d\'efinie \`a partir du spectre d'un op\'erateur de Dirac
$\DD$ et correspond au nombre des valeurs propres de l'opérateur de Dirac inf\'erieures ou
\'egales \`a une certaine \'echelle de masse $\Lambda$. En procédant à une fluctuation de la m\'etrique,
c'est-\`a-dire, au niveau algébrique, à une transformation de jauge g\'en\'eralis\'ee au groupe des unitaires de l'alg\`ebre du triplet spectral, il est possible d'obtenir le Lagrangien
du mod\`ele standard coupl\'e au Lagrangien gravitationnel d'Einstein--Hilbert, en d\'eveloppant cette fonctionnelle d'action en puissances de $\Lambda$.

\medskip

Dans cette th\`ese, nous nous sommes intéressés à certaines questions mathématiques associées 
au calcul d'action spectrale sur certains espaces non commutatifs tels que le tore non commutatif et la 3-sphère non commutative $SU_q(2)$. Nous nous sommes aussi intéressés à l'existence de tadpoles (termes linéaires associés au potentiel $A$ de la fluctuation métrique dans l'action spectrale) dans le cas de géométries riemanniennes commutatives et à la construction d'un calcul symbolique global générant un produit de Wey--Moyal sur les sections rapidement décroissantes d'un fibré d'une variété avec linéarisation. Dans tous ces travaux, l'outil fondamental a été le calcul pseudodifférentiel, qu'il soit abstrait (sur un triplet spectral quelconque), ou symbolique (sur les variétés). 

\medskip

Cette thèse est divisée en cinq parties. Voici un résumé de chacune de ces parties :

\subsection*{1. Action spectrale sur triplets spectraux}

Ce chapitre, ainsi que le chapitre suivant, a fait l'objet de la publication \emph{Spectral action on noncommutative torus} \cite{MCC}, qui est un travail en collaboration avec Driss Essouabri, Bruno Iochum et Andrzej Sitarz. 

Un triplet spectral est la donnée d'une algèbre involutive $\A$ représentée fidèlement par des opérateurs bornés sur un espace de Hilbert $\H$ et d'un opérateur autoadjoint $\DD$ sur $\H$ à résolvante compacte. On demande d'autre part que les commutateurs de $[\DD,\A]$ soient bornés. Afin de pouvoir construire un calcul pseudodifférentiel et une théorie de champ non commutative, il est utile d'introduire des hypothèses supplémentaires sur le triplet $(\A,\H,\DD)$. On dit notamment que le triplet est de dimension $n$ si les valeurs singulières $(\la_j)_j$ de l'operateur $\DD$ sont du type $\la_j=\O( j^{-1/n})$ et qu'il est régulier si $\A$ et $[\DD,\A]$ sont dans $\cap_{k} \Dom \delta^k$, où $\delta(T):=[|\DD|,T]$ (c'est à dire, qu'il est possible de "dériver" tout élément de l'algèbre).
Afin d'avoir une théorie contenant un opérateur de conjugaison de charge, il est nécessaire d'introduire une notion de structure réelle sur le triplet spectral: il s'agit d'un opérateur anti-unitaire $J$ qui commute ou anticommute avec $\DD$, selon la dimension du triplet: $\DD J = \eps J\DD$, avec $\eps\in\set{0,1}$. Dans cet environnement, les bosons de jauge sont vus comme des fluctuations internes de l'opérateur de Dirac (c'est-à-dire, au niveau classique, de la métrique): $\DD\to \DD_A:=\DD+A+\eps J A J^{-1}$, où ici $A$ est une 1-forme autoadjointe, c'est à dire une combinaison linéaire d'opérateurs du type $a[\DD,b]$, où $a$ et $b$ sont des éléments de l'algèbre $\A$.

\'Etant donné un triplet régulier $(A,\H,\DD)$ avec structure réelle $J$, un point fondamental pour faire le lien avec la physique, est d'obtenir une fonctionnelle d'action. Le principe de l'action spectrale de A. Chamseddine et A. Connes dit que la fonctionnelle suivante
$$
\SS(\DD_{A},\Phi,\Lambda):=\Tr \big( \Phi( \DD_{A} /\Lambda) \big)
$$
où $\Phi$ est une fonction de cut-off et $\Lambda$ un paramètre d'échelle de masse, est la fonctionnelle fondamentale qui peut être utilisée à la fois au niveau classique pour comparer différents espaces géométriques et au niveau quantique dans la formulation par intégrale fonctionnelle, après rotation de Wick à partir de la signature euclidienne.

Afin de pouvoir calculer précisément cette fonctionnelle en fonction de $A$ et en retirer le maximum d'information, il apparait fondamental de pouvoir développer un calcul pseudodifférentiel abstrait sur le triplet $(\A,\H,\DD)$.

En posant $ D:= \DD
+P_0$, où $P_0$ est la projection orthogonale sur $\Ker \DD$, et
\begin{align*}
OP^0&:=\{T \, : \, t\mapsto F_t(T) \in
C^\infty\big(\R,\B(\H)\big)\},\\
OP^\a&:= \set{T \,: \, T |D|^{-\a} \in  OP^0},
\end{align*}
où $F_t(T):=e^{it|\DD|}\,T\,e^{-it|\DD|}$, on peut introduire la définition suivante des opérateurs pseudodifférentiels, qui forment une algèbre $\Z$-graduée, en tenant compte à la fois de la valeur absolue de l'opérateur de Dirac, et de la structure réelle $J$ :

\begin{deffr}
Soit $\DD(\A)$ l'algèbre générée par $\A$, $J\A J^{-1}$, $\DD$ et $|\DD|$.
Un operateur pseudodifférentiel est un opérateur $T$ tel qu'il existe $d\in \Z$ tel que
pour tout $N\in \N$, il existe $p\in \N_0$, $P\in \DD(\A)$ et $R\in
OP^{-N}$ tels que $P\,D^{-2p}\in OP^d$ et
$$
T=P\,D^{-2p}+R\, .
$$
\end{deffr}

Il se trouve que si le triplet spectral est \emph{simple}, c'est dire si les fonctions $s\mapsto \zeta_{D}^P(s):=\Tr(P |D|^{-s})$, où $P$ est un opérateur pseudodifférentiel d'ordre zéro (la structure réelle $J$ étant prise en compte), sont méromorphes sur $\C$ avec uniquement des pôles simples, alors la fonctionnelle suivante (appelée intégrale non commutative)
$$
\ncint  P := \Res_{s=0} \Tr P |D|^{-s}
$$
est une trace sur l'algèbre des opérateurs pseudodifférentiels. Etant donné qu'un développement du type "noyau de la chaleur" (voir par exemple \cite[Theorem 1.145]{ConnesMarcolli}) implique
\begin{align*}
    \SS(\DD_{A},\Phi,\Lambda) \, = \,\sum_{0<k\in Sd^+} \Phi_{k}\,
    \Lambda^{k} \ncint \vert D_{A}\vert^{-k} + \Phi(0) \,
    \zeta_{D_{A}}(0) +\mathcal{O}(\Lambda^{-1}),
\end{align*}
où $\Phi_{k}= \half\int_{0}^{\infty} \Phi(t) \, t^{k/2-1} \, dt$ et
$Sd^+$ est la partie strictement positive du spectre de dimension de $(\A,\H,\DD)$ (ensemble des pôles des fonctions $\zeta_D^P$), la principale difficulté est le calcul des termes $\ncint \vert D_{A}\vert^{-k}$, $\zeta_{D_{A}}(0)$.
En utilisant le calcul pseudodifférentiel précédent et le fait que l'intégrale non commutative est nulle sur l'espace des opérateurs dans $\L^1(\H)$, on peut alors établir les résultats suivants, en notant $\wt A:=A +\eps J A J^{-1}$, $X:=\set{\wt A,\DD}+\wt A^2$, $\nabla(T):=[\DD^2,T]$, $\eps(T):=\nabla(T)D^{-2}$, $g(s,r):=\tbinom{s/2}{r}$:
\begin{align*}
&\zeta_{D_{A}}(0)-\zeta_{D}(0)=\sum_{q=1}^{n} \tfrac{(-1)^{q}}{q}
\ncint (\wt AD^{-1})^{q} \, , \\
&\ncint |D_A|^{-(n-k)}= \ncint |D|^{-(n-k)} +
\sum_{p=1}^k \sum_{r_1,\cdots, r_p =0}^{k-p}
\underset{s=n-k}{\Res} \, h(s,r,p) \, \Tr\big(\eps^{r_1}(Y)
\cdots\eps^{r_p}(Y) |D|^{-s}\big)\, , \\
&h(s,r,p):=(-s/2)^p\int_{0\leq t_1\leq \cdots \leq t_p\leq 1}
g(-st_1,r_1)\cdots
g(-st_p,r_p) \, dt\, ,\\
&Y \sim \sum_{p=1}^N\sum_{k_1,\cdots,k_p
=0}^{N-p}\tfrac{(-1)^{|k|_1+p+1}}{|k|_1+p}
\nabla^{k_p}(X\nabla^{k_{p-1}}(\cdots
X\nabla^{k_1}(X)\cdots)) D^{-2(|k|_1+p)} \mod OP^{-N-1}\, .
\end{align*}

\subsection*{2. Action spectrale sur le tore non commutatif}
Nous avons appliqué ces résultats au tore non commutatif. Il s'agit du triplet spectral non commutatif le plus simple possible. Il est basé sur l'algèbre $\Coo(\T^n_\Th)$ représentée par des fonctions rapidements décroissantes du type $\sum_{l\in \Z^n} a_l U_l$ où $(a_l)\in \S(\Z^n)$ et les éléments $U_l$ sont des unitaires vérifiant la loi de commutation 
$$
U_l U_k = e^{-i\, l\cdot\Th k} U_{k}U_l
$$
où $\Th$ est une matrice antisymétrique de déformation.
L'opérateur de Dirac est de la forme $\DD=-i\ga^\mu \delta_\mu$, où les $\ga^\mu$ sont les matrices gamma usuelles agissant sur $\C^{2^{[n/2]}}$ et $\delta_\mu(U_k):=ik_\mu U_k$. La fonctionnelle $\tau(a):=a_0$ où $a:=\sum_l a_l U_l$ génère un espace de Hilbert $\H$ de type GNS à partir duquel un triplet spectral réel régulier de dimension $n$ peut être construit. 
Le calcul des intégrales non commutatives précédentes fait intervenir des termes du type $JAJ^{-1}$. Il en résulte que nous sommes alors amenés à étudier des résidus de séries de fonctions zêta pondérées par des suites rapidement décroissantes et faisant intervenir une phase dépendante de la pondération. Plus précisément, nous avons à étudier du point de vue de l'analyse complexe, les fonctions du type
$$
g(s):= {\sum}_{l\in (\Z^n)^{q}}\,
b(l)\,f_{\Th\, \sum_{i=1}^q \eps_i l_i}(s)
$$
où $\eps_i\in \set{-1,1}$, $b\in \S((\Z^{n})^q)$, $f_a(s):={\sum}'_{k\in \Z^n}
\frac{P(k)}{|k|^s}\,
e^{2\pi i k.a}$, $a\in \R^n$ et $P$ un polynôme homogène de degré $d$.
Il apparait alors que nous pouvons connaitre précisément les pôles de ces fonctions et calculer précisément les résidus correspondants si nous faisons une hypothèse reliée à la théorie de l'approximation diophantienne sur la matrice de déformation $\Th$. Plus exactement, on établit que si $\tfrac{1}{2\pi}\Th$ est une matrice diophantienne, la fonction précédente $g$ est méromorphe sur $\C$, avec au plus un pôle simple en $s=d+n$. Le résidu en ce pôle est (Theorem \ref{analytic})
$$
\underset{s=d+n}{\Res} \, g(s) = \big(\sum_{l\in {\cal Z}} \, b(l) \big)\, 
\int_{u\in
  S^{n-1}} P(u)\, dS(u) = \sum_{l\in (\Z^n)^q} b(l) \Res_{s=d+n} f_{\Th\,\sum_{i=1}^q \eps_i l_i}(s)
$$  
où ${\cal Z}:=\{l\in(\Z^n)^{q} \ : \ \sum_{i=1}^q \eps_i
l_i= 0\}$.
Finalement, nous obtenons grâce a ces résultats, le théorème suivant:

\begin{theofr} Si $\tfrac{1}{2\pi}\Th$ est une matrice réelle antisymétrique diophantienne, alors le tore non commutatif (avec structure réelle) est un triplet spectral simple et son action spectrale est: 

\noindent $(i)$ pour $n=2$,
$$
\SS(\DD_{A},\Phi,\Lambda)=4\pi\,\Phi_{2} \, \Lambda^{2} +
\mathcal{O}(\Lambda^{-2}),
$$
\noindent $(ii)$ pour $n=4$,
$$
\SS(\DD_{A},\Phi,\Lambda)= 8\pi^2\,\Phi_{4} \, \Lambda^{4}
-\tfrac{4\pi^{2}}{3}\
\Phi(0) \,\tau(F_{\mu\nu}F^{\mu\nu})+  \mathcal{O}(\Lambda^{-2}),
$$
$(iii)$ de façon générale:
$$
\SS(\DD_{A},\Phi,\Lambda) \, = \,\sum_{k=0}^n \Phi_{n-k}\,
c_{n-k}(A) \,\Lambda^{n-k} +\mathcal{O}(\Lambda^{-1}),
$$
où $c_{n-2}(A)=0$, $c_{n-k}(A)=0$ pour $k$ impair. En particulier, $c_0(A)=0$
quand $n$ est impair.
\end{theofr}
Lorsque $n=4$, un terme de type Yang-Mills non commutatif apparait au niveau du terme invariant d'échelle (en $\Lambda^0$): $\tau(F_{\mu\nu}F^{\mu\nu})$, où $F_{\mu\nu}:=\delta_\mu(A_\nu)-\delta_\nu(A_\mu)-[A_\mu,A_\nu]$. 
La forme de l'action spectrale du tore non commutatif est donc très fortement similaire à celle du tore commutatif, à condition que la déformation $\Th$ vérifie une condition diophantienne. Nous ne savons pas cependant si cette condition, bien que suffisante, est effectivement nécessaire pour obtenir une telle action. Nous conjecturons que c'est bien le cas, comme le suggère un argument heuristique (Remark \ref{remdiophantine}).

\subsection*{3. Action spectrale sur $\SUq$}

Le troisième chapitre correspond à un travail en collaboration avec Bruno Iochum et Andrzej Sitarz: \emph{Spectral action on $\SUq$} \cite{MC}.

L'objectif consiste à appliquer les techniques pseudodifférentielles vues précédemment au calcul de l'action spectrale sur le triplet spectral de D\c{a}browski et al. \cite{DLSSV} basé sur le groupe quantique $SU_q(2)$. Ce groupe quantique (ou algèbre de Hopf) peut être vu comme une $q$-déformation de la 3-sphère commutative.

L'algèbre $\A := \A(SU_q(2))$ de ce triplet est définie comme étant l'algèbre polynomialement engendrée par deux éléments $a$ et $b$ qui sont assujettis aux règles de commutation suivantes, où $0<q<1$:
\begin{eqnarray*}
&ba = q \,ab,  \qquad  b^*a = q\,ab^*, \qquad bb^* = b^*b,
&a^*a + q^2 \,b^*b = 1,  \qquad  aa^* + bb^* = 1\, .
\end{eqnarray*}
On définit un espace de Hilbert $\H=\H^{\up} \oplus \H^{\dn}$ avec les bases orthonormales $v^{j\up}_{m,l}$ et
$v^{j\dn}_{m,l}$ où $j\in \half \N$, $0\leq m \leq 2j$, $0\leq l\leq 2j+1$
et $v^{\dn,j}_{m,l}$ est nul si $j=0$ ou $l = 2j$ or $2j+1$.

On représente alors $\A$ avec l'application $\pi$ initialement définie sur $a, b$ et qui est donnée à la section 2.2.
Cette représentation faisant intervenir des coefficients assez compliqués $\alpha^\pm_{j\mu n}$, $\beta^\pm_{j\mu n}$, et non diagonaux, il apparait très utile de définir une représentation approximée $\ul\pi$ telle que
$$
\piappr(a) := {a}_+ + {a}_-, \quad \piappr(b) := {b}_+ + {b}_-
$$
où on a posé (ici $q_n:=\sqrt{1-q^{2n}}$):
\begin{align*}
a_{+}\, v^j_{m,l} &= q_{m+1} \,q_{l+1}\, v^{j^+}_{m+1,l+1}\,,\quad
a_-\,v^j_{m,l} = q^{m+l+1}\, v^{j^-}_{m,l} \,,\nonumber\\
b_+\, v^j_{m,l} &= q^{l}\,q_{m+1}\, v^{j^+}_{m+1,l} \,,\qquad \quad
b_-\, v^j_{m,l} = -q^{m}\,q_{l}\,  v^{j^-}_{m,l-1}\,\, .
\end{align*}
Cette approximation ne modifiera pas les calculs d'intégrale non commutative puisque pour tout $x\in \A$, $\pi(x)-\piappr(x) \in \K_{q}$ où
$\K_{q}$ est un idéal inclus dans les opérateurs de type $OP^{-\infty}$. Autrement dit, $\pi(x)-\piappr(x)$ est un opérateur pseudodifférentiel régularisant.

L'opérateur de Dirac est défini de la façon suivante: $\DD\,v^j_{ml}=\genfrac{(}{)}{0pt}{1}{2j + \sesq \quad\,\,\,0}{\,0
\quad -2j- \half} \, v^j_{ml}$. 
Il possède donc le même spectre que l'opérateur de Dirac associé à la structure spinorielle de la 3-sphère standard. Cependant, l'opérateur de Dirac sur $SU_q(2)$ possède une caractéristique très particulière qui n'existe pas sur la 3-sphère commutative: le signe de $\DD$, noté $F:=\DD|\DD|^{-1}$, commute modulo $OP^{-\infty}$ avec les éléments de l'algèbre $\A$. Ceci a pour conséquence que les 1-formes $a[\DD,b]$ sont essentiellement équivalentes aux $\delta$-1-formes $a[|\DD|,b]$ dans les calculs d'intégrales non commutatives. 
On peut construire avec cet opérateur de Dirac un triplet spectral $(\A,\H,\DD)$ régulier et de spectre de dimension $\set{1,2,3}$ sur $SU_q(2)$. D'autre part, une structure réelle $J$ peut être construite sur ce triplet, avec la relation $J\DD= \DD J$.
Cette structure réelle est définie par 
\begin{align*}
&J  \, v^{j\up}_{m,l}= i^{2(m+l)-1}v^{j\up}_{2j-m,2j+1-l} \,,\qquad
J  \, v^{j\dn}_{m,l}= i^{-2(m+l)+1}v^{j\dn}_{2j-m,2j-1-l}\,.
\end{align*}
En utilisant une décomposition de type Poincaré--Birkhoff--Witt sur les $\delta$-1-formes pour calculer les intégrales $\ncint{|D_A|^{-k}}$, nous avons été amenés à étudier certains produits de séries faisant intervenir une inversion d'indice $m\mapsto 2j-m$ (sections 2.4.4 et 2.8.C, 2.8.D). Ces résultats ont permis de montrer que la structure réelle ne modifie pas le spectre de dimension $\set{1,2,3}$ de $SU_q(2)$, et que l'action spectrale du triplet $(\A,\H,\DD)$ qui tient compte de la structure réelle, c'est à dire associée à la perturbation $\DD\to \DD+ A +JAJ^{-1}$, est totalement déterminée par les intégrales suivantes (qui ne font plus intervenir la structure réelle $J$)
$$
\ncint A_\delta^{q} |\DD|^{-p}, \qquad 1\leq q\leq p\leq 3\, ,
$$
où $A_\delta$ est une $\delta$-1-forme.
Afin de calculer précisément ces intégrales, un calcul différentiel sur $SU_q(2)$ modulo un idéal $\CR$ a été défini. Cet idéal est conçu de telle sorte que tout opérateur $T$ dans $\CR$ est invisible par intégration non commutative avec $|\DD|^{-2}, |\DD|^{-3}$: $\ncint T |\DD|^{-2}=\ncint T|\DD|^{-3}=0$.
Le calcul d'intégrale du type $\ncint A_\delta^q |\DD|^{-p}$ est alors réduit à certains types particuliers de $\delta-1$-formes $A_\delta$. Par exemple, il est possible d'obtenir toutes ces intégrales avec $p=1$ à partir des intégrales suivantes (Proposition \ref{calculavecLM}):
\begin{align*}
& \ncint (bb^*)^n \,|\DD|^{-1} =
\tfrac{-2(1+q^{2n})}{(1-q^{2n})^2}\,, \\
& \ncint (bb^*)^{n} b^* \,\delta b\, |\DD|^{-1} = 
\ncint (bb^*)^{n} b \,\delta b^*\, |\DD|^{-1} =
\tfrac{2}{1-q^{2n+2}}\,, \\
& \ncint (bb^*)^n a \, \delta a^* \,|\DD|^{-1}  =
\tfrac{-2q^{4n+2}-2q^{4n}-2q^{2n+2}+6q^{2n}}{(1-q^{2n})^2(1-q^{2n+2})}\,,\\
& \ncint (bb^*)^n a^* \, \delta a \,|\DD|^{-1}  =
\tfrac{6q^{2n+2}-2q^{2n}-2q^2-2}{(1-q^{2n})^2(1-q^{2n+2})}\,.
\end{align*}
Ces résultats permettent finalement de retrouver toutes les actions spectrales possibles sur $\SUq$. Nous avons constaté, à l'aide de certains exemples, que les termes de l'action spectrale n'ont pas toujours une limite finie lorsque $q\to 1$, c'est à dire lorsque $SU_q(2)$ "s'approche" de la 3-sphère commutative $\mathbb{S}^3$. D'autre part, même lorsque cette limite existe, le terme obtenu n'est pas égal au terme correspondant à la 3-sphère.
Il existe donc un "mur" entre la $q$-géométrie de $SU_q(2)$ et $\mathbb{S}^3$ qui n'apparait pas au niveau des déformations du tore ou des plans de Moyal. 
Le calcul d'action spectrale sur d'autres géométries, telles que les sphères de Podle\'s
\cite{D'AndreaDabrowski, D'ADLV}, le plan projectif \cite{DDL}, ou les $q$-sphères euclidiennes
\cite{Landi, Dabrowski,D'ADL} pourrait faire l'objet d'investigations futures.

\subsection*{4. Tadpoles et triplets spectraux commutatifs}

Ce chapitre présente un travail fait en collaboration avec Bruno Iochum: \emph{Tadpoles and commutatives spectral triples} \cite{Tadpole}. Nous nous sommes intéressés à certaines questions concernant l'annulation d'intégrales non commutatives apparaissant dans l'action spectrale de Chamseddine--Connes. Plus particulièrement, nous avons étudié les intégrales du type $\ncint A\DD^{-1}$ (linéaires en $A$, où $A$ est une 1-forme) qui correspondent, en théorie des champs, à des tadpoles. Ici, $\DD^{-1}$ est le propagateur fermionique de Feynman et $A\DD^{-1}$ est un graphe à une boucle avec une ligne fermionique interne et une ligne bosonique externe :
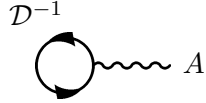
\begin{fmffile}{tadpolegraph}
\[
{\begin{picture}(100,60)
\put(0,0){\begin{fmfgraph}(50,60)
\fmfleft{l}
\fmfright{r}
\fmffreeze
\fmf{photon,tension=3}{r,w}
\fmf{photon,tension=3}{w,v}
\fmf{fermion,right}{v,l}
\fmf{fermion,right}{l,v}
\end{fmfgraph}}
\put(-10,45){\mbox{${\cal D}^{-1}$}}
\put(55,27){\mbox{$A$}}
\end{picture}} 
\]
\end{fmffile}

\vspace{-1cm}

Plus généralement, on définit un tadpole comme étant un terme linéaire en $A$ apparaissant dans l'action spectrale $\S(\DD+A,\Phi,\Lambda)$. Si $d$ est la dimension d'un triplet spectral $(\A,\H,\DD)$, on peut montrer que la partie linéaire en $A$ du terme en $\Lambda^{d-k}$ de l'action spectrale, que l'on note $\Tad_{\DD+A}(d-k)$ (tadpole d'ordre $k$) vérifie:
\begin{align*}
&\Tad_{\DD+A}(d-k) =-(d-k)\ncint A \DD |\DD|^{-(d-k) -2}, \quad \forall  k\neq d,\\
    \label{tadpole0}
& \Tad_{\DD+ A}(0)=-\ncint  A \DD^{-1}.
\end{align*}
Il apparait alors que pour tout triplet spectral riemannien, c'est à dire du type $\big( \A:=C^{\infty}(M),\, \H:=L^2(M,S),\,\DD \big)$ où $M$ est une variété compacte sans bord riemannienne à spin de dimension $d$, $\H$ l'espace de Hilbert des spineurs de carré intégrable et $\DD$ l'opérateur de Dirac associé à la structure spin, tous ces termes sont nuls. Cette propriété est basée sur le fait que ce triplet est réel et commutatif. La structure réelle de $\big( \A:=C^{\infty}(M),\, \H:=L^2(M,S),\,\DD \big)$ provient de l'existence de la structure spin, car celle-ci implique l'existence d'un opérateur de conjugaison de charge $J$, qui est une isométrie antilinéaire satisfaisant:
$$
JaJ^{-1}=a^*,\quad \forall a \in \A \, .
$$
De façon plus générale, on peut montrer, en utilisant le résidu de Wodzicki, que $\ncint B|\DD|^{-(2k+1)}$, où $B$ est un polynôme généré par $\A$ et $\DD$, est toujours nul si la dimension est paire, alors que $\ncint B \DD^{-2k}$ est toujours nul en dimension impaire. Dit autrement, $\ncint B |D|^{-(d-q)} =0$ pour tout $q$ impair.

Nous nous sommes aussi intéressés au cas d'une variété de dimension paire compacte avec bord, et d'une condition au bord de type chirale, c'est-à-dire telle que l'opérateur de Dirac perturbé $\DD + A$ agit sur le domaine 
$\set{s\in C^\infty(V) \, : \, \Pi_- s_{|\del M} =0 }$ où $\Pi_\pm:=\half(1\pm \chi)$. Ici $\chi$ est un opérateur de chiralité sur le bord, c'est-à-dire tel que $\chi^2=1$ et  
\begin{align*}
\{\chi,\ga^d\} = 0 \,, \qquad  [\chi, \ga^a] = 0 ,\qquad \forall a \in \set{1,\cdots,d-1}\,.
\end{align*}
On obtient alors une condition au bord mixte naturelle sur l'opérateur de type Laplace $(\DD+A)^2$:
$$
\B_\chi^A s:= \Pi_- (D+A)^2 s_{|\del M} \oplus \Pi_{-}s_{|\del M}\, \, .
$$  
En se basant sur les formes explicites des coefficients du noyau de la chaleur dans le cadre des conditions aux bords mixtes \cite{BG1,BG2}, on peut alors montrer qu'aucun tadpole ne peut exister dans l'action $\Tr(\Phi((D+A)^2_{B^A_\chi}/\Lambda^2))$, au moins jusqu'à l'ordre $5$. Nous nous attendons à ce que ceci se généralise à tous les ordres, et à d'autres types de conditions aux bords, notamment celles d'Atiyah-Patodi-Singer. Une approche de cette question, associée aux triplets spectraux commutatifs, fait l'objet d'un travail en cours \cite{avenir}.

\subsection*{5. Calcul pseudodifférentiel global sur variétés avec linéarisation}

Le cinquième et dernier chapitre présente la construction d'un calcul symbolique et pseudodifférentiel global sur les variétés avec linéarisation.

Il a été montré par Gayral et al. \cite{Gayral2} que les plans de Moyal sont des triplets spectraux non compacts. En d'autres termes, les plans de Moyal peuvent être vus comme des variétés à spin non compactes et non commutatives. 
Ce lien entre la théorie de la quantification par déformation et la géométrie non commutative montre en particulier que le paradigme des triplets spectraux est adapté aux questions de quantification. Le produit
de Moyal est défini sur l'espace de Schwartz $\SS(\R^{2n})$ des fonctions rapidement décroissantes
\begin{equation*}
f \star g(x) := (\pi\th)^{-2n} \int_{\R^{4n}} \,dy \,dz\,f(y)\, g(z)\,
e^{\frac{2i}{\th}(x-y)\,\.\,S(x-z)}
\end{equation*}
où $\th\in \R^*$ et $S =\genfrac{(}{)}{0pt}{1}{0
\  -1_n}{1_n\  0}$, et donne à $\S(\R^{2n})$ une structure de pré-$C^*$-algèbre de Fréchet. Le triplet spectral décrit dans \cite{Gayral2} est basé sur cette algèbre. 
L'extension de cette construction remarquable à un fibré cotangent $T^*M$ d'une variété $M$ non isométrique à $\R^n$ reste un problème ouvert. Nous proposons dans ce chapitre la construction d'un calcul pseudodifférentiel global afin d'obtenir un produit de Moyal plus général que celui défini sur $\SS(\R^{2n})$.

Le calcul pseudodifférentiel global \cite{Bokobza,Drager,Widom1,Widom2} permet d'établir une notion globale de symbole total d'un opérateur pseudodifférentiel, modulo l'algèbre résiduelle des opérateurs régularisant $\Psi^{-\infty}$ (à noyau lisse).  Il est basé sur la définition d'une connexion sur la variété (ou plus généralement, d'une linéarisation \cite{Bokobza}), et utilise l'application exponentielle, ainsi que le transport parallèle sur les géodésiques associées, pour obtenir un isomorphisme global (modulo $\Psi^{-\infty}$) entre les algèbres symboliques et opératorielles. 

Lorsque la variété $M$ n'est pas compacte, il apparaît utile, afin d'avoir une continuité de type $L^2(M)$ pour les opérateurs d'ordre 0, de considérer des espaces de symboles qui contrôlent à la fois la variable $x$ et la covariable $\th$. Ces contrôles ont été utilisés sur $\R^n$ dans le cadre du calcul pseudodifférentiel de type $SG$ (voir par exemple \cite{Schrohe}). Nous avons été amenés à définir un tel calcul dans le cadre des variétés à linéarisation, c'est à dire telles qu'il existe une application exponentielle abstraite (ou linéarisation) $\exp: TM \to M$ établissant un difféomorphisme $\exp_x:T_x M\to M$ en chaque point $x\in M$. Ce cadre est suffisamment général pour contenir les variétés de Cartan--Hadamard, qui sont les variétés simplement connexes, complètes, et de courbure négative.

L'outil essentiel de la définition du calcul global dans le cadre des variétés à linéarisation est la combinatoire liée à la formule de Faa-di-Bruno à plusieurs variables \cite{Constantine}. Celle-ci s'exprime de la façon suivante:
si $f\in C^\infty(\R^p)$ et $g\in C^\infty(\R^n,\R^p)$ alors pour tout $n$-multi-indice $\nu\neq 0$,
$$
\del^\nu (f\circ g) = \sum_{1\leq |\la|\leq |\nu|} (\del^\la f)\circ g\  \sum_{s=1}^{|\nu|} \sum_{p_s(\nu,\la)}\nu! \prod_{j=1}^s \tfrac{1}{k^j!(l^j!)^{|k^j|}} (\del^{l^j} g)^{k^j}
$$
où les multi-indices $k^j$, $l^j$ appartenant à l'ensemble $p_s(\nu,\la)$ vérifient  $\sum_{j=1}^s k^j = \la$ et $\sum_{j=1}^s |k^j| l^j= \nu$.
A l'aide de cette formule, il est possible de définir, de façon intrinsèque et dans le cadre des variétés à linéarisation, des espaces de Fréchet nucléaires de fonctions rapidement décroissantes $\S(M)$, $\S(M\times M)$, $\S(TM)$, $\S(T^*M)$ pourvu que l'application exponentielle satisfasse une hypothèse de contrôle à l'infini de type polynomial. Cette hypothèse dit plus précisément que les applications de changement de coordonnées normales $\psi_{z,z'}^{\bfr,\bfr'}:=L_\bfr \exp_z^{-1} \exp_{z'} L_{\bfr'}^{-1}$ (où $z,z'\in M$, $\bfr,\bfr'$ bases de $T_z M$, $T_{z'}M$ et $L_{\bfr}$ isomorphisme linéaire associé à $\bfr$) sont dans l'espace $\O_M(\R^n,\R^n)$ des fonctions (avec leurs dérivées) polynomialement dominées à l'infini.

Il est d'autre part possible de définir, sous ces conditions, des isomorphismes topologiques de quantifications $\Op_\la$ paramétrés par $\la\in [0,1]$ qui permettent de passer de l'espace des opérateurs (plus exactement des noyaux) $\S'(M\times M)$ à l'espace des "symboles" $\S'(T^*M)$. Ces isomorphismes envoient $S(M\times M)$ dans $\S(T^*M)$ et il apparait alors possible de définir un $\la$-produit (ou produit de Moyal si $\la=\half$) sur $\S(T^*M)$ simplement par transfert de la convolution de noyau dans $\S(M\times M)$. Plus précisément, le produit 
$$ 
a\circ_\la b\,(x,\eta) = \int_{T_x(M) \times M} d\mu_{x}(\xi)d\mu(y)\int_{V^\la_{x,\xi,y}} d\mu_{x,\xi,y}^*(\th,\th')\, g^\la_{x,\xi,y}\,e^{2\pi i \om^\la_{x,\xi,y}(\eta,\th,\th')} a(y^\la_{x,\xi},\th)\,b(y^{1-\la}_{x,-\xi},\th')
$$
où les applications $V$, $g$, $\om$ sont déterminées (Proposition \ref{la-product}) par l'application exponentielle et la densité $d\mu$ considérée sur la variété, donne à $\S(T^*M)$ une structure d'algèbre de Fréchet, et se réduit précisément au produit de Moyal classique lorsque $\la=\half$ et $M=\R^n$.

Nous avons ensuite étudié l'extension du $SG$-calcul sur les variétés à linéarisation. Les espaces de symboles $S^{l,m}$ (ayant un controle séparé en $x$ et $\th$) peuvent être définis si on renforce l'hypothèse précédente de controle polynomial sur les difféomorphismes $\psi_{z,z'}^{\bfr,\bfr'}$. Plus précisément, si les applications $\psi_{z,z'}^{\bfr,\bfr'}$ vérifient $(S_1)$, c'est-à-dire si pour tout $n$-multi-indice $\a \neq 0$ (ici $\langle \rx\rangle := (1+\norm{\rx}^2)^{1/2}$)
\begin{equation*}
 \del^\a \psi_{z,z'}^{\bfr,\bfr'} (\rx)  = \O \big( \langle \rx \rangle^{1-|\a|}\big)\, ,
\end{equation*}
alors les espaces $S^{l,m}$ de symboles $a$, définis par les estimations suivantes, pour tout système de coordonées normal $(z,\bfr)$,
\begin{equation*}
\del^{(\a,\b)}_{z,\bfr} a (x,\th)=\O\big(\langle x\rangle^{l- |\a|}_{z,\bfr}\,\langle \th\rangle^{m-|\b|}_{z,\bfr,x} \big) \, ,
\end{equation*}
sont des espaces de Fréchet homéomorphes aux espaces de $SG$-symboles classiques sur $\R^n$.
A partir de ces espaces, on peut définir les opérateurs pseudodifférentiels sur $M$ en posant $\Psi^{l,m}:=\Op_\la(S^{l,m})$ pourvu que cet espace ne dépende pas du paramètre de quantification $\la$ et qu'il stabilise à la fois $\S(M)$ et $\S'(M)$. Ceci a été rendu possible par une analyse des espaces d'amplitudes $a$ et des opérateurs associés
$$
\langle \Op_{\Ga}(a) , u\rangle :=  \int_{\R^{3n}} e^{2\pi i \langle \vth,\zeta \rangle } \Tr\big(a(\rx,\zeta,\vth)\, \Ga(u)^*(\rx,\zeta)) \, d\zeta\,d\vth\,  d\rx
$$
où $\Ga$ est un isomorphisme topologique de $\S(\R^{2n})$ vérifiant certaines hypothèses de contrôle à l'infini (Proposition \ref{ampliOP}, \ref{amplContinu} et Lemma \ref{noyauReste}).

La partie suivante est consacrée à un résultat sur la composition des opérateurs pseudodifférentiels. Il est établi (Theorem \ref{compo}) sous une hypothèse particulière $(C_\sg)$ (Definition \ref{Csigma}) sur l'application exponentielle, que $\Psi :=\cup_{l,m} \Psi^{l,m}$ est une $*$-algèbre sous la composition d'opérateurs et le symbole (pour $\la=0$) du produit de deux opérateurs $A,B$ satisfait la relation asymptotique 
$$
\sigma_{0}(AB)_{z,\bfr} \sim \sum_{\b,\ga \in \N^n} c_\b c_\ga \del_{\zeta,\vth}^{\ga,\ga}\big( a(\rx,\vth)\del^{\b}_{\zeta'}\big(e^{2\pi i \langle \vth,\varphi_{\rx,\zeta}(\zeta')\rangle} (\del^{\b}_{\vth'} f_b)(\rx,\zeta,\zeta',L_{\rx,\zeta}(\vth))\big)_{\zeta'=0}\big)_{\zeta=0}
$$
où  $a:=\sigma_0(A)_{z,\bfr}$, $b:=\sigma_0(B)_{z,\bfr}$, et les termes $L$, $\phi$ capturent la "courbure" liée à l'application exponentielle abstraite $\exp$.
Nous montrons en dernière partie que l'espace hyperbolique $\HH$ de dimension 2 est une variété avec linéarisation (l'application exponentielle étant celle venant de la structure riemanienne) vérifiant l'hypothèse de contrôle $(S1)$. Ceci permet de définir de façon globale et intrinsèque les espaces de Schwartz $\S(\HH)$, $\S(T^*\HH)$, $\S(T\HH)$ ainsi que les espaces de symboles $S_1^{l,m}(T^*\HH)$.

L'analyse détaillée des $\la$-produits sur $\S(T^*\HH)$ (ou plus généralement sur $\S(T^*M)$, pour $M$ avec $\O_M$-linéarisation) et les propriétés spectrales associées restent à étudier. Il serait par exemple intéressant de voir sous quelle condition les algèbres $(\S(T^*M),\circ_\la)$ peuvent générer un triplet spectral non compact. D'autre part, il est possible d'envisager de connecter ou d'étendre le calcul symbolique présenté ici avec les opérateurs de Fourier intégraux \cite{Coriasco,Ruzhansky3,Ruzhansky2}, les traces régularisées \cite{Paycha} et les espaces de Gelfand--Shilov \cite{Cappiello}.

\end{document}